\documentclass[12pt,letterpaper]{article}
\def\blind{1} % 0 = anonymous; 1 = with authors
\def\documentpart{0} % 0 = full; 1 = main; 2 = supplement
\input{preamble-shared}
\ifnum\documentpart>0
\usepackage{xparse}
\usepackage{refcount}
\makeatletter
\def\manuscriptreadlabel#1#2{\expandafter\gdef\csname r@external-#1\endcsname{#2}}
\begingroup
\let\newlabel\manuscriptreadlabel
\relax
\endgroup
\makeatother
\AtBeginDocument{%
  \RenewDocumentCommand{\ref}{s m}{%
    \ifcsname r@#2\endcsname
      \IfBooleanTF{#1}{\getrefnumber{#2}}{\hyperref[#2]{\getrefnumber{#2}}}%
    \else
      \getrefnumber{external-#2}%
    \fi
  }%
}
\fi
\begin{document}
\ifnum\documentpart<2
  \spacingset{1}
\ifnum\blind=1
  \title{\bfseries\papertitle}
  \author[1]{Daeyoung Ham}
  \author[2]{Sihan Wu}
  \author[2]{Yifan Cui}
  \affil[1]{Department of Statistics and Data Science, University of Texas at San Antonio\\
    \texttt{daeyoung.ham@utsa.edu}}
  \affil[2]{Center for Data Science, Zhejiang University\\
    \texttt{wush03@zju.edu.cn}, \texttt{cuiyf@zju.edu.cn}}
  \maketitle
\else
  \bigskip\bigskip\bigskip
  \begin{center}
    {\LARGE\bfseries\papertitle}
  \end{center}
  \medskip
\fi
\bigskip
\begin{abstract}
We study nonparametric inference for a continuous-treatment causal
dose-response curve and its first derivative under unmeasured confounding,
using treatment- and outcome-inducing proxies. 
We establish a proximal doubly robust pseudo-outcome whose conditional mean given treatment equals the dose-response curve whenever either bridge function is correctly specified. 
Their (smoothed) efficient influence function representations lead to cross-fitted local-polynomial estimators with explicit smoothing-bias correction. 
We establish pointwise and finite-dimensional asymptotic normality and oracle uniform Gaussian
approximations over compact interior treatment intervals, allowing mean-squared-error-optimal bandwidths without undersmoothing.
The derivative inference reuses the estimated bridges and the same form
of bandwidth-dependent nuisance-rate conditions, evaluated at the bandwidth rate appropriate for derivative estimation, without requiring derivative-specific bridge estimation. 
Employing cross-fitting accommodates flexible bridge estimators under product convergence rate conditions without entropy restrictions. 
We further develop tests against fixed-dimensional parametric dose-response null hypotheses, including constancy, linearity, and quadratic specifications, and derive their limiting distributions under the null and local alternatives, together with practical bootstrap procedures. We also develop practical bandwidth selectors, pointwise confidence intervals, and simultaneous confidence bands for each causal estimand. Extensive simulations and a data analysis, presented in the appendix, demonstrate the finite-sample performance of the proposed methods under latent confounding and multiple proxies.

\end{abstract}
\vfill
\spacingset{1.7}
\clearpage

  \section{Introduction}

Much of the classical causal inference literature is framed around a binary
treatment, with the target defined as a contrast between potential outcomes
under treatment and control.  Many empirical studies, however, concern
treatments that vary continuously.  Dichotomizing such treatments can obscure
important effect variation.  Examples include dose finding in biomedical
studies, exposure--response analysis in environmental health, and policy
evaluation with continuously varying policy intensity.  The corresponding
causal estimand is the dose-response curve
\(\theta(a)=\mathbb E\{Y(a)\}\), the mean potential outcome under treatment
level \(a\).

There is now a substantial literature on causal inference with continuous
treatments.  Early work extends the propensity score to general treatment
regimes and continuous exposures
\citep{hirano2004propensity,ImaiVanDyk2004}.  Subsequent work develops
semiparametric and nonparametric estimators of dose-response functions and
related causal parameters
\citep{GalvaoWang2015,kennedy2017non,calonico2018effect}, with extensions to quantile dose-response functions and settings with measurement error \citep{ai2022estimation,huang2023nonparametric}. 
More recent papers
develop nonparametric doubly robust estimation, testing, robust bias correction,
and confidence bands for causal curves
\citep{Doss2024nonparametric,Takatsu2025debiased,ColangeloLee2026}.  These
methods provide important tools for continuous-treatment problems, but their
causal validity relies on no unmeasured confounding.  This condition is often
difficult to defend in observational studies, particularly when relevant
socioeconomic, behavioral, or biological factors are incompletely recorded.
Latent confounding can then invalidate both estimation and uncertainty
quantification.

Proximal causal inference offers a different route to addressing unmeasured
confounding.  Rather than requiring the latent confounder to be observed, it
uses treatment-inducing and outcome-inducing confounding proxies.  Under
completeness and bridge-function conditions, these proxies can identify causal
effects when adjustment for observed covariates alone is insufficient
\citep{miao2018identifying,cui2024semiparametric,tchetgen2024introduction}.
The proximal framework has since been extended to longitudinal studies
\citep{tchetgen2024introduction,Ying2023}, survival outcomes with dependent censoring \citep{Ying2024}, invalid proxies \citep{yu2025fortified},
synthetic controls \citep{shi2026theory}, mediation analysis
\citep{Dukes2023,bai2026proximal,ghassami2025causal,wu2026proximal,park2026proximal}, individualized treatment regimes and heterogeneous effects
\citep{qi2024proximal,Shen2023,sverdrup2023proximal}, and sequential
decision-making problems
\citep{shi2022minimax,Zhang2026,wang2026blessing,gao2025multiple}.  These developments
establish proxy-based adjustment as a general statistical framework for modern
confounded data structures.

Closest to our setting, \citet{wu2024doubly} construct a kernel-localized
augmented estimator of the proximal dose-response curve, together with minimax
estimators of the bridge functions.  Their analysis derives a pointwise
smoothing bias of order \(h^2\) and variance of order \((nh)^{-1}\), yielding
the MSE-optimal bandwidth \(h\asymp n^{-1/5}\) and MSE of order
\(n^{-4/5}\) \citep{wasserman2006all}.  Their pointwise
bias--variance theorem assumes consistency of both bridge estimators and a
product-rate condition, and does not establish an asymptotic distribution or
provide feasible standard errors, confidence intervals, or simultaneous
confidence bands.  Moreover, at the MSE-optimal bandwidth, the leading bias is
generically of the same order as the standard error.  
Motivated by \citet{calonico2018effect} and
\citet{Takatsu2025debiased}, we instead use a local-linear fit with
smoothing bandwidth \(h\), estimate the second derivative entering its
leading-bias correction by a local-quadratic fit with bandwidth \(b\), and
incorporate the stochastic contribution of the correction into the limiting
variance.  This permits \(h\asymp b\asymp n^{-1/5}\), without
undersmoothing.

Our starting point is a proximal doubly robust pseudo-outcome whose conditional
mean given \(A=a\) equals \(\theta(a)\) when either bridge function is correctly
specified.  
This representation separates proximal adjustment from curve learning, thereby
allowing flexible choice of the second-stage estimator. Scientifically supported
monotonicity, convexity, or other shape restrictions can be incorporated through
shape-constrained regression \citep{westling2020causal,doss2026doubly}; beyond
these special cases, we develop smoothing-based estimation and inference for
general nonparametric dose-response curves. The same construction yields targeted
improvements in the regularity conditions required for inference.

The theoretical comparison with \citet{wu2024doubly} can most naturally be made
through \citet{ColangeloLee2026}.  In the standard setting without unmeasured
confounding, \citet{ColangeloLee2026} develop pointwise inference for a related
doubly robust estimator.
Their results clarify what would be required
to extend the proximal estimator of \citet{wu2024doubly} to pointwise inference.
Applying the pointwise inference strategies of \citet{ColangeloLee2026} to this
proximal estimator would require either \(h^2\sqrt{nh}\to0\) or a separate
estimate of the leading bias.  The latter uses a pilot bandwidth \(b\), chosen
specifically for bias estimation, and requires \(b\to0\) and
\(nb^5\to\infty\).  When \(h\asymp n^{-1/5}\), these conditions imply
\(b/h\to\infty\), so bias estimation pools observations over a much wider range
of treatment values than curve estimation.  Our local-quadratic correction
instead permits \(b\asymp h\asymp n^{-1/5}\) and accounts for its estimation
error in the limiting variance, eliminating the need for a separate, more
heavily smoothed fit.

This comparison also clarifies the relevant regularity conditions.  The
pointwise theory of \citet{ColangeloLee2026} requires the joint observed-data
density, the outcome regression, and the probability limits of both nuisance
estimators to be three times differentiable in treatment, with uniformly bounded
derivatives.  A direct extension of their argument to the proximal estimator
would impose analogous smoothness conditions on the observed-data distribution
and the probability limits of the bridge estimators.  Our theory instead
requires smoothness only for \(\theta\), the marginal treatment density, and the
conditional pseudo-outcome variance.  These are low-dimensional quantities
obtained after the covariates and proxies have been properly adjusted.  
The bridge
estimators can therefore be fitted using flexible methods, provided that their \(L_2\) errors converge at the required rates.  For uniform
inference, our theory also avoids the sup-norm convergence and
treatment-Lipschitz conditions imposed on the fitted nuisances by
\citet{ColangeloLee2026}.  Finally, unlike the non-cross-fitted analysis of
\citet{Takatsu2025debiased}, our use of cross-fitting does not require complexity
restrictions on the classes used to estimate the bridges or a condition linking
their complexity to the bandwidth, which also provides flexibility in the nuisance estimation stage. 

The derivative of the dose-response curve quantifies the marginal effect
of small dose changes and distinguishes increasing, decreasing, and
locally flat response regions. \citet{zeng2025nonparametric} develops
doubly robust derivative estimators in continuous-instrument settings.
\citet{zhang2025doubly} studies causal derivative effects under adjustment
for observed confounders, with an additive-model extension for positivity
violations. These works accommodate estimated nuisance functions but do
not develop explicit smoothing-bias correction for derivative inference
at MSE-optimal bandwidths.
We address this gap in semiparametric
regression with estimated proximal nuisance bridges. To our knowledge,
this is the first framework for bias-corrected smoothing inference in a semiparametric regression setting with cross-fitting, without undersmoothing.
We remark that the derivative estimator reuses the fitted bridges and the same form of bandwidth-dependent product-rate conditions as inference for the dose-response curve, without requiring estimation of bridges tailored for this purpose.

Global tests complement curve and derivative inference by assessing whether the dose-response curve is constant or follows a specified parametric form. \citet{Doss2024nonparametric} develop a doubly robust constancy test under no unmeasured confounding. Our framework extends this line of testing methodology to settings with unmeasured confounding accommodated through proxy variables. Moreover, we extend the testing principle to fixed-dimensional parametric families, including linear and quadratic curves, and establish the null and local-alternative limits of the resulting cross-fitted test; the constancy null of \citet{Doss2024nonparametric} arises as a special case of our parametric-null formulation. While \citet{Doss2024nonparametric} propose cross-fitting as a practical implementation of their test, we provide a formal theoretical analysis of the cross-fitted procedure, which avoids the nuisance entropy conditions required by their non-cross-fitted theory under our product-rate conditions.
To our knowledge, this is the first cross-fitted (and bias-corrected smoothing) framework for testing against general fixed-dimensional parametric null
models in semiparametric regression with estimated nuisance
functions that establishes null and local-alternative limits.

Our contributions are sixfold.  First, we establish the proximal doubly robust
regression identity, formulate a local-linear estimator, derive an influence
function representation for its smoothed target, and construct a
local-quadratic debiased estimator.  Second, we construct cross-fitted estimators
and establish pointwise and finite-dimensional asymptotic normality with
\(h\asymp b\asymp n^{-1/5}\), without undersmoothing or entropy conditions on
the fitted bridge classes.  Third, we extend the theory from fixed treatment
values to compact treatment intervals and establish an oracle uniform Gaussian
approximation.
Fourth, we develop debiased inference
for \(\theta'(a)\), deriving an influence function for its smoothed target
and combining local-quadratic slope estimation with local-cubic bias
correction. Accounting for the correction's stochastic variability
permits \(h\asymp b\asymp n^{-1/7}\) without undersmoothing and yields
pointwise and finite-dimensional asymptotic normality and an oracle
uniform Gaussian approximation under the corresponding conditions.
Fifth, we develop cross-fitted tests of fixed-dimensional parametric
dose-response models, including constant, linear, and quadratic families.
We establish null and fixed local-alternative limits without entropy
restrictions on the fitted bridge classes, under some moment and nuisance product-rate conditions.
Sixth, we develop and compare practical bandwidth selectors based on two plug-in rules and leave-one-out cross-validation. 
The simulations report pointwise and simultaneous coverage for the dose-response function and its first derivative under each nuisance model specification, and evaluate the size and power of the parametric-null tests for constant, linear, and quadratic dose-response models.

The rest of the paper is organized as follows.
Section~\ref{sec:Prox_conti} introduces the proximal continuous-treatment
setup, bridge assumptions, and pseudo-outcome.
Section~\ref{sec:estimation} derives the influence function of the
smoothed proximal target and presents the cross-fitted estimator.
Section~\ref{sec:debiased-inference-main} develops pointwise,
finite-dimensional, and oracle uniform inference for the dose-response
curve and its first derivative, covariance estimation and implementable
confidence procedures, and parametric-null testing.
The supplementary material contains the numerical experiments for curve
and derivative inference in Section~\ref{sec:numerical-experiments},
including the parametric-null experiment in
Section~\ref{sec:simulation-parametric-null-main}.
Section~\ref{sec:realdata} applies the method to the legalized abortion and crime study.  
The supplementary material also provides additional assumptions and theoretical results omitted from the main text due to space constraints, including polynomial-null examples, technical proofs and remainder analyses, bandwidth-selection details.

\section{Model setup and the existence of the pseudo-outcome}\label{sec:Prox_conti}

\subsection{Setup}
We consider a real-valued outcome $Y$, a continuously distributed scalar
treatment $A$, and possibly multivariate observed variables $X$, $Z$, and
$W$.  The observed vectors may contain discrete and continuous components.
The unmeasured confounders $U$ introduced below may be multivariate, and may
also contain discrete and continuous components.  For a treatment level $a$,
$Y(a)$ is the real-valued potential outcome under an intervention setting
$A=a$.  The target curve is the population mean of $Y(a)$ as a function of
$a$; at a fixed $a$, it is the average outcome under assignment of the
population to treatment level $a$.

The variable $X$ contains baseline covariates, $Z$ is a treatment-inducing confounding proxy and $W$
is an outcome-inducing confounding proxy.  The treatment bridge is a function of $Z$, whereas the outcome bridge is a function of $W$; see \citet{miao2018identifying,cui2024semiparametric,tchetgen2024introduction} for further discussion of the role of the proxy variables.
.

Let $O=(Y,A,Z,W,X)$, where $A$ has support $\mathcal A$ and marginal Lebesgue density $f_A$. Conditional densities of $A$ are denoted by $f(a\mid\cdot)$.
The observed data consist of independent and identically distributed copies
of $O$.

\subsection{Population pseudo-outcome}

\begin{assumptionp}{A}\label{assm:A}
\par\noindent
The following conditions hold.
\begin{enumerate}[label=\textup{(A\arabic*)},ref=A\arabic*,leftmargin=*,itemsep=0.35em]
\item\label{assm:consistency} $Y=Y(A)$ almost surely and $\mathbb E(Y^2)<\infty$.
\item\label{assm:basic.proxy} There exists an unobserved confounders $U$ such that $Y \indep \,Z \mid (U,X,A)$, and $W \indep \,(Z,A) \mid (U,X)$.
and, for Lebesgue-almost every $a\in\mathcal A$, $Y(a)\indep \,A \mid (U,X)$, and 
$0<c_1\leq f(a\mid U,X)<\infty$ a.s..
Moreover, $\sup_{a\in\mathcal A}f_A(a)<\infty$.
\item\label{assm:completeness} For almost every $(a,x)$ under the joint law of $(A,X)$ and every square-integrable function $g$,
$\mathbb{E}\!\left[g(U)\mid W, A=a, X=x\right] = 0 \ \text{a.s.}
\Longrightarrow
 g(U)=0 \ \text{a.s.}$.
For almost every $(a,x)$ under the joint law of $(A,X)$ and every square-integrable function $g$,
$\mathbb{E}\!\left[g(U)\mid Z, A=a, X=x\right] = 0 \ \text{a.s.}
\Longrightarrow
 g(U)=0 \ \text{a.s.}$.
\end{enumerate}
\end{assumptionp}

\begin{assumptionp}{B}\label{assm:B}
\par\noindent
There exist measurable functions $h$ and $q$ satisfying
$\mathbb E\{h^2(W,A,X)+q^2(Z,A,X)\}
+
\int\mathbb E\{h^2(W,a,X)\}\,dF_A(a)
<\infty$,
and the following equations.
\begin{enumerate}[label=\textup{(B\arabic*)},ref=B\arabic*,leftmargin=*,itemsep=0.35em]
\item\label{assm:outcome-bridge}
\begin{equation}
\mathbb E\{Y\mid Z,A,X\}=\int h(w,A,X)\,dF(w\mid Z,A,X). \label{eq:Bh}
\end{equation}
\item\label{assm:treatment-bridge} For $F_A$-almost every $a$,
\begin{equation}
\mathbb{E}\{q(Z,a,X)\mid W,A=a,X\}
=
\frac{f_A(a)}{f(a\mid W,X)}. \label{eq:Bq}
\end{equation}
\end{enumerate}
\end{assumptionp}

Assumption~\ref{assm:consistency} states consistency and the second-moment
condition.  The two conditional-independence restrictions in
Assumption~\ref{assm:basic.proxy} specify the proxy roles.  The restriction
$Y\indep Z\mid(U,X,A)$ is the exclusion restriction for the
treatment-inducing confounding proxy, and
$W\indep(Z,A)\mid(U,X)$ is the exclusion restriction for the
outcome-inducing confounding proxy.  These restrictions do not require either
proxy to equal $U$.  Assumption~\ref{assm:basic.proxy} imposes conditional
exchangeability given $(U,X)$, rather than
$Y(a)\indep A\mid(X,Z,W)$, and positivity of the conditional treatment density
within $(U,X)$ strata.  It also bounds the marginal treatment density.

Each implication in Assumption~\ref{assm:completeness} states that a
conditional-expectation operator on square-integrable functions of $U$ is
injective.  Completeness given $(Z,A=a,X=x)$ is used for the outcome bridge,
whereas completeness given $(W,A=a,X=x)$ is the corresponding condition for
the treatment bridge.  In a finite categorical model, each condition is a
full-column-rank restriction on the associated conditional probability matrix
and requires the relevant proxy to have at least as many categories as $U$
\citep{miao2018identifying,cui2024semiparametric,tchetgen2024introduction}.

Assumption~\ref{assm:B} imposes the existence and integrability of real-valued
outcome and treatment bridge functions.  Equation~\eqref{eq:Bh} represents the
observed conditional mean of $Y$ through the outcome bridge $h(W,A,X)$.
Equation~\eqref{eq:Bq} makes the conditional mean of $q(Z,a,X)$, rather than
$q(Z,a,X)$ itself, equal the standardized propensity function \citep{Takatsu2025debiased}.  
This is the standardized form of the continuous-treatment bridge introduced in Remark~3 of
\citet{cui2024semiparametric}.  
Both bridge equations are Fredholm integral equations of the first kind \citep{cui2024semiparametric,carrasco2007linear,kress2014linear,miao2018identifying}.  
Sufficient conditions for the existence
of square-integrable solutions are given in
Appendix~\ref{sec:existence-bridges}.

Assumptions~\ref{assm:A} and~\ref{assm:B} imply, for $F_A$-almost every $a$, $\theta(a):=\EE\{Y(a)\}=\mathbb E\{h(W,a,X)\}$.
Its proof is given in Appendix~\ref{sec: outcome bridge ID proof}.

We next define a population pseudo-outcome by adding a treatment-bridge-weighted
outcome-bridge residual to the marginal outcome-bridge mean.
For measurable functions $(\tilde q,\tilde h)$ satisfying $\int \mathbb E\{|\tilde h(W,a,X)|\}\,dF_A(a)
+
\mathbb E\bigl[|\tilde q(Z,A,X)\{Y-\tilde h(W,A,X)\}|\bigr]
<\infty,$
define $S_{\tilde h}(a)
:=
\mathbb{E}\{\tilde h(W,a,X)\}$,
and
\begin{equation}\label{def:proximal_PO}
\xi_{\mathrm{prox}}(O;\tilde q,\tilde h)
:=
\tilde q(Z,A,X)\{Y-\tilde h(W,A,X)\}
+
S_{\tilde h}(A).
\end{equation}

In Appendix~\ref{sec: PoP DR proximal proof}, we show that the pseudo-outcome has conditional mean $\theta(a)$ when
either bridge function is correctly specified.
In detail, if $\bar h=h$,
\eqref{eq:Bh} gives
$\mathbb E[\bar q(Z,a,X)\{Y-h(W,a,X)\}\mid A=a]=0$, and the identification gives $S_h(a)=\theta(a)$.  If $\bar q=q$,
\eqref{eq:Bq} converts the conditional expectation of
$h(W,a,X)-\bar h(W,a,X)$ given $A=a$ to its marginal expectation, which
cancels $S_{\bar h}(a)-S_h(a)$.  
Hence the discrepancy does not have terms involving
either bridge difference alone, and the conditional mean of the pseudo-outcome
equals $\theta(a)$ when either bridge is correct.  The result is the
continuous-treatment conditional-regression counterpart of the proximal
doubly robust identity in Theorem~3.1(3) of
\citet{cui2024semiparametric}; see also \citet{wu2024doubly}. The causal dose-response curve is therefore the conditional regression of the population
pseudo-outcome on $A$, which is the input to the local-polynomial construction
in the following section \citep{kennedy2017non,Takatsu2025debiased}.

\begin{remark}
\label{rem:forster-warmuth}
While finalizing this manuscript, we became aware of
\citet{yang2023forster}, who develop a general construction of
pseudo-outcomes and combine it with Forster--Warmuth regression
to obtain oracle risk guarantees.
Their Section~S.2.6 derives the same proximal
dose-response pseudo-outcome as in \eqref{def:proximal_PO},
treating the joint marginal distribution of the outcome-inducing proxy and observed covariates as known in the derivation.
As we show in Theorem~\ref{thm:DR_for_PO} of Appendix~\ref{sec: PoP DR proximal proof}, the same population identity holds when this distribution is unknown, under the stated conditions.

Here, we give a formal statement under explicit conditions, provide sufficient conditions for bridge
existence in Appendix~\ref{sec:existence-bridges}, and establish the exact conditional bias identity and double robustness in Theorem~\ref{thm:DR_for_PO}.
Building on this pseudo-outcome construction, our main contribution is to broaden the scope of local and global
inference for counterfactual continuous-treatment pseudo-outcome regression by developing cross-fitted (without imposing entropy conditions on the nuisance function classes), bias-corrected smoothing inference for both the dose-response curve and its higher order derivatives, together with testing against general finite-dimensional parametric
dose-response models.
These inferential guarantees are substantially different in scope and aim from the oracle risk guarantees of
\citet{yang2023forster}.
\end{remark}

\section{Estimation}
\label{sec:estimation}

\subsection{Influence function of the smoothed proximal target}

For a continuous-treatment, point evaluation of a dose-response curve is not pathwise differentiable relative to a nonparametric observed-data
model \citep{kennedy2017non}.  We therefore consider a smoothed functional
at a fixed $a_0\in\mathcal A$.  The bandwidth $h>0$ is used for the population
local-linear projection, defining the smoothed level at $a_0$, whereas $b>0$
is used for the population local-quadratic projection, defining the
second-derivative bias correction \citep{Takatsu2025debiased}.

Let $P_0$ denote the true observed-data law.  For a generic observed-data law
$P$, write $E_P$ for expectation under $P$, let $F_P$ and $P_{WX}$ denote the
marginal laws of $A$ and $(W,X)$, respectively, and write a subscript $0$ for
evaluation at $P_0$.  Let $K$ be a bounded symmetric kernel density supported
on $[-1,1]$.

\begin{assumptionp}{C}\label{assm:C}
\par\noindent
Fix $a_0\in\mathcal A$ and $h,b>0$. Let $\mathcal M^\dagger$ denote the collection of observed-data laws $P$ satisfying the following conditions.
\begin{enumerate}[label=\textup{(C\arabic*)},ref=C\arabic*,leftmargin=*,itemsep=0.35em]
\item\label{assm:smooth-model-bridges} $A$ has a marginal Lebesgue density $f_P$, and there exist unique bridge functions $h_P$ and $q_P$ satisfying $E_P\{Y-h_P(W,A,X)\mid Z,A,X\}=0$,
and, for $F_P$-almost every $a$, $E_P\{q_P(Z,a,X)\mid W,A=a,X\}=f_P(a)/f_P(a\mid W,X)$.
Moreover, $h_P\in L_2(P_{WX}\otimes F_P)$, 
$E_P\{Y^2+h_P^2(W,A,X)+q_P^2(Z,A,X)\}<\infty$,
and $E_P\!\left[
I\{|A-a_0|\le \max(h,b)\}
q_P(Z,A,X)^2\{Y-h_P(W,A,X)\}^2
\right]<\infty$.
\item\label{assm:smooth-model-positivity} There exists $\kappa_P>0$ such that $f_P(a\mid W,X)\ge \kappa_P$,
for Lebesgue-almost every $a$ satisfying $|a-a_0|\le \max(h,b)$, $P_{WX}$-almost surely.
\end{enumerate}
\end{assumptionp}

Condition~\ref{assm:smooth-model-bridges} imposes standard bridge conditions in the proximal inference literature with an adaptation to continuous treatment setting \citep{cui2024semiparametric}.
The outcome bridge defines the smoothed target, and the treatment bridge enters
its influence function representation.
Condition~\ref{assm:smooth-model-positivity} applies on the union of the
supports of the two kernel weights.

For \(P\in \mathcal M^\dagger\), define $\theta_P(a):=\mathbb E_P\{h_P(W,a,X)\}$, and
$\xi_P(O):=
q_P(Z,A,X)\{Y-h_P(W,A,X)\}+\theta_P(A)$.
The function $\theta_P$ is an observed-data functional on
$\mathcal M^\dagger$.  At $P_0$, Assumptions~\ref{assm:A} and~\ref{assm:B}
give $\theta_0(a)=E_0\{Y(a)\}$ for $F_0$-almost every $a$.  The outcome-bridge
equation and iterated expectation give
$\mathbb E_P\{\xi_P(O)\mid A=a\}=\theta_P(a)$,
for \(F_P\)-almost every \(a\).

Consequently, each population local-polynomial moment involving
$\theta_P(A)$ is unchanged when $\theta_P(A)$ is replaced by $\xi_P(O)$.
Thus the bridge equations reduce the construction of the smoothed proximal
target to population local-polynomial regression of $\xi_P(O)$ on $A$. 
In particular, this construction
is a proximal extension of the population local-polynomial target and
influence function analysis in Sections~2.3 and 2.4 of
\citet{Takatsu2025debiased}.

Define \(K_{h,a_0}(a):=h^{-1}K\{(a-a_0)/h\}\) and
\(K_{b,a_0}(a):=b^{-1}K\{(a-a_0)/b\}\).  Let
\(w_{h,a_0,1}(a):=(1,(a-a_0)/h)^\top\) and
\(w_{b,a_0,2}(a):=(1,(a-a_0)/b,((a-a_0)/b)^2)^\top\), and define
\(\widetilde w_{h,a_0,1}(a):=w_{h,a_0,1}(a)((a-a_0)/h)^2\).
Then $w_{h,a_0,1}(a),\widetilde w_{h,a_0,1}(a)\in\mathbb R^2$ and
$w_{b,a_0,2}(a)\in\mathbb R^3$.  Let
\(e_1:=(1,0)^\top\in\mathbb R^2\) and
\(e_3:=(0,0,1)^\top\in\mathbb R^3\).  For
\(P\in\mathcal M^\dagger\), define
\(D_{P,h,a_0,1}:=P\{w_{h,a_0,1}(A)w_{h,a_0,1}(A)^\top
K_{h,a_0}(A)\}\).  Likewise, define
\(D_{P,b,a_0,2}:=P\{w_{b,a_0,2}(A)w_{b,a_0,2}(A)^\top
K_{b,a_0}(A)\}\).

Condition~\ref{assm:smooth-model-positivity} implies that $f_P$ is bounded
below by $\kappa_P$ on the supports of $K_{h,a_0}$ and $K_{b,a_0}$.  Since a
nonzero polynomial of degree at most two has only finitely many zeros and $K$
is a density, for every nonzero $v\in\mathbb R^2$, $v^\top D_{P,h,a_0,1}v
=
\int \{v^\top(1,u)^\top\}^2K(u)f_P(a_0+hu)\,du
>0$,
and, for every nonzero $v\in\mathbb R^3$,
$v^\top D_{P,b,a_0,2}v
=
\int \{v^\top(1,u,u^2)^\top\}^2K(u)f_P(a_0+bu)\,du
>0$.
Thus $D_{P,h,a_0,1}\in\mathbb R^{2\times2}$ and
$D_{P,b,a_0,2}\in\mathbb R^{3\times3}$ are nonsingular.  
Define
\(M_{P,h,a_0}:=P\{w_{h,a_0,1}(A)K_{h,a_0}(A)\theta_P(A)\}\in\mathbb R^2\), and 
\(M_{P,b,a_0}:=P\{w_{b,a_0,2}(A)K_{b,a_0}(A)\theta_P(A)\}\in\mathbb R^3\).

We further define \(r_{P,h,a_0}(a):=e_1^\top D_{P,h,a_0,1}^{-1}w_{h,a_0,1}(a)K_{h,a_0}(a)\), and $\Gamma_{P,h,b,a_0}(a)
:=
r_{P,h,a_0}(a)
-\frac{1}{2}h^2 c_{P,h,a_0,2}s_{P,b,a_0}(a)$,
where $s_{P,b,a_0}(a):=2b^{-2}e_3^\top D_{P,b,a_0,2}^{-1}w_{b,a_0,2}(a)K_{b,a_0}(a)$, and $c_{P,h,a_0,2}:=e_1^\top D_{P,h,a_0,1}^{-1}
P\{\widetilde w_{h,a_0,1}(A)K_{h,a_0}(A)\}$.
The smoothed and debiased target parameter is
\[
\theta_{P,h,b}(a_0)
:=
P\{\Gamma_{P,h,b,a_0}(A)\theta_P(A)\}
=
\theta_{P,h}^{\mathrm{LL}}(a_0)
-\frac{1}{2}h^2 c_{P,h,a_0,2}\theta_{P,b}''(a_0),
\]
where
\(\theta_{P,h}^{\mathrm{LL}}(a_0):=e_1^\top D_{P,h,a_0,1}^{-1}M_{P,h,a_0}\),
\(\theta_{P,b}''(a_0):=2b^{-2}e_3^\top D_{P,b,a_0,2}^{-1}M_{P,b,a_0}\).
The vector $D_{P,h,a_0,1}^{-1}M_{P,h,a_0}$ contains the coefficients
of the population kernel-weighted local-linear projection of $\theta_P(A)$,
and $\theta_{P,h}^{\mathrm{LL}}(a_0)$ is its intercept.  Similarly,
$D_{P,b,a_0,2}^{-1}M_{P,b,a_0}$ contains the coefficients of the population
local-quadratic projection, and $\theta_{P,b}''(a_0)$ is twice its quadratic
coefficient rescaled by $b^{-2}$.  The coefficient $c_{P,h,a_0,2}$ is the
image of the quadratic monomial under the local-linear intercept functional.
The definitions also give \(P\{\Gamma_{P,h,b,a_0}(A)\}=1\),
\(P\{\Gamma_{P,h,b,a_0}(A)(A-a_0)\}=0\), and
\(P\{\Gamma_{P,h,b,a_0}(A)(A-a_0)^2\}=0\).
Its
support is contained in $\{a:|a-a_0|\le\max(h,b)\}$, and
$\theta_{P,h,b}(a_0)=\theta_P(a_0)$ whenever $\theta_P$ is a polynomial of
degree at most two on this set.

We further define $\gamma_{P,h,a_0}(a)
:=
e_1^\top D_{P,h,a_0,1}^{-1}w_{h,a_0,1}(a)K_{h,a_0}(a)
\,w_{h,a_0,1}(a)^\top D_{P,h,a_0,1}^{-1}M_{P,h,a_0}$, 
$\gamma_{P,b,a_0}''(a)
:=
2b^{-2}e_3^\top D_{P,b,a_0,2}^{-1}w_{b,a_0,2}(a)K_{b,a_0}(a)
\,w_{b,a_0,2}(a)^\top D_{P,b,a_0,2}^{-1}M_{P,b,a_0}$,
and
$\gamma_{P,h,a_0}^c(a):=
e_1^\top D_{P,h,a_0,1}^{-1}
\Bigl[
\widetilde w_{h,a_0,1}(a)-
w_{h,a_0,1}(a)w_{h,a_0,1}(a)^\top
D_{P,h,a_0,1}^{-1}
P\{\widetilde w_{h,a_0,1}(A)K_{h,a_0}(A)\}
\Bigr]K_{h,a_0}(a)$.
The following theorem establishes the efficient influence function of the smoothed proximal estimand.
\begin{theorem}\label{thm:EIF of the smoothed proximal target}
Define
$\gamma_{P,h,b,a_0}(a)
:=
\gamma_{P,h,a_0}(a)
-
\frac{1}{2}h^2 c_{P,h,a_0,2}\gamma_{P,b,a_0}''(a)
+
\frac{1}{2}h^2 \theta_{P,b}''(a_0)\gamma_{P,h,a_0}^c(a)$.
Under Assumptions~\ref{assm:smooth-model-bridges} and
\ref{assm:smooth-model-positivity} and the regularity condition in
Appendix~\ref{app:pathwise-bridge-regularity}, the map
\(P\mapsto \theta_{P,h,b}(a_0)\) is pathwise differentiable at \(P_0\), and
$\varphi_{0,h,b,a_0}(O)
=
\Gamma_{0,h,b,a_0}(A)\xi_0(O)-\gamma_{0,h,b,a_0}(A)+
\int \Gamma_{0,h,b,a_0}(\bar a)
\{h_0(W,\bar a,X)-\theta_0(\bar a)\}\,dF_0(\bar a)$
is an influence function at $P_0$.
\end{theorem}
\noindent Proof is given in Appendix~\ref{app:pathwise-bridge-regularity}.
The product $q_0(Z,A,X)\{Y-h_0(W,A,X)\}$ is the weighted proximal residual, and its coefficient $\Gamma_{0,h,b,a_0}(A)$ represents the contribution of the
outcome bridge to the pathwise derivative.  The term
$\Gamma_{0,h,b,a_0}(A)\theta_0(A)$ represents variation of the marginal law
of $A$.  The subtraction of $\gamma_{0,h,b,a_0}(A)$ accounts for the
dependence of the local-polynomial moment matrices and $c_{P,h,a_0,2}$ on this
law.  The integral term accounts for variation of the marginal law of $(W,X)$ in
$\theta_P(a)=E_P\{h_P(W,a,X)\}$.
The definitions above yield $P_0\{\gamma_{0,h,a_0}(A)\}=\theta_{0,h}^{\mathrm{LL}}(a_0)$, $P_0\{\gamma_{0,b,a_0}''(A)\}=\theta_{0,b}''(a_0)$, and 
$P_0\{\gamma_{0,h,a_0}^c(A)\}=0$.
Consequently, one has
$P_0\{\gamma_{0,h,b,a_0}(A)\}
=
P_0\{\Gamma_{0,h,b,a_0}(A)\theta_0(A)\}
=
\theta_{0,h,b}(a_0)$.
The integral term has mean zero by the definition of $\theta_0$.  Hence
$P_0\{\varphi_{0,h,b,a_0}(O)\}=0$.  For fixed $h$ and $b$,
Theorem~\ref{thm:EIF of the smoothed proximal target} supplies the pathwise
first-order correction used in the cross-fitted estimator in the following section.

\subsection{Cross-fitted estimator}\label{sec:cross-fitted-estimator}

Let $P_n$ denote the empirical distribution of $O_1,\ldots,O_n$.  Fix
$h_n,b_n>0$ and an integer $K\ge2$, and let $I_1,\ldots,I_K$ be a partition of
$\{1,\ldots,n\}$ into nonempty folds.  Write $n_k:=|I_k|$ and
$P_{n,k}f:=n_k^{-1}\sum_{i\in I_k}f(O_i)$.  Let $F_{n,k}$ and
$Q_{n,k}^{WX}$ denote the empirical distributions of
$\{A_i:i\in I_k\}$ and $\{(W_i,X_i):i\in I_k\}$, respectively.  For every
measurable function $g$ for which the following expression is finite, define
$(Q_{n,k}^{WX}\times F_{n,k})g
:=(1/n_k^2)\sum_{i\in I_k}\sum_{j\in I_k}g(W_i,X_i,A_j)$.

For each fold $k$, let $\hat q_n^{(-k)}$ and $\hat h_n^{(-k)}$ be estimators
constructed using only $\{O_i:i\notin I_k\}$.  Conditional on these training
observations, the fitted bridges are fixed when their contributions on $I_k$
are evaluated.  
In contrast to the non-cross-fitted construction of the canonical dose-response curve estimators in \citet{Takatsu2025debiased}, due to the benefit of using cross-fitting, the pointwise, finite-dimensional, and uniform results below impose foldwise conditional
convergence conditions without the need for Donsker or entropy conditions on the classes containing the fitted bridge functions, thereby improving theoretical guarantees in high-complexity settings \citep{kennedy2023towards,Chernozhukov2018,vanderlaan2011crossvalidated,belloni2018uniformly}.  
The empirical local-polynomial
quantities below are computed from $P_n$ and are common to all folds.

For each $a_0$, use the kernel weights and polynomial vectors defined in the
preceding subsection with $h=h_n$ and $b=b_n$. 
Define $D_{n,h_n,a_0,1}:=
P_n\{w_{h_n,a_0,1}(A)w_{h_n,a_0,1}(A)^\top K_{h_n,a_0}(A)\}$, and 
$D_{n,b_n,a_0,2}:=
P_n\{w_{b_n,a_0,2}(A)w_{b_n,a_0,2}(A)^\top K_{b_n,a_0}(A)\}$.
Whenever $D_{n,h_n,a_0,1}$ is nonsingular, define $c_{n,h_n,a_0,2}
:=
e_1^\top D_{n,h_n,a_0,1}^{-1}
P_n\{\widetilde w_{h_n,a_0,1}(A)K_{h_n,a_0}(A)\}$,
and otherwise set $c_{n,h_n,a_0,2}:=0$.  Whenever both matrices are
nonsingular, define $\Gamma_{n,h_n,b_n,a_0}(a)
:=
e_1^\top D_{n,h_n,a_0,1}^{-1}
w_{h_n,a_0,1}(a)K_{h_n,a_0}(a)
-
\frac{1}{2}h_n^2c_{n,h_n,a_0,2}\,2b_n^{-2}
e_3^\top D_{n,b_n,a_0,2}^{-1}
w_{b_n,a_0,2}(a)K_{b_n,a_0}(a)$
If either matrix is singular, set $\Gamma_{n,h_n,b_n,a_0}(a):=0$ for every
$a$.

For each fold $k$, define $\hat\xi_{n,1}^{(-k)}(O):=\hat q_n^{(-k)}(Z,A,X)
\{Y-\hat h_n^{(-k)}(W,A,X)\}$, $\hat m_{n,k}^{(-k)}(a):=
\int \hat h_n^{(-k)}(w,a,x)\,dQ_{n,k}^{WX}(w,x)
=
\frac{1}{n_k}\sum_{i\in I_k}\hat h_n^{(-k)}(W_i,a,X_i)$, and $\hat\xi_n^{(-k)}(O)
:=
\hat\xi_{n,1}^{(-k)}(O)+\hat m_{n,k}^{(-k)}(A)$.
The cross-fitted estimator has the equivalent representations
\begin{equation}\label{eq:DB-cf-representations}
\begin{aligned}
\hat\theta_n^{DB,cf}(a_0)
:={}&
\sum_{k=1}^K\frac{n_k}{n}
\Bigl[
P_{n,k}\{\Gamma_{n,h_n,b_n,a_0}(A)\hat\xi_{n,1}^{(-k)}(O)\}\\
&\hspace{6.2em}+
(Q_{n,k}^{WX}\times F_{n,k})
\{\Gamma_{n,h_n,b_n,a_0}(A)\hat h_n^{(-k)}(W,A,X)\}
\Bigr]\\
={}&
\sum_{k=1}^K\frac{n_k}{n}
P_{n,k}\{\Gamma_{n,h_n,b_n,a_0}(A)\hat\xi_n^{(-k)}(O)\}.
\end{aligned}
\end{equation}
Expanding the empirical measures gives
\begin{equation}\label{def:DB_cf_estimation}
\begin{aligned}
\hat\theta_n^{DB,cf}(a_0)
={}&
\frac{1}{n}\sum_{k=1}^K\sum_{i\in I_k}
\Gamma_{n,h_n,b_n,a_0}(A_i)
\hat q_n^{(-k)}(Z_i,A_i,X_i)
\{Y_i-\hat h_n^{(-k)}(W_i,A_i,X_i)\}\\
&+
\sum_{k=1}^K\frac{1}{n n_k}
\sum_{i\in I_k}\sum_{j\in I_k}
\Gamma_{n,h_n,b_n,a_0}(A_j)
\hat h_n^{(-k)}(W_i,A_j,X_i).
\end{aligned}
\end{equation}
Appendix~\ref{sec:DB-cf-equivalence} verifies the equality of these
representations.
The same folds and fitted bridges are used for every $a_0$.

\section{Debiased inference for proximal dose-response}
\label{sec:debiased-inference-main}

\subsection{Pointwise and finite-dimensional inference}

In contrast to root-$n$ inference for the pathwise differentiable
average treatment effect in Theorem~3.2 of \citet{cui2024semiparametric},
point evaluation of a continuously indexed dose-response curve is not pathwise differentiable in a nonparametric model, regular estimation at the parametric rate is impossible \citep{bickel1982adaptive,pfanzagl2012contributions,kennedy2017non}.
Thus, our smoothed target functionals indexed by a smoothing parameter, each of which is pathwise differentiable and converges to the original parameter as the smoothing vanishes, providing a basis for valid semiparametric inference \citep{Takatsu2025debiased,ColangeloLee2026}.

Fix an exposure value $a_0$ in the interior of the support of $A$, and let
$(q_0,h_0)$ be a fixed pair satisfying Assumption~\ref{assm:B} under $P_0$.
For $\delta>0$, define $B_\delta(a_0):=\{a\in\mathcal A:|a-a_0|\le\delta\}$.
For deterministic measurable $(q_\infty,h_\infty)$ representing the limiting nuisance bridge functions, define $m_\infty(a)
:=E_0\{h_\infty(W,a,X)\}$,
$\xi_\infty(O):=q_\infty(Z,A,X)\{Y-h_\infty(W,A,X)\}+m_\infty(A)$,
and $\sigma_0^2(a):=
E_0\left([\xi_\infty(O)-\theta_0(A)]^2\mid A=a\right)$.

\begin{assumptionp}{D}\label{assm:D}
\par\noindent
The kernel and bandwidth conditions are as follows.
\begin{enumerate}[label=\textup{(D\arabic*)},ref=D\arabic*,leftmargin=*,itemsep=0.35em]
\item\label{assm:kernel} The kernel $K$ is a mean-zero, symmetric, nonnegative, and Lipschitz continuous density function with support contained in $[-1,1]$. Additionally, $K$ belongs to the linear span of the functions whose subgraph can be represented as a finite number of Boolean operations among sets of the form $\{(s,u)\in\mathbb R\times\mathbb R:\ p(s,u)\le \varphi(u)\}$, where $p$ is a polynomial and $\varphi$ is a real function.
\item\label{assm:bandwidths} As $n\to\infty$, the bandwidths $h=h_n$ and $b=b_n$ satisfy $h_n\to0$, $nh_n\to\infty$, $b_n\to0$, and $\tau_n:=h_n/b_n\to\tau\in[0,\infty)$.
\end{enumerate}
\end{assumptionp}

\begin{assumptionp}{E}\label{assm:E}
\par\noindent
The pointwise smoothness and moment conditions are as follows.
\begin{enumerate}[label=\textup{(E\arabic*)},ref=E\arabic*,leftmargin=*,itemsep=0.35em]
\item\label{assm:regul1} There exists $\delta_1>0$ such that:
\begin{enumerate}[label=\textup{(\alph*)},leftmargin=2.2em,itemsep=0.25em]
\item $\theta_0$ is twice continuously differentiable on $B_{\delta_1}(a_0)$;
\item $f_A$ is positive and Lipschitz continuous on $B_{\delta_1}(a_0)$;
\item there exist $\delta_2>0$ and $C_6>0$ such that $\mathbb E_0[|Y|^{2+\delta_2}\mid A=a]<C_6$ for all $a\in B_{\delta_1}(a_0)$ and $\mathbb E_0[|Y|^4]<\infty$;
\item $a\mapsto\sigma_0^2(a)$ is bounded and continuous on $B_{\delta_1}(a_0)$.
\end{enumerate}
\end{enumerate}
\end{assumptionp}

For $\mathcal I\subseteq\mathcal A$ and measurable $\mathcal S\subseteq\mathcal I\times\mathcal Z\times\mathcal W\times\mathcal X$, define $d_{q,k}(\bar q_1,\bar q_2;\mathcal I,\mathcal S)$ by
\[
\begin{aligned}
\sup_{a\in\mathcal I}
\Biggl[
E_0\Bigl\{
I_{\mathcal S}(a,Z,W,X)
[\bar q_1(Z,a,X)-\bar q_2(Z,a,X)]^2
\,\Bigm|\, A=a,\trainingsigma_k
\Bigr\}
\Biggr]^{1/2}
\end{aligned}
\]
and $d_{h,k}(\bar h_1,\bar h_2;\mathcal I,\mathcal S)$ by
\[
\begin{aligned}
\sup_{a\in\mathcal I}
\Biggl[
E_0\Bigl\{
I_{\mathcal S}(a,Z,W,X)
[\bar h_1(W,a,X)-\bar h_2(W,a,X)]^2
\,\Bigm|\, A=a,\trainingsigma_k
\Bigr\}
\Biggr]^{1/2}.
\end{aligned}
\]
For pointwise inference, take $\mathcal I=B_{\delta_1}(a_0)$, where
$\delta_1$ is as in Assumption~\ref{assm:E}.

\begin{assumptionp}{F}\label{assm:F}
\par\noindent
The pointwise cross-fitting and conditional nuisance conditions are as follows.
\begin{enumerate}[label=\textup{(F\arabic*)},ref=F\arabic*,leftmargin=*,itemsep=0.35em]
\item\label{assm:folds} The number of folds is the fixed integer $K$, and the fold partition is deterministic or independent of the observations. For some constant $\underline\pi>0$, $\min_{1\le k\le K}(n_k/n)\ge \underline\pi$,
for all sufficiently large $n$. For each $k$, let
$\trainingsigma_k:=\sigma\{I_1,\ldots,I_K,\{O_i:i\notin I_k\}\}$
denote the training sigma-field. Then $\hat q_n^{(-k)}$ and $\hat h_n^{(-k)}$ are $\trainingsigma_k$-measurable, whereas $\{O_i:i\in I_k\}$ is conditionally IID from $P_0$ given $\trainingsigma_k$. Whenever $P_0$ or $E_0$ is applied to a $\trainingsigma_k$-measurable random function, integration is with respect to an independent generic observation $O\sim P_0$, conditionally on $\trainingsigma_k$.

\item\label{assm:nuisance1} There exist finite constants $C_q$ and $C_h$ such that, with probability tending to one, $\max_{1\le k\le K}\sup_{\substack{a\in B_{\delta_1}(a_0)\\ z,x}}|\hat q_n^{(-k)}(z,a,x)|\le C_q$, and 
$\max_{1\le k\le K}\sup_{\substack{a\in B_{\delta_1}(a_0)\\ w,x}}|\hat h_n^{(-k)}(w,a,x)|\le C_h$.

\item\label{assm:DR_and_rates} The functions $q_\infty$ and $h_\infty$ are uniformly bounded on $B_{\delta_1}(a_0)$: for finite constants $C_q$ and $C_h$, $\sup_{\substack{a\in B_{\delta_1}(a_0)\\ z,x}}|q_\infty(z,a,x)|\le C_q$ and $\sup_{\substack{a\in B_{\delta_1}(a_0)\\ w,x}}|h_\infty(w,a,x)|\le C_h$.
In addition, there exists a measurable partition
$\mathcal S_1\cup\mathcal S_2\cup\mathcal S_3
=
B_{\delta_1}(a_0)\times\mathcal Z\times\mathcal W\times\mathcal X$
up to $P_0(\cdot\mid A=a)$-null sets for $f_A$-almost every $a\in B_{\delta_1}(a_0)$. Moreover, for $f_A$-almost every $a\in B_{\delta_1}(a_0)$,
\[
\begin{aligned}
&I_{\mathcal S_1\cup\mathcal S_3}(a,Z,W,X)
\{q_\infty(Z,a,X)-q_0(Z,a,X)\}=0,\\
&I_{\mathcal S_2\cup\mathcal S_3}(a,Z,W,X)
\{h_\infty(W,a,X)-h_0(W,a,X)\}=0
\end{aligned}
\]
$P_0(\cdot\mid A=a)$-almost surely, and
\[
\sup_{a\in B_{\delta_1}(a_0)}
E_0\!\left[
\begin{aligned}
&\{q_\infty(Z,a,X)-q_0(Z,a,X)\}^2\\
&\quad+
\{h_\infty(W,a,X)-h_0(W,a,X)\}^2
\end{aligned}
\,\middle|\, A=a
\right]
<\infty.
\]
The pointwise doubly robust nuisance rates are
\[
\begin{aligned}
&\max_{1\le k\le K}
d_{q,k}(\hat q_n^{(-k)},q_\infty;B_{\delta_1}(a_0),\mathcal S_1)
=o_p\{(nh_n)^{-1/2}\},\\
&\max_{1\le k\le K}
d_{h,k}(\hat h_n^{(-k)},h_\infty;B_{\delta_1}(a_0),\mathcal S_1)
=o_p(1),\\
&\max_{1\le k\le K}
d_{h,k}(\hat h_n^{(-k)},h_\infty;B_{\delta_1}(a_0),\mathcal S_2)
=o_p\{(nh_n)^{-1/2}\},\\
&\max_{1\le k\le K}
d_{q,k}(\hat q_n^{(-k)},q_\infty;B_{\delta_1}(a_0),\mathcal S_2)
=o_p(1),\\
&\max_{1\le k\le K}
\Bigl\{
d_{q,k}(\hat q_n^{(-k)},q_\infty;B_{\delta_1}(a_0),\mathcal S_3)\\[-0.2em]
&\hspace{10em}\times
d_{h,k}(\hat h_n^{(-k)},h_\infty;B_{\delta_1}(a_0),\mathcal S_3)
\Bigr\}
=o_p\{(nh_n)^{-1/2}\}.
\end{aligned}
\]

\item\label{assm:pointwise-pseudooutcome-rate} The nuisance estimators satisfy
\[
\begin{aligned}
&\sup_{\substack{1\le k\le K\\a\in B_{\delta_1}(a_0)}}
\left[
E_0\!\left\{
(\hat q_n^{(-k)}(Z,a,X)-q_\infty(Z,a,X))^2
\,\middle|\, A=a,\trainingsigma_k
\right\}
\right]^{1/2}=o_p(1),\\
&\sup_{\substack{1\le k\le K\\a\in B_{\delta_1}(a_0)}}
\left[
E_0\!\left\{
(\hat h_n^{(-k)}(W,a,X)-h_\infty(W,a,X))^2
\,\middle|\, A=a,\trainingsigma_k
\right\}
\right]^{1/2}=o_p(1).
\end{aligned}
\]
\end{enumerate}
\end{assumptionp}

Assumption~\ref{assm:D} imposes the local-polynomial kernel and bandwidth conditions. 
The bandwidth $b_n$ localizes the
local-quadratic bias correction, and the limit of $\tau_n=h_n/b_n$ determines
whether that correction contributes to the first-order variance. 
Assumption~\ref{assm:E} is local to $a_0$. 
Positivity of $f_A$ makes the
population local-polynomial moment matrices nonsingular for all sufficiently
small bandwidths, while its Lipschitz continuity permits replacement of the
local density by $f_A(a_0)$ in the variance calculation.  Continuity of
$\sigma_0^2$ has the analogous role for the conditional variance. 
The moment conditions provide the moment bounds required for the Lyapunov condition and for controlling the remainder term in the Taylor expansion. Further details are provided by \citet{Takatsu2025debiased}.

In Assumption~\ref{assm:F}, conditional on $\trainingsigma_k$,
$\hat q_n^{(-k)}$ and $\hat h_n^{(-k)}$ are fixed relative to the observations in $I_k$.  
This assumption indicates that the nuisance-dependent components of the empirical remainder and the fold-specific empirical bridge-relevant remainder can therefore be
controlled by the moment bounds and the required $L_2$-rates.
By contrast, the condition (c) of \citet{Takatsu2025debiased} requires polynomial uniform covering-number bounds for classes containing the bridge estimators. 
The usage of cross-fitting avoids these nuisance-class entropy restrictions and the corresponding
stochastic-equicontinuity arguments over the possible nuisance fits \citep{VanderVaartWellner2023weak}.
The boundedness condition on the treatment bridge differs structurally from the condition (c)(i) of \citet{Takatsu2025debiased}.  
Their pseudo-outcome
weights the outcome-regression residual by the reciprocal standardized propensity, and their condition requires uniform bounds on both functions.  
In this proximal construction, \eqref{eq:Bq} makes the conditional mean of $q_0(Z,a,X)$ equal
$f_A(a)/f(a\mid W,X)$, the reciprocal standardized treatment density. Assumption~\ref{assm:F} therefore
imposes boundedness of the fitted treatment bridges but neither a positive
lower bound on $\hat q_n^{(-k)}$ nor a bound on its reciprocal.

Assumption~\ref{assm:DR_and_rates} is the continuous-treatment
analogue of the proximal doubly robust moment condition for the binary treatment setting in Theorem~3.1(3) of \citet{cui2024semiparametric}.  
As stated in Remark~6 of \citet{cui2024semiparametric}, this condition does not require uniqueness of the corresponding bridge solution.  
Here the partition in Assumption~\ref{assm:DR_and_rates} implies, for $f_A$-almost
every $a\in B_{\delta_1}(a_0)$,
$\{q_\infty(Z,a,X)-q_0(Z,a,X)\}
\{h_0(W,a,X)-h_\infty(W,a,X)\}=0$,
$P_0(\cdot\mid A=a)$-almost surely.  The product-bias identity in
Theorem~\ref{thm:DR_for_PO} consequently gives $E_0\{\xi_\infty(O)\mid A=a\}=\theta_0(a)$,
for $f_A$-almost every $a\in B_{\delta_1}(a_0)$.
The associated rate conditions are the proximal bridge extension of condition~(d) of \citet{Takatsu2025debiased}.  On $\mathcal S_1$,
$q_\infty=q_0$, so the treatment-bridge estimation error must be
$o_p\{(nh_n)^{-1/2}\}$, whereas the outcome-bridge estimator need only
converge to $h_\infty$.  On $\mathcal S_2$, the corresponding requirements
hold with the roles reversed.  On $\mathcal S_3$, both limiting bridges agree
with the population bridges, and only the product of their conditional
$L_2$ rates must be $o_p\{(nh_n)^{-1/2}\}$.  Hence neither limiting bridge is
required to agree with the population bridge throughout the relevant
conditional support.  

For $j\ge0$, let $c_j:=\int u^jK(u)\,du$, and set
$V_{K,\tau}:=\int\{
K(u)-\tau^3c_2\frac{(\tau u)^2-c_2}{c_4-c_2^2}K(\tau u)\}^2\,du$.
The term inside braces is the limiting combined kernel weight: its first
term is the local-linear weight, and its second term is the local-quadratic correction that removes the leading bias.  Thus, $V_{K,\tau}$ is the integral
of the squared combined weight and determines the kernel-dependent component
of the asymptotic variance.

We let $\Gamma_{0,h,b,a_0}$ and
$\gamma_{0,h,b,a_0}$ denote the population local-polynomial quantities defined in Section \ref{sec:estimation}, evaluated under $P_0$ with
$\theta_P=\theta_0$.  
Define
\begin{equation*}
\begin{aligned}
\phi_{\infty,h,b,a_0}(O)
:=
\Gamma_{0,h,b,a_0}(A)\xi_\infty(O)
-\gamma_{0,h,b,a_0}(A)+
\int\Gamma_{0,h,b,a_0}(\bar a)
\{h_\infty(W,\bar a,X)-m_\infty(\bar a)\}\,dF_0(\bar a),
\end{aligned}
\end{equation*}
and let $\phi_{\infty,a_0}:=\phi_{\infty,h_n,b_n,a_0}$.  The conditions above and some local polynomial moment identities give $P_0\phi_{\infty,a_0}=0$.

\begin{theorem}\label{thm:pointwise-proximal-inference}
Fix $a_0$ in the interior of the support of $A$.  Under Assumptions~\ref{assm:A}, \ref{assm:B}, and~\ref{assm:D}--\ref{assm:F},
$\hat\theta_n^{DB,cf}(a_0)-\theta_0(a_0)
=
P_n\phi_{\infty,a_0}+o_p\{(nh_n)^{-1/2}+h_n^2\}$.
Moreover, $(nh_n)^{1/2}P_n\phi_{\infty,a_0}
\rightsquigarrow
N(0,\frac{V_{K,\tau}\sigma_0^2(a_0)}{f_A(a_0)})$.
Consequently, if $nh_n^5=O(1)$, then
$(nh_n)^{1/2}\{\hat\theta_n^{DB,cf}(a_0)-\theta_0(a_0)\}
\rightsquigarrow
N(0,\frac{V_{K,\tau}\sigma_0^2(a_0)}{f_A(a_0)})$.
\end{theorem}
\noindent The first display in Theorem~\ref{thm:pointwise-proximal-inference} is a first-order representation.  At the
$(nh_n)^{1/2}$ scale, the localized term
$\Gamma_{0,h_n,b_n,a_0}(A)\{\xi_\infty(O)-\theta_0(A)\}$ determines the
Gaussian limit; the terms involving $\gamma_{0,h_n,b_n,a_0}$ and marginal
averaging over $(W,X)$ are of smaller order.  
The kernel $K$ and the bandwidth ratio $\tau$ determine $V_{K,\tau}$, whereas $f_A(a_0)$ and $\sigma_0^2(a_0)$ characterize the contributions of the exposure density and the proximal pseudo-outcome, respectively.

The bias correction reduces the deterministic smoothing error to $o(h_n^2)$.
Hence $nh_n^5=O(1)$ makes this error negligible at the $(nh_n)^{1/2}$ scale
and permits $h_n$ of order $n^{-1/5}$, and the stronger undersmoothing condition $nh_n^5\to0$ is no longer required \citep{Takatsu2025debiased}. 
If $\tau=0$, then
$V_{K,0}=\int K(u)^2\,du$ and the bias correction does not contribute to the
first-order variance.  If $\tau>0$, its stochastic contribution is included
in $V_{K,\tau}$.
In particular, when $b_n$ is of the same order as $h_n$, the
local-quadratic second-derivative component need not be consistent by itself; see Sections 2--3 of \citet{Takatsu2025debiased}.
Proofs for Theorem \ref{thm:pointwise-proximal-inference} and the following Theorem \ref{thm:finite-dimensional-proximal-inference} are given in Section~\ref{sec: multiple a inference proofs}.

\begin{theorem}\label{thm:finite-dimensional-proximal-inference}
Fix distinct points $a_1,\ldots,a_m$ in the interior of the support of $A$.
Suppose that Assumptions~\ref{assm:A}, \ref{assm:B}, and~\ref{assm:D}--\ref{assm:F}
hold at each $a_j$, $j=1,\ldots,m$, for the same pair
$(q_\infty,h_\infty)$ and that $nh_n^5=O(1)$.  Then $(nh_n)^{1/2}(\hat\theta_n^{DB,cf}(a_1)-\theta_0(a_1),\ldots, \hat\theta_n^{DB,cf}(a_m)-\theta_0(a_m))$
converges in distribution to a mean-zero multivariate normal vector with
diagonal covariance matrix ${\rm diag}(V_{K,\tau}\sigma_0^2(a_1)/f_A(a_1),\ldots,V_{K,\tau}\sigma_0^2(a_m)/f_A(a_m))$.
\end{theorem}
\noindent The same limiting pair $(q_\infty,h_\infty)$ defines a common first-order representation at all $m$ points.  
Since the points are fixed and distinct and $h_n,b_n\to0$, the compact supports of the corresponding
localized leading terms are eventually disjoint.  Their cross-covariances are
therefore eventually zero, while the remaining terms are negligible
at the $(nh_n)^{1/2}$ scale.  The diagonal limit appears naturally.  
The theorem permits
joint inference for fixed dose contrasts by the delta method.

\subsection{Uniform inference}\label{sec: uniform inference}

Let $\mathcal A_0$ be a compact subset of the interior of $\mathcal A$. 
Theorem~\ref{thm:finite-dimensional-proximal-inference} implies that the
studentized estimators at any fixed collection of distinct exposure values
converge jointly to independent standard normal variables.  Suppose that
$\mathcal A_0$ is infinite, and choose distinct points
$a_1,a_2,\ldots\in\mathcal A_0$.  If the studentized estimator process
converged weakly in $\ell^\infty(\mathcal A_0)$ to a tight process $G$, then
continuity of the coordinate projections would imply
$\bigl(G(a_1),\ldots,G(a_m)\bigr)^\top\sim N_m(0,I_m)$,
for every $m\ge1$ \citep{VanderVaartWellner2023weak}.  Hence, for every $M<\infty$, $\Pr\{\sup_{j\ge1}|G(a_j)|\le M\}
=
\lim_{m\to\infty}
\{\Pr\bigl(|N(0,1)|\le M\bigr)\}^{m}
=0$.
Therefore, $\sup_{a\in\mathcal A_0}|G(a)|=\infty$ almost surely, which
contradicts $G\in\ell^\infty(\mathcal A_0)$.  The resulting
finite-dimensional law, consisting of independent standard normal
coordinates at distinct exposure values, is referred to as Gaussian white noise.  This white-noise argument is discussed in Section 3.3 of \citet{Takatsu2025debiased}. 
Uniform inference therefore requires a bandwidth-dependent Gaussian approximation rather than weak
convergence to a fixed process in $\ell^\infty(\mathcal A_0)$.
Following the local empirical process approximation machinery of \citep{Chernozhukov2014Gaussian} and its covariate-adjusted extension technique of \citet{Takatsu2025debiased}, we therefore use a Gaussian process whose covariance agrees with that of the normalized leading term.

Assumption~\ref{assm:G}, which specifies the bandwidth, nuisance-rate, smoothness, and grid conditions for the uniform inference results below, is provided in Appendix~\ref{sec:uniform-inference-conditions}.
For $a_0\in\mathcal A_0$, define
$\sigma_{\infty,h,b}^2(a_0):= hP_0\{\phi_{\infty,h,b,a_0}^2\}$.
Under Assumptions~\ref{assm:A}, \ref{assm:B}, \ref{assm:D}, and
conditions~\ref{assm:G-bandwidth}--\ref{assm:G-smooth},
$\sup_{a_0\in\mathcal A_0}
\left|
\sigma_{\infty,h_n,b_n}^2(a_0)
-V_{K,\tau}\sigma_0^2(a_0)/f_A(a_0)
\right|
\to0$.
Consequently, $\sigma_{\infty,h_n,b_n}^2$ is bounded away from zero and
infinity uniformly on $\mathcal A_0$ for all sufficiently large $n$.  For
such $n$, let $Z_{\infty,h_n,b_n}$ be a desired
mean-zero Gaussian process indexed by $\mathcal A_0$, with covariance
$\Sigma_{\infty,h_n,b_n}(u,v)
:=h_nP_0\{\phi_{\infty,h_n,b_n,u}\phi_{\infty,h_n,b_n,v}\}/(\sigma_{\infty,h_n,b_n}(u)
\sigma_{\infty,h_n,b_n}(v))$.
Define $T_n:=\sup_{a_0\in\mathcal A_0}
\frac{(nh_n)^{1/2}} {\sigma_{\infty,h_n,b_n}(a_0)} \left|
\hat\theta_n^{DB,cf}(a_0)-\theta_0(a_0)
\right|$.

\begin{theorem}\label{thm:uniform-proximal-inference}
Under Assumptions~\ref{assm:A}, \ref{assm:B}, \ref{assm:D}, and
\ref{assm:G}, $\sup_{a_0\in\mathcal A_0}
\left|
\hat\theta_n^{DB,cf}(a_0)-\theta_0(a_0)
-P_n\phi_{\infty,h_n,b_n,a_0}
\right|
=o_p\{(nh_n\log n)^{-1/2}\}$.
Moreover, $\sup_{t\in\mathbb R}
|
P_0(T_n\le t)-
P_0(
\sup_{a_0\in\mathcal A_0}
|Z_{\infty,h_n,b_n}(a_0)|
\le t )|
\to0$
In addition,
$\sup_{t\in\mathbb R}|
P_0(T_n\le t)-
P_0(
\max_{a\in\mathcal A_n}
|Z_{\infty,h_n,b_n}(a)|
\le t)|
\to0.$
\end{theorem}

The first conclusion of
Theorem~\ref{thm:uniform-proximal-inference} gives the uniform
triangular-array expansion at the rate required by Gaussian
anti-concentration \citep{Chernozhukov2014Gaussian}.  The second conclusion applies the Gaussian
approximation technique of \citet{Chernozhukov2014Gaussian} to the normalized
influence function class.  The process $Z_{\infty,h_n,b_n}$ is therefore an
$n$-dependent oracle process rather than a fixed weak limit.  Its
expected supremum is $O\{\sqrt{\log(h_n^{-1})}\}$, so the corresponding
oracle half-width is
$O[\sqrt{\log(h_n^{-1})/(nh_n)}]$.
Assumption~\ref{assm:G-grid} makes the oscillation of the oracle Gaussian
process between adjacent mesh points negligible at the Gaussian
anti-concentration scale.

\subsection{Debiased inference for derivative of the dose response function}
\label{sec:derivative-inference-main}

We next consider debiased inference for the first derivative
$\theta_0'(a)=d\theta_0(a)/da$ at interior exposure values. 
All notation from the preceding sections continues to apply.  We introduce the higher-order objects required for
the derivative estimation.  
Extend the base local-polynomial notation to
$t>0$ and $j\in\{2,3\}$ by writing $K_{t,a}(u):=t^{-1}K\left((u-a)/t\right)$, $w_{t,a,j}(u):=\left(1,(u-a)/t,\ldots, \left\{(u-a)/t\right\}^{j}\right)^\top$, $e_2:=(0,1,0)^\top$, and $e_4:=(0,0,0,1)^\top$.
We let $D_{P,t,a,2}$ denote the degree-two design matrix defined in Section~\ref{sec:estimation}, with $(b,a_0)$ replaced by $(t,a)$, and set $D_{P,t,a,3}
:=P\{w_{t,a,3}(A)w_{t,a,3}(A)^\top K_{t,a}(A)\}$, and $M_{P,t,a,j}
:=P\{w_{t,a,j}(A)K_{t,a}(A)\theta_P(A)\}$, for $j\in\{2,3\}$.
We further define
$U_{P,h,a,3}^{(1)}
:=P\!\left\{w_{h,a,2}(A)\left((A-a)/h\right)^3K_{h,a}(A)\right\}$, and 
$c_{P,h,a,3}^{(1)}
:=e_2^\top D_{P,h,a,2}^{-1}U_{P,h,a,3}^{(1)}$.

The local-quadratic derivative weight, local-cubic pilot weight, and their bias-corrected combination are $r_{P,h,a}^{(1)}(u)
:=h^{-1}e_2^\top D_{P,h,a,2}^{-1}w_{h,a,2}(u)K_{h,a}(u)$, $s_{P,b,a}^{(3)}(u)
:=6b^{-3}e_4^\top D_{P,b,a,3}^{-1}w_{b,a,3}(u)K_{b,a}(u)$, and $\Gamma_{P,h,b,a}^{(1)}(u)
:=r_{P,h,a}^{(1)}(u)
-\frac16h^2c_{P,h,a,3}^{(1)}s_{P,b,a}^{(3)}(u)$, respectively.
Denote $\theta_{P,h}^{(1)}(a):=P\{r_{P,h,a}^{(1)}(A)\theta_P(A)\}$, and $\theta_{P,b}^{(3)}(a):=P\{s_{P,b,a}^{(3)}(A)\theta_P(A)\}$.
Thus the smoothed, debiased derivative target is
\begin{equation}
\theta_{P,h,b}^{(1),\mathrm{DB}}(a)
:=P\{\Gamma_{P,h,b,a}^{(1)}(A)\theta_P(A)\}
=\theta_{P,h}^{(1)}(a)
-\frac16h^2c_{P,h,a,3}^{(1)}\theta_{P,b}^{(3)}(a).
\label{eq:deriv-smoothed-target}
\end{equation}
We define $\gamma_{P,h,a}^{(1)}(u)
:=h^{-1}e_2^\top D_{P,h,a,2}^{-1}w_{h,a,2}(u)K_{h,a}(u)w_{h,a,2}(u)^\top D_{P,h,a,2}^{-1}M_{P,h,a,2}$, $\gamma_{P,b,a}^{(3)}(u)
:=6b^{-3}e_4^\top D_{P,b,a,3}^{-1}w_{b,a,3}(u)K_{b,a}(u) w_{b,a,3}(u)^\top D_{P,b,a,3}^{-1}M_{P,b,a,3}$, $\gamma_{P,h,a}^{c,1}(u):=e_2^\top D_{P,h,a,2}^{-1}[
w_{h,a,2}(u)\left(\frac{u-a}{h}\right)^3
-w_{h,a,2}(u)w_{h,a,2}(u)^\top
D_{P,h,a,2}^{-1}U_{P,h,a,3}^{(1)}]K_{h,a}(u)$,
and
$\gamma_{P,h,b,a}^{(1)}(u)
:=\gamma_{P,h,a}^{(1)}(u)
-\frac16h^2c_{P,h,a,3}^{(1)}\gamma_{P,b,a}^{(3)}(u)
+\frac16h^2\theta_{P,b}^{(3)}(a)\gamma_{P,h,a}^{c,1}(u)$.
The derivative results reuse the nuisance-rate conditions for the
dose-response curve, which depends on $h_n$: Assumption~\ref{assm:F}
pointwise and \ref{assm:G-crossfit}--\ref{assm:G-nuisance-L2}
uniformly. The extra factor $h_n^{-1}$ in the derivative weights is
matched by the change in stochastic scale from $(nh_n)^{-1/2}$ to
$(nh_n^3)^{-1/2}$. Tables~\ref{tab:pointwise-remainder-comparison}
and~\ref{tab:uniform-remainder-comparison} in
Appendix~\ref{sec:derivative-remainder-comparison} compare all six
remainder bounds. 
The derivative-specific model,
smoothness, and bandwidth conditions in
Assumptions~\ref{assm:deriv-model}--\ref{assm:deriv-uniform}, together with
proofs of the results in this section, are given in
Appendix~\ref{sec:derivative-assumptions-proofs}.

\begin{theorem}\label{thm:deriv-if}
Under Assumption~\ref{assm:deriv-model} and condition~\ref{assm:kernel}, for
fixed $a_0$ and fixed $h,b>0$ the map
$P\mapsto\theta_{P,h,b}^{(1),\mathrm{DB}}(a_0)$ is pathwise differentiable at
$P_0$, with efficient influence function
$\varphi_{0,h,b,a_0}^{(1)}(O)
=\Gamma_{0,h,b,a_0}^{(1)}(A)\xi_0(O)
-\gamma_{0,h,b,a_0}^{(1)}(A)+\int\Gamma_{0,h,b,a_0}^{(1)}(\bar a)
\{h_0(W,\bar a,X)-\theta_0(\bar a)\}\,dF_0(\bar a)$.
\end{theorem}

For estimation, define $D_{n,h,a,2}$, $D_{n,b,a,3}$, and
$U_{n,h,a,3}^{(1)}$ by replacing $P$ with $P_n$ in the preceding
definitions. 
Define $c_{n,h,a,3}^{(1)}=e_2^\top D_{n,h,a,2}^{-1}U_{n,h,a,3}^{(1)}$, and 
$\Gamma_{n,h,b,a}^{(1)}$ be the corresponding empirical version of $\Gamma_{P,h,b,a}^{(1)}$.
Using the same folds, fitted nuisance bridges, and fold pseudo-outcomes as in
\eqref{eq:DB-cf-representations}, define
$\hat\theta_n^{(1),\mathrm{DB},\mathrm{cf}}(a)
:=\sum_{k=1}^K\frac{n_k}{n}
P_{n,k}\{\Gamma_{n,h_n,b_n,a}^{(1)}(A)\hat\xi_n^{(-k)}(O)\}$.

For the limiting variance, write $r_2(u)=(1,u,u^2)^\top$, $r_3(u)=(1,u,u^2,u^3)^\top$, $S_3^{(1)}:=\int r_2(u)r_2(u)^\top K(u)\,du$, $S_4^{(1)}:=\int r_3(u)r_3(u)^\top K(u)\,du$,
and $c_{K,3}^{(1)}
:=e_2^\top\{S_3^{(1)}\}^{-1}\int r_2(u)u^3K(u)\,du$.
For $\tau\in[0,\infty)$, define $\Lambda_{K,\tau}^{(1)}(u)
:=e_2^\top\{S_3^{(1)}\}^{-1}r_2(u)K(u)-\tau^4c_{K,3}^{(1)}e_4^\top\{S_4^{(1)}\}^{-1}
r_3(\tau u)K(\tau u)$, and  $V_{K,\tau}^{(1)}:=\int\{\Lambda_{K,\tau}^{(1)}(u)\}^2\,du$.
Finally, let
\begin{equation*}
\begin{split}
\phi_{\infty,h,b,a}^{(1)}(O)
:=\Gamma_{0,h,b,a}^{(1)}(A)\xi_\infty(O)
-\gamma_{0,h,b,a}^{(1)}(A)+\int\Gamma_{0,h,b,a}^{(1)}(\bar a)
\{h_\infty(W,\bar a,X)-m_\infty(\bar a)\}\,dF_0(\bar a).
\end{split}
\end{equation*}

\subsubsection{Pointwise and finite-dimensional inference}

\begin{theorem}\label{thm:deriv-pointwise}
Suppose that Assumptions~\ref{assm:A}, \ref{assm:B}, \ref{assm:D},
\ref{assm:F}, \ref{assm:deriv-model}, 
and~\ref{assm:deriv-pointwise} hold.  Suppose also that
condition~\ref{assm:regul1}\textup{(b)}, the conditional
$(2+\delta_2)$-moment requirement in condition~\ref{assm:regul1}\textup{(c)},
and the continuity requirement in condition~\ref{assm:regul1}\textup{(d)},
imposed at $a_0$, hold.  Then
$\hat\theta_n^{(1),\mathrm{DB},\mathrm{cf}}(a_0)-\theta_0'(a_0)
=P_n\phi_{\infty,h_n,b_n,a_0}^{(1)}
+o_p\{(nh_n^3)^{-1/2}+h_n^2\}$, 
$(nh_n^3)^{1/2}P_n\phi_{\infty,h_n,b_n,a_0}^{(1)}
\rightsquigarrow
N\!\left(0,
\frac{V_{K,\tau}^{(1)}\sigma_0^2(a_0)}{f_A(a_0)}\right)$,
and consequently
$(nh_n^3)^{1/2}
\{\hat\theta_n^{(1),\mathrm{DB},\mathrm{cf}}(a_0)-\theta_0'(a_0)\}
\rightsquigarrow
N\!\left(0,
\frac{V_{K,\tau}^{(1)}\sigma_0^2(a_0)}{f_A(a_0)}\right)$.
\end{theorem}

\begin{theorem}\label{thm:deriv-finite}
Fix distinct interior points $a_1,\ldots,a_m$.  If the assumptions of
Theorem~\ref{thm:deriv-pointwise} hold at every $a_j$ for the same limiting
pair $(q_\infty,h_\infty)$, then
$(nh_n^3)^{1/2}
\left(
\hat\theta_n^{(1),\mathrm{DB},\mathrm{cf}}(a_j)-\theta_0'(a_j)
\right)_{j=1}^m
\rightsquigarrow N_m(0,\Sigma^{(1)})$,
where $\Sigma^{(1)}={\rm diag}\left(V_{K,\tau}^{(1)}\sigma_0^2(a_1)/f_A(a_1),\ldots,V_{K,\tau}^{(1)}\sigma_0^2(a_m)/ f_A(a_m)\right)$.
\end{theorem}

\subsubsection{Uniform inference}

For $a\in\mathcal A_0$, define
$\sigma_{\infty,h,b}^{(1)2}(a)
:=h^3P_0\{\phi_{\infty,h,b,a}^{(1)}(O)^2\}$.
Let $Z_{\infty,h,b}^{(1)}$ be a mean-zero Gaussian process on
$\mathcal A_0$ with covariance
$\Sigma_{\infty,h,b}^{(1)}(u,v)
:=h^3P_0\{\phi_{\infty,h,b,u}^{(1)}
\phi_{\infty,h,b,v}^{(1)}\}/(\sigma_{\infty,h,b}^{(1)}(u)\sigma_{\infty,h,b}^{(1)}(v))$.

\begin{theorem}\label{thm:deriv-uniform}
Under Assumptions~\ref{assm:A}, \ref{assm:B}, \ref{assm:D},
\ref{assm:deriv-model}, \ref{assm:deriv-point-centering},
and~\ref{assm:deriv-uniform}, conditions~\ref{assm:G-crossfit}--\ref{assm:G-nuisance-L2},
and condition~\ref{assm:G-smooth}\textup{(ii)--(iv)}, $\sup_{a\in\mathcal A_0}
|\hat\theta_n^{(1),\mathrm{DB},\mathrm{cf}}(a)-\theta_0'(a)
-P_n\phi_{\infty,h_n,b_n,a}^{(1)}|
=o_p\{(nh_n^3\log n)^{-1/2}\}$.
Moreover,
$\sup_{t\in\mathbb R}|
P_0\{
\sup_{a\in\mathcal A_0}
|\frac{(nh_n^3)^{1/2}
\{\hat\theta_n^{(1),\mathrm{DB},\mathrm{cf}}(a)-\theta_0'(a)\}}
{\sigma_{\infty,h_n,b_n}^{(1)}(a)}|\le t\}-P_0\{\sup_{a\in\mathcal A_0}
|Z_{\infty,h_n,b_n}^{(1)}(a)|\le t\}| \rightarrow0$.
If a finite mesh $\mathcal A_n\subset\mathcal A_0$ satisfies the mesh-size
requirement in condition~\ref{assm:G-grid}, then the last display remains
true after the replacement of $\sup_{a\in\mathcal A_0}|Z_{\infty,h_n,b_n}^{(1)}(a)|$ with $\max_{a\in\mathcal A_n}|Z_{\infty,h_n,b_n}^{(1)}(a)|$.
\end{theorem}

\subsection{Covariance estimation}
\label{sec:inference-implementation-covariance}

Following Sections~3.1--3.3 of \citet{Takatsu2025debiased}, we estimate the
covariance using finite-sample versions of the influence functions of the
smoothed targets.  Let \(\nu=0\) denote the dose-response function and
\(\nu=1\) its first derivative, with
\(\widehat\theta_n^{(0),\mathrm{DB,cf}}
:=\widehat\theta_n^{\mathrm{DB,cf}}\).  
Denote
\(\widehat\Gamma_{n,h,b,a}^{(\nu)}\) and
\(\widehat\gamma_{n,h,b,a}^{(\nu)}\) for the empirical analogues of the
corresponding weight and design-correction functions.
For \(i\in I_k\), define 
$\widehat\psi_{n,h,b,a,i}^{(\nu)}
:=
\widehat\Gamma_{n,h,b,a}^{(\nu)}(A_i)
\widehat\xi_n^{(-k)}(O_i)
-\widehat\gamma_{n,h,b,a}^{(\nu)}(A_i)
+
P_{n,k}\![
\widehat\Gamma_{n,h,b,a}^{(\nu)}(A)\{
\widehat h_n^{(-k)}(W_i,A,X_i)
-\widehat m_{n,k}^{(-k)}(A)\}]$,
and its centered version $\widehat\phi_{n,h,b,a,i}^{(\nu)}
:=
\widehat\psi_{n,h,b,a,i}^{(\nu)}
-\frac{1}{n}\sum_{j=1}^n
\widehat\psi_{n,h,b,a,j}^{(\nu)}$.
For \(\nu=0\), this is the finite-sample version of
\(\phi_{\infty,h,b,a_0}(O)\); for \(\nu=1\), it is the corresponding
version of \(\phi_{\infty,h,b,a}^{(1)}\).

For exposure values \(u\) and \(v\), define $\widehat\Sigma_{n,h,b}^{(\nu)}(u,v)
:=
\frac{h^{2\nu+1}}{n}\sum_{i=1}^n
\widehat\phi_{n,h,b,u,i}^{(\nu)}
\widehat\phi_{n,h,b,v,i}^{(\nu)}$, and 
$\widehat\sigma_{n,h,b}^{(\nu)}(u)^2
:=
\widehat\Sigma_{n,h,b}^{(\nu)}(u,u)$.
Thus the estimator-scale covariance is
$\widehat{\operatorname{Cov}}\!\left\{
\widehat\theta_n^{(\nu),\mathrm{DB,cf}}(u),
\widehat\theta_n^{(\nu),\mathrm{DB,cf}}(v)
\right\}
=
\frac{\widehat\Sigma_{n,h,b}^{(\nu)}(u,v)}{nh^{2\nu+1}}
=
\frac{1}{n^2}\sum_{i=1}^n
\widehat\phi_{n,h,b,u,i}^{(\nu)}
\widehat\phi_{n,h,b,v,i}^{(\nu)}$.
The estimated pointwise standard error at \(a\) is therefore $\widehat{\operatorname{se}}^{(\nu)}(a)=\widehat\sigma_{n,h,b}^{(\nu)}(a)/\sqrt{nh^{2\nu+1}}$.
In particular, the denominators are \(\sqrt{nh}\) for the dose-response
function and \(\sqrt{nh^3}\) for its derivative.

For a finite evaluation grid \(\mathcal A_n=\{a_1,\ldots,a_m\}\), define $\widehat R_{n,h,b}^{(\nu)}(u,v)
:=\widehat\Sigma_{n,h,b}^{(\nu)}(u,v)/(\widehat\sigma_{n,h,b}^{(\nu)}(u)
 \widehat\sigma_{n,h,b}^{(\nu)}(v))$, $u,v\in\mathcal A_n$.
The numerical implementation symmetrizes the resulting matrix and applies a positive-semidefinite
correlation projection if required.  Conditional on
\(\widehat R_{n,h,b}^{(\nu)}\), generate independent draws $G^{(\ell)}\sim
N_m\!(0,\widehat R_{n,h,b}^{(\nu)})$, for $\ell=1,\ldots,B$.
Let \(\widehat t_{1-\alpha}^{(\nu)}\) denote the empirical
\((1-\alpha)\)-quantile of
\(\max_{1\le r\le m}|G_r^{(\ell)}|\).  
Then, the simultaneous intervals are $\widehat\theta_n^{(\nu),\mathrm{DB,cf}}(a_r)
\ \pm\ 
\widehat t_{1-\alpha}^{(\nu)}
\frac{\widehat\sigma_{n,h,b}^{(\nu)}(a_r)}
{\sqrt{nh^{2\nu+1}}}$, $r=1,\ldots,m$.

\begin{remark}\label{remark: which to cross fit?}
For the dose-response and derivative intervals and bands above, the bridge
estimators are trained outside each evaluation fold.  The empirical design
matrices, the bias coefficient $c_{n,h,a_0,2}$, and
$\Gamma_{n,h,b,a_0}$ use the full empirical treatment distribution $P_n$.
The moment vectors entering $\widehat\gamma_{n,h,b,a_0}$, their centering,
and the variance and covariance estimates are computed from the pooled
cross-fitted pseudo-outcomes.  For the test in
\eqref{eq:parametric-null-statistic}, the smoothing moments and null fit
instead use $P_{n,k}$ within each summand; the fold statistics are averaged
with weights $n_k/n$.
\end{remark}

\subsection{Parametric null testing}
\label{sec:parametric-null-main}

Let $\mathcal M=\{m(\cdot;\beta):\beta\in\mathcal B\subset\mathbb R^p\}$, where $p$ is fixed.  
We extend the integrated squared-deviation testing principle of
\citet{Doss2024nonparametric} from constancy to a fixed-dimensional
parametric family under proximal identification. The choice
$m(a;\beta)=\beta$ recovers their constancy null, while linear and
quadratic families assess departures from more flexible response shapes.
Specifically, we test the following null hypothesis:
\begin{equation}
H_{0,\mathcal M}:\quad
\theta_0(a)=m(a;\beta_0),\qquad a\in\mathcal A,
\label{eq:parametric-null}
\end{equation}
for some $\beta_0\in\mathcal B$.
To test this null, we keep the cross-fitting scheme described above.
Within each evaluation fold, we compare the estimated dose-response
curve with the fitted parametric null.
Before comparing them, we smooth the fitted null using the same
weights as the debiased curve estimator.
We then average the fold-specific test statistics with weights $n_k/n$.
Recall that, for each fold $k$, $\Gamma_{n,k,h_n,b_n,a}$ denotes the debiased weight obtained from $\Gamma_{n,h_n,b_n,a}$ by replacing $P_n$ by $P_{n,k}$ and the sample size
$n$ by $n_k$ in its empirical moment matrices and bias coefficient,
while keeping the bandwidths $h_n,b_n$. 
We denote $\hat\theta_{n,k}^{DB}(a)
:=P_{n,k}\{\Gamma_{n,k,h_n,b_n,a}(A)\hat\xi_n^{(-k)}(O)\}$.
Then we fit the null separately on each evaluation fold:
\begin{equation}
\hat\beta_{n,k}
\in\argmin_{\beta\in\mathcal B}
P_{n,k}\left[\{\hat\xi_n^{(-k)}(O)-m(A;\beta)\}^2\right].
\label{eq:parametric-null-beta}
\end{equation}
Next we pass that fit through the same fold operator; $\hat m_{n,k,h_n,b_n}(a;\beta)
:=P_{n,k}\{\Gamma_{n,k,h_n,b_n,a}(A)m(A;\beta)\}$.
The assumptions used in this section are deferred to
Section~\ref{sec:parametric-null-proofs}.
For the nonnegative weight $w$ in Assumption~\ref{assm:parametric-null}, motivated by \citet{hardle1993comparing}, we compute the test statistic as follows:
\begin{equation}
\begin{aligned}
\widehat T_{n,\mathcal M}^{DB}
:=\sum_{k=1}^K\frac{n_k}{n}\,n_k\sqrt{h_n}
\int_{\mathcal A_0}
\bigl\{\hat\theta_{n,k}^{DB}(a)
-\hat m_{n,k,h_n,b_n}(a;\hat\beta_{n,k})\bigr\}^2w(a)\,da.
\end{aligned}
\label{eq:parametric-null-statistic}
\end{equation}
For $s>0$, define $\Lambda_s(u)
:=K(u)-s^3c_2((su)^2-c_2)/(c_4-c_2^2)K(su)$, where $\Lambda_s^{(2)}(0)=\int\Lambda_s(u)^2\,du$, and 
$\Lambda_s^{(4)}(0)
=\int\left\{\int\Lambda_s(v)\Lambda_s(v+r)\,dv\right\}^2dr$.
We further define $b_{0h,n}^{DB}:=h_n^{-1/2}\Lambda_{\tau_n}^{(2)}(0)
\int_{\mathcal A_0}\frac{\sigma_0^2(a)w(a)}{f_A(a)}\,da$, and $V_\tau^{DB}
:=2\Lambda_\tau^{(4)}(0)
\int_{\mathcal A_0}
\left\{\frac{\sigma_0^2(a)w(a)}{f_A(a)}\right\}^2da$.

\begin{theorem}
\label{thm:parametric-null}
Under $H_{0,\mathcal M}$ and Assumptions~\ref{assm:A}, \ref{assm:B},
\ref{assm:D}, \ref{assm:G}, and~\ref{assm:parametric-null},
$(V_\tau^{DB}\sum_{k=1}^K(n_k/n)^2)^{-1/2}\{\widehat T_{n,\mathcal M}^{DB}-b_{0h,n}^{DB}\}
\rightsquigarrow N(0,1)$.
\end{theorem}
\noindent The proof is given in Section~\ref{sec:parametric-null-proofs}. 
The bias correction addresses the
deterministic smoothing error of the estimated curve.  Even if that
error were absent, squaring a mean-zero estimation error would still
produce a positive expectation equal to its variance.
For fold $k$, the leading pointwise variance is of order
$(n_kh_n)^{-1}$.  Integrating the squared error and multiplying by
$n_k\sqrt{h_n}$ therefore produces a centering term of order
$h_n^{-1/2}$.  This equals the term $b_{0h,n}^{DB}$.  Averaging the fold statistics with
weights $n_k/n$ preserves this term because $\sum_k n_k/n=1$.
Thus $b_{0h,n}^{DB}$ reflects stochastic variation, not residual smoothing
bias.  
The debiasing step changes its kernel factor
$\Lambda_{\tau_n}^{(2)}(0)$ compared to its canonical counterpart, but cannot make that factor zero.
Thus, we note that debiasing the curve estimator does not remove the positive centering term of the test statistic that we use.

Let $\varrho_n:=(n\sqrt{h_n})^{-1/2}$.
For a fixed bounded function $g$ and a fixed $c\in\mathbb R$, consider the local alternative sequence
\begin{equation}
\theta_{0,n}(a)
=m(a;\beta_0)+c\varrho_ng(a),
\qquad a\in\mathcal A.
\label{eq:parametric-null-local-alternative}
\end{equation}
Define $g_{\mathcal M}^{\perp}(a):=
g(a)-\dot m_0(a)^\top J_0^{-1}
P_0\{\dot m_0(A)g(A)\}$.
Here $\dot m_0$ and $J_0$ are formally defined in condition
\ref{assm:parametric-null-family}.  
We have
$P_0\{r(A)\}:=\int r(a)\,dF_0(a)$ for every $F_0$-integrable $r$.

\begin{theorem}
\label{thm:parametric-null-local-power}
Under Assumptions~\ref{assm:parametric-null-triangular} and
\ref{assm:parametric-null-local}, along the sequence in
\eqref{eq:parametric-null-local-alternative},
$(V_{\tau,n}^{DB}\sum_{k=1}^K(n_k/n)^2)^{-1/2}\{\widehat T_{n,\mathcal M}^{DB}-b_{0h,n}^{DB}
-c^2\sum_{k=1}^K(n_k/n)^2
\int_{\mathcal A_0}\{g_{\mathcal M}^{\perp}(a)\}^2w(a)\,da\}$
converges to $N(0,1)$, where $b_{0h,n}^{DB}$ and $V_{\tau,n}^{DB}$ are defined in
\eqref{eq:parametric-null-local-centering} and
\eqref{eq:parametric-null-local-variance} in Section~\ref{sec:parametric-null-local-alternatives}.
Consequently, for fixed $\alpha_0\in(0,1)$, if $\Phi$ is the standard normal
distribution function and
$z_{1-\alpha_0}:=\Phi^{-1}(1-\alpha_0)$, then, conditional on the partition,
\[
\begin{aligned}
&P_{0,n}\left\{
\widehat T_{n,\mathcal M}^{DB}
>
b_{0h,n}^{DB}
+\sqrt{V_{\tau,n}^{DB}\sum_{k=1}^K(n_k/n)^2}\,z_{1-\alpha_0}
\,\middle|\,I_1,\ldots,I_K
\right\}\\
&\qquad-\left\{1-\Phi\left[
z_{1-\alpha_0}
-
\frac{c^2\sum_{k=1}^K(n_k/n)^2\int_{\mathcal A_0}
\{g_{\mathcal M}^{\perp}(a)\}^2w(a)\,da}
{\sqrt{V_{\tau,n}^{DB}\sum_{k=1}^K(n_k/n)^2}}
\right]\right\}=o_{P_{0,n}}(1).
\end{aligned}
\]
\end{theorem}
\noindent Its proof is given in Section~\ref{sec:parametric-null-local-alternatives}. We note that Theorem~3.2 and displays~(F.1)--(F.9) of
\citet{Doss2024nonparametric} consider
$\theta_0(a)=c_0+\delta_n\varrho_ng(a)$ with
$\delta_n\to\infty$ and $n^{1/40}/\delta_n\to0$, and establish power
converging to one under their assumptions.
In contrast, Theorem~\ref{thm:parametric-null-local-power} gives the Gaussian limit and
power for each fixed $c$, and allows the parametric family $\mathcal M$.

In Appendix~\ref{sec:parametric-null-polynomial-examples}, we provide the constant, linear, and quadratic null examples, including their explicit test statistics and projected local-alternative directions in Table~\ref{tab:parametric-null-polynomial-families}, together with the corresponding null and local-alternative corollaries.

For numerical calibration, motivated by \citet{Doss2024nonparametric}, we use a wild response bootstrap. Conditional on the data, the multipliers $\omega_i$ are independent and take the values $(1-\sqrt5)/2$ and $(1+\sqrt5)/2$ with probabilities $(\sqrt5+1)/(2\sqrt5)$ and $(\sqrt5-1)/(2\sqrt5)$, respectively.
These multipliers satisfy $E(\omega_i)=0$ and
$E(\omega_i^2)=E(\omega_i^3)=1$.
Form one bootstrap response vector by setting, for
$i\in I_k$, $\hat\xi_{n,i}^{*}
:=m(A_i;\hat\beta_{n,k})+\omega_i
\{\hat\xi_n^{(-k)}(O_i)-m(A_i;\hat\beta_{n,k})\}$.
An alternative replaces the residual term by $\hat\xi_n^{(-k)}(O_i)-\hat\theta_{n,k}^{DB}(A_i)$, for $i\in I_k$.  
The proposed construction parallels Theorem~3.3 of
\citet{Doss2024nonparametric}; the alternative adapts Theorem~3.4 and
Appendix~J.1 of \citet{Doss2024nonparametric}.
We then refit the null separately on every fold and apply the same fold smoother: $\hat\beta_{n,k}^{*}
\in\argmin_{\beta\in\mathcal B}
\frac1{n_k}\sum_{i\in I_k}
\{\hat\xi_{n,i}^{*}-m(A_i;\beta)\}^2$, and 
$\hat\theta_{n,k}^{DB,*}(a)
:=\frac1{n_k}\sum_{i\in I_k}
\Gamma_{n,k,h_n,b_n,a}(A_i)\hat\xi_{n,i}^{*}$.
The proposed bootstrap statistic is
\[
\begin{aligned}
\widehat T_{n,\mathcal M}^{DB,*}
:=\sum_{k=1}^K\frac{n_k}{n}\,n_k\sqrt{h_n}
\int_{\mathcal A_0}
\bigl\{\hat\theta_{n,k}^{DB,*}(a)
-\hat m_{n,k,h_n,b_n}(a;\hat\beta_{n,k}^{*})\bigr\}^2w(a)\,da.
\end{aligned}
\]
In practice, we repeat this construction independently $B$ times and use
the corrected bootstrap $p$-value described in
Section~\ref{sec:simulation-parametric-null-main}, rejecting at level $\alpha_0$.
We keep the folds, fitted bridges, bandwidths, and fold smoothing operators fixed across random draws.

We expect that the proofs of Theorems~3.3--3.4 of
\citet{Doss2024nonparametric} (see Appendix~L--M therein)
can be adapted to establish the validity of both bootstrap procedures.
We omit these theoretical extensions and compare the finite-sample
calibration and power of the two bootstrap procedures in
Section~\ref{sec:simulation-parametric-null-main}.
In addition, Remark~\ref{rem:debiased-local-linear-variance} in
Section~\ref{sec:simulation-parametric-null-main} discusses the variance
comparison with canonical local-linear smoothing and its implications for mean squared error and testing power.

\begin{acks}
The authors were partially supported by the National Key R\&D Program of China (2024YFA1015600), the National Natural Science Foundation of China (12471266 and U23A2064). Correspondence to Daeyoung Ham and Yifan Cui. 
We thank Dr. Charles R. Doss for helpful discussions on parametric null
testing and Dr. Ted Westling for helpful discussions on debiased derivative inference.
We also thank Dr. Eric J. Tchetgen Tchetgen for bringing the work of \citet{yang2023forster} to our attention.
\end{acks}

\section*{Data availability statement}
Public benchmark data for the abortion--crime application are available in the repository accompanying \citet{mastouri2021proximal} at \url{https://github.com/yuchen-zhu/kernel_proxies}. As described in Section~\ref{sec:realdata}, our evaluation uses observations generated from a model fitted to the original abortion--crime data.

\else
  \setcounter{equation}{10}
\setcounter{table}{0}
\setcounter{figure}{0}
\setcounter{alg}{0}
\setcounter{proposition}{0}
\setcounter{corollary}{0}
\setcounter{defin}{0}
\setcounter{assumption}{0}
\setcounter{theorem}{10}
\setcounter{lemma}{0}
\setcounter{remark}{2}
\setcounter{assumptionalt}{6}

\fi
\ifnum\documentpart=0
  % Preserve the manuscript's continuous main/supplement numbering.
  \newwrite\manuscriptcounterfile
  \immediate\openout\manuscriptcounterfile=\jobname-main-counters.tex
  \makeatletter
  \@for\manuscriptcounter:=equation,table,figure,alg,proposition,corollary,defin,assumption,theorem,lemma,remark,assumptionalt\do{%
    \immediate\write\manuscriptcounterfile{\string\setcounter{\manuscriptcounter}{\number\value{\manuscriptcounter}}}%
  }
  \makeatother
  \immediate\closeout\manuscriptcounterfile
\fi
\ifnum\documentpart=1
  \spacingset{1.55}
\else
  \ifnum\documentpart=0\clearpage\fi
  \spacingset{1.2}
\bigskip\bigskip\bigskip
\begin{center}
  {\LARGE\bfseries Supplementary Materials for ``\papertitle''}
  \ifnum\blind=1
    \\[1.2em]
    {\large Daeyoung Ham}\\
    Department of Statistics and Data Science, University of Texas at San Antonio\\
    \texttt{daeyoung.ham@utsa.edu}\\[0.75em]
    {\large Sihan Wu and Yifan Cui}\\
    Center for Data Science, Zhejiang University\\
    \texttt{wush03@zju.edu.cn}, \texttt{cuiyf@zju.edu.cn}
  \fi
\end{center}
\medskip
\spacingset{1.7}

  \appendix
  \section{Numerical studies}\label{sec: simulations}

\subsection{Numerical experiments}
\label{sec:numerical-experiments}

\subsubsection{Simulation setup}
\label{sec:simulation}

We evaluate the finite-sample behavior of the proposed estimators under
correct and misspecified bridge models for both the dose-response function
and its first derivative.  Let $X:=(X_1,X_2)^\top$ and
$O=(Y,A,Z,W,X_1,X_2)$.  The mutually independent primitive variables are
generated according to $X_1,X_2\overset{\mathrm{iid}}{\sim}N(0,0.25^2)$, $U\sim N(0,1)$., $\varepsilon_A\sim N(0,0.79)$, $\varepsilon_Z,\varepsilon_W\overset{\mathrm{iid}}{\sim}N(0,0.5^2)$, and 
$\varepsilon_Y\sim\operatorname{Unif}[-0.5,0.5]$.
Set
\begin{equation}\label{eqs:simulation design}
\begin{split}
 A&:=0.25X_1+0.25X_2+0.15U+\varepsilon_A,\\
 Z&:=0.25+0.25A+0.25X_1+0.25X_2+U+\varepsilon_Z,\\
 W&:=0.25X_1+0.25X_2+U+\varepsilon_W,\\
 Y&:=2.5+\sin(A+0.2)+\tanh(W)+\varepsilon_Y.
\end{split}
\end{equation}
For each fixed $a$, define
\[
Y(a):=2.5+\sin(a+0.2)+\tanh(W)+\varepsilon_Y.
\]
This construction satisfies $Y=Y(A)$,
$W\perp\!\!\!\perp(Z,A)\mid(U,X)$,
$Y\perp\!\!\!\perp Z\mid(U,A,X)$, and
$Y(a)\perp\!\!\!\perp A\mid(U,X)$.
The outcome bridge and dose-response function are
\begin{equation}
  h_0(w,a,x)=2.5+\sin(a+0.2)+\tanh(w)
  \label{eq:simulation-h0}
\end{equation}
and
\begin{equation}
  \theta_0(a)=E\{h_0(W,a,X)\}=2.5+\sin(a+0.2),
  \label{eq:simulation-theta0}
\end{equation}
respectively, where symmetry of $W$ gives $E\{\tanh(W)\}=0$.  Moreover,
\begin{equation}
\theta_0'(a)=\cos(a+0.2).
\label{eq:simulation-derivative-truth}
\end{equation}

Let $d(a,x):=a-0.25(x_1+x_2)$,
$r(z,a,x):=z-0.25-0.25a-0.25(x_1+x_2)-d(a,x)/0.15$,
$V_A:=105/128$, $V_R:=0.5^2+0.79/0.15^2$, and
$Q:=0.79/(0.15\sqrt{V_AV_R})$.  The treatment bridge is
\begin{equation}
  q_0(z,a,x)=Q\exp\!\left\{
  -\frac{a^2}{2V_A}
  +\frac{r(z,a,x)^2}{2V_R}
  \right\},
  \label{eq:simulation-q0}
\end{equation}
which satisfies
\[
  E\{q_0(Z,a,X)\mid W,A=a,X\}
  =
  \frac{f_A(a)}{f(a\mid W,X)}.
\]
The bridge equations are verified in
Section~\ref{sec:simulation-bridge-compatibility}.

The bridge functions are estimated using five-fold cross-fitting.  On each
training sample, the outcome bridge is fit by generalized method of moments
and the treatment bridge by exponential calibration.  
The scenarios CC, CM, MC, and MM specify, respectively, that both
bridge models are correct, only the outcome bridge model is correct, only the
treatment bridge model is correct, and neither bridge model is correct.  The
first letter refers to the outcome bridge and the second to the treatment
bridge.  
The computational procedure of the nuisance estimators for each scenario is illustrated in Section~\ref{sec:simulation-nuisance}.
The scenario MM is included only as a joint-misspecification diagnostic, and its results are therefore deferred to Section~\ref{sec: MM reports}.

We consider $n\in\{2500,5000,7500,10000\}$ and use $5{,}000$ independent
replications at each sample size.  Estimation and inference are evaluated on
$\mathcal G:=\{-0.90,-0.87,\ldots,0.90\}$, where $|\mathcal G|=61$.
The four PDR estimators set $b=h$.  For dose-response inference, we also
include ProxLL--US, a constructed proximal analogue of the undersmoothed
local-linear estimator of \citet{kennedy2017non}.  For derivative inference,
ProxLQ--US replaces the degree-one slope by the degree-two local-polynomial
slope.  Within each estimand, the five estimators are evaluated using the same
generated sample, fold assignment, and scenario-specific cross-fitted
pseudo-outcomes within each replication.  Pointwise confidence intervals are
constructed for all five estimators.  Simultaneous confidence bands are
constructed for the four PDR estimators using $B=10{,}000$ Gaussian random draws and
the full $61\times61$ estimated correlation matrix.

\subsubsection{Dose-response inference}
\label{sec:simulation-dose-response}

For the dose-response function, we consider four different bandwidth selectors; two plug-ins (PIs) and the other two cross-validated schemes (CVs), where all estimators used $b=h$.  PDR--PI1 uses an integrated
plug-in criterion over $[-0.9,0.9]$, and PDR--PI2 uses the direct scale rule
$1.5\widehat{\operatorname{sd}}(A)n^{-1/5}$.  PDR--CV1 minimizes the
degree-two leave-one-out criterion over a grid, and
PDR--CV2 minimizes the same criterion over the shifted grid centered at the
PDR--PI2 bandwidth.  
As in \citet{Takatsu2025debiased}, ProxLL--US uses the degree-one local-polynomial intercept
with its selected local-linear bandwidth divided by $\log_{10}(n)$.  The
selectors are illustrated in Section~\ref{sec: DRF badn selectors}.

Figure~\ref{fig:sim-pointwise-coverage-dose} shows broadly similar
pointwise coverage near the nominal level in CC, CM, and MC.
All four PDR procedures produce narrower intervals and smaller mean grid
MSE than ProxLL--US, with PDR--CV2 best on both measures and PDR--CV1 next
(see Table~\ref{tab:sim-dose-main} and the tables and plots in
Section~\ref{sec:simulation-additional-results}).
The two plug-in methods yielded wider intervals and selected larger bandwidths, but still outperformed ProxLL--US on both measures.
The coordinatewise results show that intervals widen toward both ends of
the treatment grid, with the PDR width advantage throughout.
The scenarios CC and CM behave similarly; MC has wider intervals and larger MSE.
Increasing $n$ reduces width and MSE for every method.

For simultaneous inference over the fixed grid,
Figure~\ref{fig:sim-band-coverage-dose} shows the most consistently
near-nominal coverage for PDR--PI2, which also gives the widest PDR bands.
PDR--CV2 gives the narrowest bands but mildly undercovers, including at
the largest sample size. PDR--CV1 offers intermediate widths and
coverage closer to nominal at larger sample sizes.
All four band widths decrease with $n$.

\begin{table}[!t]
\centering
\caption{Dose-response inference at $n=10{,}000$.  Panel A averages over CC,
CM, MC, and the 61 grid coordinates.  Panel B reports fixed-grid simultaneous
coverage by scenario and mean width averaged over the three scenarios.}
\label{tab:sim-dose-main}
\small
\begin{tabular}{@{}lccc@{}}
\toprule
\multicolumn{4}{c}{Panel A: Pointwise inference and estimation}\\
\midrule
Method & Coverage & Mean width & Mean grid MSE\\
\midrule
PDR--PI1    & .9489 & .06546 & $3.2949\times10^{-4}$\\
PDR--PI2    & .9510 & .07269 & $3.9759\times10^{-4}$\\
PDR--CV1    & .9494 & .05707 & $2.4612\times10^{-4}$\\
PDR--CV2    & .9482 & .05003 & $1.8918\times10^{-4}$\\
ProxLL--US  & .9520 & .08613 & $5.6567\times10^{-4}$\\
\bottomrule
\end{tabular}

\vspace{0.8em}

\begin{tabular}{@{}lcccc@{}}
\toprule
\multicolumn{5}{c}{Panel B: Fixed-grid simultaneous inference}\\
\midrule
Method & CC & CM & MC & Mean width\\
\midrule
PDR--PI1 & .9454 & .9438 & .9468 & .10153\\
PDR--PI2 & .9508 & .9510 & .9460 & .11513\\
PDR--CV1 & .9462 & .9454 & .9484 & .08508\\
PDR--CV2 & .9430 & .9426 & .9440 & .07160\\
\bottomrule
\end{tabular}
\end{table}

\begin{figure}[!p]
\centering
\includegraphics[width=\textwidth]{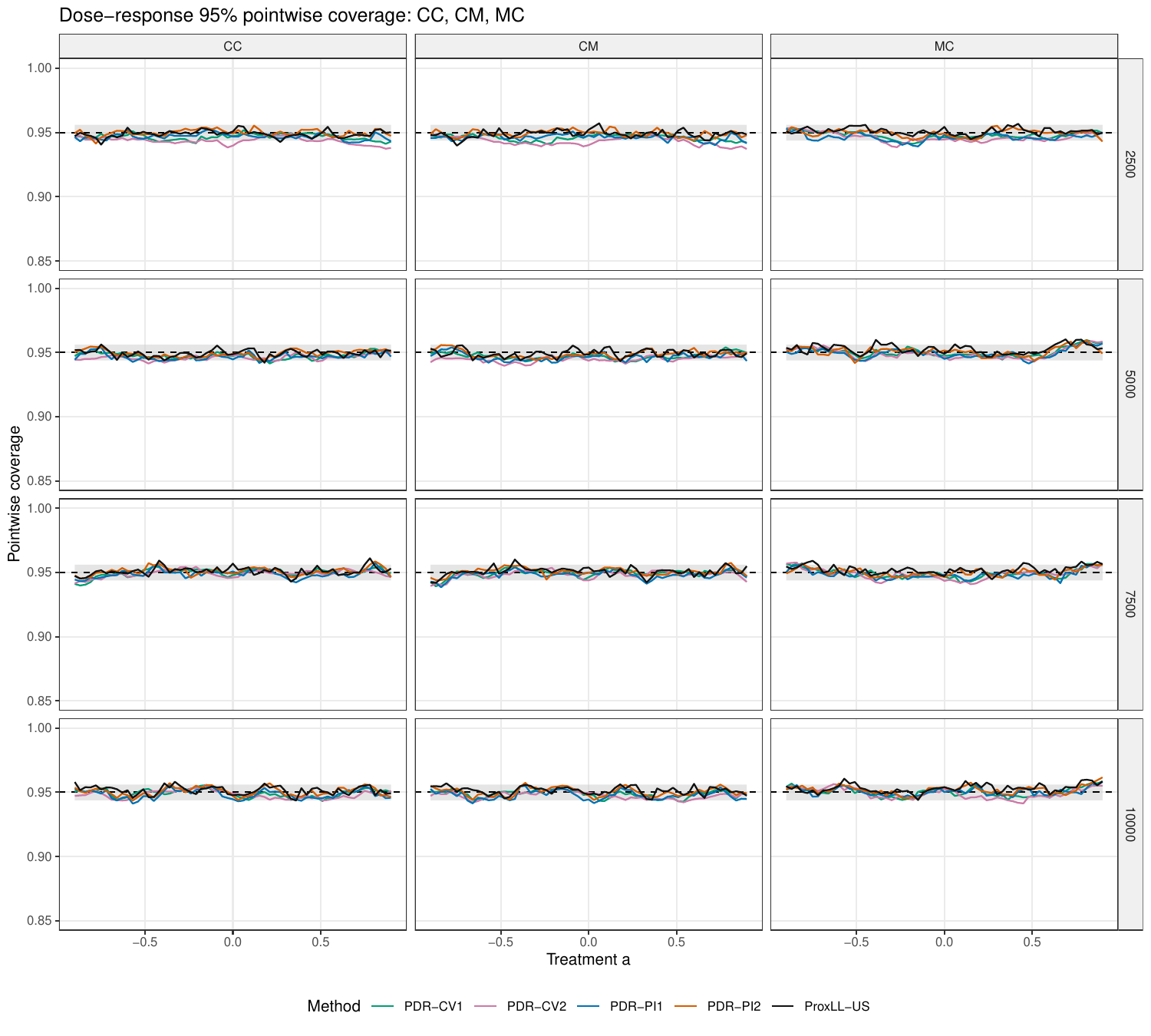}
\caption{Empirical coverage of nominal 95\% pointwise confidence intervals
for the dose-response function in CC, CM, and MC.}
\label{fig:sim-pointwise-coverage-dose}
\end{figure}

\begin{figure}[!p]
\centering
\includegraphics[width=\textwidth]{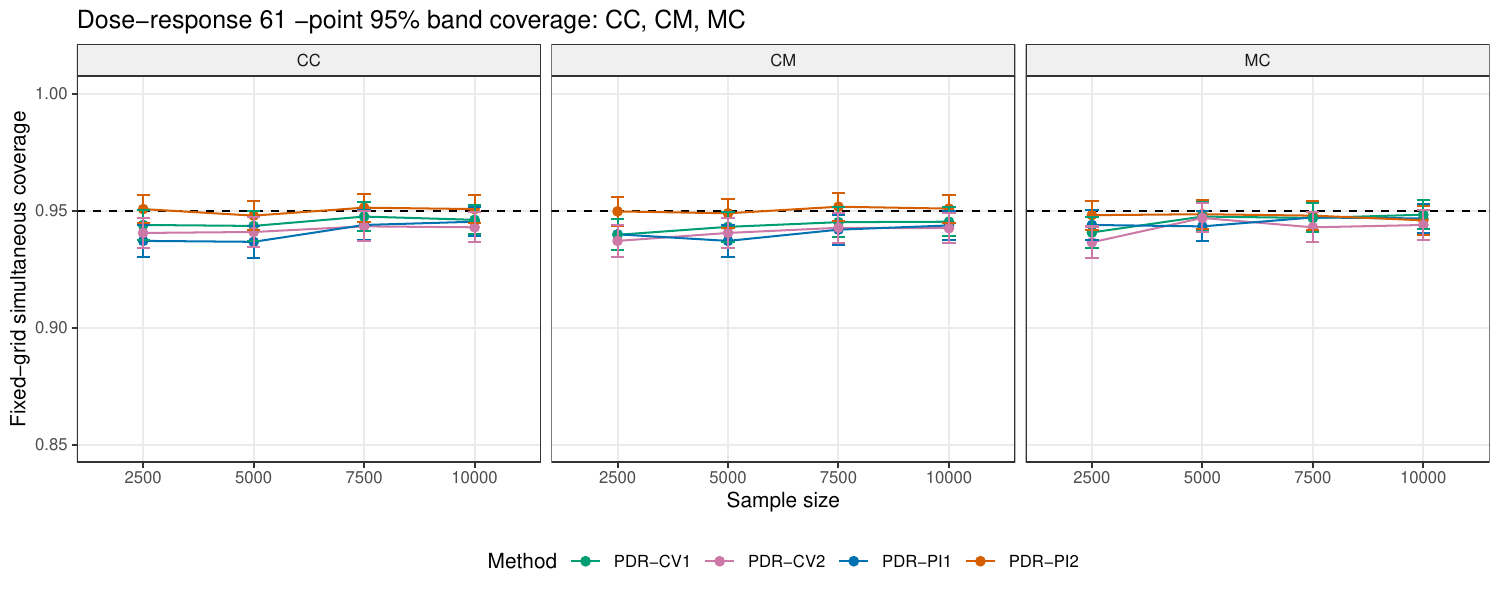}
\caption{Empirical coverage of nominal 95\% simultaneous confidence bands
for the dose-response function over the fixed 61-point grid in CC, CM, and
MC.}
\label{fig:sim-band-coverage-dose}
\end{figure}

\subsubsection{Derivative inference}
\label{sec:simulation-derivative}

As in Section~\ref{sec:simulation-dose-response}, we still consider four different bandwidth selectors. 
PDR--PI1 minimizes the estimated AMISE of the primary
degree-two slope, and PDR--PI2 uses
$1.5\widehat{\operatorname{sd}}(A)n^{-1/7}$.  PDR--CV1 and PDR--CV2 minimize
the tapered derivative-risk criterion over their
respective grids.  
ProxLQ--US uses the
degree-two local-polynomial slope.  
Its raw bandwidth is the refined minimizer
of the derivative-risk criterion, and its final bandwidth is the
raw bandwidth divided by $\log_{10}(n)$. The details for these selectors are postponed to Section~\ref{sec: derivative band selectors}.

The comparison between plug-in and CV selectors reverses for derivative
inference:
PDR--PI1 has the smallest interval widths and mean grid MSE, followed by
PDR--PI2, PDR--CV2, and PDR--CV1
(Table~\ref{tab:sim-derivative-main}).
The width and MSE plots in
Section~\ref{sec:simulation-additional-results} show much larger
relative differences between selectors than for dose-response inference.
Every PDR procedure improves on ProxLQ--US on both measures; the width advantage extends across the treatment grid.  In particular, despite producing intervals more than three times as wide as those of the CV-based PDR procedures, ProxLQ--US still has substantially larger MSE.
Then intervals again widen toward the grid endpoints; MC has wider intervals
and larger MSE than CC and CM. 
The width and MSE again decrease with $n$.
Figure~\ref{fig:sim-pointwise-coverage-derivative} shows broadly
near-nominal pointwise coverage, although ProxLQ--US undercovers more
clearly at the smallest sample size despite its wider intervals.

The plug-in bands combine substantially smaller widths with coverage
generally close to nominal
(see Figure~\ref{fig:sim-band-coverage-derivative} and
Section~\ref{sec:simulation-additional-results}).
The CV bands are much wider yet under-cover more noticeably than their
dose-response counterparts, particularly at smaller sample sizes and in
MC. Their coverage improves overall with $n$, but deficits remain in MC
at the largest sample size.

\begin{table}[!t]
\centering
\caption{Derivative inference at $n=10{,}000$.  Panel A averages over CC,
CM, MC, and the 61 grid coordinates.  Panel B reports fixed-grid simultaneous
coverage by scenario and mean width averaged over the three scenarios.}
\label{tab:sim-derivative-main}
\small
\begin{tabular}{@{}lccc@{}}
\toprule
\multicolumn{4}{c}{Panel A: Pointwise inference and estimation}\\
\midrule
Method & Coverage & Mean width & Mean grid MSE\\
\midrule
PDR--PI1    & .9523 & .3846 & .01278\\
PDR--PI2    & .9499 & .4483 & .01618\\
PDR--CV1    & .9480 & 2.0386 & .45016\\
PDR--CV2    & .9479 & 1.3035 & .17004\\
ProxLQ--US  & .9475 & 7.4556 & 5.68545\\
\bottomrule
\end{tabular}

\vspace{0.8em}

\begin{tabular}{@{}lcccc@{}}
\toprule
\multicolumn{5}{c}{Panel B: Fixed-grid simultaneous inference}\\
\midrule
Method & CC & CM & MC & Mean width\\
\midrule
PDR--PI1 & .9532 & .9534 & .9518 & .6054\\
PDR--PI2 & .9482 & .9502 & .9482 & .7120\\
PDR--CV1 & .9468 & .9450 & .9372 & 3.4342\\
PDR--CV2 & .9430 & .9452 & .9400 & 2.1769\\
\bottomrule
\end{tabular}
\end{table}

\begin{figure}[!p]
\centering
\includegraphics[width=\textwidth]{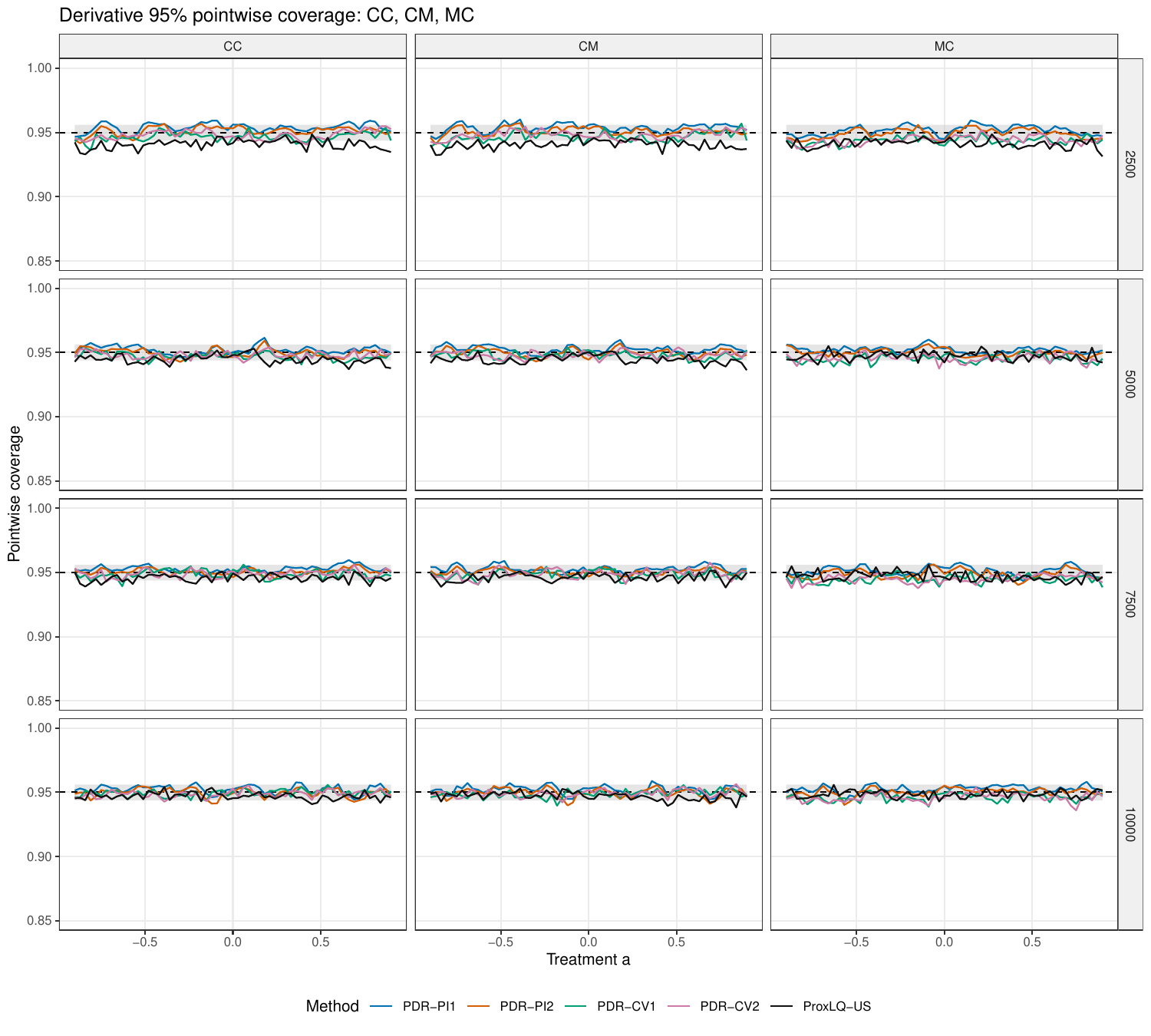}
\caption{Empirical coverage of nominal 95\% pointwise confidence intervals
for the derivative in CC, CM, and MC.}
\label{fig:sim-pointwise-coverage-derivative}
\end{figure}

\begin{figure}[!p]
\centering
\includegraphics[width=\textwidth]{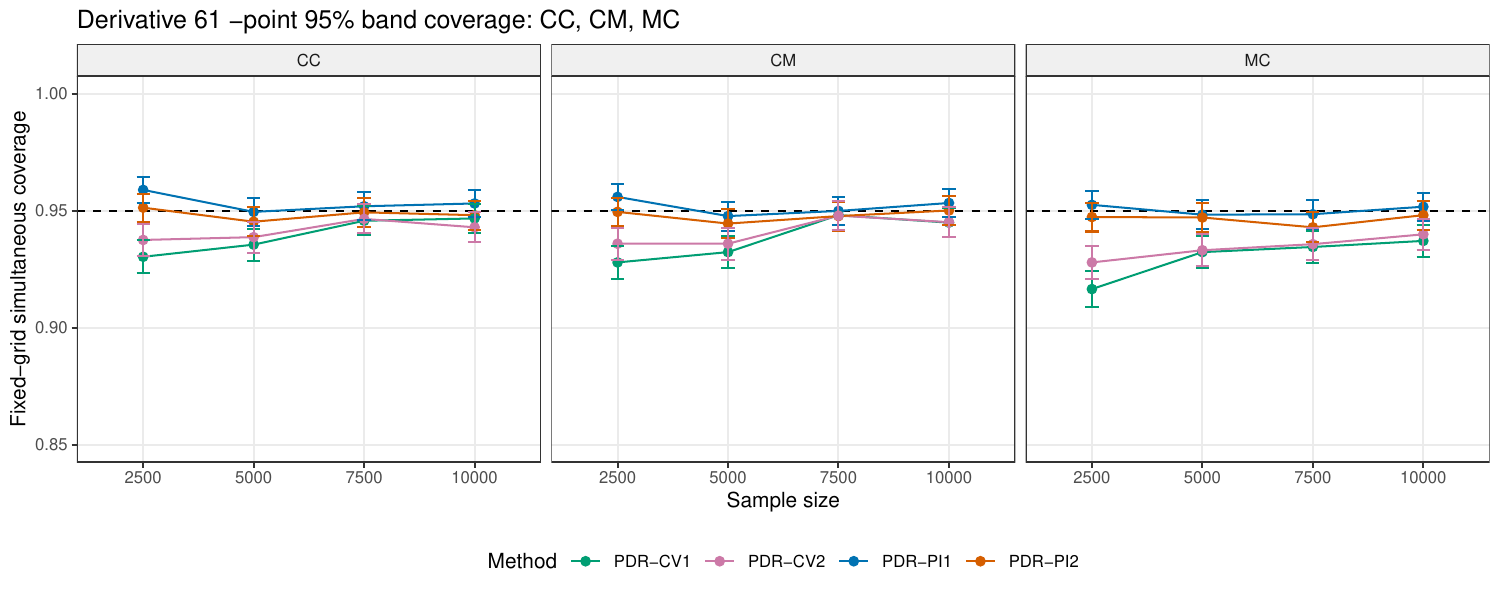}
\caption{Empirical coverage of nominal 95\% simultaneous confidence bands
for the derivative over the fixed 61-point grid in CC, CM, and MC.}
\label{fig:sim-band-coverage-derivative}
\end{figure}

\section{Real Data Analysis: Legalized Abortion and Crime}
\label{sec:realdata}

We evaluate the proposed proximal dose-response estimators on the legalized
abortion and crime application originally studied by
\citet{donohue2001impact} and later used in proximal causal learning benchmarks
\citep{mastouri2021proximal,wu2024doubly}. The data are based on a state-year
panel of 48 U.S. states observed from 1985 to 1997, yielding 624 state-year
observations. Following 
\citet{woody2020estimating}, we take the treatment \(A\) to be the effective
abortion rate and the outcome \(Y\) to be the murder rate. The
treatment-arm proxy \(Z\) is the generosity toward families with dependent
children, and the outcome-arm proxies \(W\) are beer consumption per
capita, log-prisoner population per capita, and concealed weapons laws.
Following \citet{mastouri2021proximal}, the ground truth dose-response curve
\(\beta(a)\) is obtained from a generative model fitted to the real
abortion-crime panel. Our evaluation uses \(n=4000\) observations generated from
this fitted real data design. We estimate all nuisance functions with 5-fold cross-fitting and report the four proposed estimators: PDR--PI1, PDR--PI2, PDR--CV1, and
PDR--CV2.

\begin{figure}[t]
    \centering
    \includegraphics[width=0.86\textwidth]{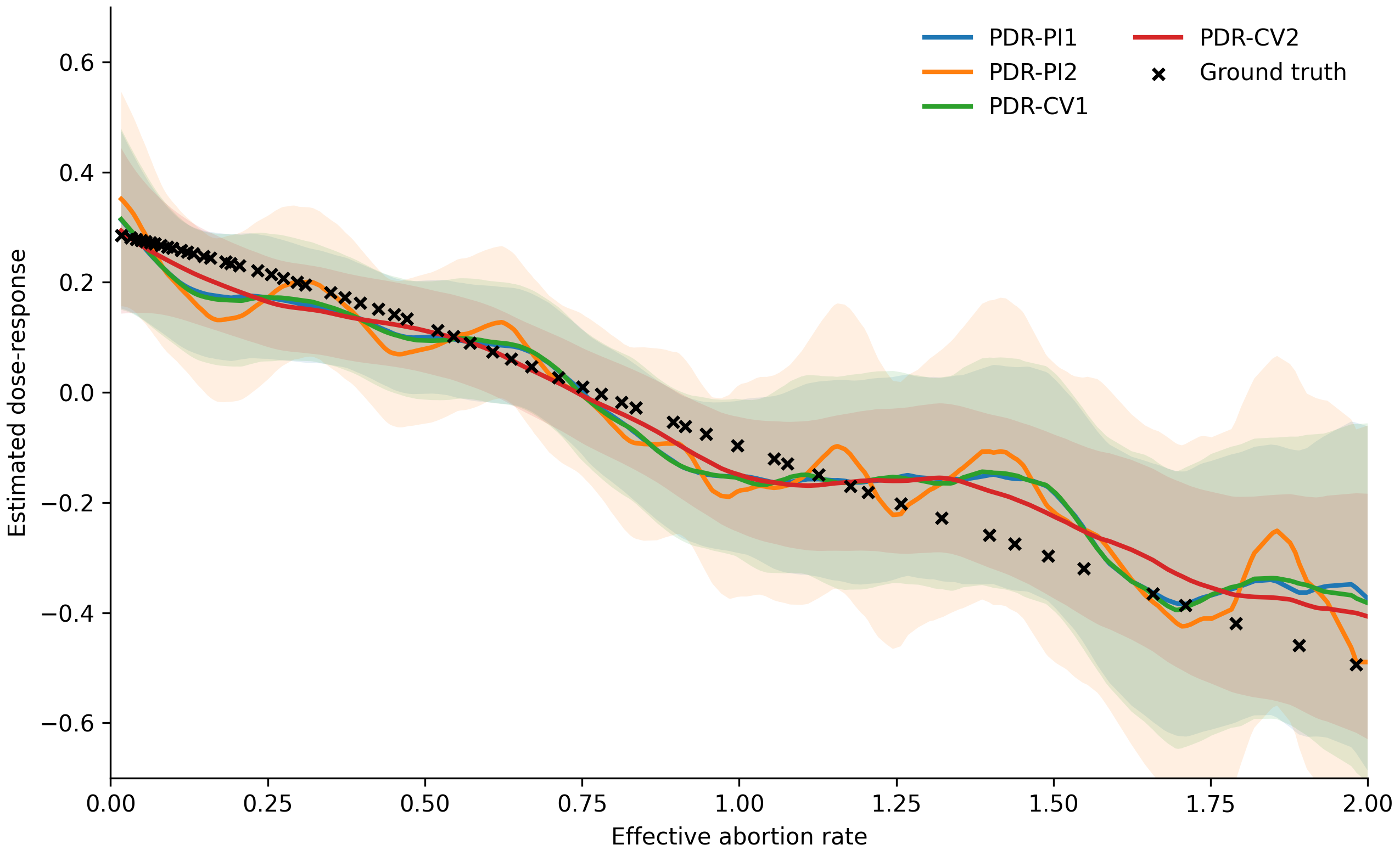}
    \caption{
    Estimated dose-response curves for the legalized abortion and crime
    application. Black crosses denote the ground-truth curve obtained from the
    fitted generative model. All four PDR estimators closely track the
    ground-truth dose-response curve over the support of the treatment.
    }
    \label{fig:abortion_crime_realdata}
\end{figure}

\begin{table}[t]
    \centering
    \caption{
    Estimation accuracy on the abortion-crime data design over
    \(a \in [0,2]\). Following \citet{mastouri2021proximal}, the MSE is
    computed on 100 equally spaced treatment values in this interval.
    }
    \label{tab:abortion_crime_realdata}
    \begin{tabular}{lccc}
        \toprule
        Estimator & MSE & RMSE & MAE \\
        \midrule
        PDR--PI1 & 0.003871 & 0.062219 & 0.050004 \\
        PDR--PI2 & 0.005140 & 0.071693 & 0.056407 \\
        PDR--CV1 & 0.003915 & 0.062571 & 0.050782 \\
        PDR--CV2 & 0.002472 & 0.049715 & 0.042384 \\
        \bottomrule
    \end{tabular}
\end{table}

The results are shown in Figure~\ref{fig:abortion_crime_realdata} and
Table~\ref{tab:abortion_crime_realdata}. All four estimators recover the
ground-truth dose-response curve well, with PDR--CV2 achieving the smallest
MSE, RMSE, and MAE. Following
 \citet{wu2024doubly}, we compute MSE over 100 equally spaced treatment values
in \(a \in [0,2]\). As a benchmark, \citet{wu2024doubly} report MSEs between
0.01 and 0.07 for several competing methods on the same abortion-crime design,
with their best-performing methods around 0.01. In comparison, our four PDR
estimators achieve MSEs between 0.0025 and 0.0052 over the same treatment
range. This comparison suggests that the proposed estimators are highly
competitive in this real-data application.

Substantively, the fitted curves suggest a negative dose-response relationship
between the effective abortion rate and subsequent murder rates. As the
effective abortion rate increases, the estimated causal response decreases. This
pattern is consistent with the mechanism proposed by
\citet{donohue2001impact}. Legalized abortion may reduce the number of unwanted
births, especially among women facing adverse socioeconomic conditions, limited
family resources, or unstable household environments. These conditions are
closely related to later-life risk factors for crime, including poverty,
reduced parental investment, lower educational attainment, and exposure to
disadvantaged neighborhoods. Therefore, higher abortion access can change the
composition of later birth cohorts by reducing the share of children born into
high-risk environments. When these cohorts reach the ages at which criminal
activity is most common, aggregate crime and murder rates may decline. Our estimated dose-response curves are consistent with this channel after
adjusting for latent confounding through proxy variables.

At larger treatment values near the boundary of the support, the uncertainty
increases, which is expected because fewer observations are available in the region with
high abortion rate. Overall, this real data experiment demonstrates that
the proposed proximal dose-response estimators remain useful in a
policy-relevant setting with latent confounding, multiple proxy variables, and a continuous-treatment.

\section{Existence of solutions to the bridge equations}
\label{sec:existence-bridges}

We give sufficient conditions for the existence of
square-integrable solutions to the bridge equations.  The argument is the
continuous-treatment analogue of the singular value decomposition argument in Appendix~B of \citet{cui2024semiparametric} and follows the Picard characterization
for compact operators in \citet[Theorem~2.41]{carrasco2007linear} and
\citet[Theorems~15.16 and~15.18]{kress2014linear}.  
The outcome and treatment
bridge equations are \eqref{eq:Bh} and \eqref{eq:Bq}, respectively.  
Note that the
right-hand side of \eqref{eq:Bq} is the generalized propensity score
\(f_A(a)/f(a\mid W,X)\).

Fix \(a\in\mathcal A\) and \(x\) in the support of \(X\).  Define
\[
T_{a,x}:L_2\{F(w\mid a,x)\}\longrightarrow L_2\{F(z\mid a,x)\}
\]
by
\[
(T_{a,x}g)(z)
:=
\mathbb E\{g(W)\mid Z=z,A=a,X=x\}.
\]
Its adjoint is $T_{a,x}^*:L_2\{F(z\mid a,x)\}\longrightarrow L_2\{F(w\mid a,x)\}$, where
\[
(T_{a,x}^*r)(w)
:=
\mathbb E\{r(Z)\mid W=w,A=a,X=x\}.
\]
Let
\[
\{\lambda_{a,x,j},\phi_{a,x,j},\psi_{a,x,j}\}_{j\geq 1}
\]
be a singular system of \(T_{a,x}\), so that
\[
T_{a,x}\phi_{a,x,j}
=
\lambda_{a,x,j}\psi_{a,x,j},
\qquad
T_{a,x}^*\psi_{a,x,j}
=
\lambda_{a,x,j}\phi_{a,x,j}.
\]
Here, \(\lambda_{a,x,j}>0\),
\(\{\phi_{a,x,j}\}_{j\geq 1}\) is an orthonormal basis of
\(N(T_{a,x})^\perp\), and
\(\{\psi_{a,x,j}\}_{j\geq 1}\) is an orthonormal basis of
\(N(T_{a,x}^*)^\perp\).

Assume that the singular systems can be selected such that, for every \(j\),
the maps
\[
(a,x)\mapsto\lambda_{a,x,j},
\qquad
(w,a,x)\mapsto\phi_{a,x,j}(w),
\qquad
(z,a,x)\mapsto\psi_{a,x,j}(z)
\]
are measurable.  Assume that \(T_{a,x}\) is compact for almost every
\((a,x)\) under the joint law of \((A,X)\).  Whenever the corresponding
conditional laws admit Lebesgue densities, a sufficient Hilbert--Schmidt
condition for compactness is
\[
\iint
f(w\mid z,a,x)f(z\mid w,a,x)\,dw\,dz
<\infty.
\]
Assume further that
\begin{equation}
\mathbb E\!\left[\mathbb E(Y\mid Z,A,X)^2\right]
<\infty
\label{eq:bridge-existence-h-square}
\end{equation}
and
\begin{equation}
\mathbb E\!\left[
\sum_{j=1}^{\infty}
\lambda_{A,X,j}^{-2}
\left|
\left\langle
\mathbb E(Y\mid Z,A,X),
\psi_{A,X,j}
\right\rangle_{z\mid A,X}
\right|^2
\right]
<\infty.
\label{eq:bridge-existence-h-picard}
\end{equation}
For the treatment bridge, assume that
\begin{equation}
\mathbb E\!\left[
\left\{
\frac{f_A(A)}{f(A\mid W,X)}
\right\}^2
\right]
<\infty
\label{eq:bridge-existence-q-square}
\end{equation}
and
\begin{equation}
\mathbb E\!\left[
\sum_{j=1}^{\infty}
\lambda_{A,X,j}^{-2}
\left|
\left\langle
\frac{f_A(A)}{f(A\mid W,X)},
\phi_{A,X,j}
\right\rangle_{w\mid A,X}
\right|^2
\right]
<\infty.
\label{eq:bridge-existence-q-picard}
\end{equation}

For almost every \((a,x)\) under the joint law of \((A,X)\), assume that, for
every square-integrable \(g\),
\begin{equation}
\mathbb E\{g(Z)\mid W,A=a,X=x\}=0\ \text{a.s.}
\quad\Longrightarrow\quad
g(Z)=0\ \text{a.s.}
\label{eq:bridge-existence-completeness-z}
\end{equation}
and
\begin{equation}
\mathbb E\{g(W)\mid Z,A=a,X=x\}=0\ \text{a.s.}
\quad\Longrightarrow\quad
g(W)=0\ \text{a.s.}
\label{eq:bridge-existence-completeness-w}
\end{equation}
The conditions in \eqref{eq:bridge-existence-completeness-z} and
\eqref{eq:bridge-existence-completeness-w} imply
\[
N(T_{a,x}^*)=\{0\}
\qquad\text{and}\qquad
N(T_{a,x})=\{0\},
\]
respectively.

For almost every \((a,x)\), \eqref{eq:bridge-existence-h-square} implies that
\(\mathbb E(Y\mid Z,A=a,X=x)\) belongs to
\(L_2\{F(z\mid a,x)\}\).  By
\eqref{eq:bridge-existence-completeness-z},
\[
\mathbb E(Y\mid Z,A=a,X=x)\in N(T_{a,x}^*)^\perp.
\]
The compactness of \(T_{a,x}\) and
\eqref{eq:bridge-existence-h-picard} therefore imply, by Picard's theorem \citep{kress2014linear},
that
\[
T_{a,x}h(\cdot,a,x)
=
\mathbb E(Y\mid Z,A=a,X=x)
\]
has a square-integrable solution.  Its minimum-norm solution is
\begin{equation}
h^\dagger(w,a,x)
=
\sum_{j=1}^{\infty}
\lambda_{a,x,j}^{-1}
\left\langle
\mathbb E(Y\mid Z,A=a,X=x),
\psi_{a,x,j}
\right\rangle_{z\mid a,x}
\phi_{a,x,j}(w).
\label{eq:bridge-existence-h-solution}
\end{equation}
By \eqref{eq:bridge-existence-h-picard}, the measurable partial sums in
\eqref{eq:bridge-existence-h-solution} are Cauchy in the \(L_2\) space of the
law of \((W,A,X)\).  Their limit has a measurable representative and satisfies
\[
\begin{aligned}
\mathbb E\{h^\dagger(W,A,X)^2\}
&=
\mathbb E\!\left[
\sum_{j=1}^{\infty}
\lambda_{A,X,j}^{-2}
\left|
\left\langle
\mathbb E(Y\mid Z,A,X),
\psi_{A,X,j}
\right\rangle_{z\mid A,X}
\right|^2
\right]<\infty.
\end{aligned}
\]
If \eqref{eq:bridge-existence-completeness-w} holds, then
\(N(T_{a,x})=\{0\}\), and the outcome bridge solution is unique.  Without
\eqref{eq:bridge-existence-completeness-w}, all square-integrable solutions
for each such \((a,x)\) are of the form
\[
h^\dagger(\cdot,a,x)+\nu(\cdot,a,x),
\qquad
\nu(\cdot,a,x)\in N(T_{a,x}).
\]

For almost every \((a,x)\), \eqref{eq:bridge-existence-q-square} implies that
\(f_A(a)/f(a\mid W,X=x)\) belongs to \(L_2\{F(w\mid a,x)\}\).  By
\eqref{eq:bridge-existence-completeness-w},
\[
\frac{f_A(a)}{f(a\mid W,X=x)}
\in N(T_{a,x})^\perp.
\]
The compactness of \(T_{a,x}^*\) and
\eqref{eq:bridge-existence-q-picard} therefore imply, by Picard's theorem,
that
\[
T_{a,x}^*q(\cdot,a,x)
=
\frac{f_A(a)}{f(a\mid W,X=x)}
\]
has a square-integrable solution.  Its minimum-norm solution is
\begin{equation}
q^\dagger(z,a,x)
=
\sum_{j=1}^{\infty}
\lambda_{a,x,j}^{-1}
\left\langle
\frac{f_A(a)}{f(a\mid W,X=x)},
\phi_{a,x,j}
\right\rangle_{w\mid a,x}
\psi_{a,x,j}(z).
\label{eq:bridge-existence-q-solution}
\end{equation}
By \eqref{eq:bridge-existence-q-picard}, the measurable partial sums in
\eqref{eq:bridge-existence-q-solution} are Cauchy in the \(L_2\) space of the
law of \((Z,A,X)\).  Their limit has a measurable representative and satisfies
\[
\begin{aligned}
\mathbb E\{q^\dagger(Z,A,X)^2\}
&=
\mathbb E\!\left[
\sum_{j=1}^{\infty}
\lambda_{A,X,j}^{-2}
\left|
\left\langle
\frac{f_A(A)}{f(A\mid W,X)},
\phi_{A,X,j}
\right\rangle_{w\mid A,X}
\right|^2
\right] <\infty.
\end{aligned}
\]
If \eqref{eq:bridge-existence-completeness-z} holds, then
\(N(T_{a,x}^*)=\{0\}\), and the treatment bridge solution is unique.  Without
\eqref{eq:bridge-existence-completeness-z}, all square-integrable solutions
for each such \((a,x)\) are of the form
\[
q^\dagger(\cdot,a,x)+\mu(\cdot,a,x),
\qquad
\mu(\cdot,a,x)\in N(T_{a,x}^*).
\]

Consequently, the displayed series define measurable functions \(h^\dagger\)
and \(q^\dagger\) satisfying
\[
\mathbb E\{h^\dagger(W,A,X)^2+q^\dagger(Z,A,X)^2\}<\infty.
\]
Moreover, Assumption~\ref{assm:basic.proxy} implies, for \(F_A\)-almost every
\(a\),
\[
f(a\mid W,X)
=
\mathbb E\{f(a\mid U,X)\mid W,X\}
\geq c_1
\quad\text{almost surely}.
\]
Therefore,
\[
\begin{aligned}
\int\mathbb E\{h^\dagger(W,a,X)^2\}\,dF_A(a)
&=
\mathbb E\left[
h^\dagger(W,A,X)^2
\frac{f_A(A)}{f(A\mid W,X)}
\right] \\
&\leq
\frac{\sup_{a\in\mathcal A}f_A(a)}{c_1}
\mathbb E\{h^\dagger(W,A,X)^2\}<\infty.
\end{aligned}
\]
The functions \(h^\dagger\) and \(q^\dagger\) satisfy
\eqref{eq:Bh} and \eqref{eq:Bq}, respectively, for almost every
\((a,x)\) under the joint law of \((A,X)\).  Hence
Assumption~\ref{assm:B} holds.

\section{Auxiliary results for identification and estimation}\label{app:aux-identification-estimation}

\subsection{Proof of an outcome-bridge identification}\label{sec: outcome bridge ID proof}
\begin{proof}[Proof of the identification result]
For $F_A$-almost every $a$, Assumptions~\ref{assm:basic.proxy} and~\ref{assm:outcome-bridge} give
\[
\begin{aligned}
0
&=
\mathbb E\{Y-h(W,a,X)\mid Z,A=a,X\}\\
&=
\mathbb E\bigl[
\mathbb E\{Y\mid U,A=a,X\}
-
\mathbb E\{h(W,a,X)\mid U,X\}
\mid Z,A=a,X
\bigr].
\end{aligned}
\]
Assumption~\ref{assm:completeness} implies
\[
\mathbb E\{Y\mid U,A=a,X\}
=
\mathbb E\{h(W,a,X)\mid U,X\}
\quad\text{almost surely}.
\]
Since $f(a\mid U,X)\ge c_1$, the preceding equality also holds under the marginal law of $(U,X)$. Therefore,
\[
\begin{aligned}
\mathbb E\{Y(a)\}
&=
\mathbb E\bigl[\mathbb E\{Y(a)\mid U,X\}\bigr]\\
&=
\mathbb E\bigl[\mathbb E\{Y(a)\mid U,A=a,X\}\bigr]\\
&=
\mathbb E\bigl[\mathbb E\{Y\mid U,A=a,X\}\bigr]\\
&=
\mathbb E\bigl[\mathbb E\{h(W,a,X)\mid U,X\}\bigr]\\
&=
\mathbb E\{h(W,a,X)\}.
\end{aligned}
\]
\end{proof}

\subsection{Auxiliary weighting identity}

\begin{lemma}\label{lem:DR_lemma}
Under Assumption~\ref{assm:basic.proxy}, let \(g\) be any integrable measurable function of \((W,X)\). Then, for
\(F_A\)-almost every \(a\),
\[
\mathbb{E}\left\{
\frac{f_A(a)}{f(a\mid W,X)}g(W,X)
\,\middle|\,
A=a
\right\}
=
\mathbb{E}\{g(W,X)\}.
\]
\end{lemma}
\begin{proof}
Assumption~\ref{assm:basic.proxy} implies
\[
f(a\mid W,X)
=
\mathbb E\{f(a\mid U,X)\mid W,X\}
\ge c_1
\]
almost surely for $F_A$-almost every $a$. Hence all expectations below are absolutely integrable.
Let \(r\) be any bounded measurable function of \(A\). By iterated expectation,
\[
\mathbb{E}\left[
r(A)\frac{f_A(A)}{f(A\mid W,X)}g(W,X)
\right]
=
\mathbb{E}\left[
g(W,X)
\,
\mathbb{E}\left\{
r(A)\frac{f_A(A)}{f(A\mid W,X)}
\,\middle|\,
W,X
\right\}
\right].
\]
Since \(f(\cdot\mid W,X)\) is a conditional density of \(A\) given \((W,X)\), we have
\[
\mathbb{E}\left\{
r(A)\frac{f_A(A)}{f(A\mid W,X)}
\,\middle|\,
W,X
\right\}
=
\int r(a)f_A(a)\,da
=
\mathbb{E}\{r(A)\}.
\]
Therefore
\[
\mathbb{E}\left[
r(A)\frac{f_A(A)}{f(A\mid W,X)}g(W,X)
\right]
=
\mathbb{E}\{g(W,X)\}\mathbb{E}\{r(A)\}.
\]
Equivalently,
\[
\mathbb{E}\left[
r(A)
\left\{
\mathbb{E}\left(
\frac{f_A(A)}{f(A\mid W,X)}g(W,X)
\,\middle|\,
A
\right)
-
\mathbb{E}\{g(W,X)\}
\right\}
\right]
=
0
\]
for every bounded measurable \(r\). Hence
\[
\mathbb{E}\left(
\frac{f_A(A)}{f(A\mid W,X)}g(W,X)
\,\middle|\,
A
\right)
=
\mathbb{E}\{g(W,X)\}
\]
almost surely, which concludes the claim.
\end{proof}

\subsection{Population proximal double robustness}\label{sec: PoP DR proximal proof}

\begin{theorem}\label{thm:DR_for_PO}
Let $(\bar q,\bar h)$ be measurable functions satisfying
\[
\int \mathbb E\{|\bar h(W,a,X)|\}\,dF_A(a)
+
\mathbb E\bigl[|\bar q(Z,A,X)\{Y-\bar h(W,A,X)\}|\bigr]
<\infty
\]
and, for $F_A$-almost every $a$,
\[
\begin{aligned}
\mathbb E\bigl[
&|\bar q(Z,a,X)|\,|Y-h(W,a,X)|\\
&+\{|\bar q(Z,a,X)|+|q(Z,a,X)|\}
 |h(W,a,X)-\bar h(W,a,X)|
\mid A=a
\bigr]<\infty.
\end{aligned}
\]
For \(F_A\)-almost every \(a\),
\[
\begin{aligned}
&\mathbb{E}\{\xi_{\mathrm{prox}}(O;\bar q,\bar h)\mid A=a\}-\theta(a)\\
&\qquad=
\mathbb E\bigl[
\{\bar q(Z,a,X)-q(Z,a,X)\}
\{h(W,a,X)-\bar h(W,a,X)\}
\mid A=a
\bigr].
\end{aligned}
\]
Consequently, $\mathbb{E}\{\xi_{\mathrm{prox}}(O;\bar q,\bar h)\mid A=a\}
=
\theta(a)$, if either \(\bar h=h\) or \(\bar q=q\).
\end{theorem}

\begin{proof}
We have
\[
\begin{aligned}
&\EE\left[\xi_{\mathrm{prox}}(O;\bar q,\bar h)\mid A=a\right]\\
&=
\E\left[\bar q(Z,A,X)\{Y-\bar h(W,A,X)\}\mid A=a\right]
+
\E\left[S_{\bar h}(A)\mid A=a\right].
\end{aligned}
\]
Therefore
\[
E\left[\xi_{\mathrm{prox}}(O;\bar q,\bar h)\mid A=a\right]
=
E\left[\bar q(Z,a,X)\{Y-\bar h(W,a,X)\}\mid A=a\right]
+
S_{\bar h}(a).
\]
Subtracting $\theta(a)=E\{h(W,a,X)\}$ from both sides yields
\[
\begin{aligned}
&
\E\left[\xi_{\mathrm{prox}}(O;\bar q,\bar h)\mid A=a\right]-\theta(a)
\\
&=
\E\left[\bar q(Z,a,X)\{Y-\bar h(W,a,X)\}\mid A=a\right]
+
S_{\bar h}(a)-\E\{h(W,a,X)\}.
\end{aligned}
\]
By $S_{\bar h}(a)=\E\{\bar h(W,a,X)\}$, we have
\[
S_{\bar h}(a)-\E\{h(W,a,X)\}
=
-\{\E[h(W,a,X)]-\E[\bar h(W,a,X)]\}.
\]
Hence
\[
\begin{aligned}
&
\E\left[\xi_{\mathrm{prox}}(O;\bar q,\bar h)\mid A=a\right]-\theta(a)
\\
&=
\E\left[\bar q(Z,a,X)\{Y-\bar h(W,a,X)\}\mid A=a\right]
-
\{\E[h(W,a,X)]-\E[\bar h(W,a,X)]\}.
\end{aligned}
\]
Now decompose
\[
Y-\bar h(W,a,X)=\{Y-h(W,a,X)\}+\{h(W,a,X)-\bar h(W,a,X)\},
\]
so that
\[
\begin{aligned}
&
\E\left[\bar q(Z,a,X)\{Y-\bar h(W,a,X)\}\mid A=a\right]
\\
&=
\E\left[\bar q(Z,a,X)\{Y-h(W,a,X)\}\mid A=a\right]
\\
&\quad
+
\E\left[\bar q(Z,a,X)\{h(W,a,X)-\bar h(W,a,X)\}\mid A=a\right].
\end{aligned}
\]
We claim that
\[
\E\left[\bar q(Z,a,X)\{Y-h(W,a,X)\}\mid A=a\right]=0.
\]
Indeed, by iterated expectation,
\[
\begin{aligned}
&
\E\left[\bar q(Z,a,X)\{Y-h(W,a,X)\}\mid A=a\right]
\\
&=
\E\left[
\E\left[\bar q(Z,a,X)\{Y-h(W,a,X)\}\mid Z,X,A=a\right]
\middle| A=a
\right].
\end{aligned}
\]
Since $\bar q(Z,a,X)$ is measurable with respect to $(Z,X)$,
\[
\begin{aligned}
&
\E\left[\bar q(Z,a,X)\{Y-h(W,a,X)\}\mid Z,X,A=a\right]
\\
&=
\bar q(Z,a,X)\,
\E\left[Y-h(W,a,X)\mid Z,X,A=a\right].
\end{aligned}
\]
By the outcome bridge equation \eqref{eq:Bh}, we get
\[
\E\left[Y-h(W,a,X)\mid Z,X,A=a\right]=0.
\]
Therefore
\[
\E\left[\bar q(Z,a,X)\{Y-h(W,a,X)\}\mid A=a\right]=0.
\]
Substituting this into the preceding display gives
\[
\begin{aligned}
&\E\left[\bar q(Z,a,X)\{Y-\bar h(W,a,X)\}\mid A=a\right]\\
&=
\E\left[\bar q(Z,a,X)\{h(W,a,X)-\bar h(W,a,X)\}\mid A=a\right].
\end{aligned}
\]
Consequently, we obtain
\[
\begin{aligned}
&
\E\left[\xi_{\mathrm{prox}}(O;\bar q,\bar h)\mid A=a\right]-\theta(a)
\\
&=
\E\left[\bar q(Z,a,X)\{h(W,a,X)-\bar h(W,a,X)\}\mid A=a\right]
-
\{\E[h(W,a,X)]-\E[\bar h(W,a,X)]\}.
\end{aligned}
\]
Next add and subtract $q(Z,a,X)$ inside the first expectation:
\[
\begin{aligned}
&
\E\left[\bar q(Z,a,X)\{h(W,a,X)-\bar h(W,a,X)\}\mid A=a\right]
\\
&=
\E\left[(\bar q(Z,a,X)-q(Z,a,X))(h(W,a,X)-\bar h(W,a,X))\mid A=a\right]
\\
&\qquad
+
\E\left[q(Z,a,X)(h(W,a,X)-\bar h(W,a,X))\mid A=a\right].
\end{aligned}
\]
Hence
\[
\begin{aligned}
&
\E\left[\xi_{\mathrm{prox}}(O;\bar q,\bar h)\mid A=a\right]-\theta(a)
\\
&=
\E\left[((\bar q(Z,a,X)-q(Z,a,X))(h(W,a,X)-\bar h(W,a,X))\mid A=a\right]
\\
&\qquad
+
\E\left[q(Z,a,X)(h(W,a,X)-\bar h(W,a,X))\mid A=a\right]
\\
&\quad
-
\{\E[h(W,a,X)]-\E[\bar h(W,a,X)]\}.
\end{aligned}
\]
It remains to simplify the last two terms. To this end, apply Lemma \ref{lem:DR_lemma} with
\[
g(W,X)=h(W,a,X)-\bar h(W,a,X).
\]
Lemma \ref{lem:DR_lemma} yields
\[
\E\left[\frac{f_A(a)}{f(a\mid W,X)}\{h(W,a,X)-\bar h(W,a,X)\}\mid A=a\right]
=
\E[h(W,a,X)]-\E[\bar h(W,a,X)].
\]
Also, by the treatment bridge equation \eqref{eq:Bq} and iterated expectation,
\[
\begin{aligned}
&
\E\left[q(Z,a,X)(h(W,a,X)-\bar h(W,a,X))\mid A=a\right]
\\
&=
\E\left[
\E\left[q(Z,a,X)\mid W,X,A=a\right]
\{h(W,a,X)-\bar h(W,a,X)\}
\middle| A=a
\right]
\\
&=
\E\left[
\frac{f_A(a)}{f(a\mid W,X)}
\{h(W,a,X)-\bar h(W,a,X)\}
\middle| A=a
\right].
\end{aligned}
\]
Combining the last two displays, we obtain
\[
\E\left[q(Z,a,X)(h(W,a,X)-\bar h(W,a,X))\mid A=a\right]
=
\E[h(W,a,X)]-\E[\bar h(W,a,X)].
\]
Therefore these terms cancel, and we conclude that
\begin{equation*}
\begin{split}
&\E\left[\xi_{\mathrm{prox}}(O;\bar q,\bar h)\mid A=a\right]-\theta(a)\\
&\qquad=
\E\left[(\bar q(Z,a,X)-q(Z,a,X))(h(W,a,X)-\bar h(W,a,X))\mid A=a\right].
\end{split}
\end{equation*}
This proves the statement.
\end{proof}

\subsection{Pathwise regularity of the smoothed target}
\label{app:pathwise-bridge-regularity}

In this section, we study pathwise differentiability of our smoothed causal target along regular parametric submodels.

The following condition justifies differentiation of the outcome-bridge
equation along regular parametric submodels.  
We remark that a similar differentiation step is
used in the proximal influence function calculations in Sections~D and G.2 of \citet{cui2024semiparametric} (the operator regularity condition used there is stated separately Appendix~G.1 therein).

For every regular one-dimensional parametric submodel
$\{P_\varepsilon:|\varepsilon|<\delta\}\subset\mathcal M^\dagger$ through
$P_0$, with score $\dot\ell_0\in L_2^0(P_0)$, there exists a function
$\dot h_0\in L_2(P_{0,WX}\otimes F_0)$ such that
\[
\int \left|\frac{h_{P_\varepsilon}(w,a,x)-h_{P_0}(w,a,x)}{\varepsilon}-\dot h_0(w,a,x)\right|^2
\,d(P_{0,WX}\otimes F_0)(w,a,x)\to 0
\]
and
\[
E_0\!\left[
I\{|A-a_0|\le \max(h,b)\}
\left\{1+q_0(Z,A,X)^2\right\}\dot h_0(W,A,X)^2
\right]<\infty,
\]
and
\[
E_0\!\left[
\begin{aligned}
&I\{|A-a_0|\le \max(h,b)\}
\{1+q_0(Z,A,X)^2\}\\[-0.2em]
&\quad\times
\left|
\frac{h_{P_\varepsilon}(W,A,X)-h_0(W,A,X)}{\varepsilon}
-\dot h_0(W,A,X)
\right|^2
\end{aligned}
\right]\to0
\]
as $\varepsilon\to0$.

The first display gives the $L_2(P_{0,WX}\otimes F_0)$ derivative of the
outcome bridge.  The second display gives the integrability of this derivative
under the weight used in the pathwise calculation, and the third gives
convergence under the same weight.  These conditions permit differentiation
of the outcome-bridge equation after multiplication by the kernel-supported
weights and by $q_0$.  The map $P\mapsto\theta_{P,h,b}(a_0)$ depends on $h_P$,
$P_{WX}$, and $F_P$, but not on $q_P$.  Consequently, we do not require pathwise derivative of $q_P$.

\begin{lemma}\label{lem:generic_P_ID}
Let $P\in\mathcal M^\dagger$. Suppose that $r:\mathcal A\to\mathbb R$ is bounded and measurable,
\[
r(a)=0\qquad\text{if }|a-a_0|>\max(h,b),
\]
and $u:\mathcal A\times\mathcal W\times\mathcal X\to\mathbb R$ is measurable with
\[
E_P\{|r(A)q_P(Z,A,X)u(A,W,X)|\}<\infty
\]
and
\[
\int |r(a)|E_P\{|u(a,W,X)|\}\,dF_P(a)<\infty.
\]
Then
\[
\mathbb E_P\{r(A)q_P(Z,A,X)u(A,W,X)\}
=
\int r(a)\mathbb E_P\{u(a,W,X)\}\,dF_P(a).
\]
\end{lemma}
\begin{proof}
By iterated expectation, we get
\[
\mathbb E_P\{r(A)q_P(Z,A,X)u(A,W,X)\}
=
\int r(a)\,
\mathbb E_P\{q_P(Z,a,X)u(a,W,X)\mid A=a\}\,dF_P(a).
\]
For \(F_P\)-almost every \(a\), one has
\begin{equation*}
\begin{split}
&\mathbb E_P\{q_P(Z,a,X)u(a,W,X)\mid A=a\}\\
&=
\mathbb E_P\!\left[
u(a,W,X)
\mathbb E_P\{q_P(Z,a,X)\mid W,A=a,X\}
\,\middle|\,A=a
\right]\\
&=
\mathbb E_P\!\left[
u(a,W,X)\frac{f_P(a)}{f_P(a\mid W,X)}
\,\middle|\,A=a
\right].
\end{split}
\end{equation*}
For $F_P$-almost every $a$ such that $r(a)\ne0$, Assumption~\ref{assm:smooth-model-positivity} and
\[
dP_{W,X\mid A=a}(w,x)
=
\frac{f_P(a\mid w,x)}{f_P(a)}\,dP_{W,X}(w,x),
\]
give
\[
\mathbb E_P\!\left[
u(a,W,X)\frac{f_P(a)}{f_P(a\mid W,X)}
\,\middle|\,A=a
\right]
=
\iint u(a,w,x)\,dP_{W,X}(w,x)
=
\mathbb E_P\{u(a,W,X)\}.
\]
Substituting this identity into the first display proves the claim.
\end{proof}

\begin{lemma}\label{lem:Fixed-weight pathwise derivative}
Let \(r:\mathcal A\to\mathbb R\) be bounded and measurable and satisfy
\[
r(a)=0\qquad\text{if }|a-a_0|>\max(h,b).
\]
Define
\[
\Psi_r(P)
:=
\int r(a)\theta_P(a)\,dF_P(a).
\]
Then, for every regular one-dimensional parametric submodel
\(\{P_\varepsilon:|\varepsilon|<\delta\}\subset\mathcal M^\dagger\) through \(P_0\)
with score \(\dot\ell_0\),
\[
\left.\frac{d}{d\varepsilon}\Psi_r(P_\varepsilon)\right|_{\varepsilon=0}
=
\mathbb E_0\!\left[
\left\{
r(A)\xi_0(O)
+
\int r(\bar a)\{h_0(W,\bar a,X)-\theta_0(\bar a)\}\,dF_0(\bar a)
\right\}
\dot\ell_0(O)
\right].
\]
\end{lemma}
\begin{proof}
Write \(h_\varepsilon:=h_{P_\varepsilon}\), \(\theta_\varepsilon:=\theta_{P_\varepsilon}\),
and let \(P_{\varepsilon,WX}\) and \(F_\varepsilon\) denote the
\((W,X)\)-marginal and \(A\)-marginal under \(P_\varepsilon\), respectively.
Then
\[
\Psi_r(P_\varepsilon)
=
\int\!\!\int r(a)h_\varepsilon(w,a,x)\,dP_{\varepsilon,WX}(w,x)\,dF_\varepsilon(a).
\]
The quadratic-mean differentiability of the submodel, the
\(L_2(P_{0,WX}\otimes F_0)\)-differentiability of \(h_\varepsilon\), and the
Cauchy--Schwarz inequality give
\[
\left.\frac{d}{d\varepsilon}\Psi_r(P_\varepsilon)\right|_{\varepsilon=0}
=
I_1+I_2+I_3,
\]
where
\[
I_1
:=
\int r(a)\mathbb E_0\{\dot h_0(W,a,X)\}\,dF_0(a),
\]
\[
I_2
:=
\mathbb E_0\!\left[
\left\{
\int r(\bar a)h_0(W,\bar a,X)\,dF_0(\bar a)
\right\}\dot\ell_0(O)
\right],
\]
and
\[
I_3
:=
\mathbb E_0\{r(A)\theta_0(A)\dot\ell_0(O)\}.
\]

It remains to identify \(I_1\). Since \(h_\varepsilon\) satisfies
\[
\mathbb E_\varepsilon\{Y-h_\varepsilon(W,A,X)\mid Z,A,X\}=0,
\]
for every $M>0$ and every \(\varepsilon\), one has
\[
0
=
\mathbb E_\varepsilon\!\left[
r(A)\bigl[\{-M\vee q_0(Z,A,X)\}\wedge M\bigr]
\{Y-h_\varepsilon(W,A,X)\}
\right].
\]
Assumption~\ref{assm:smooth-model-bridges} and the regularity condition in
this section,
quadratic-mean differentiability, and the Cauchy--Schwarz inequality give, after
differentiation at $\varepsilon=0$ and passage to the limit $M\to\infty$,
\[
0
=
\mathbb E_0\!\left[
r(A)q_0(Z,A,X)\{Y-h_0(W,A,X)\}\dot\ell_0(O)
\right]
-
\mathbb E_0\{r(A)q_0(Z,A,X)\dot h_0(W,A,X)\}.
\]
Therefore
\[
\mathbb E_0\{r(A)q_0(Z,A,X)\dot h_0(W,A,X)\}
=
\mathbb E_0\!\left[
r(A)q_0(Z,A,X)\{Y-h_0(W,A,X)\}\dot\ell_0(O)
\right].
\]
Applying Lemma~\ref{lem:generic_P_ID} with \(P=P_0\) and \(u(a,w,x)=\dot h_0(w,a,x)\), we get
\[
\mathbb E_0\{r(A)q_0(Z,A,X)\dot h_0(W,A,X)\}
=
\int r(a)\mathbb E_0\{\dot h_0(W,a,X)\}\,dF_0(a)
=
I_1.
\]
Hence
\[
I_1
=
\mathbb E_0\!\left[
r(A)q_0(Z,A,X)\{Y-h_0(W,A,X)\}\dot\ell_0(O)
\right].
\]

Combining \(I_1\), \(I_2\), and \(I_3\), we obtain
\begin{equation}
\begin{split}
\left.\frac{d}{d\varepsilon}\Psi_r(P_\varepsilon)\right|_{\varepsilon=0}
&=
\mathbb E_0\!\Bigg[
\Bigg\{
r(A)q_0(Z,A,X)\{Y-h_0(W,A,X)\}
+
r(A)\theta_0(A)
\\
&
\qquad\qquad\quad
+
\int r(\bar a)h_0(W,\bar a,X)\,dF_0(\bar a)
\Bigg\}
\dot\ell_0(O)
\Bigg].
\end{split}
\end{equation}
Since
\[
\xi_0(O)=q_0(Z,A,X)\{Y-h_0(W,A,X)\}+\theta_0(A),
\]
and
\[
\int r(\bar a)\theta_0(\bar a)\,dF_0(\bar a)
=
\Psi_r(P_0),
\]
while \(\mathbb E_0\{\dot\ell_0(O)\}=0\), the preceding display is equivalent to
\[
\left.\frac{d}{d\varepsilon}\Psi_r(P_\varepsilon)\right|_{\varepsilon=0}
=
\mathbb E_0\!\left[
\left\{
r(A)\xi_0(O)
+
\int r(\bar a)\{h_0(W,\bar a,X)-\theta_0(\bar a)\}\,dF_0(\bar a)
\right\}
\dot\ell_0(O)
\right],
\]
which proves the claim.
\end{proof}

\begin{proposition}\label{prop:theta_LL_Pathwise_derivative_part1}
Define
\[
\gamma_{P,h,a_0}(a)
:=
e_1^\top D_{P,h,a_0,1}^{-1}w_{h,a_0,1}(a)K_{h,a_0}(a)
\,w_{h,a_0,1}(a)^\top D_{P,h,a_0,1}^{-1}M_{P,h,a_0}.
\]
Then, for every regular one-dimensional parametric submodel
\(\{P_\varepsilon\}\subset\mathcal M^\dagger\) through \(P_0\),
\begin{equation*}
\begin{split}
\left.\frac{d}{d\varepsilon}\theta_{P_\varepsilon,h}^{\mathrm{LL}}(a_0)\right|_{\varepsilon=0}
&=
\mathbb E_0\!\Bigg[
\Bigg\{
r_{0,h,a_0}(A)\xi_0(O)
-
\gamma_{0,h,a_0}(A)\\
&\qquad\qquad\quad+
\int r_{0,h,a_0}(\bar a)\{h_0(W,\bar a,X)-\theta_0(\bar a)\}\,dF_0(\bar a)
\Bigg\}
\dot\ell_0(O)
\Bigg].
\end{split}
\end{equation*}
\end{proposition}
\begin{proof}
Write
\[
D_\varepsilon:=D_{P_\varepsilon,h,a_0,1},
\qquad
M_\varepsilon:=M_{P_\varepsilon,h,a_0}.
\]
Then we have
$\theta_{P_\varepsilon,h}^{\mathrm{LL}}(a_0)
=
e_1^\top D_\varepsilon^{-1}M_\varepsilon$.
Differentiating at \(\varepsilon=0\) and using the matrix identity
\[
\left.\frac{d}{d\varepsilon}D_\varepsilon^{-1}\right|_{\varepsilon=0}
=
-D_0^{-1}\dot D_0 D_0^{-1},
\]
we obtain
\[
\left.\frac{d}{d\varepsilon}\theta_{P_\varepsilon,h}^{\mathrm{LL}}(a_0)\right|_{\varepsilon=0}
=
-e_1^\top D_0^{-1}\dot D_0 D_0^{-1}M_0
+
e_1^\top D_0^{-1}\dot M_0.
\]

We first identify \(\dot D_0\). Since
\[
D_\varepsilon
=
P_\varepsilon\{w_{h,a_0,1}(A)w_{h,a_0,1}(A)^\top K_{h,a_0}(A)\},
\]
the score identity gives
\[
\dot D_0
=
\mathbb E_0\!\left[
w_{h,a_0,1}(A)w_{h,a_0,1}(A)^\top K_{h,a_0}(A)\dot\ell_0(O)
\right].
\]
Therefore, we have
\begin{equation}
\begin{split}
-e_1^\top D_0^{-1}\dot D_0 D_0^{-1}M_0
&=
-\mathbb E_0\!\left[
e_1^\top D_0^{-1}w_{h,a_0,1}(A)K_{h,a_0}(A)
w_{h,a_0,1}(A)^\top D_0^{-1}M_0
\dot\ell_0(O)
\right]\\
&=
-\mathbb E_0\{\gamma_{0,h,a_0}(A)\dot\ell_0(O)\}
\end{split}
\end{equation}

We next identify \(\dot M_0\). 
Note that each component of \(M_\varepsilon\) is a fixed-weight
functional of the form \(\Psi_r(P_\varepsilon)\) with weight equal to a component of
\(w_{h,a_0,1}(a)K_{h,a_0}(a)\). Applying Lemma \ref{lem:Fixed-weight pathwise derivative}, we get
\begin{equation*}
\begin{split}
\dot M_0
&=
\mathbb E_0\!\Bigg[
\Bigg\{
w_{h,a_0,1}(A)K_{h,a_0}(A)\xi_0(O)\\
&\qquad\qquad\quad+
\int w_{h,a_0,1}(\bar a)K_{h,a_0}(\bar a)
\{h_0(W,\bar a,X)-\theta_0(\bar a)\}\,dF_0(\bar a)
\Bigg\}
\dot\ell_0(O)
\Bigg].
\end{split}
\end{equation*}
Multiplying by \(e_1^\top D_0^{-1}\) yields
\begin{equation*}
\begin{split}
&e_1^\top D_0^{-1}\dot M_0\\
&\quad=
\mathbb E_0\!\left[
\left\{
r_{0,h,a_0}(A)\xi_0(O)
+
\int r_{0,h,a_0}(\bar a)
\{h_0(W,\bar a,X)-\theta_0(\bar a)\}\,dF_0(\bar a)
\right\}
\dot\ell_0(O)
\right].
\end{split}
\end{equation*}
Combining the last two displays proves the proposition.
\end{proof}

\begin{proposition}\label{prop:theta_LL_Pathwise_derivative_part2}
Define
\[
\gamma_{P,b,a_0}''(a)
:=
2b^{-2}e_3^\top D_{P,b,a_0,2}^{-1}w_{b,a_0,2}(a)K_{b,a_0}(a)
\,w_{b,a_0,2}(a)^\top D_{P,b,a_0,2}^{-1}M_{P,b,a_0}.
\]
Then, for every regular one-dimensional parametric submodel
\(\{P_\varepsilon\}\subset\mathcal M^\dagger\) through \(P_0\),
\begin{equation*}
\begin{split}
\left.\frac{d}{d\varepsilon}\theta_{P_\varepsilon,b}''(a_0)\right|_{\varepsilon=0}
&=
\mathbb E_0\!\Bigg[
\Bigg\{
s_{0,b,a_0}(A)\xi_0(O)
-
\gamma_{0,b,a_0}''(A)\\
&\qquad\qquad\quad+
\int s_{0,b,a_0}(\bar a)\{h_0(W,\bar a,X)-\theta_0(\bar a)\}\,dF_0(\bar a)
\Bigg\}
\dot\ell_0(O)
\Bigg].
\end{split}
\end{equation*}
\end{proposition}
\begin{proof}
Write $\widetilde D_\varepsilon:=D_{P_\varepsilon,b,a_0,2}$, and $\widetilde M_\varepsilon:=M_{P_\varepsilon,b,a_0}$.
Then we get
$\theta_{P_\varepsilon,b}''(a_0)
=
2b^{-2}e_3^\top \widetilde D_\varepsilon^{-1}\widetilde M_\varepsilon$.
By differentiating at \(\varepsilon=0\), we obtain
\[
\left.\frac{d}{d\varepsilon}\theta_{P_\varepsilon,b}''(a_0)\right|_{\varepsilon=0}
=
-2b^{-2}e_3^\top \widetilde D_0^{-1}\dot{\widetilde D}_0 \widetilde D_0^{-1}\widetilde M_0
+
2b^{-2}e_3^\top \widetilde D_0^{-1}\dot{\widetilde M}_0.
\]
The score identity gives
\[
\dot{\widetilde D}_0
=
\mathbb E_0\!\left[
w_{b,a_0,2}(A)w_{b,a_0,2}(A)^\top K_{b,a_0}(A)\dot\ell_0(O)
\right],
\]
which further indicates
\[
-2b^{-2}e_3^\top \widetilde D_0^{-1}\dot{\widetilde D}_0 \widetilde D_0^{-1}\widetilde M_0
=
-\mathbb E_0\{\gamma_{0,b,a_0}''(A)\dot\ell_0(O)\}.
\]

Note that each component of \(\widetilde M_\varepsilon\) is again a fixed-weight functional,
now with weight equal to a component of \(w_{b,a_0,2}(a)K_{b,a_0}(a)\).
Applying Lemma \ref{lem:Fixed-weight pathwise derivative}, we have
\begin{equation*}
\begin{split}
\dot{\widetilde M}_0
&=
\mathbb E_0\!\Bigg[
\Bigg\{
w_{b,a_0,2}(A)K_{b,a_0}(A)\xi_0(O)\\
&\qquad\qquad\qquad+
\int w_{b,a_0,2}(\bar a)K_{b,a_0}(\bar a)
\{h_0(W,\bar a,X)-\theta_0(\bar a)\}\,dF_0(\bar a)
\Bigg\}
\dot\ell_0(O)
\Bigg].
\end{split}
\end{equation*}
Therefore
\[
\begin{aligned}
&2b^{-2}e_3^\top \widetilde D_0^{-1}\dot{\widetilde M}_0\\
&=
\mathbb E_0\!\left[
\left\{
s_{0,b,a_0}(A)\xi_0(O)
+
\int s_{0,b,a_0}(\bar a)
\{h_0(W,\bar a,X)-\theta_0(\bar a)\}\,dF_0(\bar a)
\right\}
\dot\ell_0(O)
\right].
\end{aligned}
\]
Combining the last two displays proves the claim.
\end{proof}

\begin{proposition}\label{prop:theta_LL_Pathwise_derivative_part3}
Define
\[
U_{P,h,a_0}
:=
P\{\widetilde w_{h,a_0,1}(A)K_{h,a_0}(A)\}
\]
and
\[
\begin{aligned}
&\gamma_{P,h,a_0}^c(a)\\
&:=
e_1^\top D_{P,h,a_0,1}^{-1}
\left[
\widetilde w_{h,a_0,1}(a)
-
w_{h,a_0,1}(a)w_{h,a_0,1}(a)^\top
D_{P,h,a_0,1}^{-1}
U_{P,h,a_0}
\right]
K_{h,a_0}(a).
\end{aligned}
\]
Then, for every regular one-dimensional parametric submodel
\(\{P_\varepsilon\}\subset\mathcal M^\dagger\) through \(P_0\),
\[
\left.\frac{d}{d\varepsilon}c_{P_\varepsilon,h,a_0,2}\right|_{\varepsilon=0}
=
\mathbb E_0\{\gamma_{0,h,a_0}^c(A)\dot\ell_0(O)\}.
\]
Moreover,
\[
P\{\gamma_{P,h,a_0}^c(A)\}=0
\]
for every \(P\in \mathcal M^\dagger\).
\end{proposition}
\begin{proof}
Write
\[
D_\varepsilon:=D_{P_\varepsilon,h,a_0,1},
\qquad
U_\varepsilon:=U_{P_\varepsilon,h,a_0}.
\]
Then
$c_{P_\varepsilon,h,a_0,2}
=
e_1^\top D_\varepsilon^{-1}U_\varepsilon$.
By differentiating at \(\varepsilon=0\), we get
\[
\left.\frac{d}{d\varepsilon}c_{P_\varepsilon,h,a_0,2}\right|_{\varepsilon=0}
=
-e_1^\top D_0^{-1}\dot D_0 D_0^{-1}U_0
+
e_1^\top D_0^{-1}\dot U_0.
\]
Using the score identity,
\[
\dot D_0
=
\mathbb E_0\!\left[
w_{h,a_0,1}(A)w_{h,a_0,1}(A)^\top K_{h,a_0}(A)\dot\ell_0(O)
\right]
\]
and
\[
\dot U_0
=
\mathbb E_0\!\left[
\widetilde w_{h,a_0,1}(A)K_{h,a_0}(A)\dot\ell_0(O)
\right].
\]
Therefore, we obtain
\begin{equation*}
\begin{split}
&\left.\frac{d}{d\varepsilon}c_{P_\varepsilon,h,a_0,2}\right|_{\varepsilon=0}\\
&\qquad=
\mathbb E_0\!\left[
e_1^\top D_0^{-1}
\left\{
\widetilde w_{h,a_0,1}(A)
-
w_{h,a_0,1}(A)w_{h,a_0,1}(A)^\top D_0^{-1}U_0
\right\}
K_{h,a_0}(A)
\dot\ell_0(O)
\right],
\end{split}
\end{equation*}
which is the desired representation.
For the mean part, note that
\[
P\{\gamma_{P,h,a_0}^c(A)\}
=
e_1^\top D_{P,h,a_0,1}^{-1}
\left[
U_{P,h,a_0}
-
D_{P,h,a_0,1}D_{P,h,a_0,1}^{-1}U_{P,h,a_0}
\right]
=
0.
\]
\end{proof}

\begin{proof}[Proof of Theorem~\ref{thm:EIF of the smoothed proximal target}]
Recall
\[
\theta_{P,h,b}(a_0)
=
\theta_{P,h}^{\mathrm{LL}}(a_0)
-\frac{1}{2}h^2 c_{P,h,a_0,2}\theta_{P,b}''(a_0).
\]
Differentiating along a regular one-dimensional parametric submodel
\(\{P_\varepsilon\}\subset \mathcal M^\dagger\) through \(P_0\), we get
\begin{equation*}
\begin{split}
&\left.\frac{d}{d\varepsilon}\theta_{P_\varepsilon,h,b}(a_0)\right|_{\varepsilon=0}\\
&\quad=
\left.\frac{d}{d\varepsilon}\theta_{P_\varepsilon,h}^{\mathrm{LL}}(a_0)\right|_{\varepsilon=0}
-
\frac{1}{2}h^2
\left\{
\left.\frac{d}{d\varepsilon}c_{P_\varepsilon,h,a_0,2}\right|_{\varepsilon=0}\theta_{0,b}''(a_0)
+
c_{0,h,a_0,2}
\left.\frac{d}{d\varepsilon}\theta_{P_\varepsilon,b}''(a_0)\right|_{\varepsilon=0}
\right\}.
\end{split}
\end{equation*}
Applying Propositions~\ref{prop:theta_LL_Pathwise_derivative_part1}--\ref{prop:theta_LL_Pathwise_derivative_part3} gives
\begin{equation*}
\begin{split}
&\left.\frac{d}{d\varepsilon}\theta_{P_\varepsilon,h,b}(a_0)\right|_{\varepsilon=0}\\
&=
\mathbb E_0\!\left[
\left\{
r_{0,h,a_0}(A)\xi_0(O)
-
\gamma_{0,h,a_0}(A)
+
\int r_{0,h,a_0}(\bar a)\{h_0(W,\bar a,X)-\theta_0(\bar a)\}\,dF_0(\bar a)
\right\}
\dot\ell_0(O)
\right]\\
&\quad
-
\frac{1}{2}h^2 c_{0,h,a_0,2}
\,
\mathbb E_0\!\left[
\left\{
s_{0,b,a_0}(A)\xi_0(O)-
\gamma_{0,b,a_0}''(A)\right.\right.\\
&\left.\left.\qquad\qquad\qquad\qquad\qquad+
\int s_{0,b,a_0}(\bar a)\{h_0(W,\bar a,X)-\theta_0(\bar a)\}\,dF_0(\bar a)
\right\}
\dot\ell_0(O)
\right]\\
&\quad
-
\frac{1}{2}h^2 \theta_{0,b}''(a_0)
\,
\mathbb E_0\{\gamma_{0,h,a_0}^c(A)\dot\ell_0(O)\}.
\end{split}
\end{equation*}
Collecting terms, we get
\begin{equation}
\begin{split}
&\left.\frac{d}{d\varepsilon}\theta_{P_\varepsilon,h,b}(a_0)\right|_{\varepsilon=0}\\
&=
\mathbb E_0\!\left[
\left\{
\left(
r_{0,h,a_0}(A)-\frac{1}{2}h^2 c_{0,h,a_0,2}s_{0,b,a_0}(A)
\right)\xi_0(O)
\right.
\right.\\
&
\left.
\left.
\quad
-
\left(
\gamma_{0,h,a_0}(A)
-\frac{1}{2}h^2 c_{0,h,a_0,2}\gamma_{0,b,a_0}''(A)
+\frac{1}{2}h^2\theta_{0,b}''(a_0)\gamma_{0,h,a_0}^c(A)
\right)
\right.
\right.\\
&\left.
\left.
\quad
+
\int
\left(
r_{0,h,a_0}(\bar a)
-\frac{1}{2}h^2 c_{0,h,a_0,2}s_{0,b,a_0}(\bar a)
\right)
\{h_0(W,\bar a,X)-\theta_0(\bar a)\}\,dF_0(\bar a)
\right\}
\dot\ell_0(O)
\right].
\end{split}
\end{equation}
By definition of \(\Gamma_{0,h,b,a_0}\) and \(\gamma_{0,h,b,a_0}\), we have
\[
\left.\frac{d}{d\varepsilon}\theta_{P_\varepsilon,h,b}(a_0)\right|_{\varepsilon=0}
=
\mathbb E_0\{\varphi_{0,h,b,a_0}(O)\dot\ell_0(O)\}.
\]
Thus \(\varphi_{0,h,b,a_0}\) is a pathwise gradient of \(\theta_{P,h,b}(a_0)\) at \(P_0\).

It remains to verify that \(\varphi_{0,h,b,a_0}\) has mean zero. First, because
\(\mathbb E_0\{\xi_0(O)\mid A\}=\theta_0(A)\),
\[
\mathbb E_0\{\Gamma_{0,h,b,a_0}(A)\xi_0(O)\}
=
\mathbb E_0\{\Gamma_{0,h,b,a_0}(A)\theta_0(A)\}
=
\theta_{0,h,b}(a_0).
\]
Second,
\[
\begin{aligned}
&\mathbb E_0\!\left[
\int \Gamma_{0,h,b,a_0}(\bar a)\{h_0(W,\bar a,X)-\theta_0(\bar a)\}\,dF_0(\bar a)
\right]\\
&=
\int \Gamma_{0,h,b,a_0}(\bar a)
\{\theta_0(\bar a)-\theta_0(\bar a)\}\,dF_0(\bar a)
=
0.
\end{aligned}
\]
Third,
\begin{equation*}
\begin{split}
\mathbb E_0\{\gamma_{0,h,a_0}(A)\}
&=
e_1^\top D_{0,h,a_0,1}^{-1}
P_0\{w_{h,a_0,1}(A)K_{h,a_0}(A)w_{h,a_0,1}(A)^\top\}
D_{0,h,a_0,1}^{-1}M_{0,h,a_0}\\
&=
e_1^\top D_{0,h,a_0,1}^{-1}M_{0,h,a_0},
\end{split}
\end{equation*}
which is equal to $\theta_{0,h}^{\mathrm{LL}}(a_0)$.
Likewise,
\begin{equation*}
\begin{split}
\mathbb E_0\{\gamma_{0,b,a_0}''(A)\}
&=
2b^{-2}e_3^\top D_{0,b,a_0,2}^{-1}
P_0\{w_{b,a_0,2}(A)K_{b,a_0}(A)w_{b,a_0,2}(A)^\top\}
D_{0,b,a_0,2}^{-1}M_{0,b,a_0}\\
&=
2b^{-2}e_3^\top D_{0,b,a_0,2}^{-1}M_{0,b,a_0}\\
&=
\theta_{0,b}''(a_0),
\end{split}
\end{equation*}
and, by Proposition \ref{prop:theta_LL_Pathwise_derivative_part3}, we obtain
\[
\mathbb E_0\{\gamma_{0,h,a_0}^c(A)\}=0.
\]
Therefore, we obtain
\[
\mathbb E_0\{\gamma_{0,h,b,a_0}(A)\}
=
\theta_{0,h}^{\mathrm{LL}}(a_0)
-
\frac{1}{2}h^2 c_{0,h,a_0,2}\theta_{0,b}''(a_0)
+
\frac{1}{2}h^2 \theta_{0,b}''(a_0)\cdot 0
=
\theta_{0,h,b}(a_0).
\]
Combining the three mean identities, we get
\[
\mathbb E_0\{\varphi_{0,h,b,a_0}(O)\}=0.
\]

Assumption~\ref{assm:smooth-model-bridges}, the regularity condition in
Appendix~\ref{app:pathwise-bridge-regularity}, and the
bounded support of the fixed kernel weights imply
$\varphi_{0,h,b,a_0}\in L_2^0(P_0)$.
The derivative identity and the preceding mean identity prove the theorem.
\end{proof}

\subsection{Equivalent representations of the cross-fitted estimator}
\label{sec:DB-cf-equivalence}

\begin{lemma}\label{lem:DB-cf-equivalence}
For every $a_0$, the two right-hand sides in
\eqref{eq:DB-cf-representations} and the right-hand side in
\eqref{def:DB_cf_estimation} are equal.
\end{lemma}
\begin{proof}
For each fold $k$, the definition of $\hat m_{n,k}^{(-k)}$ gives
\[
\begin{aligned}
P_{n,k}\{\Gamma_{n,h_n,b_n,a_0}(A)\hat m_{n,k}^{(-k)}(A)\}
&=
\frac{1}{n_k}\sum_{j\in I_k}
\Gamma_{n,h_n,b_n,a_0}(A_j)
\left\{
\frac{1}{n_k}\sum_{i\in I_k}
\hat h_n^{(-k)}(W_i,A_j,X_i)
\right\}\\
&=
\frac{1}{n_k^2}
\sum_{i\in I_k}\sum_{j\in I_k}
\Gamma_{n,h_n,b_n,a_0}(A_j)
\hat h_n^{(-k)}(W_i,A_j,X_i)\\
&=
(Q_{n,k}^{WX}\times F_{n,k})
\{\Gamma_{n,h_n,b_n,a_0}(A)\hat h_n^{(-k)}(W,A,X)\}.
\end{aligned}
\]
Since
$\hat\xi_n^{(-k)}=\hat\xi_{n,1}^{(-k)}+\hat m_{n,k}^{(-k)}$, it follows that
\[
\begin{aligned}
P_{n,k}\{\Gamma_{n,h_n,b_n,a_0}(A)\hat\xi_n^{(-k)}(O)\}
={}&
P_{n,k}\{\Gamma_{n,h_n,b_n,a_0}(A)\hat\xi_{n,1}^{(-k)}(O)\}\\
&+
(Q_{n,k}^{WX}\times F_{n,k})
\{\Gamma_{n,h_n,b_n,a_0}(A)\hat h_n^{(-k)}(W,A,X)\}.
\end{aligned}
\]
This proves the equality of the two right-hand sides in
\eqref{eq:DB-cf-representations}.  Moreover,
\[
\begin{aligned}
\frac{n_k}{n}
&P_{n,k}\{\Gamma_{n,h_n,b_n,a_0}(A)\hat\xi_{n,1}^{(-k)}(O)\}\\
&=
\frac{1}{n}\sum_{i\in I_k}
\Gamma_{n,h_n,b_n,a_0}(A_i)
\hat q_n^{(-k)}(Z_i,A_i,X_i)
\{Y_i-\hat h_n^{(-k)}(W_i,A_i,X_i)\},\\
\frac{n_k}{n}
&P_{n,k}\{\Gamma_{n,h_n,b_n,a_0}(A)\hat m_{n,k}^{(-k)}(A)\}
=
\frac{1}{n n_k}
\sum_{i\in I_k}\sum_{j\in I_k}
\Gamma_{n,h_n,b_n,a_0}(A_j)
\hat h_n^{(-k)}(W_i,A_j,X_i).
\end{aligned}
\]
Summing these two identities over $k=1,\ldots,K$ yields
\eqref{def:DB_cf_estimation}.
\end{proof}

\section{Auxiliary results for debiased inference}\label{app:debiased-inference-aux}

\subsection{Conditions for uniform inference}\label{sec:uniform-inference-conditions}

\begin{assumptionp}{G}\label{assump:regularity}\label{assm:G}
\par\noindent
Let $\mathcal A_0$ be a compact subset of the interior of $\mathcal A$.  For
$\delta>0$, define $\mathcal A_\delta
:=
\{a\in\mathbb R:|a-a_0|\le\delta
\text{ for some }a_0\in\mathcal A_0\}$.
There exists $\delta_3>0$ such that
$\mathcal A_{\delta_3}\subset\mathcal A$ and the following conditions hold.
\begin{enumerate}[label=\textup{(G\arabic*)},ref=G\arabic*,leftmargin=*,itemsep=0.35em]
\item\label{assm:G-bandwidth}
The bandwidth sequences are deterministic, the limit in
Assumption~\ref{assm:bandwidths} satisfies $\tau\in(0,\infty)$, and $nh_n^5=O(1)$, $nh_n^\kappa\to\infty$, for some $\kappa\in(1,5)$.

\item\label{assm:G-crossfit}
Assumption~\ref{assm:folds} holds.  There exist finite constants $C_q$ and
$C_h$ such that, with probability tending to one, $\max_{1\le k\le K}
\sup_{\substack{a\in\mathcal A_{\delta_3}\\ z,x}}
|\hat q_n^{(-k)}(z,a,x)|
\le C_q$, and $\max_{1\le k\le K}
\sup_{\substack{a\in\mathcal A_{\delta_3}\\ w,x}}
|\hat h_n^{(-k)}(w,a,x)|
\le C_h$.

\item\label{assm:G-DR-rates}
The functions $q_\infty$ and $h_\infty$ satisfy
$\sup_{\substack{a\in\mathcal A_{\delta_3}\\z,x}}
|q_\infty(z,a,x)|\le C_q$, and 
$\sup_{\substack{a\in\mathcal A_{\delta_3}\\w,x}}
|h_\infty(w,a,x)|\le C_h$.
There exist measurable sets
$\mathcal S_1',\mathcal S_2',\mathcal S_3'$ in
$\mathcal A_{\delta_3}\times\mathcal Z\times\mathcal W\times\mathcal X$
such that, for $f_A$-almost every $a\in\mathcal A_{\delta_3}$, the sets
$\{(z,w,x):(a,z,w,x)\in\mathcal S_j'\}$, $j\in\{1,2,3\}$, form a partition
up to $P_0(\cdot\mid A=a)$-null sets and
\[
\begin{aligned}
&I_{\mathcal S_1'\cup\mathcal S_3'}(a,Z,W,X)
\{q_\infty(Z,a,X)-q_0(Z,a,X)\}=0,\\
&I_{\mathcal S_2'\cup\mathcal S_3'}(a,Z,W,X)
\{h_\infty(W,a,X)-h_0(W,a,X)\}=0
\end{aligned}
\]
$P_0(\cdot\mid A=a)$-almost surely.  Moreover,
\[
\sup_{a\in\mathcal A_{\delta_3}}
E_0\!\left[
\begin{aligned}
&\{q_\infty(Z,a,X)-q_0(Z,a,X)\}^2\\[-0.2em]
&\quad+\{h_\infty(W,a,X)-h_0(W,a,X)\}^2
\end{aligned}
\,\middle|\,A=a
\right]
<\infty.
\]
The nuisance rates satisfy
\[
\begin{aligned}
\max_{1\le k\le K}
d_{q,k}(\hat q_n^{(-k)},q_\infty;
\mathcal A_{\delta_3},\mathcal S_1')
&=o_p\{(nh_n\log n)^{-1/2}\},\\
\max_{1\le k\le K}
d_{h,k}(\hat h_n^{(-k)},h_\infty;
\mathcal A_{\delta_3},\mathcal S_2')
&=o_p\{(nh_n\log n)^{-1/2}\},
\end{aligned}
\]
and
\[
\begin{aligned}
\max_{1\le k\le K}
&d_{q,k}(\hat q_n^{(-k)},q_\infty;
\mathcal A_{\delta_3},\mathcal S_3')
d_{h,k}(\hat h_n^{(-k)},h_\infty;
\mathcal A_{\delta_3},\mathcal S_3')
=o_p\{(nh_n\log n)^{-1/2}\}.
\end{aligned}
\]

\item\label{assm:G-nuisance-L2}
The foldwise conditional $L_2$ errors satisfy
\[
\begin{aligned}
\max_{1\le k\le K}
\sup_{a\in\mathcal A_{\delta_3}}
&
\left[
E_0\!\left\{
(\hat q_n^{(-k)}(Z,a,X)-q_\infty(Z,a,X))^2
\,\middle|\,A=a,\trainingsigma_k
\right\}
\right]^{1/2}\\[-0.2em]
&\quad\times\sqrt{\log(h_n^{-1})\log n}
=o_p(1),\\
\max_{1\le k\le K}
\sup_{a\in\mathcal A_{\delta_3}}
&
\left[
E_0\!\left\{
(\hat h_n^{(-k)}(W,a,X)-h_\infty(W,a,X))^2
\,\middle|\,A=a,\trainingsigma_k
\right\}
\right]^{1/2}\\[-0.2em]
&\quad\times\sqrt{\log(h_n^{-1})\log n}
=o_p(1).
\end{aligned}
\]

\item\label{assm:regul2}\label{assm:G-smooth}
The following conditions hold on $\mathcal A_{\delta_3}$:
\begin{enumerate}[label=\textup{(\roman*)},leftmargin=2.2em,itemsep=0.25em]
\item $\theta_0$ is twice continuously differentiable, and $\theta_0''$ is
H\"older continuous with exponent $\alpha\in(0,1]$;
\item $f_A$ is Lipschitz continuous and
$0<\inf_{a\in\mathcal A_{\delta_3}}f_A(a)
\le\sup_{a\in\mathcal A_{\delta_3}}f_A(a)<\infty$;
\item $|Y|$ is $P_0$-almost surely bounded; and
\item $\sigma_0^2$ is continuous on
$\mathcal A_{\delta_3}$ and
$\inf_{a\in\mathcal A_0}\sigma_0^2(a)>0$.
\end{enumerate}

\item\label{assm:G-grid}
The set $\mathcal A_n\subset\mathcal A_0$ is nonempty, deterministic, and
finite.  Its mesh size $\omega_n
:=
\sup_{a_0\in\mathcal A_0}
\inf_{a\in\mathcal A_n}|a_0-a|$
satisfies $\omega_n=o(h_n^\rho)$ for some $\rho>1$.
\end{enumerate}
\end{assumptionp}

Assumption~\ref{assm:G-bandwidth} places $b_n$ and $h_n$ on the same
asymptotic scale and permits $h_n\asymp n^{-1/5}$.  Together with
Assumption~\ref{assm:G-smooth}\textup{(i)}, the robust bias correction then
leaves a uniform smoothing remainder of order $O(h_n^{2+\alpha})$, which is
$o\{(nh_n\log n)^{-1/2}\}$.  The two bandwidth conditions also imply
$\log(h_n^{-1})=O(\log n)$,
$\log(h_n^{-1})\{\log n/(nh_n)\}^{1/2}\to0$, and
$h_n\log(h_n^{-1})\log n\to0$.
Assumptions~\ref{assm:G-DR-rates} and
\ref{assm:G-nuisance-L2} are the uniform counterparts of the pointwise
proximal nuisance conditions in \ref{assm:DR_and_rates} and \ref{assm:pointwise-pseudooutcome-rate}. 
The strengthened
rates make the resulting product-bias remainder
$o_p\{(nh_n\log n)^{-1/2}\}$ uniformly over $\mathcal A_0$.
In contrast to non-cross-fitting, again our approach avoids the uniform entropy restrictions on the nuisance-estimator class, including the stronger
outcome-regression restriction used in \citet{Takatsu2025debiased} to control the empirical $U$-process, and the associated entropy--bandwidth restriction in condition (f) therein.

\subsection{Multiple-$a$ inference}\label{sec: multiple a inference proofs}

Fix an exposure value $a_0$ in the interior of the support of $A$.
Let $(q_{\infty},h_{\infty})$ be a limiting nuisance pair and define
\[
m_{\infty}(a)
:=
E_0\{h_{\infty}(W,a,X)\},
\]
and
\[
\xi_{\infty}(O)
:=
q_{\infty}(Z,A,X)\{Y-h_{\infty}(W,A,X)\}
+
m_{\infty}(A).
\]
Under Assumptions~\ref{assm:A}, \ref{assm:B}, and~\ref{assm:F}, the proximal double-robust identity implies that
\[
E_0\{\xi_{\infty}(O)\mid A=a\}=\theta_0(a)
\]
for $f_A$-almost every $a\in B_{\delta_1}(a_0)$.

For exposure value $a$, define
\[
\sigma_0^2(a)
:=
E_0\left(
[\xi_{\infty}(O)-\theta_0(A)]^2
\mid A=a
\right).
\]

Also define
\[
\begin{aligned}
\phi_{\infty,a_0}(O)
:={}&
\Gamma_{0,h_n,b_n,a_0}(A)\xi_{\infty}(O)
-\gamma_{0,h_n,b_n,a_0}(A)\\
&\qquad+
\int
\Gamma_{0,h_n,b_n,a_0}(\bar a)
\{h_{\infty}(W,\bar a,X)-m_{\infty}(\bar a)\}\,dF_0(\bar a),
\end{aligned}
\]
and
\[
\begin{aligned}
\psi_{n,a_0}^{sm}(O)
:={}&
\Gamma_{0,h_n,b_n,a_0}(A)\theta_0(A)
-\gamma_{0,h_n,b_n,a_0}(A)\\
&\qquad+
\int
\Gamma_{0,h_n,b_n,a_0}(\bar a)
\{h_{\infty}(W,\bar a,X)-m_{\infty}(\bar a)\}\,dF_0(\bar a).
\end{aligned}
\]

For $j \geq 0$, define kernel moments
\[
c_j := \int u^j K(u)\,du,
\qquad
c_j^{\ast} := \int u^j K(u)^2\,du,
\qquad
c_{j,\tau}^{\ast} := \int u^j K(u)K(\tau u)\,du,
\]
and matrices
\[
S_2 :=
\begin{pmatrix}
c_0 & c_1 \\
c_1 & c_2
\end{pmatrix},
\qquad
S_3 :=
\begin{pmatrix}
c_0 & c_1 & c_2 \\
c_1 & c_2 & c_3 \\
c_2 & c_3 & c_4
\end{pmatrix}.
\]
Recall
\[
V_{K,\tau}
:=
\int
\left\{
K(u)-\tau^3 c_2
\frac{(\tau u)^2-c_2}{c_4-c_2^2}
K(\tau u)
\right\}^2\,du.
\]

The assumptions used in the results below are Assumptions~\ref{assm:A}, \ref{assm:B}, and~\ref{assm:D}--\ref{assm:F} for the pointwise and finite-dimensional statements, and Assumptions~\ref{assm:A}, \ref{assm:B}, \ref{assm:D}, and~\ref{assm:G} for the uniform statements.

\subsection{Remainder decomposition}
\label{sec: remainder decomposition}

Following the decomposition in Section D of
\citet{Takatsu2025debiased}, write
\[
\Gamma_n(a):=\Gamma_{n,h_n,b_n,a_0}(a),\qquad
\Gamma_0(a):=\Gamma_{0,h_n,b_n,a_0}(a),\qquad
\gamma_0(a):=\gamma_{0,h_n,b_n,a_0}(a),
\]
and recall that \(\tau_n=h_n/b_n\).  Let \(Q_0^{WX}\) denote the
\(P_0\)-law of \((W,X)\).  Whenever \(P_0\) is applied to a
sample-dependent function, the observed sample is held fixed and \(P_0\)
denotes integration with respect to an independent observation
\(O\sim P_0\).

Define
\[
R_{n,1}(a_0)
:=
\theta_{0,h_n,b_n}(a_0)-\theta_0(a_0).
\]

The nuisance empirical-process remainder is
\[
\begin{aligned}
&R_{n,2}(a_0)
:=
\sum_{k=1}^K\frac{n_k}{n}(P_{n,k}-P_0)
\Bigg[
\Gamma_0(A)
\left\{
\hat\xi_{n,1}^{(-k)}(O)
+
S_{\hat h_n^{(-k)}}(A)
\right\} \\
&\quad
+
\int
\Gamma_0(\bar a)
\hat h_n^{(-k)}(W,\bar a,X)
\,dF_0(\bar a) 
-
\Gamma_0(A)\xi_\infty(O)
-
\int
\Gamma_0(\bar a)
h_\infty(W,\bar a,X)
\,dF_0(\bar a)
\Bigg].
\end{aligned}
\]

The empirical-process remainder for the estimated local-polynomial weight is
\[
\begin{aligned}
R_{n,3}(a_0)
:=
\sum_{k=1}^K\frac{n_k}{n}(P_{n,k}-P_0)
\Bigg[
&\{\Gamma_n(A)-\Gamma_0(A)\}
\left\{
\hat\xi_{n,1}^{(-k)}(O)
+
S_{\hat h_n^{(-k)}}(A)
\right\} \\
&\qquad+
\int
\{\Gamma_n(\bar a)-\Gamma_0(\bar a)\}
\hat h_n^{(-k)}(W,\bar a,X)
\,dF_0(\bar a)
\Bigg].
\end{aligned}
\]

The proximal doubly robust remainder is
\[
\begin{aligned}
R_{n,4}(a_0)
:=
\sum_{k=1}^K\frac{n_k}{n}
\int
\Gamma_n(a)
E_0\Big[
&\{\hat q_n^{(-k)}(Z,a,X)-q_0(Z,a,X)\} \\
&\quad\times
\{h_0(W,a,X)-\hat h_n^{(-k)}(W,a,X)\}
\,\Big|\,A=a
\Big]
\,dF_0(a).
\end{aligned}
\]

The product empirical-measure remainder is
\[
\begin{aligned}
R_{n,5}(a_0)
:=
\sum_{k=1}^K\frac{n_k}{n}
\iint
&\Gamma_n(a)\hat h_n^{(-k)}(w,a,x)
d(Q_{n,k}^{WX}-Q_0^{WX})(w,x)
\,d(F_{n,k}-F_0)(a).
\end{aligned}
\]

Finally, define the local-polynomial matrix remainder by
\[
\begin{aligned}
R_{n,6}(a_0)
&:=
e_1^\top D_{0,h_n,a_0,1}^{-1}
(D_{0,h_n,a_0,1}-D_{n,h_n,a_0,1})
(D_{n,h_n,a_0,1}^{-1}-D_{0,h_n,a_0,1}^{-1})
\\
&\qquad\times
P_0\!\left\{
w_{h_n,a_0,1}(A)K_{h_n,a_0}(A)\theta_0(A)
\right\}
\\
&\quad-
c_{0,h_n,a_0,2}\tau_n^2
e_3^\top D_{0,b_n,a_0,2}^{-1}
(D_{0,b_n,a_0,2}-D_{n,b_n,a_0,2})
(D_{n,b_n,a_0,2}^{-1}-D_{0,b_n,a_0,2}^{-1})
\\
&\qquad\times
P_0\!\left\{
w_{b_n,a_0,2}(A)K_{b_n,a_0}(A)\theta_0(A)
\right\}
\\
&\quad-
\tau_n^2e_1^\top
(D_{n,h_n,a_0,1}^{-1}-D_{0,h_n,a_0,1}^{-1})
(P_n-P_0)
\left\{
\widetilde w_{h_n,a_0,1}(A)K_{h_n,a_0}(A)
\right\}
\\
&\qquad\times
e_3^\top D_{0,b_n,a_0,2}^{-1}
P_0\!\left\{
w_{b_n,a_0,2}(A)K_{b_n,a_0}(A)\theta_0(A)
\right\}
\\
&\quad-
\tau_n^2e_1^\top
D_{0,h_n,a_0,1}^{-1}
(D_{0,h_n,a_0,1}-D_{n,h_n,a_0,1})
(D_{n,h_n,a_0,1}^{-1}-D_{0,h_n,a_0,1}^{-1})
\\
&\qquad\times
P_0\!\left\{
\widetilde w_{h_n,a_0,1}(A)K_{h_n,a_0}(A)
\right\} e_3^\top D_{0,b_n,a_0,2}^{-1}
P_0\!\left\{
w_{b_n,a_0,2}(A)K_{b_n,a_0}(A)\theta_0(A)
\right\}
\\
&\quad-
\tau_n^2
(c_{n,h_n,a_0,2}-c_{0,h_n,a_0,2})
e_3^\top
(D_{n,b_n,a_0,2}^{-1}-D_{0,b_n,a_0,2}^{-1})
\\
&\qquad\times
P_0\!\left\{
w_{b_n,a_0,2}(A)K_{b_n,a_0}(A)\theta_0(A)
\right\}.
\end{aligned}
\]

\begin{lemma}
\label{lem:prox-six-rem-decomp}
For all sufficiently large $n$, on the event that the empirical moment
matrices are nonsingular,
\[
\hat\theta_n^{DB,cf}(a_0)-\theta_0(a_0)
=
P_n\phi_{\infty,a_0}
+
\sum_{j=1}^6R_{n,j}(a_0).
\]
\end{lemma}
\begin{proof}
Recall \(P_0\phi_{\infty,a_0}=0\). 
We now expand the estimator fold by fold.  Fix \(k\).  Bilinearity of
product measures gives
\[
\begin{aligned}
Q_{n,k}^{WX}\times F_{n,k}
={}&
Q_0^{WX}\times F_0
+
Q_0^{WX}\times(F_{n,k}-F_0) \\
&\qquad+
(Q_{n,k}^{WX}-Q_0^{WX})\times F_0
+
(Q_{n,k}^{WX}-Q_0^{WX})\times(F_{n,k}-F_0).
\end{aligned}
\]
The first three components satisfy
\[
\begin{aligned}
&(Q_0^{WX}\times F_0)
\{\Gamma_n(A)\hat h_n^{(-k)}(W,A,X)\}  =
P_0\left[
\Gamma_n(A)S_{\hat h_n^{(-k)}}(A)
\right],
\end{aligned}
\]
\[
\begin{aligned}
&\{Q_0^{WX}\times(F_{n,k}-F_0)\}
\{\Gamma_n(A)\hat h_n^{(-k)}(W,A,X)\} =
(P_{n,k}-P_0)
\left[
\Gamma_n(A)S_{\hat h_n^{(-k)}}(A)
\right],
\end{aligned}
\]
and
\[
\begin{aligned}
&\{(Q_{n,k}^{WX}-Q_0^{WX})\times F_0\}
\{\Gamma_n(A)\hat h_n^{(-k)}(W,A,X)\} \\
&\qquad =
(P_{n,k}-P_0)
\left[
\int
\Gamma_n(\bar a)
\hat h_n^{(-k)}(W,\bar a,X)
\,dF_0(\bar a)
\right].
\end{aligned}
\]
It follows that the \(k\)th fold contribution to the estimator equals
\[
\begin{aligned}
&
P_{n,k}
\left\{
\Gamma_n(A)\hat\xi_{n,1}^{(-k)}(O)
\right\}
+
(Q_{n,k}^{WX}\times F_{n,k})
\left\{
\Gamma_n(A)\hat h_n^{(-k)}(W,A,X)
\right\}
\\
&=
(P_{n,k}-P_0)
\Bigg[
\Gamma_n(A)
\left\{
\hat\xi_{n,1}^{(-k)}(O)
+
S_{\hat h_n^{(-k)}}(A)
\right\}
+
\int
\Gamma_n(\bar a)
\hat h_n^{(-k)}(W,\bar a,X)
\,dF_0(\bar a)
\Bigg]
\\
&\qquad+
P_0\left[
\Gamma_n(A)
\left\{
\hat\xi_{n,1}^{(-k)}(O)
+
S_{\hat h_n^{(-k)}}(A)
\right\}
\right]
\\
&\qquad+
\iint
\Gamma_n(a)\hat h_n^{(-k)}(w,a,x)
\,d(Q_{n,k}^{WX}-Q_0^{WX})(w,x)
\,d(F_{n,k}-F_0)(a).
\end{aligned}
\]
Multiplying by \(n_k/n\) and summing over \(k\), the last line gives
\(R_{n,5}(a_0)\).

Since \(P_0\phi_{\infty,a_0}=0\), we may write
\[
\begin{aligned}
P_n\phi_{\infty,a_0}
&=
\sum_{k=1}^K\frac{n_k}{n}(P_{n,k}-P_0)
\Bigg[
\Gamma_0(A)\xi_\infty(O)
+
\int
\Gamma_0(\bar a)h_\infty(W,\bar a,X)
\,dF_0(\bar a)
\Bigg]
\\
&\quad
-
(P_n-P_0)\{\gamma_0(A)\}.
\end{aligned}
\]
Indeed, the omitted term $\int\Gamma_0(\bar a)m_\infty(\bar a)\,dF_0(\bar a)$
is constant and is therefore cancelled by \(P_n-P_0\).

Subtracting this representation from the preceding fold expansion and then
adding and subtracting the same fold-specific expression with \(\Gamma_n\)
replaced by \(\Gamma_0\) gives
\[
\begin{aligned}
\hat\theta_n^{DB,cf}(a_0)-\theta_0(a_0)-P_n\phi_{\infty,a_0}
&=
R_{n,2}(a_0)+R_{n,3}(a_0)+R_{n,5}(a_0)
\\
&\qquad+
\sum_{k=1}^K\frac{n_k}{n}
P_0\left[
\Gamma_n(A)
\left\{
\hat\xi_{n,1}^{(-k)}(O)
+
S_{\hat h_n^{(-k)}}(A)
\right\}
\right]
\\
&\qquad+
(P_n-P_0)\{\gamma_0(A)\}
-\theta_0(a_0).
\end{aligned}
\]

By the proximal product-bias identity in
Theorem~\ref{thm:DR_for_PO}, applied pathwise with
\[
\bar q=\hat q_n^{(-k)},
\qquad
\bar h=\hat h_n^{(-k)},
\]
we have, for \(F_0\)-almost every \(a\),
\[
\begin{aligned}
&E_0\left[
\hat\xi_{n,1}^{(-k)}(O)
+
S_{\hat h_n^{(-k)}}(A)
\,\middle|\,A=a
\right]
-\theta_0(a)
\\
&\qquad =
E_0\Big[
\{\hat q_n^{(-k)}(Z,a,X)-q_0(Z,a,X)\}
\{h_0(W,a,X)-\hat h_n^{(-k)}(W,a,X)\}
\,\Big|\,A=a
\Big].
\end{aligned}
\]
Consequently,
\[
\begin{aligned}
&\sum_{k=1}^K\frac{n_k}{n}
P_0\left[
\Gamma_n(A)
\left\{
\hat\xi_{n,1}^{(-k)}(O)
+
S_{\hat h_n^{(-k)}}(A)
\right\}
\right]
 =
P_0\{\Gamma_n(A)\theta_0(A)\}
+
R_{n,4}(a_0).
\end{aligned}
\]
It follows that
\[
\begin{aligned}
\hat\theta_n^{DB,cf}(a_0)-\theta_0(a_0)-P_n\phi_{\infty,a_0}
&=
R_{n,2}(a_0)+R_{n,3}(a_0)+R_{n,4}(a_0)+R_{n,5}(a_0)
\\
&\qquad+
(P_n-P_0)\{\gamma_0(A)\}
+
P_0\{\Gamma_n(A)\theta_0(A)\}
-
\theta_0(a_0).
\end{aligned}
\]
Adding and subtracting
\[
P_0\{\Gamma_0(A)\theta_0(A)\}
=
\theta_{0,h_n,b_n}(a_0)
\]
shows that the last line equals
\[
R_{n,1}(a_0)
+
(P_n-P_0)\{\gamma_0(A)\}
+
P_0\!\left[
\{\Gamma_n(A)-\Gamma_0(A)\}\theta_0(A)
\right].
\]

It remains to identify the last two terms.  Direct substitution of the
definitions of \(\Gamma_n,\Gamma_0\), and \(\gamma_0\) gives
\[
\begin{aligned}
&P_0\!\left[
\{\Gamma_n(A)-\Gamma_0(A)\}\theta_0(A)
\right]\\
&=
e_1^\top
(D_{n,h_n,a_0,1}^{-1}-D_{0,h_n,a_0,1}^{-1})
P_0\!\left\{
w_{h_n,a_0,1}(A)K_{h_n,a_0}(A)\theta_0(A)
\right\}
\\
&\quad-
\tau_n^2
\left\{
c_{n,h_n,a_0,2}e_3^\top D_{n,b_n,a_0,2}^{-1}
-
c_{0,h_n,a_0,2}e_3^\top D_{0,b_n,a_0,2}^{-1}
\right\}
P_0\!\left\{
w_{b_n,a_0,2}(A)K_{b_n,a_0}(A)\theta_0(A)
\right\},
\end{aligned}
\]
whereas
\[
\begin{aligned}
&(P_n-P_0)\{\gamma_0(A)\}\\
&=
e_1^\top D_{0,h_n,a_0,1}^{-1}
(D_{n,h_n,a_0,1}-D_{0,h_n,a_0,1})
D_{0,h_n,a_0,1}^{-1}
P_0\!\left\{
w_{h_n,a_0,1}(A)K_{h_n,a_0}(A)\theta_0(A)
\right\}
\\
&\qquad-
c_{0,h_n,a_0,2}\tau_n^2
e_3^\top D_{0,b_n,a_0,2}^{-1}
(D_{n,b_n,a_0,2}-D_{0,b_n,a_0,2})
D_{0,b_n,a_0,2}^{-1}
P_0\!\left\{
w_{b_n,a_0,2}(A)K_{b_n,a_0}(A)\theta_0(A)
\right\}
\\
&\qquad+
\tau_n^2e_1^\top D_{0,h_n,a_0,1}^{-1}
\Big[
(P_n-P_0)
\{\widetilde w_{h_n,a_0,1}(A)K_{h_n,a_0}(A)\}\\
&\qquad\qquad\qquad\qquad\qquad\quad-
(D_{n,h_n,a_0,1}-D_{0,h_n,a_0,1})
D_{0,h_n,a_0,1}^{-1}
P_0\{\widetilde w_{h_n,a_0,1}(A)K_{h_n,a_0}(A)\}
\Big]
\\
&\qquad\qquad\times
e_3^\top D_{0,b_n,a_0,2}^{-1}
P_0\!\left\{
w_{b_n,a_0,2}(A)K_{b_n,a_0}(A)\theta_0(A)
\right\}.
\end{aligned}
\]
Using
$D_n^{-1}-D_0^{-1}
=
D_0^{-1}(D_0-D_n)D_n^{-1}$
for each of the two local-polynomial moment matrices, together with
\[
\begin{aligned}
c_{n,h_n,a_0,2}-c_{0,h_n,a_0,2}
={}&
e_1^\top
(D_{n,h_n,a_0,1}^{-1}-D_{0,h_n,a_0,1}^{-1})
P_n\{\widetilde w_{h_n,a_0,1}K_{h_n,a_0}\}
\\
&\qquad+
e_1^\top D_{0,h_n,a_0,1}^{-1}
(P_n-P_0)
\{\widetilde w_{h_n,a_0,1}K_{h_n,a_0}\},
\end{aligned}
\]
and
\[
\begin{aligned}
&e_1^\top D_{0,h_n,a_0,1}^{-1}\\
&\times \Big[
(P_n-P_0)\{\widetilde w_{h_n,a_0,1}K_{h_n,a_0}\}-
(D_{n,h_n,a_0,1}-D_{0,h_n,a_0,1})D_{0,h_n,a_0,1}^{-1}
P_0\{\widetilde w_{h_n,a_0,1}K_{h_n,a_0}\}
\Big]
\\
&\quad-
(c_{n,h_n,a_0,2}-c_{0,h_n,a_0,2})
\\
&=-e_1^\top
(D_{n,h_n,a_0,1}^{-1}-D_{0,h_n,a_0,1}^{-1})
(P_n-P_0)\{\widetilde w_{h_n,a_0,1}K_{h_n,a_0}\}
\\
&\qquad-
e_1^\top D_{0,h_n,a_0,1}^{-1}
(D_{0,h_n,a_0,1}-D_{n,h_n,a_0,1})
(D_{n,h_n,a_0,1}^{-1}-D_{0,h_n,a_0,1}^{-1})
P_0\{\widetilde w_{h_n,a_0,1}K_{h_n,a_0}\},
\end{aligned}
\]
and collecting terms yields
\[
(P_n-P_0)\{\gamma_0(A)\}
+
P_0\!\left[
\{\Gamma_n(A)-\Gamma_0(A)\}\theta_0(A)
\right]
=
R_{n,6}(a_0).
\]
Combining the preceding identities proves the claimed decomposition.
\end{proof}

\begin{lemma}\label{lem:D-expansion}
If conditions \ref{assm:kernel}, \ref{assm:bandwidths}, and \ref{assm:regul1}\textup{(b)} hold, then
\[
   D_{0,h,a_0,1}^{-1}
  = f_A(a_0)^{-1} S_2^{-1} + O(h),
  \quad
  c_{0,h,a_0,2}
  = c_2 + O(h)
\]
as $h\to 0$, and
\[
  D_{0,b,a_0,2}^{-1}
  = f_A(a_0)^{-1} S_3^{-1} + O(b)
\]
as $b\to 0$.

If conditions \ref{assm:kernel}, \ref{assm:bandwidths}, and \ref{assm:G-smooth}\textup{(ii)} hold, then
\[
  \sup_{a_0\in\mathcal A_0}
  \left\|
    D_{0,h,a_0,1}^{-1}
    - f_A(a_0)^{-1} S_2^{-1}
  \right\|_\infty
  = O(h),
\]
\[
  \sup_{a_0\in\mathcal A_0}
  \bigl|
    c_{0,h,a_0,2} - c_2
  \bigr|
  = O(h),
\]
and
\[
  \sup_{a_0\in\mathcal A_0}
  \left\|
    D_{0,b,a_0,2}^{-1}
    - f_A(a_0)^{-1}S_3^{-1}
  \right\|_\infty
  = O(b)
\]
as $h,b\to 0$.
\end{lemma}
\noindent Proof is given in Lemma 4 of \citet{Takatsu2025debiased}, so we omit the technical details here.

\begin{lemma}\label{lem:single_a_leading_term}
Fix $a_0$ in the interior of the support of $A$. Assume Assumptions~\ref{assm:A}, \ref{assm:B}, and~\ref{assm:D}--\ref{assm:F}. Then
\[
(nh_n)^{1/2}
P_n\left[
\Gamma_{0,h_n,b_n,a_0}(A)\{\xi_{\infty}(O)-\theta_0(A)\}
\right]
\rightsquigarrow
N\left(
0,
\frac{V_{K,\tau}\sigma_0^2(a_0)}{f_A(a_0)}
\right).
\]
\end{lemma}
\noindent Recall 
\[
\phi_{\infty,a_0}(O)
=
\Gamma_{0,h_n,b_n,a_0}(A)\{\xi_\infty(O)-\theta_0(A)\}
+
\psi_{n,a_0}^{sm}(O).
\]
\begin{lemma}\label{lem:smoothing_negligible_single_a}
Fix $a_0$ in the interior of the support of $A$. Under Assumptions~\ref{assm:A}, \ref{assm:B}, and~\ref{assm:D}--\ref{assm:F}, it holds that
\[
P_0\psi_{n,a_0}^{sm}=0,
\]
and
\[
(nh_n)^{1/2}
\left|
P_n\phi_{\infty,a_0}
-
P_n\left[
\Gamma_{0,h_n,b_n,a_0}(A)\{\xi_\infty(O)-\theta_0(A)\}
\right]
\right|
=
(nh_n)^{1/2}|P_n\psi_{n,a_0}^{sm}|
\to_p 0.
\]
\end{lemma}

\begin{lemma}\label{lem:multiple_a_leading_term}
Fix distinct points $a_1,\ldots,a_m$ in the interior of the support of $A$.
Assume Assumptions~\ref{assm:A}, \ref{assm:B}, and~\ref{assm:D}--\ref{assm:F} hold at each $a_k$ for the same deterministic pair $(q_\infty,h_\infty)$. Then
\[
(nh_n)^{1/2}
\begin{pmatrix}
P_n\left[
\Gamma_{0,h_n,b_n,a_1}(A)\{\xi_{\infty}(O)-\theta_0(A)\}
\right]
\\
\vdots
\\
P_n\left[
\Gamma_{0,h_n,b_n,a_m}(A)\{\xi_{\infty}(O)-\theta_0(A)\}
\right]
\end{pmatrix}
\]
converges in distribution to a mean-zero multivariate normal vector with diagonal
covariance matrix
\[
\operatorname{diag}
\left(
\frac{V_{K,\tau}\sigma_0^2(a_1)}{f_A(a_1)},
\ldots,
\frac{V_{K,\tau}\sigma_0^2(a_m)}{f_A(a_m)}
\right).
\]
\end{lemma}
\noindent We omit the proofs of the preceding three lemmas, as the corresponding arguments are provided in the proofs of Lemmas 5 and 6 in \citet{Takatsu2025debiased}.

\subsection{Asymptotic confidence band}

We write $h=h_n$ and $b=b_n$ in this section.
\begin{lemma}\label{lem:limiting_variance}
Under Assumptions~\ref{assm:A}, \ref{assm:B},
\ref{assm:D}, and~\ref{assm:G},
\[
\sup_{a_0\in\mathcal A_0}
\left|
\sigma_{\infty,h,b}^2(a_0)
-
\frac{V_{K,\tau}\sigma_0^2(a_0)}{f_A(a_0)}
\right|
\longrightarrow0.
\]
Consequently, there exist constants $0<c<C<\infty$ such that, for all
sufficiently large $n$,
\[
c\le \inf_{a_0\in\mathcal A_0}\sigma_{\infty,h,b}^2(a_0)
\le \sup_{a_0\in\mathcal A_0}\sigma_{\infty,h,b}^2(a_0)\le C.
\]
\end{lemma}

\begin{proof}
The decomposition
\[
\phi_{\infty,h,b,a_0}
=\Gamma_{0,h,b,a_0}(A)\{\xi_\infty-\theta_0(A)\}
 +\psi_{n,a_0}^{sm}
\]
and the uniform version of the two bounds used in
Lemma~\ref{lem:smoothing_negligible_single_a} imply
\(\sup_{a_0}P_0(\psi_{n,a_0}^{sm})^2=O(1)\).  Since the localized first
term has squared norm \(O(h^{-1})\), the Cauchy--Schwarz inequality yields
\[
\sup_{a_0\in\mathcal A_0}\left|
\sigma_{\infty,h,b}^2(a_0)
-h\int\Gamma_{0,h,b,a_0}(a)^2\sigma_0^2(a)f_A(a)\,da
\right|=O(h^{1/2}).
\]
Uniformly over \(a_0\in\mathcal A_0\) and bounded \(u\),
Lemma~\ref{lem:D-expansion} gives
\[
h\Gamma_{0,h,b,a_0}(a_0+hu)
=\frac{1}{f_A(a_0)}
\left\{K(u)-\tau_n^3c_2
\frac{(\tau_nu)^2-c_2}{c_4-c_2^2}K(\tau_nu)\right\}+o(1).
\]
The functions have a common compact support.  Uniform continuity of
\(\sigma_0^2f_A\), followed by dominated convergence, proves the claimed
uniform limit.  We note that this uses a similar arguments underlying Lemmas~5 and~22 of
\citet{Takatsu2025debiased}, with
\(\mu_\infty(a,W)\) replaced by \(h_\infty(W,a,X)\).

Finally, the equivalent kernel in the definition of \(V_{K,\tau}\) has
integral one.  Indeed, for \(\tau>0\), the change of variables
\(v=\tau u\) gives
\[
\int \tau^3c_2
\frac{(\tau u)^2-c_2}{c_4-c_2^2}K(\tau u)\,du
=\frac{\tau^2c_2}{c_4-c_2^2}
\int(v^2-c_2)K(v)\,dv=0,
\]
and the correction is zero when \(\tau=0\).  Hence that kernel cannot
vanish in \(L_2(du)\), so \(V_{K,\tau}>0\).  The lower and upper bounds now
follow from Assumption~\ref{assm:G-smooth}\textup{(ii), (iv)}.
\end{proof}

\begin{lemma}\label{lem:leading_term_L_k_bound}
Under Assumptions~\ref{assm:A}, \ref{assm:B},
\ref{assm:D}, and~\ref{assm:G}, for every integer $k\ge1$,
\[
\sup_{a_0\in\mathcal A_0}
P_0|h\phi_{\infty,h,b,a_0}|^k\lesssim h.
\]
Equivalently,
\[
\sup_{a_0\in\mathcal A_0}
P_0|\phi_{\infty,h,b,a_0}|^k\lesssim h^{-(k-1)}.
\]
In particular,
\[
\sup_{a_0\in\mathcal A_0}
P_0\left|
\frac{h^{1/2}\phi_{\infty,h,b,a_0}}
{\sigma_{\infty,h,b}(a_0)}
\right|^3
\lesssim h^{-1/2}.
\]
\end{lemma}

\begin{lemma}\label{lem:limiting_gaussian_process}
Under Assumptions~\ref{assm:A}, \ref{assm:B},
\ref{assm:D}, and~\ref{assm:G}, the mean-zero Gaussian process
$Z_{\infty,h,b}$ is tight in $\ell^\infty(\mathcal A_0)$ and
\[
E_0\left\{\sup_{a_0\in\mathcal A_0}|Z_{\infty,h,b}(a_0)|\right\}
\le C\{\log(h^{-1})\}^{1/2}.
\]
For every $\delta>0$,
\[
\begin{aligned}
&E_0\left[
\sup_{\substack{u,v\in\mathcal A_0\\|u-v|<\delta}}
|Z_{\infty,h,b}(u)-Z_{\infty,h,b}(v)|
\right]\\
&\le
Ch^{-1/2}(\delta\wedge h)^{1/2}
\left[
1+\log\left\{
1+\frac{\operatorname{diam}(\mathcal A_0)}
{\delta\wedge h}
\right\}
\right]^{1/2}.
\end{aligned}
\]
\end{lemma}

\begin{lemma}\label{lem:leading_term_gaussian_approx}
Under Assumptions~\ref{assm:A}, \ref{assm:B},
\ref{assm:D}, and~\ref{assm:G}, define
\[
\eta_{h,b,a_0}
:=
\frac{h^{1/2}\phi_{\infty,h,b,a_0}}
{\sigma_{\infty,h,b}(a_0)},
\qquad
\mathcal H_{h,b}
:=
\{\eta_{h,b,a_0}:a_0\in\mathcal A_0\},
\]
and $\mathcal F_{h,b}:=\mathcal H_{h,b}\cup(-\mathcal H_{h,b})$.
The class $\mathcal F_{h,b}$ is pointwise measurable and VC--type, with
characteristics independent of $n$, and has an envelope $F_{h,b}$ satisfying
\[
\|F_{h,b}\|_{P_0,3}\le Ch^{-1/2},
\qquad
\|F_{h,b}\|_{P_0,4}\le Ch^{-1/2}.
\]
There exists a random variable
\[
\widetilde Z_{n,h,b}
\overset{d}{=}
\sup_{a_0\in\mathcal A_0}|Z_{\infty,h,b}(a_0)|
\]
such that
\[
\left|
\sup_{a_0\in\mathcal A_0}|G_n\eta_{h,b,a_0}|
-\widetilde Z_{n,h,b}
\right|
=O_p(r_n),
\]
where
\[
r_n
=
(nh)^{-1/6}\log n
+
(nh)^{-1/4}(\log n)^{5/4}
+
(nh)^{-1/2}(\log n)^{3/2}.
\]
Moreover,
\[
\sup_{t\in\mathbb R}
\left|
P_0\left\{
\sup_{a_0\in\mathcal A_0}|G_n\eta_{h,b,a_0}|\le t
\right\}
-
P_0\left\{
\sup_{a_0\in\mathcal A_0}|Z_{\infty,h,b}(a_0)|\le t
\right\}
\right|
\longrightarrow0.
\]
\end{lemma}
\noindent We omit the proofs of the preceding three lemmas, as similar arguments are provided in the proofs of Lemmas 23--25 in \citet{Takatsu2025debiased}.

\begin{corollary}
Under the conditions of the preceding lemma,
\[
\sup_{t\in\mathbb R}
\left|
P_0\left(
\sup_{a_0\in\mathcal A_0}
\left|
\frac{(nh)^{1/2}P_n\phi_{\infty,h,b,a_0}}
{\sigma_{\infty,h,b}(a_0)}
\right|
\le t
\right)
-
P_0\left(
\sup_{a_0\in\mathcal A_0}|Z_{\infty,h,b}(a_0)|\le t
\right)
\right|
\to0.
\]
\end{corollary}

\begin{proof}
The centering identity established before the Gaussian coupling gives
\[
G_n\eta_{h,b,a_0}
=\frac{(nh)^{1/2}P_n\phi_{\infty,h,b,a_0}}
 {\sigma_{\infty,h,b}(a_0)}.
\]
The conclusion is the second conclusion of
Lemma~\ref{lem:leading_term_gaussian_approx}.
\end{proof}

\begin{lemma}\label{lem:finite-mesh-gaussian}
Under the conditions of Lemma~\ref{lem:limiting_gaussian_process}, let
$\mathcal A_n\subset\mathcal A_0$ be a finite mesh and define
\[
\omega_n
:=
\sup_{a_0\in\mathcal A_0}\inf_{a\in\mathcal A_n}|a_0-a|.
\]
If $\omega_n=o(h^p)$ for some $p>1$, then
\[
\sup_{t\in\mathbb R}
\left|
P_0\left\{
\sup_{a_0\in\mathcal A_0}|Z_{\infty,h,b}(a_0)|\le t
\right\}
-
P_0\left\{
\max_{a\in\mathcal A_n}|Z_{\infty,h,b}(a)|\le t
\right\}
\right|
\longrightarrow0.
\]
\end{lemma}
\noindent We omit the proof of the preceding lemma, since it closely follows the argument used for the proofs of Lemmas 26 and 28 in \citet{Takatsu2025debiased}.

\subsection{Proofs of the core inference statements}

\begin{proof}[Proof of Theorem~\ref{thm:pointwise-proximal-inference}]
By Lemma~\ref{lem:prox-six-rem-decomp},
\[
\hat\theta_n^{DB,cf}(a_0)-\theta_0(a_0)=P_n\phi_{\infty,a_0}+\sum_{j=1}^6R_{n,j}(a_0).
\]
Lemma~\ref{lem:R1-rate} gives $R_{n,1}(a_0)=o(h_n^2)$, and Lemmas~\ref{lem:R6-rate}, \ref{lem:R4-control}, \ref{lem:R5-control}, \ref{lem:R2-control}, \ref{lem:R3-control} give $R_{n,j}(a_0)=o_p\{(nh_n)^{-1/2}\}$ for $j=2,\ldots,6$.  This proves the displayed first-order expansion.  Lemmas~\ref{lem:single_a_leading_term} and~\ref{lem:smoothing_negligible_single_a} yield the claimed normal limit for $(nh_n)^{1/2}P_n\phi_{\infty,a_0}$.  If $nh_n^5=O(1)$, then $(nh_n)^{1/2}R_{n,1}(a_0)=o(1)$, and the stated limit for the estimator follows by Slutsky's theorem.
\end{proof}

\begin{proof}[Proof of Theorem~\ref{thm:finite-dimensional-proximal-inference}]
Lemma~\ref{lem:prox-six-rem-decomp} gives the expansion componentwise at
$a_1,\ldots,a_m$. The pointwise remainder bounds in the proof of
Theorem~\ref{thm:pointwise-proximal-inference} apply componentwise. The
deterministic smoothing term is negligible under $nh_n^5=O(1)$.
Lemmas~\ref{lem:multiple_a_leading_term}
and~\ref{lem:smoothing_negligible_single_a} give the finite-dimensional
Gaussian limit with diagonal covariance matrix.
\end{proof}

\begin{proof}[Proof of Theorem~\ref{thm:uniform-proximal-inference}]
By Lemma~\ref{lem:prox-six-rem-decomp}, uniformly over $a_0\in\mathcal A_0$,
\[
\hat\theta_n^{DB,cf}(a_0)-\theta_0(a_0)-P_n\phi_{\infty,h_n,b_n,a_0}
=
\sum_{j=1}^6R_{n,j}(a_0).
\]
Assumptions~\ref{assm:G-bandwidth} and~\ref{assm:G-smooth} and Lemma~\ref{lem:R1-rate} give
\[
\begin{aligned}
\sqrt{nh_n\log n}
\sup_{a_0\in\mathcal A_0}|R_{n,1}(a_0)|
&=O\!\left(\sqrt{nh_n^5}\,h_n^\alpha\sqrt{\log n}\right)\\
&=o(1).
\end{aligned}
\]
Next, Lemmas~\ref{lem:R6-rate}, \ref{lem:R4-control}, \ref{lem:R5-control}, \ref{lem:R2-control}, \ref{lem:R3-control} give
\[
\max_{2\le j\le6}
\sup_{a_0\in\mathcal A_0}|R_{n,j}(a_0)|
=o_p\{(nh_n\log n)^{-1/2}\}.
\]
The uniform expansion follows.

The corollary to Lemma~\ref{lem:leading_term_gaussian_approx} gives the Gaussian approximation for
\[
\sup_{a_0\in\mathcal A_0}
\left|
\frac{(nh_n)^{1/2}P_n\phi_{\infty,h_n,b_n,a_0}}
{\sigma_{\infty,h_n,b_n}(a_0)}
\right|.
\]
By Lemma~\ref{lem:limiting_variance} and the uniform expansion, the difference between this random variable and the corresponding supremum with $P_n\phi_{\infty,h_n,b_n,a_0}$ replaced by $\hat\theta_n^{DB,cf}(a_0)-\theta_0(a_0)$ is $o_p\{(\log n)^{-1/2}\}$.  Assumption~\ref{assm:G-bandwidth} implies $\log(h_n^{-1})=O(\log n)$.  The anti-concentration bound used in Lemma~\ref{lem:leading_term_gaussian_approx} therefore gives the stated Kolmogorov-distance convergence.  Lemma~\ref{lem:finite-mesh-gaussian} and Assumption~\ref{assm:G-grid} give the finite-mesh statement.
\end{proof}

\section{Remainder control}
\subsection{The smoothing and local-polynomial matrix remainders}

We first control the deterministic smoothing remainder
\[
R_{n,1}(a_0)
=
\theta_{0,h_n,b_n}(a_0)-\theta_0(a_0).
\]

\begin{lemma}\label{lem:R1-rate}
If \ref{assm:kernel}, \ref{assm:bandwidths}, and
\ref{assm:regul1}(a)--(b) hold, then
\(R_{n,1}(a_0)=o(h_n^2)\).
If Assumption~\ref{assm:G-smooth}\textup{(i)}--\textup{(ii)} also holds, then
\[
\sup_{a_0\in\mathcal A_0}|R_{n,1}(a_0)|
=
O(h_n^{2+\alpha}+h_n^2b_n^\alpha).
\]
If Assumption~\ref{assm:G-bandwidth} also holds, then
\[
\sup_{a_0\in\mathcal A_0}|R_{n,1}(a_0)|
=
o\{(nh_n\log n)^{-1/2}\}.
\]
\end{lemma}

\begin{proof}
We note that the equivalent kernel reproduces polynomials through degree two:
\[
P_0\{\Gamma_0(A)\}=1,\quad
P_0\{\Gamma_0(A)(A-a_0)\}=0,\quad
P_0\{\Gamma_0(A)(A-a_0)^2\}=0.
\]
Taylor's expansion and Lemma~\ref{lem:D-expansion} consequently give
\[
|R_{n,1}(a_0)|
\le Ch_n^2\!\sup_{|u-a_0|\le h_n}
 |\theta_0''(u)-\theta_0''(a_0)|
 +Ch_n^2\!\sup_{|u-a_0|\le b_n}
 |\theta_0''(u)-\theta_0''(a_0)|.
\]
The continuity proves the pointwise result, while uniform H\"older continuity
gives \(O(h_n^{2+\alpha}+h_n^2b_n^\alpha)\).  
Under Assumption~\ref{assm:G-bandwidth}, \(b_n\asymp h_n\) and
\[
\sqrt{nh_n\log n}\,h_n^{2+\alpha}
=\sqrt{nh_n^5}\,h_n^\alpha\sqrt{\log n}\longrightarrow0,
\]
which proves the final assertion.
\end{proof}

We next control the local-polynomial matrix remainder \(R_{n,6}\).  Throughout
the proof, \(\|\cdot\|_\infty\) denotes the elementwise maximum norm.

\begin{lemma}\label{lem:R6-rate}
If \ref{assm:kernel}, \ref{assm:bandwidths}, and
\ref{assm:regul1}(a)--(b) hold, then
\[
R_{n,6}(a_0)=O_p\{(nh_n)^{-1}\}.
\]
If Assumption~\ref{assm:G-smooth}\textup{(i)}--\textup{(ii)} also holds and
\(nh_n/\log(h_n^{-1})\to\infty\), then
\[
\sup_{a_0\in\mathcal A_0}|R_{n,6}(a_0)|
=
O_p\!\left\{\frac{\log(h_n^{-1})}{nh_n}\right\}.
\]
\end{lemma}

\begin{proof}
The pointwise and uniform moment-matrix calculations given in Lemmas~8 and 12 of \citet{Takatsu2025debiased} yield
\[
\begin{aligned}
&D_{n,h,a_0,1}-D_{0,h,a_0,1},
\quad (P_n-P_0)\{\widetilde w_{h,a_0,1}K_{h,a_0}\}
 =O_p\{(nh)^{-1/2}\},\\
&D_{n,b,a_0,2}-D_{0,b,a_0,2}=O_p\{(nb)^{-1/2}\},
\end{aligned}
\]
entrywise.  The inverse identity and Lemma~\ref{lem:D-expansion} give the
same rates for the corresponding inverse differences, and the expansion of \(c_{n,h,a_0,2}-c_{0,h,a_0,2}\) gives
\(O_p\{(nh)^{-1/2}\}\).  Every summand defining \(R_{n,6}\) is a product of
two such differences and bounded population factors.  Since \(h/b=O(1)\), exactly as in Lemma~20 of
\citet{Takatsu2025debiased},
we get the pointwise and uniform results.
\end{proof}

\subsection{Control of the proximal drift and product-empirical remainders}
\label{sec:R4-R5-control}

\subsubsection{Control of the proximal doubly robust remainder}
\begin{lemma}
\label{lem:R4-control}
Under Assumptions~\ref{assm:A}, \ref{assm:B}, and~\ref{assm:D}--\ref{assm:F},
\[
    R_{n,4}(a_0)=o_p\{(nh_n)^{-1/2}\}.
\]
Under Assumptions~\ref{assm:A}, \ref{assm:B}, \ref{assm:D}, and~\ref{assm:G},
\[
    \sup_{a_0\in\mathcal A_0}|R_{n,4}(a_0)|
    =
    o_p\{(nh_n\log n)^{-1/2}\}.
\]
\end{lemma}

\begin{proof}
Write \(h=h_n\), \(b=b_n\), and \(\tau_n=h/b\).  We first establish a
bound on the \(L_1(F_0)\)-norm of the estimated equivalent kernel.  As in \citet{Takatsu2025debiased}, we have (also used in the analysis of preceding lemma)
\[
\begin{aligned}
\|D_{n,h,a_0,1}^{-1}\|_\infty
+
\|D_{n,b,a_0,2}^{-1}\|_\infty
+
|c_{n,h,a_0,2}|
=
O_p(1).
\end{aligned}
\]
Moreover, boundedness and compact support of \(K\), together with local
boundedness of \(f_A\), imply
\[
\begin{aligned}
\int
\|w_{h,a_0,1}(a)\|
|K_{h,a_0}(a)|\,dF_0(a)
&=O(1),\\
\int
\|w_{b,a_0,2}(a)\|
|K_{b,a_0}(a)|\,dF_0(a)
&=O(1).
\end{aligned}
\]
Consequently, we get
\begin{equation}
\label{eq:Gamma-L1-R4}
\int|\Gamma_{n,h,b,a_0}(a)|\,dF_0(a)
=
O_p(1),
\end{equation}
since \(\tau_n=O(1)\).  Under the uniform matrix bounds, the same
argument gives
\begin{equation}
\label{eq:Gamma-L1-R4-uniform}
\sup_{a_0\in\mathcal A_0}
\int|\Gamma_{n,h,b,a_0}(a)|\,dF_0(a)
=
O_p(1).
\end{equation}

We next compute the conditional proximal product-bias identity.  Fix a fold
\(k\) and condition on \(\trainingsigma_k\).  The nuisance estimators are then
fixed measurable functions, while an independent generic observation
\(O\) has law \(P_0\).  For \(F_0\)-almost every \(a\),
the outcome bridge equation gives
\[
\begin{aligned}
E_0\Big[
\hat q_n^{(-k)}(Z,a,X)
\{Y-h_0(W,a,X)\}
\mid A=a,\trainingsigma_k
\Big]
=0.
\end{aligned}
\]
Moreover, the treatment bridge equation and the density-ratio identity give
\[
\begin{aligned}
&E_0\Big[
q_0(Z,a,X)
\{h_0(W,a,X)-\hat h_n^{(-k)}(W,a,X)\}
\mid A=a,\trainingsigma_k
\Big]\\
&\qquad =
E_0\Big[
h_0(W,a,X)-\hat h_n^{(-k)}(W,a,X)
\mid\trainingsigma_k
\Big].
\end{aligned}
\]
Subtracting and adding the true bridges therefore yields
\begin{equation}
\begin{split}
\label{eq:conditional-proximal-product-bias}
&E_0\left[
\hat\xi_{n,1}^{(-k)}(O)
+
S_{\hat h_n^{(-k)}}(A)
\mid A=a,\trainingsigma_k
\right]
-\theta_0(a)\\
&\quad =
E_0\Big[
\{\hat q_n^{(-k)}(Z,a,X)-q_0(Z,a,X)\}
\{h_0(W,a,X)-\hat h_n^{(-k)}(W,a,X)\}
\mid A=a,\trainingsigma_k
\Big].
\end{split}
\end{equation}

For brevity, denote the right-hand side of
\eqref{eq:conditional-proximal-product-bias} by
\(\Delta_{n,k}(a)\).  From the definition of \(R_{n,4}\), we get
\[
\begin{aligned}
|R_{n,4}(a_0)|
\le
\sum_{k=1}^K\frac{n_k}{n}
\int
|\Gamma_{n,h,b,a_0}(a)|
|\Delta_{n,k}(a)|
\,dF_0(a).
\end{aligned}
\]
For all sufficiently large \(n\), the support of
\(\Gamma_{n,h,b,a_0}\) is contained in
\(B_{\delta_1}(a_0)\).  Hence
\begin{equation}
\label{eq:R4-reduced-to-drift}
|R_{n,4}(a_0)|
\le
\left\{
\int|\Gamma_{n,h,b,a_0}(a)|\,dF_0(a)
\right\}
\max_{1\le k\le K}
\sup_{a\in B_{\delta_1}(a_0)}
|\Delta_{n,k}(a)|.
\end{equation}

It remains to control the conditional drift.  Since
\(\mathcal S_1,\mathcal S_2,\mathcal S_3\) partition the relevant
conditional support,
\[
\begin{aligned}
\Delta_{n,k}(a)
=
\sum_{j=1}^3
E_0\Big[
&I_{\mathcal S_j}(a,Z,W,X)
\{\hat q_n^{(-k)}(Z,a,X)-q_0(Z,a,X)\}\\
&\qquad\times
\{h_0(W,a,X)-\hat h_n^{(-k)}(W,a,X)\}
\mid A=a,\trainingsigma_k
\Big].
\end{aligned}
\]
The Cauchy--Schwarz inequality gives
\[
\begin{aligned}
\sup_{a\in\mathcal I}|\Delta_{n,k}(a)|
\le
\sum_{j=1}^3
&d_{q,k}(
\hat q_n^{(-k)},q_0;\mathcal I,\mathcal S_j)
d_{h,k}(
\hat h_n^{(-k)},h_0;\mathcal I,\mathcal S_j).
\end{aligned}
\]

On \(\mathcal S_1\), \(q_\infty=q_0\), and therefore
\[
d_{q,k}(
\hat q_n^{(-k)},q_0;\mathcal I,\mathcal S_1)
=
d_{q,k}(
\hat q_n^{(-k)},q_\infty;\mathcal I,\mathcal S_1).
\]
Furthermore,
\[
\begin{aligned}
d_{h,k}(
\hat h_n^{(-k)},h_0;\mathcal I,\mathcal S_1)
\le 
d_{h,k}(
\hat h_n^{(-k)},h_\infty;\mathcal I,\mathcal S_1)+
d_{h,k}(
h_\infty,h_0;\mathcal I,\mathcal S_1)
=
O_p(1).
\end{aligned}
\]
Thus, the \(\mathcal S_1\) contribution is
\(o_p\{(nh)^{-1/2}\}\).  The same argument with the roles of the
bridges reversed shows that the \(\mathcal S_2\) contribution is
\(o_p\{(nh)^{-1/2}\}\).  On \(\mathcal S_3\), both limiting
bridges are correct, and the product-rate condition gives the same order.
Therefore, we obtain
\[
\max_{1\le k\le K}
\sup_{a\in B_{\delta_1}(a_0)}
|\Delta_{n,k}(a)|
=
o_p\{(nh)^{-1/2}\}.
\]
Combining this result with \eqref{eq:Gamma-L1-R4} and
\eqref{eq:R4-reduced-to-drift} proves the pointwise assertion.

For the uniform assertion, the support of
\(\Gamma_{n,h,b,a_0}\) is contained in
\(\mathcal A_{\delta_3}\), uniformly over \(a_0\in\mathcal A_0\),
for all sufficiently large \(n\).  Hence
\[
\begin{aligned}
\sup_{a_0\in\mathcal A_0}|R_{n,4}(a_0)|
\le{}&
\sup_{a_0\in\mathcal A_0}
\int|\Gamma_{n,h,b,a_0}(a)|\,dF_0(a)
\max_{1\le k\le K}
\sup_{a\in\mathcal A_{\delta_3}}
|\Delta_{n,k}(a)|.
\end{aligned}
\]
The first factor is \(O_p(1)\) by
\eqref{eq:Gamma-L1-R4-uniform}, and the partitioned uniform rate
conditions make the second factor
\(o_p\{(nh\log n)^{-1/2}\}\).  This proves the result.
\end{proof}

\subsubsection{A conditional product empirical-process bound}

For \(\ell\in\{h_n,b_n\}\), \(j\in\{0,1,2\}\), and \(u\in\mathbb R\),
define
\[
 p_{\ell,u,j}(a)
 :=\left(\frac{a-u}{\ell}\right)^jK_{\ell,u}(a),
 \qquad
 \mathcal P_{\ell,j}
 :=\{p_{\ell,u,j}:u\in\mathcal A_0\}.
\]
We use this notation throughout the remainder of this section.

\begin{lemma}
\label{lem:conditional-product-process}
Under Assumptions~\ref{assm:basic.proxy}, \ref{assm:kernel},
\ref{assm:bandwidths}, \ref{assm:folds}, and~\ref{assm:nuisance1}, for
\(\ell\in\{h_n,b_n\}\), \(j\in\{0,1,2\}\), and fixed \(a_0\),
\[
\begin{aligned}
&\iint p_{\ell,a_0,j}(a)\hat h_n^{(-k)}(w,a,x)
 \,d(Q_{n,k}^{WX}-Q_0^{WX})(w,x)\,d(F_{n,k}-F_0)(a)\\
&\qquad=O_p\{(n_k\sqrt\ell)^{-1}\}.
\end{aligned}
\]
Under Assumptions~\ref{assm:basic.proxy}, \ref{assm:kernel},
\ref{assm:G-bandwidth}, and~\ref{assm:G-crossfit},
\[
\begin{aligned}
&\sup_{a_0\in\mathcal A_0}
 \left|\iint p_{\ell,a_0,j}(a)\hat h_n^{(-k)}(w,a,x)
 \,d(Q_{n,k}^{WX}-Q_0^{WX})(w,x)\,d(F_{n,k}-F_0)(a)\right|\\
&\qquad=O_p\{(n_k\ell)^{-1}\}.
\end{aligned}
\]
\end{lemma}

\begin{proof}
Fix \(k\), \(\ell\), and \(j\), and condition on
\(\trainingsigma_k\).  The observations indexed by \(I_k\) are then
independent with law \(P_0\), by Assumption~\ref{assm:folds}.
The conditional bounds below hold on the \(\trainingsigma_k\)-measurable
event
\[
 \sup_{a\in\mathcal I,\,w,x}
 |\hat h_n^{(-k)}(w,a,x)|\le C_h,
\]
where \(\mathcal I=B_{\delta_1}(a_0)\) for the pointwise assertion and
\(\mathcal I=\mathcal A_{\delta_3}\) for the uniform assertion; this event concerns the fixed fold \(k\). We note that the constants in the conditional
bounds do not depend on the training realization.  The probability
that this bound fails tends to zero by
Assumption~\ref{assm:nuisance1} or~\ref{assm:G-crossfit}, respectively.
For all sufficiently large \(n\), the support of every kernel under
consideration is contained in the corresponding \(\mathcal I\).
In this proof, \(C\) denotes a finite constant independent of \(n\),
\(k\), \(\ell\), and the training folds; its value may vary between
displays.

For \(o_i=(y_i,a_i,z_i,w_i,x_i)\), define
\[
 f_{n,k,\ell,j,u}(o_1,o_2)
 :=p_{\ell,u,j}(a_2)\hat h_n^{(-k)}(w_1,a_2,x_1),
\]
and define its centering and symmetrization by
\[
\begin{aligned}
 f_{n,k,\ell,j,u}^{\circ}(o_1,o_2)
 :={}&f_{n,k,\ell,j,u}(o_1,o_2)
   -\int f_{n,k,\ell,j,u}(o_1,v)\,dP_0(v)\\
 &-\int f_{n,k,\ell,j,u}(v,o_2)\,dP_0(v)
   +\iint f_{n,k,\ell,j,u}(v_1,v_2)
                  \,dP_0(v_1)\,dP_0(v_2),\\
 \widetilde f_{n,k,\ell,j,u}^{\circ}(o_1,o_2)
 :={}&\frac12\{f_{n,k,\ell,j,u}^{\circ}(o_1,o_2)
             +f_{n,k,\ell,j,u}^{\circ}(o_2,o_1)\}.
\end{aligned}
\]
These definitions give
\[
\begin{aligned}
 \int f_{n,k,\ell,j,u}^{\circ}(o_1,v)\,dP_0(v)
 &=\int f_{n,k,\ell,j,u}^{\circ}(v,o_2)\,dP_0(v)=0,\\
 \int\widetilde f_{n,k,\ell,j,u}^{\circ}(o_1,v)\,dP_0(v)
 &=\int\widetilde f_{n,k,\ell,j,u}^{\circ}(v,o_2)\,dP_0(v)=0,\\
 \widetilde f_{n,k,\ell,j,u}^{\circ}(o_1,o_2)
 &=\widetilde f_{n,k,\ell,j,u}^{\circ}(o_2,o_1).
\end{aligned}
\]
Expanding the signed measures yields
\[
\begin{aligned}
 &\iint f_{n,k,\ell,j,u}(o_1,o_2)
       \,d(P_{n,k}-P_0)(o_1)\,d(P_{n,k}-P_0)(o_2)\\
 &\quad=\frac1{n_k^2}\sum_{r,s\in I_k}f_{n,k,\ell,j,u}(O_r,O_s)
       -\frac1{n_k}\sum_{r\in I_k}
          \int f_{n,k,\ell,j,u}(O_r,v)\,dP_0(v)\\
 &\qquad\quad-\frac1{n_k}\sum_{s\in I_k}
          \int f_{n,k,\ell,j,u}(v,O_s)\,dP_0(v)
       +\iint f_{n,k,\ell,j,u}(v_1,v_2)\,dP_0(v_1)\,dP_0(v_2)\\
 &\quad=\frac1{n_k^2}\sum_{r,s\in I_k}
           f_{n,k,\ell,j,u}^{\circ}(O_r,O_s).
\end{aligned}
\]
Consequently, we get
\begin{equation}
\label{eq:conditional-product-exact}
\begin{aligned}
 &\iint p_{\ell,u,j}(a)\hat h_n^{(-k)}(w,a,x)
 \,d(Q_{n,k}^{WX}-Q_0^{WX})(w,x)\,d(F_{n,k}-F_0)(a)\\
 &\quad=\frac1{n_k^2}
   \sum_{\substack{r,s\in I_k\\r\ne s}}
        \widetilde f_{n,k,\ell,j,u}^{\circ}(O_r,O_s)
       +\frac1{n_k^2}\sum_{r\in I_k}
        f_{n,k,\ell,j,u}^{\circ}(O_r,O_r).
\end{aligned}
\end{equation}
For fixed \(u=a_0\), the change of variables \(a=a_0+\ell t\) gives,
for \(r\in\{1,2\}\),
\[
\begin{aligned}
 \int |p_{\ell,a_0,j}(a)|^r\,dF_0(a)
 &=\ell^{1-r}\int_{-1}^1
       |t|^{jr}|K(t)|^r f_A(a_0+\ell t)\,dt\\
 &\le \ell^{1-r}\|f_A\|_\infty
        \int_{-1}^1 |t|^{jr}|K(t)|^r\,dt
 \le C\ell^{1-r}.
\end{aligned}
\]
Consequently, we have
\[
 \|f_{n,k,\ell,j,a_0}\|_{P_0^{\otimes2},2}\le C\ell^{-1/2},
 \qquad
 \|\widetilde f_{n,k,\ell,j,a_0}^{\circ}\|_{P_0^{\otimes2},2}
 \le4\|f_{n,k,\ell,j,a_0}\|_{P_0^{\otimes2},2},
\]
where the second inequality follows by applying Jensen's inequality to each
integral in the centering and to the average defining the symmetrization.
For \(g=\widetilde f_{n,k,\ell,j,a_0}^{\circ}\) and three distinct
indices \(r,s,t\), conditional independence gives
\[
 E_0\{g(O_r,O_s)g(O_r,O_t)\mid O_r,\trainingsigma_k\}
 =\left\{\int g(O_r,v)\,dP_0(v)\right\}^{\!2}=0.
\]
Symmetry covers every arrangement with exactly one shared index.
For \(r\ne s\) and \(t\ne v\), independence for disjoint pairs and the
preceding identity therefore give
\[
 E_0\{g(O_r,O_s)g(O_t,O_v)\mid\trainingsigma_k\}
 =P_0^{\otimes2}g^2
   \{I(r=t,s=v)+I(r=v,s=t)\}.
\]
There are \(n_k(n_k-1)\) ordered pairs, each with precisely two matches.
Thus
\[
\begin{aligned}
E_0\left[
 \left\{\frac1{n_k^2}\sum_{\substack{r,s\in I_k\\r\ne s}}
       \widetilde f_{n,k,\ell,j,a_0}^{\circ}(O_r,O_s)\right\}^{\!2}
       \,\middle|\,\trainingsigma_k\right]=
 \frac{2n_k(n_k-1)}{n_k^4}
 P_0^{\otimes2}(\widetilde f_{n,k,\ell,j,a_0}^{\circ})^2
 \le \frac{C}{n_k^2\ell}.
\end{aligned}
\]
The four terms in \(f_{n,k,\ell,j,a_0}^{\circ}(O,O)\) satisfy
\[
\begin{aligned}
 |f_{n,k,\ell,j,a_0}(O,O)|
 &\le C_h|p_{\ell,a_0,j}(A)|,\\
 \left|\int f_{n,k,\ell,j,a_0}(O,v)\,dP_0(v)\right|
 &\le C_h\int|p_{\ell,a_0,j}(a)|\,dF_0(a),\\
 \left|\int f_{n,k,\ell,j,a_0}(v,O)\,dP_0(v)\right|
 &\le C_h|p_{\ell,a_0,j}(A)|,\\
 \left|\iint f_{n,k,\ell,j,a_0}(v_1,v_2)
                       \,dP_0(v_1)\,dP_0(v_2)\right|
 &\le C_h\int|p_{\ell,a_0,j}(a)|\,dF_0(a).
\end{aligned}
\]
Taking conditional expectations and summing gives
\[
 E_0\left[
  \left|\frac1{n_k^2}\sum_{r\in I_k}
       f_{n,k,\ell,j,a_0}^{\circ}(O_r,O_r)\right|
       \,\middle|\,\trainingsigma_k\right]
 \le\frac C{n_k}.
\]
Next, Chebyshev's and Markov's inequalities applied to
\eqref{eq:conditional-product-exact} prove the pointwise assertion,
since \(\ell\to0\).

For the uniform assertion, we first derive a covering bound for
\(\mathcal P_{\ell,j}\).  The function \(t\mapsto t^jK(t)\) is
Lipschitz continuous and vanishes outside \([-1,1]\); hence it has bounded
variation.  It can therefore be written as \(L_{j,1}-L_{j,2}\), where
\(L_{j,1}\) and \(L_{j,2}\) are bounded nondecreasing functions.
Explicitly, if \(V_j(t)\) is the total variation of
\(s\mapsto s^jK(s)\) on \(( -\infty,t]\), take
\[
 L_{j,1}(t)=\tfrac12\{V_j(t)+t^jK(t)\},
 \qquad
 L_{j,2}(t)=\tfrac12\{V_j(t)-t^jK(t)\}.
\]
For \(r\in\{1,2\}\), define
\[
 \mathcal P_{\ell,j,r}
 :=\{a\mapsto\ell^{-1}L_{j,r}((a-u)/\ell):u\in\mathcal A_0\}.
\]
For \(u\le v\) and every \(a\),
\[
 \ell^{-1}L_{j,r}((a-v)/\ell)
 \le \ell^{-1}L_{j,r}((a-u)/\ell),\qquad r=1,2.
\]
Thus each class has VC dimension at most one by
Lemma~2.6.17 of \citet{VanderVaartWellner2023weak}.
Their envelopes are bounded by \(C/\ell\).  Theorem~2.6.7 of that
reference, with exponent two in its \(L_r\) norm, gives
\[
 \sup_Q N\left(\frac{\varepsilon}{4\ell},
              \mathcal P_{\ell,j,r},L_2(Q)\right)
 \le C\varepsilon^{-2},
 \qquad r=1,2,\quad 0<\varepsilon<1,
\]
where \(Q\) ranges over probability measures on \(\mathcal A\).
Cover the two classes at radius \(\varepsilon/(4\ell)\)
and classify each \(u\) by the pair of covering balls containing its two
components.  Choose one \(u\) from each nonempty intersection.  If \(u\)
and \(v\) are in the same intersection, the triangle inequality gives
\(\|p_{\ell,u,j}-p_{\ell,v,j}\|_{Q,2}\le\varepsilon/\ell\).
There are at most the product of the two covering numbers such
intersections.  Consequently, for a constant depending on \(K\) only, we get
\begin{equation}
\label{eq:kernel-polynomial-cover}
 \sup_Q N\left(\frac{\varepsilon}{\ell},
               \mathcal P_{\ell,j},L_2(Q)\right)
 \le C\varepsilon^{-4},
 \qquad 0<\varepsilon<1.
\end{equation}
This argument allows the centers of the final covering balls to belong to
\(\mathcal P_{\ell,j}\).

Define the two classes of functions on pairs of observations by
\[
\begin{aligned}
 \mathcal F_{n,k,\ell,j}
 &:=\{f_{n,k,\ell,j,u}:u\in\mathcal A_0\},\\
 \widetilde{\mathcal F}_{n,k,\ell,j}^{\circ}
 &:=\{\widetilde f_{n,k,\ell,j,u}^{\circ}:u\in\mathcal A_0\}\cup\{0\}.
\end{aligned}
\]
For a probability measure \(Q\) on pairs of observations, write
\(Q_{A_2}\) for its second treatment-coordinate marginal.
Because the supports of \(p_{\ell,u,j}\) and \(p_{\ell,v,j}\) lie in
\(\mathcal A_{\delta_3}\),
\[
 \|f_{n,k,\ell,j,u}-f_{n,k,\ell,j,v}\|_{Q,2}
 \le C_h\|p_{\ell,u,j}-p_{\ell,v,j}\|_{Q_{A_2},2}.
\]
Consequently, \eqref{eq:kernel-polynomial-cover} gives
\[
 \sup_Q N\left(\frac{\varepsilon}{\ell},
             \mathcal F_{n,k,\ell,j},L_2(Q)\right)
 \le \sup_Q N\left(\frac{\varepsilon}{(1+C_h)\ell},
                  \mathcal P_{\ell,j},L_2(Q)\right)
 \le C\varepsilon^{-4}.
\]

To handle the centered class, let \(Q^\top\) be the image of \(Q\) under
\((o_1,o_2)\mapsto(o_2,o_1)\), let
\(\overline Q=(Q+Q^\top)/2\), and let
\(\overline Q_1\) and \(\overline Q_2\) be its two marginals.  Set
\[
 Q^\star
 :=\frac14\{\overline Q+\overline Q_1\otimes P_0
                     +P_0\otimes\overline Q_2+P_0^{\otimes2}\}.
\]
For any two functions \(f,g\in\mathcal F_{n,k,\ell,j}\), Jensen's
inequality first for symmetrization and then for the two integrated
coordinates yields
\[
\begin{aligned}
 \|\widetilde f^{\circ}-\widetilde g^{\circ}\|_{Q,2}
 &\le \|f^{\circ}-g^{\circ}\|_{\overline Q,2}\\
 &\le \|f-g\|_{\overline Q,2}
       +\|f-g\|_{\overline Q_1\otimes P_0,2}
       +\|f-g\|_{P_0\otimes\overline Q_2,2}
       +\|f-g\|_{P_0^{\otimes2},2}\\
 &\le4\|f-g\|_{Q^\star,2}.
\end{aligned}
\]
In the last step, apply the Cauchy--Schwarz inequality to the sum of the
four norms.  Covering \(\mathcal F_{n,k,\ell,j}\) under \(Q^\star\)
and taking the centered, symmetrized images of its centers therefore gives
\begin{equation}
\label{eq:conditional-product-cover}
 \sup_Q N\left(\frac{\varepsilon}{\ell},
       \widetilde{\mathcal F}_{n,k,\ell,j}^{\circ},L_2(Q)\right)
 \le1+C\varepsilon^{-4},
 \qquad 0<\varepsilon<1.
\end{equation}
The four centering terms also give
\[
 \sup_{u,o_1,o_2}|f_{n,k,\ell,j,u}^{\circ}(o_1,o_2)|
 \le \frac{4C_h\|K\|_\infty}{\ell},
 \qquad
 \sup_{g\in\widetilde{\mathcal F}_{n,k,\ell,j}^{\circ}}
       \|g\|_\infty\le \frac{4C_h\|K\|_\infty}{\ell}.
\]
Since \(u\mapsto p_{\ell,u,j}(a)\) is
continuous, all the centered and symmetrized kernels are continuous in
\(u\), by dominated convergence.  A countable dense subset of the compact
set \(\mathcal A_0\) consequently gives the same suprema in the
calculations below.

We apply Propositions~2.1, 2.2, and~2.6 of
\citet{Arcones1993limit} at order two.  To give the precise normalization,
let \(\{\epsilon_r:r\in I_k\}\) be independent Rademacher variables,
independent of all observations, and, for
\(g\in\widetilde{\mathcal F}_{n,k,\ell,j}^{\circ}\), define
\[
 Z_g:=\frac1{n_k}
       \sum_{\substack{r,s\in I_k\\r\ne s}}
         \epsilon_r\epsilon_s g(O_r,O_s),
 \qquad
 Q_{n,k}^{(2)}
 :=\frac1{n_k(n_k-1)}
       \sum_{\substack{r,s\in I_k\\r\ne s}}\delta_{(O_r,O_s)}.
\]
Here \(Q_{n,k}^{(2)}\) is a probability measure for all sufficiently
large \(n\).  Conditional on the observations, symmetry gives
\[
 d(g,g')^2
 :=E_{\epsilon}(Z_g-Z_{g'})^2
 =\frac{2(n_k-1)}{n_k}
       \|g-g'\|_{Q_{n,k}^{(2)},2}^{\,2}.
\]
Proposition~2.2 of \citet{Arcones1993limit} gives, for \(1<q<p<\infty\),
\[
 \{E_\epsilon|Z_g-Z_{g'}|^p\}^{1/p}
 \le\frac{p-1}{q-1}
       \{E_\epsilon|Z_g-Z_{g'}|^q\}^{1/q}.
\]
Thus condition~\textup{(2.11)} therein holds with order two;
in particular,
\[
 \{E_\epsilon|Z_g-Z_{g'}|^p\}^{1/p}
 \le(p-1)\{E_\epsilon|Z_g-Z_{g'}|^2\}^{1/2}
 =(p-1)d(g,g'),\qquad p\ge2.
\]
The zero function belongs to the class, so \(Z_0=0\), and the diameter
in \(d\) is at most \(C/\ell\).  Proposition~2.6 of \citet{Arcones1993limit} and
\eqref{eq:conditional-product-cover} yield
\[
\begin{aligned}
 E_{\epsilon}\sup_{g\in\widetilde{\mathcal F}_{n,k,\ell,j}^{\circ}}|Z_g|
 &\le C\int_0^{\operatorname{diam}
             (\widetilde{\mathcal F}_{n,k,\ell,j}^{\circ},d)}
       \log N(\varepsilon,\widetilde{\mathcal F}_{n,k,\ell,j}^{\circ},d)
       \,d\varepsilon\\
 &\le \frac C\ell
       \int_0^1\{1+\log(1/t)\}\,dt
 \le\frac C\ell.
\end{aligned}
\]
The comparison of a symmetric canonical sum with its Rademacher
randomization in Proposition~2.1 of \citet{Arcones1993limit}, specifically its equivalence with the display
\textup{(2.3)} therein, now yields
\[
\begin{aligned}
 E_0\left[
  \sup_{u\in\mathcal A_0}
   \left|\frac1{n_k^2}
     \sum_{\substack{r,s\in I_k\\r\ne s}}
           \widetilde f_{n,k,\ell,j,u}^{\circ}(O_r,O_s)\right|
     \,\middle|\,\trainingsigma_k\right]\le\frac C{n_k}
 E_0\left[
 E_{\epsilon}
   \sup_{g\in\widetilde{\mathcal F}_{n,k,\ell,j}^{\circ}}|Z_g|
      \,\middle|\,\trainingsigma_k\right]
 \le\frac C{n_k\ell}.
\end{aligned}
\]
In particular, the normalization in this display is related to the
usual order-two \(U\)-statistic normalization by
\[
 \frac1{n_k^2}\sum_{\substack{r,s\in I_k\\r\ne s}}g(O_r,O_s)
 =\frac{n_k-1}{n_k}
   \binom{n_k}{2}^{-1}
     \sum_{\substack{r<s\\r,s\in I_k}}g(O_r,O_s).
\]
Finally, the diagonal in \eqref{eq:conditional-product-exact} satisfies
\[
 \sup_{u\in\mathcal A_0}
 \left|\frac1{n_k^2}\sum_{r\in I_k}
    f_{n,k,\ell,j,u}^{\circ}(O_r,O_r)\right|
 \le\frac1{n_k^2}\sum_{r\in I_k}
      \sup_{u\in\mathcal A_0}|f_{n,k,\ell,j,u}^{\circ}(O_r,O_r)|
 \le\frac C{n_k\ell}.
\]
Applying Markov's inequality proves the
uniform statement.
\end{proof}

\subsubsection{Control of the product empirical-measure remainder}

\begin{lemma}
\label{lem:R5-control}
Under Assumptions~\ref{assm:A} and~\ref{assm:D}--\ref{assm:F},
\[
    R_{n,5}(a_0)
    =O_p\{(n\sqrt{h_n})^{-1}\}
    =o_p\{(nh_n)^{-1/2}\}.
\]
Under Assumptions~\ref{assm:A}, \ref{assm:D}, and~\ref{assm:G},
\[
    \sup_{a_0\in\mathcal A_0}|R_{n,5}(a_0)|
    =O_p\{(nh_n)^{-1}\}
    =o_p\{(nh_n\log n)^{-1/2}\}.
\]
\end{lemma}

\begin{proof}
Again write \(h=h_n\) and \(b=b_n\). We will work first on the event that both
empirical moment matrices are nonsingular.  The matrix estimates in the
proof of Lemma~\ref{lem:D-expansion} and Lemma~\ref{lem:R6-rate} imply that its probability tends to one.
Substituting \(\Gamma_n\) into \(R_{n,5}\) gives
\[
\begin{aligned}
R_{n,5}(a_0)
={}&\sum_{k=1}^K\frac{n_k}{n}
\sum_{j=0}^1
\left[e_1^\top D_{n,h,a_0,1}^{-1}\right]_{j+1}
\\
&\quad\times
\iint
 p_{h,a_0,j}(a)
\hat h_n^{(-k)}(w,a,x)\,
d(Q_{n,k}^{WX}-Q_0^{WX})(w,x)\,d(F_{n,k}-F_0)(a)
\\
&-\tau_n^2c_{n,h,a_0,2}
\sum_{k=1}^K\frac{n_k}{n}
\sum_{j=0}^2
\left[e_3^\top D_{n,b,a_0,2}^{-1}\right]_{j+1}
\\
&\quad\times
\iint
p_{b,a_0,j}(a)
\hat h_n^{(-k)}(w,a,x)\,
d(Q_{n,k}^{WX}-Q_0^{WX})(w,x)\,d(F_{n,k}-F_0)(a).
\end{aligned}
\]
Here \([v]_{j+1}\) denotes the \((j+1)\)th entry of a row vector \(v\).
The moment-matrix estimates in the proof of
Lemma~\ref{lem:R6-rate}, together with Lemma~\ref{lem:D-expansion}, imply
\[
\|D_{n,h,a_0,1}^{-1}\|_\infty
+\|D_{n,b,a_0,2}^{-1}\|_\infty
+|c_{n,h,a_0,2}|=O_p(1).
\]
For all sufficiently large \(n\), both kernel supports are contained in
\(B_{\delta_1}(a_0)\).  Applying the pointwise assertion of
Lemma~\ref{lem:conditional-product-process} separately to every integral
in the preceding identity, then using the triangle inequality, yields
\[
\begin{aligned}
|R_{n,5}(a_0)|
=O_p\!\left[
\sum_{k=1}^K\frac{n_k}{n}
\left\{\frac1{n_k\sqrt h}
+\frac{\tau_n^2}{n_k\sqrt b}\right\}\right]
=O_p\!\left\{\frac K{n\sqrt h}(1+\tau_n^{5/2})\right\}
=O_p\{(n\sqrt h)^{-1}\}.
\end{aligned}
\]
The coefficients in this identity may depend on the validation
observations.  The conditional assertion of
Lemma~\ref{lem:conditional-product-process} is applied to the integrals;
the displayed bound then follows by multiplication of their unconditional
orders with the matrix bounds. Thus, the pointwise assertion follows.

For the uniform assertion, the supports are contained in
\(\mathcal A_{\delta_3}\) for every \(a_0\in\mathcal A_0\) once
\(\max(h,b)<\delta_3\).  The uniform moment-matrix estimates in the proof
of Lemma~\ref{lem:R6-rate} give
\[
\sup_{a_0\in\mathcal A_0}
\left\{
\|D_{n,h,a_0,1}^{-1}\|_\infty
+\|D_{n,b,a_0,2}^{-1}\|_\infty
+|c_{n,h,a_0,2}|
\right\}=O_p(1).
\]
The empirical matrices are nonsingular simultaneously for all
\(a_0\in\mathcal A_0\) with probability tending to one, by the same
matrix estimates and Lemma~\ref{lem:D-expansion}.
Taking the supremum in the exact identity above and applying the uniform
assertion of Lemma~\ref{lem:conditional-product-process} therefore indicates
\[
\begin{aligned}
\sup_{a_0\in\mathcal A_0}|R_{n,5}(a_0)|
=O_p\!\left[
\sum_{k=1}^K\frac{n_k}{n}
\left\{\frac1{n_kh}+\frac{\tau_n^2}{n_kb}\right\}\right]
=O_p\!\left\{\frac K{nh}(1+\tau_n^3)\right\}
=O_p\{(nh)^{-1}\}.
\end{aligned}
\]
Assumption~\ref{assm:G-bandwidth} implies \(nh\ge n^{1-1/\kappa}\)
for all sufficiently large \(n\), since \(nh^\kappa\to\infty\).
Consequently,
\[
\frac{(nh)^{-1}}{(nh\log n)^{-1/2}}
=\left(\frac{\log n}{nh}\right)^{1/2}
\le\left(\frac{\log n}{n^{1-1/\kappa}}\right)^{1/2}
\longrightarrow0,
\]
which proves the uniform assertion.
\end{proof}

\subsubsection{Control of the second and third remainders}
\label{sec:R2-R3-control}

Conditionally on $\trainingsigma_k$, the estimated bridges are fixed
functions and the held-out observations are IID from $P_0$.
The empirical-process classes below vary only over the kernel center.

For \(k=1,\ldots,K\), define
\[
    \hat\vartheta_{n,k}(O)
    :=
    \hat\xi_{n,1}^{(-k)}(O)
    +
    S_{\hat h_n^{(-k)}}(A),
    \qquad
    \vartheta_\infty(O):=\xi_\infty(O).
\]
For pointwise inference we take \(\mathcal I=B_{\delta_1}(a_0)\).  For uniform
inference we take \(\mathcal I=\mathcal A_{\delta_3}\).  Throughout this
subsection, \(K\) is fixed and \(\min_k n_k/n\) is bounded away from zero.
We retain the bandwidth convention \(h_n\to0\) and \(b_n\to0\) from
Assumption~\ref{assm:bandwidths}, and use the functions
\(p_{\ell,u,j}\) used in Lemma~\ref{lem:conditional-product-process}.

\begin{lemma}
\label{lem:conditional-kernel-ep}
Under Assumptions~\ref{assm:basic.proxy}, \ref{assm:kernel},
and~\ref{assm:folds}, let
\(\ell\in\{h_n,b_n\}\) and \(j\in\{0,1,2\}\).
Let \(\psi_{n,k}\) be a random measurable function of \(O\)
whose dependence on the training data is
\(\trainingsigma_k\)-measurable.  For
\(\mathcal I\subseteq\mathcal A\), define
\[
    \delta_{\psi,k}(\mathcal I)
    :=
    \sup_{a\in\mathcal I}
    \{E_0(\psi_{n,k}^2\mid A=a,\trainingsigma_k)\}^{1/2}.
\]
If \(\delta_{\psi,k}(\mathcal I)<\infty\) and
\(\{a:|a-u|\le\ell\}\subseteq\mathcal I\), then, for fixed
\(u\in\mathcal I\),
\[
    (P_{n,k}-P_0)\{p_{\ell,u,j}(A)\psi_{n,k}(O)\}
    =
    O_p\left(\frac{\delta_{\psi,k}(\mathcal I)}{\sqrt{n_k\ell}}\right).
\]
If Assumption~\ref{assm:G-bandwidth} also holds,
\[
\|\psi_{n,k}(O)I(A\in\mathcal A_{\delta_3})\|_{P_0,\infty}=O_p(1),
\]
and
\(\{a:|a-u|\le\ell\}\subseteq\mathcal A_{\delta_3}\) for every
\(u\in\mathcal A_0\), then
\[
\sup_{u\in\mathcal A_0}
\left|
(P_{n,k}-P_0)\{p_{\ell,u,j}(A)\psi_{n,k}(O)\}
\right|
=
O_p\left\{
\delta_{\psi,k}(\mathcal A_{\delta_3})
\sqrt{\frac{\log(\ell^{-1})}{n_k\ell}}
+
\frac{\log(\ell^{-1})}{n_k\ell}
\right\}.
\]
\end{lemma}

\begin{proof}
Write $N=n_k$, $\delta=\delta_{\psi,k}(\mathcal I)$, and
$\phi_j(t)=t^jK(t)$.  All expectations in this proof are conditional on
$\trainingsigma_k$.  Put $M_f=\sup_a f_A(a)<\infty$.  The support condition
and the substitution $a=u+\ell t$ give
\begin{align*}
P_0\{p_{\ell,u,j}^2(A)\psi_{n,k}^2(O)\}
&=\int p_{\ell,u,j}^2(a)
 E_0\{\psi_{n,k}^2(O)\mid A=a,\trainingsigma_k\}f_A(a)\,da\\
&\le M_f\delta^2\int\ell^{-2}\phi_j^2\{(a-u)/\ell\}\,da
=\frac{M_f\delta^2}{\ell}\int\phi_j^2(t)\,dt.
\end{align*}
The conditional variance of the empirical average is at most this
quantity divided by $N$.  For $t>0$, Chebyshev's inequality therefore gives
\[
P\left\{\left|(P_{n,k}-P_0)(p_{\ell,u,j}\psi_{n,k})\right|
>\frac{t\delta}{\sqrt{N\ell}}\,\middle|\,\trainingsigma_k\right\}
\le \frac{M_f\int\phi_j^2}{t^2},
\]
which further validates the pointwise claim.

For the uniform assertion set $\mathcal I=\mathcal A_{\delta_3}$.
Fix $M\ge1$ and restrict the training data to the event
\[
\|\psi_{n,k}I(A\in\mathcal A_{\delta_3})\|_{P_0,\infty}\le M.
\]
Replace $\psi_{n,k}I(A\in\mathcal A_{\delta_3})$ by its truncation to
$[-M,M]$, denoted by $\psi^M$.  On this event the two functions agree
$P_0$-almost surely.  Conditional IID sampling implies agreement at every
held-out observation almost surely as well.  The truncation thus changes
neither empirical averages nor their $P_0$ integrals, and supplies a
pointwise bounded version for covering numbers under arbitrary measures.
Consider the following function class:
\[
\mathcal F=\{O\mapsto p_{\ell,u,j}(A)\psi^M(O):u\in\mathcal A_0\},
\qquad B=2M\|\phi_j\|_\infty/\ell.
\]
Every member has absolute value at most $B/2$.  For any probability measure
$Q$ with $Q(\psi^M)^2>0$, define a probability measure on the corresponding treatment space by
\[
Q_{\psi}(D)
=\frac{Q\{(\psi^M)^2I(A\in D)\}}{Q(\psi^M)^2}.
\]
Then, for all centers $u,v$, one has
\[
\|p_{\ell,u,j}\psi^M-p_{\ell,v,j}\psi^M\|_{Q,2}
=\{Q(\psi^M)^2\}^{1/2}
 \|p_{\ell,u,j}-p_{\ell,v,j}\|_{Q_{\psi},2}.
\]
The covering bound~\eqref{eq:kernel-polynomial-cover}, together with
$\{Q(\psi^M)^2\}^{1/2}\le M$, consequently implies
\[
\sup_QN(\varepsilon B,\mathcal F,L_2(Q))
\le C\varepsilon^{-4},\qquad 0<\varepsilon<1.
\]
Define $\mathcal F/B:=\{f/B:f\in\mathcal F\}$, with constant
envelope one.  Its uniform entropy integral is
\[
J(t):=\sup_Q\int_0^t
 \sqrt{1+\log N(\varepsilon,\mathcal F/B,L_2(Q))}\,d\varepsilon,
\]
where the supremum is over finitely supported probability measures.
It satisfies
\[
J(t)\le C\int_0^t\sqrt{\log(C/\varepsilon)}\,d\varepsilon
\le Ct\sqrt{\log(C/t)},\qquad 0<t<1.
\]
The variance calculation above gives
\[
 \sup_{f\in\mathcal F}P_0f^2\le\sigma^2,
 \qquad
 \sigma^2:=M_f\delta^2\ell^{-1}\int\phi_j(t)^2\,dt.
\]
For $N\ge4$ choose
\[
t=\min\{3/4,\max(2\sigma/B,N^{-1/2})\}.
\]
If the minimum selects $3/4$, the bound $|f|/B\le1/2$ gives
$P_0(f/B)^2<t^2$.  Otherwise the same strict inequality follows from
$t\ge2\sigma/B$; it also holds when $\sigma=0$.
Moreover, note that
\[
t\ge N^{-1/2},\qquad
Bt\le2\sigma+B N^{-1/2},\qquad
\log(C/t)\le C\log N.
\]
Applying Theorem~2.14.2 of
\citet{VanderVaartWellner2023weak} to $\mathcal F/B$ yields
\begin{align*}
E\left[\sup_{f\in\mathcal F}|(P_{n,k}-P_0)f|
 \,\middle|\,\trainingsigma_k\right]
&\le C B\left\{\frac{J(t)}{\sqrt N}
             +\frac{J(t)^2}{Nt^2}\right\}\\
&\le C B\left\{t\sqrt{\frac{\log(C/t)}{N}}
             +\frac{\log(C/t)}{N}\right\}\\
&\le C_M\left\{\delta\sqrt{\frac{\log N}{N\ell}}
             +\frac{\log N}{N\ell}\right\}.
\end{align*}
The $N^{-1/2}$ floor in $t$ covers arbitrarily small $\delta$.
Now, Markov's inequality, followed by integration over the
training data and then by $M\to\infty$, proves the bound with $\log N$.
Here constants are uniform in the training data on the displayed
boundedness event, and its complement has arbitrarily small probability.
Finally, Assumption~\ref{assm:G-bandwidth} gives, for all sufficiently large $n$,
\[
\frac15\log n-O(1)\le\log(h_n^{-1})\le\frac1\kappa\log n,
\qquad \log(b_n^{-1})=\log(h_n^{-1})+O(1).
\]
Since $\underline\pi n\le N\le n$, replacing $\log N$ by
$\log(\ell^{-1})$ changes only constants and proves the statement.
\end{proof}

\begin{lemma}
\label{lem:conditional-integrated-kernel}
Under Assumptions~\ref{assm:basic.proxy}, \ref{assm:kernel},
and~\ref{assm:folds}, let
\(\ell\in\{h_n,b_n\}\), \(j\in\{0,1,2\}\), and let
\(\chi_{n,k}(w,a,x)\) be \(\trainingsigma_k\)-measurable.
Let \(\mathcal I\subseteq\mathcal A\).
Define
\[
    H_{n,k,\ell,j,u}(O)
    :=
    \int p_{\ell,u,j}(a)\chi_{n,k}(W,a,X)\,dF_0(a).
\]
If
\[
\sup_{\substack{a\in\mathcal I\\w,x}}
|\chi_{n,k}(w,a,x)|=O_p(1)
\]
and \(\{a:|a-u|\le\ell\}\subseteq\mathcal I\), then, for fixed \(u\),
\[
    (P_{n,k}-P_0)H_{n,k,\ell,j,u}
    =
    O_p(n_k^{-1/2}).
\]
If
\[
\sup_{\substack{a\in\mathcal A_{\delta_3}\\w,x}}
|\chi_{n,k}(w,a,x)|=O_p(1)
\]
and \(\{a:|a-u|\le\ell\}\subseteq\mathcal A_{\delta_3}\) for every
\(u\in\mathcal A_0\), then
\[
    \sup_{u\in\mathcal A_0}
    |(P_{n,k}-P_0)H_{n,k,\ell,j,u}|
    =
    O_p\left\{
    \sqrt{\frac{\log(\ell^{-1})}{n_k}}
    +
    \frac{\log(\ell^{-1})}{n_k}
    \right\}.
\]
\end{lemma}

\begin{proof}
Fix the training data and an event on which the stated supremum of
$|\chi_{n,k}|$ is at most a deterministic $M\ge1$.  Note that such events have
arbitrarily high probability as $M$ increases.  Write
$\phi_j(t)=t^jK(t)$ and $M_f=\sup_a f_A(a)$.
For every admissible center $u$, we have
\[
|H_{n,k,\ell,j,u}(O)|
\le M M_f\int|p_{\ell,u,j}(a)|\,da
=M M_f\int|\phi_j(t)|\,dt=:B_M.
\]
We have
\[
E\left[\{(P_{n,k}-P_0)H_{n,k,\ell,j,u}\}^2
 \,\middle|\,\trainingsigma_k\right]\le B_M^2/n_k.
\]
Applying Chebyshev's inequality proves the
fixed-center conclusion.

For uniform inference, define the following function class:
\[
\mathcal H=\{H_{n,k,\ell,j,u}:u\in\mathcal A_0\}.
\]
The function $\phi_j$ is globally Lipschitz and supported on $[-1,1]$.
For $|d|\le1$, the union of the supports of
$\phi_j(\cdot)$ and $\phi_j(\cdot+d)$ has length at most three.  Hence
\[
\int|\phi_j(t)-\phi_j(t+d)|\,dt
\le3\operatorname{Lip}(\phi_j)|d|.
\]
For $|d|>1$ the integral is at most $2\|\phi_j\|_1$.
It follows, by the substitution $a=u+\ell t$, that
\begin{align*}
\|H_{n,k,\ell,j,u}-H_{n,k,\ell,j,v}\|_\infty
&\le M M_f\int|p_{\ell,u,j}(a)-p_{\ell,v,j}(a)|\,da\\
&\le C M\min\{1,|u-v|/\ell\}.
\end{align*}
The integrands vanish outside the union of the two kernel supports, so
only the assumed bound on $\chi_{n,k}$ over $\mathcal A_{\delta_3}$ is used.
Cover a fixed bounded interval containing $\mathcal A_0$ by intervals of
length $c\ell\varepsilon$, and select one point of $\mathcal A_0$ from
each nonempty interval.  The preceding bound gives, for $0<\ell\le1$,
\[
\sup_QN(\varepsilon B_M,\mathcal H,L_2(Q))
\le N(\varepsilon B_M,\mathcal H,L_\infty)
\le\frac{C}{\ell\varepsilon},\qquad 0<\varepsilon<1.
\]
Define $\mathcal H/B_M:=\{H/B_M:H\in\mathcal H\}$ and its
uniform entropy integral by
\[
J(t):=\sup_Q\int_0^t
 \sqrt{1+\log N(\varepsilon,\mathcal H/B_M,L_2(Q))}\,d\varepsilon,
\]
where $Q$ ranges over finitely supported probability measures.
Then
\[
J(1)\le C\int_0^1\sqrt{\log\{C/(\ell\varepsilon)\}}\,d\varepsilon
\le C\sqrt{\log(e/\ell)}.
\]
We apply Theorem~2.14.1 of
\citet{VanderVaartWellner2023weak} again.  Its $p=1$ bound,
conditionally on the training data, gives
\[
E\left[\sup_{u\in\mathcal A_0}
 |(P_{n,k}-P_0)H_{n,k,\ell,j,u}|
 \,\middle|\,\trainingsigma_k\right]
\le C B_M\sqrt{\frac{\log(e/\ell)}{n_k}}.
\]
Finally, Markov's inequality proves the bound
$O_p[\{\log(e/\ell)/n_k\}^{1/2}]$.
Because $\ell\to0$, this implies the statement.
\end{proof}

\begin{lemma}
\label{lem:R2-control}
Under Assumptions~\ref{assm:A}, \ref{assm:D}--\ref{assm:F},
\[
    R_{n,2}(a_0)=o_p\{(nh_n)^{-1/2}\}.
\]
Under Assumptions~\ref{assm:A}, \ref{assm:D}, and~\ref{assm:G},
\[
    \sup_{a_0\in\mathcal A_0}|R_{n,2}(a_0)|
    =
    o_p\{(nh_n\log n)^{-1/2}\}.
\]
\end{lemma}

\begin{proof}
For either neighborhood $\mathcal I$ in the statement, let
$\mathcal S=\mathcal I\times\mathcal Z\times\mathcal W\times\mathcal X$.
We use the defined (conditional) norms
$d_{q,k}(\hat q_n^{(-k)},q_\infty;\mathcal I,\mathcal S)$ and
$d_{h,k}(\hat h_n^{(-k)},h_\infty;\mathcal I,\mathcal S)$.
We remark that all expectations involving a fit are conditional on
$\trainingsigma_k$ and integrate an independent observation.

Assumption~\ref{assm:basic.proxy} gives, for almost every $a$,
\[
\begin{aligned}
 f(a\mid W,X)&=E_0\{f(a\mid U,X)\mid W,X\}\ge c_1,\\
 dP_{0,WX\mid A=a}
 &=\frac{f(a\mid W,X)}{f_A(a)}\,dP_{0,WX}.
\end{aligned}
\]
Jensen's inequality and this change of measure imply
\begin{align*}
&|S_{\hat h_n^{(-k)}}(a)-m_\infty(a)|^2\\
&\quad\le\int\{\hat h_n^{(-k)}(w,a,x)-h_\infty(w,a,x)\}^2
 \,dP_{0,WX}(w,x)\\
&\quad=E_0\left[\{\hat h_n^{(-k)}(W,a,X)-h_\infty(W,a,X)\}^2
 \frac{f_A(a)}{f(a\mid W,X)}
 \,\middle|\,A=a,\trainingsigma_k\right]\\
&\quad\le\frac{\sup_a f_A(a)}{c_1}
 d_{h,k}(\hat h_n^{(-k)},h_\infty;\mathcal I,\mathcal S)^2.
\end{align*}
Define $\psi_{n,k}=\hat\vartheta_{n,k}-\vartheta_\infty$ as a random
function of $O$.  Subtraction gives the exact identity
\begin{align*}
\psi_{n,k}(O)
={}&\{\hat q_n^{(-k)}(Z,A,X)-q_\infty(Z,A,X)\}
       \{Y-h_\infty(W,A,X)\}\\
&-\hat q_n^{(-k)}(Z,A,X)
       \{\hat h_n^{(-k)}(W,A,X)-h_\infty(W,A,X)\}\\
&+S_{\hat h_n^{(-k)}}(A)-m_\infty(A).
\end{align*}

For pointwise inference take $\mathcal I=B_{\delta_1}(a_0)$.
Assumption~\ref{assm:regul1}(c) and boundedness of $h_\infty$ imply
\[
\sup_{a\in\mathcal I}
E_0\{|Y-h_\infty(W,a,X)|^{2+\delta_2}
       \mid A=a,\trainingsigma_k\}
\le2^{1+\delta_2}(C_6+C_h^{2+\delta_2}).
\]
H\"older's inequality with exponents $(2+\delta_2)/2$ and
$(2+\delta_2)/\delta_2$ gives
\begin{align*}
&E_0\left[\{\hat q_n^{(-k)}(Z,a,X)-q_\infty(Z,a,X)\}^2
 \{Y-h_\infty(W,a,X)\}^2
 \,\middle|\,A=a,\trainingsigma_k\right]\\
&\quad\le
 \left[E_0\{|Y-h_\infty(W,a,X)|^{2+\delta_2}
       \mid A=a,\trainingsigma_k\}\right]^{2/(2+\delta_2)}\\
&\qquad\times
 \left[E_0\left\{
 |\hat q_n^{(-k)}(Z,a,X)-q_\infty(Z,a,X)|^{2(2+\delta_2)/\delta_2}
 \,\middle|\,A=a,\trainingsigma_k\right\}\right]^{\delta_2/(2+\delta_2)}\\
&\quad\le C\left[(2C_q)^{4/\delta_2}
 E_0\{[\hat q_n^{(-k)}(Z,a,X)-q_\infty(Z,a,X)]^2
       \mid A=a,\trainingsigma_k\}\right]^{\delta_2/(2+\delta_2)}\\
&\quad\le C
 d_{q,k}(\hat q_n^{(-k)},q_\infty;\mathcal I,\mathcal S)^{2\delta_2/(2+\delta_2)}.
\end{align*}
The second inequality uses
$|\hat q_n^{(-k)}-q_\infty|\le2C_q$.  Also,
\begin{align*}
&E_0\left[\hat q_n^{(-k)}(Z,a,X)^2
 \{\hat h_n^{(-k)}(W,a,X)-h_\infty(W,a,X)\}^2
 \,\middle|\,A=a,\trainingsigma_k\right]\\
&\quad\le C_q^2
 d_{h,k}(\hat h_n^{(-k)},h_\infty;\mathcal I,\mathcal S)^2.
\end{align*}
Taking square roots and applying Minkowski's inequality to the three terms in
$\psi_{n,k}$ gives
\begin{equation}
\begin{aligned}
\delta_{\psi,k}(\mathcal I)
\le C\big\{d_{q,k}(\hat q_n^{(-k)},q_\infty;
                     \mathcal I,\mathcal S)^{\delta_2/(2+\delta_2)}+d_{h,k}(\hat h_n^{(-k)},h_\infty;
                     \mathcal I,\mathcal S)\big\}.
\end{aligned}
\label{eq:R2-pointwise-pseudo-radius}
\end{equation}
Assumption~\ref{assm:pointwise-pseudooutcome-rate} therefore gives
$\max_k\delta_{\psi,k}(\mathcal I)=o_p(1)$.

For uniform inference take $\mathcal I=\mathcal A_{\delta_3}$.
By Assumption~\ref{assm:G-smooth}(iii), $|Y|$ is bounded almost surely,
so the first term of $\psi_{n,k}$ instead satisfies
\begin{align*}
&E_0\left[\{\hat q_n^{(-k)}(Z,a,X)-q_\infty(Z,a,X)\}^2
 \{Y-h_\infty(W,a,X)\}^2
 \,\middle|\,A=a,\trainingsigma_k\right]\\
&\quad\le C
 d_{q,k}(\hat q_n^{(-k)},q_\infty;\mathcal I,\mathcal S)^2.
\end{align*}
The other two bounds are unchanged.  Minkowski's inequality and
Assumption~\ref{assm:G-nuisance-L2} thus give
\begin{equation}
\begin{aligned}
&\delta_{\psi,k}(\mathcal I)
\le C\big\{d_{q,k}(\hat q_n^{(-k)},q_\infty;\mathcal I,\mathcal S)
          +d_{h,k}(\hat h_n^{(-k)},h_\infty;\mathcal I,\mathcal S)\big\},\\
&\max_k\delta_{\psi,k}(\mathcal A_{\delta_3})
\sqrt{\log(h_n^{-1})\log n}
=o_p(1).
\end{aligned}
\label{eq:R2-uniform-pseudo-radius}
\end{equation}
The formula for $\psi_{n,k}$ and boundedness of $Y$, the fitted bridges,
and their limits also show that
\[
\max_k\|\psi_{n,k}I(A\in\mathcal A_{\delta_3})\|_{P_0,\infty}=O_p(1).
\]
The conditional-law identities above hold for almost every $a$, which
suffices for the kernel integrals.  On any exceptional $f_A$-null set,
choose the conditional second moment of $\psi_{n,k}$ to be zero for the
application of Lemma~\ref{lem:conditional-kernel-ep}.  This changes
neither population integrals nor empirical averages, and the displayed
upper bounds continue to hold.

Split the remainder exactly as
\[
\begin{aligned}
R_{n,2}^{(1)}(u)
&=\sum_{k=1}^K\frac{n_k}{n}(P_{n,k}-P_0)
       \{\Gamma_{0,h,b,u}(A)\psi_{n,k}(O)\},\\
R_{n,2}^{(2)}(u)
&=\sum_{k=1}^K\frac{n_k}{n}(P_{n,k}-P_0)
 \int\Gamma_{0,h,b,u}(a)
       \{\hat h_n^{(-k)}(W,a,X)-h_\infty(W,a,X)\}\,dF_0(a).
\end{aligned}
\]
Thus $R_{n,2}=R_{n,2}^{(1)}+R_{n,2}^{(2)}$ by its definition.
We have
\begin{align*}
\Gamma_{0,h,b,u}(a)
=\sum_{j=0}^1[e_1^\top D_{0,h,u,1}^{-1}]_{j+1}p_{h,u,j}(a)-\tau_n^2c_{0,h,u,2}
  \sum_{j=0}^2[e_3^\top D_{0,b,u,2}^{-1}]_{j+1}p_{b,u,j}(a).
\end{align*}
Lemma~\ref{lem:D-expansion} gives
\[
 \|D_{0,h,u,1}^{-1}\|_\infty
 +\|D_{0,b,u,2}^{-1}\|_\infty+|c_{0,h,u,2}|\le C,
\]
at $u=a_0$ for the pointwise assertion and for every $u\in\mathcal A_0$
for the uniform assertion, for all sufficiently large $n$.
For the integrated terms use $H_{n,k,\ell,j,u}$ of
Lemma~\ref{lem:conditional-integrated-kernel} with
\[
\chi_{n,k}(w,a,x)=\hat h_n^{(-k)}(w,a,x)-h_\infty(w,a,x),
\qquad
\sup_{a\in\mathcal I,w,x}|\chi_{n,k}(w,a,x)|\le2C_h.
\]
The triangle inequality and the displayed row expansion give
\begin{equation}
\label{eq:R2-basis-bounds}
\begin{aligned}
 |R_{n,2}^{(1)}(u)|
 \le C\sum_{k=1}^K\frac{n_k}{n}
 \Big[&\max_{0\le j\le1}
  |(P_{n,k}-P_0)\{p_{h,u,j}(A)\psi_{n,k}(O)\}|\\
 &+\tau_n^2\max_{0\le j\le2}
  |(P_{n,k}-P_0)\{p_{b,u,j}(A)\psi_{n,k}(O)\}|\Big],\\
 |R_{n,2}^{(2)}(u)|
 \le C\sum_{k=1}^K\frac{n_k}{n}
 \Big[&\max_{0\le j\le1}
  |(P_{n,k}-P_0)H_{n,k,h,j,u}|\\
 &+\tau_n^2\max_{0\le j\le2}
  |(P_{n,k}-P_0)H_{n,k,b,j,u}|\Big].
\end{aligned}
\end{equation}

For fixed $a_0$, Lemma~\ref{lem:conditional-kernel-ep} and
\eqref{eq:R2-pointwise-pseudo-radius}, with $n_k\asymp n$, give
\[
R_{n,2}^{(1)}(a_0)
=o_p\{(nh)^{-1/2}\}
 +\tau_n^2o_p\{(nb)^{-1/2}\}
=o_p\{(nh)^{-1/2}\},
\]
since $\tau_n^2(nb)^{-1/2}=\tau_n^{5/2}(nh)^{-1/2}$ and
$\tau_n=O(1)$.
Lemma~\ref{lem:conditional-integrated-kernel} implies
\[
R_{n,2}^{(2)}(a_0)=O_p\{(1+\tau_n^2)n^{-1/2}\}
=o_p\{(nh)^{-1/2}\},
\]
where the last equality follows from $h\to0$.

For the uniform assertion, Assumption~\ref{assm:G-bandwidth} gives
\[
b\asymp h,\qquad
\log(b^{-1})\asymp\log(h^{-1})\asymp\log n,\qquad
nh\ge n^{1-1/\kappa}
\]
for all sufficiently large $n$.  Applying
Lemma~\ref{lem:conditional-kernel-ep} to each term of the row expansion
therefore yields
\[
\sup_u|R_{n,2}^{(1)}(u)|
=O_p\left\{
 \max_k\delta_{\psi,k}(\mathcal A_{\delta_3})
       \sqrt{\frac{\log(h^{-1})}{nh}}
 +\frac{\log(h^{-1})}{nh}\right\}.
\]
After multiplication by $\sqrt{nh\log n}$, the first term is $o_p(1)$ by
\eqref{eq:R2-uniform-pseudo-radius}.  The second tends to zero because
\[
\frac{\log(h^{-1})\sqrt{\log n}}{\sqrt{nh}}
\le\frac{C(\log n)^{3/2}}{n^{(1-1/\kappa)/2}}\longrightarrow0.
\]
Lemma~\ref{lem:conditional-integrated-kernel}, applied to the same
expansion, gives
\[
\sup_u|R_{n,2}^{(2)}(u)|
=O_p\left\{\sqrt{\frac{\log(h^{-1})}{n}}
             +\frac{\log(h^{-1})}{n}\right\}.
\]
The two ratios to $(nh\log n)^{-1/2}$ are
\[
\sqrt{h\log(h^{-1})\log n}
\quad\hbox{and}\quad
\sqrt{h\log(h^{-1})\log n}\sqrt{\log(h^{-1})/n}.
\]
Both tend to zero: $nh^5=O(1)$ implies
\[
h\log(h^{-1})\log n\le Cn^{-1/5}(\log n)^2\longrightarrow0,
\qquad \log(h^{-1})/n\longrightarrow0.
\]
Combining the two parts proves the uniform statement.
\end{proof}

\begin{lemma}
\label{lem:Gamma-difference-expansion}
Under Assumptions~\ref{assm:D}--\ref{assm:F}, for fixed \(a_0\),
\[
\begin{aligned}
\Gamma_n(a)-\Gamma_0(a)
=\sum_{j=0}^1\alpha_{n,h,j}(a_0)
p_{h_n,a_0,j}(a)
+\tau_n^2\sum_{j=0}^2\alpha_{n,b,j}(a_0)
p_{b_n,a_0,j}(a),
\end{aligned}
\]
where
\[
\max_{0\le j\le1}|\alpha_{n,h,j}(a_0)|
+\max_{0\le j\le2}|\alpha_{n,b,j}(a_0)|
=O_p\{(nh_n)^{-1/2}\}.
\]
Under Assumptions~\ref{assm:D} and~\ref{assm:G},
\[
\begin{aligned}
\sup_{a_0\in\mathcal A_0}
\left\{
\max_{0\le j\le1}|\alpha_{n,h,j}(a_0)|
+\max_{0\le j\le2}|\alpha_{n,b,j}(a_0)|
\right\}
=O_p\!\left\{\sqrt{\frac{\log(h_n^{-1})}{nh_n}}\right\}.
\end{aligned}
\]
\end{lemma}

\begin{proof}
On the event that the empirical
moment matrices are nonsingular, define
\[
\begin{aligned}
\alpha_{n,h,j}(a_0)
:=
\left[e_1^\top
\{D_{n,h,a_0,1}^{-1}-D_{0,h,a_0,1}^{-1}\}\right]_{j+1},
\qquad j=0,1,
\end{aligned}
\]
and
\[
\begin{aligned}
\alpha_{n,b,j}(a_0)
:=-\left[
c_{n,h,a_0,2}e_3^\top D_{n,b,a_0,2}^{-1}
-c_{0,h,a_0,2}e_3^\top D_{0,b,a_0,2}^{-1}
\right]_{j+1},
\qquad j=0,1,2.
\end{aligned}
\]
If either empirical matrix is singular, define instead
\[
\begin{aligned}
\alpha_{n,h,j}(a_0)
&:=-[e_1^\top D_{0,h,a_0,1}^{-1}]_{j+1},
&&j=0,1,\\
\alpha_{n,b,j}(a_0)
&:=c_{0,h,a_0,2}[e_3^\top D_{0,b,a_0,2}^{-1}]_{j+1},
&&j=0,1,2.
\end{aligned}
\]
These definitions give \(\Gamma_n-\Gamma_0=-\Gamma_0\) under the
convention \(\Gamma_n=0\) on that event.
Subtraction of the two definitions of \(\Gamma\) gives the displayed
expansion on the nonsingular event.  On this event, we have
\[
\max_{0\le j\le1}|\alpha_{n,h,j}(a_0)|
\le
\|D_{n,h,a_0,1}^{-1}-D_{0,h,a_0,1}^{-1}\|_\infty
\]
and
\[
\begin{aligned}
\max_{0\le j\le2}|\alpha_{n,b,j}(a_0)|
&\leq
|c_{n,h,a_0,2}-c_{0,h,a_0,2}|
\|D_{n,b,a_0,2}^{-1}\|_\infty
\\
&\quad
+|c_{0,h,a_0,2}|
\|D_{n,b,a_0,2}^{-1}-D_{0,b,a_0,2}^{-1}\|_\infty.
\end{aligned}
\]
Indeed, the second inequality follows by writing the row vector inside
the definition of \(-\alpha_{n,b,j}\) as
\[
\begin{aligned}
(c_{n,h,a_0,2}-c_{0,h,a_0,2})e_3^\top D_{n,b,a_0,2}^{-1}
+
c_{0,h,a_0,2}e_3^\top
\{D_{n,b,a_0,2}^{-1}-D_{0,b,a_0,2}^{-1}\}.
\end{aligned}
\]
By Lemma~\ref{lem:D-expansion} and the moment-matrix estimates in the
proof of Lemma~\ref{lem:R6-rate}, we obtain
\[
\begin{aligned}
&\|D_{n,h,a_0,1}^{-1}-D_{0,h,a_0,1}^{-1}\|_\infty
+|c_{n,h,a_0,2}-c_{0,h,a_0,2}|
=O_p\{(nh)^{-1/2}\},
\\
&\|D_{n,b,a_0,2}^{-1}-D_{0,b,a_0,2}^{-1}\|_\infty
=O_p\{(nb)^{-1/2}\},
\\
&\|D_{n,b,a_0,2}^{-1}\|_\infty+|c_{0,h,a_0,2}|
=O_p(1).
\end{aligned}
\]
Since \((nb)^{-1/2}=\tau_n^{1/2}(nh)^{-1/2}\) and
\(\tau_n=O(1)\), the pointwise conclusion follows.

For the uniform assertion, Assumption~\ref{assm:G-bandwidth} gives
\(b\asymp h\), \(\log(b^{-1})\asymp\log(h^{-1})\), and
\(nh/\log(h^{-1})\to\infty\).  Thus the uniform matrix estimates in the
proof of Lemma~\ref{lem:R6-rate} apply and give
\[
\begin{aligned}
&\sup_{a_0\in\mathcal A_0}
\|D_{n,h,a_0,1}^{-1}-D_{0,h,a_0,1}^{-1}\|_\infty
=O_p\!\left\{\sqrt{\frac{\log(h^{-1})}{nh}}\right\},
\\
&\sup_{a_0\in\mathcal A_0}
\|D_{n,b,a_0,2}^{-1}-D_{0,b,a_0,2}^{-1}\|_\infty
=O_p\!\left\{\sqrt{\frac{\log(b^{-1})}{nb}}\right\},
\\
&\sup_{a_0\in\mathcal A_0}
|c_{n,h,a_0,2}-c_{0,h,a_0,2}|
=O_p\!\left\{\sqrt{\frac{\log(h^{-1})}{nh}}\right\},
\\
&\sup_{a_0\in\mathcal A_0}
\{\|D_{n,b,a_0,2}^{-1}\|_\infty+|c_{0,h,a_0,2}|\}
=O_p(1).
\end{aligned}
\]
Taking suprema in the coefficient inequalities now yields
\[
\begin{aligned}
&\sup_{a_0\in\mathcal A_0}
\max_{0\le j\le1}|\alpha_{n,h,j}(a_0)|
\leq
\sup_{a_0\in\mathcal A_0}
\|D_{n,h,a_0,1}^{-1}-D_{0,h,a_0,1}^{-1}\|_\infty,
\\
&\sup_{a_0\in\mathcal A_0}
\max_{0\le j\le2}|\alpha_{n,b,j}(a_0)|
\leq
\sup_{a_0\in\mathcal A_0}|c_{n,h,a_0,2}-c_{0,h,a_0,2}|
\sup_{a_0\in\mathcal A_0}\|D_{n,b,a_0,2}^{-1}\|_\infty
\\
&\qquad\qquad\qquad\qquad\qquad\quad+
\sup_{a_0\in\mathcal A_0}|c_{0,h,a_0,2}|
\sup_{a_0\in\mathcal A_0}
\|D_{n,b,a_0,2}^{-1}-D_{0,b,a_0,2}^{-1}\|_\infty.
\end{aligned}
\]
These bounds prove the uniform conclusion.
\end{proof}

\begin{lemma}
\label{lem:R3-control}
Under Assumptions~\ref{assm:A} and~\ref{assm:D}--\ref{assm:F},
\[
R_{n,3}(a_0)=O_p\{(nh_n)^{-1}\}
=o_p\{(nh_n)^{-1/2}\}.
\]
Under Assumptions~\ref{assm:A}, \ref{assm:D}, and~\ref{assm:G},
\[
\sup_{a_0\in\mathcal A_0}|R_{n,3}(a_0)|
=o_p\{(nh_n\log n)^{-1/2}\}.
\]
\end{lemma}

\begin{proof}
For fixed \(a_0\), the supports of
\(p_{h,a_0,j}\) and \(p_{b,a_0,j}\) are contained in
\(\mathcal I=B_{\delta_1}(a_0)\) for all sufficiently large \(n\).
On the events in Assumption~\ref{assm:nuisance1}, one has
\[
\begin{aligned}
E_0\{\hat\vartheta_{n,k}^2\mid A=a,\trainingsigma_k\}
&\le
3C_q^2E_0(Y^2\mid A=a)+3C_q^2C_h^2+3C_h^2
\end{aligned}
\]
for \(a\in\mathcal I\), since
\(|S_{\hat h_n^{(-k)}}(a)|\le C_h\).
Assumption~\ref{assm:regul1}\textup{(c)} and Jensen's inequality yields
\[
\sup_{a\in\mathcal I}E_0(Y^2\mid A=a)
\le C_6^{2/(2+\delta_2)}.
\]
Thus the quantity defined in Lemma~\ref{lem:conditional-kernel-ep}
satisfies
\[
\max_k\delta_{\hat\vartheta,k}(\mathcal I)=O_p(1).
\]
Therefore we have
\[
\left|(P_{n,k}-P_0)
\{p_{\ell,a_0,j}(A)\hat\vartheta_{n,k}(O)\}\right|
=O_p\{(n_k\ell)^{-1/2}\},
\qquad \ell\in\{h,b\}.
\]
Lemma~\ref{lem:conditional-integrated-kernel}, applied with
\(\chi_{n,k}=\hat h_n^{(-k)}\), gives
\[
\left|(P_{n,k}-P_0)
\left\{\int p_{\ell,a_0,j}(a)
\hat h_n^{(-k)}(W,a,X)\,dF_0(a)\right\}\right|
=O_p(n_k^{-1/2}).
\]

Substitute Lemma~\ref{lem:Gamma-difference-expansion} into the
definition of \(R_{n,3}\).
We have
\[
\begin{aligned}
R_{n,3}(a_0)
={}&\sum_{k=1}^K\frac{n_k}{n}
\sum_{j=0}^1\alpha_{n,h,j}(a_0)
\Bigg[
(P_{n,k}-P_0)
\{p_{h,a_0,j}(A)\hat\vartheta_{n,k}(O)\}
\\
&\qquad+
(P_{n,k}-P_0)
\left\{\int p_{h,a_0,j}(a)
\hat h_n^{(-k)}(W,a,X)\,dF_0(a)\right\}
\Bigg]
\\
&+\tau_n^2\sum_{k=1}^K\frac{n_k}{n}
\sum_{j=0}^2\alpha_{n,b,j}(a_0)
\Bigg[
(P_{n,k}-P_0)
\{p_{b,a_0,j}(A)\hat\vartheta_{n,k}(O)\}
\\
&\qquad+
(P_{n,k}-P_0)
\left\{\int p_{b,a_0,j}(a)
\hat h_n^{(-k)}(W,a,X)\,dF_0(a)\right\}
\Bigg].
\end{aligned}
\]
Note that the coefficients are computed from the full sample.  
Since \(K\) is
fixed, \(n_k\asymp n\), and \(b^{-1/2}\le Ch^{-1/2}\), the triangle
inequality and Lemma~\ref{lem:Gamma-difference-expansion} now give
\[
\begin{aligned}
|R_{n,3}(a_0)|
&=O_p\!\left[
(nh)^{-1/2}
\sum_{k=1}^K\frac{n_k}{n}
\left\{
\frac1{\sqrt{n_kh}}
+\frac{\tau_n^2}{\sqrt{n_kb}}
+\frac{1+\tau_n^2}{\sqrt{n_k}}
\right\}\right]
\\
&=O_p\{(nh)^{-1}+(n\sqrt h)^{-1}\}
=O_p\{(nh)^{-1}\}.
\end{aligned}
\]

For the uniform assertion, take \(\mathcal I=\mathcal A_{\delta_3}\).
The kernel supports are contained in this set uniformly for
\(a_0\in\mathcal A_0\) for all sufficiently large \(n\).
Assumptions~\ref{assm:G-crossfit} and~\ref{assm:G-smooth}\textup{(iii)}
imply
\[
\max_k
\|\hat\vartheta_{n,k}(O)I(A\in\mathcal A_{\delta_3})\|_{P_0,\infty}
=O_p(1),
\qquad
\max_k\delta_{\hat\vartheta,k}(\mathcal A_{\delta_3})=O_p(1).
\]
The uniform assertions of
Lemmas~\ref{lem:conditional-kernel-ep}
and~\ref{lem:conditional-integrated-kernel} thus imply, respectively,
\[
\begin{aligned}
\sup_{a_0\in\mathcal A_0}
\left|(P_{n,k}-P_0)
\{p_{\ell,a_0,j}(A)\hat\vartheta_{n,k}(O)\}\right|
=O_p\!\left\{
\sqrt{\frac{\log(\ell^{-1})}{n_k\ell}}
+\frac{\log(\ell^{-1})}{n_k\ell}\right\},
\end{aligned}
\]
and
\[
\begin{aligned}
&\sup_{a_0\in\mathcal A_0}
\left|(P_{n,k}-P_0)
\left\{\int p_{\ell,a_0,j}(a)
\hat h_n^{(-k)}(W,a,X)\,dF_0(a)\right\}\right|
\\
&=O_p\!\left\{
\sqrt{\frac{\log(\ell^{-1})}{n_k}}
+\frac{\log(\ell^{-1})}{n_k}\right\},
\end{aligned}
\]
for $\ell\in\{h,b\}$.
Taking the supremum in the exact expansion of \(R_{n,3}\), applying the
triangle inequality to its finitely many terms, and using
\(b\asymp h\), \(n_k\asymp n\), and
Lemma~\ref{lem:Gamma-difference-expansion}, we obtain
\[
\begin{aligned}
\sup_{a_0\in\mathcal A_0}|R_{n,3}(a_0)|
&=O_p\!\left[
\sqrt{\frac{\log(h^{-1})}{nh}}
\left\{
\sqrt{\frac{\log(h^{-1})}{nh}}
+\frac{\log(h^{-1})}{nh}
\right\}\right]
\\
&\qquad+O_p\!\left[
\sqrt{\frac{\log(h^{-1})}{nh}}
\left\{\sqrt{\frac{\log(h^{-1})}{n}}
+\frac{\log(h^{-1})}{n}
\right\}\right]
\\
&=O_p\!\left\{\frac{\log(h^{-1})}{nh}\right\}.
\end{aligned}
\]
For the last equality, \(h\to0\) and
\(\log(h^{-1})/(nh)\to0\) follow from
Assumption~\ref{assm:G-bandwidth}.  The same assumption also gives
\(\log(h^{-1})=O(\log n)\) and
\(nh\ge n^{1-1/\kappa}\) for all sufficiently large \(n\).  Therefore,
\[
\frac{\log(h^{-1})/(nh)}{(nh\log n)^{-1/2}}
=\log(h^{-1})\sqrt{\frac{\log n}{nh}}
\lesssim\frac{(\log n)^{3/2}}{n^{(1-1/\kappa)/2}}
\longrightarrow0,
\]
which proves the result.
\end{proof}

\section{Assumptions and proofs for derivative inference}
\label{sec:derivative-assumptions-proofs}

This section states the derivative conditions and proves the results in
Section~\ref{sec:derivative-inference-main}.  The definitions of
$d_{q,k}$, $d_{h,k}$, $B_{\delta}(a_0)$, $\mathcal A_\delta$, and
$\trainingsigma_k$ are used throughout Section~\ref{sec:debiased-inference-main}.

\begin{assumptionp}{H}\label{assm:deriv-model}
For fixed $a_0\in\mathcal A$ and fixed $h,b>0$, a law $P$ belongs to the
derivative model whenever the following properties hold.
\begin{enumerate}[label=\textup{(H\arabic*)},ref=H\arabic*,leftmargin=*,itemsep=0.35em]
\item\label{assm:deriv-model-bridges}
The treatment has a marginal Lebesgue density $f_P$, and there are unique
bridge functions $h_P$ and $q_P$ satisfying
\[
\int E_P\{h_P(W,a,X)^2\}\,dF_P(a)<\infty,
\qquad E_P\{q_P(Z,A,X)^2\}<\infty,
\]
\[
E_P\{Y-h_P(W,A,X)\mid Z,A,X\}=0,
\]
and, for $F_P$-almost every $a$,
\[
E_P\{q_P(Z,a,X)\mid W,A=a,X\}
=\frac{f_P(a)}{f_P(a\mid W,X)}.
\]
\item\label{assm:deriv-model-positivity}
There is $\kappa>0$ such that
\[
f_P(a\mid W,X)\ge\kappa
\]
for $F_P$-almost every $a$ with $|a-a_0|\le\max(h,b)$,
$P_{WX}$-almost surely.
\item\label{assm:deriv-bridge-path}
For every regular one-dimensional submodel in this model
$\{P_\varepsilon:|\varepsilon|<\delta\}$ through $P_0$, with score
$\dot\ell_0\in L_2^0(P_0)$, there is
$\dot h_0\in L_2(P_{0,WX}\otimes F_0)$ such that
\[
\int\left|
\frac{h_{P_\varepsilon}(w,a,x)-h_{P_0}(w,a,x)}{\varepsilon}
-\dot h_0(w,a,x)\right|^2
\,d(P_{0,WX}\otimes F_0)(w,a,x)\longrightarrow0,
\]
\[
E_0\{q_0(Z,A,X)^2\dot h_0(W,A,X)^2\}<\infty,
\]
\[
E_{P_\varepsilon}\!\left[q_0(Z,A,X)^2
\left|\frac{h_{P_\varepsilon}(W,A,X)-h_{P_0}(W,A,X)}{\varepsilon}
-\dot h_0(W,A,X)\right|^2\right]\longrightarrow0,
\]
and, for some $\delta'>0$,
\[
\sup_{0<|\varepsilon|<\delta'}
E_{P_\varepsilon}\!\left[q_0(Z,A,X)^2
\{Y-h_{P_\varepsilon}(W,A,X)\}^2\right]<\infty.
\]
\item\label{assm:deriv-if-integrability}
\[
E_0\!\left[
\{q_0(Z,A,X)(Y-h_0(W,A,X))+\theta_0(A)\}^2\right]<\infty.
\]
\item\label{assm:deriv-tangent}
The observed-data tangent space at $P_0$ is $L_2^0(P_0)$.
\end{enumerate}
\end{assumptionp}

\begin{assumptionp}{I}\label{assm:deriv-pointwise}
For a fixed interior point $a_0$ and the neighborhood $B_{\delta_1}(a_0)$
in condition~\ref{assm:regul1}, the following conditions hold.
\begin{enumerate}[label=\textup{(I\arabic*)},ref=I\arabic*,leftmargin=*,itemsep=0.35em]
\item\label{assm:deriv-point-smooth}
The function $\theta_0$ is three times continuously differentiable on
$B_{\delta_1}(a_0)$.
\item\label{assm:deriv-point-moment}
$\sigma_0^2(a_0)>0$.
\item\label{assm:deriv-point-centering}
$nh_n^3\longrightarrow\infty$, and $nh_n^7=O(1)$.
\end{enumerate}
\end{assumptionp}

\begin{assumptionp}{J}\label{assm:deriv-uniform}
For $\mathcal A_0$ and $\mathcal A_{\delta_3}$ as defined in
Assumption~\ref{assm:G}, the following conditions hold.
\begin{enumerate}[label=\textup{(J\arabic*)},ref=J\arabic*,leftmargin=*,itemsep=0.35em]
\item\label{assm:deriv-uniform-smooth}
The function $\theta_0$ is three times continuously differentiable on
$\mathcal A_{\delta_3}$, and $\theta_0^{(3)}$ is H\"older continuous there
with exponent $\alpha\in(0,1]$.
\item\label{assm:deriv-uniform-bandwidth}
For the exponent $\alpha$ in condition~\ref{assm:deriv-uniform-smooth},
\[
(nh_n^3\log n)^{1/2}
\{h_n^{2+\alpha}+h_n^2b_n^\alpha\}\longrightarrow0.
\]
\end{enumerate}
\end{assumptionp}

Below, ``the pointwise derivative assumptions'' means the assumptions and
conditions stated in Theorem~\ref{thm:deriv-pointwise};
``the uniform derivative assumptions'' means those stated in
Theorem~\ref{thm:deriv-uniform}, before its conditional mesh assertion.
In these asymptotic statements, Assumption~\ref{assm:deriv-model} is imposed
with $(h,b)=(h_n,b_n)$ for all sufficiently large $n$; in a uniform statement
it is imposed with $a_0=a$ for every $a\in\mathcal A_0$.

\subsection{Pathwise derivative}

\begin{proof}[Proof of Theorem~\ref{thm:deriv-if}]
Fix $a=a_0$.
We first record the score calculation under the derivative model itself.
Let $r$ be bounded and satisfy $r(\bar a)=0$ whenever
$|\bar a-a|>\max(h,b)$.  Use the fixed-weight
functional $\Psi_r$ defined in
Lemma~\ref{lem:Fixed-weight pathwise derivative}, and write
\[
 h_\varepsilon=h_{P_\varepsilon},
\quad
\Delta_\varepsilon=\frac{h_\varepsilon-h_0}{\varepsilon}.
\]
Condition~\ref{assm:deriv-bridge-path} and differentiation of the two
marginal laws give
\begin{align*}
\left.\frac{d}{d\varepsilon}\Psi_r(P_\varepsilon)\right|_0
={}&\int r(\bar a)E_0\{\dot h_0(W,\bar a,X)\}\,dF_0(\bar a)\\
&+E_0\!\left[\left\{\int r(\bar a)h_0(W,\bar a,X)\,dF_0(\bar a)\right\}
\dot\ell_0(O)\right]
+E_0\{r(A)\theta_0(A)\dot\ell_0(O)\}.
\end{align*}
For every $\varepsilon$, the outcome-bridge equation gives
\[
E_{P_\varepsilon}[r(A)q_0(Z,A,X)\{Y-h_\varepsilon(W,A,X)\}]=0.
\]
Set $g_0=r(A)q_0(Z,A,X)\{Y-h_0(W,A,X)\}$.  Jensen's inequality and
condition~\ref{assm:deriv-model-bridges} give
$E_0\{\theta_0(A)^2\}<\infty$, so
condition~\ref{assm:deriv-if-integrability} gives $g_0\in L_2(P_0)$.
Consequently,
\begin{align*}
0
&=\frac{(P_\varepsilon-P_0)g_0}{\varepsilon}
-E_{P_\varepsilon}\{r(A)q_0(Z,A,X)\Delta_\varepsilon(W,A,X)\}\\
&\longrightarrow
E_0\{g_0(O)\dot\ell_0(O)\}
-E_0\{r(A)q_0(Z,A,X)\dot h_0(W,A,X)\}.
\end{align*}
The two weighted convergence requirements in
condition~\ref{assm:deriv-bridge-path} justify the second limit.  On the
support of $r$, condition~\ref{assm:deriv-model-positivity}, the treatment
bridge, and Bayes' formula give directly
\[
E_0\{r(A)q_0(Z,A,X)\dot h_0(W,A,X)\}
=\int r(\bar a)E_0\{\dot h_0(W,\bar a,X)\}\,dF_0(\bar a).
\]
Substitution in the preceding marginal derivative, followed by subtracting
the constant $\Psi_r(P_0)$ inside the score covariance, yields
\begin{align*}
\left.\frac{d}{d\varepsilon}\Psi_r(P_\varepsilon)\right|_0
=E_0\Bigg[\Bigg\{r(A)\xi_0(O)
+\int r(\bar a)\{h_0(W,\bar a,X)-\theta_0(\bar a)\}\,dF_0(\bar a)
\Bigg\}\dot\ell_0(O)\Bigg].
\end{align*}
Each coordinate of $w_{h,a,2}K_{h,a}$ and $w_{b,a,3}K_{b,a}$ satisfies the
support and boundedness requirements, so this identity applies
componentwise.

With
$D_{\varepsilon h}=D_{P_\varepsilon,h,a,2}$ and
$M_{\varepsilon h}=M_{P_\varepsilon,h,a,2}$, the score identity gives
\[
\left.\frac{d}{d\varepsilon}D_{\varepsilon h}\right|_0
=E_0\{w_{h,a,2}(A)w_{h,a,2}(A)^\top K_{h,a}(A)\dot\ell_0(O)\}.
\]
Apply the preceding componentwise calculation to $M_{\varepsilon h}$, and
use $dD^{-1}=-D^{-1}(dD)D^{-1}$.
This gives
\begin{align*}
\left.\frac{d}{d\varepsilon}\theta_{P_\varepsilon,h}^{(1)}(a)
\right|_{0}
={}&E_0\Bigg[\Bigg\{r_{0,h,a}^{(1)}(A)\xi_0(O)
-\gamma_{0,h,a}^{(1)}(A)
+\int r_{0,h,a}^{(1)}(\bar a)
\notag\\[-1mm]
&\hspace{16mm}\times
\{h_0(W,\bar a,X)-\theta_0(\bar a)\}\,dF_0(\bar a)
\Bigg\}\dot\ell_0(O)\Bigg].
\end{align*}
Replacing
\[
(h,w_{h,a,2},e_2,h^{-1})
\quad\text{by}\quad
(b,w_{b,a,3},e_4,6b^{-3})
\]
in these two calculations gives the corresponding formula with
$(r^{(1)},\gamma^{(1)})$ replaced by $(s^{(3)},\gamma^{(3)})$.  Direct
differentiation of
$c_{P,h,a,3}^{(1)}=e_2^\top D_{P,h,a,2}^{-1}U_{P,h,a,3}^{(1)}$ yields
\[
\left.\frac{d}{d\varepsilon}c_{P_\varepsilon,h,a,3}^{(1)}\right|_0
=E_0\{\gamma_{0,h,a}^{c,1}(A)\dot\ell_0(O)\}.
\]
The product rule in \eqref{eq:deriv-smoothed-target} now gives
$E_0\{\varphi_{0,h,b,a}^{(1)}\dot\ell_0\}$.
The moment-matrix definitions imply
\[
P_0\gamma_{0,h,a}^{(1)}=\theta_{0,h}^{(1)}(a),\qquad
P_0\gamma_{0,b,a}^{(3)}=\theta_{0,b}^{(3)}(a),\qquad
P_0\gamma_{0,h,a}^{c,1}=0.
\]
Together with $E_0\{\xi_0(O)\mid A\}=\theta_0(A)$, these identities show
$P_0\varphi_{0,h,b,a}^{(1)}=0$. 
Conditions \ref{assm:deriv-model-bridges} and~\ref{assm:deriv-if-integrability}, followed
by the Cauchy--Schwarz inequality, give
$\varphi_{0,h,b,a}^{(1)}\in L_2^0(P_0)$, and \ref{assm:deriv-tangent} then makes this gradient the efficient influence function.
\end{proof}

\subsection{Derivative decomposition and remainder bounds}
\label{sec:derivative-remainder-proofs}

Throughout this subsection $h=h_n$, $b=b_n$, and $\tau_n=h/b$.  At a
fixed evaluation point $a_0$, abbreviate
\[
\Gamma_n^{(1)}=\Gamma_{n,h,b,a_0}^{(1)},\qquad
\Gamma_0^{(1)}=\Gamma_{0,h,b,a_0}^{(1)},\qquad
\gamma_0^{(1)}=\gamma_{0,h,b,a_0}^{(1)}.
\]
Use $\hat\vartheta_{n,k}$ from Section~\ref{sec:R2-R3-control}, and set
\begin{align}
R_{n,1}^{(1)}(a_0)
&:=\theta_{0,h,b}^{(1),\mathrm{DB}}(a_0)-\theta_0'(a_0),
\label{eq:deriv-R1}\\
R_{n,2}^{(1)}(a_0)
&:=\sum_{k=1}^K\frac{n_k}{n}(P_{n,k}-P_0)
\Bigg[\Gamma_0^{(1)}(A)\{\hat\vartheta_{n,k}(O)-\xi_\infty(O)\}
\notag\\
&\quad+\int\Gamma_0^{(1)}(\bar a)
\{\hat h_n^{(-k)}(W,\bar a,X)-h_\infty(W,\bar a,X)\}
\,dF_0(\bar a)\Bigg],
\label{eq:deriv-R2}\\
R_{n,3}^{(1)}(a_0)
&:=\sum_{k=1}^K\frac{n_k}{n}(P_{n,k}-P_0)
\Bigg[(\Gamma_n^{(1)}-\Gamma_0^{(1)})(A)
\hat\vartheta_{n,k}(O)\notag\\
&\quad+\int(\Gamma_n^{(1)}-\Gamma_0^{(1)})(\bar a)
\hat h_n^{(-k)}(W,\bar a,X)\,dF_0(\bar a)\Bigg],
\label{eq:deriv-R3}\\
R_{n,4}^{(1)}(a_0)
&:=\sum_{k=1}^K\frac{n_k}{n}\int\Gamma_n^{(1)}(a)
E_0\!\left[\{\hat q_n^{(-k)}-q_0\}
\{h_0-\hat h_n^{(-k)}\}
\,\middle|\, A=a,\trainingsigma_k\right]dF_0(a),
\label{eq:deriv-R4}\\
R_{n,5}^{(1)}(a_0)
&:=\sum_{k=1}^K\frac{n_k}{n}
\iint\Gamma_n^{(1)}(a)\hat h_n^{(-k)}(w,a,x)
\notag\\
&\hspace{15mm}\times d(Q_{n,k}^{WX}-Q_0^{WX})(w,x)
\,d(F_{n,k}-F_0)(a),
\label{eq:deriv-R5}\\
R_{n,6}^{(1)}(a_0)
&:=(P_n-P_0)\gamma_0^{(1)}(A)
+P_0\{(\Gamma_n^{(1)}-\Gamma_0^{(1)})(A)\theta_0(A)\}.
\label{eq:deriv-R6}
\end{align}
In \eqref{eq:deriv-R4}, the suppressed function arguments are
$(Z,a,X)$ and $(W,a,X)$.  As in
the proof of Lemma~\ref{lem:prox-six-rem-decomp}, $P_0$ integrates an independent
observation while holding the observed sample fixed, including the
full-sample design matrices.

\begin{lemma}\label{lem:deriv-decomposition}
On the event that the derivative design matrices are nonsingular,
\begin{equation}
\hat\theta_n^{(1),\mathrm{DB},\mathrm{cf}}(a_0)-\theta_0'(a_0)
=P_n\phi_{\infty,h,b,a_0}^{(1)}+\sum_{j=1}^6R_{n,j}^{(1)}(a_0).
\label{eq:deriv-decomposition}
\end{equation}
\end{lemma}
\noindent Since its proof closely follows that of Lemma~\ref{lem:prox-six-rem-decomp}, we omit the proof.

\begin{lemma}\label{lem:deriv-matrices}
At a fixed $a_0$, under Assumptions~\ref{assm:D} and
\ref{assm:deriv-point-centering} and condition~\ref{assm:regul1}\textup{(b)},
\[
\|D_{n,h,a_0,2}-D_{0,h,a_0,2}\|
+\|U_{n,h,a_0,3}^{(1)}-U_{0,h,a_0,3}^{(1)}\|
=O_p\{(nh)^{-1/2}\},
\]
\[
\begin{aligned}
\|D_{n,b,a_0,3}-D_{0,b,a_0,3}\|&=O_p\{(nb)^{-1/2}\},\\
|c_{n,h,a_0,3}^{(1)}-c_{0,h,a_0,3}^{(1)}|
&=O_p\{(nh)^{-1/2}\}.
\end{aligned}
\]
The same orders hold for the corresponding inverse-matrix differences.
Under Assumptions~\ref{assm:D} and~\ref{assm:deriv-point-centering} and
condition~\ref{assm:G-smooth}\textup{(ii)}, the fixed-point orders are, uniformly in
$a_0\in\mathcal A_0$,
\[
O_p\!\left\{\sqrt{\frac{\log(h^{-1})}{nh}}\right\}
\quad\text{and}\quad
O_p\!\left\{\sqrt{\frac{\log(b^{-1})}{nb}}\right\},
\]
at the $h$ and $b$ scales, respectively.  Consequently, with probability
tending to one the empirical matrices are nonsingular and, pointwise and
uniformly as appropriate,
\begin{align}
&|\Gamma_{n,h,b,a_0}^{(1)}(u)|
\le C_nh^{-2}I(|u-a_0|\le h)
+C_nh^2b^{-4}I(|u-a_0|\le b),
\qquad C_n=O_p(1),
\label{eq:deriv-weight-envelope}\\
&\int|\Gamma_{n,h,b,a_0}^{(1)}(u)|\,dF_0(u)
=O_p(h^{-1}),\\
&\int\{\Gamma_{0,h,b,a_0}^{(1)}(u)\}^2\,dF_0(u)
=O(h^{-3}).
\label{eq:deriv-weight-norms}
\end{align}
\end{lemma}

\begin{proof}
For $\ell=h,b$ and every fixed integer $0\le j\le6$, change of variables
gives, on the relevant neighborhood,
\[
P_0\!\left[\left\{\left(\frac{A-a_0}{\ell}\right)^j
K_{\ell,a_0}(A)\right\}^{\!2}\right]
=\ell^{-1}\int_{-1}^1v^{2j}K(v)^2f_A(a_0+\ell v)\,dv
\le C\ell^{-1}.
\]
The entries of $D_{n,h,a_0,2}$ and $U_{n,h,a_0,3}^{(1)}$ have degrees at
most four and five, respectively; those of $D_{n,b,a_0,3}$ have degrees at
most six.  Chebyshev's inequality gives their pointwise rates.
For the uniform rates, apply the translated polynomial-kernel calculation
in the proof of Lemma~\ref{lem:R6-rate} to these finitely many degrees.
Indeed, $v\mapsto v^jK(v)$ has bounded variation for each such $j$, so the
proof of \eqref{eq:kernel-polynomial-cover} gives the same covering bound,
with a constant depending on the maximum degree.  Theorem~2.14.2 of \citet{VanderVaartWellner2023weak}, applied to
the normalized class with envelope one and $L_2(P_0)$ radius
$C\sqrt\ell$, gives
\[
\begin{aligned}
\sup_{a_0\in\mathcal A_0}
\left|(P_n-P_0)
\left\{\left(\frac{A-a_0}{\ell}\right)^jK_{\ell,a_0}(A)\right\}\right|=O_p\!\left\{
\sqrt{\frac{\log(\ell^{-1})}{n\ell}}
+\frac{\log(\ell^{-1})}{n\ell}\right\}.
\end{aligned}
\]
Assumptions~\ref{assm:D} and~\ref{assm:deriv-point-centering} imply
$n h/\log(h^{-1})\to\infty$.  Also, for some finite $C\ge1$, eventually
$b\ge h/C$ and $\log(b^{-1})\le\log(h^{-1})+\log C$.  Hence
$n b/\log(b^{-1})\to\infty$, and the second term is smaller than the
first at both scales.

The population calculation in Lemma~\ref{lem:D-expansion}, with
local degrees $(1,2)$ replaced by $(2,3)$, gives
\[
D_{0,h,a_0,2}=f_A(a_0)S_3^{(1)}+O(h),
\qquad
D_{0,b,a_0,3}=f_A(a_0)S_4^{(1)}+O(b).
\]
Also $U_{0,h,a_0,3}^{(1)}=O(1)$.  These bounds are uniform under
condition~\ref{assm:G-smooth}\textup{(ii)}.  Both fixed Gram matrices are
positive definite because their quadratic forms integrate nonzero squared
polynomials against the density $K$.  The population inverse norms are
therefore bounded, and the empirical matrix rates give nonsingularity
with probability tending to one.  Apply the inverse and coefficient
argument in the proof of Lemma~\ref{lem:Gamma-difference-expansion} with
the degree pair $(1,2)$ replaced by $(2,3)$, selectors $(e_1,e_3)$ by
$(e_2,e_4)$, and the quadratic moment by $U_{n,h,a_0,3}^{(1)}$.
In particular,
\[
D_n^{-1}-D_0^{-1}=D_n^{-1}(D_0-D_n)D_0^{-1}
\]
gives the inverse rates, and the coefficient identity in the proof of Lemma~\ref{lem:Gamma-difference-expansion} provides the stated $c_{n,h,a_0,3}^{(1)}$ rate under these substitutions.
Consequently, we get
\[
\|D_{n,h,a_0,2}^{-1}\|+\|D_{n,b,a_0,3}^{-1}\|
+|c_{n,h,a_0,3}^{(1)}|=O_p(1),
\]
pointwise, and also after taking $\sup_{a_0\in\mathcal A_0}$ under the uniform assumptions.

Substitution into $\Gamma_{P,h,b,a}^{(1)}$ gives
\eqref{eq:deriv-weight-envelope}; the same bound holds for the population
weight with a deterministic constant.  Since
$F_0([a_0-\ell,a_0+\ell])\le C\ell$ on the relevant neighborhoods,
\[
\begin{aligned}
\int|\Gamma_{n,h,b,a_0}^{(1)}|\,dF_0
&\le C C_n(h^{-1}+h^2b^{-3})
=C C_nh^{-1}(1+\tau_n^3),\\
\int(\Gamma_{0,h,b,a_0}^{(1)})^2\,dF_0
&\le C(h^{-3}+h^4b^{-7})
=Ch^{-3}(1+\tau_n^7).
\end{aligned}
\]
These inequalities prove \eqref{eq:deriv-weight-norms}, since
$\tau_n=O(1)$.
\end{proof}

\begin{lemma}\label{lem:deriv-R1}
Under the pointwise derivative assumptions,
$R_{n,1}^{(1)}(a_0)=o(h_n^2)$.  Under the uniform derivative assumptions,
\[
\sup_{a_0\in\mathcal A_0}|R_{n,1}^{(1)}(a_0)|
=O(h_n^{2+\alpha}+h_n^2b_n^\alpha).
\]
\end{lemma}

\begin{proof}
Use the polynomial-reproduction argument in the proof of
Lemma~\ref{lem:R1-rate}, with local degrees two and three and derivative
orders one and three.  In detail, the moment-matrix definitions give
\[
\begin{aligned}
P_0\{r_{0,h,a_0}^{(1)}(A)(A-a_0)^j\}
&=I(j=1),&&j=0,1,2,\\
P_0\{s_{0,b,a_0}^{(3)}(A)(A-a_0)^j\}
&=6I(j=3),&&j=0,1,2,3,\\
P_0\{r_{0,h,a_0}^{(1)}(A)(A-a_0)^3\}
&=h^2c_{0,h,a_0,3}^{(1)}.
\end{aligned}
\]
Consequently, we get
\[
P_0\{\Gamma_0^{(1)}(A)(A-a_0)^j\}=I(j=1),
\qquad j=0,1,2,3.
\]
Next, Taylor's expansion yields
\[
\begin{aligned}
\left|\theta_0(u)-\sum_{j=0}^3
\frac{\theta_0^{(j)}(a_0)}{j!}(u-a_0)^j\right|\leq\frac{|u-a_0|^3}{6}
\sup_{|v-a_0|\le|u-a_0|}
|\theta_0^{(3)}(v)-\theta_0^{(3)}(a_0)|.
\end{aligned}
\]
The population envelope in Lemma~\ref{lem:deriv-matrices} and
\[
\int_{|u-a_0|\le\ell}|u-a_0|^3\,dF_0(u)\le C\ell^4
\]
therefore imply
\[
\begin{aligned}
|R_{n,1}^{(1)}(a_0)|
\le Ch^2\sup_{|v-a_0|\le h}
|\theta_0^{(3)}(v)-\theta_0^{(3)}(a_0)|+Ch^2\sup_{|v-a_0|\le b}
|\theta_0^{(3)}(v)-\theta_0^{(3)}(a_0)|.
\end{aligned}
\]
The continuity proves the pointwise claim.  The uniform H\"older bound
makes the two terms $O(h^{2+\alpha})$ and $O(h^2b^\alpha)$, respectively.
\end{proof}

\begin{lemma}\label{lem:deriv-R6}
Under the pointwise derivative assumptions,
\[
R_{n,6}^{(1)}(a_0)=O_p\{(nh_n^2)^{-1}\}.
\]
Under the uniform derivative assumptions,
\[
\sup_{a_0\in\mathcal A_0}|R_{n,6}^{(1)}(a_0)|
=O_p\!\left\{\frac{\log(h_n^{-1})}{nh_n^2}\right\}.
\]
\end{lemma}

\begin{proof}
Suppress $a_0$ and write $D_{nh}=D_{n,h,a_0,2}$ and
$D_{nb}=D_{n,b,a_0,3}$, with analogous notation for
$D_{0h},D_{0b},M_{0h},M_{0b},U_{nh},U_{0h},c_n,c_0$.
We work on the nonsingularity event in Lemma~\ref{lem:deriv-matrices}.
The final matrix calculation in the proof of
Lemma~\ref{lem:prox-six-rem-decomp} applies with degree pair $(1,2)$
replaced by $(2,3)$, selector pair $(e_1,e_3)$ replaced by $(e_2,e_4)$,
the cubic moment $U_{nh}$ in place of the quadratic moment, and the
outside factors $1,\tau_n^2$ replaced by $h^{-1},\tau_n^3h^{-1}$.
For completeness, the cancellation uses the ordered identity
\[
\begin{aligned}
D_{n\ell}^{-1}-D_{0\ell}^{-1}
+D_{0\ell}^{-1}(D_{n\ell}-D_{0\ell})D_{0\ell}^{-1}=D_{0\ell}^{-1}(D_{0\ell}-D_{n\ell})
(D_{n\ell}^{-1}-D_{0\ell}^{-1}),\quad \ell=h,b,
\end{aligned}
\]
and the exact coefficient identity
\[
\begin{aligned}
&e_2^\top D_{0h}^{-1}
\{U_{nh}-U_{0h}-(D_{nh}-D_{0h})D_{0h}^{-1}U_{0h}\}
-(c_n-c_0)\\
&\quad=-e_2^\top(D_{nh}^{-1}-D_{0h}^{-1})(U_{nh}-U_{0h})-e_2^\top D_{0h}^{-1}(D_{0h}-D_{nh})
(D_{nh}^{-1}-D_{0h}^{-1})U_{0h}.
\end{aligned}
\]
Substituting these identities into \eqref{eq:deriv-R6} yields
\begin{align}
R_{n,6}^{(1)}
={}&h^{-1}e_2^\top D_{0h}^{-1}(D_{0h}-D_{nh})
(D_{nh}^{-1}-D_{0h}^{-1})M_{0h}\notag\\
&-\tau_n^3h^{-1}c_0e_4^\top D_{0b}^{-1}(D_{0b}-D_{nb})
(D_{nb}^{-1}-D_{0b}^{-1})M_{0b}\notag\\
&-\tau_n^3h^{-1}e_2^\top(D_{nh}^{-1}-D_{0h}^{-1})
(U_{nh}-U_{0h})e_4^\top D_{0b}^{-1}M_{0b}\notag\\
&-\tau_n^3h^{-1}e_2^\top D_{0h}^{-1}(D_{0h}-D_{nh})
(D_{nh}^{-1}-D_{0h}^{-1})U_{0h}\notag\\
&\hspace{15mm}\times e_4^\top D_{0b}^{-1}M_{0b}\notag\\
&-\tau_n^3h^{-1}(c_n-c_0)e_4^\top
(D_{nb}^{-1}-D_{0b}^{-1})M_{0b}.
\label{eq:deriv-R6-exact}
\end{align}
All population factors in this display are bounded: this follows from
Lemma~\ref{lem:deriv-matrices}, local boundedness of $\theta_0$, and
\[
\|M_{0h}\|+\|M_{0b}\|
\le C\sum_{\ell\in\{h,b\}}
\int\|w_{\ell,a_0,j}(u)\|\,|K_{\ell,a_0}(u)|\,dF_0(u)
\le C,
\]
where $j=2$ for $\ell=h$ and $j=3$ for $\ell=b$.
Multiplying the matrix-error bounds in Lemma~\ref{lem:deriv-matrices}
therefore gives, in the order of the five terms,
\[
\begin{aligned}
|R_{n,6}^{(1)}(a_0)|
&=O_p\!\left\{
\frac1{nh^2}+
\frac{\tau_n^3}{nhb}+
\frac{\tau_n^3}{nh^2}+
\frac{\tau_n^3}{nh^2}+
\frac{\tau_n^3}{nh\sqrt{hb}}\right\}\\
&=O_p\!\left\{
\frac{1+\tau_n^4+2\tau_n^3+\tau_n^{7/2}}{nh^2}
\right\}
=O_p\{(nh^2)^{-1}\}.
\end{aligned}
\]
Under the uniform assumptions the population bounds hold uniformly, and
exactly the same five products give
\[
\begin{aligned}
&\sup_{a_0\in\mathcal A_0}|R_{n,6}^{(1)}(a_0)|\\
&=O_p\!\Bigg[\frac1{nh^2}\Big\{
\log(h^{-1})+\tau_n^4\log(b^{-1})+2\tau_n^3\log(h^{-1})+\tau_n^{7/2}\sqrt{\log(h^{-1})\log(b^{-1})}
\Big\}\Bigg].
\end{aligned}
\]
Since $\tau_n=O(1)$ and
$\log(b^{-1})\le\log(h^{-1})+O(1)$, this proves the uniform claim.
\end{proof}

\begin{lemma}\label{lem:deriv-fixed-degree-extensions}
Under the corresponding pointwise or uniform derivative assumptions, the
product-process conclusions of Lemma~\ref{lem:conditional-product-process}
hold for every fixed $j\le6$, and the integrated-process conclusions of
Lemma~\ref{lem:conditional-integrated-kernel} hold for every fixed $j\le3$.
For the kernel process in the notation of
Lemma~\ref{lem:conditional-kernel-ep}, its fixed-center conclusion holds for
every fixed $j\le3$.  Uniformly, write
$\delta_{\psi,k}=\delta_{\psi,k}(\mathcal A_{\delta_3})$.  If the multiplier
is bounded and $0\le\delta_{\psi,k}\le1$, then
\[
\sup_{u\in\mathcal A_0}
\left|(P_{n,k}-P_0)\{p_{\ell,u,j}(A)\psi_{n,k}(O)\}\right|
=O_p\!\left[
\delta_{\psi,k}
\sqrt{\frac{\log(\ell^{-1})+\log(1/\delta_{\psi,k})}{n_k\ell}}
+\frac{\log n}{n_k\ell}
\right].
\]
\end{lemma}

\begin{proof}
Write $h=h_n$, $b=b_n$, $L_h=\log(h^{-1})$, and $L_b=\log(b^{-1})$.
Take $\mathcal I=B_{\delta_1}(a_0)$ pointwise and
$\mathcal I=\mathcal A_{\delta_3}$ uniformly, and put
$M_f=\sup_{a\in\mathcal I}f_A(a)<\infty$.
For the uniform assertions, Assumption~\ref{assm:deriv-point-centering} and
condition~\ref{assm:deriv-uniform-bandwidth} give, for all sufficiently large $n$,
\[
n^{-1/3}\le h\le(n\log n)^{-1/(7+2\alpha)},\qquad
\frac{\log n+\log\log n}{7+2\alpha}\le L_h\le\frac{\log n}{3}.
\]
Condition~\ref{assm:bandwidths} gives $\tau_n=h/b=O(1)$.  Therefore,
\begin{equation}\label{eq:deriv-bandwidth-log-consequences}
L_h\asymp\log n,\qquad L_b\le L_h+O(1),\qquad
\frac{(\log n)^3}{nh}\longrightarrow0,\qquad
hL_h\log n\longrightarrow0,
\end{equation}
where the last two limits follow from
\[
\frac{(\log n)^3}{nh}\le n^{-2/3}(\log n)^3,\qquad
hL_h\log n\le\frac{(n\log n)^{-1/(7+2\alpha)}(\log n)^2}{3}.
\]
The coupling rate required below also follows:
\[
\begin{aligned}
&\{(nh)^{-1/6}\log n+(nh)^{-1/4}(\log n)^{5/4}
 +(nh)^{-1/2}(\log n)^{3/2}\}\sqrt{L_h}\\
&\qquad\le C\{n^{-1/9}(\log n)^{3/2}
 +n^{-1/6}(\log n)^{7/4}+n^{-1/3}(\log n)^2\}\longrightarrow0.
\end{aligned}
\]

For $0\le j\le6$, extend the earlier definitions by
\[
\phi_j(t)=t^jK(t),\quad p_{\ell,u,j}(a)=\ell^{-1}\phi_j\{(a-u)/\ell\},
\quad \mathcal P_{\ell,j}=\{p_{\ell,u,j}:u\in\mathcal A_0\}.
\]
For each such $j$, $\operatorname{supp}(\phi_j)\subseteq[-1,1]$, and
\[
\|\phi_j\|_\infty\le\|K\|_\infty,\quad
\operatorname{Lip}(\phi_j)\le j\|K\|_\infty+\operatorname{Lip}(K),\quad
\operatorname{TV}(\phi_j)\le2\{j\|K\|_\infty+\operatorname{Lip}(K)\}.
\]
Thus the bounded-variation decomposition and the representatives from
nonempty intersections in the proof of
\eqref{eq:kernel-polynomial-cover} apply to every $j\le6$; in particular,
Lemma~2.6.17 and Theorem~2.6.7 of \citet{VanderVaartWellner2023weak} give
\[
\sup_QN(\varepsilon/\ell,\mathcal P_{\ell,j},L_2(Q))
\le C\varepsilon^{-4},\qquad0<\varepsilon<1.
\]
For the product process, substitute these seven degrees into the already
defined functions $f_{n,k,\ell,j,u}$, $f_{n,k,\ell,j,u}^{\circ}$,
$\widetilde f_{n,k,\ell,j,u}^{\circ}$ and their classes in the proof of Lemma~\ref{lem:conditional-product-process}; all other arguments and
centering measures remain unchanged.  On the foldwise
$\trainingsigma_k$-measurable event
$\sup_{a\in\mathcal I,w,x}|\hat h_n^{(-k)}(w,a,x)|\le C_h$,
\eqref{eq:conditional-product-cover} holds with the same exponent four.
The exact identity \eqref{eq:conditional-product-exact} still has ordered
sum normalization $n_k^{-2}$, and its proof indicates, with $N=n_k$,
\[
\begin{aligned}
E_0\!\left[\left\{N^{-2}\sum_{r\ne s}
\widetilde f_{n,k,\ell,j,a_0}^{\circ}(O_r,O_s)\right\}^2
\middle|\trainingsigma_k\right]&\le C/(N^2\ell),\\
E_0\!\left[\left|N^{-2}\sum_r f_{n,k,\ell,j,a_0}^{\circ}(O_r,O_r)\right|
\middle|\trainingsigma_k\right]&\le C/N,\\
E_0\!\left[\sup_{u\in\mathcal A_0}\left|
N^{-2}\sum_{r,s}f_{n,k,\ell,j,u}^{\circ}(O_r,O_s)\right|
\middle|\trainingsigma_k\right]&\le C/(N\ell),
\end{aligned}
\]
as in the proof of Lemma~\ref{lem:conditional-product-process}.
Here all indices are in $I_k$.  

For the kernel process, substitute $j\in\{0,1,2,3\}$ into the proof of
Lemma~\ref{lem:conditional-kernel-ep}.  
Conditional on $\trainingsigma_k$, its variance
calculation is
\[
\begin{aligned}
P_0(p_{\ell,u,j}^2\psi_{n,k}^2)
&=\ell^{-1}\int_{-1}^1\phi_j(t)^2
 E_0(\psi_{n,k}^2\mid A=u+\ell t,\trainingsigma_k)
 f_A(u+\ell t)\,dt\\
&\le M_f\delta^2\ell^{-1}\|\phi_j\|_2^2,
\end{aligned}
\]
where $\delta=\delta_{\psi,k}(\mathcal I)$; every displayed kernel support
is contained in $\mathcal I$ for all sufficiently large $n$.
This proves the fixed-center claim by conditional Chebyshev's inequality.
For uniform inference take $\mathcal I=\mathcal A_{\delta_3}$; on the
foldwise event $\|\psi_{n,k}I(A\in\mathcal I)\|_{P_0,\infty}\le M$, use
its truncation $\psi^M$ to $[-M,M]$ and define $\mathcal F=\{p_{\ell,u,j}(A)\psi^M(O):u\in\mathcal A_0\}$, $B=2M\|\phi_j\|_\infty/\ell$, $\sigma^2=M_f\delta^2\ell^{-1}\|\phi_j\|_2^2$, $t=\min\{3/4,\max(2\sigma/B,N^{-1/2})\}$, and $N=n_k$.
The change of measure $Q_\psi$ in the proof of
Lemma~\ref{lem:conditional-kernel-ep}, applied to this $\mathcal F$,
gives $\sup_QN(\varepsilon B,\mathcal F,L_2(Q))\le C_M\varepsilon^{-4}$.
With
\[
J(t)=\sup_Q\int_0^t
\sqrt{1+\log N(\varepsilon B,\mathcal F,L_2(Q))}\,d\varepsilon
\le C_Mt\sqrt{\log(C_M/t)},
\]
Theorem~2.14.2 of \citet{VanderVaartWellner2023weak} applies exactly as
in the earlier proof: $P_0(f/B)^2<t^2$ for every $f\in\mathcal F$,
$Bt\le2\sigma+B/\sqrt N$, and $\log(C_M/t)\le C_M\log N$.
For $0<\delta\le1$ one also has
\[
t\ge\min\{3/4,c_M\delta\sqrt\ell\},\qquad
\log(C_M/t)\le C_M+\tfrac12\log(\ell^{-1})+\log(1/\delta).
\]
Consequently, we obtain
\[
\begin{aligned}
E_0\!\left[\sup_{f\in\mathcal F}|(P_{n,k}-P_0)f|
\middle|\trainingsigma_k\right]
&\le C_MB\left\{\frac{J(t)}{\sqrt N}+\frac{J(t)^2}{Nt^2}\right\}\\
&\le C_M\left[
\delta\sqrt{\frac{\log(\ell^{-1})+\log(1/\delta)}{N\ell}}
+\frac{\log n}{N\ell}\right].
\end{aligned}
\]
If $\delta=0$, the $\sigma$ term vanishes and the $N^{-1/2}$ floor gives
this bound with its first term equal to zero.  Using only
$\log(C_M/t)\le C_M\log N$ also proves, without requiring $\delta\le1$,
\begin{equation}\label{eq:deriv-kernel-logn}
\sup_{u\in\mathcal A_0}|(P_{n,k}-P_0)(p_{\ell,u,j}\psi_{n,k})|
=O_p\!\left\{\delta_{\psi,k}\sqrt{\frac{\log n}{n_k\ell}}
+\frac{\log n}{n_k\ell}\right\}.
\end{equation}
Next, Markov's inequality and iterated expectation imply both probability bounds. 

Finally, for the integrated process, substitute $j\in\{0,1,2,3\}$ into
Lemma~\ref{lem:conditional-integrated-kernel}, keeping the same
$\chi_{n,k}$, $H_{n,k,\ell,j,u}$, and explicitly defined class
$\mathcal H=\{H_{n,k,\ell,j,u}:u\in\mathcal A_0\}$.
On $\sup_{a\in\mathcal I,w,x}|\chi_{n,k}(w,a,x)|\le M$, the argument used in the proof of Lemma~\ref{lem:conditional-integrated-kernel} gives
\[
|H_{n,k,\ell,j,u}|\le MM_f\|\phi_j\|_1=:B_M,\qquad
\sup_QN(\varepsilon B_M,\mathcal H,L_2(Q))\le C_M/(\ell\varepsilon).
\]
Its conditional variance and Theorem~2.14.1, $p=1$, of
\citet{VanderVaartWellner2023weak} consequently give
\[
\begin{aligned}
E_0\{[(P_{n,k}-P_0)H_{n,k,\ell,j,u}]^2\mid\trainingsigma_k\}
&\le B_M^2/N,\\
E_0\!\left[\sup_{H\in\mathcal H}|(P_{n,k}-P_0)H|
\middle|\trainingsigma_k\right]
&\le C_M\sqrt{\log(e/\ell)/N}.
\end{aligned}
\]
These imply both integrated-process conclusions.  As in the earlier proofs, continuity of $u\mapsto p_{\ell,u,j}$ and bounded dominated
convergence allow all indexed suprema to be taken over a countable dense
subset of $\mathcal A_0$.
\end{proof}

\begin{lemma}\label{lem:deriv-R2-R5}
Under the pointwise derivative assumptions,
\[
R_{n,2}^{(1)}(a_0)=o_p\{(nh_n^3)^{-1/2}\},\qquad
R_{n,3}^{(1)}(a_0)=O_p\{(nh_n^2)^{-1}\},
\]
\[
R_{n,4}^{(1)}(a_0)=o_p\{(nh_n^3)^{-1/2}\},\qquad
R_{n,5}^{(1)}(a_0)=O_p(n^{-1}h_n^{-3/2}).
\]
Under the uniform derivative assumptions,
\[
\sup_{a_0\in\mathcal A_0}|R_{n,2}^{(1)}(a_0)|
=o_p\{(nh_n^3\log n)^{-1/2}\},
\]
\[
\sup_{a_0\in\mathcal A_0}|R_{n,3}^{(1)}(a_0)|
=O_p\!\left\{\frac{\log(h_n^{-1})}{nh_n^2}\right\},
\]
\[
\sup_{a_0\in\mathcal A_0}|R_{n,4}^{(1)}(a_0)|
=o_p\{(nh_n^3\log n)^{-1/2}\},
\qquad
\sup_{a_0\in\mathcal A_0}|R_{n,5}^{(1)}(a_0)|
=O_p(n^{-1}h_n^{-2}).
\]
\end{lemma}

\begin{proof}
Write $L_h=\log(h^{-1})$,
$L_b=\log(b^{-1})$, and take $\mathcal I=B_{\delta_1}(a_0)$ pointwise
and $\mathcal I=\mathcal A_{\delta_3}$ uniformly.  At the displayed center $a_0$,
\[
D_{th}=D_{t,h,a_0,2},\qquad D_{tb}=D_{t,b,a_0,3},\qquad
c_t=c_{t,h,a_0,3}^{(1)},\qquad t\in\{0,n\}.
\]
All conditional bounds below are proved for a fold given
$\trainingsigma_k$, on its own measurable boundedness event.  There are $K$ folds and $n_k\ge\underline\pi n$; hence these bounds hold
jointly. 

For $R_{n,4}^{(1)}$, use the conditional drift from the proof of
Lemma~\ref{lem:R4-control},
\[
\Delta_{n,k}(a)=E_0\{(\hat q_n^{(-k)}-q_0)(h_0-\hat h_n^{(-k)})
\mid A=a,\trainingsigma_k\}.
\]
Its partition calculation applies under condition~\ref{assm:DR_and_rates}
with $(\mathcal I,\mathcal S_j)=(B_{\delta_1}(a_0),\mathcal S_j)$;
uniformly, use conditions~\ref{assm:G-DR-rates}--\ref{assm:G-nuisance-L2}
with $(\mathcal I,\mathcal S_j)=(\mathcal A_{\delta_3},\mathcal S_j')$.
The bound proved there is
\[
\sup_{a\in\mathcal I}|\Delta_{n,k}(a)|
\le\sum_{j=1}^3d_{q,k}(\hat q_n^{(-k)},q_0;\mathcal I,\mathcal S_j)
 d_{h,k}(\hat h_n^{(-k)},h_0;\mathcal I,\mathcal S_j).
\]
It has order $o_p\{(nh)^{-1/2}\}$ pointwise and
$o_p\{(nh\log n)^{-1/2}\}$ uniformly, by the respective substitutions.
Replacing $\Gamma_n$ by $\Gamma_n^{(1)}$ in
\eqref{eq:R4-reduced-to-drift} gives
\[
\begin{aligned}
|R_{n,4}^{(1)}(a_0)|
&\le\left\{\int|\Gamma_{n,h,b,a_0}^{(1)}(a)|\,dF_0(a)\right\}
\max_k\sup_{a\in\mathcal I}|\Delta_{n,k}(a)|,\\
\int|\Gamma_{n,h,b,a_0}^{(1)}|\,dF_0
&=O_p(h^{-1}+h^2b^{-3})=O_p\{h^{-1}(1+\tau_n^3)\},
\end{aligned}
\]
where the last bound follows from \eqref{eq:deriv-weight-norms} and also
holds uniformly.  
It further proves both $R_{n,4}^{(1)}$ statements.

For $R_{n,5}^{(1)}$, the exact expansion is
\[
\begin{aligned}
R_{n,5}^{(1)}(a_0)
={}&h^{-1}\sum_{k=1}^K\frac{n_k}n\sum_{j=0}^2
[e_2^\top D_{nh}^{-1}]_{j+1}\\
&\qquad\times\iint p_{h,a_0,j}(a)\hat h_n^{(-k)}(w,a,x)
\,d(Q_{n,k}^{WX}-Q_0^{WX})\,d(F_{n,k}-F_0)\\
&-h^2b^{-3}c_n\sum_{k=1}^K\frac{n_k}n\sum_{j=0}^3
[e_4^\top D_{nb}^{-1}]_{j+1}\\
&\qquad\times\iint p_{b,a_0,j}(a)\hat h_n^{(-k)}(w,a,x)
\,d(Q_{n,k}^{WX}-Q_0^{WX})\,d(F_{n,k}-F_0).
\end{aligned}
\]
Lemma~\ref{lem:deriv-matrices} bounds the displayed rows and $c_n$ in
probability, pointwise and uniformly.  Applying the product conclusion of
Lemma~\ref{lem:deriv-fixed-degree-extensions} with
$$
(\ell,j)=(h,0),
(h,1),(h,2),(b,0),\ldots,(b,3),
$$
and using
$\sum_k(n_k/n)n_k^{-1}=K/n$, gives
\[
\begin{aligned}
|R_{n,5}^{(1)}(a_0)|
&=O_p\!\left\{\frac K n(h^{-3/2}+h^2b^{-7/2})\right\}
=O_p\!\left\{\frac{1+\tau_n^{7/2}}{nh^{3/2}}\right\},\\
\sup_{a_0\in\mathcal A_0}|R_{n,5}^{(1)}(a_0)|
&=O_p\!\left\{\frac K n(h^{-2}+h^2b^{-4})\right\}
=O_p\!\left\{\frac{1+\tau_n^4}{nh^2}\right\}.
\end{aligned}
\]

For $R_{n,2}^{(1)}$ put $\psi_{n,k}=\hat\vartheta_{n,k}-\xi_\infty$.
Its exact expansion is
\begin{align}
\psi_{n,k}(O)
={}&\{\hat q_n^{(-k)}(Z,A,X)-q_\infty(Z,A,X)\}
\{Y-h_\infty(W,A,X)\}\notag\\
&+\hat q_n^{(-k)}(Z,A,X)
\{h_\infty(W,A,X)-\hat h_n^{(-k)}(W,A,X)\}\notag\\
&+S_{\hat h_n^{(-k)}}(A)-m_\infty(A).
\label{eq:deriv-pseudo-difference}
\end{align}
Recall $\mathcal S=\mathcal I\times\mathcal Z\times\mathcal W\times\mathcal X$.
%Choose the versions of the conditional second moments defining
%$\delta_{\psi,k}$ to be zero on the exceptional $F_0$-null set in the
%Bayes identity used below.  Essential suprema below are with respect to $F_0$.
The H\"older calculation proving \eqref{eq:R2-pointwise-pseudo-radius}
uses conditions~\ref{assm:nuisance1}--\ref{assm:DR_and_rates} and the
conditional moment requirement in condition~\ref{assm:regul1}\textup{(c)},
with exponents $(2+\delta_2)/2$ and $(2+\delta_2)/\delta_2$.  It gives
\[
\begin{aligned}
&\sup_{a\in\mathcal I}
E_0\{(\hat q_n^{(-k)}-q_\infty)^2(Y-h_\infty)^2
\mid A=a,\trainingsigma_k\}\le C\,d_{q,k}(\hat q_n^{(-k)},q_\infty;
\mathcal I,\mathcal S)^{2\delta_2/(2+\delta_2)},\\
&\delta_{\psi,k}(\mathcal I)\le C\,d_{q,k}(\hat q_n^{(-k)},q_\infty;
\mathcal I,\mathcal S)^{\delta_2/(2+\delta_2)}
\\
&\qquad\qquad+C\,d_{h,k}(\hat h_n^{(-k)},h_\infty;\mathcal I,\mathcal S)
+\mathop{\mathrm{ess\,sup}}_{a\in\mathcal I}|S_{\hat h_n^{(-k)}}(a)-m_\infty(a)|.
\end{aligned}
\]
The Bayes' formula and Jensen's inequality in the proof of
Lemma~\ref{lem:R2-control}, using condition~\ref{assm:basic.proxy}, yields
\[
\begin{aligned}
\mathop{\mathrm{ess\,sup}}_{a\in\mathcal I}|S_{\hat h_n^{(-k)}}(a)-m_\infty(a)|^2
&\le\mathop{\mathrm{ess\,sup}}_{a\in\mathcal I}E_0\{(\hat h_n^{(-k)}-h_\infty)^2(W,a,X)
 \mid\trainingsigma_k\}\\
&\le\frac{\sup_{a\in\mathcal I}f_A(a)}{c_1}
 d_{h,k}(\hat h_n^{(-k)},h_\infty;\mathcal I,\mathcal S)^2.
\end{aligned}
\]
Condition~\ref{assm:pointwise-pseudooutcome-rate} therefore gives
$\max_k\delta_{\psi,k}(B_{\delta_1}(a_0))=o_p(1)$.
Under condition~\ref{assm:G-smooth}\textup{(iii)}, the bound for
$\delta_{\psi,k}$ holds with the exponent on $d_{q,k}$ replaced by one.
Conditions~\ref{assm:G-crossfit}--\ref{assm:G-nuisance-L2} then yield
\[
\begin{gathered}
\max_k\|\psi_{n,k}I(A\in\mathcal A_{\delta_3})\|_{P_0,\infty}=O_p(1),\\
\sqrt{L_h\log n}\max_k\delta_{\psi,k}(\mathcal A_{\delta_3})=o_p(1).
\end{gathered}
\]

Use $H_{n,k,\ell,j,a_0}$ from
Lemma~\ref{lem:conditional-integrated-kernel} with the explicit substitution
\[
\begin{aligned}
\chi_{n,k}(w,a,x)&=\hat h_n^{(-k)}(w,a,x)-h_\infty(w,a,x),\\
H_{n,k,\ell,j,a_0}(O)&=\int p_{\ell,a_0,j}(a)
\{\hat h_n^{(-k)}(W,a,X)-h_\infty(W,a,X)\}\,dF_0(a).
\end{aligned}
\]
Then \eqref{eq:deriv-R2} is exactly
\[
\begin{aligned}
R_{n,2}^{(1)}(a_0)
={}&h^{-1}\sum_k\frac{n_k}n\sum_{j=0}^2[e_2^\top D_{0h}^{-1}]_{j+1}\\
&\qquad\times(P_{n,k}-P_0)\{p_{h,a_0,j}(A)\psi_{n,k}+H_{n,k,h,j,a_0}\}\\
&-h^2b^{-3}c_0\sum_k\frac{n_k}n\sum_{j=0}^3[e_4^\top D_{0b}^{-1}]_{j+1}\\
&\qquad\times(P_{n,k}-P_0)\{p_{b,a_0,j}(A)\psi_{n,k}+H_{n,k,b,j,a_0}\}.
\end{aligned}
\]
The population coefficients are bounded by
Lemma~\ref{lem:deriv-matrices}.  The fixed-center conclusions give
\[
|R_{n,2}^{(1)}(a_0)|
=o_p\!\left\{\frac{1+\tau_n^{7/2}}{\sqrt{nh^3}}\right\}
+O_p\!\left\{\frac{1+\tau_n^3}{\sqrt n\,h}\right\}
=o_p\{(nh^3)^{-1/2}\},
\]
since the ratio of the second term to $(nh^3)^{-1/2}$ is $O(\sqrt h)$.
Uniformly set $\delta_n=\max_k\delta_{\psi,k}(\mathcal A_{\delta_3})$.
The bound \eqref{eq:deriv-kernel-logn} and the integrated-process bound give
\[
\begin{aligned}
\sup_{a_0\in\mathcal A_0}|R_{n,2}^{(1)}(a_0)|
=O_p\!\Bigg[\delta_n(1+\tau_n^{7/2})
\sqrt{\frac{\log n}{nh^3}}
+(1+\tau_n^4)\frac{\log n}{nh^2}+\frac{1+\tau_n^3}{h}
\left\{\sqrt{\frac{L_h}n}+\frac{L_h}n\right\}\Bigg].
\end{aligned}
\]
The respective ratios to $(nh^3\log n)^{-1/2}$ are bounded in probability
by constants times
\[
\delta_n\log n=o_p(1),\qquad
\frac{(\log n)^{3/2}}{\sqrt{nh}}\longrightarrow0,\qquad
\sqrt{hL_h\log n}\left(1+\sqrt{L_h/n}\right)\longrightarrow0,
\]
by \eqref{eq:deriv-bandwidth-log-consequences}.  This proves the uniform
$R_{n,2}^{(1)}$ assertion.

For $R_{n,3}^{(1)}$, when both empirical matrices are nonsingular define
\[
\begin{aligned}
\alpha_{n,h,j}(a_0)&=[e_2^\top(D_{nh}^{-1}-D_{0h}^{-1})]_{j+1},
&&0\le j\le2,\\
\alpha_{n,b,j}(a_0)&=-[c_ne_4^\top D_{nb}^{-1}-c_0e_4^\top D_{0b}^{-1}]_{j+1},
&&0\le j\le3.
\end{aligned}
\]
If either matrix is singular, define instead
\[
\alpha_{n,h,j}(a_0)=-[e_2^\top D_{0h}^{-1}]_{j+1},\qquad
\alpha_{n,b,j}(a_0)=c_0[e_4^\top D_{0b}^{-1}]_{j+1}.
\]
These definitions respect the convention $\Gamma_n^{(1)}=0$ and give,
on both events, the exact identity
\begin{align}
\Gamma_n^{(1)}(u)-\Gamma_0^{(1)}(u)
={}&h^{-1}\sum_{j=0}^2\alpha_{n,h,j}(a_0)p_{h,a_0,j}(u)
+h^2b^{-3}\sum_{j=0}^3\alpha_{n,b,j}(a_0)p_{b,a_0,j}(u).
\label{eq:deriv-weight-difference}
\end{align}
On the nonsingular event, we get
\[
\begin{aligned}
\max_{j\le2}|\alpha_{n,h,j}|&\le\|D_{nh}^{-1}-D_{0h}^{-1}\|_\infty,\\
\max_{j\le3}|\alpha_{n,b,j}|&\le
|c_n-c_0|\|D_{nb}^{-1}\|_\infty
+|c_0|\|D_{nb}^{-1}-D_{0b}^{-1}\|_\infty.
\end{aligned}
\]
Lemma~\ref{lem:deriv-matrices} and
$ (nb)^{-1/2}=\tau_n^{1/2}(nh)^{-1/2}$ consequently yields
\[
\begin{aligned}
\max_{j\le2}|\alpha_{n,h,j}|+\max_{j\le3}|\alpha_{n,b,j}|
&=O_p\{(nh)^{-1/2}\},\\
\sup_{a_0\in\mathcal A_0}
\left\{\max_{j\le2}|\alpha_{n,h,j}|+\max_{j\le3}|\alpha_{n,b,j}|\right\}
&=O_p\!\left\{\sqrt{\frac{L_h}{nh}}\right\}.
\end{aligned}
\]
For the uniform bound, the inverse error is
$O_p\{\sqrt{L_b/(nb)}\}=O_p\{\sqrt{L_h/(nh)}\}$ by
\eqref{eq:deriv-bandwidth-log-consequences}.  The singular event has
probability tending to zero, uniformly in the uniform assertion, so the
piecewise definitions have the same probability bounds.

For the kernel process now substitute $\psi_{n,k}=\hat\vartheta_{n,k}$.
On the foldwise boundedness event, one has
\[
\sup_{a\in\mathcal I}E_0\{\hat\vartheta_{n,k}^2\mid A=a,\trainingsigma_k\}
\le3C_q^2\sup_{a\in\mathcal I}E_0(Y^2\mid A=a)
+3C_q^2C_h^2+3C_h^2=O(1),
\]
where we used the conditional moment requirement in
condition~\ref{assm:regul1}\textup{(c)}; uniformly,
conditions~\ref{assm:G-crossfit}--\ref{assm:G-DR-rates} and
\ref{assm:G-smooth}\textup{(iii)} bound $\hat\vartheta_{n,k}$ on the
relevant neighborhood.
For the integrated process, substitute
\[
\begin{aligned}
\chi_{n,k}(w,a,x)&=\hat h_n^{(-k)}(w,a,x),\\
H_{n,k,\ell,j,a_0}(O)&=\int p_{\ell,a_0,j}(a)\hat h_n^{(-k)}(W,a,X)\,dF_0(a).
\end{aligned}
\]
Then equation~\eqref{eq:deriv-R3} becomes
\[
\begin{aligned}
R_{n,3}^{(1)}(a_0)
={}&h^{-1}\sum_k\frac{n_k}n\sum_{j=0}^2\alpha_{n,h,j}(a_0)\\
&\qquad\times(P_{n,k}-P_0)\{p_{h,a_0,j}(A)\hat\vartheta_{n,k}+H_{n,k,h,j,a_0}\}\\
&+h^2b^{-3}\sum_k\frac{n_k}n\sum_{j=0}^3\alpha_{n,b,j}(a_0)\\
&\qquad\times(P_{n,k}-P_0)\{p_{b,a_0,j}(A)\hat\vartheta_{n,k}+H_{n,k,b,j,a_0}\}.
\end{aligned}
\]
Therefore, we have
\[
\begin{aligned}
|R_{n,3}^{(1)}(a_0)|
&=O_p\!\left[\frac1{\sqrt{nh}}
\left\{\frac{1+\tau_n^{7/2}}{\sqrt{nh^3}}
+\frac{1+\tau_n^3}{\sqrt n\,h}\right\}\right]\\
&=O_p\!\left\{\frac{1+\tau_n^{7/2}}{nh^2}
+\frac{1+\tau_n^3}{nh^{3/2}}\right\}
=O_p\{(nh^2)^{-1}\}.
\end{aligned}
\]
Uniformly, \eqref{eq:deriv-kernel-logn}, $\log n\asymp L_h$, and
the integrated-process bound give
\[
\begin{aligned}
\sup_{a_0\in\mathcal A_0}|R_{n,3}^{(1)}(a_0)|
=O_p\!\Bigg\{\frac{(1+\tau_n^{7/2})L_h}{nh^2}
+\frac{(1+\tau_n^4)L_h^{3/2}}{n^{3/2}h^{5/2}}+\frac{(1+\tau_n^3)L_h}{nh^{3/2}}
+\frac{(1+\tau_n^3)L_h^{3/2}}{n^{3/2}h^{3/2}}\Bigg\}.
\end{aligned}
\]
The ratios of the last three orders to $L_h/(nh^2)$ are respectively
\[
\sqrt{L_h/(nh)}\longrightarrow0,\qquad
\sqrt h\longrightarrow0,\qquad
\sqrt{hL_h/n}\longrightarrow0.
\]
Since $\tau_n=O(1)$, this proves the remaining claim.
\end{proof}

\begin{proposition}\label{prop:deriv-remainders}
Under the pointwise derivative assumptions,
\[
\sum_{j=2}^6|R_{n,j}^{(1)}(a_0)|
=o_p\{(nh_n^3)^{-1/2}\}.
\]
Under the uniform derivative assumptions,
\[
\sum_{j=2}^6\sup_{a\in\mathcal A_0}|R_{n,j}^{(1)}(a)|
=o_p\{(nh_n^3\log n)^{-1/2}\}.
\]
\end{proposition}

\begin{proof}
Apply Lemmas~\ref{lem:deriv-R6} and~\ref{lem:deriv-R2-R5}.
Pointwise, $nh^3\to\infty$ and $h\to0$ imply $nh\to\infty$, and
\[
\frac{(nh^2)^{-1}}{(nh^3)^{-1/2}}=(nh)^{-1/2}\longrightarrow0,
\qquad
\frac{n^{-1}h^{-3/2}}{(nh^3)^{-1/2}}=n^{-1/2}\longrightarrow0.
\]
Uniformly, condition~\ref{assm:deriv-uniform-bandwidth} and
\eqref{eq:deriv-bandwidth-log-consequences} give
\[
\frac{L_h/(nh^2)}{(nh^3\log n)^{-1/2}}
=L_h\sqrt{\frac{\log n}{nh}}\longrightarrow0,
\qquad
\frac{n^{-1}h^{-2}}{(nh^3\log n)^{-1/2}}
=\sqrt{\frac{\log n}{nh}}\longrightarrow0.
\]
The sum has five terms, so the asserted conclusions follow.
\end{proof}

\subsection{Comparison of the six remainders}
\label{sec:derivative-remainder-comparison}

The pointwise six-remainder comparison is

\begin{table}[!htbp]
\centering
\caption{Pointwise remainder bounds for the dose-response curve and its first derivative.}
\label{tab:pointwise-remainder-comparison}
\[
\begin{array}{c@{\qquad}c@{\qquad}c}
\toprule
&\text{dose-response}&\text{first derivative}\\
\midrule
R_{n,1}&o(h_n^2)&o(h_n^2)\\
R_{n,2}&o_p\{(nh_n)^{-1/2}\}&o_p\{(nh_n^3)^{-1/2}\}\\
R_{n,3}&O_p\{(nh_n)^{-1}\}&O_p\{(nh_n^2)^{-1}\}\\
R_{n,4}&o_p\{(nh_n)^{-1/2}\}&o_p\{(nh_n^3)^{-1/2}\}\\
R_{n,5}&O_p(n^{-1}h_n^{-1/2})&O_p(n^{-1}h_n^{-3/2})\\
R_{n,6}&O_p\{(nh_n)^{-1}\}&O_p\{(nh_n^2)^{-1}\}\\
\bottomrule
\end{array}
\]
\end{table}

and the uniform comparison is

\begin{table}[!htbp]
\centering
\caption{Uniform remainder bounds for the dose-response curve and its first derivative.}
\label{tab:uniform-remainder-comparison}
\[
\begin{array}{c@{\quad}c@{\quad}c}
\toprule
&\text{dose-response}&\text{first derivative}\\
\midrule
R_{n,1}
&O(h_n^{2+\alpha}+h_n^2b_n^\alpha)
&O(h_n^{2+\alpha}+h_n^2b_n^\alpha)\\
R_{n,2}
&o_p\{(nh_n\log n)^{-1/2}\}
&o_p\{(nh_n^3\log n)^{-1/2}\}\\
R_{n,3}
&O_p\!\left\{\dfrac{\log(h_n^{-1})}{nh_n}\right\}
&O_p\!\left\{\dfrac{\log(h_n^{-1})}{nh_n^2}\right\}\\
R_{n,4}
&o_p\{(nh_n\log n)^{-1/2}\}
&o_p\{(nh_n^3\log n)^{-1/2}\}\\
R_{n,5}&O_p\{(nh_n)^{-1}\}&O_p(n^{-1}h_n^{-2})\\
R_{n,6}
&O_p\!\left\{\dfrac{\log(h_n^{-1})}{nh_n}\right\}
&O_p\!\left\{\dfrac{\log(h_n^{-1})}{nh_n^2}\right\}\\
\bottomrule
\end{array}
\]
\end{table}

\subsection{Leading term and Gaussian approximation}

\begin{lemma}\label{lem:deriv-equivalent-kernel}
Under the pointwise derivative assumptions,
\[
D_{0,h,a_0,2}^{-1}
=f_A(a_0)^{-1}\{S_3^{(1)}\}^{-1}+O(h),
\qquad
D_{0,b,a_0,3}^{-1}
=f_A(a_0)^{-1}\{S_4^{(1)}\}^{-1}+O(b),
\]
\[
c_{0,h,a_0,3}^{(1)}=c_{K,3}^{(1)}+O(h),
\]
and
\begin{equation}
h^2f_A(a_0)\Gamma_{0,h,b,a_0}^{(1)}(a_0+hu)
\longrightarrow\Lambda_{K,\tau}^{(1)}(u)
\quad\text{in }L_2(du).
\label{eq:deriv-equivalent-kernel-limit}
\end{equation}
Moreover, $V_{K,\tau}^{(1)}>0$.  Every assertion is uniform in
$a_0\in\mathcal A_0$ under the uniform derivative assumptions.
\end{lemma}

\begin{proof}
The change of variables used for Lemma~\ref{lem:D-expansion}, with
polynomial degrees two and three, gives
\[
D_{0,h,a,2}=f_A(a)S_3^{(1)}+O(h),\qquad
D_{0,b,a,3}=f_A(a)S_4^{(1)}+O(b).
\]
The Gram matrices are positive definite; applying the same calculation to
$U_{0,h,a,3}^{(1)}$ proves the three coefficient expansions.  Substitution
in $\Gamma_{P,h,b,a}^{(1)}$ gives the more precise bound
\[
\left\|h^2f_A(a)\Gamma_{0,h,b,a}^{(1)}(a+h\,\cdot)
-\Lambda_{K,\tau_n}^{(1)}\right\|_2
\le C\{h+(h+b)\tau_n^{7/2}\}.
\]
For $\tau>0$, we get
$\Lambda_{K,\tau_n}^{(1)}\to\Lambda_{K,\tau}^{(1)}$ in $L_2$.
The constants are uniform on $\mathcal A_0$ under the uniform assumptions.
Finally, we obtain
\[
e_2^\top\{S_3^{(1)}\}^{-1}\int r_2(u)uK(u)\,du=1,
\qquad
e_4^\top\{S_4^{(1)}\}^{-1}\int r_3(v)vK(v)\,dv=0.
\]
Another changing variables therefore gives
$\int u\Lambda_{K,\tau}^{(1)}(u)\,du=1$.
Thus $V_{K,\tau}^{(1)}>0$.
\end{proof}
Recall
\begin{align}
\psi_{h,b,a}^{(1),\mathrm{sm}}(O)
:={}&\Gamma_{0,h,b,a}^{(1)}(A)\theta_0(A)
-\gamma_{0,h,b,a}^{(1)}(A)\notag\\
&+\int\Gamma_{0,h,b,a}^{(1)}(\bar a)
\{h_\infty(W,\bar a,X)-m_\infty(\bar a)\}\,dF_0(\bar a),
\label{eq:deriv-smoothing-component}
\end{align}
which implies
\[
\phi_{\infty,h,b,a}^{(1)}
=\Gamma_{0,h,b,a}^{(1)}(A)\{\xi_\infty-\theta_0(A)\}
+\psi_{h,b,a}^{(1),\mathrm{sm}}.
\]

\begin{lemma}\label{lem:deriv-point-leading}
Under the pointwise derivative assumptions,
\[
(nh_n^3)^{1/2}P_n\left[
\Gamma_{0,h_n,b_n,a_0}^{(1)}(A)
\{\xi_\infty(O)-\theta_0(A)\}\right]
\rightsquigarrow
N\!\left(0,
\frac{V_{K,\tau}^{(1)}\sigma_0^2(a_0)}{f_A(a_0)}\right),
\]
and
\[
P_0\psi_{h_n,b_n,a_0}^{(1),\mathrm{sm}}=0,
\qquad
(nh_n^3)^{1/2}P_n\psi_{h_n,b_n,a_0}^{(1),\mathrm{sm}}=o_p(1).
\]
\end{lemma}

\begin{proof}
This is the derivative version of
Lemmas~\ref{lem:single_a_leading_term}
and~\ref{lem:smoothing_negligible_single_a}; we verify the changed
normalization and smoothing bounds.  Theorem~\ref{thm:DR_for_PO} and
condition~\ref{assm:DR_and_rates} give
$E_0\{\xi_\infty-\theta_0(A)\mid A\}=0$.
For $Z_{n,i}=h^{3/2}\Gamma_{0,h,b,a_0}^{(1)}(A_i)
\{\xi_\infty(O_i)-\theta_0(A_i)\}$, the preceding lemma and continuity
of $\sigma_0^2f_A$ give
\[
E_0Z_{n,i}^2\longrightarrow
V_{K,\tau}^{(1)}\sigma_0^2(a_0)/f_A(a_0).
\]
Both original kernel supports shrink to $a_0$, including when $\tau=0$.
For $r=2+\delta_2$, the conditional moment bound and
\eqref{eq:deriv-weight-envelope} yield
\[
\begin{aligned}
E_0|Z_{n,i}|^r
&\le C h^{3r/2}\{h^{1-2r}+h^{2r}b^{1-4r}\}\\
&=C h^{1-r/2}(1+\tau_n^{4r-1}).
\end{aligned}
\]
The Lyapunov ratio is therefore $O\{(nh)^{-\delta_2/2}\}=o(1)$.

Write $\psi_{h,b,a}^{(1),\mathrm{sm}}=U_a+V_a$, where
$U_a=\Gamma_{0,h,b,a}^{(1)}(A)\theta_0(A)-\gamma_{0,h,b,a}^{(1)}(A)$
and $V_a$ is the integrated term in
\eqref{eq:deriv-smoothing-component}.  The moment identities
$P_0\gamma_{0,h,a}^{(1)}=\theta_{0,h}^{(1)}(a)$,
$P_0\gamma_{0,b,a}^{(3)}=\theta_{0,b}^{(3)}(a)$, and
$P_0\gamma_{0,h,a}^{c,1}=0$ give $P_0U_a=P_0V_a=0$.
The Taylor's expansion argument of Lemma~\ref{lem:deriv-R1} gives, for
$(\ell,j)=(h,2),(b,3)$,
\[
\sup_{|s|\le\ell}
\left|\theta_0(a+s)
-w_{\ell,a,j}(a+s)^\top D_{0,\ell,a,j}^{-1}M_{0,\ell,a,j}\right|
\le C\ell^3.
\]
Indeed, both projections reproduce the quadratic polynomial, its
remainder is bounded by $C\ell^3$, and the inverse-matrix and kernel
$L_1(F_0)$ bounds make the projection operators bounded in supremum norm.
The cubic coefficient of that polynomial is zero, so the same bounds give
$|\theta_{0,b}^{(3)}(a)|\le C$.  Including the last term of $\gamma_{P,h,b,a}^{(1)}(u)$, we obtain
\[
|U_a|\le C\{hI(|A-a|\le h)+h^2b^{-1}I(|A-a|\le b)\},
\qquad \|V_a\|_\infty\le Ch^{-1}.
\]
Consequently, we have
\[
P_0U_a^2\le C(h^3+h^4/b)=O(h^3),\qquad
h^3P_0(U_a+V_a)^2=O(h).
\]
Chebyshev's inequality proves the smoothing assertion.  All these bounds
are uniform on $\mathcal A_0$ under the uniform derivative assumptions.
\end{proof}

\begin{lemma}\label{lem:deriv-uniform-leading}
Under the uniform derivative assumptions,
\[
\sup_{a\in\mathcal A_0}\left|
\sigma_{\infty,h,b}^{(1)2}(a)
-\frac{V_{K,\tau}^{(1)}\sigma_0^2(a)}{f_A(a)}\right|\longrightarrow0.
\]
In particular, the variances are uniformly bounded away from zero and
infinity.  If
\[
g_{h,b,a}^{(1)}(O)
:=\frac{h^{3/2}\phi_{\infty,h,b,a}^{(1)}(O)}
{\sigma_{\infty,h,b}^{(1)}(a)},
\]
then
\begin{align*}
\sup_{t\in\mathbb R}\Bigg|&
P_0\left\{\sup_{a\in\mathcal A_0}|\GGn g_{h,b,a}^{(1)}|\le t\right\}
-P_0\left\{\sup_{a\in\mathcal A_0}|Z_{\infty,h,b}^{(1)}(a)|\le t\right\}
\Bigg|\longrightarrow0.
\end{align*}
If a finite mesh $\mathcal A_n\subset\mathcal A_0$ satisfies the mesh-size
requirement in condition~\ref{assm:G-grid}, the continuous Gaussian supremum
can be replaced in this display by
\[
\max_{a\in\mathcal A_n}|Z_{\infty,h,b}^{(1)}(a)|.
\]
\end{lemma}

\begin{proof}
Use the variance decomposition in the proof of
Lemma~\ref{lem:limiting_variance} with derivative weights and variance
normalization $h^3$.
The preceding lemma gives a normalized smoothing variance $O(h)$, and the Cauchy--Schwarz inequality bounds its covariance with the localized term by
$O(h^{1/2})$.  Lemma~\ref{lem:deriv-equivalent-kernel} and uniform
continuity of $\sigma_0^2f_A$ then prove the variance limit.  Apply this
continuity on the original supports of radii $h,b\to0$.
The same moment identities give $P_0\phi_{\infty,h,b,a}^{(1)}=0$.

We adapt the metric and entropy arguments in the proofs of
Lemmas~24--25 of \citet{Takatsu2025debiased}, which underlie
Lemmas~\ref{lem:limiting_gaussian_process}
and~\ref{lem:leading_term_gaussian_approx}, as follows.  Put $d_\ell(u,v)=1\wedge(|u-v|/\ell)$ for
$\ell\in\{h,b\}$, and use $p_{\ell,a,j}$ from
Lemma~\ref{lem:conditional-kernel-ep}.  The zero extensions of
$s^jK(s)$ are Lipschitz and compactly supported, so for $j\le6$,
\[
\|p_{\ell,u,j}-p_{\ell,v,j}\|_{P_0,2}^2
\le C\ell^{-1}d_\ell(u,v),\qquad
\|p_{\ell,u,j}-p_{\ell,v,j}\|_{P_0,1}
\le Cd_\ell(u,v).
\]
The population moment matrices, their inverses, and their $M,U,c$
coefficients are bounded and Lipschitz in the center.  Moreover,
\[
|\theta_{0,b}^{(3)}(a)|\le C,\qquad
|\theta_{0,b}^{(3)}(u)-\theta_{0,b}^{(3)}(v)|
\le Cb^{-3}|u-v|.
\]
Using these coefficient bounds in $\Gamma_{P,h,b,a}^{(1)}$ gives
\[
\begin{aligned}
h^3\|(\Gamma_{0,h,b,u}^{(1)}-\Gamma_{0,h,b,v}^{(1)})\xi_\infty\|_{P_0,2}^2\le C\{d_h(u,v)+\tau_n^7d_b(u,v)
+(1+\tau_n^7)|u-v|^2\}.
\end{aligned}
\]
For the last term of $\gamma_{P,h,b,a}^{(1)}(u)$, the kernel
translation and the displayed bound for $\theta_{0,b}^{(3)}$ give
\[
\begin{aligned}
&h^3\left\|\frac{h^2}{6}
\{\theta_{0,b}^{(3)}(u)\gamma_{0,h,u}^{c,1}
-\theta_{0,b}^{(3)}(v)\gamma_{0,h,v}^{c,1}\}\right\|_{P_0,2}^2
\\
&\le Ch^7h^{-1}
\{d_h(u,v)+|u-v|^2+b^{-6}|u-v|^2\}\\
&=C\{h^6d_h(u,v)+(h^6+\tau_n^6)|u-v|^2\}.
\end{aligned}
\]
Consequently, we obtain
\[
\begin{aligned}
&h^3\|\gamma_{0,h,b,u}^{(1)}-\gamma_{0,h,b,v}^{(1)}\|_{P_0,2}^2\\
&\qquad\le C\{(1+h^6)d_h(u,v)+\tau_n^7d_b(u,v)
+(1+\tau_n^7+h^6+\tau_n^6)|u-v|^2\}.
\end{aligned}
\]
For the integrated term $V_a$, the $L_1$ translation bound gives
\[
\|V_u-V_v\|_\infty\le Ch^{-1}d_h(u,v).
\]
Indeed, $\tau_n=O(1)$ and $d_b(u,v)\le\tau_nd_h(u,v)$ when
$|u-v|<h$; for the remaining pairs all the normalized norms are bounded.
Combining these bounds yields
\[
h^3\|\phi_{\infty,h,b,u}^{(1)}-
\phi_{\infty,h,b,v}^{(1)}\|_{P_0,2}^2\le Cd_h(u,v).
\]
By the reverse triangle inequality for
$\sigma_{\infty,h,b}^{(1)}(a)
=h^{3/2}\|\phi_{\infty,h,b,a}^{(1)}\|_{P_0,2}$ and the variance lower
bound, the canonical metric of $Z_{\infty,h,b}^{(1)}$ therefore satisfies
\[
\rho^{(1)}(u,v)\le Cd_h(u,v)^{1/2},\qquad
N(\varepsilon,\mathcal A_0,\rho^{(1)})\le C/(h\varepsilon^2).
\]
The chaining argument of Theorem~2.2.4 and Corollary~2.2.9 of
\citet{VanderVaartWellner2023weak} gives a tight version.  Anchoring at
one variance-one coordinate adds at most one to the bound, giving
\[
E_0\sup_{a\in\mathcal A_0}|Z_{\infty,h,b}^{(1)}(a)|
=O\{\sqrt{\log(h^{-1})}\}.
\]

Restrict the observation space to the full-$P_0$-measure set on which
the outcome bound holds.  The weight and design-correction envelopes
and $\|V_a\|_\infty\le Ch^{-1}$
give $|g_{h,b,a}^{(1)}|\le Ch^{-1/2}$.  Since
$P_0(g_{h,b,a}^{(1)})^2=1$, they also give
\[
\sup_aP_0|g_{h,b,a}^{(1)}|^3\le Ch^{-1/2},\qquad
\sup_aP_0|g_{h,b,a}^{(1)}|^4\le Ch^{-1}.
\]
The preceding translation calculation also holds in supremum norm:
\[
\|p_{\ell,u,j}-p_{\ell,v,j}\|_\infty
\le C\ell^{-1}d_\ell(u,v),\qquad
h^{3/2}\|\phi_{\infty,h,b,u}^{(1)}-
\phi_{\infty,h,b,v}^{(1)}\|_\infty
\le Ch^{-1/2}d_h(u,v).
\]
If $|u-v|<h$, the extra design correction contributes a coefficient
term at most $Chb^{-3}|u-v|$, bounded by $Ch^{-2}d_h(u,v)$;
for larger separations use its envelope.  Combining this with
the variance increment and lower bound already proved gives
\[
\|g_{h,b,u}^{(1)}-g_{h,b,v}^{(1)}\|_\infty
\le Ch^{-1/2}d_h(u,v)^{1/2}.
\]
Define $\mathcal F_{h,b}^{(1)}=\{\pm g_{h,b,a}^{(1)}:a\in\mathcal A_0\}$.
A grid of spacing $ch\varepsilon^2$ in the center therefore gives
\[
\sup_QN\!\left(\varepsilon Ch^{-1/2},
\mathcal F_{h,b}^{(1)},L_2(Q)\right)
\le C/(h\varepsilon^2),\qquad 0<\varepsilon\le1.
\]
This supremum-norm continuity also proves pointwise measurability.

Apply Corollary~2.2 of \citet{Chernozhukov2014Gaussian}, as in the proof
of Lemma~25 of \citet{Takatsu2025debiased}, to this centered,
variance-one class.  Its parameters are $B_n=Ch^{-1/2}$, $q=\infty$,
$\sigma=1$, and $\gamma_n=(\log n)^{-1}$; the logarithmic entropy factor
is $O(\log n)$.  The resulting scalar coupling has error $O_p(r_n)$, where
\[
r_n=(nh)^{-1/6}\log n+(nh)^{-1/4}(\log n)^{5/4}
+(nh)^{-1/2}(\log n)^{3/2}.
\]
The bound following \eqref{eq:deriv-bandwidth-log-consequences} gives
$r_n\sqrt{\log(h^{-1})}\to0$.  Lemma~2.4 of \citet{Chernozhukov2014Gaussian} therefore converts this
coupling into the stated Kolmogorov approximation, using the expected
supremum bound above and the normalization.

For the grid mesh, Corollary~2.2.9 of \citet{VanderVaartWellner2023weak}
and the metric bound give
\[
E_0\sup_{|u-v|\le\omega_n}
|Z_{\infty,h,b}^{(1)}(u)-Z_{\infty,h,b}^{(1)}(v)|
\le C\sqrt{\omega_n/h}\sqrt{\log(C/\omega_n)}
=o\{\log(h^{-1})^{-1/2}\}.
\]
The right side is interpreted as zero when $\omega_n=0$; the last
equality follows from $\omega_n=o(h^p)$, $p>1$.
Use the argument in the proof of Lemma~28 of
\citet{Takatsu2025debiased}, also used for
Lemma~\ref{lem:finite-mesh-gaussian}:
Markov's inequality controls this modulus at
$\varepsilon/\sqrt{\log(h^{-1})}$, and Lemma~A.1 of \citet{Chernozhukov2014Gaussian} bounds the
Gaussian probability of an interval of twice that width by
$C\varepsilon+o(1)$.  A distribution-function sandwich followed by
$\varepsilon\downarrow0$ proves the mesh assertion.
\end{proof}

\subsection{Proofs of the derivative inference theorems}

\begin{proof}[Proof of Theorem~\ref{thm:deriv-pointwise}]
Lemma~\ref{lem:deriv-decomposition},
Proposition~\ref{prop:deriv-remainders}, and Lemma~\ref{lem:deriv-R1}
give the expansion, and Lemma~\ref{lem:deriv-point-leading} gives the
normal limit for its leading term.  As in the proof of
Theorem~\ref{thm:pointwise-proximal-inference}, centering at the target
follows from
\[
\sqrt{nh^3}|R_{n,1}^{(1)}(a_0)|=\sqrt{nh^7}\,o(1)=o(1)
\]
and Slutsky's theorem, using condition~\ref{assm:deriv-point-centering}.
\end{proof}

\begin{proof}[Proof of Theorem~\ref{thm:deriv-finite}]
Use the Cram\'er--Wold device argument in the proof of Lemma~6 of
\citet{Takatsu2025debiased}, as for Lemma~\ref{lem:multiple_a_leading_term},
with normalization $\sqrt{nh^3}$.  The localized leading terms have mean
zero and eventually disjoint supports, since $\max(h,b)\to0$, so their
cross-covariances are zero.  The normalized covariances involving the
smoothing terms are $O(h^{1/2})$ by the bounds in
Lemma~\ref{lem:deriv-point-leading}.  The moment calculation there applies
to every fixed linear combination.  The componentwise remainder and bias
bounds used for Theorem~\ref{thm:deriv-pointwise} complete the result.
\end{proof}

\begin{proof}[Proof of Theorem~\ref{thm:deriv-uniform}]
Lemma~\ref{lem:deriv-decomposition},
Proposition~\ref{prop:deriv-remainders}, and Lemma~\ref{lem:deriv-R1}
give
\[
\begin{aligned}
&\sup_{a\in\mathcal A_0}
\left|\hat\theta_n^{(1),\mathrm{DB},\mathrm{cf}}(a)-\theta_0'(a)
-P_n\phi_{\infty,h,b,a}^{(1)}\right|\\
&\qquad\le O(h^{2+\alpha}+h^2b^\alpha)
+o_p\{(nh^3\log n)^{-1/2}\}
=o_p\{(nh^3\log n)^{-1/2}\}.
\end{aligned}
\]
The last equality is condition~\ref{assm:deriv-uniform-bandwidth}.
Centering and the uniform variance lower bound in
Lemma~\ref{lem:deriv-uniform-leading} show that the oracle-normalized
supremum in the theorem differs from
$\sup_a|\GGn g_{h,b,a}^{(1)}|$ by $o_p\{(\log n)^{-1/2}\}$.
Since $\log(h^{-1})=O(\log n)$, the coupling and Gaussian
anti-concentration argument in that lemma transfer its Kolmogorov
approximation to the estimator, exactly as in the proof of
Theorem~\ref{thm:uniform-proximal-inference}.  Its mesh argument validates the final conclusion.
\end{proof}

\section{Additional technical details for parametric null tests}
\label{sec:parametric-null-proofs}

\subsection{Constant, linear, and quadratic null families}\label{sec:parametric-null-polynomial-examples}

For a fixed integer $d\ge0$, the polynomial null
$m(a;\beta)=\sum_{j=0}^{d}\beta_{j+1}a^j$ specifies a dose-response
curve of degree at most $d$.
For $d=0,1,2$, smoothing the fitted null curve leaves it unchanged:
the debiased smoother exactly preserves constant, linear, and quadratic
functions, as shown in
Proposition~\ref{prop:parametric-null-reproduction} in the supplement.
Thus, the smoothed null equals $m(a;\hat\beta_{n,k})$, and the test
directly compares the debiased dose-response estimate with the fitted
null curve, giving the simpler formulas in
Table~\ref{tab:parametric-null-polynomial-families}.
For fixed $d>2$, smoothing may change the fitted polynomial.
The test therefore compares the estimated curve with the smoothed
fitted null defined in $\hat m_{n,k,h_n,b_n}(a;\beta)$, rather than
with $m(a;\hat\beta_{n,k})$ itself.
To leave these higher-degree fitted curves unchanged as well, one would need a smoother that exactly preserves polynomials through degree $d$.
A suitably constructed higher-order local-polynomial smoother could achieve this, but its asymptotic properties would require a separate analysis.

For the constant null family, define $\hat c_{n,k}:=P_{n,k}\hat\xi_n^{(-k)}$.
For $d\in\{1,2\}$, write $v_d(a)=(1,a,\ldots,a^d)^\top$ and define $\hat\beta_{n,k,d}
:=\left[P_{n,k}\{v_d(A)v_d(A)^\top\}\right]^{-1}
P_{n,k}\{v_d(A)\hat\xi_n^{(-k)}(O)\}$.
The following table gives the three
families, their statistics, and their local-alternative directions.
\begin{table}[t]
\centering
\caption{Mean models, test statistics, and projected local-alternative
directions for the constant, linear, and quadratic families.}
\label{tab:parametric-null-polynomial-families}
\begin{tabular}{@{}l@{\qquad}l@{}}
\toprule
Family & Mean model and true coefficient \\
\midrule
Constant & $m(a;\beta)=\beta$, \quad $\beta_0=c_0$ \\[0.5ex]
\multicolumn{2}{@{}l@{}}{$\displaystyle
\widehat T_{n,\mathrm{const}}^{DB}
:=\sum_{k=1}^K\frac{n_k}{n}\,n_k\sqrt{h_n}\int_{\mathcal A_0}
\{\hat\theta_{n,k}^{DB}(a)-\hat c_{n,k}\}^2w(a)\,da$} \\[0.8ex]
\multicolumn{2}{@{}l@{}}{$\displaystyle
 g_{\mathcal M}^{\perp}(a)=g(a)-P_0g(A)$} \\
\midrule
Linear & $m(a;\beta)=v_1(a)^\top\beta$,
\quad $\beta_0=(\beta_{0,1},\beta_{0,2})^\top$ \\[0.5ex]
\multicolumn{2}{@{}l@{}}{$\displaystyle
\widehat T_{n,\mathrm{lin}}^{DB}
:=\sum_{k=1}^K\frac{n_k}{n}\,n_k\sqrt{h_n}\int_{\mathcal A_0}
\{\hat\theta_{n,k}^{DB}(a)-v_1(a)^\top\hat\beta_{n,k,1}\}^2w(a)\,da$} \\[0.8ex]
\multicolumn{2}{@{}l@{}}{$\displaystyle
 g_{\mathcal M}^{\perp}(a)=g(a)-v_1(a)^\top
\{P_0[v_1(A)v_1(A)^\top]\}^{-1}P_0\{v_1(A)g(A)\}$} \\
\midrule
Quadratic & $m(a;\beta)=v_2(a)^\top\beta$,
\quad $\beta_0=(\beta_{0,1},\beta_{0,2},\beta_{0,3})^\top$ \\[0.5ex]
\multicolumn{2}{@{}l@{}}{$\displaystyle
\widehat T_{n,\mathrm{quad}}^{DB}
:=\sum_{k=1}^K\frac{n_k}{n}\,n_k\sqrt{h_n}\int_{\mathcal A_0}
\{\hat\theta_{n,k}^{DB}(a)-v_2(a)^\top\hat\beta_{n,k,2}\}^2w(a)\,da$} \\[0.8ex]
\multicolumn{2}{@{}l@{}}{$\displaystyle
 g_{\mathcal M}^{\perp}(a)=g(a)-v_2(a)^\top
\{P_0[v_2(A)v_2(A)^\top]\}^{-1}P_0\{v_2(A)g(A)\}$} \\
\bottomrule
\end{tabular}
\end{table}

\begin{corollary}
\label{cor:parametric-null-polynomials}
Under the assumptions of Theorem~\ref{thm:parametric-null}, specialized to
the corresponding family in Table~\ref{tab:parametric-null-polynomial-families},
\[
\left(V_\tau^{DB}\sum_{k=1}^K(n_k/n)^2\right)^{-1/2}\left\{\widehat T_{n,\mathcal M}^{DB}-b_{0h,n}^{DB}\right\}\rightsquigarrow N(0,1)
\]
holds with the famliy-specific replacements of $\widehat T_{n,\mathcal M}^{DB}$.
% :
% \begin{enumerate}[label=\textup{(\roman*)},ref=\thecorollary(\roman*),leftmargin=*]
% \item\label{cor:parametric-null-constancy}
% $\widehat T_{n,\mathrm{const}}^{DB}$ under
% $H_{0,\mathrm{const}}:\theta_0(a)=c_0$;
% \item\label{cor:parametric-null-linear}
% $\widehat T_{n,\mathrm{lin}}^{DB}$ under
% $H_{0,\mathrm{lin}}:\theta_0(a)=\beta_{0,1}+\beta_{0,2}a$;
% \item\label{cor:parametric-null-quadratic}
% $\widehat T_{n,\mathrm{quad}}^{DB}$ under
% $H_{0,\mathrm{quad}}:\theta_0(a)=\beta_{0,1}+\beta_{0,2}a+\beta_{0,3}a^2$.
% \end{enumerate}
\end{corollary}

\begin{corollary}
\label{cor:parametric-null-polynomials-local-power}
Under the assumptions of Theorem~\ref{thm:parametric-null}, specialized to
the corresponding family in Table~\ref{tab:parametric-null-polynomial-families}, along
\eqref{eq:parametric-null-local-alternative},
\[
\left(V_{\tau,n}^{DB}\sum_{k=1}^K(n_k/n)^2\right)^{-1/2}\left\{\widehat T_{n,\mathcal M}^{DB}-b_{0h,n}^{DB}
-c^2\sum_{k=1}^K(n_k/n)^2
\int_{\mathcal A_0}\{g_{\mathcal M}^{\perp}(a)\}^2w(a)\,da\right\}
\]
converges to $N(0,1)$ with the family-specific $\widehat T_{n,\mathcal M}^{DB}$, and $g_{\mathcal M}^{\perp}$.
\end{corollary}

\subsection{Assumptions and auxiliary results}

For $a\in\mathcal A$, recall
\[
\Gamma_{n,k,a}:=\Gamma_{n,k,h_n,b_n,a},
\qquad
\Gamma_{0,a}:=\Gamma_{0,h_n,b_n,a},
\qquad
\zeta_\infty(O):=\xi_\infty(O)-\theta_0(A).
\]
For every fold $k$, $\Delta_{n,k}$ is the function already defined by
\eqref{eq:conditional-proximal-product-bias}.

\begin{assumptionp}{L}\label{assm:parametric-null}
\par\medskip\noindent
\begin{enumerate}[label=\textup{(L\arabic*)},ref=L\arabic*,leftmargin=*,itemsep=0.35em]
\item\label{assm:parametric-null-support}
The function $w$ is nonnegative, nonzero, and continuously differentiable,
and its compact support is contained in the interior of $\mathcal A_0$.
The functions $f_A$ and $\sigma_0^2$ are continuously differentiable on
$\mathcal A_{\delta_3}$ and are bounded above and away from zero there.  Also,
\[
\sup_{a\in\mathcal A_{\delta_3}}
E_0\{|\zeta_\infty(O)|^4\mid A=a\}<\infty.
\]

\item\label{assm:parametric-null-bandwidth}
The bandwidths satisfy $nh_n^2\to\infty$.

\item\label{assm:parametric-null-family}
The set $\mathcal B\subset\mathbb R^p$ is compact, $p$ is fixed, and
$\beta_0$ is an interior point.  On a neighborhood of $\beta_0$,
$m(a;\beta)$ is twice continuously differentiable in $\beta$, uniformly in
$a\in\mathcal A$, and $m$ and its first two derivatives in $\beta$ are
uniformly bounded.  With
\[
\dot m_0(a):=
\left.\frac{\partial}{\partial\beta}m(a;\beta)\right|_{\beta=\beta_0},
\qquad
J_0:=P_0\{\dot m_0(A)\dot m_0(A)^\top\},
\]
$J_0$ is positive definite.  The fold-specific estimators in
\eqref{eq:parametric-null-beta} satisfy
\[
\max_{1\le k\le K}\|\hat\beta_{n,k}-\beta_0\|=o_p(1)
\]
and are interior solutions with probability tending to one.

\item\label{assm:parametric-null-l2}
\[
\max_k\sup_{a\in\mathcal A_{\delta_3}}
\left[
E_0\left\{
|\hat q_n^{(-k)}(Z,a,X)-q_\infty(Z,a,X)|^4
\mid A=a,\trainingsigma_k
\right\}
\right]^{1/4}=o_p(\sqrt{h_n}),
\]
\[
\max_k d_{h,k}(\hat h_n^{(-k)},h_\infty;
\mathcal A_{\delta_3},
\mathcal A_{\delta_3}\times\mathcal Z\times\mathcal W\times\mathcal X)
=o_p(\sqrt{h_n}).
\]

\item\label{assm:parametric-null-product}
The drift in \eqref{eq:conditional-proximal-product-bias} satisfies
\[
\max_k
\left\{\int_{\mathcal A_{\delta_3}}
\Delta_{n,k}(a)^2\,dF_0(a)\right\}^{1/2}
=o_p\{(n\sqrt{h_n})^{-1/2}\},
\]
\[
\max_k\int_{\mathcal A}|\Delta_{n,k}(a)|\,dF_0(a)=o_p(1),
\]
\[
\max_k\left\|
\int_{\mathcal A}\dot m_0(a)\Delta_{n,k}(a)\,dF_0(a)
\right\|
=o_p\{(n\sqrt{h_n})^{-1/2}\}.
\]

\item\label{assm:parametric-null-score-moment}
If $(W',X')$ is an independent copy of $(W,X)$, independent of
$\trainingsigma_k$ and of a validation observation, then
\[
\max_k E_0\left[
\begin{aligned}
&|\hat q_n^{(-k)}(Z,A,X)
\{Y-\hat h_n^{(-k)}(W,A,X)\}|^4\\
&\quad+|\hat h_n^{(-k)}(W',A,X')|^4
+|\hat h_n^{(-k)}(W,A,X)|^4
\end{aligned}
\,\middle|\,\trainingsigma_k\right]=O_p(1).
\]
\end{enumerate}
\end{assumptionp}
\noindent Condition~\ref{assm:parametric-null-product} bounds both the local
$L_2(F_0)$ norm of $\Delta_{n,k}$ and its integral against $\dot m_0$.
The first is used in the curve remainder; the latter is used in the
foldwise score equation.

Define the oracle residual process
\begin{align}
L_{n,k,\mathcal M}(a)
:={}&P_{n,k}\left[
\Gamma_{0,a}(A)\{\xi_\infty(O)-m(A;\hat\beta_{n,k})\}\right]\notag\\
&+P_{n,k}\left[
\int\Gamma_{0,a}(\bar a)
\{h_\infty(W,\bar a,X)-m_\infty(\bar a)\}\,dF_0(\bar a)
\right]
\label{eq:parametric-null-leading}
\end{align}
and
\begin{equation}
T_{n,\mathcal M,0}^{DB}
:=\sum_{k=1}^K\frac{n_k}{n}\,n_k\sqrt{h_n}
\int_{\mathcal A_0}L_{n,k,\mathcal M}(a)^2w(a)\,da.
\label{eq:parametric-null-oracle}
\end{equation}
For $a\in\mathcal A_0$ and each fold $k$, define the following remainder terms:
\begin{align}
R_{n,k,2}^{\mathcal M}(a)
:={}&(P_{n,k}-P_0)
\Bigg[
\Gamma_{0,a}(A)\{\hat\vartheta_{n,k}(O)-\xi_\infty(O)\}\notag\\
&\qquad+\int\Gamma_{0,a}(\bar a)
\{\hat h_n^{(-k)}(W,\bar a,X)-h_\infty(W,\bar a,X)\}
\,dF_0(\bar a)
\Bigg],
\label{eq:parametric-null-R2}\\
R_{n,k,3}^{\mathcal M}(a)
:={}&(P_{n,k}-P_0)
\Bigg[
\{\Gamma_{n,k,a}(A)-\Gamma_{0,a}(A)\}
\{\hat\vartheta_{n,k}(O)-\theta_0(A)\}\notag\\
&\qquad+\int
\{\Gamma_{n,k,a}(\bar a)-\Gamma_{0,a}(\bar a)\}
\hat h_n^{(-k)}(W,\bar a,X)\,dF_0(\bar a)
\Bigg],
\label{eq:parametric-null-R3}\\
R_{n,k,4}^{\mathcal M}(a)
:={}&\int\Gamma_{n,k,a}(\bar a)\Delta_{n,k}(\bar a)\,dF_0(\bar a),
\notag\\
R_{n,k,5}^{\mathcal M}(a)
:={}&
\int\Gamma_{n,k,a}(\bar a)\hat h_n^{(-k)}(w,\bar a,x)
\,d\{(Q_{n,k}^{WX}-Q_0^{WX})\times(F_{n,k}-F_0)\}(w,x,\bar a),
\label{eq:parametric-null-R45}\\
R_{n,k,6}^{\mathcal M}(a)
:={}&-P_{n,k}\left[
\{\Gamma_{n,k,a}(A)-\Gamma_{0,a}(A)\}
\{m(A;\hat\beta_{n,k})-m(A;\beta_0)\}
\right].
\label{eq:parametric-null-R6}
\end{align}
Here $Q_{n,k}^{WX}$ and $F_{n,k}$ are the validation-fold empirical
marginals defined in Section~\ref{sec: remainder decomposition}.
The function $\Delta_{n,k}$ is the signed conditional bias in
\eqref{eq:conditional-proximal-product-bias}.

\begin{lemma}
\label{lem:parametric-null-decomposition}
Under the assumptions of Theorem~\ref{thm:parametric-null},
for all sufficiently large $n$, on the event that the fold empirical
moment matrices are nonsingular,
\begin{equation}
\hat\theta_{n,k}^{DB}(a)
-\hat m_{n,k,h_n,b_n}(a;\hat\beta_{n,k})
=L_{n,k,\mathcal M}(a)+\sum_{j=2}^6R_{n,k,j}^{\mathcal M}(a)
\label{eq:parametric-null-curve-decomposition}
\end{equation}
for every $a\in\mathcal A_0$ and every fold $k$.
\end{lemma}

\begin{proof}
Recall the exact product-measure identity:
\begin{align*}
Q_{n,k}^{WX}\times F_{n,k}
={}&Q_0^{WX}\times F_0
+Q_0^{WX}\times(F_{n,k}-F_0)
+(Q_{n,k}^{WX}-Q_0^{WX})\times F_0\\
&+(Q_{n,k}^{WX}-Q_0^{WX})\times(F_{n,k}-F_0).
\end{align*}
Now we apply it to
$g(w,x,\bar a)=\Gamma_{n,k,a}(\bar a)
\hat h_n^{(-k)}(w,\bar a,x)$ in the fold curve representation.
The residual term and the first three product-measure terms equal
\[
\begin{aligned}
&(P_{n,k}-P_0)\left[
\Gamma_{n,k,a}(A)\hat\vartheta_{n,k}(O)
+\int\Gamma_{n,k,a}(\bar a)
\hat h_n^{(-k)}(W,\bar a,X)\,dF_0(\bar a)
\right]\\
&\qquad+P_0\{\Gamma_{n,k,a}(A)\hat\vartheta_{n,k}(O)\},
\end{aligned}
\]
and the fourth is $R_{n,k,5}^{\mathcal M}(a)$.
Holding the realized observations fixed, the conditional-bias identity
\eqref{eq:conditional-proximal-product-bias} gives
\[
P_0\{\Gamma_{n,k,a}(A)\hat\vartheta_{n,k}(O)\}
=P_0\{\Gamma_{n,k,a}(A)\theta_0(A)\}
+\int\Gamma_{n,k,a}(\bar a)\Delta_{n,k}(\bar a)\,dF_0(\bar a).
\]
Consequently, we get
\[
\begin{aligned}
&\hat\theta_{n,k}^{DB}(a)
-P_{n,k}\{\Gamma_{n,k,a}(A)\theta_0(A)\}\\
&\quad=(P_{n,k}-P_0)\Bigg[
\Gamma_{n,k,a}(A)\{\hat\vartheta_{n,k}(O)-\theta_0(A)\}\\
&\hspace{8em}+\int\Gamma_{n,k,a}(\bar a)
\hat h_n^{(-k)}(W,\bar a,X)\,dF_0(\bar a)\Bigg]
+R_{n,k,4}^{\mathcal M}(a)+R_{n,k,5}^{\mathcal M}(a).
\end{aligned}
\]
Split $\Gamma_{n,k,a}$ into $\Gamma_{0,a}$ and its difference from
$\Gamma_{0,a}$, and add and subtract $\xi_\infty$ and $h_\infty$ in the
population-kernel part.  On the local kernel support,
condition~\ref{assm:G-DR-rates} and the conditional-bias identity provide
\[
E_0(\xi_\infty\mid A)=\theta_0(A),\qquad
E_0\{h_\infty(W,t,X)\}=m_\infty(t).
\]
These identities yield
\[
\begin{aligned}
&P_{n,k}\left[
\Gamma_{0,a}(A)\{\xi_\infty-\theta_0(A)\}
+\int\Gamma_{0,a}(\bar a)
\{h_\infty(W,\bar a,X)-m_\infty(\bar a)\}\,dF_0(\bar a)
\right]\\
&\quad+\sum_{j=2}^5R_{n,k,j}^{\mathcal M}(a).
\end{aligned}
\]
Under the null $\theta_0=m(\cdot;\beta_0)$.  Subtract
$P_{n,k}[\Gamma_{n,k,a}(A)\{m(A;\hat\beta_{n,k})-m(A;\beta_0)\}]$,
and split its kernel into the same two parts.  This gives
\eqref{eq:parametric-null-leading} and
\eqref{eq:parametric-null-R6}, proving the claimed decomposition.
The deterministic smoothing term cancels since
\[
(P_{n,k}-P_0)\{\Gamma_{n,k,a}(A)\theta_0(A)\}
+P_0\{\Gamma_{n,k,a}(A)\theta_0(A)\}
-P_{n,k}\{\Gamma_{n,k,a}(A)\theta_0(A)\}=0.
\]
\end{proof}

For a function $g$ on $\mathcal A_0$, define
\[
\|g\|_w^2:=\int_{\mathcal A_0}g(a)^2w(a)\,da,
\qquad
\langle g_1,g_2\rangle_w
:=\int_{\mathcal A_0}g_1(a)g_2(a)w(a)\,da.
\]

\begin{lemma}
\label{lem:parametric-null-remainder-rates}
Under the assumptions of Theorem~\ref{thm:parametric-null},
\[
\max_k\|R_{n,k,2}^{\mathcal M}\|_w=o_p(n^{-1/2}),
\qquad
\max_k\|R_{n,k,3}^{\mathcal M}\|_w=O_p\{(nh_n)^{-1}\},
\]
\[
\max_k\|R_{n,k,4}^{\mathcal M}\|_w=o_p(\varrho_n),
\qquad
\max_k\|R_{n,k,5}^{\mathcal M}\|_w=O_p(n^{-1}h_n^{-1/2}),
\]
\[
\max_k\|R_{n,k,6}^{\mathcal M}\|_w=o_p(n^{-1/2}).
\]
\end{lemma}

\begin{proof}
Throughout, $b_n\asymp h_n$ and $n_k\asymp n$.
The uniform matrix bounds in the proof of Lemma~\ref{lem:R6-rate},
together with Lemma~\ref{lem:D-expansion}, give an event of probability
tending to one on which all empirical inverses and coefficients in
$\Gamma_{n,k,a}$ are uniformly bounded.  
The conditional bounds for a fixed
fold are applied on its own bounded-fit event based on condition~\ref{assm:G-crossfit}; their constants are independent
of the training split.  The validation-fold matrix coefficients are
multiplied into the resulting bounds afterwards.

For $R_{n,k,2}^{\mathcal M}$, write the bracket in
\eqref{eq:parametric-null-R2} as the $L_2(\mathcal A_0,w)$-valued function
\[
\begin{aligned}
g_{n,k}(O)(a):={}&
\Gamma_{0,a}(A)\{\hat\vartheta_{n,k}(O)-\xi_\infty(O)\}\\
&+\int\Gamma_{0,a}(t)
\{\hat h_n^{(-k)}(W,t,X)-h_\infty(W,t,X)\}\,dF_0(t).
\end{aligned}
\]
The pseudo-outcome expansion and the Bayes's rule and the change of measure in the proof
of Lemma~\ref{lem:R2-control} still apply.  In particular,
condition~\ref{assm:parametric-null-l2} and H\"older's
inequality give, uniformly in $a\in\mathcal A_{\delta_3}$,
\[
\begin{aligned}
&E_0\left[
 (\hat q_n^{(-k)}-q_\infty)^2(Y-h_\infty)^2
 \mid A=a,\trainingsigma_k\right]\\
&\quad\le
 \left[E_0\{|\hat q_n^{(-k)}-q_\infty|^4
 \mid A=a,\trainingsigma_k\}\right]^{1/2}
 \left[E_0\{|Y-h_\infty|^4\mid A=a\}\right]^{1/2}
 =o_p(h_n).
\end{aligned}
\]
On the bounded-fit event in each fold, we get
\[
\begin{aligned}
&E_0\!\left[
 \{\hat q_n^{(-k)}(h_\infty-\hat h_n^{(-k)})\}^2
 \,\middle|\,A=a,\trainingsigma_k\right]\\
&\qquad\le C_q^2
 E_0\!\left[
 \{\hat h_n^{(-k)}(W,a,X)-h_\infty(W,a,X)\}^2
 \,\middle|\,A=a,\trainingsigma_k\right]
=o_p(h_n)
\end{aligned}
\]
uniformly for almost every $a\in\mathcal A_{\delta_3}$.
The change of measure in the proof of Lemma~\ref{lem:R2-control} gives
\[
\begin{aligned}
&|S_{\hat h_n^{(-k)}}(a)-m_\infty(a)|^2\\
&\quad\le
\int\{\hat h_n^{(-k)}(w,a,x)-h_\infty(w,a,x)\}^2\,dQ_0^{WX}(w,x)\\
&\quad\le \frac{\sup_{t\in\mathcal A_{\delta_3}} f_A(t)}{c_1}\,
 E_0\!\left[
 \{\hat h_n^{(-k)}(W,a,X)-h_\infty(W,a,X)\}^2
 \,\middle|\,A=a,\trainingsigma_k\right]
=o_p(h_n).
\end{aligned}
\]
Consequently, we obtain
\[
\sup_{a\in\mathcal A_{\delta_3}}
E_0\{(\hat\vartheta_{n,k}-\xi_\infty)^2
 \mid A=a,\trainingsigma_k\}=o_p(h_n).
\]
The bounds in Lemma~\ref{lem:parametric-null-kernel}, followed by the Fubini-Tonelli theorem and the Cauchy--Schwarz inequality yield
\[
\begin{aligned}
P_0\|g_{n,k}\|_w^2
&\le \frac C{h_n}
 \sup_{t\in\mathcal A_{\delta_3}}
 E_0\{(\hat\vartheta_{n,k}-\xi_\infty)^2
       \mid A=t,\trainingsigma_k\}\\
&\quad+C\sup_{t\in\mathcal A_{\delta_3}}
 \int\{\hat h_n^{(-k)}(w,t,x)-h_\infty(w,t,x)\}^2
      \,dQ_0^{WX}(w,x)\\
&=h_n^{-1}o_p(h_n)+o_p(h_n)=o_p(1).
\end{aligned}
\]
The conditional independence provides the exact variance identity
\[
\begin{aligned}
E_0\left[\|(P_{n,k}-P_0)g_{n,k}\|_w^2
 \mid\trainingsigma_k\right]
&=\frac{P_0\|g_{n,k}\|_w^2-\|P_0g_{n,k}\|_w^2}{n_k}\\
&=o_p(n_k^{-1}).
\end{aligned}
\]
Next, Markov's inequality and the fixed number of folds prove the first claim.  This strengthens the rate in Lemma~\ref{lem:R2-control}
by using the integrated second moment and
condition~\ref{assm:parametric-null-l2}.

For $R_{n,k,3}^{\mathcal M}$, use the exact coefficient expansion in
Lemma~\ref{lem:Gamma-difference-expansion}, with every empirical moment
computed under $P_{n,k}$.  Write its coefficients as
$\alpha_{n,k,h,j}$ and $\alpha_{n,k,b,j}$.
Use the existing functions
$p_{\ell,a,j}(t)=\{(t-a)/\ell\}^jK_{\ell,a}(t)$.
For every applicable fixed degree $j$, compact kernel support gives
$P_0|p_{\ell,a,j}|^r\le C\ell^{1-r}$ for $r=2,4$.
The fourth moment of a centered empirical mean is thus
\[
\begin{aligned}
E_0[\{(P_{n,k}-P_0)p_{\ell,a,j}\}^4\mid\trainingsigma_k]
&=\frac{P_0(p_{\ell,a,j}-P_0p_{\ell,a,j})^4}{n_k^3}\\
&\quad+\frac{3(n_k-1)}{n_k^3}
 \{P_0(p_{\ell,a,j}-P_0p_{\ell,a,j})^2\}^2\\
&\leq C\{n_k^{-3}\ell^{-3}+n_k^{-2}\ell^{-2}\}.
\end{aligned}
\]
Integrating this bound and using the inverse and coefficient identities
in the proof of Lemma~\ref{lem:Gamma-difference-expansion} yields
\[
\max_{j\le1}\int|\alpha_{n,k,h,j}(a)|^4w(a)\,da
+\max_{j\le2}\int|\alpha_{n,k,b,j}(a)|^4w(a)\,da
=O_p\{(nh_n)^{-2}\}.
\]
Indeed, on the event where the empirical inverse matrices are uniformly
bounded, the identities in the proof of
Lemma~\ref{lem:Gamma-difference-expansion} give
\[
\begin{aligned}
&\max_{0\le j\le1}|\alpha_{n,k,h,j}(a)|
+\max_{0\le j\le2}|\alpha_{n,k,b,j}(a)|\\
&\qquad\le C\left\{
\sum_{j=0}^3|(P_{n,k}-P_0)p_{h_n,a,j}|
+\sum_{j=0}^4|(P_{n,k}-P_0)p_{b_n,a,j}|
\right\}.
\end{aligned}
\]
The preceding fourth-moment bound therefore implies the 
integrated fourth-moment bound for the coefficients.

Substituting the same coefficient expansion into
\eqref{eq:parametric-null-R3}, each coefficient multiplies one of
the two centered empirical means
\[
(P_{n,k}-P_0)
\{p_{\ell,a,j}(A)(\hat\vartheta_{n,k}-\theta_0(A))\},
\quad
(P_{n,k}-P_0)
\left\{\int p_{\ell,a,j}(t)
\hat h_n^{(-k)}(W,t,X)\,dF_0(t)\right\},
\]
where $\ell=h_n$, $j=0,1$, or $\ell=b_n$, $j=0,1,2$.
Conditional on $\trainingsigma_k$, both functions inside these
empirical means are fixed.

On the bounded-fit event in each fold,
conditions~\ref{assm:G-crossfit} and~\ref{assm:G-smooth}
imply that $\hat\vartheta_{n,k}-\theta_0$ is uniformly bounded
on the kernel support.  Hence, for $r=2,4$, we get
\[
P_0\left|
p_{\ell,a,j}(A)(\hat\vartheta_{n,k}-\theta_0(A))
\right|^r
\le C P_0|p_{\ell,a,j}|^r
\le C\ell^{1-r}.
\]
Applying the preceding fourth-moment identity to this function
in place of $p_{\ell,a,j}$ and integrating over $a$ gives
\[
\begin{aligned}
&E_0\left[
\int_{\mathcal A_0}\left|
(P_{n,k}-P_0)
\{p_{\ell,a,j}(A)(\hat\vartheta_{n,k}-\theta_0(A))\}
\right|^4w(a)\,da
\,\middle|\,\trainingsigma_k\right]\\
&\qquad\le C\{n_k^{-3}\ell^{-3}+n_k^{-2}\ell^{-2}\}
=O\{(n_k\ell)^{-2}\},
\end{aligned}
\]
where the last equality uses $n_k\ell\to\infty$.
For the second empirical mean, we have
\[
\left|\int p_{\ell,a,j}(t)
\hat h_n^{(-k)}(W,t,X)\,dF_0(t)\right|
\le C_h\int|p_{\ell,a,j}(t)|\,dF_0(t)
\le C.
\]
The same fourth-moment identity, therefore, yields
\[
E_0\left[
\int_{\mathcal A_0}\left|
(P_{n,k}-P_0)
\left\{\int p_{\ell,a,j}(t)
\hat h_n^{(-k)}(W,t,X)\,dF_0(t)\right\}
\right|^4w(a)\,da
\,\middle|\,\trainingsigma_k\right]
\le Cn_k^{-2}.
\]

We now multiply these empirical means by their coefficients.
H\"older's inequality with respect to $w(a)\,da$ gives
\[
\begin{aligned}
&\left\|\alpha_{n,k,h,j}(\cdot)(P_{n,k}-P_0)
\{p_{h_n,\cdot,j}(A)(\hat\vartheta_{n,k}-\theta_0(A))\}
\right\|_w\\
&\quad\le
\left\{\int_{\mathcal A_0}
|\alpha_{n,k,h,j}(a)|^4w(a)\,da\right\}^{1/4}\\
&\qquad\times
\left\{\int_{\mathcal A_0}\left|
(P_{n,k}-P_0)
\{p_{h_n,a,j}(A)(\hat\vartheta_{n,k}-\theta_0(A))\}
\right|^4w(a)\,da\right\}^{1/4}\\
&\quad=
O_p\{(nh_n)^{-1/2}\}\,
O_p\{(nh_n)^{-1/2}\}
=O_p\{(nh_n)^{-1}\}.
\end{aligned}
\]
Here the second factor rate follows from Markov's
inequality.  For the empirical mean involving
$\int p_{\ell,a,j}\hat h_n^{(-k)}\,dF_0$, the same inequality gives
\[
O_p\{(nh_n)^{-1/2}\}\,O_p(n^{-1/2})
=O_p(n^{-1}h_n^{-1/2}).
\]
The terms with $\alpha_{n,k,b,j}$ satisfy the same bounds because
$b_n\asymp h_n$.  Summing the finitely many terms and using
$\tau_n^2=O(1)$ yields
$\|R_{n,k,3}^{\mathcal M}\|_w=O_p\{(nh_n)^{-1}\}$.

For $R_{n,k,4}^{\mathcal M}$, bounded empirical coefficients give
\[
\sup_{a\in\mathcal A_0}\int|\Gamma_{n,k,a}(t)|\,dF_0(t)\le C,
\qquad
\sup_{t\in\mathcal A_{\delta_3}}
\int_{\mathcal A_0}|\Gamma_{n,k,a}(t)|w(a)\,da\le C
\]
on an event of probability tending to one, uniformly over $k$.
For all large $n$, every kernel window with center in the support of $w$
is contained in $\mathcal A_{\delta_3}$.
The Cauchy--Schwarz inequality with respect to
$|\Gamma_{n,k,a}(t)|dF_0(t)$ and the Fubini-Tonelli theorem therefore yields
\begin{align*}
\|R_{n,k,4}^{\mathcal M}\|_w^2
&\le C\int_{\mathcal A_0}\int
 |\Gamma_{n,k,a}(t)|\Delta_{n,k}(t)^2\,dF_0(t)w(a)\,da\\
&\le C\int_{\mathcal A_{\delta_3}}\Delta_{n,k}(t)^2\,dF_0(t)
=o_p(\varrho_n^2).
\end{align*}
The last equality is the first display in
condition~\ref{assm:parametric-null-product}.

For $R_{n,k,5}^{\mathcal M}$, use
$f_{n,k,\ell,j,a}^{\circ}$ and
$\widetilde f_{n,k,\ell,j,a}^{\circ}$ from the proof of
Lemma~\ref{lem:conditional-product-process}.
Its exact identity~\eqref{eq:conditional-product-exact} separates the
ordered off-diagonal sum and the repeated-index sum.  The kernel bounds
in that proof, integrated against $w(a)da$, give
\[
\begin{aligned}
E_0\left[
\|\widetilde f_{n,k,\ell,j,\cdot}^{\circ}(O_1,O_2)\|_w^2
\mid\trainingsigma_k\right]\le C\ell^{-1},\qquad
E_0\left[
\|f_{n,k,\ell,j,\cdot}^{\circ}(O,O)\|_w^2
\mid\trainingsigma_k\right]\le C\ell^{-1}.
\end{aligned}
\]
For $r\ne s$ and $t\ne u$, conditional independence and the
zero conditional means in each coordinate imply, for $\{r,s\}\ne\{t,u\}$,
\[
E_0\!\left[
\widetilde f^\circ(O_r,O_s)\widetilde f^\circ(O_t,O_u)
\,\middle|\,\trainingsigma_k\right]=0.
\]
Thus only $(t,u)=(r,s)$ and $(t,u)=(s,r)$ contribute to the
second moment.  By symmetry, each contributes
$P_0^{\otimes2}(\widetilde f^\circ)^2$.
Hence
\[
\begin{aligned}
E_0\left[
\left\|\frac1{n_k^2}\sum_{\substack{r,s\in I_k\\r\ne s}}
\widetilde f_{n,k,\ell,j,\cdot}^{\circ}(O_r,O_s)
\right\|_w^2\,\middle|\,\trainingsigma_k\right]&=\frac{2n_k(n_k-1)}{n_k^4}
E_0\left[
\|\widetilde f_{n,k,\ell,j,\cdot}^{\circ}(O_1,O_2)\|_w^2
\mid\trainingsigma_k\right]\\
&\leq\frac C{n_k^2\ell}.
\end{aligned}
\]
The diagonal need not have mean zero.  Its exact second moment is
\[
\begin{aligned}
&E_0\left[
\left\|\frac1{n_k^2}\sum_{r\in I_k}
f_{n,k,\ell,j,\cdot}^{\circ}(O_r,O_r)\right\|_w^2
\,\middle|\,\trainingsigma_k\right]\\
&\quad=\frac1{n_k^3}E_0\left[
\|f_{n,k,\ell,j,\cdot}^{\circ}(O,O)\|_w^2
\mid\trainingsigma_k\right]+\frac{n_k-1}{n_k^3}
\left\|E_0\{f_{n,k,\ell,j,\cdot}^{\circ}(O,O)
\mid\trainingsigma_k\}\right\|_w^2
\le\frac C{n_k^2\ell},
\end{aligned}
\]
where the last bound uses Jensen's inequality.  
Next, Markov's inequality yields $O_p(n_k^{-1}\ell^{-1/2})$ for their sum.
The kernel-coordinate expansion in the proof of
Lemma~\ref{lem:R5-control} then gives
$\|R_{n,k,5}^{\mathcal M}\|_w=O_p(n^{-1}h_n^{-1/2})$ after
multiplication by its uniformly bounded empirical coefficients.

Finally, write
\[
P_{n,k}|p_{\ell,a,j}|
=P_0|p_{\ell,a,j}|
+(P_{n,k}-P_0)|p_{\ell,a,j}|.
\]
Since $P_0|p_{\ell,a,j}|\le C$ uniformly in $a$, the inequality
$|x+y|^4\le8(|x|^4+|y|^4)$ gives
\[
\begin{aligned}
E_0\left[
\int_{\mathcal A_0}\{P_{n,k}|p_{\ell,a,j}|\}^4w(a)\,da
\,\middle|\,\trainingsigma_k\right]
&\le
8\int_{\mathcal A_0}\{P_0|p_{\ell,a,j}|\}^4w(a)\,da\\
&\qquad+
8E_0\left[
\int_{\mathcal A_0}
\left|(P_{n,k}-P_0)|p_{\ell,a,j}|\right|^4w(a)\,da
\,\middle|\,\trainingsigma_k\right]\\
&\le C\{1+(n_k\ell)^{-2}+(n_k\ell)^{-3}\}
=O(1).
\end{aligned}
\]
The second inequality follows from the preceding fourth-moment
identity applied to $|p_{\ell,a,j}|$.
Then, Markov's inequality yields
\[
\int_{\mathcal A_0}\{P_{n,k}|p_{\ell,a,j}|\}^4w(a)\,da
=O_p(1).
\]
The coefficient expansion in Lemma~\ref{lem:Gamma-difference-expansion}
and H\"older's inequality give the following bound for every realized
sample and every fixed $C<\infty$:
\[
\begin{aligned}
&\sup_{\|v\|_\infty\le C}
 \left\|P_{n,k}\{(\Gamma_{n,k,\cdot}-\Gamma_{0,\cdot})v\}\right\|_w\\
&\quad\le C\sum_{j=0}^1
 \left\{\int|\alpha_{n,k,h,j}(a)|^4w(a)\,da\right\}^{1/4}
 \left\{\int(P_{n,k}|p_{h_n,a,j}|)^4w(a)\,da\right\}^{1/4}\\
&\qquad+C\tau_n^2\sum_{j=0}^2
 \left\{\int|\alpha_{n,k,b,j}(a)|^4w(a)\,da\right\}^{1/4}
 \left\{\int(P_{n,k}|p_{b_n,a,j}|)^4w(a)\,da\right\}^{1/4}\\
&\quad=O_p\{(nh_n)^{-1/2}\}.
\end{aligned}
\]
Taylor's expansion of $m(\cdot;\hat\beta_{n,k})-m(\cdot;\beta_0)$ in
\eqref{eq:parametric-null-R6}, condition~\ref{assm:parametric-null-family},
and Lemma~\ref{lem:parametric-null-beta-rate} give
\[
\begin{aligned}
\|R_{n,k,6}^{\mathcal M}\|_w
&=O_p\{(nh_n)^{-1/2}\}
\{\|\hat\beta_{n,k}-\beta_0\|+\|\hat\beta_{n,k}-\beta_0\|^2\}\\
&=o_p(n^{-1}h_n^{-3/4})
=o_p(n^{-1/2}).
\end{aligned}
\]
For the last equality, $nh_n^2\to\infty$ implies
$nh_n^{3/2}\to\infty$.
\end{proof}

\begin{proposition}
\label{prop:parametric-null-statistic-decomposition}
Under $H_{0,\mathcal M}$, on the event in
Lemma~\ref{lem:parametric-null-decomposition},
\begin{align}
\widehat T_{n,\mathcal M}^{DB}-T_{n,\mathcal M,0}^{DB}
={}&\sum_{k=1}^K\frac{n_k}{n}\,n_k\sqrt{h_n}
\Bigg[
2\sum_{j=2}^6
\langle L_{n,k,\mathcal M},R_{n,k,j}^{\mathcal M}\rangle_w
+\sum_{j=2}^6\|R_{n,k,j}^{\mathcal M}\|_w^2\notag\\
&\hspace{8em}
+2\sum_{2\le j<\ell\le6}
\langle R_{n,k,j}^{\mathcal M},R_{n,k,\ell}^{\mathcal M}\rangle_w
\Bigg].
\label{eq:parametric-null-statistic-decomposition}
\end{align}
Under the assumptions of Theorem~\ref{thm:parametric-null}, the right-hand
side is $o_p(1)$.
\end{proposition}

\begin{proof}
The identity follows by squaring
\eqref{eq:parametric-null-curve-decomposition}, subtracting
$\|L_{n,k,\mathcal M}\|_w^2$, and summing with the displayed weights.
For each fold, write
\[
\begin{aligned}
Y_{n,k}(a)&:=P_{n,k}\{\Gamma_{0,a}(A)\zeta_\infty(O)\},\\
D_{n,k}(a)&:=P_{n,k}\left[
\Gamma_{0,a}(A)\{m(A;\hat\beta_{n,k})-m(A;\beta_0)\}\right],\\
U_{n,k}(a)&:=P_{n,k}\left[
\int\Gamma_{0,a}(t)\{h_\infty(W,t,X)-m_\infty(t)\}\,dF_0(t)\right].
\end{aligned}
\]
Then $L_{n,k,\mathcal M}=Y_{n,k}-D_{n,k}+U_{n,k}$.
Lemma~\ref{lem:parametric-null-beta-rate}, the bounded parametric
derivatives, and the kernel bounds in
Lemma~\ref{lem:parametric-null-kernel} give
\[
\max_k\|Y_{n,k}\|_w=O_p\{(nh_n)^{-1/2}\},~
\max_k\|D_{n,k}\|_w=o_p(\varrho_n),~ \max_k\|U_{n,k}\|_w=O_p(n^{-1/2}).
\]
For the first and third bounds, integrate the variance of the centered
empirical mean.  For the second, use
$\|m(\cdot;\hat\beta_{n,k})-m(\cdot;\beta_0)\|_\infty
\le C\|\hat\beta_{n,k}-\beta_0\|$ and
$\sup_a P_{n,k}|\Gamma_{0,a}|=O_p(1)$.
In particular, $\max_k\|L_{n,k,\mathcal M}\|_w
=O_p\{(nh_n)^{-1/2}\}$.

For $a,u\in\mathcal A_0$, conditional independence and
$E_0(\zeta_\infty\mid A)=0$ on the kernel support give
\[
\begin{aligned}
E_0\{Y_{n,k}(a)\mid\trainingsigma_k\}&=0,\\
E_0\{Y_{n,k}(a)Y_{n,k}(u)\mid\trainingsigma_k\}
&=\frac1{n_k^2}\sum_{i\in I_k}
E_0\!\left[
\Gamma_{0,a}(A_i)\Gamma_{0,u}(A_i)\zeta_\infty(O_i)^2
\,\middle|\,\trainingsigma_k\right]\\
&=\frac1{n_k}\int
\Gamma_{0,a}(s)\Gamma_{0,u}(s)\sigma_0^2(s)\,dF_0(s).
\end{aligned}
\]
Since $\Delta_{n,k}$ is $\trainingsigma_k$-measurable,
this covariance identity implies
\[
\begin{aligned}
&E_0\Bigg[
\Bigg\{n_k\sqrt{h_n}\int_{\mathcal A_0}Y_{n,k}(a)
\left[\int\Gamma_{0,a}(t)\Delta_{n,k}(t)\,dF_0(t)\right]
w(a)\,da\Bigg\}^{\!2}
\,\Biggm|\,\trainingsigma_k\Bigg]\\
&\quad=n_kh_n\int\sigma_0^2(s)
\Bigg\{\int_{\mathcal A_0}\Gamma_{0,a}(s)
\left[\int\Gamma_{0,a}(t)\Delta_{n,k}(t)\,dF_0(t)\right]
w(a)\,da\Bigg\}^{\!2}dF_0(s).
\end{aligned}
\]

For all sufficiently large $n$, $\Gamma_{0,a}(t)=0$ whenever
$a\in\mathcal A_0$ and $t\notin\mathcal A_{\delta_3}$.
Lemma~\ref{lem:parametric-null-kernel} gives
\[
\sup_{a\in\mathcal A_0}
\int|\Gamma_{0,a}(t)|\,dF_0(t)\le C,
\qquad
\sup_t\int_{\mathcal A_0}
|\Gamma_{0,a}(t)|w(a)\,da\le C.
\]
Thus, the Cauchy--Schwarz inequality yields, for each $a\in\mathcal A_0$,
\[
\begin{aligned}
\left|\int\Gamma_{0,a}(t)\Delta_{n,k}(t)\,dF_0(t)\right|^2
&\le
\left[\int|\Gamma_{0,a}(t)|\,dF_0(t)\right]
\left[\int|\Gamma_{0,a}(t)|\Delta_{n,k}(t)^2\,dF_0(t)\right]\\
&\le C\int|\Gamma_{0,a}(t)|
\Delta_{n,k}(t)^2\,dF_0(t).
\end{aligned}
\]
Applying another Cauchy--Schwarz inequality to the integration over $a$, and using
the bound on $\sigma_0^2$ in
condition~\ref{assm:parametric-null-support}, gives
\[
\begin{aligned}
&\int\sigma_0^2(s)
\Bigg|\int_{\mathcal A_0}\Gamma_{0,a}(s)
\left[\int\Gamma_{0,a}(t)\Delta_{n,k}(t)\,dF_0(t)\right]
w(a)\,da\Bigg|^2dF_0(s)\\
&\quad\le C\int\int_{\mathcal A_0}|\Gamma_{0,a}(s)|
\left|\int\Gamma_{0,a}(t)\Delta_{n,k}(t)\,dF_0(t)\right|^2
w(a)\,da\,dF_0(s)\\
&\quad\le C\int_{\mathcal A_0}
\left|\int\Gamma_{0,a}(t)\Delta_{n,k}(t)\,dF_0(t)\right|^2
w(a)\,da\\
&\quad\le C\int_{\mathcal A_{\delta_3}}\Delta_{n,k}(t)^2
\left[\int_{\mathcal A_0}|\Gamma_{0,a}(t)|w(a)\,da\right]
dF_0(t)\\
&\quad\le C\int_{\mathcal A_{\delta_3}}
\Delta_{n,k}(t)^2\,dF_0(t).
\end{aligned}
\]
Consequently, we obtain
\[
\begin{aligned}
&E_0\Bigg[
\Bigg\{n_k\sqrt{h_n}\int_{\mathcal A_0}Y_{n,k}(a)
\left[\int\Gamma_{0,a}(t)\Delta_{n,k}(t)\,dF_0(t)\right]
w(a)\,da\Bigg\}^{\!2}
\,\Biggm|\,\trainingsigma_k\Bigg]\\
&\quad\le Cn_kh_n
\int_{\mathcal A_{\delta_3}}\Delta_{n,k}(t)^2\,dF_0(t)
=o_p(nh_n\varrho_n^2)
=o_p(\sqrt{h_n}),
\end{aligned}
\]
where the last two equalities follow from
condition~\ref{assm:parametric-null-product},
$n_k\le n$, and $nh_n\varrho_n^2=\sqrt{h_n}$.

To control the difference obtained by replacing $\Gamma_{0,a}$
with $\Gamma_{n,k,a}$, Lemma~\ref{lem:Gamma-difference-expansion}
and the uniform empirical moment bound in the proof of
Lemma~\ref{lem:R6-rate} give
\[
\max_k\left\{
\max_{0\le j\le1}\sup_{a\in\mathcal A_0}|\alpha_{n,k,h,j}(a)|
+\max_{0\le j\le2}\sup_{a\in\mathcal A_0}|\alpha_{n,k,b,j}(a)|
\right\}
=O_p\left(\sqrt{\frac{\log n}{nh_n}}\right).
\]
Here we used $n_k\asymp n$ and $b_n\asymp h_n$.

For $\ell\in\{h_n,b_n\}$ and $j=0,1,2$, boundedness and compact
support of $K$, together with boundedness of $f_A$ and $w$, imply
\[
\sup_{a\in\mathcal A_0}
\int|p_{\ell,a,j}(t)|\,dF_0(t)\le C,
\qquad
\sup_t\int_{\mathcal A_0}|p_{\ell,a,j}(t)|w(a)\,da\le C.
\]
Recall that these kernels vanish outside
$\mathcal A_{\delta_3}$ when $a\in\mathcal A_0$, for all sufficiently large $n$.
The Cauchy--Schwarz inequality and integration over $a$ therefore yields
\[
\begin{aligned}
&\left\|\int p_{\ell,\cdot,j}(t)
\Delta_{n,k}(t)\,dF_0(t)\right\|_w^2\\
&\quad\le
\int_{\mathcal A_0}
\left[\int|p_{\ell,a,j}(t)|\,dF_0(t)\right]
\left[\int|p_{\ell,a,j}(t)|
\Delta_{n,k}(t)^2\,dF_0(t)\right]w(a)\,da\\
&\quad\le C\int_{\mathcal A_{\delta_3}}\Delta_{n,k}(t)^2
\left[\int_{\mathcal A_0}|p_{\ell,a,j}(t)|w(a)\,da\right]
dF_0(t)\\
&\quad\le C\int_{\mathcal A_{\delta_3}}
\Delta_{n,k}(t)^2\,dF_0(t).
\end{aligned}
\]
Using the coefficient expansion and the triangle inequality,
we obtain
\[
\begin{aligned}
&\left\|\int
\{\Gamma_{n,k,\cdot}(t)-\Gamma_{0,\cdot}(t)\}
\Delta_{n,k}(t)\,dF_0(t)\right\|_w\\
&\quad\le
\sum_{j=0}^1
\sup_{a\in\mathcal A_0}|\alpha_{n,k,h,j}(a)|
\left\|\int p_{h_n,\cdot,j}(t)
\Delta_{n,k}(t)\,dF_0(t)\right\|_w\\
&\qquad+\tau_n^2\sum_{j=0}^2
\sup_{a\in\mathcal A_0}|\alpha_{n,k,b,j}(a)|
\left\|\int p_{b_n,\cdot,j}(t)
\Delta_{n,k}(t)\,dF_0(t)\right\|_w\\
&\quad=
O_p\left(\sqrt{\frac{\log n}{nh_n}}\right)
\left\{\int_{\mathcal A_{\delta_3}}
\Delta_{n,k}(t)^2\,dF_0(t)\right\}^{1/2},
\end{aligned}
\]
uniformly over $k$, since $\tau_n^2=O(1)$.
Consequently, by the Cauchy--Schwarz inequality, we obtain
\[
\begin{aligned}
&n_k\sqrt{h_n}\left|
\int_{\mathcal A_0}Y_{n,k}(a)
\left[\int\{\Gamma_{n,k,a}(t)-\Gamma_{0,a}(t)\}
\Delta_{n,k}(t)\,dF_0(t)\right]w(a)\,da
\right|\\
&\quad\le n_k\sqrt{h_n}\,\|Y_{n,k}\|_w
\left\|\int\{\Gamma_{n,k,\cdot}(t)-\Gamma_{0,\cdot}(t)\}
\Delta_{n,k}(t)\,dF_0(t)\right\|_w\\
&\quad=
O_p\left(\sqrt{\frac{\log n}{h_n}}\right)
\left\{\int_{\mathcal A_{\delta_3}}
\Delta_{n,k}(t)^2\,dF_0(t)\right\}^{1/2}\\
&\quad=
o_p\left(\sqrt{\frac{\log n}{nh_n^{3/2}}}\right)
=o_p(1).
\end{aligned}
\]
The first rate uses
$\|Y_{n,k}\|_w=O_p\{(nh_n)^{-1/2}\}$;
the second uses condition~\ref{assm:parametric-null-product}.
For the last equality,
condition~\ref{assm:parametric-null-bandwidth} gives
\[
\sqrt{\frac{\log n}{nh_n^{3/2}}}
=n^{-1/8}\sqrt{\log n}\,(nh_n^2)^{-3/8}
\longrightarrow0.
\]
Thus
\[
\max_k n_k\sqrt{h_n}
|\langle Y_{n,k},R_{n,k,4}^{\mathcal M}\rangle_w|=o_p(1).
\]
The other two parts of the leading process satisfy
\begin{align*}
\max_k n_k\sqrt{h_n}\,
\|D_{n,k}\|_w\|R_{n,k,4}^{\mathcal M}\|_w
&=n\sqrt{h_n}\,o_p(\varrho_n^2)=o_p(1),\\
\max_k n_k\sqrt{h_n}\,
\|U_{n,k}\|_w\|R_{n,k,4}^{\mathcal M}\|_w
&=o_p(\sqrt{nh_n}\,\varrho_n)
=o_p(h_n^{1/4}).
\end{align*}
Consequently, the leading cross term for $j=4$ is $o_p(1)$.

For $j=2,3,5,6$, Cauchy--Schwarz and
Lemma~\ref{lem:parametric-null-remainder-rates} give
\[
\max_k n_k\sqrt{h_n}
|\langle L_{n,k,\mathcal M},R_{n,k,j}^{\mathcal M}\rangle_w|
=
\begin{cases}
o_p(1),&j=2,6,\\
O_p\{(nh_n^2)^{-1/2}\},&j=3,\\
O_p\{(nh_n)^{-1/2}\},&j=5.
\end{cases}
\]
The squared terms satisfy
\[
\max_k n_k\sqrt{h_n}\|R_{n,k,j}^{\mathcal M}\|_w^2
=
\begin{cases}
o_p(\sqrt{h_n}),&j=2,6,\\
O_p\{(nh_n^{3/2})^{-1}\},&j=3,\\
o_p(1),&j=4,\\
O_p\{(n\sqrt{h_n})^{-1}\},&j=5.
\end{cases}
\]
All these bounds converge to zero.  Each remainder--to--remainder product
is bounded by the geometric mean of its two squared terms:
\[
\begin{aligned}
n_k\sqrt{h_n}
|\langle R_{n,k,j}^{\mathcal M},R_{n,k,\ell}^{\mathcal M}\rangle_w|\leq
\{n_k\sqrt{h_n}\|R_{n,k,j}^{\mathcal M}\|_w^2\}^{1/2}
\{n_k\sqrt{h_n}\|R_{n,k,\ell}^{\mathcal M}\|_w^2\}^{1/2}
=o_p(1).
\end{aligned}
\]
Since $K$ is fixed and $\sum_kn_k/n=1$, summing the bounds proves the
proposition.
\end{proof}

\begin{proof}[Proof of Theorem~\ref{thm:parametric-null}]
Lemma~\ref{lem:parametric-null-oracle} gives
\[
\frac{T_{n,\mathcal M,0}^{DB}-b_{0h,n}^{DB}}
{\{V_\tau^{DB}\sum_{k=1}^K(n_k/n)^2\}^{1/2}}
\rightsquigarrow N(0,1).
\]
Proposition~\ref{prop:parametric-null-statistic-decomposition} gives
$\widehat T_{n,\mathcal M}^{DB}-T_{n,\mathcal M,0}^{DB}=o_p(1)$.
Since $V_\tau^{DB}>0$ and
$\sum_k(n_k/n)^2\ge1/K$, Slutsky's theorem proves the result.
\end{proof}

\begin{proof}[Proof of Corollary~\ref{cor:parametric-null-polynomials}]
For the constant family, the foldwise first-order condition in
\eqref{eq:parametric-null-beta}, on the event in condition
\ref{assm:parametric-null-family}, gives
$\hat\beta_{n,k}=\hat c_{n,k}$.
Proposition~\ref{prop:parametric-null-reproduction} below gives $\hat m_{n,k,h_n,b_n}(a;\hat\beta_{n,k})=\hat c_{n,k}$.
For $d\in\{1,2\}$, the foldwise least-squares equations are
\[
P_{n,k}\{v_d(A)v_d(A)^\top\}\hat\beta_{n,k}
=P_{n,k}\{v_d(A)\hat\xi_n^{(-k)}(O)\}.
\]
All these matrices are nonsingular with probability tending to one:
$K$ is fixed, $n_k\asymp n$, and
$P_0\{v_d(A)v_d(A)^\top\}$ is positive definite by
condition~\ref{assm:parametric-null-support}.
Proposition~\ref{prop:parametric-null-reproduction} therefore gives
\[
\hat m_{n,k,h_n,b_n}(a;\hat\beta_{n,k})
=v_d(a)^\top\hat\beta_{n,k}.
\]
Each weighted statistic in the table equals
\eqref{eq:parametric-null-statistic} for its stated family with
probability tending to one. Now, Theorem~\ref{thm:parametric-null} directly applies.
\end{proof}

The bounds in
Proposition~\ref{prop:parametric-null-statistic-decomposition} are uniform
over the fixed number of folds:
\begin{center}
\small
\setlength{\tabcolsep}{4pt}
\begin{tabular}{@{}cccc@{}}
\toprule
$j$ & $\|R_{n,k,j}^{\mathcal M}\|_w$
& $n_k\sqrt{h_n}|\langle L_{n,k,\mathcal M},R_{n,k,j}^{\mathcal M}\rangle_w|$
& $n_k\sqrt{h_n}\|R_{n,k,j}^{\mathcal M}\|_w^2$\\
\midrule
$2$&$o_p(n^{-1/2})$&$o_p(1)$&$o_p(\sqrt{h_n})$\\
$3$&$O_p\{(nh_n)^{-1}\}$&$O_p\{(nh_n^2)^{-1/2}\}$
&$O_p\{(nh_n^{3/2})^{-1}\}$\\
$4$&$o_p(\varrho_n)$&$o_p(1)$&$o_p(1)$\\
$5$&$O_p(n^{-1}h_n^{-1/2})$&$O_p\{(nh_n)^{-1/2}\}$
&$O_p\{(n\sqrt{h_n})^{-1}\}$\\
$6$&$o_p(n^{-1/2})$&$o_p(1)$&$o_p(\sqrt{h_n})$\\
\bottomrule
\end{tabular}
\end{center}

\begin{lemma}
\label{lem:parametric-null-beta-rate}
Under $H_{0,\mathcal M}$ and Assumptions~\ref{assm:A}, \ref{assm:B},
\ref{assm:G}, and~\ref{assm:parametric-null},
\[
\max_{1\le k\le K}\|\hat\beta_{n,k}-\beta_0\|
=O_p(n^{-1/2})+o_p\{(n\sqrt{h_n})^{-1/2}\}
=o_p\{(n\sqrt{h_n})^{-1/2}\}.
\]
\end{lemma}

\begin{proof}
Under $H_{0,\mathcal M}$, $\theta_0(a)=m(a;\beta_0)$ for
$a\in\mathcal A$.  Condition~\ref{assm:parametric-null-family}
therefore gives $\sup_{a\in\mathcal A}|\theta_0(a)|<\infty$ and,
with probability tending to one,
\[
P_{n,k}\left[
\left.\frac{\partial m(A;\beta)}{\partial\beta}
\right|_{\beta=\hat\beta_{n,k}}
\{\hat\xi_n^{(-k)}-m(A;\hat\beta_{n,k})\}\right]=0,
\qquad k=1,\ldots,K.
\]
For a bounded deterministic vector function $v$ on $\mathcal A$,
the definition of $\hat\xi_n^{(-k)}$ gives
\[
\begin{aligned}
&P_{n,k}[v(A)\{\hat\xi_n^{(-k)}-\theta_0(A)\}]\\
&\quad=P_{n,k}\left[
v(A)\{\hat q_n^{(-k)}(Y-\hat h_n^{(-k)})-\theta_0(A)\}
\right]
+\frac1{n_k^2}\sum_{i,j\in I_k}
v(A_i)\hat h_n^{(-k)}(W_j,A_i,X_j).
\end{aligned}
\]
Conditional independence of the validation observations implies
\[
\begin{aligned}
&E_0\left[
\left\|(P_{n,k}-P_0)
\left[v(A)\{\hat q_n^{(-k)}(Y-\hat h_n^{(-k)})
-\theta_0(A)\}\right]\right\|^2
\,\middle|\,\trainingsigma_k\right]\\
&\quad\le\frac1{n_k}
E_0\left[
\left\|v(A)\{\hat q_n^{(-k)}(Y-\hat h_n^{(-k)})
-\theta_0(A)\}\right\|^2
\,\middle|\,\trainingsigma_k\right]
=O_p(n_k^{-1}),
\end{aligned}
\]
by condition~\ref{assm:parametric-null-score-moment}.

For the double sum, for $\{i,j\}\cap\{r,s\}=\varnothing,$ conditional independence yields 
\[
\begin{aligned}
&\operatorname{Cov}_0\!\left(
v(A_i)\hat h_n^{(-k)}(W_j,A_i,X_j),
v(A_r)\hat h_n^{(-k)}(W_s,A_r,X_s)
\,\middle|\,\trainingsigma_k\right)=0.
\end{aligned}
\]
Moreover, one has
\[
\#\{(i,j,r,s)\in I_k^4:
\{i,j\}\cap\{r,s\}\ne\varnothing\}
\le4n_k^3.
\]
Centering each summand by its own conditional expectation and applying the Cauchy--Schwarz inequality to each of the remaining covariances consequently yields
\[
\begin{aligned}
&E_0\Bigg[
\Bigg\|\frac1{n_k^2}\sum_{i,j\in I_k}
\Big[
v(A_i)\hat h_n^{(-k)}(W_j,A_i,X_j)
-E_0\{v(A_i)\hat h_n^{(-k)}(W_j,A_i,X_j)
\mid\trainingsigma_k\}
\Big]\Bigg\|^2
\,\Biggm|\,\trainingsigma_k\Bigg]\\
&\quad\le\frac4{n_k}
E_0\left[
\|v(A)\hat h_n^{(-k)}(W',A,X')\|^2
+\|v(A)\hat h_n^{(-k)}(W,A,X)\|^2
\,\middle|\,\trainingsigma_k\right]
=O_p(n_k^{-1}),
\end{aligned}
\]
where $(W',X')$ has the marginal law of $(W,X)$ and is independent
of $O$ and $\trainingsigma_k$.
The diagonal and off-diagonal expectations, together with
\eqref{eq:conditional-proximal-product-bias}, give the exact identity
\[
\begin{aligned}
&E_0\left[
P_{n,k}[v(A)\{\hat\xi_n^{(-k)}-\theta_0(A)\}]
\,\middle|\,\trainingsigma_k\right]\\
&\quad=\int_{\mathcal A}v(a)\Delta_{n,k}(a)\,dF_0(a)+\frac1{n_k}E_0\left[
v(A)\{\hat h_n^{(-k)}(W,A,X)
-\hat h_n^{(-k)}(W',A,X')\}
\,\middle|\,\trainingsigma_k\right].
\end{aligned}
\]
The squared norm of the last conditional expectation is $O_p(1)$
by Jensen's inequality and condition~\ref{assm:parametric-null-score-moment}.
Combining the preceding bounds with
$\|x+y\|^2\le2\|x\|^2+2\|y\|^2$ yields
\[
E_0\left[
\left\|
P_{n,k}[v(A)\{\hat\xi_n^{(-k)}-\theta_0(A)\}]
-\int_{\mathcal A}v(a)\Delta_{n,k}(a)\,dF_0(a)
\right\|^2
\,\middle|\,\trainingsigma_k\right]
=O_p(n_k^{-1}).
\]
For $M>0$, Markov's inequality implies
\[
\begin{aligned}
&P_0\left(
\sqrt{n_k}\left\|
P_{n,k}[v(A)\{\hat\xi_n^{(-k)}-\theta_0(A)\}]
-\int_{\mathcal A}v(a)\Delta_{n,k}(a)\,dF_0(a)
\right\|>M\right)\\
&\quad\le
P_0\Bigg(
n_kE_0\left[
\left\|
\begin{aligned}
&P_{n,k}[v(A)\{\hat\xi_n^{(-k)}-\theta_0(A)\}]\\
&\quad-\int_{\mathcal A}v(a)\Delta_{n,k}(a)\,dF_0(a)
\end{aligned}
\right\|^2
\,\middle|\,\trainingsigma_k\right]>M
\Bigg)+M^{-1}.
\end{aligned}
\]
Taking $\limsup_{n\to\infty}$ and then $M\to\infty$, and using
$K<\infty$ and $n_k\asymp n$, proves
\[
\max_k\left\|
P_{n,k}[v(A)\{\hat\xi_n^{(-k)}-\theta_0(A)\}]
-\int_{\mathcal A}v(a)\Delta_{n,k}(a)\,dF_0(a)
\right\|=O_p(n^{-1/2}).
\]
The same conclusion applies entrywise to a bounded deterministic
matrix function $v$, since its dimension is fixed.
Taking $v=\dot m_0$ and applying
condition~\ref{assm:parametric-null-product} gives
\[
\max_k\left\|
P_{n,k}[\dot m_0(A)\{\hat\xi_n^{(-k)}-m(A;\beta_0)\}]
\right\|
=O_p(n^{-1/2})+o_p(\varrho_n).
\]
Taking instead
$v(a)=\partial^2m(a;\beta_0)/(\partial\beta\partial\beta^\top)$
gives
\[
\begin{aligned}
\max_k\left\|
P_{n,k}\left[
\frac{\partial^2m(A;\beta_0)}{\partial\beta\partial\beta^\top}
\{\hat\xi_n^{(-k)}-m(A;\beta_0)\}
\right]\right\|\leq O_p(n^{-1/2})
+C\max_k\int_{\mathcal A}|\Delta_{n,k}(a)|\,dF_0(a),
\end{aligned}
\]
which is $o_p(1)$.
Since $\dot m_0$ is bounded, we get
\[
\max_k\|P_{n,k}(\dot m_0\dot m_0^\top)-J_0\|
=O_p(n^{-1/2}).
\]
Thus, we obtain
\[
\max_k\left\|
P_{n,k}\left[
\dot m_0(A)\dot m_0(A)^\top
-\frac{\partial^2m(A;\beta_0)}{\partial\beta\partial\beta^\top}
\{\hat\xi_n^{(-k)}-m(A;\beta_0)\}
\right]-J_0\right\|=o_p(1).
\]

We next control the change in this matrix near $\beta_0$.
Jensen's inequality applied to the empirical marginalization in
$\hat\xi_n^{(-k)}$ yields
\[
\begin{aligned}
P_{n,k}|\hat\xi_n^{(-k)}|^2
\leq 2P_{n,k}|\hat q_n^{(-k)}(Y-\hat h_n^{(-k)})|^2+\frac2{n_k^2}\sum_{i,j\in I_k}
|\hat h_n^{(-k)}(W_j,A_i,X_j)|^2.
\end{aligned}
\]
Consequently, we obtain
\[
\begin{aligned}
E_0[P_{n,k}|\hat\xi_n^{(-k)}|^2\mid\trainingsigma_k]
\le{}&2E_0[|\hat q_n^{(-k)}(Y-\hat h_n^{(-k)})|^2
\mid\trainingsigma_k]\\
&+2\left(1-\frac1{n_k}\right)
E_0[|\hat h_n^{(-k)}(W',A,X')|^2\mid\trainingsigma_k]\\
&+\frac2{n_k}
E_0[|\hat h_n^{(-k)}(W,A,X)|^2\mid\trainingsigma_k]
=O_p(1).
\end{aligned}
\]
In particular, for $M>0$, we have
\[
\begin{aligned}
P_0\left(\max_kP_{n,k}|\hat\xi_n^{(-k)}|^2>M^2\right)
&\le\sum_{k=1}^K
P_0\left(
E_0[P_{n,k}|\hat\xi_n^{(-k)}|^2\mid\trainingsigma_k]>M
\right)+\frac K M,
\end{aligned}
\]
so $\max_kP_{n,k}|\hat\xi_n^{(-k)}|^2=O_p(1)$.

For sufficiently small $\delta>0$, the closed ball
$\{\beta:\|\beta-\beta_0\|\le\delta\}$ lies in the
differentiability neighborhood and in the interior of $\mathcal B$.
The boundedness and continuity of the second derivatives imply
\[
\lim_{\delta\downarrow0}
P_0\left[
\sup_{\|\beta-\beta_0\|\le\delta}
\left\|
\frac{\partial^2m(A;\beta)}{\partial\beta\partial\beta^\top}
-\frac{\partial^2m(A;\beta_0)}{\partial\beta\partial\beta^\top}
\right\|^2
\right]=0
\]
by dominated convergence.
For every fold, the Cauchy--Schwarz inequality implies
\[
\begin{aligned}
&\sup_{\|\beta-\beta_0\|\le\delta}
\left\|P_{n,k}\left[
\left\{
\frac{\partial^2m(A;\beta)}{\partial\beta\partial\beta^\top}
-\frac{\partial^2m(A;\beta_0)}{\partial\beta\partial\beta^\top}
\right\}
\{\hat\xi_n^{(-k)}-m(A;\beta_0)\}
\right]\right\|\\
&\quad\le
\left\{P_{n,k}\left[
\sup_{\|\beta-\beta_0\|\le\delta}
\left\|
\frac{\partial^2m(A;\beta)}{\partial\beta\partial\beta^\top}
-\frac{\partial^2m(A;\beta_0)}{\partial\beta\partial\beta^\top}
\right\|^2
\right]\right\}^{1/2}
\{P_{n,k}(|\hat\xi_n^{(-k)}|+C)^2\}^{1/2}.
\end{aligned}
\]
Therefore, for $\varepsilon,M>0$, we have
\[
\begin{aligned}
&P_0\Bigg(
\max_k\sup_{\|\beta-\beta_0\|\le\delta}
\Bigg\|P_{n,k}\left[
\left\{
\begin{aligned}
&\frac{\partial^2m(A;\beta)}{\partial\beta\partial\beta^\top}
\\
&\quad-\frac{\partial^2m(A;\beta_0)}{\partial\beta\partial\beta^\top}
\end{aligned}
\right\}
\{\hat\xi_n^{(-k)}-m(A;\beta_0)\}
\right]\Bigg\|>\varepsilon
\Bigg)\\
&\quad\le P_0\left(\max_kP_{n,k}(|\hat\xi_n^{(-k)}|+C)^2>M\right)\\
&\quad+\frac{KM}{\varepsilon^2}
P_0\left[
\sup_{\|\beta-\beta_0\|\le\delta}
\left\|
\begin{aligned}
&\frac{\partial^2m(A;\beta)}{\partial\beta\partial\beta^\top}\\
&\quad-\frac{\partial^2m(A;\beta_0)}{\partial\beta\partial\beta^\top}
\end{aligned}
\right\|^2
\right].
\end{aligned}
\]
Taking $\limsup_{n\to\infty}$, then
$\delta\downarrow0$, and finally $M\to\infty$ makes the last
probability bound zero.

The bounded first and second derivatives also give
\[
\sup_{a\in\mathcal A}|m(a;\beta)-m(a;\beta_0)|
+\sup_{a\in\mathcal A}
\left\|\frac{\partial m(a;\beta)}{\partial\beta}-\dot m_0(a)\right\|
\le C\delta,
\qquad \|\beta-\beta_0\|\le\delta.
\]
It follows that, for every fold,
\[
\begin{aligned}
&\sup_{\|\beta-\beta_0\|\le\delta}
\left\|P_{n,k}\left[
\frac{\partial m(A;\beta)}{\partial\beta}
\frac{\partial m(A;\beta)}{\partial\beta^\top}
-\dot m_0(A)\dot m_0(A)^\top
\right]\right\|\\
&\quad+
\sup_{\|\beta-\beta_0\|\le\delta}
\left\|P_{n,k}\left[
\frac{\partial^2m(A;\beta)}{\partial\beta\partial\beta^\top}
\{m(A;\beta)-m(A;\beta_0)\}
\right]\right\|
\le C\delta.
\end{aligned}
\]
Using the exact identity
\[
\begin{aligned}
&\frac{\partial^2m(A;\beta)}{\partial\beta\partial\beta^\top}
\{\hat\xi_n^{(-k)}-m(A;\beta)\}
-\frac{\partial^2m(A;\beta_0)}{\partial\beta\partial\beta^\top}
\{\hat\xi_n^{(-k)}-m(A;\beta_0)\}\\
&\quad=
\left\{
\frac{\partial^2m(A;\beta)}{\partial\beta\partial\beta^\top}
-\frac{\partial^2m(A;\beta_0)}{\partial\beta\partial\beta^\top}
\right\}
\{\hat\xi_n^{(-k)}-m(A;\beta_0)\}\\
&\qquad-
\frac{\partial^2m(A;\beta)}{\partial\beta\partial\beta^\top}
\{m(A;\beta)-m(A;\beta_0)\},
\end{aligned}
\]
the preceding bounds imply, for every $\varepsilon>0$,
\[
\begin{aligned}
\lim_{\delta\downarrow0}\limsup_{n\to\infty}
P_0\Bigg(&
\max_k\sup_{\|\beta-\beta_0\|\le\delta}
\Bigg\|P_{n,k}\Bigg[
\frac{\partial m(A;\beta)}{\partial\beta}
\frac{\partial m(A;\beta)}{\partial\beta^\top}\\
&\qquad\qquad\qquad\quad-
\frac{\partial^2m(A;\beta)}{\partial\beta\partial\beta^\top}
\{\hat\xi_n^{(-k)}-m(A;\beta)\}
\Bigg]-J_0\Bigg\|>\varepsilon
\Bigg)=0.
\end{aligned}
\]
Condition~\ref{assm:parametric-null-family} gives
$P_0(\max_k\|\hat\beta_{n,k}-\beta_0\|>\delta)\to0$
for every fixed $\delta>0$.  Since
\[
\sup_{0\le t\le1}
\|\beta_0+t(\hat\beta_{n,k}-\beta_0)-\beta_0\|
\le\|\hat\beta_{n,k}-\beta_0\|,
\]
the probability of a matrix deviation exceeding $\varepsilon$ on any
of these segments is at most
$P_0(\max_k\|\hat\beta_{n,k}-\beta_0\|>\delta)$ plus the probability
inside the preceding iterated limit.  Hence the matrix differs from
$J_0$ by $o_p(1)$ uniformly over the segments and folds.

On the event that all estimators are interior and these segments lie
in the differentiability neighborhood, the score equations and the
fundamental theorem of calculus give
\[
\begin{aligned}
&\Bigg(\int_0^1
\left.P_{n,k}\Bigg[
\frac{\partial m(A;\beta)}{\partial\beta}
\frac{\partial m(A;\beta)}{\partial\beta^\top}
-\frac{\partial^2m(A;\beta)}{\partial\beta\partial\beta^\top}
\{\hat\xi_n^{(-k)}-m(A;\beta)\}
\Bigg]\right|_{\beta=\beta_0+t(\hat\beta_{n,k}-\beta_0)}
dt\Bigg)\\
&\hspace{8em}\times(\hat\beta_{n,k}-\beta_0)
=P_{n,k}[\dot m_0(A)\{\hat\xi_n^{(-k)}-m(A;\beta_0)\}].
\end{aligned}
\]
With probability tending to one, each integral matrix has smallest
eigenvalue at least $\lambda_{\min}(J_0)/2$.  Hence, on that event,
\[
\begin{aligned}
\max_k\|\hat\beta_{n,k}-\beta_0\|
&\le\frac{2}{\lambda_{\min}(J_0)}
\max_k\left\|
P_{n,k}[\dot m_0(A)\{\hat\xi_n^{(-k)}-m(A;\beta_0)\}]
\right\|\\
&=O_p(n^{-1/2})+o_p(\varrho_n)
=o_p(\varrho_n),
\end{aligned}
\]
since $n^{-1/2}=h_n^{1/4}\varrho_n=o(\varrho_n)$.
\end{proof}

\begin{lemma}
\label{lem:parametric-null-kernel}
Under conditions \ref{assm:kernel}, \ref{assm:G-bandwidth}, and
\ref{assm:parametric-null-support}, uniformly over $a\in\mathcal A_0$,
\[
|\Gamma_{0,a}(t)|
\le Ch_n^{-1}I(|t-a|\le Ch_n),
\]
and, on the fixed union of the rescaled kernel supports,
\begin{equation}
\sup_u\left|
h_nf_A(a)\Gamma_{0,a}(a+h_nu)-\Lambda_{\tau_n}(u)
\right|\le Ch_n.
\label{eq:parametric-null-equivalent-kernel}
\end{equation}
If
\[
J_{n,h}(s,t)
:=\int_{\mathcal A_0}\Gamma_{0,a}(s)\Gamma_{0,a}(t)w(a)\,da,
\]
then
\[
|J_{n,h}(s,t)|
\le Ch_n^{-1}I(|s-t|\le Ch_n).
\]
\end{lemma}

\begin{proof}
By the definitions of the population moment matrices, symmetry of $K$,
and condition~\ref{assm:parametric-null-support}, uniformly in
$a\in\mathcal A_0$,
\[
\begin{aligned}
D_{0,h_n,a,1}
&=\int (1,u)^\top(1,u)K(u)f_A(a+h_nu)\,du
=f_A(a)\begin{pmatrix}1&0\\0&c_2\end{pmatrix}+O(h_n),\\
D_{0,b_n,a,2}
&=\int (1,u,u^2)^\top(1,u,u^2)K(u)f_A(a+b_nu)\,du\\
&=f_A(a)\begin{pmatrix}1&0&c_2\\0&c_2&0\\c_2&0&c_4\end{pmatrix}+O(b_n).
\end{aligned}
\]
Note that Condition~\ref{assm:kernel} gives $c_2>0$ and $c_4-c_2^2>0$.
The inverse identity used in Lemma~\ref{lem:Gamma-difference-expansion}
therefore gives
\[
\begin{aligned}
f_A(a)e_1^\top D_{0,h_n,a,1}^{-1}=(1,0)+O(h_n),\qquad
f_A(a)e_3^\top D_{0,b_n,a,2}^{-1}
=\frac{(-c_2,0,1)}{c_4-c_2^2}+O(b_n),
\end{aligned}
\]
and $c_{0,h_n,a,2}=c_2+O(h_n)$.
Since $b_n\asymp h_n$, substitution in $\Gamma_{0,a}$ gives, uniformly
on the union of the rescaled kernel supports,
\[
\begin{aligned}
h_nf_A(a)\Gamma_{0,a}(a+h_nu)
&=K(u)-\frac{h_n^3}{b_n^3}c_2
\frac{(h_nu/b_n)^2-c_2}{c_4-c_2^2}K(h_nu/b_n)+O(h_n)\\
&=\Lambda_{\tau_n}(u)+O(h_n).
\end{aligned}
\]
The union is contained in a fixed bounded interval by
condition~\ref{assm:G-bandwidth}.  This proves
\eqref{eq:parametric-null-equivalent-kernel} and the stated envelope.
Finally,
\[
\begin{aligned}
|J_{n,h}(s,t)|
&\le Ch_n^{-2}\int
I(|s-a|\le Ch_n)I(|t-a|\le Ch_n)\,da\\
&\le Ch_n^{-1}I(|s-t|\le Ch_n),
\end{aligned}
\]
where the constant in the last line may be enlarged.
\end{proof}

\begin{lemma}
\label{lem:parametric-null-oracle}
Under $H_{0,\mathcal M}$ and Assumptions~\ref{assm:A}, \ref{assm:B},
\ref{assm:D}, \ref{assm:G}, and~\ref{assm:parametric-null},
\[
\frac{T_{n,\mathcal M,0}^{DB}-b_{0h,n}^{DB}}
{\left\{V_\tau^{DB}\sum_{k=1}^K(n_k/n)^2\right\}^{1/2}}
\rightsquigarrow N(0,1).
\]
Moreover,
\[
\max_{1\le k\le K}\|L_{n,k,\mathcal M}\|_w
=O_p\{(nh_n)^{-1/2}\}.
\]
\end{lemma}

\begin{proof}
We use $Y_{n,k},D_{n,k},U_{n,k}$ as defined in the proof of
Proposition~\ref{prop:parametric-null-statistic-decomposition}, so that
$L_{n,k,\mathcal M}=Y_{n,k}-D_{n,k}+U_{n,k}$.
Condition~\ref{assm:folds} gives $K<\infty$, $n_k\ge\underline\pi n$,
and independence of the observations conditional on $I_1,\ldots,I_K$.
Lemma~\ref{lem:parametric-null-kernel} implies
\[
\begin{aligned}
&\sup_{a\in\mathcal A_0}\int|\Gamma_{0,a}(t)|\,dF_0(t)
+\sup_t\int_{\mathcal A_0}|\Gamma_{0,a}(t)|w(a)\,da\le C,\\
&\sup_s\int|J_{n,h}(s,t)|\,dF_0(t)\le C,
\qquad
\sup_s\int J_{n,h}(s,t)^2\,dF_0(t)\le Ch_n^{-1}.
\end{aligned}
\]
For all sufficiently large $n$, $\Gamma_{0,a}(t)=0$ if
$a\in\mathcal A_0$ and $t\notin\mathcal A_{\delta_3}$.
Also, the compactness of $\operatorname{supp}(w)\subset\operatorname{int}
(\mathcal A_0)$ implies
\[
|J_{n,h}(s,t)|\le
Ch_n^{-1}I(|s-t|\le Ch_n)I(s,t\in\mathcal A_0).
\]

Since $E_0\{\zeta_\infty(O)\mid A\}=0$ on
$\mathcal A_{\delta_3}$, we get
\[
E_0\{\|Y_{n,k}\|_w^2\mid I_1,\ldots,I_K\}=\frac1{n_k}\int_{\mathcal A_0}\int
\Gamma_{0,a}(t)^2\sigma_0^2(t)\,dF_0(t)w(a)\,da
\le\frac C{n_kh_n},
\]
and
\[
\begin{aligned}
&E_0\left[
\int_{\mathcal A_0}\{P_{n,k}|\Gamma_{0,a}|\}^2w(a)\,da
\,\middle|\,I_1,\ldots,I_K\right]\\
&\quad=\frac1{n_k}\int_{\mathcal A_0}P_0\Gamma_{0,a}^2w(a)\,da
+\frac{n_k-1}{n_k}\int_{\mathcal A_0}
\{P_0|\Gamma_{0,a}|\}^2w(a)\,da\\
&\quad\le C\{1+(n_kh_n)^{-1}\}.
\end{aligned}
\]
The boundedness of $\dot m_0$ also implies
\[
\int_{\mathcal A_0}
\|P_{n,k}\{\Gamma_{0,a}(A)\dot m_0(A)\}\|^2w(a)\,da
\le C\int_{\mathcal A_0}\{P_{n,k}|\Gamma_{0,a}|\}^2w(a)\,da.
\]
Therefore, Markov's inequality, $K<\infty$, and $n_k\asymp n$ provide
\[
\max_k\|Y_{n,k}\|_w=O_p\{(nh_n)^{-1/2}\},
\]
and
\[
\begin{aligned}
\max_k\left\|P_{n,k}|\Gamma_{0,\cdot}|\right\|_w
+\max_k\left\{
\int_{\mathcal A_0}
\|P_{n,k}(\Gamma_{0,a}\dot m_0)\|^2w(a)\,da
\right\}^{1/2}&=O_p(1).
\end{aligned}
\]

For the cross product involving the first derivative of $m$,
one has
\[
\left\langle Y_{n,k},
P_{n,k}\{\Gamma_{0,\cdot}(A)\dot m_0(A)\}\right\rangle_w
=\frac1{n_k^2}\sum_{i,j\in I_k}
\zeta_\infty(O_i)\dot m_0(A_j)J_{n,h}(A_i,A_j).
\]
For $i,j,r,s\in I_k$ with $i\ne r$, note that
\[
E_0\left[
\zeta_\infty(O_i)\zeta_\infty(O_r)
J_{n,h}(A_i,A_j)J_{n,h}(A_r,A_s)
\,\middle|\,(A_\ell)_{\ell\in I_k},I_1,\ldots,I_K
\right]=0.
\]
Splitting $\sum_j$ into $j=i$ and $j\ne i$ and using
$\|x+y\|^2\le2\|x\|^2+2\|y\|^2$ gives
\begin{equation}
\begin{aligned}
&E_0\left[
\left\|\left\langle Y_{n,k},
P_{n,k}\{\Gamma_{0,\cdot}(A)\dot m_0(A)\}\right\rangle_w\right\|^2
\,\middle|\,I_1,\ldots,I_K\right]\\
&\quad=\frac1{n_k^4}\sum_{i\in I_k}
E_0\left[
\sigma_0^2(A_i)
\left\|\sum_{j\in I_k}\dot m_0(A_j)J_{n,h}(A_i,A_j)\right\|^2
\,\middle|\,I_1,\ldots,I_K\right]\\
&\quad\le\frac2{n_k^3}
\int\sigma_0^2(s)\|\dot m_0(s)\|^2J_{n,h}(s,s)^2\,dF_0(s)\\
&\qquad+\frac{2(n_k-1)}{n_k^3}
\iint\sigma_0^2(s)\|\dot m_0(t)\|^2J_{n,h}(s,t)^2
\,dF_0(t)\,dF_0(s)\\
&\qquad+\frac{2(n_k-1)(n_k-2)}{n_k^3}
\int\sigma_0^2(s)
\left\|\int\dot m_0(t)J_{n,h}(s,t)\,dF_0(t)\right\|^2dF_0(s)\\
&\quad\le C\{n_k^{-3}h_n^{-2}+n_k^{-2}h_n^{-1}+n_k^{-1}\}
\le Cn_k^{-1}.
\end{aligned}
\label{eq:parametric-null-oracle-score-moment}
\end{equation}
Since $n_kh_n\to\infty$, the last bound in
\eqref{eq:parametric-null-oracle-score-moment} holds for all sufficiently
large $n$.  
Markov's inequality applied to
\eqref{eq:parametric-null-oracle-score-moment} yields
\[
\max_k\left\|
\left\langle Y_{n,k},
P_{n,k}\{\Gamma_{0,\cdot}(A)\dot m_0(A)\}\right\rangle_w
\right\|=O_p(n^{-1/2}).
\]

The bounded second derivatives in
condition~\ref{assm:parametric-null-family} imply
\[
\sup_{a\in\mathcal A}
\left|m(a;\hat\beta_{n,k})-m(a;\beta_0)
-\dot m_0(a)^\top(\hat\beta_{n,k}-\beta_0)\right|
\le C\|\hat\beta_{n,k}-\beta_0\|^2
\]
on an event whose probability tends to one.  Consequently,
\[
\begin{aligned}
\left\|D_{n,k}
-\left[P_{n,k}\{\Gamma_{0,\cdot}(A)\dot m_0(A)\}\right]^\top
(\hat\beta_{n,k}-\beta_0)\right\|_w \leq C\|\hat\beta_{n,k}-\beta_0\|^2
\left\|P_{n,k}|\Gamma_{0,\cdot}|\right\|_w.
\end{aligned}
\]
Lemma~\ref{lem:parametric-null-beta-rate} now gives
\[
\max_k\|D_{n,k}\|_w=o_p(\varrho_n),
\qquad
\max_k n_k\sqrt{h_n}\|D_{n,k}\|_w^2
=o_p(n\sqrt{h_n}\varrho_n^2)=o_p(1).
\]
By the Cauchy--Schwarz inequality, we get
\[
\begin{aligned}
&n_k\sqrt{h_n}|\langle Y_{n,k},D_{n,k}\rangle_w|\\
&\quad\le n_k\sqrt{h_n}\|\hat\beta_{n,k}-\beta_0\|
\left\|\left\langle Y_{n,k},
P_{n,k}\{\Gamma_{0,\cdot}(A)\dot m_0(A)\}\right\rangle_w\right\|\\
&\qquad+Cn_k\sqrt{h_n}\|\hat\beta_{n,k}-\beta_0\|^2
\|Y_{n,k}\|_w\left\|P_{n,k}|\Gamma_{0,\cdot}|\right\|_w\\
&\quad=o_p(\sqrt{nh_n}\varrho_n)
+o_p(\sqrt n\varrho_n^2)
=o_p(h_n^{1/4})+o_p\{(nh_n)^{-1/2}\},
\end{aligned}
\]
which is $o_p(1)$, uniformly over $k$.

The definition $m_\infty(t)=E_0h_\infty(W,t,X)$ implies
\[
E_0\left[
\int\Gamma_{0,a}(t)
\{h_\infty(W,t,X)-m_\infty(t)\}\,dF_0(t)
\right]=0.
\]
The bound on $h_\infty$ in condition~\ref{assm:G-DR-rates} then implies
\[
\begin{aligned}
&E_0\{\|U_{n,k}\|_w^2\mid I_1,\ldots,I_K\}\\
&\quad=\frac1{n_k}\int_{\mathcal A_0}
E_0\left[
\left\{\int\Gamma_{0,a}(t)
\{h_\infty(W,t,X)-m_\infty(t)\}\,dF_0(t)\right\}^2
\right]w(a)\,da\\
&\quad\le\frac C{n_k}.
\end{aligned}
\]
Moreover, we have
\[
\langle Y_{n,k},U_{n,k}\rangle_w
=\frac1{n_k^2}\sum_{i,j\in I_k}\zeta_\infty(O_i)
\int J_{n,h}(A_i,t)
\{h_\infty(W_j,t,X_j)-m_\infty(t)\}\,dF_0(t).
\]
For independent observations $O_1,O_2$, the two conditional means are
\begin{equation}
\begin{aligned}
&E_0\left[
\zeta_\infty(O_1)\int J_{n,h}(A_1,t)
\{h_\infty(W_2,t,X_2)-m_\infty(t)\}\,dF_0(t)
\,\middle|\,O_1\right]\\
&\quad=\zeta_\infty(O_1)\int J_{n,h}(A_1,t)
\{E_0h_\infty(W_2,t,X_2)-m_\infty(t)\}\,dF_0(t)=0,\\
&E_0\left[
\zeta_\infty(O_1)\int J_{n,h}(A_1,t)
\{h_\infty(W_2,t,X_2)-m_\infty(t)\}\,dF_0(t)
\,\middle|\,O_2\right]\\
&\quad=\int E_0\{\zeta_\infty(O_1)J_{n,h}(A_1,t)\}
\{h_\infty(W_2,t,X_2)-m_\infty(t)\}\,dF_0(t)=0.
\end{aligned}
\label{eq:parametric-null-oracle-U-means}
\end{equation}
For $j=1,2$, we get
\[
\begin{aligned}
&E_0\left[
\left\{\zeta_\infty(O_1)\int J_{n,h}(A_1,t)
\{h_\infty(W_j,t,X_j)-m_\infty(t)\}\,dF_0(t)\right\}^2
\right]\\
&\quad\le C E_0\left[
\zeta_\infty(O_1)^2I(A_1\in\mathcal A_0)
\left\{\int|J_{n,h}(A_1,t)|\,dF_0(t)\right\}^2
\right]\le C.
\end{aligned}
\]
Equation~\eqref{eq:parametric-null-oracle-U-means} implies, for
$i\ne j$ and $r\ne s$,
\[
\begin{aligned}
&\Bigg|E_0\Bigg[
\zeta_\infty(O_i)\zeta_\infty(O_r)
\iint J_{n,h}(A_i,t)J_{n,h}(A_r,u)\\
&\qquad\times
\{h_\infty(W_j,t,X_j)-m_\infty(t)\}
\{h_\infty(W_s,u,X_s)-m_\infty(u)\}
\,dF_0(t)\,dF_0(u)\Bigg]\Bigg|\\
&\quad\le C\{I(i=r,j=s)+I(i=s,j=r)\},
\end{aligned}
\]
where the Cauchy--Schwarz inequality bounds the two nonzero cases.
Summing over the indices gives
\[
\begin{aligned}
&E_0\Bigg[
\Bigg\{\sum_{\substack{i,j\in I_k\\i\ne j}}
\zeta_\infty(O_i)\int J_{n,h}(A_i,t)
\{h_\infty(W_j,t,X_j)-m_\infty(t)\}\,dF_0(t)
\Bigg\}^{\!2}\Biggm|\,I_1,\ldots,I_K\Bigg]
\\&\qquad\le2Cn_k(n_k-1).
\end{aligned}
\]
For the diagonal sum, another Cauchy--Schwarz inequality implies
\[
\begin{aligned}
&E_0\Bigg[
\Bigg\{\sum_{i\in I_k}
\zeta_\infty(O_i)\int J_{n,h}(A_i,t)
\{h_\infty(W_i,t,X_i)-m_\infty(t)\}\,dF_0(t)
\Bigg\}^{\!2}\Biggm|\,I_1,\ldots,I_K\Bigg]\\
&\quad\le n_k\sum_{i\in I_k}
E_0\left[
\left\{\zeta_\infty(O_i)\int J_{n,h}(A_i,t)
\{h_\infty(W_i,t,X_i)-m_\infty(t)\}\,dF_0(t)\right\}^2
\middle|\,I_1,\ldots,I_K\right]\\
&\quad\le Cn_k^2.
\end{aligned}
\]
Using $(x+y)^2\le2x^2+2y^2$ for the diagonal and off-diagonal sums, we have
\[
E_0\{|\langle Y_{n,k},U_{n,k}\rangle_w|^2
\mid I_1,\ldots,I_K\}
\le\frac{C\{n_k(n_k-1)+n_k^2\}}{n_k^4}
\le\frac C{n_k^2}.
\]
Markov's inequality and the Cauchy--Schwarz inequality
therefore imply, uniformly over $k$,
\[
\begin{aligned}
&\|U_{n,k}\|_w=O_p(n^{-1/2}),\\
&n_k\sqrt{h_n}\|U_{n,k}\|_w^2
+n_k\sqrt{h_n}|\langle Y_{n,k},U_{n,k}\rangle_w|
=O_p(\sqrt{h_n})=o_p(1),\\
&n_k\sqrt{h_n}|\langle D_{n,k},U_{n,k}\rangle_w|\le
\{n_k\sqrt{h_n}\|D_{n,k}\|_w^2\}^{1/2}
\{n_k\sqrt{h_n}\|U_{n,k}\|_w^2\}^{1/2}
=o_p(1).
\end{aligned}
\]
Since
\[
\begin{aligned}
\|L_{n,k,\mathcal M}\|_w^2-\|Y_{n,k}\|_w^2
={}&\|D_{n,k}\|_w^2+\|U_{n,k}\|_w^2
-2\langle Y_{n,k},D_{n,k}\rangle_w\\
&+2\langle Y_{n,k},U_{n,k}\rangle_w
-2\langle D_{n,k},U_{n,k}\rangle_w,
\end{aligned}
\]
$K<\infty$ and $\sum_k n_k/n=1$ give
\begin{equation}
T_{n,\mathcal M,0}^{DB}
=\sum_{k=1}^K\frac{n_k}{n}\,
n_k\sqrt{h_n}\|Y_{n,k}\|_w^2+o_p(1).
\label{eq:parametric-null-oracle-reduction}
\end{equation}

Define $H_{n,ij}:=\zeta_\infty(O_i)\zeta_\infty(O_j)
J_{n,h}(A_i,A_j)$.  Then
\begin{equation}
n_k\sqrt{h_n}\|Y_{n,k}\|_w^2
=\frac{\sqrt{h_n}}{n_k}\sum_{i\in I_k}H_{n,ii}
+\frac{\sqrt{h_n}}{n_k}
\sum_{\substack{i,j\in I_k\\i\ne j}}H_{n,ij}.
\label{eq:parametric-null-diagonal-offdiagonal}
\end{equation}
Equation~\eqref{eq:parametric-null-diagonal-offdiagonal} uses the
diagonal and off-diagonal decomposition in the proofs of Proposition~1
of \citet{hardle1993comparing} and Theorem~2.1 of
\citet{alcala1999goodness}, with the kernel $J_{n,h}$ defined from
the debiased smoothing weights.
Condition~\ref{assm:parametric-null-support} and
Lemma~\ref{lem:parametric-null-kernel} give
\[
E_0H_{n,11}^2
=\int_{\mathcal A_0}
E_0\{\zeta_\infty(O)^4\mid A=s\}
J_{n,h}(s,s)^2\,dF_0(s)
\le Ch_n^{-2}.
\]
For $r=2,4$, we get
\[
\begin{aligned}
E_0|H_{n,12}|^r
&=\int_{\mathcal A_0}\int_{\mathcal A_0}
E_0\{|\zeta_\infty(O)|^r\mid A=s\}
E_0\{|\zeta_\infty(O)|^r\mid A=t\}
\\
&\qquad\times |J_{n,h}(s,t)|^r\,dF_0(t)\,dF_0(s)\\
&\le Ch_n^{-r}\int_{\mathcal A_0}\int
I(|s-t|\le Ch_n)\,dt\,ds
\le Ch_n^{1-r}.
\end{aligned}
\]

Equation~\eqref{eq:parametric-null-equivalent-kernel} implies,
uniformly on the common bounded support in $u$,
\[
\left|h_n^2f_A(a)^2\Gamma_{0,a}(a+h_nu)^2
-\Lambda_{\tau_n}(u)^2\right|\le Ch_n.
\]
Condition~\ref{assm:parametric-null-support} gives
\[
\sup_{\substack{a\in\mathcal A_0\\|u|\le C}}
|\sigma_0^2(a+h_nu)f_A(a+h_nu)-\sigma_0^2(a)f_A(a)|
\le Ch_n.
\]
Using $t=a+h_nu$ in the integral for $E_0H_{n,11}$ gives
\[
\begin{aligned}
\sqrt{h_n}E_0H_{n,11}
&=\sqrt{h_n}\int_{\mathcal A_0}\int
\Gamma_{0,a}(t)^2\sigma_0^2(t)f_A(t)\,dt\,w(a)\,da\\
&=h_n^{-1/2}\int_{\mathcal A_0}\frac{w(a)}{f_A(a)^2}
\int\Lambda_{\tau_n}(u)^2
\sigma_0^2(a+h_nu)f_A(a+h_nu)\,du\,da\\
&\qquad+O(\sqrt{h_n})\\
&=h_n^{-1/2}\Lambda_{\tau_n}^{(2)}(0)
\int_{\mathcal A_0}\frac{\sigma_0^2(a)w(a)}{f_A(a)}\,da
+O(\sqrt{h_n})\\
&=b_{0h,n}^{DB}+O(\sqrt{h_n}).
\end{aligned}
\]

The conditional squared error of the diagonal sum is therefore
\begin{equation}
\begin{aligned}
E_0\left[
\left\{\frac{\sqrt{h_n}}{n_k}\sum_{i\in I_k}H_{n,ii}
-b_{0h,n}^{DB}\right\}^2
\,\middle|\,I_1,\ldots,I_K\right]
&=\frac{h_n}{n_k}\operatorname{Var}_0(H_{n,11})
\\
&\quad+\{\sqrt{h_n}E_0H_{n,11}-b_{0h,n}^{DB}\}^2\\
&\leq C\{(n_kh_n)^{-1}+h_n\}=o(1).
\end{aligned}
\label{eq:parametric-null-oracle-diagonal-control}
\end{equation}

We extend $w$ by zero outside $\mathcal A_0$.
Condition~\ref{assm:parametric-null-support} implies
\[
\sup_{\substack{s\in\mathcal A_0\\|u|\le C}}
\left|
\frac{w(s-h_nu)}{f_A(s-h_nu)^2}
-\frac{w(s)}{f_A(s)^2}
\right|\le Ch_n.
\]
Substituting $a=s-h_nu$ and applying
\eqref{eq:parametric-null-equivalent-kernel} gives, uniformly for
$s\in\mathcal A_0$ and bounded $v$,
\begin{equation}
\begin{aligned}
J_{n,h}(s,s+h_nv)
&=h_n\int_{\{u:s-h_nu\in\mathcal A_0\}}
\Gamma_{0,s-h_nu}(s)\Gamma_{0,s-h_nu}(s+h_nv)
w(s-h_nu)\,du\\
&=h_n^{-1}\int
\Lambda_{\tau_n}(u)\Lambda_{\tau_n}(u+v)
\frac{w(s-h_nu)}{f_A(s-h_nu)^2}\,du+O(1)\\
&=h_n^{-1}\frac{w(s)}{f_A(s)^2}
\int\Lambda_{\tau_n}(u)\Lambda_{\tau_n}(u+v)\,du+O(1).
\end{aligned}
\label{eq:parametric-null-oracle-J-expansion}
\end{equation}
For the next calculation, the $v$-integrals are over a fixed bounded
interval containing the supports of both $J_{n,h}(s,s+h_nv)$ and
$\int\Lambda_{\tau_n}(u)\Lambda_{\tau_n}(u+v)\,du$.
Equation~\eqref{eq:parametric-null-oracle-J-expansion} implies
\[
\begin{aligned}
2h_nE_0H_{n,12}^2
&=2h_n^2\int_{\mathcal A_0}\int
\sigma_0^2(s)f_A(s)\sigma_0^2(s+h_nv)f_A(s+h_nv)
J_{n,h}(s,s+h_nv)^2\,dv\,ds\\
&=2\int_{\mathcal A_0}\int
\frac{w(s)^2\sigma_0^2(s)f_A(s)}{f_A(s)^4}
\sigma_0^2(s+h_nv)f_A(s+h_nv)\\
&\qquad\times
\left\{\int\Lambda_{\tau_n}(u)\Lambda_{\tau_n}(u+v)\,du\right\}^2
\,dv\,ds+O(h_n)\\
&=2\Lambda_{\tau_n}^{(4)}(0)
\int_{\mathcal A_0}
\left\{\frac{\sigma_0^2(s)w(s)}{f_A(s)}\right\}^2ds+O(h_n).
\end{aligned}
\]
Condition~\ref{assm:kernel} and $\tau_n\to\tau\in(0,\infty)$ provide
\[
\|\Lambda_{\tau_n}-\Lambda_\tau\|_1
+\|\Lambda_{\tau_n}-\Lambda_\tau\|_2
\le C|\tau_n-\tau|\longrightarrow0.
\]
Young's inequality yields
\[
\begin{aligned}
&\left\|
\int\Lambda_{\tau_n}(u)\Lambda_{\tau_n}(u+\cdot)\,du
-\int\Lambda_\tau(u)\Lambda_\tau(u+\cdot)\,du
\right\|_2\\
&\quad\le
(\|\Lambda_{\tau_n}\|_1+\|\Lambda_\tau\|_1)
\|\Lambda_{\tau_n}-\Lambda_\tau\|_2\longrightarrow0.
\end{aligned}
\]
Thus $\Lambda_{\tau_n}^{(4)}(0)\to\Lambda_\tau^{(4)}(0)$ and
$2h_nE_0H_{n,12}^2\to V_\tau^{DB}$.
Moreover, we have
\[
\int\left\{\int\Lambda_\tau(u)\Lambda_\tau(u+v)\,du\right\}dv
=\left\{\int\Lambda_\tau(u)\,du\right\}^2=1,
\]
so $\Lambda_\tau^{(4)}(0)>0$.
Condition~\ref{assm:parametric-null-support} also gives
\[
\int_{\mathcal A_0}
\left\{\frac{\sigma_0^2(s)w(s)}{f_A(s)}\right\}^2ds>0,
\qquad V_\tau^{DB}>0.
\]

For independent observations $O_1,O_2,O_3$, note that $H_{n,12}=H_{n,21}$, and
\[
\begin{aligned}
E_0(H_{n,12}\mid O_1)
=\zeta_\infty(O_1)\int J_{n,h}(A_1,t)
E_0\{\zeta_\infty(O_2)\mid A_2=t\}\,dF_0(t)=0.
\end{aligned}
\]
Also, it is obvious that
\[
\begin{aligned}
E_0(H_{n,13}H_{n,23}\mid O_1,O_2)
={}&\zeta_\infty(O_1)\zeta_\infty(O_2)
\int\sigma_0^2(t)J_{n,h}(A_1,t)J_{n,h}(A_2,t)\,dF_0(t).
\end{aligned}
\]
Lemma~\ref{lem:parametric-null-kernel} implies
\[
\begin{aligned}
&|E_0(H_{n,13}H_{n,23}\mid O_1,O_2)|\\
&\qquad\le Ch_n^{-2}|\zeta_\infty(O_1)\zeta_\infty(O_2)|
I(A_1,A_2\in\mathcal A_0)\\
&\qquad\qquad\times\int
I(|A_1-t|\le Ch_n)I(|A_2-t|\le Ch_n)\,dt\\
&\qquad\le Ch_n^{-1}|\zeta_\infty(O_1)\zeta_\infty(O_2)|
I(A_1,A_2\in\mathcal A_0)I(|A_1-A_2|\le Ch_n),
\end{aligned}
\]
where the constant in the last line may be enlarged.  Therefore, we get
\[
\begin{aligned}
&E_0\left[\{E_0(H_{n,13}H_{n,23}\mid O_1,O_2)\}^2\right]\\
&\quad\le Ch_n^{-2}
\int_{\mathcal A_0}\int_{\mathcal A_0}
\sigma_0^2(s)\sigma_0^2(t)I(|s-t|\le Ch_n)
\,dF_0(t)\,dF_0(s)
\le Ch_n^{-1}.
\end{aligned}
\]
Since $E_0H_{n,12}^2\asymp h_n^{-1}$, condition~(2.1) in
Theorem~1 of \citet{Hall1984central}, with sample size $n_k$, follows
from
\[
\frac{
E_0[\{E_0(H_{n,13}H_{n,23}\mid O_1,O_2)\}^2]
+n_k^{-1}E_0H_{n,12}^4}
{(E_0H_{n,12}^2)^2}
\le C\{h_n+(n_kh_n)^{-1}\}\longrightarrow0.
\]
The proof of Theorem~1 of \citet{Hall1984central} applies with kernel $H_{n,12}$ and row
length $n_k\to\infty$.
Indeed, Theorem~1 of \citet{Hall1984central} implies
\[
\frac{\displaystyle\sum_{\substack{i<j\\i,j\in I_k}}H_{n,ij}}
{\{n_k(n_k-1)E_0H_{n,12}^2/2\}^{1/2}}
\rightsquigarrow N(0,1).
\]
We have
\[
\operatorname{Var}_0\left(
\frac{\sqrt{h_n}}{n_k}\sum_{\substack{i,j\in I_k\\i\ne j}}H_{n,ij}
\,\middle|\,I_1,\ldots,I_K\right)
=2(1-n_k^{-1})h_nE_0H_{n,12}^2\longrightarrow V_\tau^{DB}.
\]
Theorem~1 of \citet{Hall1984central} and
\eqref{eq:parametric-null-oracle-diagonal-control} implies that
\begin{equation}
\begin{aligned}
\sup_{\substack{I_1,\ldots,I_K\text{ a partition of }\{1,\ldots,n\}\\
\min_k n_k\ge\underline\pi n}}\max_{1\le k\le K}\sup_{t\in\mathbb R}\left|
P_0\left\{
\frac{n_k\sqrt{h_n}\|Y_{n,k}\|_w^2-b_{0h,n}^{DB}}
{\sqrt{V_\tau^{DB}}}
\le t\,\middle|\,I_1,\ldots,I_K
\right\}-\Phi(t)\right|
\end{aligned}
\label{eq:parametric-null-oracle-fold-cdf}
\end{equation}
converges to $0$. 
% Indeed, the conditional distribution of a fold statistic depends only
% on $n_k$.  Failure of \eqref{eq:parametric-null-oracle-fold-cdf} would produce a
% deterministic sequence of admissible partitions and, since $K$ is fixed, a fixed fold
% along a subsequence for which the uniform distance from $\Phi$ is
% bounded below by a positive constant. Theorem~1 of \citet{Hall1984central} and continuity of $\Phi$ imply that the supremum over
% $t$ in \eqref{eq:parametric-null-oracle-fold-cdf} tends to zero along
% the selected sequence, contradicting the positive lower bound.
Conditional on $I_1,\ldots,I_K$, the statistics
$n_k\sqrt{h_n}\|Y_{n,k}\|_w^2$, $k=1,\ldots,K$, are independent.
We have
\begin{equation}
\begin{aligned}
&\sup_{t\in\mathbb R}\Bigg|
P_0\Bigg\{
\frac{\sum_{k=1}^K(n_k/n)n_k\sqrt{h_n}\|Y_{n,k}\|_w^2
-b_{0h,n}^{DB}}
{\{V_\tau^{DB}\sum_{k=1}^K(n_k/n)^2\}^{1/2}}
\le t\,\Biggm|\,I_1,\ldots,I_K
\Bigg\}-\Phi(t)\Bigg|\\
&\quad\le\sum_{k=1}^K\sup_{t\in\mathbb R}
\left|
P_0\left\{
\frac{n_k\sqrt{h_n}\|Y_{n,k}\|_w^2-b_{0h,n}^{DB}}
{\sqrt{V_\tau^{DB}}}
\le t\,\middle|\,I_1,\ldots,I_K
\right\}-\Phi(t)\right|.
\end{aligned}
\label{eq:parametric-null-oracle-weighted-cdf}
\end{equation}
% Positive rescaling preserves the uniform distance between distribution
% functions.  To obtain \eqref{eq:parametric-null-oracle-weighted-cdf},
% replace one fold distribution at a time by the standard normal
% distribution before applying its fold weight.  Integration against the
% remaining independent distributions bounds the $k$th change by the
% $k$th supremum in \eqref{eq:parametric-null-oracle-weighted-cdf}.  The final normal distribution has mean zero and variance
% one because
% \[
% \sum_{k=1}^K\frac{n_k}{n}=1,
% \qquad
% \sum_{k=1}^K
% \left\{\frac{n_k/n}{[\sum_{\ell=1}^K(n_\ell/n)^2]^{1/2}}\right\}^2=1.
% \]
Since $K<\infty$, equations~\eqref{eq:parametric-null-oracle-fold-cdf}
and~\eqref{eq:parametric-null-oracle-weighted-cdf} imply uniform
convergence over admissible partitions.  Taking expectation over the
partition gives
\[
\frac{\sum_{k=1}^K(n_k/n)n_k\sqrt{h_n}\|Y_{n,k}\|_w^2
-b_{0h,n}^{DB}}
{\{V_\tau^{DB}\sum_{k=1}^K(n_k/n)^2\}^{1/2}}
\rightsquigarrow N(0,1).
\]
Moreover, we have
\[
\left\{V_\tau^{DB}\sum_{k=1}^K(n_k/n)^2\right\}^{-1/2}
\le\sqrt{K/V_\tau^{DB}}.
\]
Equation~\eqref{eq:parametric-null-oracle-reduction} and Slutsky's
theorem therefore establish the asserted limit for
$T_{n,\mathcal M,0}^{DB}$.
Finally, the triangle inequality implies
\[
\begin{aligned}
\max_k\|L_{n,k,\mathcal M}\|_w
&\le\max_k\|Y_{n,k}\|_w
+\max_k\|D_{n,k}\|_w+\max_k\|U_{n,k}\|_w\\
&=O_p\{(nh_n)^{-1/2}\}+o_p(\varrho_n)+O_p(n^{-1/2})
=O_p\{(nh_n)^{-1/2}\},
\end{aligned}
\]
since $\varrho_n=h_n^{1/4}(nh_n)^{-1/2}$.
\end{proof}

\begin{proposition}
\label{prop:parametric-null-reproduction}
For every fold $k$ and $j\in\{0,1,2\}$, whenever both fold empirical
moment matrices at $a$ are nonsingular,
\[
P_{n,k}\{\Gamma_{n,k,a}(A)(A-a)^j\}=I(j=0).
\]
Consequently, $\hat m_{n,k,h_n,b_n}(a;\beta)$ equals $m(a;\beta)$
for every polynomial $m(\cdot;\beta)$ of degree at most two.  In
particular,
\[
\hat m_{n,k,h_n,b_n}(a;\hat\beta_{n,k})
=m(a;\hat\beta_{n,k})
\]
for each of the constant, linear, and quadratic families.
\end{proposition}

\begin{proof}
Fix $k$ and $a$ for which the two fold moment matrices are nonsingular.
Denote $D_{n,k,h_n,a,1}$, $D_{n,k,b_n,a,2}$, and
$c_{n,k,h_n,a,2}$ for the corresponding empirical matrices and bias
coefficient with $P_n$ replaced by $P_{n,k}$.
The local-linear part of $\Gamma_{n,k,a}$ satisfies
\[
e_1^\top D_{n,k,h_n,a,1}^{-1}
P_{n,k}\{w_{h_n,a,1}(A)K_{h_n,a}(A)\}=1,
\]
\[
e_1^\top D_{n,k,h_n,a,1}^{-1}
P_{n,k}\{w_{h_n,a,1}(A)K_{h_n,a}(A)(A-a)\}=0,
\]
and
\[
e_1^\top D_{n,k,h_n,a,1}^{-1}
P_{n,k}\{w_{h_n,a,1}(A)K_{h_n,a}(A)(A-a)^2\}
=h_n^2c_{n,k,h_n,a,2}.
\]
The local-quadratic second-derivative part satisfies
\[
2b_n^{-2}e_3^\top D_{n,k,b_n,a,2}^{-1}
P_{n,k}\{w_{b_n,a,2}(A)K_{b_n,a}(A)\}=0,
\]
\[
2b_n^{-2}e_3^\top D_{n,k,b_n,a,2}^{-1}
P_{n,k}\{w_{b_n,a,2}(A)K_{b_n,a}(A)(A-a)\}=0,
\]
and
\[
2b_n^{-2}e_3^\top D_{n,k,b_n,a,2}^{-1}
P_{n,k}\{w_{b_n,a,2}(A)K_{b_n,a}(A)(A-a)^2\}=2.
\]
The first two identities for each part are the corresponding columns of
the fold normal equations.  The third local-linear identity is the
definition of $c_{n,k,h_n,a,2}$, and the third local-quadratic identity
is exact reproduction of the second derivative of $(A-a)^2$.
By the definition of $\Gamma_{n,k,a}$, recall
\[
\begin{aligned}
\Gamma_{n,k,a}(t)
={}&e_1^\top D_{n,k,h_n,a,1}^{-1}
w_{h_n,a,1}(t)K_{h_n,a}(t)\\
&-\frac{h_n^2c_{n,k,h_n,a,2}}2
\left\{
2b_n^{-2}e_3^\top D_{n,k,b_n,a,2}^{-1}
w_{b_n,a,2}(t)K_{b_n,a}(t)
\right\}.
\end{aligned}
\]
Consequently, one can easily check
\begin{equation}
\begin{aligned}
P_{n,k}\left[
\Gamma_{n,k,a}(A)
\begin{pmatrix}
1\\ A-a\\ (A-a)^2
\end{pmatrix}
\right]=
\begin{pmatrix}
1\\ 0\\ h_n^2c_{n,k,h_n,a,2}
\end{pmatrix}
-\frac{h_n^2c_{n,k,h_n,a,2}}2
\begin{pmatrix}
0\\ 0\\ 2
\end{pmatrix}
=
\begin{pmatrix}
1\\ 0\\ 0
\end{pmatrix}.
\end{aligned}
\label{eq:parametric-null-fold-polynomial-moments}
\end{equation}

For every polynomial $m(\cdot;\beta)$ of degree at most two,
\[
m(t;\beta)
=m(a;\beta)
+\frac{\partial m(a;\beta)}{\partial a}(t-a)
+\frac12\frac{\partial^2m(a;\beta)}{\partial a^2}(t-a)^2.
\]
Equation~\eqref{eq:parametric-null-fold-polynomial-moments} therefore yields, for $\beta\in\mathcal B$,
\[
\begin{aligned}
P_{n,k}\{\Gamma_{n,k,a}(A)m(A;\beta)\}
={}&m(a;\beta)P_{n,k}\{\Gamma_{n,k,a}(A)\}\\
&+\frac{\partial m(a;\beta)}{\partial a}
P_{n,k}\{\Gamma_{n,k,a}(A)(A-a)\}\\
&+\frac12\frac{\partial^2m(a;\beta)}{\partial a^2}
P_{n,k}\{\Gamma_{n,k,a}(A)(A-a)^2\}\\
={}&m(a;\beta).
\end{aligned}
\]
For $\beta=\hat\beta_{n,k}$, the same substitution into
\eqref{eq:parametric-null-fold-polynomial-moments}, together with $\hat m_{n,k,h_n,b_n}(a;\beta)$, yields
\[
\begin{aligned}
\hat m_{n,k,h_n,b_n}(a;\hat\beta_{n,k})
&=P_{n,k}\{\Gamma_{n,k,a}(A)m(A;\hat\beta_{n,k})\}\\
&=m(a;\hat\beta_{n,k})\cdot1
+\frac{\partial m(a;\hat\beta_{n,k})}{\partial a}\cdot0
+\frac12\frac{\partial^2m(a;\hat\beta_{n,k})}{\partial a^2}\cdot0\\
&=m(a;\hat\beta_{n,k}).
\end{aligned}
\]
\end{proof}

\subsection{Local alternatives}
\label{sec:parametric-null-local-alternatives}

\begin{assumptionp}{M}\label{assm:parametric-null-triangular}
\par\medskip\noindent
\begin{enumerate}[label=\textup{(M\arabic*)},ref=M\arabic*,leftmargin=*,itemsep=0.35em]
\item\label{assm:parametric-null-triangular-law}
For each $n$, let $O_1,\ldots,O_n$ be independent with common law
$P_{0,n}$, let $E_{0,n}$ denote expectation under $P_{0,n}$, and let $P_n$
denote their empirical distribution, with fold empirical distributions
$P_{n,k}$.  The marginal law of
$A$ under $P_{0,n}$ is $F_0$ for every $n$.  Thus, for every
$F_0$-integrable $r$,
\[
P_{0,n}\{r(A)\}=P_0\{r(A)\}=\int r(a)\,dF_0(a).
\]
For each $n$, let $(q_{0,n},h_{0,n})$ denote the bridge pair in
Assumption~\ref{assm:B} under $P_{0,n}$, and let
$(q_{\infty,n},h_{\infty,n})$ be deterministic functions.  Define
\[
m_{\infty,n}(a):=E_{0,n}\{h_{\infty,n}(W,a,X)\},
\]
\[
\begin{aligned}
\xi_{\infty,n}(O)
&:=q_{\infty,n}(Z,A,X)
\{Y-h_{\infty,n}(W,A,X)\}+m_{\infty,n}(A),\\
\zeta_{\infty,n}(O)
&:=\xi_{\infty,n}(O)-\theta_{0,n}(A),
\end{aligned}
\]
and
\[
\sigma_{0,n}^2(a)
:=E_{0,n}\{\zeta_{\infty,n}(O)^2\mid A=a\}.
\]

\item\label{assm:parametric-null-triangular-inherited}
Assumptions~\ref{assm:A}, \ref{assm:B}, and~\ref{assm:D},
conditions~\ref{assm:G-bandwidth}--\ref{assm:G-smooth}, and
conditions~\ref{assm:parametric-null-support}--
\ref{assm:parametric-null-score-moment} are imposed along
$\{P_{0,n}:n\ge1\}$.  The consistency and interior-solution requirement for
$\hat\beta_{n,k}$, $1\le k\le K$, in
condition~\ref{assm:parametric-null-family} is replaced by
condition~\ref{assm:parametric-null-local-family}.  In these conditions, the
substitutions are
\[
\begin{aligned}
(P_0,E_0,\theta_0,q_0,h_0)
&\mapsto(P_{0,n},E_{0,n},\theta_{0,n},q_{0,n},h_{0,n}),\\
(q_\infty,h_\infty,m_\infty)
&\mapsto(q_{\infty,n},h_{\infty,n},m_{\infty,n}),\\
(\xi_\infty,\zeta_\infty,\sigma_0^2)
&\mapsto(\xi_{\infty,n},\zeta_{\infty,n},\sigma_{0,n}^2).
\end{aligned}
\]
Every $o_p$ and $O_p$ is replaced by $o_{P_{0,n}}$ and
$O_{P_{0,n}}$, respectively.  The functions denoted by $h$ and $q$ in
Assumption~\ref{assm:B} are $h_{0,n}$ and $q_{0,n}$.  The definitions of
$d_{q,k}$ and $d_{h,k}$ are read with $E_0$ replaced by $E_{0,n}$.
The measurable partition in condition~\ref{assm:G-DR-rates} may depend on
$n$.
The quantity $\Delta_{n,k}$ in
\eqref{eq:conditional-proximal-product-bias} is read with
$E_0,q_0,h_0$ replaced by $E_{0,n},q_{0,n},h_{0,n}$. 

\item\label{assm:parametric-null-triangular-uniform}
The number of folds $K$, the dimension $p$, and the constants $\delta_3$
and $\alpha$ are fixed.  The constants $C_q,C_h$ and the finite conditional
second-moment bound in condition~\ref{assm:G-DR-rates} are uniform in $n$.
The constant $c_1$ in condition~\ref{assm:basic.proxy} is independent of
$n$.
The H\"older constant in condition~\ref{assm:G-smooth}\textup{(i)}, the
essential bound in condition~\ref{assm:G-smooth}\textup{(iii)}, and the
upper, lower, and fourth-moment bounds in conditions
\ref{assm:G-smooth}\textup{(iv)} and
\ref{assm:parametric-null-support} are independent of $n$.  In addition,
\[
\sup_n\sup_{a\in\mathcal A_{\delta_3}}
\left|\frac{d}{da}\sigma_{0,n}^2(a)\right|<\infty.
\]

\end{enumerate}
\end{assumptionp}
We remark that the definitions of
$\Gamma_{n,k,a}$, $L_{n,k,\mathcal M}$, and
$R_{n,k,j}^{\mathcal M}$, including
\eqref{eq:parametric-null-leading}--\eqref{eq:parametric-null-R45},
are interpreted with the substitutions in the preceding display.
In the definition of $R_{n,k,5}^{\mathcal M}$,
$Q_0^{WX}$ denotes the $P_{0,n}$-marginal distribution of $(W,X)$.
The objects $F_0$, $f_A$, $m$, $\beta_0$, $\dot m_0$, and $J_0$ are fixed.
We further define
\begin{align}
b_{0h,n}^{DB}
&:=h_n^{-1/2}\Lambda_{\tau_n}^{(2)}(0)
\int_{\mathcal A_0}\frac{\sigma_{0,n}^2(a)w(a)}{f_A(a)}\,da,
\label{eq:parametric-null-local-centering}\\
V_{\tau,n}^{DB}
&:=2\Lambda_\tau^{(4)}(0)
\int_{\mathcal A_0}
\left\{\frac{\sigma_{0,n}^2(a)w(a)}{f_A(a)}\right\}^2da.
\label{eq:parametric-null-local-variance}
\end{align}

\begin{assumptionp}{N}\label{assm:parametric-null-local}
\par\medskip\noindent
\begin{enumerate}[label=\textup{(N\arabic*)},ref=N\arabic*,leftmargin=*,itemsep=0.35em]
\item\label{assm:parametric-null-local-law}
Equation~\eqref{eq:parametric-null-local-alternative} holds under
$P_{0,n}$.  The function $g$ is bounded on $\mathcal A$, is twice
continuously differentiable on $\mathcal A_{\delta_3}$, and $g''$ is
H\"older continuous on $\mathcal A_{\delta_3}$ with exponent $\alpha$ in condition
\ref{assm:G-smooth}\textup{(i)}.

\item\label{assm:parametric-null-local-family}
For every $\epsilon>0$,
\[
\inf_{\substack{\beta\in\mathcal B\\
\|\beta-\beta_0\|\ge\epsilon}}
P_0\left[
\{m(A;\beta)-m(A;\beta_0)\}^2
\right]>0.
\]
For all sufficiently large $n$, the minimizer
\begin{equation}
\beta_n^\dagger
:=
\argmin_{\beta\in\mathcal B}
P_0\left[
\{\theta_{0,n}(A)-m(A;\beta)\}^2
\right]
\label{eq:parametric-null-local-population-beta}
\end{equation}
is unique and lies in the interior of $\mathcal B$, and
\[
\max_{1\le k\le K}\|\hat\beta_{n,k}-\beta_n^\dagger\|
=o_{P_{0,n}}(1).
\]
\end{enumerate}
\end{assumptionp}

\begin{lemma}
\label{lem:parametric-null-local-projection}
Under Assumptions~\ref{assm:parametric-null-triangular} and
\ref{assm:parametric-null-local},
\[
\beta_n^\dagger-\beta_0
=c\varrho_nJ_0^{-1}P_0\{\dot m_0(A)g(A)\}
+O(\varrho_n^2),
\]
\[
\sup_{a\in\mathcal A}
\left|
\theta_{0,n}(a)-m(a;\beta_n^\dagger)
-c\varrho_ng_{\mathcal M}^{\perp}(a)
\right|=O(\varrho_n^2),
\]
and
\[
\max_k\|\hat\beta_{n,k}-\beta_n^\dagger\|
=O_{P_{0,n}}(n^{-1/2})+o_{P_{0,n}}(\varrho_n)
=o_{P_{0,n}}(\varrho_n).
\]
Moreover,
\begin{equation}
\left\|
P_0\{\Gamma_{0,\cdot}(A)g_{\mathcal M}^{\perp}(A)\}
-g_{\mathcal M}^{\perp}
\right\|_w\longrightarrow0.
\label{eq:parametric-null-local-smoothed-direction}
\end{equation}
\end{lemma}

\begin{proof}
The minimization in
\eqref{eq:parametric-null-local-population-beta} gives
\[
\begin{aligned}
\|m(\cdot;\beta_n^\dagger)-m(\cdot;\beta_0)\|_{L_2(F_0)}
&\le
\|m(\cdot;\beta_n^\dagger)-\theta_{0,n}\|_{L_2(F_0)}
+\|\theta_{0,n}-m(\cdot;\beta_0)\|_{L_2(F_0)}\\
&\le2|c|\varrho_n\|g\|_{L_2(F_0)}\longrightarrow0.
\end{aligned}
\]
Condition~\ref{assm:parametric-null-local-family} therefore implies
$\beta_n^\dagger\to\beta_0$ and, for all sufficiently large $n$,
\begin{equation}
P_0\left[
\frac{\partial m(A;\beta_n^\dagger)}{\partial\beta}
\{\theta_{0,n}(A)-m(A;\beta_n^\dagger)\}
\right]=0.
\label{eq:parametric-null-local-population-score}
\end{equation}
The bounded first and second derivatives in
condition~\ref{assm:parametric-null-family} imply
\begin{equation}
\begin{aligned}
&\sup_{a\in\mathcal A}|m(a;\beta_n^\dagger)-m(a;\beta_0)|
\le C\|\beta_n^\dagger-\beta_0\|,\\
&\sup_{a\in\mathcal A}\left\|
\frac{\partial m(a;\beta_n^\dagger)}{\partial\beta}
-\dot m_0(a)\right\|
\le C\|\beta_n^\dagger-\beta_0\|,\\
&\sup_{a\in\mathcal A}\big|
m(a;\beta_n^\dagger)-m(a;\beta_0)
-\dot m_0(a)^\top(\beta_n^\dagger-\beta_0)\big|
\le C\|\beta_n^\dagger-\beta_0\|^2.
\end{aligned}
\label{eq:parametric-null-local-taylor}
\end{equation}
By \eqref{eq:parametric-null-local-alternative} and
\eqref{eq:parametric-null-local-population-score}, we get
\[
\begin{aligned}
&J_0(\beta_n^\dagger-\beta_0)
-c\varrho_nP_0\{\dot m_0(A)g(A)\}\\
&\quad=P_0\left[
\left\{\frac{\partial m(A;\beta_n^\dagger)}{\partial\beta}
-\dot m_0(A)\right\}
\{c\varrho_ng(A)-m(A;\beta_n^\dagger)+m(A;\beta_0)\}
\right]\\
&\qquad-P_0\left[
\dot m_0(A)\{m(A;\beta_n^\dagger)-m(A;\beta_0)
-\dot m_0(A)^\top(\beta_n^\dagger-\beta_0)\}\right].
\end{aligned}
\]
Consequently, \eqref{eq:parametric-null-local-taylor} further implies
\begin{equation}
\begin{aligned}
\left\|J_0(\beta_n^\dagger-\beta_0)
-c\varrho_nP_0\{\dot m_0(A)g(A)\}\right\|\leq C\{\varrho_n\|\beta_n^\dagger-\beta_0\|
+\|\beta_n^\dagger-\beta_0\|^2\}.
\end{aligned}
\label{eq:parametric-null-local-population-error}
\end{equation}
Since $\beta_n^\dagger\to\beta_0$ and $\varrho_n\to0$, one has
\[
\begin{aligned}
\frac{\lambda_{\min}(J_0)}2\|\beta_n^\dagger-\beta_0\|
&\le\{\lambda_{\min}(J_0)-C\varrho_n
-C\|\beta_n^\dagger-\beta_0\|\}\|\beta_n^\dagger-\beta_0\|\\
&\le |c|\varrho_n\|P_0\{\dot m_0(A)g(A)\}\|
\end{aligned}
\]
for all sufficiently large $n$.  Substitution of
$\|\beta_n^\dagger-\beta_0\|=O(\varrho_n)$ into
\eqref{eq:parametric-null-local-population-error} gives
\begin{equation}
\left\|\beta_n^\dagger-\beta_0
-c\varrho_nJ_0^{-1}P_0\{\dot m_0(A)g(A)\}\right\|
=O(\varrho_n^2).
\label{eq:parametric-null-local-coefficient-expansion}
\end{equation}
By the definition of $g_{\mathcal M}^{\perp}(a)$, for every
$a\in\mathcal A$, we get
\[
\begin{aligned}
&\theta_{0,n}(a)-m(a;\beta_n^\dagger)
-c\varrho_ng_{\mathcal M}^{\perp}(a)\\
&\quad=-\dot m_0(a)^\top
\left[\beta_n^\dagger-\beta_0
-c\varrho_nJ_0^{-1}P_0\{\dot m_0(A)g(A)\}\right]\\
&\qquad-\{m(a;\beta_n^\dagger)-m(a;\beta_0)
-\dot m_0(a)^\top(\beta_n^\dagger-\beta_0)\}.
\end{aligned}
\]
Displays~\eqref{eq:parametric-null-local-taylor} and
\eqref{eq:parametric-null-local-coefficient-expansion} yield
\begin{equation}
\begin{aligned}
\sup_{a\in\mathcal A}|\theta_{0,n}(a)-m(a;\beta_n^\dagger)
-c\varrho_ng_{\mathcal M}^{\perp}(a)|&=O(\varrho_n^2),\\
\sup_{a\in\mathcal A}|\theta_{0,n}(a)-m(a;\beta_n^\dagger)|
&=O(\varrho_n).
\end{aligned}
\label{eq:parametric-null-local-population-residual}
\end{equation}

For the bounded deterministic vector or matrix functions $v$ used
below, apply the conditional mean and second-moment calculations in
the proof of Lemma~\ref{lem:parametric-null-beta-rate} with
\[
(P_0,E_0,\theta_0,q_0,h_0)
\mapsto(P_{0,n},E_{0,n},\theta_{0,n},q_{0,n},h_{0,n}).
\]
Note that the conditional moment bounds depend on $v$ only through its
uniform bound and its fixed dimension. 
The required bound on
$\theta_{0,n}$ follows from
$\sup_{a\in\mathcal A}|\theta_{0,n}(a)|
\le\sup_{a\in\mathcal A}|m(a;\beta_0)|
+|c|\varrho_n\sup_{a\in\mathcal A}|g(a)|=O(1)$.
Thus, we have
\[
\begin{aligned}
&E_{0,n}\left[
P_{n,k}[v(A)\{\hat\xi_n^{(-k)}-\theta_{0,n}(A)\}]
\,\middle|\,\trainingsigma_k\right]\\
&\quad=\int_{\mathcal A}v(a)\Delta_{n,k}(a)\,dF_0(a)\\
&\qquad+\frac1{n_k}E_{0,n}\left[
v(A)\{\hat h_n^{(-k)}(W,A,X)-\hat h_n^{(-k)}(W',A,X')\}
\,\middle|\,\trainingsigma_k\right],
\end{aligned}
\]
where $(W',X')$ is the independent copy in
condition~\ref{assm:parametric-null-score-moment}, and
\[
\begin{aligned}
E_{0,n}\left[
\left\|
P_{n,k}[v(A)\{\hat\xi_n^{(-k)}-\theta_{0,n}(A)\}]
-\int_{\mathcal A}v(a)\Delta_{n,k}(a)\,dF_0(a)
\right\|^2
\,\middle|\,\trainingsigma_k\right]=O_{P_{0,n}}(n_k^{-1}).
\end{aligned}
\]
Applying the conditional probability truncation used to obtain the
$O_p(n^{-1/2})$ bound for weighted fold averages in the proof of
Lemma~\ref{lem:parametric-null-beta-rate}, with
$K$ fixed and $n_k\asymp n$, we have
\begin{equation}
\max_k\left\|
P_{n,k}[v(A)\{\hat\xi_n^{(-k)}-\theta_{0,n}(A)\}]
-\int_{\mathcal A}v(a)\Delta_{n,k}(a)\,dF_0(a)
\right\|=O_{P_{0,n}}(n^{-1/2}).
\label{eq:parametric-null-local-weighted-mean}
\end{equation}

Take $v(a)=\partial m(a;\beta_n^\dagger)/\partial\beta$ in
\eqref{eq:parametric-null-local-weighted-mean}.
By \eqref{eq:parametric-null-local-taylor} and
condition~\ref{assm:parametric-null-product}, we obtain
\[
\begin{aligned}
&\max_k\left\|\int_{\mathcal A}
\frac{\partial m(a;\beta_n^\dagger)}{\partial\beta}
\Delta_{n,k}(a)\,dF_0(a)\right\|\\
&\quad\le\max_k\left\|\int_{\mathcal A}
\dot m_0(a)\Delta_{n,k}(a)\,dF_0(a)\right\|
+C\|\beta_n^\dagger-\beta_0\|
\max_k\int_{\mathcal A}|\Delta_{n,k}(a)|\,dF_0(a)\\
&\quad=o_{P_{0,n}}(\varrho_n)
+O(\varrho_n)o_{P_{0,n}}(1)
=o_{P_{0,n}}(\varrho_n).
\end{aligned}
\]
Equations~\eqref{eq:parametric-null-local-population-score} and
\eqref{eq:parametric-null-local-population-residual} give, respectively,
\[
\begin{aligned}
P_{n,k}\left[
\frac{\partial m(A;\beta_n^\dagger)}{\partial\beta}
\{\hat\xi_n^{(-k)}-m(A;\beta_n^\dagger)\}\right]
&=P_{n,k}\left[
\frac{\partial m(A;\beta_n^\dagger)}{\partial\beta}
\{\hat\xi_n^{(-k)}-\theta_{0,n}(A)\}\right]\\
&\quad+(P_{n,k}-P_0)\left[
\frac{\partial m(A;\beta_n^\dagger)}{\partial\beta}
\{\theta_{0,n}(A)-m(A;\beta_n^\dagger)\}\right],
\end{aligned}
\]
and
\[
\begin{aligned}
&E_{0,n}\Bigg[
\left\|(P_{n,k}-P_0)\left[
\frac{\partial m(A;\beta_n^\dagger)}{\partial\beta}
\{\theta_{0,n}(A)-m(A;\beta_n^\dagger)\}\right]\right\|^2
\,\Biggm|\,\trainingsigma_k\Bigg]\\
&\qquad\le\frac1{n_k}P_0\left[
\left\|\frac{\partial m(A;\beta_n^\dagger)}{\partial\beta}\right\|^2
\{\theta_{0,n}(A)-m(A;\beta_n^\dagger)\}^2\right]
=O(\varrho_n^2/n_k).
\end{aligned}
\]
Consequently, we obtain
\begin{equation}
\max_k\left\|P_{n,k}\left[
\frac{\partial m(A;\beta_n^\dagger)}{\partial\beta}
\{\hat\xi_n^{(-k)}-m(A;\beta_n^\dagger)\}\right]\right\|
=O_{P_{0,n}}(n^{-1/2})+o_{P_{0,n}}(\varrho_n).
\label{eq:parametric-null-local-score-rate}
\end{equation}

For the negative score derivative, apply
\eqref{eq:parametric-null-local-weighted-mean} entrywise with
\[
v(a)=\frac{\partial^2m(a;\beta_0)}
{\partial\beta\partial\beta^\top}.
\]
Since $\theta_{0,n}(a)-m(a;\beta_0)=c\varrho_ng(a)$, we get
\[
\begin{aligned}
&\max_k\left\|P_{n,k}\left[
\frac{\partial^2m(A;\beta_0)}{\partial\beta\partial\beta^\top}
\{\hat\xi_n^{(-k)}-m(A;\beta_0)\}\right]\right\|\\
&\quad\le O_{P_{0,n}}(n^{-1/2})
+C\max_k\int_{\mathcal A}|\Delta_{n,k}(a)|\,dF_0(a)
+C\varrho_n=o_{P_{0,n}}(1).
\end{aligned}
\]
Also,
$\max_k\|P_{n,k}(\dot m_0\dot m_0^\top)-J_0\|
=O_{P_{0,n}}(n^{-1/2})$ because $\dot m_0$ is bounded and the
marginal law of $A$ is $F_0$.  Hence
\[
\max_k\left\|P_{n,k}\left[
\dot m_0(A)\dot m_0(A)^\top
-\frac{\partial^2m(A;\beta_0)}{\partial\beta\partial\beta^\top}
\{\hat\xi_n^{(-k)}-m(A;\beta_0)\}\right]-J_0\right\|
=o_{P_{0,n}}(1).
\]

The proof of Lemma~\ref{lem:parametric-null-beta-rate} establishes
continuity of the negative score derivative near $\beta_0$ by
bounding $P_{n,k}|\hat\xi_n^{(-k)}|^2$ and applying
the Cauchy--Schwarz inequality to the difference of the second
derivatives multiplied by $\hat\xi_n^{(-k)}-m(A;\beta_0)$.
For $P_{0,n}$, the required moment and expectation are
\[
\max_kP_{n,k}|\hat\xi_n^{(-k)}|^2=O_{P_{0,n}}(1)
\]
and, for every sufficiently small $\delta>0$,
\[
\begin{aligned}
&E_{0,n}\left[P_{n,k}\left\{
\sup_{\|\beta-\beta_0\|\le\delta}
\left\|\frac{\partial^2m(A;\beta)}{\partial\beta\partial\beta^\top}
-\frac{\partial^2m(A;\beta_0)}{\partial\beta\partial\beta^\top}
\right\|^2\right\}\right]\\
&\quad=P_0\left[
\sup_{\|\beta-\beta_0\|\le\delta}
\left\|\frac{\partial^2m(A;\beta)}{\partial\beta\partial\beta^\top}
-\frac{\partial^2m(A;\beta_0)}{\partial\beta\partial\beta^\top}
\right\|^2\right]\longrightarrow0
\qquad(\delta\downarrow0).
\end{aligned}
\]
The bound on $\max_kP_{n,k}|\hat\xi_n^{(-k)}|^2$ follows from the
conditional second-moment calculation in the proof of
Lemma~\ref{lem:parametric-null-beta-rate} with $E_0$ replaced by
$E_{0,n}$, under
condition~\ref{assm:parametric-null-score-moment}.
The marginal law $F_0$ does not depend on $n$; boundedness and
continuity of the second derivatives permit dominated convergence.
The remaining changes in the negative score derivative satisfy
\[
\begin{aligned}
&\sup_{\|\beta-\beta_0\|\le\delta}
\left\|P_{n,k}\left[
\frac{\partial m(A;\beta)}{\partial\beta}
\frac{\partial m(A;\beta)}{\partial\beta^\top}
-\dot m_0(A)\dot m_0(A)^\top\right]\right\|\\
&\quad+
\sup_{\|\beta-\beta_0\|\le\delta}
\left\|P_{n,k}\left[
\frac{\partial^2m(A;\beta)}{\partial\beta\partial\beta^\top}
\{m(A;\beta)-m(A;\beta_0)\}\right]\right\|
\le C\delta.
\end{aligned}
\]
Therefore, the probability bound in the proof of
Lemma~\ref{lem:parametric-null-beta-rate}, with $P_0,E_0$
replaced by $P_{0,n},E_{0,n}$ and with the reference point
$\beta_0$ unchanged, yields, for every $\varepsilon>0$,
\begin{equation}
\begin{aligned}
\lim_{\delta\downarrow0}\limsup_{n\to\infty}
P_{0,n}\Bigg(&\max_k\sup_{\|\beta-\beta_0\|\le\delta}
\Bigg\|P_{n,k}\Bigg[
\frac{\partial m(A;\beta)}{\partial\beta}
\frac{\partial m(A;\beta)}{\partial\beta^\top}\\
&\qquad\qquad\qquad\quad-\frac{\partial^2m(A;\beta)}{\partial\beta\partial\beta^\top}
\{\hat\xi_n^{(-k)}-m(A;\beta)\}\Bigg]-J_0\Bigg\|
>\varepsilon\Bigg)=0.
\end{aligned}
\label{eq:parametric-null-local-derivative-continuity}
\end{equation}
Condition~\ref{assm:parametric-null-local-family} and
$\beta_n^\dagger\to\beta_0$ implies
\[
\max_k\sup_{0\le t\le1}
\|\beta_n^\dagger+t(\hat\beta_{n,k}-\beta_n^\dagger)-\beta_0\|
\le\|\beta_n^\dagger-\beta_0\|
+\max_k\|\hat\beta_{n,k}-\beta_n^\dagger\|
=o_{P_{0,n}}(1).
\]
By \eqref{eq:parametric-null-local-derivative-continuity}, the
negative score derivative therefore equals $J_0+o_{P_{0,n}}(1)$
uniformly on the segments
$\{\beta_n^\dagger+t(\hat\beta_{n,k}-\beta_n^\dagger):0\le t\le1\}$
and over $k$.  With probability tending to one, the segments lie in
the differentiability neighborhood and in the interior of $\mathcal B$.
The score equations and the fundamental theorem of calculus yield
\begin{equation}
\begin{aligned}
&\Bigg(\int_0^1
\left.P_{n,k}\Bigg[
\frac{\partial m(A;\beta)}{\partial\beta}
\frac{\partial m(A;\beta)}{\partial\beta^\top}-\frac{\partial^2m(A;\beta)}{\partial\beta\partial\beta^\top}
\{\hat\xi_n^{(-k)}-m(A;\beta)\}
\Bigg]\right|_{\beta=\beta_n^\dagger+t(\hat\beta_{n,k}-\beta_n^\dagger)}
\,dt\Bigg)\\
&\qquad\times(\hat\beta_{n,k}-\beta_n^\dagger)
=P_{n,k}\left[
\frac{\partial m(A;\beta_n^\dagger)}{\partial\beta}
\{\hat\xi_n^{(-k)}-m(A;\beta_n^\dagger)\}\right].
\end{aligned}
\label{eq:parametric-null-local-score-linearization}
\end{equation}
Again, by \eqref{eq:parametric-null-local-derivative-continuity}, the
integral matrices in \eqref{eq:parametric-null-local-score-linearization}
have smallest eigenvalue at least $\lambda_{\min}(J_0)/2$ with
probability tending to one.  On the same event, we have
\[
\begin{aligned}
\max_k\|\hat\beta_{n,k}-\beta_n^\dagger\|&\leq\frac{2}{\lambda_{\min}(J_0)}
\max_k\left\|P_{n,k}\left[
\frac{\partial m(A;\beta_n^\dagger)}{\partial\beta}
\{\hat\xi_n^{(-k)}-m(A;\beta_n^\dagger)\}\right]\right\|\\
&\quad=O_{P_{0,n}}(n^{-1/2})+o_{P_{0,n}}(\varrho_n)
=o_{P_{0,n}}(\varrho_n),
\end{aligned}
\]
by \eqref{eq:parametric-null-local-score-rate} and
$n^{-1/2}=h_n^{1/4}\varrho_n=o(\varrho_n)$.

For \eqref{eq:parametric-null-local-smoothed-direction}, first note
that the definition of $g_{\mathcal M}^{\perp}(a)$ gives
\[
\sup_{a\in\mathcal A}|g_{\mathcal M}^{\perp}(a)|
\le\sup_{a\in\mathcal A}|g(a)|
+\sup_{a\in\mathcal A}\|\dot m_0(a)\|
\|J_0^{-1}P_0\{\dot m_0(A)g(A)\}\|<\infty.
\]
The same argument in the proof of
Proposition~\ref{prop:parametric-null-reproduction}, with
\[
(P_{n,k},D_{n,k,h_n,a,1},D_{n,k,b_n,a,2},c_{n,k,h_n,a,2})
\mapsto(P_0,D_{0,h_n,a,1},D_{0,b_n,a,2},c_{0,h_n,a,2}),
\]
yields, for all sufficiently large $n$ and every $a\in\mathcal A_0$,
\[
\begin{aligned}
P_0\{\Gamma_{0,a}(A)\}
=e_1^\top D_{0,h_n,a,1}^{-1}D_{0,h_n,a,1}e_1-\frac{h_n^2}{b_n^2}c_{0,h_n,a,2}
e_3^\top D_{0,b_n,a,2}^{-1}D_{0,b_n,a,2}e_1
=1.
\end{aligned}
\]
The inverses exist uniformly over $a\in\mathcal A_0$ by the matrix
bounds in the proof of Lemma~\ref{lem:parametric-null-kernel}.
The Cauchy--Schwarz inequality, the bound
$|\Gamma_{0,a}(t)|\le Ch_n^{-1}I(|t-a|\le Ch_n)$ in
Lemma~\ref{lem:parametric-null-kernel}, and the boundedness of $f_A$
and $w$ give
\begin{equation}
\begin{aligned}
&\left\|P_0\{\Gamma_{0,\cdot}(A)g_{\mathcal M}^{\perp}(A)\}
-g_{\mathcal M}^{\perp}\right\|_w^2\\
&\quad=\int_{\mathcal A_0}
\left|\int\Gamma_{0,a}(t)
\{g_{\mathcal M}^{\perp}(t)-g_{\mathcal M}^{\perp}(a)\}
\,dF_0(t)\right|^2w(a)\,da\\
&\quad\le\int_{\mathcal A_0}w(a)
\left\{\int|\Gamma_{0,a}(t)|\,dF_0(t)\right\}
\left\{\int|\Gamma_{0,a}(t)|
|g_{\mathcal M}^{\perp}(t)-g_{\mathcal M}^{\perp}(a)|^2
\,dF_0(t)\right\}da\\
&\quad\le\frac C{h_n}\int_{|r|\le Ch_n}\int_{\mathcal A_0}
|g_{\mathcal M}^{\perp}(a+r)-g_{\mathcal M}^{\perp}(a)|^2
\,da\,dr\\
&\quad\le C\sup_{|r|\le Ch_n}\int_{\mathcal A_0}
|g_{\mathcal M}^{\perp}(a+r)-g_{\mathcal M}^{\perp}(a)|^2\,da.
\end{aligned}
\label{eq:parametric-null-local-translation-bound}
\end{equation}
The expression $g_{\mathcal M}^{\perp}I_{\mathcal A_{\delta_3}}$
equals zero outside $\mathcal A_{\delta_3}$ and belongs to
$L_2(\mathbb R)$, since $g_{\mathcal M}^{\perp}$ is bounded and
$\mathcal A_{\delta_3}$ is compact.  For all sufficiently large $n$,
$a\in\mathcal A_0$ and $|r|\le Ch_n$ imply
$a+r\in\mathcal A_{\delta_3}$.  
Therefore, we obtain
\[
\begin{aligned}
&\sup_{|r|\le Ch_n}\int_{\mathcal A_0}
|g_{\mathcal M}^{\perp}(a+r)-g_{\mathcal M}^{\perp}(a)|^2\,da\\
&\quad\le\sup_{|r|\le Ch_n}\int_{\mathbb R}
\big|g_{\mathcal M}^{\perp}(a+r)I_{\mathcal A_{\delta_3}}(a+r)
-g_{\mathcal M}^{\perp}(a)I_{\mathcal A_{\delta_3}}(a)\big|^2\,da
\longrightarrow0.
\end{aligned}
\]
Equation~\eqref{eq:parametric-null-local-translation-bound} proves
\eqref{eq:parametric-null-local-smoothed-direction}.
\end{proof}

Define
\begin{equation}
\widetilde R_{n,k,6}^{\mathcal M}(a)
:=P_{n,k}\left[
\{\Gamma_{n,k,a}(A)-\Gamma_{0,a}(A)\}
\{\theta_{0,n}(A)-m(A;\hat\beta_{n,k})\}
\right].
\label{eq:parametric-null-local-R6}
\end{equation}

\begin{lemma}
\label{lem:parametric-null-local-decomposition}
Under Assumptions~\ref{assm:parametric-null-triangular} and
\ref{assm:parametric-null-local}, for all sufficiently large $n$,
on the event that the
fold moment matrices are nonsingular, for every $k$,
\[
\hat\theta_{n,k}^{DB}(a)
-\hat m_{n,k,h_n,b_n}(a;\hat\beta_{n,k})
=L_{n,k,\mathcal M}(a)
+\sum_{j=2}^5R_{n,k,j}^{\mathcal M}(a)
+\widetilde R_{n,k,6}^{\mathcal M}(a),
\]
uniformly over $1\le k\le K$, and
\[
\|R_{n,k,2}^{\mathcal M}\|_w=o_{P_{0,n}}(n^{-1/2}),
\qquad
\|R_{n,k,3}^{\mathcal M}\|_w=O_{P_{0,n}}\{(nh_n)^{-1}\},
\]
\[
\|R_{n,k,4}^{\mathcal M}\|_w=o_{P_{0,n}}(\varrho_n),
\qquad
\|R_{n,k,5}^{\mathcal M}\|_w=O_{P_{0,n}}(n^{-1}h_n^{-1/2}),
\]
\[
\|\widetilde R_{n,k,6}^{\mathcal M}\|_w
=O_{P_{0,n}}\{\varrho_n(nh_n)^{-1/2}\}
=o_{P_{0,n}}(n^{-1/2}).
\]
\end{lemma}

\begin{proof}
Apply the same argument used in the proof of
Lemma~\ref{lem:parametric-null-decomposition} with the substitutions in
condition~\ref{assm:parametric-null-triangular-inherited}.
Theorem~\ref{thm:DR_for_PO} and condition~\ref{assm:G-DR-rates} give
\begin{equation}
\begin{aligned}
&E_{0,n}\{\xi_{\infty,n}(O)\mid A=a\}-\theta_{0,n}(a)\\
&=E_{0,n}\left[(q_{\infty,n}(Z,a,X)-q_{0,n}(Z,a,X))(h_{0,n}(W,a,X)-h_{\infty,n}(W,a,X))\middle|\, A=a\right]\\
&=0
\end{aligned}
\label{eq:parametric-null-local-DR}
\end{equation}
for $F_0$-almost every $a\in\mathcal A_{\delta_3}$, and
\begin{equation}
E_{0,n}\{h_{\infty,n}(W,a,X)-m_{\infty,n}(a)\}=0.
\label{eq:parametric-null-local-h-centering}
\end{equation}
For all sufficiently large $n$, both $\Gamma_{n,k,a}(t)$ and
$\Gamma_{0,a}(t)$ vanish when $a\in\mathcal A_0$ and
$t\notin\mathcal A_{\delta_3}$.  Thus the same calculation as in the proof of Lemma~\ref{lem:parametric-null-decomposition} yields
\[
\begin{aligned}
&\hat\theta_{n,k}^{DB}(a)-P_{n,k}\{\Gamma_{n,k,a}(A)\theta_{0,n}(A)\}\\
&=P_{n,k}\{\Gamma_{0,a}(A)\zeta_{\infty,n}(O)\}\\
&\quad+P_{n,k}\left[\int\Gamma_{0,a}(t)
\{h_{\infty,n}(W,t,X)-m_{\infty,n}(t)\}\,dF_0(t)\right]
+\sum_{j=2}^5R_{n,k,j}^{\mathcal M}(a).
\end{aligned}
\]
Add $P_{n,k}[\Gamma_{n,k,a}(A)\{\theta_{0,n}(A)-m(A;\hat\beta_{n,k})\}]$
and split $\Gamma_{n,k,a}=\Gamma_{0,a}+(\Gamma_{n,k,a}-\Gamma_{0,a})$.
The resulting terms are $L_{n,k,\mathcal M}$ and
$\widetilde R_{n,k,6}^{\mathcal M}$ by
\eqref{eq:parametric-null-leading} and
\eqref{eq:parametric-null-local-R6}.  This proves the claimed decomposition.

For $j=2,3,5$, apply the respective calculations in the proof of
Lemma~\ref{lem:parametric-null-remainder-rates}, replacing
$P_0,E_0,\xi_\infty,h_\infty,\theta_0$ by
$P_{0,n},E_{0,n},\xi_{\infty,n},h_{\infty,n},\theta_{0,n}$.
The kernel moments and their population matrices are unchanged, since the marginal distribution of $A$ is still $F_0$.  Conditions
\ref{assm:parametric-null-triangular-inherited} and
\ref{assm:parametric-null-triangular-uniform} indicate that the same conditional
rates and common bounds are used in these calculations.  Specifically, with
$g_{n,k}$, $\alpha_{n,k,h,j}$, $\alpha_{n,k,b,j}$, and
$f_{n,k,\ell,j,a}^{\circ}$ read under these substitutions, they yield
\[
E_{0,n}\left[\|(P_{n,k}-P_{0,n})g_{n,k}\|_w^2
\mid\trainingsigma_k\right]=o_{P_{0,n}}(n_k^{-1}),
\]
\[
\max_{j\le1}\int|\alpha_{n,k,h,j}(a)|^4w(a)\,da
+\max_{j\le2}\int|\alpha_{n,k,b,j}(a)|^4w(a)\,da
=O_{P_{0,n}}\{(nh_n)^{-2}\},
\]
and, for $\ell=h_n$, $j\in\{0,1\}$, or $\ell=b_n$, $j\in\{0,1,2\}$,
\[
\begin{aligned}
&E_{0,n}\left[\int\left|(P_{n,k}-P_{0,n})
\{p_{\ell,a,j}(A)(\hat\vartheta_{n,k}-\theta_{0,n}(A))\}
\right|^4w(a)\,da\,\middle|\,\trainingsigma_k\right]\\
&\qquad\le C\{n_k^{-3}\ell^{-3}+n_k^{-2}\ell^{-2}\},\\
&E_{0,n}\left[
\left\|\frac1{n_k^2}\sum_{r,s\in I_k}
f_{n,k,\ell,j,\cdot}^{\circ}(O_r,O_s)\right\|_w^2
\,\middle|\,\trainingsigma_k\right]\le\frac C{n_k^2\ell}.
\end{aligned}
\]
The fourth moment of the integrated outcome-bridge factor is
$O_{P_{0,n}}(n_k^{-2})$, as in the proof of
Lemma~\ref{lem:parametric-null-remainder-rates}. 
Similarly, Markov's
inequality, H\"older's inequality for the coefficient products,
and $n_k\asymp n$, $b_n\asymp h_n$ prove the three claimed rates.
The two kernel integral bounds in the same proof give, for every $k$,
\[
\|R_{n,k,4}^{\mathcal M}\|_w
\le C\left\{\int_{\mathcal A_{\delta_3}}
\Delta_{n,k}(a)^2\,dF_0(a)\right\}^{1/2}
=o_{P_{0,n}}(\varrho_n).
\]

For the remaining term, Lemma~\ref{lem:parametric-null-local-projection}
and boundedness of the first derivative of $m$ give
\[
\sup_{a\in\mathcal A}|\theta_{0,n}(a)-m(a;\hat\beta_{n,k})|
\le\sup_{a\in\mathcal A}|\theta_{0,n}(a)-m(a;\beta_n^\dagger)|
+C\|\hat\beta_{n,k}-\beta_n^\dagger\|
=O_{P_{0,n}}(\varrho_n).
\]
The last calculation for $R_{n,k,6}^{\mathcal M}$ in the proof of
Lemma~\ref{lem:parametric-null-remainder-rates} uses only the empirical
treatment distribution and bounds deterministic in the function $v$.
It therefore gives
\[
\sup_{\|v\|_\infty\le C}
\left\|P_{n,k}\{(\Gamma_{n,k,\cdot}-\Gamma_{0,\cdot})v\}\right\|_w
=O_{P_{0,n}}\{(nh_n)^{-1/2}\}.
\]
On events where the preceding supremum norm divided by $\varrho_n$ is
bounded, substitute
$v=\{\theta_{0,n}-m(\cdot;\hat\beta_{n,k})\}/\varrho_n$.  Then
\[
\|\widetilde R_{n,k,6}^{\mathcal M}\|_w
=O_{P_{0,n}}\{\varrho_n(nh_n)^{-1/2}\}
=O_{P_{0,n}}(n^{-1}h_n^{-3/4})=o_{P_{0,n}}(n^{-1/2}),
\]
since $nh_n^2\to\infty$ implies $nh_n^{3/2}\to\infty$.
\end{proof}

\begin{lemma}
\label{lem:parametric-null-local-oracle}
Under Assumptions~\ref{assm:parametric-null-triangular}
and~\ref{assm:parametric-null-local},
\[
\frac{T_{n,\mathcal M,0}^{DB}-b_{0h,n}^{DB}
-c^2\|g_{\mathcal M}^{\perp}\|_w^2\sum_{k=1}^K(n_k/n)^2}
{\left\{V_{\tau,n}^{DB}\sum_{k=1}^K(n_k/n)^2\right\}^{1/2}}
\rightsquigarrow N(0,1),
\]
and
\[
\max_{1\le k\le K}\|L_{n,k,\mathcal M}\|_w
=O_{P_{0,n}}\{(nh_n)^{-1/2}\}.
\]
\end{lemma}

\begin{proof}
Use $Y_{n,k},D_{n,k},U_{n,k}$ from the proof of
Proposition~\ref{prop:parametric-null-statistic-decomposition}, with
\[
(\zeta_\infty,h_\infty,m_\infty,\beta_0)
\mapsto(\zeta_{\infty,n},h_{\infty,n},m_{\infty,n},\beta_n^\dagger).
\]
Equation~\eqref{eq:parametric-null-leading} gives
\begin{equation}
\begin{aligned}
L_{n,k,\mathcal M}(a)
=Y_{n,k}(a)+U_{n,k}(a)-D_{n,k}(a)+P_{n,k}\left[
\Gamma_{0,a}(A)\{\theta_{0,n}(A)-m(A;\beta_n^\dagger)\}\right].
\end{aligned}
\label{eq:parametric-null-local-oracle-decomposition}
\end{equation}
The term involving $\theta_{0,n}-m(\cdot;\beta_n^\dagger)$ in
\eqref{eq:parametric-null-local-oracle-decomposition} satisfies
\begin{equation}
\begin{aligned}
&P_{n,k}\left[
\Gamma_{0,\cdot}(A)\{\theta_{0,n}(A)-m(A;\beta_n^\dagger)\}\right]\\
&\quad=c\varrho_nP_0\{\Gamma_{0,\cdot}(A)g_{\mathcal M}^{\perp}(A)\}+c\varrho_n(P_{n,k}-P_0)
\{\Gamma_{0,\cdot}(A)g_{\mathcal M}^{\perp}(A)\}\\
&\qquad+P_{n,k}\left[\Gamma_{0,\cdot}(A)
\{\theta_{0,n}(A)-m(A;\beta_n^\dagger)
-c\varrho_ng_{\mathcal M}^{\perp}(A)\}\right].
\end{aligned}
\label{eq:parametric-null-local-signal-decomposition}
\end{equation}

Recall that Condition~\ref{assm:folds} gives $K<\infty$ and
$n_k\ge\underline\pi n$.  
Lemma~\ref{lem:parametric-null-kernel},
\eqref{eq:parametric-null-local-DR}, and
condition~\ref{assm:parametric-null-triangular-uniform} imply
\[
\begin{aligned}
&E_{0,n}\{\|Y_{n,k}\|_w^2\mid I_1,\ldots,I_K\}=\frac1{n_k}\int_{\mathcal A_0}\int
\Gamma_{0,a}(t)^2\sigma_{0,n}^2(t)\,dF_0(t)w(a)\,da
\le\frac C{n_kh_n},\\
&E_{0,n}\left[
\int_{\mathcal A_0}\{P_{n,k}|\Gamma_{0,a}|\}^2w(a)\,da
\,\middle|\,I_1,\ldots,I_K\right]\\
&\quad=\frac1{n_k}\int_{\mathcal A_0}P_0\Gamma_{0,a}^2w(a)\,da
+\frac{n_k-1}{n_k}\int_{\mathcal A_0}
\{P_0|\Gamma_{0,a}|\}^2w(a)\,da\\
&\quad\le C\{1+(n_kh_n)^{-1}\}.
\end{aligned}
\]
Consequently, uniformly over $k$, we get
\begin{equation}
\|Y_{n,k}\|_w=O_{P_{0,n}}\{(n_kh_n)^{-1/2}\},
\qquad
\int_{\mathcal A_0}\{P_{n,k}|\Gamma_{0,a}|\}^2w(a)\,da
=O_{P_{0,n}}(1).
\label{eq:parametric-null-local-Y-bound}
\end{equation}

In the calculation establishing
\eqref{eq:parametric-null-oracle-score-moment} in the proof of
Lemma~\ref{lem:parametric-null-oracle}, make the substitutions
\[
(E_0,\zeta_\infty,\sigma_0^2,\dot m_0)
\mapsto\left(E_{0,n},\zeta_{\infty,n},\sigma_{0,n}^2,
\frac{\partial m(\cdot;\beta_n^\dagger)}{\partial\beta}\right).
\]
The derivative is deterministic and uniformly bounded by
Lemma~\ref{lem:parametric-null-local-projection} and
condition~\ref{assm:parametric-null-family}.
Equation~\eqref{eq:parametric-null-local-DR} gives
$E_{0,n}(\zeta_{\infty,n}\mid A)=0$ on
$\mathcal A_{\delta_3}$.  Thus
\[
\begin{aligned}
&E_{0,n}\Bigg[
\left\|\left\langle Y_{n,k},
P_{n,k}\left\{\Gamma_{0,\cdot}(A)
\frac{\partial m(A;\beta_n^\dagger)}{\partial\beta}\right\}
\right\rangle_w\right\|^2
\,\Biggm|\,I_1,\ldots,I_K\Bigg]\\
&\quad\le C\{n_k^{-3}h_n^{-2}+n_k^{-2}h_n^{-1}+n_k^{-1}\}
\le Cn_k^{-1}.
\end{aligned}
\]
For the expansion of $D_{n,k}$, the bounded second derivatives give
\[
\begin{aligned}
\left\|D_{n,k}
-\left[P_{n,k}\left\{\Gamma_{0,\cdot}(A)
\frac{\partial m(A;\beta_n^\dagger)}{\partial\beta}\right\}\right]^\top
(\hat\beta_{n,k}-\beta_n^\dagger)\right\|_w\leq C\|\hat\beta_{n,k}-\beta_n^\dagger\|^2
\left\|P_{n,k}|\Gamma_{0,\cdot}|\right\|_w
\end{aligned}
\]
on an event whose probability tends to one.
By \eqref{eq:parametric-null-local-Y-bound},
\eqref{eq:parametric-null-oracle-score-moment} under the stated
substitutions, Lemma~\ref{lem:parametric-null-local-projection},
and the Cauchy--Schwarz inequality, we have
\begin{equation}
\begin{aligned}
\|D_{n,k}\|_w
&\le C\|\hat\beta_{n,k}-\beta_n^\dagger\|
\left\|P_{n,k}|\Gamma_{0,\cdot}|\right\|_w
=o_{P_{0,n}}(\varrho_n),\\
n_k\sqrt{h_n}\|D_{n,k}\|_w^2
&=o_{P_{0,n}}(n\sqrt{h_n}\varrho_n^2)=o_{P_{0,n}}(1),\\
n_k\sqrt{h_n}|\langle Y_{n,k},D_{n,k}\rangle_w|
&\le n_k\sqrt{h_n}\|\hat\beta_{n,k}-\beta_n^\dagger\|\\
&\qquad\times\left\|\left\langle Y_{n,k},
P_{n,k}\left\{\Gamma_{0,\cdot}(A)
\frac{\partial m(A;\beta_n^\dagger)}{\partial\beta}\right\}
\right\rangle_w\right\|\\
&\quad+Cn_k\sqrt{h_n}\|\hat\beta_{n,k}-\beta_n^\dagger\|^2
\|Y_{n,k}\|_w\left\|P_{n,k}|\Gamma_{0,\cdot}|\right\|_w\\
&=o_{P_{0,n}}(\sqrt{nh_n}\varrho_n)
+o_{P_{0,n}}(\sqrt n\varrho_n^2)\\
&=o_{P_{0,n}}(h_n^{1/4})
+o_{P_{0,n}}\{(nh_n)^{-1/2}\},
\end{aligned}
\label{eq:parametric-null-local-D-rates}
\end{equation}
which is $o_{P_{0,n}}(1)$.
We note that all bounds in \eqref{eq:parametric-null-local-D-rates} are uniform over the fixed number of folds.

For $U_{n,k}$, use the calculation of $E_0\|U_{n,k}\|_w^2$ and
the diagonal and off-diagonal calculations following
\eqref{eq:parametric-null-oracle-U-means} in the proof of
Lemma~\ref{lem:parametric-null-oracle}, with
\[
(E_0,\zeta_\infty,h_\infty,m_\infty)
\mapsto(E_{0,n},\zeta_{\infty,n},h_{\infty,n},m_{\infty,n}).
\]
Equations~\eqref{eq:parametric-null-local-DR} and
\eqref{eq:parametric-null-local-h-centering} give both conditional
means equal to zero in \eqref{eq:parametric-null-oracle-U-means}
after substitution.  Using
condition~\ref{assm:parametric-null-triangular-uniform},
\eqref{eq:parametric-null-local-D-rates}, and
the Cauchy--Schwarz inequality gives
\begin{equation}
\begin{aligned}
E_{0,n}\{\|U_{n,k}\|_w^2\mid I_1,\ldots,I_K\}&\le Cn_k^{-1},\\
E_{0,n}\{|\langle Y_{n,k},U_{n,k}\rangle_w|^2
\mid I_1,\ldots,I_K\}&\le Cn_k^{-2},\\
\|U_{n,k}\|_w^2&=O_{P_{0,n}}(n_k^{-1}),\\
|\langle Y_{n,k},U_{n,k}\rangle_w|&=O_{P_{0,n}}(n_k^{-1}),\\
n_k\sqrt{h_n}\{\|U_{n,k}\|_w^2
+|\langle Y_{n,k},U_{n,k}\rangle_w|\}
&=O_{P_{0,n}}(\sqrt{h_n})=o_{P_{0,n}}(1),\\
n_k\sqrt{h_n}|\langle D_{n,k},U_{n,k}\rangle_w|
&\le\{n_k\sqrt{h_n}\|D_{n,k}\|_w^2\}^{1/2}
\{n_k\sqrt{h_n}\|U_{n,k}\|_w^2\}^{1/2}\\
&=o_{P_{0,n}}(1).
\end{aligned}
\label{eq:parametric-null-local-U-rates}
\end{equation}

Since $g_{\mathcal M}^{\perp}$ is bounded, we get
\[
\begin{aligned}
&E_{0,n}\left[
\left\|(P_{n,k}-P_0)
\{\Gamma_{0,\cdot}(A)g_{\mathcal M}^{\perp}(A)\}\right\|_w^2
\,\middle|\,I_1,\ldots,I_K\right]\\
&\quad=\frac1{n_k}\int_{\mathcal A_0}
\left[P_0\{\Gamma_{0,a}^2(g_{\mathcal M}^{\perp})^2\}
-\{P_0(\Gamma_{0,a}g_{\mathcal M}^{\perp})\}^2\right]w(a)\,da
\le\frac C{n_kh_n}.
\end{aligned}
\]
Hence
\begin{equation}
\left\|c\varrho_n(P_{n,k}-P_0)
\{\Gamma_{0,\cdot}(A)g_{\mathcal M}^{\perp}(A)\}\right\|_w
=O_{P_{0,n}}\{\varrho_n(n_kh_n)^{-1/2}\}.
\label{eq:parametric-null-local-empirical-signal-rate}
\end{equation}
Equations~\eqref{eq:parametric-null-local-population-residual} and
\eqref{eq:parametric-null-local-Y-bound} imply
\begin{equation}
\begin{aligned}
&\left\|P_{n,k}\left[\Gamma_{0,\cdot}(A)
\{\theta_{0,n}(A)-m(A;\beta_n^\dagger)
-c\varrho_ng_{\mathcal M}^{\perp}(A)\}\right]\right\|_w\\
&\quad\le
\sup_{a\in\mathcal A}|\theta_{0,n}(a)-m(a;\beta_n^\dagger)
-c\varrho_ng_{\mathcal M}^{\perp}(a)|
\left\|P_{n,k}|\Gamma_{0,\cdot}|\right\|_w
=O_{P_{0,n}}(\varrho_n^2).
\end{aligned}
\label{eq:parametric-null-local-Taylor-signal-rate}
\end{equation}
Using the Cauchy--Schwarz inequality with
\eqref{eq:parametric-null-local-Y-bound} and
\eqref{eq:parametric-null-local-empirical-signal-rate} gives
\begin{equation}
\begin{aligned}
&n_k\sqrt{h_n}\left|
\left\langle Y_{n,k},c\varrho_n(P_{n,k}-P_0)
\{\Gamma_{0,\cdot}(A)g_{\mathcal M}^{\perp}(A)\}
\right\rangle_w\right|\\
&\quad\le n_k\sqrt{h_n}\|Y_{n,k}\|_w
\left\|c\varrho_n(P_{n,k}-P_0)
\{\Gamma_{0,\cdot}(A)g_{\mathcal M}^{\perp}(A)\}\right\|_w\\
&\quad=O_{P_{0,n}}(\varrho_nh_n^{-1/2})
=O_{P_{0,n}}\{(nh_n^{3/2})^{-1/2}\},
\end{aligned}
\label{eq:parametric-null-local-empirical-signal-Y-rate}
\end{equation}
which is $o_{P_{0,n}}(1)$.
The Cauchy--Schwarz inequality applied to
\eqref{eq:parametric-null-local-Y-bound} and
\eqref{eq:parametric-null-local-Taylor-signal-rate} also gives
\begin{equation}
\begin{aligned}
&n_k\sqrt{h_n}\left|\left\langle Y_{n,k},
P_{n,k}\left[\Gamma_{0,\cdot}(A)
\{\theta_{0,n}(A)-m(A;\beta_n^\dagger)
-c\varrho_ng_{\mathcal M}^{\perp}(A)\}\right]\right\rangle_w\right|\\
&\quad\le n_k\sqrt{h_n}\|Y_{n,k}\|_w
\left\|P_{n,k}\left[\Gamma_{0,\cdot}(A)
\{\theta_{0,n}(A)-m(A;\beta_n^\dagger)
-c\varrho_ng_{\mathcal M}^{\perp}(A)\}\right]\right\|_w\\
&\quad=O_{P_{0,n}}(\sqrt n\varrho_n^2)
=O_{P_{0,n}}\{(nh_n)^{-1/2}\}=o_{P_{0,n}}(1).
\end{aligned}
\label{eq:parametric-null-local-Taylor-Y-rate}
\end{equation}
The limits in
\eqref{eq:parametric-null-local-empirical-signal-Y-rate} and
\eqref{eq:parametric-null-local-Taylor-Y-rate} use
$nh_n^{3/2},nh_n\to\infty$, which follow from
condition~\ref{assm:parametric-null-bandwidth} and $h_n\to0$.

For the term involving $P_0(\Gamma_{0,a}g_{\mathcal M}^{\perp})$,
the definition of $Y_{n,k}$ yields
\[
\begin{aligned}
\left\langle Y_{n,k},
P_0\{\Gamma_{0,\cdot}(A)g_{\mathcal M}^{\perp}(A)\}\right\rangle_w=P_{n,k}\left[\zeta_{\infty,n}(O)
\int_{\mathcal A_0}\Gamma_{0,a}(A)
P_0\{\Gamma_{0,a}(A)g_{\mathcal M}^{\perp}(A)\}w(a)\,da\right].
\end{aligned}
\]
Lemma~\ref{lem:parametric-null-kernel} implies
\[
\begin{aligned}
\sup_{a\in\mathcal A_0}
|P_0\{\Gamma_{0,a}(A)g_{\mathcal M}^{\perp}(A)\}|
&\le\|g_{\mathcal M}^{\perp}\|_\infty
\sup_{a\in\mathcal A_0}\int|\Gamma_{0,a}(t)|\,dF_0(t)
\le C,\\
\left|\int_{\mathcal A_0}\Gamma_{0,a}(s)
P_0\{\Gamma_{0,a}(A)g_{\mathcal M}^{\perp}(A)\}w(a)\,da\right|
&\le C\int_{\mathcal A_0}|\Gamma_{0,a}(s)|w(a)\,da\le C I(s\in\mathcal A_{\delta_3}).
\end{aligned}
\]
By \eqref{eq:parametric-null-local-DR}, we have
\[
\begin{aligned}
&E_{0,n}\left[
\left\langle Y_{n,k},P_0\{\Gamma_{0,\cdot}(A)
g_{\mathcal M}^{\perp}(A)\}\right\rangle_w^2
\,\middle|\,I_1,\ldots,I_K\right]\\
&\quad=\frac1{n_k}\int\sigma_{0,n}^2(s)
\left\{\int_{\mathcal A_0}\Gamma_{0,a}(s)
P_0\{\Gamma_{0,a}(A)g_{\mathcal M}^{\perp}(A)\}w(a)\,da\right\}^2
\,dF_0(s)
\le\frac C{n_k}.
\end{aligned}
\]
Since $n_kh_n\varrho_n^2=(n_k/n)\sqrt{h_n}$, we get
\begin{equation}
\begin{aligned}
n_k\sqrt{h_n}\varrho_n
\left\langle Y_{n,k},P_0\{\Gamma_{0,\cdot}(A)
g_{\mathcal M}^{\perp}(A)\}\right\rangle_w=O_{P_{0,n}}\{(n_k/n)^{1/2}h_n^{1/4}\}
=o_{P_{0,n}}(1).
\end{aligned}
\label{eq:parametric-null-local-signal-Y-rate}
\end{equation}
Equations~\eqref{eq:parametric-null-local-oracle-decomposition} and
\eqref{eq:parametric-null-local-signal-decomposition} imply
\begin{equation}
\begin{aligned}
&n_k\sqrt{h_n}\left\|L_{n,k,\mathcal M}-Y_{n,k}
-c\varrho_nP_0\{\Gamma_{0,\cdot}(A)g_{\mathcal M}^{\perp}(A)\}
\right\|_w^2\\
&\quad\le 4n_k\sqrt{h_n}\Bigg[\|D_{n,k}\|_w^2+\|U_{n,k}\|_w^2+\left\|c\varrho_n(P_{n,k}-P_0)
\{\Gamma_{0,\cdot}(A)g_{\mathcal M}^{\perp}(A)\}\right\|_w^2\\
&\qquad\qquad\qquad\quad+\left\|P_{n,k}\left[\Gamma_{0,\cdot}(A)
\{\theta_{0,n}(A)-m(A;\beta_n^\dagger)
-c\varrho_ng_{\mathcal M}^{\perp}(A)\}\right]\right\|_w^2\Bigg]\\
&\quad=o_{P_{0,n}}(1)+O_{P_{0,n}}(\sqrt{h_n})
+O_{P_{0,n}}\{(nh_n)^{-1}\}+O_{P_{0,n}}(\varrho_n^2),
\end{aligned}
\label{eq:parametric-null-local-oracle-small-square}
\end{equation}
which is $o_{P_{0,n}}(1)$,
by \eqref{eq:parametric-null-local-D-rates},
\eqref{eq:parametric-null-local-U-rates},
\eqref{eq:parametric-null-local-empirical-signal-rate}, and
\eqref{eq:parametric-null-local-Taylor-signal-rate}.
Likewise, \eqref{eq:parametric-null-local-D-rates},
\eqref{eq:parametric-null-local-U-rates},
\eqref{eq:parametric-null-local-empirical-signal-Y-rate}, and
\eqref{eq:parametric-null-local-Taylor-Y-rate} imply
\begin{equation}
\begin{aligned}
&n_k\sqrt{h_n}\left|
\left\langle Y_{n,k},L_{n,k,\mathcal M}-Y_{n,k}
-c\varrho_nP_0\{\Gamma_{0,\cdot}(A)g_{\mathcal M}^{\perp}(A)\}
\right\rangle_w\right|\\
&~\le n_k\sqrt{h_n}\Bigg[
|\langle Y_{n,k},D_{n,k}\rangle_w|
+|\langle Y_{n,k},U_{n,k}\rangle_w|\\
&\quad+\left|\left\langle Y_{n,k},c\varrho_n(P_{n,k}-P_0)
\{\Gamma_{0,\cdot}(A)g_{\mathcal M}^{\perp}(A)\}\right\rangle_w\right|\\
&\quad+\left|\left\langle Y_{n,k},
P_{n,k}\left[\Gamma_{0,\cdot}(A)
\{\theta_{0,n}(A)-m(A;\beta_n^\dagger)
-c\varrho_ng_{\mathcal M}^{\perp}(A)\}\right]\right\rangle_w\right|
\Bigg],
\end{aligned}
\label{eq:parametric-null-local-oracle-small-Y}
\end{equation}
which is also $o_{P_{0,n}}(1)$.
The Cauchy--Schwarz inequality and
\eqref{eq:parametric-null-local-oracle-small-square} yields
\begin{equation}
\begin{aligned}
&n_k\sqrt{h_n}\Big|
\Big\langle c\varrho_nP_0\{\Gamma_{0,\cdot}(A)g_{\mathcal M}^{\perp}(A)\},L_{n,k,\mathcal M}-Y_{n,k}
-c\varrho_nP_0\{\Gamma_{0,\cdot}(A)g_{\mathcal M}^{\perp}(A)\}
\Big\rangle_w\Big|\\
&\quad\le
\left[c^2\frac{n_k}{n}
\left\|P_0\{\Gamma_{0,\cdot}(A)g_{\mathcal M}^{\perp}(A)\}\right\|_w^2
\right]^{1/2}\left[n_k\sqrt{h_n}
\left\|L_{n,k,\mathcal M}-Y_{n,k}
-c\varrho_nP_0\{\Gamma_{0,\cdot}(A)g_{\mathcal M}^{\perp}(A)\}
\right\|_w^2\right]^{1/2}\\
&\quad=o_{P_{0,n}}(1).
\end{aligned}
\label{eq:parametric-null-local-oracle-small-signal}
\end{equation}
Expanding the squared norm in
\eqref{eq:parametric-null-local-oracle-small-square} and using
\eqref{eq:parametric-null-local-oracle-small-Y},
\eqref{eq:parametric-null-local-oracle-small-signal}, and
\eqref{eq:parametric-null-local-signal-Y-rate} yields
\[
\begin{aligned}
n_k\sqrt{h_n}\|L_{n,k,\mathcal M}\|_w^2
=n_k\sqrt{h_n}\|Y_{n,k}\|_w^2+c^2\frac{n_k}{n}
\left\|P_0\{\Gamma_{0,\cdot}(A)g_{\mathcal M}^{\perp}(A)\}\right\|_w^2
+o_{P_{0,n}}(1).
\end{aligned}
\]
Since $K<\infty$ and $\sum_k n_k/n=1$, we get
\begin{equation}
\begin{aligned}
T_{n,\mathcal M,0}^{DB}
=\sum_{k=1}^K\frac{n_k}{n}\,
n_k\sqrt{h_n}\|Y_{n,k}\|_w^2+c^2\sum_{k=1}^K\left(\frac{n_k}{n}\right)^2
\left\|P_0\{\Gamma_{0,\cdot}(A)g_{\mathcal M}^{\perp}(A)\}\right\|_w^2
+o_{P_{0,n}}(1).
\end{aligned}
\label{eq:parametric-null-local-oracle-reduction}
\end{equation}
Use $H_{n,ij}$ from the proof of
Lemma~\ref{lem:parametric-null-oracle}, replacing $\zeta_\infty$
by $\zeta_{\infty,n}$.  Then
\begin{equation}
n_k\sqrt{h_n}\|Y_{n,k}\|_w^2
=\frac{\sqrt{h_n}}{n_k}\sum_{i\in I_k}H_{n,ii}
+\frac{\sqrt{h_n}}{n_k}
\sum_{\substack{i,j\in I_k\\i\ne j}}H_{n,ij}.
\label{eq:parametric-null-local-diagonal-offdiagonal}
\end{equation}
For independent observations $O_1,O_2$ with law $P_{0,n}$, recall that
\[
\begin{aligned}
H_{n,12}&=H_{n,21},\\
E_{0,n}(H_{n,12}\mid O_1)
&=\zeta_{\infty,n}(O_1)\int J_{n,h}(A_1,t)
E_{0,n}\{\zeta_{\infty,n}(O_2)\mid A_2=t\}\,dF_0(t)=0,
\end{aligned}
\]
by \eqref{eq:parametric-null-local-DR}.
Apply the diagonal calculation establishing
\eqref{eq:parametric-null-oracle-diagonal-control}, the variance
calculation using \eqref{eq:parametric-null-oracle-J-expansion},
and the calculation of
$E_0[\{E_0(H_{n,13}H_{n,23}\mid O_1,O_2)\}^2]$ in the proof of
Lemma~\ref{lem:parametric-null-oracle}, with
\[
(E_0,\zeta_\infty,\sigma_0^2,V_\tau^{DB})
\mapsto(E_{0,n},\zeta_{\infty,n},\sigma_{0,n}^2,V_{\tau,n}^{DB}).
\]
Condition~\ref{assm:parametric-null-triangular-uniform} bounds
$\sigma_{0,n}^2$, its derivative, and the conditional fourth moment
of $\zeta_{\infty,n}$ uniformly in $n$.  Therefore, we obtain
\begin{equation}
\begin{aligned}
&E_{0,n}\left[
\left\{\frac{\sqrt{h_n}}{n_k}\sum_{i\in I_k}H_{n,ii}
-b_{0h,n}^{DB}\right\}^2
\,\middle|\,I_1,\ldots,I_K\right]
\le C\{(n_kh_n)^{-1}+h_n\},\\
&\left|2h_nE_{0,n}H_{n,12}^2-V_{\tau,n}^{DB}\right|
\le C\{h_n+|\tau_n-\tau|\},
\qquad C^{-1}\le V_{\tau,n}^{DB}\le C,\\
&E_{0,n}H_{n,12}^4\le Ch_n^{-3},
\qquad
E_{0,n}\left[\{E_{0,n}(H_{n,13}H_{n,23}\mid O_1,O_2)\}^2\right]
\le Ch_n^{-1}.
\end{aligned}
\label{eq:parametric-null-local-oracle-moments}
\end{equation}
In particular, $E_{0,n}H_{n,12}^2\asymp h_n^{-1}$, and condition
(2.1) in Theorem~1 of \citet{Hall1984central}, with sample size
$n_k$, is verified by
\[
\frac{E_{0,n}[\{E_{0,n}(H_{n,13}H_{n,23}\mid O_1,O_2)\}^2]
+n_k^{-1}E_{0,n}H_{n,12}^4}
{(E_{0,n}H_{n,12}^2)^2}
\le C\{h_n+(n_kh_n)^{-1}\}\longrightarrow0.
\]
The proof of Theorem~1 of \citet{Hall1984central}
applies to independent observations with a common distribution within
each row: its moment calculations use $E_{0,n}$ throughout row $n$.
For every deterministic sequence of partitions satisfying
$\min_k n_k\ge\underline\pi n$, we obtain
\[
\frac{\displaystyle\sum_{\substack{i<j\\i,j\in I_k}}H_{n,ij}}
{\{n_k(n_k-1)E_{0,n}H_{n,12}^2/2\}^{1/2}}
\rightsquigarrow N(0,1).
\]
Moreover, we get
\[
\operatorname{Var}_{0,n}\left(
\frac{\sqrt{h_n}}{n_k}\sum_{\substack{i,j\in I_k\\i\ne j}}H_{n,ij}
\,\middle|\,I_1,\ldots,I_K\right)
=2(1-n_k^{-1})h_nE_{0,n}H_{n,12}^2
=V_{\tau,n}^{DB}+o(1),
\]
uniformly over the admissible partitions, by
\eqref{eq:parametric-null-local-oracle-moments}.

Apply the argument establishing
\eqref{eq:parametric-null-oracle-fold-cdf} in the proof of
Lemma~\ref{lem:parametric-null-oracle}, with
$(P_0,V_\tau^{DB})$ replaced by $(P_{0,n},V_{\tau,n}^{DB})$.
By the diagonal bound in
\eqref{eq:parametric-null-local-oracle-moments} and the decomposition
\eqref{eq:parametric-null-local-diagonal-offdiagonal}, we have that
\begin{equation}
\begin{aligned}
\sup_{\substack{I_1,\ldots,I_K\text{ a partition of }\{1,\ldots,n\}\\
\min_k n_k\ge\underline\pi n}}\max_{1\le k\le K}\sup_{t\in\mathbb R}\left|
P_{0,n}\left\{
\frac{n_k\sqrt{h_n}\|Y_{n,k}\|_w^2-b_{0h,n}^{DB}}
{\sqrt{V_{\tau,n}^{DB}}}
\le t\,\middle|\,I_1,\ldots,I_K\right\}-\Phi(t)
\right|
\end{aligned}
\label{eq:parametric-null-local-oracle-fold-cdf}
\end{equation}
converges to $0$.
Conditional on $I_1,\ldots,I_K$, the random variables
$n_k\sqrt{h_n}\|Y_{n,k}\|_w^2$, $1\le k\le K$, are independent.
Equation~\eqref{eq:parametric-null-oracle-weighted-cdf} yields,
after replacing $P_0$ by $P_{0,n}$ and $V_\tau^{DB}$ by
$V_{\tau,n}^{DB}$,
\[
\begin{aligned}
&\sup_{t\in\mathbb R}\Bigg|
P_{0,n}\Bigg\{
\frac{\sum_k(n_k/n)n_k\sqrt{h_n}\|Y_{n,k}\|_w^2-b_{0h,n}^{DB}}
{\{V_{\tau,n}^{DB}\sum_k(n_k/n)^2\}^{1/2}}
\le t\,\Biggm|\,I_1,\ldots,I_K\Bigg\}-\Phi(t)\Bigg|\\
&\quad\le\sum_{k=1}^K\sup_{t\in\mathbb R}
\left|P_{0,n}\left\{
\frac{n_k\sqrt{h_n}\|Y_{n,k}\|_w^2-b_{0h,n}^{DB}}
{\sqrt{V_{\tau,n}^{DB}}}
\le t\,\middle|\,I_1,\ldots,I_K\right\}-\Phi(t)\right|
\longrightarrow0,
\end{aligned}
\]
uniformly over the admissible partitions by
\eqref{eq:parametric-null-local-oracle-fold-cdf}.
Taking expectation over the partition proves
\begin{equation}
\frac{\sum_k(n_k/n)n_k\sqrt{h_n}\|Y_{n,k}\|_w^2-b_{0h,n}^{DB}}
{\{V_{\tau,n}^{DB}\sum_k(n_k/n)^2\}^{1/2}}
\rightsquigarrow N(0,1).
\label{eq:parametric-null-local-oracle-noise-limit}
\end{equation}

By \eqref{eq:parametric-null-local-smoothed-direction}, we get
\[
\begin{aligned}
\left|
\left\|P_0\{\Gamma_{0,\cdot}(A)g_{\mathcal M}^{\perp}(A)\}\right\|_w^2
-\|g_{\mathcal M}^{\perp}\|_w^2\right|
&\le
\left\|P_0\{\Gamma_{0,\cdot}(A)g_{\mathcal M}^{\perp}(A)\}
-g_{\mathcal M}^{\perp}\right\|_w\\
&\qquad\times\left(
\left\|P_0\{\Gamma_{0,\cdot}(A)g_{\mathcal M}^{\perp}(A)\}\right\|_w
+\|g_{\mathcal M}^{\perp}\|_w\right),
\end{aligned}
\]
which is $o(1)$.
Since $K^{-1}\le\sum_k(n_k/n)^2\le1$ and
$C^{-1}\le V_{\tau,n}^{DB}\le C$ by
\eqref{eq:parametric-null-local-oracle-moments},
\eqref{eq:parametric-null-local-oracle-reduction} and
\eqref{eq:parametric-null-local-oracle-noise-limit} prove the first
conclusion.
Finally, \eqref{eq:parametric-null-local-oracle-decomposition},
\eqref{eq:parametric-null-local-signal-decomposition},
\eqref{eq:parametric-null-local-Y-bound},
\eqref{eq:parametric-null-local-D-rates},
\eqref{eq:parametric-null-local-U-rates},
\eqref{eq:parametric-null-local-empirical-signal-rate}, and
\eqref{eq:parametric-null-local-Taylor-signal-rate} yield
\[
\begin{aligned}
\max_k\|L_{n,k,\mathcal M}\|_w
&\le\max_k\|Y_{n,k}\|_w+\max_k\|U_{n,k}\|_w
+\max_k\|D_{n,k}\|_w\\
&\quad+|c|\varrho_n
\left\|P_0\{\Gamma_{0,\cdot}(A)g_{\mathcal M}^{\perp}(A)\}\right\|_w\\
&\quad+\max_k\left\|c\varrho_n(P_{n,k}-P_0)
\{\Gamma_{0,\cdot}(A)g_{\mathcal M}^{\perp}(A)\}\right\|_w\\
&\quad+\max_k\left\|P_{n,k}\left[\Gamma_{0,\cdot}(A)
\{\theta_{0,n}(A)-m(A;\beta_n^\dagger)
-c\varrho_ng_{\mathcal M}^{\perp}(A)\}\right]\right\|_w\\
&=O_{P_{0,n}}\{(nh_n)^{-1/2}\}
+O_{P_{0,n}}(n^{-1/2})+o_{P_{0,n}}(\varrho_n)
+O(\varrho_n)\\
&\quad+O_{P_{0,n}}\{\varrho_n(nh_n)^{-1/2}\}
+O_{P_{0,n}}(\varrho_n^2)
=O_{P_{0,n}}\{(nh_n)^{-1/2}\},
\end{aligned}
\]
since $n_k\asymp n$ and
$\varrho_n/(nh_n)^{-1/2}=h_n^{1/4}\to0$.
\end{proof}

\begin{proof}[Proof of Theorem~\ref{thm:parametric-null-local-power}]
Lemma~\ref{lem:parametric-null-local-decomposition} gives the exact identity
\[
\begin{aligned}
&\widehat T_{n,\mathcal M}^{DB}-T_{n,\mathcal M,0}^{DB}
\\
&=\sum_{k=1}^K\frac{n_k}{n}\,n_k\sqrt{h_n}
\Bigg[2\left\langle L_{n,k,\mathcal M},
\sum_{j=2}^5R_{n,k,j}^{\mathcal M}
+\widetilde R_{n,k,6}^{\mathcal M}\right\rangle_w+\left\|
\sum_{j=2}^5R_{n,k,j}^{\mathcal M}
+\widetilde R_{n,k,6}^{\mathcal M}\right\|_w^2\Bigg].
\end{aligned}
\]
For $j\in\{2,3,5\}$, the Cauchy--Schwarz inequality and the rates in
Lemmas~\ref{lem:parametric-null-local-decomposition} and
\ref{lem:parametric-null-local-oracle} provide, uniformly in $k$,
\[
\begin{aligned}
n_k\sqrt{h_n}|\langle L_{n,k,\mathcal M},R_{n,k,2}^{\mathcal M}\rangle_w|
&=o_{P_{0,n}}(1),\\
n_k\sqrt{h_n}|\langle L_{n,k,\mathcal M},R_{n,k,3}^{\mathcal M}\rangle_w|
&=O_{P_{0,n}}\{(nh_n^2)^{-1/2}\},\\
n_k\sqrt{h_n}|\langle L_{n,k,\mathcal M},R_{n,k,5}^{\mathcal M}\rangle_w|
&=O_{P_{0,n}}\{(nh_n)^{-1/2}\},\\
n_k\sqrt{h_n}|\langle L_{n,k,\mathcal M},\widetilde R_{n,k,6}^{\mathcal M}\rangle_w|
&=O_{P_{0,n}}\{(nh_n^{3/2})^{-1/2}\}.
\end{aligned}
\]
For $R_{n,k,4}^{\mathcal M}$, apply the conditional second-moment and
empirical-coefficient calculations in the proof of
Proposition~\ref{prop:parametric-null-statistic-decomposition}, with
\[
(P_0,E_0,\zeta_\infty,\sigma_0^2)
\mapsto(P_{0,n},E_{0,n},\zeta_{\infty,n},\sigma_{0,n}^2).
\]
Equation~\eqref{eq:parametric-null-local-DR} and condition
\ref{assm:parametric-null-triangular-uniform} give, respectively, the
conditional mean zero and the same variance bound.  
Thus, we obtain
\[
 n_k\sqrt{h_n}\langle Y_{n,k},R_{n,k,4}^{\mathcal M}\rangle_w
 =o_{P_{0,n}}(1).
\]
Equations~\eqref{eq:parametric-null-local-oracle-decomposition}--
\eqref{eq:parametric-null-local-signal-decomposition} and
Lemma~\ref{lem:parametric-null-local-projection} give
\[
\|L_{n,k,\mathcal M}-Y_{n,k}\|_w=O_{P_{0,n}}(\varrho_n).
\]
Since $n\sqrt{h_n}\varrho_n^2=1$, we have
\[
 n_k\sqrt{h_n}\left|\langle L_{n,k,\mathcal M}-Y_{n,k},
 R_{n,k,4}^{\mathcal M}\rangle_w\right|
 \le n\sqrt{h_n}\,O_{P_{0,n}}(\varrho_n)o_{P_{0,n}}(\varrho_n)
 =o_{P_{0,n}}(1).
\]
The five statistic-scale squares, in the order
$R_{n,k,2}^{\mathcal M},R_{n,k,3}^{\mathcal M},
R_{n,k,4}^{\mathcal M},R_{n,k,5}^{\mathcal M},
\widetilde R_{n,k,6}^{\mathcal M}$, have bounds
\[
\begin{gathered}
o_{P_{0,n}}(\sqrt{h_n}),\quad
O_{P_{0,n}}\{(nh_n^{3/2})^{-1}\},\quad
o_{P_{0,n}}(1),\quad
O_{P_{0,n}}\{(n\sqrt{h_n})^{-1}\},\quad
O_{P_{0,n}}\{(nh_n)^{-1}\},
\end{gathered}
\]
where each bound tends to zero by condition~\ref{assm:parametric-null-bandwidth}.
Since $K$ is fixed and
$\sum_k n_k/n=1$, we get
\[
\widehat T_{n,\mathcal M}^{DB}-T_{n,\mathcal M,0}^{DB}
=o_{P_{0,n}}(1).
\]
Now, Lemma~\ref{lem:parametric-null-local-oracle} proves the standardized limit.

For the conditional power statement,
\eqref{eq:parametric-null-local-oracle-reduction} and
\eqref{eq:parametric-null-local-smoothed-direction} give
\begin{equation}
\begin{aligned}
&\widehat T_{n,\mathcal M}^{DB}
-\sum_{k=1}^K\frac{n_k}{n}\,n_k\sqrt{h_n}\|Y_{n,k}\|_w^2
-c^2\|g_{\mathcal M}^{\perp}\|_w^2\sum_{k=1}^K(n_k/n)^2\\
&\quad=\widehat T_{n,\mathcal M}^{DB}-T_{n,\mathcal M,0}^{DB}
+o_{P_{0,n}}(1)=o_{P_{0,n}}(1).
\end{aligned}
\label{eq:parametric-null-local-power-reduction}
\end{equation}
For every $\epsilon,\eta>0$, the tower property and Markov's
inequality imply
\begin{equation}
\begin{aligned}
&P_{0,n}\Bigg(
P_{0,n}\Bigg\{
\Bigg|\widehat T_{n,\mathcal M}^{DB}
-\sum_{k=1}^K\frac{n_k}{n}\,n_k\sqrt{h_n}\|Y_{n,k}\|_w^2
-c^2\|g_{\mathcal M}^{\perp}\|_w^2\sum_{k=1}^K(n_k/n)^2
\Bigg|>\epsilon\,\Biggm|\,I_1,\ldots,I_K
\Bigg\}>\eta\Bigg)\\
&\quad\le\eta^{-1}P_{0,n}\Bigg\{
\Bigg|\widehat T_{n,\mathcal M}^{DB}
-\sum_{k=1}^K\frac{n_k}{n}\,n_k\sqrt{h_n}\|Y_{n,k}\|_w^2
-c^2\|g_{\mathcal M}^{\perp}\|_w^2\sum_{k=1}^K(n_k/n)^2
\Bigg|>\epsilon\Bigg\}\longrightarrow0,
\end{aligned}
\label{eq:parametric-null-local-power-conditional-remainder}
\end{equation}
by \eqref{eq:parametric-null-local-power-reduction}.
Equation~\eqref{eq:parametric-null-local-oracle-moments} and
$\sum_k n_k/n=1$ implies
\[
\begin{gathered}
C^{-1}\le V_{\tau,n}^{DB}\le C,
\quad
K^{-1}\le\sum_{k=1}^K(n_k/n)^2\le1,\quad
\left\{V_{\tau,n}^{DB}\sum_{k=1}^K(n_k/n)^2\right\}^{1/2}
\ge(CK)^{-1/2}.
\end{gathered}
\]
For every $\epsilon>0$, applying
\eqref{eq:parametric-null-oracle-weighted-cdf} with
$P_0$ replaced by $P_{0,n}$ and $V_\tau^{DB}$ by
$V_{\tau,n}^{DB}$ at thresholds $t-\epsilon$ and $t+\epsilon$ yields
\begin{equation}
\begin{aligned}
&\sup_{t\in\mathbb R}\Bigg|P_{0,n}\Bigg\{
\frac{\widehat T_{n,\mathcal M}^{DB}-b_{0h,n}^{DB}
-c^2\|g_{\mathcal M}^{\perp}\|_w^2\sum_k(n_k/n)^2}
{\{V_{\tau,n}^{DB}\sum_k(n_k/n)^2\}^{1/2}}
\le t\,\Biggm|\,I_1,\ldots,I_K
\Bigg\}-\Phi(t)\Bigg|\\
&\quad\le
\sum_{k=1}^K\sup_{t\in\mathbb R}\left|P_{0,n}\left\{
\frac{n_k\sqrt{h_n}\|Y_{n,k}\|_w^2-b_{0h,n}^{DB}}
{\sqrt{V_{\tau,n}^{DB}}}
\le t\,\middle|\,I_1,\ldots,I_K
\right\}-\Phi(t)\right|\\
&\qquad+
P_{0,n}\Bigg\{
\Bigg|\widehat T_{n,\mathcal M}^{DB}
-\sum_{k=1}^K\frac{n_k}{n}\,n_k\sqrt{h_n}\|Y_{n,k}\|_w^2
-c^2\|g_{\mathcal M}^{\perp}\|_w^2\sum_k(n_k/n)^2
\Bigg|>\frac{\epsilon}{\sqrt{CK}}
\,\Biggm|\,I_1,\ldots,I_K\Bigg\}\\
&\qquad+\sup_{t\in\mathbb R}\{\Phi(t+\epsilon)-\Phi(t)\}\\
&\quad\le\frac{\epsilon}{\sqrt{2\pi}}+o_{P_{0,n}}(1).
\end{aligned}
\label{eq:parametric-null-local-power-cdf-bound}
\end{equation}
The bound in \eqref{eq:parametric-null-local-power-cdf-bound} follows from
\eqref{eq:parametric-null-local-oracle-fold-cdf},
\eqref{eq:parametric-null-local-power-conditional-remainder}, and
\[
\sup_{t\in\mathbb R}\{\Phi(t+\epsilon)-\Phi(t)\}
=\sup_{t\in\mathbb R}\int_t^{t+\epsilon}
\frac{e^{-x^2/2}}{\sqrt{2\pi}}\,dx
\le\frac{\epsilon}{\sqrt{2\pi}}.
\]
For every $\eta>0$, choose $0<\epsilon<\eta\sqrt{2\pi}/2$ in
\eqref{eq:parametric-null-local-power-cdf-bound}.  The probability
that the supremum in \eqref{eq:parametric-null-local-power-cdf-bound}
exceeds $\eta$ then tends to zero.
For each realized partition, take
\[
t=z_{1-\alpha_0}
-\frac{c^2\|g_{\mathcal M}^{\perp}\|_w^2
\sqrt{\sum_k(n_k/n)^2}}{\sqrt{V_{\tau,n}^{DB}}}.
\]
Equation~\eqref{eq:parametric-null-local-power-cdf-bound} therefore gives
\[
\begin{aligned}
&P_{0,n}\left\{
\widehat T_{n,\mathcal M}^{DB}>
 b_{0h,n}^{DB}
+z_{1-\alpha_0}\left\{V_{\tau,n}^{DB}
\sum_{k=1}^K(n_k/n)^2\right\}^{1/2}
\,\middle|\,I_1,\ldots,I_K\right\}\\
&\quad=1-\Phi\left(
 z_{1-\alpha_0}
-\frac{c^2\|g_{\mathcal M}^{\perp}\|_w^2
\sqrt{\sum_{k=1}^K(n_k/n)^2}}{\sqrt{V_{\tau,n}^{DB}}}
\right)+o_{P_{0,n}}(1).
\end{aligned}
\]
\end{proof}

\begin{proof}[Proof of Corollary~\ref{cor:parametric-null-polynomials-local-power}]
For the three families in Table~\ref{tab:parametric-null-polynomial-families}, note that
\begin{equation}
\begin{aligned}
\dot m_0(a)&=
\begin{cases}
1,&m(a;\beta)=\beta,\\
v_d(a),&m(a;\beta)=v_d(a)^\top\beta,\quad d\in\{1,2\},
\end{cases}\\
m(a;\beta)&=\dot m_0(a)^\top\beta,
\qquad
\frac{\partial m(a;\beta)}{\partial\beta}=\dot m_0(a),
\qquad
\frac{\partial^2m(a;\beta)}{\partial\beta\partial\beta^\top}=0.
\end{aligned}
\label{eq:parametric-null-polynomial-derivatives}
\end{equation}
Condition~\ref{assm:parametric-null-family} gives
$\sup_{a\in\mathcal A}\|\dot m_0(a)\|<\infty$.
For the constant family, $J_0=P_0(1)=1$.
For $d\in\{1,2\}$ and $u\in\mathbb R^{d+1}\setminus\{0\}$,
\[
\#\{a\in\mathbb R:u^\top v_d(a)=0\}\le d,
\]
and condition~\ref{assm:parametric-null-support} implies
\[
\begin{aligned}
u^\top J_0u
=\int_{\mathcal A}\{u^\top v_d(a)\}^2\,dF_0(a)\left\{\inf_{a\in\mathcal A_{\delta_3}}f_A(a)\right\}
\int_{\mathcal A_{\delta_3}}\{u^\top v_d(a)\}^2\,da>0.
\end{aligned}
\]
Here $\int_{\mathcal A_{\delta_3}}\{u^\top v_d(a)\}^2\,da>0$
because $\mathcal A_{\delta_3}$ contains an interval of positive
length and $u^\top v_d$ has at most $d$ zeros.  Thus, for every family and every $\epsilon>0$, one has
\begin{equation}
\begin{aligned}
\inf_{\substack{\beta\in\mathcal B\\\|\beta-\beta_0\|\ge\epsilon}}
P_0[\{m(A;\beta)-m(A;\beta_0)\}^2]=\inf_{\substack{\beta\in\mathcal B\\\|\beta-\beta_0\|\ge\epsilon}}
(\beta-\beta_0)^\top J_0(\beta-\beta_0)
\ge\lambda_{\min}(J_0),
\end{aligned}
\label{eq:parametric-null-polynomial-separation}
\end{equation}
which is strictly positive.
By the definition of $g_{\mathcal M}^{\perp}(a)$, we get
\[
\begin{aligned}
P_0\{\dot m_0(A)g_{\mathcal M}^{\perp}(A)\}
&=P_0\{\dot m_0(A)g(A)\}
-J_0J_0^{-1}P_0\{\dot m_0(A)g(A)\}=0,\\
\|g_{\mathcal M}^{\perp}\|_\infty
&\le\|g\|_\infty
+\sup_{a\in\mathcal A}\|\dot m_0(a)\|
\|J_0^{-1}P_0\{\dot m_0(A)g(A)\}\|<\infty.
\end{aligned}
\]
Equations~\eqref{eq:parametric-null-local-alternative} and
\eqref{eq:parametric-null-polynomial-derivatives} give, for every
$\beta\in\mathcal B$,
\begin{equation}
\begin{aligned}
&P_0[\{\theta_{0,n}(A)-m(A;\beta)\}^2]\\
&\quad=c^2\varrho_n^2P_0\{g_{\mathcal M}^{\perp}(A)^2\}\\
&\qquad+
\left[\beta-\beta_0-c\varrho_nJ_0^{-1}P_0\{\dot m_0(A)g(A)\}\right]^\top
J_0\\
&\hspace{5em}\times
\left[\beta-\beta_0-c\varrho_nJ_0^{-1}P_0\{\dot m_0(A)g(A)\}\right].
\end{aligned}
\label{eq:parametric-null-polynomial-population-loss}
\end{equation}
Since $\beta_0\in\operatorname{int}(\mathcal B)$ and
\[
\|c\varrho_nJ_0^{-1}P_0\{\dot m_0(A)g(A)\}\|=O(\varrho_n)\longrightarrow0,
\]
$\beta_0+c\varrho_nJ_0^{-1}P_0\{\dot m_0(A)g(A)\}$ belongs to
$\operatorname{int}(\mathcal B)$ for all sufficiently large $n$.
Equation~\eqref{eq:parametric-null-polynomial-population-loss} and
$\lambda_{\min}(J_0)>0$ therefore yields the unique minimizer in
\eqref{eq:parametric-null-local-population-beta}:
\begin{equation}
\begin{aligned}
&\beta_n^\dagger
=\beta_0+c\varrho_nJ_0^{-1}P_0\{\dot m_0(A)g(A)\},\\
&\theta_{0,n}(a)-m(a;\beta_n^\dagger)
=c\varrho_ng_{\mathcal M}^{\perp}(a),\qquad a\in\mathcal A,\\
&P_0\left[\dot m_0(A)
\{\theta_{0,n}(A)-m(A;\beta_n^\dagger)\}\right]=0.
\end{aligned}
\label{eq:parametric-null-polynomial-population-projection}
\end{equation}
For $m(a;\beta)=\beta$, substitution in
\eqref{eq:parametric-null-polynomial-population-projection} implies
\[
\beta_n^\dagger=\beta_0+c\varrho_nP_0g(A),
\qquad g_{\mathcal M}^{\perp}(a)=g(a)-P_0g(A).
\]
For $m(a;\beta)=v_d(a)^\top\beta$, $d\in\{1,2\}$, this yields
\[
\begin{aligned}
\beta_n^\dagger
&=\beta_0+c\varrho_n
\{P_0[v_d(A)v_d(A)^\top]\}^{-1}P_0\{v_d(A)g(A)\},\\
g_{\mathcal M}^{\perp}(a)
&=g(a)-v_d(a)^\top
\{P_0[v_d(A)v_d(A)^\top]\}^{-1}P_0\{v_d(A)g(A)\}.
\end{aligned}
\]

Apply the conditional moment calculation from the proof of
Lemma~\ref{lem:parametric-null-beta-rate} used to obtain
\eqref{eq:parametric-null-local-weighted-mean}, with the replacement:
\[
(P_0,E_0,\theta_0,q_0,h_0,v)
\mapsto(P_{0,n},E_{0,n},\theta_{0,n},q_{0,n},h_{0,n},\dot m_0).
\]
The same calculation uses condition~\ref{assm:parametric-null-score-moment},
$\sup_a\|\dot m_0(a)\|<\infty$,
$K<\infty$, $n_k\asymp n$, and
\[
\sup_{a\in\mathcal A}|\theta_{0,n}(a)|
\le\sup_{a\in\mathcal A}|\dot m_0(a)^\top\beta_0|
+|c|\varrho_n\|g\|_\infty=O(1).
\]
We obtain
\begin{equation}
\max_k\left\|
P_{n,k}[\dot m_0(A)\{\hat\xi_n^{(-k)}-\theta_{0,n}(A)\}]
-\int_{\mathcal A}\dot m_0(a)\Delta_{n,k}(a)\,dF_0(a)
\right\|=O_{P_{0,n}}(n^{-1/2}).
\label{eq:parametric-null-polynomial-weighted-mean}
\end{equation}
The exact residual in
\eqref{eq:parametric-null-polynomial-population-projection} yields
\[
\begin{aligned}
&E_{0,n}\Bigg[
\left\|(P_{n,k}-P_0)\left[
\dot m_0(A)\{\theta_{0,n}(A)-m(A;\beta_n^\dagger)\}\right]\right\|^2
\,\Biggm|\,I_1,\ldots,I_K\Bigg]\\
&\quad=\frac{c^2\varrho_n^2}{n_k}
P_0\{\|\dot m_0(A)\|^2g_{\mathcal M}^{\perp}(A)^2\}
\le C\varrho_n^2/n_k.
\end{aligned}
\]
Also, condition~\ref{assm:parametric-null-product} implies
\[
\max_k\left\|\int_{\mathcal A}
\dot m_0(a)\Delta_{n,k}(a)\,dF_0(a)\right\|
=o_{P_{0,n}}(\varrho_n).
\]
By \eqref{eq:parametric-null-polynomial-population-projection} and
\eqref{eq:parametric-null-polynomial-weighted-mean}, we have
\begin{equation}
\begin{aligned}
&P_{n,k}[\dot m_0(A)\{\hat\xi_n^{(-k)}-m(A;\beta_n^\dagger)\}]\\
&=P_{n,k}[\dot m_0(A)\{\hat\xi_n^{(-k)}-\theta_{0,n}(A)\}]\\
&\qquad+(P_{n,k}-P_0)
[\dot m_0(A)\{\theta_{0,n}(A)-m(A;\beta_n^\dagger)\}],\\
&\max_k\left\|P_{n,k}[\dot m_0(A)
\{\hat\xi_n^{(-k)}-m(A;\beta_n^\dagger)\}]\right\|\\
&=O_{P_{0,n}}(n^{-1/2})
+o_{P_{0,n}}(\varrho_n)
+O_{P_{0,n}}(\varrho_n n^{-1/2})\\
&\quad=O_{P_{0,n}}(n^{-1/2})+o_{P_{0,n}}(\varrho_n).
\end{aligned}
\label{eq:parametric-null-polynomial-score-rate}
\end{equation}

For $d\in\{1,2\}$, boundedness of $v_d$ implies
\[
\begin{aligned}
E_{0,n}\left[
\|P_{n,k}\{v_d(A)v_d(A)^\top\}-J_0\|^2
\,\middle|\,I_1,\ldots,I_K\right]\le\frac1{n_k}\sum_{r,s=0}^d
\operatorname{Var}_0(A^{r+s})\le Cn_k^{-1}.
\end{aligned}
\]
Hence, we obtain
\begin{equation}
\begin{aligned}
&P_{0,n}\left\{
\max_k\|P_{n,k}(v_dv_d^\top)-J_0\|
>\frac{\lambda_{\min}(J_0)}2\right\}\le\frac Cn,\\
&\max_k\|\{P_{n,k}(v_dv_d^\top)\}^{-1}\|
\le\frac2{\lambda_{\min}(J_0)}
\quad\text{if }\max_k\|P_{n,k}(v_dv_d^\top)-J_0\|
\le\frac{\lambda_{\min}(J_0)}2.
\end{aligned}
\label{eq:parametric-null-polynomial-Gram-bound}
\end{equation}
For the constant family, we have $P_{n,k}(1)=1$ and
$\hat c_{n,k}-\beta_n^\dagger
=P_{n,k}\{\hat\xi_n^{(-k)}-\beta_n^\dagger\}$.
For $d\in\{1,2\}$, whenever
$\max_k\|P_{n,k}(v_dv_d^\top)-J_0\|
\le\lambda_{\min}(J_0)/2$, the definition of
$\hat\beta_{n,k,d}$ gives
\[
\begin{aligned}
\hat\beta_{n,k,d}-\beta_n^\dagger
=\{P_{n,k}[v_d(A)v_d(A)^\top]\}^{-1} P_{n,k}\left[
v_d(A)\{\hat\xi_n^{(-k)}-v_d(A)^\top\beta_n^\dagger\}\right].
\end{aligned}
\]
Equations~\eqref{eq:parametric-null-polynomial-score-rate} and
\eqref{eq:parametric-null-polynomial-Gram-bound} imply
\begin{equation}
\begin{aligned}
\max_k|\hat c_{n,k}-\beta_n^\dagger|
&=O_{P_{0,n}}(n^{-1/2})+o_{P_{0,n}}(\varrho_n)
=o_{P_{0,n}}(\varrho_n),\\
\max_k\|\hat\beta_{n,k,d}-\beta_n^\dagger\|
&=O_{P_{0,n}}(n^{-1/2})+o_{P_{0,n}}(\varrho_n)
=o_{P_{0,n}}(\varrho_n),\quad d\in\{1,2\},
\end{aligned}
\label{eq:parametric-null-polynomial-coefficient-rates}
\end{equation}
for the constant and degree-$d$ families, respectively, since
$n^{-1/2}=h_n^{1/4}\varrho_n=o(\varrho_n)$.

Choose $\epsilon>0$ such that
$\{\beta:\|\beta-\beta_0\|\le2\epsilon\}
\subset\operatorname{int}(\mathcal B)$.
For all sufficiently large $n$,
$\|\beta_n^\dagger-\beta_0\|\le\epsilon$ by
\eqref{eq:parametric-null-polynomial-population-projection}.
Equation~\eqref{eq:parametric-null-polynomial-coefficient-rates} yields
\[
\begin{aligned}
P_{0,n}\{\exists k:\hat c_{n,k}\notin\operatorname{int}(\mathcal B)\}
&\le P_{0,n}\{\max_k|\hat c_{n,k}-\beta_n^\dagger|>\epsilon\}
\longrightarrow0,\\
P_{0,n}\{\exists k:\hat\beta_{n,k,d}\notin\operatorname{int}(\mathcal B)\}
&\le P_{0,n}\{\max_k\|\hat\beta_{n,k,d}-\beta_n^\dagger\|>\epsilon\}
\longrightarrow0.
\end{aligned}
\]
For the constant family and every $\beta\in\mathbb R$, we get
\[
P_{n,k}\{(\hat\xi_n^{(-k)}-\beta)^2\}
-P_{n,k}\{(\hat\xi_n^{(-k)}-\hat c_{n,k})^2\}
=(\beta-\hat c_{n,k})^2.
\]
For $d\in\{1,2\}$, whenever
$\max_k\|P_{n,k}(v_dv_d^\top)-J_0\|
\le\lambda_{\min}(J_0)/2$, we obtain
\begin{equation}
\begin{aligned}
&P_{n,k}\{[\hat\xi_n^{(-k)}-v_d(A)^\top\beta]^2\}
-P_{n,k}\{[\hat\xi_n^{(-k)}-v_d(A)^\top\hat\beta_{n,k,d}]^2\}\\
&\quad=(\beta-\hat\beta_{n,k,d})^\top
P_{n,k}\{v_d(A)v_d(A)^\top\}(\beta-\hat\beta_{n,k,d})\\
&\quad\ge\frac{\lambda_{\min}(J_0)}2
\|\beta-\hat\beta_{n,k,d}\|^2,
\qquad\beta\in\mathbb R^{d+1}.
\end{aligned}
\label{eq:parametric-null-polynomial-empirical-loss}
\end{equation}
Equation~\eqref{eq:parametric-null-polynomial-empirical-loss}, its
constant-family identity, and
\eqref{eq:parametric-null-polynomial-coefficient-rates}
identify the minimizers in \eqref{eq:parametric-null-beta}:
\begin{equation}
\begin{aligned}
P_{0,n}\{\hat\beta_{n,k}=\hat c_{n,k},\ 1\le k\le K\}
&\longrightarrow1
&&\text{for }m(a;\beta)=\beta,\\
P_{0,n}\{\hat\beta_{n,k}=\hat\beta_{n,k,d},\ 1\le k\le K\}
&\longrightarrow1
&&\text{for }m(a;\beta)=v_d(a)^\top\beta,\\
\max_k\|\hat\beta_{n,k}-\beta_n^\dagger\|
&=o_{P_{0,n}}(\varrho_n)=o_{P_{0,n}}(1).
\end{aligned}
\label{eq:parametric-null-polynomial-fit-identification}
\end{equation}
Equations~\eqref{eq:parametric-null-polynomial-separation},
\eqref{eq:parametric-null-polynomial-population-loss},
\eqref{eq:parametric-null-polynomial-population-projection}, and
\eqref{eq:parametric-null-polynomial-fit-identification} establish
condition~\ref{assm:parametric-null-local-family}.

For the three dose polynomials, note that
\[
\theta_{0,n}''(a)=
\begin{cases}
c\varrho_ng''(a),&m(a;\beta)=\beta\text{ or }v_1(a)^\top\beta,\\
2\beta_{0,3}+c\varrho_ng''(a),&m(a;\beta)=v_2(a)^\top\beta.
\end{cases}
\]
For $a,t\in\mathcal A_{\delta_3}$,
condition~\ref{assm:parametric-null-local-law} gives
\[
|\theta_{0,n}''(a)-\theta_{0,n}''(t)|
=|c|\varrho_n|g''(a)-g''(t)|
\le C|c|\varrho_n|a-t|^\alpha
\le C|a-t|^\alpha.
\]
Since $\theta_{0,n}\in C^2(\mathcal A_{\delta_3})$,
condition~\ref{assm:G-smooth}\textup{(i)} holds with a constant
independent of $n$ for all sufficiently large $n$.

Equation~\eqref{eq:parametric-null-polynomial-derivatives} also implies
\begin{equation}
\begin{aligned}
&\int_{\mathcal A}
\left\{\frac{\partial m(a;\beta_n^\dagger)}{\partial\beta}
-\dot m_0(a)\right\}\Delta_{n,k}(a)\,dF_0(a)=0,\\
&P_{n,k}\left[
\frac{\partial^2m(A;\beta)}{\partial\beta\partial\beta^\top}
\{\hat\xi_n^{(-k)}-m(A;\beta)\}\right]=0,
\qquad\beta\in\mathcal B.
\end{aligned}
\label{eq:parametric-null-polynomial-zero-drift-terms}
\end{equation}
The proof of Lemma~\ref{lem:parametric-null-local-projection} uses
$\max_k\int_{\mathcal A}|\Delta_{n,k}|\,dF_0=o_{P_{0,n}}(1)$
to bound the derivative difference in the score and the term involving
the second derivative in the score derivative.
Then, display~\eqref{eq:parametric-null-polynomial-zero-drift-terms}
sets both expressions equal to zero.
The exact residual in
\eqref{eq:parametric-null-polynomial-population-projection} and
the coefficient rate in
\eqref{eq:parametric-null-polynomial-fit-identification} therefore
supply the population and coefficient conclusions of
Lemma~\ref{lem:parametric-null-local-projection} used in
Lemmas~\ref{lem:parametric-null-local-decomposition} and
\ref{lem:parametric-null-local-oracle} without the global
$L_1(F_0)$ rate.  The smoothing conclusion
\eqref{eq:parametric-null-local-smoothed-direction} follows from
the kernel and translation calculation in the proof of
Lemma~\ref{lem:parametric-null-local-projection};
boundedness of $g$ and $\dot m_0$ in
$g_{\mathcal M}^{\perp}(a)$ gives
$\|g_{\mathcal M}^{\perp}\|_\infty<\infty$.

For the fold smoothing matrices, apply the same calculations used in the proof of
Lemma~\ref{lem:Gamma-difference-expansion} after replacing
$P_n$ by $P_{n,k}$ and $n$ by $n_k$.
Under $P_{0,n}$, the law of $A$ is $F_0$ and $n_k\asymp n$.
Hence
\[
\begin{aligned}
\max_k\sup_{a\in\mathcal A_0}\big\{
\|D_{n,k,h_n,a,1}-D_{0,h_n,a,1}\|+\|D_{n,k,b_n,a,2}-D_{0,b_n,a,2}\|\big\}
=O_{P_{0,n}}\left(\sqrt{\frac{\log n}{nh_n}}\right),
\end{aligned}
\]
which is $o_{P_{0,n}}(1)$.
The population matrix bounds in Lemma~\ref{lem:D-expansion} then imply
\[
P_{0,n}\left\{
\min_k\inf_{a\in\mathcal A_0}
\min\big\{\lambda_{\min}(D_{n,k,h_n,a,1}),
\lambda_{\min}(D_{n,k,b_n,a,2})\big\}>0
\right\}\longrightarrow1.
\]
By Proposition~\ref{prop:parametric-null-reproduction} and
\eqref{eq:parametric-null-polynomial-fit-identification}, with
probability tending to one, simultaneously over $a\in\mathcal A_0$
and $1\le k\le K$,
\[
\begin{aligned}
\hat m_{n,k,h_n,b_n}(a;\hat\beta_{n,k})
&=P_{n,k}\{\Gamma_{n,k,a}(A)m(A;\hat\beta_{n,k})\}\\
&=\begin{cases}
\hat c_{n,k},&m(a;\beta)=\beta,\\
v_d(a)^\top\hat\beta_{n,k,d},&m(a;\beta)=v_d(a)^\top\beta.
\end{cases}
\end{aligned}
\]
Substitution in \eqref{eq:parametric-null-statistic} yields
\begin{equation}
\widehat T_{n,\mathcal M}^{DB}
=\begin{cases}
\widehat T_{n,\mathrm{const}}^{DB},&m(a;\beta)=\beta,\\
\widehat T_{n,\mathrm{lin}}^{DB},&m(a;\beta)=v_1(a)^\top\beta,\\
\widehat T_{n,\mathrm{quad}}^{DB},&m(a;\beta)=v_2(a)^\top\beta,
\end{cases}
\label{eq:parametric-null-polynomial-statistic-identities}
\end{equation}
with probability tending to one.
Apply the proofs of
Lemmas~\ref{lem:parametric-null-local-decomposition} and
\ref{lem:parametric-null-local-oracle}, and
Theorem~\ref{thm:parametric-null-local-power}, using
\eqref{eq:parametric-null-polynomial-population-projection} and
\eqref{eq:parametric-null-polynomial-fit-identification}
for the population residual and coefficient rate.
Equation~\eqref{eq:parametric-null-polynomial-zero-drift-terms}
removes the two uses of the global $L_1(F_0)$ drift rate.
Equation~\eqref{eq:parametric-null-polynomial-statistic-identities}
then gives conclusions \textup{(i)}--\textup{(iii)}.
\end{proof}

\section{Additional numerical details}
\label{sec:additional-numerical-details}

\subsection{Bandwidth selection}
\label{sec:implementation-bandwidth}

This subsection describes the data-adaptive bandwidth selectors used in the
numerical experiments.  The criteria are evaluated conditional on the
cross-fitted pseudo-outcomes, without re-estimating the nuisance functions
across candidate bandwidths or leave-one-out samples.  Each selector is
evaluated separately within each nuisance-model scenario.

For \(p\geq0\), let \(w_p(u):=(1,u,\ldots,u^p)^\top\), and define
\[
U_i(a,h):=\frac{A_i-a}{h},\qquad
D_{n,p}(a,h):=
P_n\!\left[
w_p\{U_i(a,h)\}w_p\{U_i(a,h)\}^\top K_{h,a}(A_i)
\right].
\]
For a cross-fitted pseudo-outcome \(\widehat\xi_i^{\mathrm{cf}}\), set
\[
\widehat M_{n,p}(a,h):=
P_n\!\left[
w_p\{U_i(a,h)\}K_{h,a}(A_i)\widehat\xi_i^{\mathrm{cf}}
\right].
\]
The direct local-polynomial estimator of derivative order \(\nu\leq p\) is
\begin{equation}
\widehat\theta_{n,p,\nu}(a;h)
:=
\nu!h^{-\nu}e_{\nu+1}^{\top}
D_{n,p}(a,h)^{-1}\widehat M_{n,p}(a,h).
\label{eq:simulation-direct-local-polynomial}
\end{equation}
When \(b=h\) is used, the PDR dose-response and derivative estimators equal
\(\widehat\theta_{n,2,0}\) and \(\widehat\theta_{n,3,1}\), respectively.
ProxLL--US uses \(\widehat\theta_{n,1,0}\) for the dose-response function,
whereas ProxLQ--US uses \(\widehat\theta_{n,2,1}\) for the derivative.

\subsubsection{Dose-response bandwidth selection}\label{sec: DRF badn selectors}

Let
\[
\widehat s_R
:=
\min\left\{
\widehat{\operatorname{sd}}(A),
\frac{\widehat{\operatorname{IQR}}(A)}{1.349}
\right\},
\qquad
h_{\mathrm{NR},0}:=1.06\,\widehat s_Rn^{-1/5}.
\]
The initial CV1 candidate set is
\begin{equation}
\mathcal H_{0}^{(0)}
:=
\left\{
h_{\mathrm{NR},0}
\exp\!\left[
\log(0.35)+\frac{j-1}{14}
\{\log(2.75)-\log(0.35)\}
\right]:
j=1,\ldots,15
\right\}.
\label{eq:simulation-level-cv1-grid}
\end{equation}
After an initial minimization, the geometric midpoints of the grid intervals
adjacent to the preliminary minimizer are added, and the criterion is
reevaluated on this refined grid.

Let \(\widehat q_{.05}\) and \(\widehat q_{.95}\) denote the empirical
\(0.05\) and \(0.95\) quantiles of \(A\), and define
\(\mathcal T_n:=\{i:\widehat q_{.05}\leq A_i\leq\widehat q_{.95}\}\).
For \(p\in\{1,2\}\), write
\(\widehat\theta_{n,p,0,-i}(A_i;h)\) for the exact leave-one-out version of
\eqref{eq:simulation-direct-local-polynomial}.  The same step
is applied to the degree-one and degree-two risks.  We only keep the window containing at least 20 neighbors.  Let
\(\mathcal V_{n,p,s}^{(0)}(h)\subseteq\mathcal T_n\) contain the observations
that pass this filter under
scenario \(s\).  A candidate is retained for risk \(p\) only when
\[
|\mathcal V_{n,p,s}^{(0)}(h)|
\geq \max\{1,\lfloor0.98|\mathcal T_n|\rfloor\}.
\]
The criterion is
\begin{equation}
\widehat{\operatorname{CV}}_{0,p,s}(h)
:=
\frac{1}{|\mathcal V_{n,p,s}^{(0)}(h)|}
\sum_{i\in\mathcal V_{n,p,s}^{(0)}(h)}
\left\{
\widehat\xi_{i,s}^{\mathrm{cf}}
-\widehat\theta_{n,p,0,-i}(A_i;h)
\right\}^2.
\label{eq:implementation-db-loocv}
\end{equation}
PDR--CV1 minimizes \eqref{eq:implementation-db-loocv} with \(p=2\), whereas
the ProxLL selector uses \(p=1\).

For PDR--PI1, let
\[
h_{1,s}
:=
\Pi_{[\min\mathcal H_0^{(0)},\max\mathcal H_0^{(0)}]}
\{1.5\widehat h_{\mathrm{LL,CV},s}\},
\]
where \(\widehat h_{\mathrm{LL,CV},s}\) is the refined \(p=1\) CV1
minimizer.  A degree-two fit with bandwidth \(h_{1,s}\) estimates
\(\theta_0''(a)\).  Let \(\widehat\sigma_{\xi,s}^2(a)\) be the sample
variance of the \(k_n\) nearest cross-fitted pseudo-outcomes, where
$k_n:=\max\{15,\lceil n^{1/2}\rceil\}$ \citep{Takatsu2025debiased}.
The boundary windows are shifted to contain \(k_n\)
observations.
Using the degree-one local-polynomial design at \(h_{1,s}\), define
\begin{align}
\widehat B_{0,s}(a)
&:=
e_1^\top D_{n,1}(a,h_{1,s})^{-1}
P_n\!\left[
w_1\{U_i(a,h_{1,s})\}K_{h_{1,s},a}(A_i)
\frac{\widehat\theta_s''(a)}{2}U_i(a,h_{1,s})^2
\right],
\label{eq:implementation-plugin-bias}\\
\widehat V_{0,s}(a)
&:=
h_{1,s}e_1^\top D_{n,1}(a,h_{1,s})^{-1}\notag\\[-2pt]
&\quad\times
P_n\!\left[
w_1\{U_i(a,h_{1,s})\}K_{h_{1,s},a}(A_i)^2
\widehat\sigma_{\xi,s}^2(A_i)
w_1\{U_i(a,h_{1,s})\}^\top
\right]\notag\\[-2pt]
&\quad\times D_{n,1}(a,h_{1,s})^{-1}e_1.
\label{eq:implementation-plugin-variance}
\end{align}
The integrals
\[
\widehat I_{B,0,s}
:=
\int_{-0.9}^{0.9}\widehat B_{0,s}(a)^2\,da,
\qquad
\widehat I_{V,0,s}
:=
\int_{-0.9}^{0.9}\widehat V_{0,s}(a)\,da
\]
are evaluated. Note that the unconstrained minimizer of
\[
h^4\widehat I_{B,0,s}+\frac{\widehat I_{V,0,s}}{nh}
\]
is
\begin{equation}
\widehat h_{\mathrm{PI1},0,s}^{\mathrm{raw}}
=
\left\{
\frac{\widehat I_{V,0,s}}
{4n\widehat I_{B,0,s}}
\right\}^{1/5}.
\label{eq:implementation-plugin-bandwidth}
\end{equation}
PDR--PI1 truncates this value to
\([\min\mathcal H_0^{(0)},\max\mathcal H_0^{(0)}]\).

Let
\[
\widehat s_A
:=
\left\{\frac{1}{n-1}\sum_{i=1}^n(A_i-\bar A)^2\right\}^{1/2}.
\]
PDR--PI2 uses
\begin{equation}
\widehat h_{\mathrm{PI2},0}
:=
1.5\,\widehat s_A n^{-1/5}.
\label{eq:simulation-pi2-bandwidth}
\end{equation}
PDR--CV2 minimizes \eqref{eq:implementation-db-loocv} with \(p=2\) over
\begin{equation}
\mathcal H_{2,n}^{(0)}
:=
\left\{
\widehat h_{\mathrm{PI2},0}2^{j/10}:j=-15,\ldots,15
\right\}.
\label{eq:simulation-cv2-grid}
\end{equation}
The ProxLL--US dose-response bandwidth is
\begin{equation}
\widehat h_{\mathrm{US},0,s}
:=
\frac{\widehat h_{\mathrm{LL,CV},s}}{\log_{10}(n)}.
\label{eq:implementation-undersmoothing}
\end{equation}
All four PDR estimators set \(b=h\).

\subsubsection{Derivative bandwidth selection}\label{sec: derivative band selectors}

Let
\[
h_{\mathrm{NR},1}:=1.06\,\widehat s_Rn^{-1/7}.
\]
The derivative CV1 initial set \(\mathcal H_0^{(1)}\) is defined by replacing
\(h_{\mathrm{NR},0}\) by \(h_{\mathrm{NR},1}\) in
\eqref{eq:simulation-level-cv1-grid}; the same adjacent geometric refinement
is then applied.

The derivative selectors minimize a derivative-risk criterion rather than a
pseudo-outcome prediction criterion.  Define the taper
\[
\rho(a)=
\begin{cases}
1, & |a|\leq0.9,\\
1-10t^3+15t^4-6t^5,
  & 0.9<|a|<1,\quad t=(|a|-0.9)/0.1,\\
0, & |a|\geq1.
\end{cases}
\]
Then \(\rho=\rho'=0\) at \(\{-1,1\}\).  Let
\(d_{p,-i,s}(a;h):=\widehat\theta_{n,p,1,-i}(a;h)\), and let
\(\partial_a d_{p,-i,s}(a;h)\) be its derivative with.  The leave-one-out exposure density is computed
with the Epanechnikov kernel, bandwidth
\(1.06\widehat s_Rn^{-1/5}\).
For \(p\in\{2,3\}\), let
\(\mathcal V_{n,p,s}^{(1)}(h)\) contain the observations with
\(-1<A_i<1\) for which at least 20 neighbors are present, as above.  
With
\(N_1:=\sum_{i=1}^n I(-1<A_i<1)\), a candidate is allowed for risk \(p\)
only when
\[
|\mathcal V_{n,p,s}^{(1)}(h)|
\geq \max\{1,\lfloor0.98N_1\rfloor\}.
\]
The implemented criterion is
\begin{align}
\widehat{\operatorname{CV}}_{1,p,s}(h)
:={}&
\frac{1}{|\mathcal V_{n,p,s}^{(1)}(h)|}
\sum_{i\in\mathcal V_{n,p,s}^{(1)}(h)}
\frac{1}{\widehat f_{A,-i}(A_i)}
\Bigl[
\rho(A_i)d_{p,-i,s}(A_i;h)^2
\nonumber\\[-0.3em]
&\hspace{27mm}+
2\widehat\xi_{i,s}^{\mathrm{cf}}
\{\rho(A_i)\partial_a d_{p,-i,s}(A_i;h)
+\rho'(A_i)d_{p,-i,s}(A_i;h)\}
\Bigr].
\label{eq:simulation-derivative-cv}
\end{align}
Note that integration by parts yields
\[
-2\int\rho(a)d(a)\theta_0'(a)\,da
=
2\int\theta_0(a)\{\rho(a)d(a)\}'\,da,
\]
since the boundary term vanishes.  Thus
\eqref{eq:simulation-derivative-cv} estimates tapered integrated squared
derivative error up to a term that does not depend on \(h\).  PDR--CV1 and
PDR--CV2 use \(p=3\); the ProxLQ selector uses \(p=2\).

For PDR--PI1, a degree-three pilot with bandwidth
\[
h_{3}:=1.5\,\widehat s_Rn^{-1/9}
\]
estimates \(\theta_0^{(3)}(a)\).  Using the primary degree-two slope design,
define
\[
\widehat c_{3,s}(a)
:=
e_2^\top D_{n,2}(a,h_3)^{-1}
P_n\!\left[
w_2\{U_i(a,h_3)\}K_{h_3,a}(A_i)U_i(a,h_3)^3
\right],
\]
\[
\widehat B_{1,s}(a)
:=
\frac{\widehat c_{3,s}(a)\widehat\theta_s^{(3)}(a)}{6},
\]
and
\[
\begin{aligned}
\widehat V_{1,s}(a)
&:=
h_3e_2^\top D_{n,2}(a,h_3)^{-1}\\[-2pt]
&\quad\times
P_n\!\left[
w_2\{U_i(a,h_3)\}K_{h_3,a}(A_i)^2
\widehat\sigma_{\xi,s}^2(A_i)
w_2\{U_i(a,h_3)\}^\top
\right]\\[-2pt]
&\quad\times D_{n,2}(a,h_3)^{-1}e_2.
\end{aligned}
\]
Let
\[
\widehat I_{B,1,s}:=\int_{-0.9}^{0.9}\widehat B_{1,s}(a)^2\,da,
\qquad
\widehat I_{V,1,s}:=\int_{-0.9}^{0.9}\widehat V_{1,s}(a)\,da.
\]
The derivative AMISE criterion and its unconstrained
minimizer are
\begin{equation}
h^4\widehat I_{B,1,s}+\frac{\widehat I_{V,1,s}}{nh^3},
\qquad
\widehat h_{\mathrm{PI1},1,s}^{\mathrm{raw}}
=
\left\{
\frac{3\widehat I_{V,1,s}}
{4n\widehat I_{B,1,s}}
\right\}^{1/7}.
\label{eq:simulation-derivative-pi1}
\end{equation}
PDR--PI1 truncates this value to
\([\min\mathcal H_0^{(1)},\max\mathcal H_0^{(1)}]\).

PDR--PI2 and the derivative CV2 candidate set are
\begin{equation}
\widehat h_{\mathrm{PI2},1}
:=
1.5\,\widehat s_A n^{-1/7},
\qquad
\mathcal H_{2,n}^{(1)}
:=
\left\{
\widehat h_{\mathrm{PI2},1}2^{j/10}:j=-15,\ldots,15
\right\}.
\label{eq:simulation-derivative-pi2-cv2}
\end{equation}
PDR--CV2 minimizes \eqref{eq:simulation-derivative-cv} with \(p=3\) over
\(\mathcal H_{2,n}^{(1)}\).  Let \(\widehat h_{\mathrm{LQ,CV},s}\) denote the
refined \(p=2\) CV1 minimizer.  The ProxLQ--US bandwidth is
\begin{equation}
\widehat h_{\mathrm{US},1,s}
:=
\frac{\widehat h_{\mathrm{LQ,CV},s}}{\log_{10}(n)}.
\label{eq:simulation-derivative-proxlq-bandwidth}
\end{equation}
Its pointwise standard errors use the finite-sample influence-function
construction in Section~\ref{sec:inference-implementation-covariance},
with the direct degree-two first-derivative weight and its design-correction
function. All four PDR derivative fits use \(b=h\).

\subsection{Nuisance estimation}
\label{sec:simulation-nuisance}

The nuisance functions are estimated separately on each training sample.
For fold \(k\), write
\[
n_{-k}:=n-n_k,
\qquad
P_{n,-k}f:=
\frac{1}{n_{-k}}\sum_{i\notin I_k}f(O_i).
\]
Let \(c_A:=\exp(-V_A/2)\) and define
\[
S_A(a):=\sin(a)-c_Aa,\qquad
C_A(a):=\cos(a)-c_A.
\]
The outcome-bridge instrument vector is
\[
B_h(Z,A,X)
:=
\{1,A,X_1,X_2,S_A(A),C_A(A),Z\}^\top.
\]
The correct and misspecified regressor vectors are
\[
\begin{aligned}
R_h^{\mathrm C}(W,A,X)
&:=
\{1,A,X_1,X_2,S_A(A),C_A(A),\tanh(W)\}^\top,\\
R_h^{\mathrm M}(W,A,X)
&:=
\{1,A,X_1,X_2,S_A(A),C_A(A)\}^\top.
\end{aligned}
\]
For \(R_h\in\{R_h^{\mathrm C},R_h^{\mathrm M}\}\), define
\[
\widehat g_{h,k}(b)
:=
P_{n,-k}\!\left[
B_h\{Y-b^\top R_h\}
\right],
\qquad
\widehat\Omega_{h,k}
:=
P_{n,-k}(B_hB_h^\top).
\]
With \(\lambda=10^{-10}\), let
\begin{equation}
\widehat b^{(-k)}
:=
\arg\min_b
\widehat g_{h,k}(b)^\top
(\widehat\Omega_{h,k}+\lambda I)^{-1}
\widehat g_{h,k}(b)
+
\lambda\|b\|_2^2,
\label{eq:simulation-outcome-gmm}
\end{equation}
and set
\[
\widehat h_n^{(-k)}(w,a,x)
:=
\widehat b^{(-k)\top}R_h(w,a,x).
\]
The fixed ridge in
\eqref{eq:simulation-outcome-gmm} is used for numerical stability.

For the treatment bridge, define
\[
\begin{aligned}
B_q^{\mathrm C}(Z,A,X)
:=(&1,Z,A,X_1,X_2,Z^2,A^2,X_1^2,X_2^2,\\
   &ZA,ZX_1,ZX_2,AX_1,AX_2,X_1X_2)^\top,
\end{aligned}
\]
\[
B_q^{\mathrm M}(Z,A,X)
:=
(1,Z,A,X_1,X_2)^\top,
\]
and
\[
\begin{aligned}
G_q(W,A,X)
:=(&1,W,A,X_1,X_2,W^2,A^2,X_1^2,X_2^2,\\
   &WA,WX_1,WX_2,AX_1,AX_2,X_1X_2)^\top.
\end{aligned}
\]

The treatment-bridge equation \eqref{eq:Bq} implies
\[
E_0\!\left[
G_q(W,A,X)q_0(Z,A,X)
\right]
=
\iint G_q(w,a,x)\,dP_{0,WX}(w,x)\,dF_0(a).
\]
We estimate the right-hand side using the product of the two
training-sample empirical marginal distributions:
\[
\widehat\tau_{q,k}
:=
\frac{1}{n_{-k}^2}
\sum_{i\notin I_k}\sum_{j\notin I_k}
G_q(W_i,A_j,X_i).
\]
The treatment bridge is estimated
by matching its weighted empirical moments to
$\widehat\tau_{q,k}$.

For $B_q\in\{B_q^{\mathrm C},B_q^{\mathrm M}\}$, define diagonal
matrices $S_{B,k}$ and $S_{G,k}$ by the column scales
\[
\begin{aligned}
(S_{B,k})_{jj}
&:=
\left\{
\frac{1}{n_{-k}}
\sum_{i\notin I_k}
B_{q,j}(Z_i,A_i,X_i)^2
\right\}^{1/2},\\
(S_{G,k})_{jj}
&:=
\left\{
\frac{1}{n_{-k}}
\sum_{i\notin I_k}
G_{q,j}(W_i,A_i,X_i)^2
\right\}^{1/2}.
\end{aligned}
\]
The intercept entries are set to one.
We numerically minimize
\begin{equation}
\eta\longmapsto
\left\|
S_{G,k}^{-1}
\left[
P_{n,-k}\!\left\{
G_q\exp\!\left(\eta^\top S_{B,k}^{-1}B_q\right)
\right\}
-\widehat\tau_{q,k}
\right]
\right\|_2^2.
\label{eq:simulation-treatment-gmm}
\end{equation}
Here $S_{G,k}^{-1}$ divides each moment discrepancy by the
empirical root-mean-square value of the corresponding coordinate
of $G_q$. The matrix $S_{B,k}^{-1}$ rescales the coefficients
without changing the working function class, since
\[
\eta^\top S_{B,k}^{-1}B_q
=
(S_{B,k}^{-1}\eta)^\top B_q.
\]
Let $\widehat\eta^{(-k)}$ denote the coefficient vector returned by
the procedure. We then compute
\[
\widehat q_n^{(-k)}(z,a,x)
:=
\exp\!\left\{
\widehat\eta^{(-k)\top}
S_{B,k}^{-1}B_q(z,a,x)
\right\}.
\]

\begin{table}[!t]
\centering
\caption{Working nuisance models under the four specification scenarios.
The first and second letters indicate the outcome- and treatment-bridge
models, respectively.}
\label{tab:simulation-scenarios}
\begin{tabular}{@{}clll@{}}
\toprule
Scenario & Outcome bridge & Treatment bridge & Population implication\\
\midrule
CC & \(R_h^{\mathrm C}\) & \(B_q^{\mathrm C}\) & Both bridges correct\\
CM & \(R_h^{\mathrm C}\) & \(B_q^{\mathrm M}\) & Outcome bridge correct\\
MC & \(R_h^{\mathrm M}\) & \(B_q^{\mathrm C}\) & Treatment bridge correct\\
MM & \(R_h^{\mathrm M}\) & \(B_q^{\mathrm M}\) & Neither bridge correct\\
\bottomrule
\end{tabular}
\end{table}

For \(i\in I_k\), let
\[
\widehat m_{n,k}^{(-k)}(a)
:=
P_{n,k}\{\widehat h_n^{(-k)}(W,a,X)\}.
\]
The cross-fitted pseudo-outcome under each scenario is
\[
\widehat\xi_{i}^{\mathrm{cf}}
:=
\widehat q_n^{(-k)}(Z_i,A_i,X_i)
\{Y_i-\widehat h_n^{(-k)}(W_i,A_i,X_i)\}
+
\widehat m_{n,k}^{(-k)}(A_i).
\]
The fitted bridges and pseudo-outcomes are reused over all evaluation
coordinates and candidate bandwidths.  All \(20{,}000\) requested
replications satisfy the prespecified nuisance-convergence and artifact
validation checks.

\subsection{Additional tables and figures}
\label{sec:simulation-additional-results}

This subsection supplements
Sections~\ref{sec:simulation-dose-response} and
\ref{sec:simulation-derivative}.  

\begin{figure}[!t]
\centering
\includegraphics[width=\textwidth]{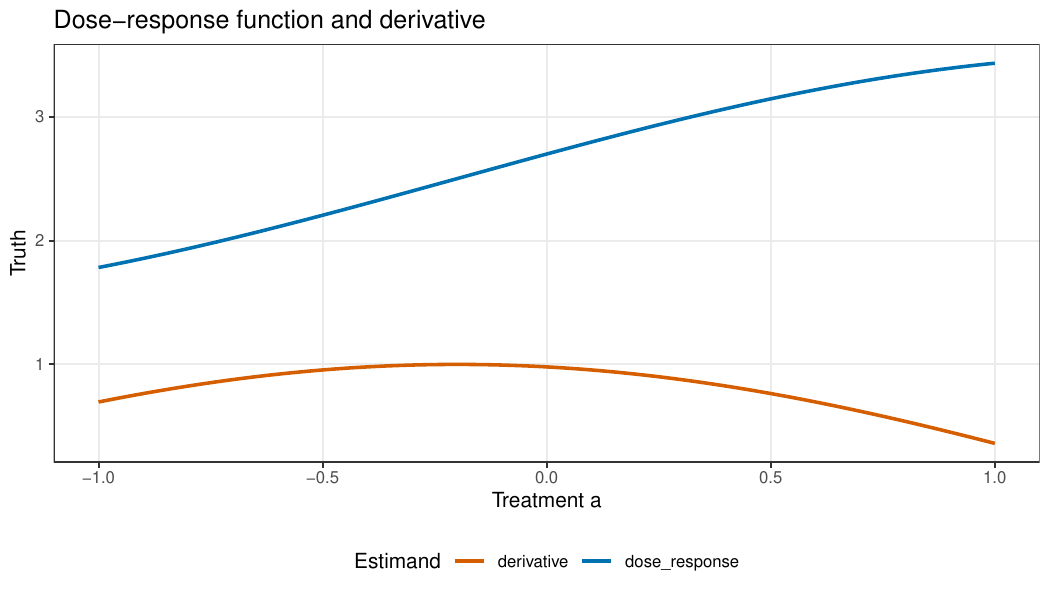}
\caption{Dose-response function
\(\theta_0(a)=2.5+\sin(a+0.2)\) and derivative
\(\theta_0'(a)=\cos(a+0.2)\) on \([-1,1]\).}
\label{fig:sim-dgp-truth}
\end{figure}

\subsubsection{Estimation and pointwise inference}

Tables~\ref{tab:sim-pointwise-grid-summary}
and~\ref{tab:sim-pointwise-grid-summary-derivative} report pointwise coverage
and mean interval width averaged over the 61 coordinates in CC, CM, and MC.
Tables~\ref{tab:sim-ci-nine-point}
and~\ref{tab:sim-ci-nine-point-derivative} report the coordinatewise results
in CC, CM, and MC.
\begin{table}[!p]
\centering
\caption{Empirical coverage and mean width of nominal 95\% pointwise
confidence intervals for the dose-response function, averaged over the 61
grid coordinates.  Each entry gives coverage, followed by mean width in
brackets, over \(5{,}000\) replications.}
\label{tab:sim-pointwise-grid-summary}
\scriptsize
\setlength{\tabcolsep}{2.8pt}
\begin{tabular}{@{}rrccccc@{}}
\toprule
\(n\) & Scenario & PDR--PI1 & PDR--PI2 & PDR--CV1 & PDR--CV2 & ProxLL--US\\
\midrule
% Source: Derivative R study/joint_simulation/results/config_4e7d28a1/summary/integrated_performance.csv (configuration 4e7d28a1).
% Schema: n & scenario & PDR--PI1 coverage [mean CI width] & PDR--PI2 coverage [mean CI width] & PDR--CV1 coverage [mean CI width] & PDR--CV2 coverage [mean CI width] & ProxLL--US coverage [mean CI width] \\
% Scope: arithmetic means over the fixed 61-point grid; 5000 replications; scenarios CC, CM, and MC.
% Rounding: coverage and mean CI width are rounded to four decimal places.
2500 & CC & 0.9468 [0.0924] & 0.9498 [0.0951] & 0.9465 [0.0801] & 0.9436 [0.0736] & 0.9489 [0.1107] \\
2500 & CM & 0.9461 [0.0913] & 0.9492 [0.0934] & 0.9459 [0.0789] & 0.9432 [0.0726] & 0.9488 [0.1095] \\
2500 & MC & 0.9463 [0.1696] & 0.9495 [0.1999] & 0.9478 [0.1527] & 0.9462 [0.1304] & 0.9510 [0.2123] \\
5000 & CC & 0.9468 [0.0684] & 0.9493 [0.0704] & 0.9475 [0.0585] & 0.9460 [0.0535] & 0.9493 [0.0852] \\
5000 & CM & 0.9464 [0.0674] & 0.9488 [0.0692] & 0.9471 [0.0576] & 0.9448 [0.0528] & 0.9493 [0.0840] \\
5000 & MC & 0.9494 [0.1273] & 0.9508 [0.1511] & 0.9504 [0.1143] & 0.9493 [0.0972] & 0.9527 [0.1647] \\
7500 & CC & 0.9488 [0.0571] & 0.9509 [0.0592] & 0.9499 [0.0487] & 0.9490 [0.0443] & 0.9510 [0.0727] \\
7500 & CM & 0.9482 [0.0564] & 0.9506 [0.0581] & 0.9492 [0.0480] & 0.9482 [0.0437] & 0.9508 [0.0717] \\
7500 & MC & 0.9484 [0.1083] & 0.9501 [0.1284] & 0.9493 [0.0969] & 0.9487 [0.0822] & 0.9522 [0.1433] \\
10000 & CC & 0.9484 [0.0505] & 0.9506 [0.0523] & 0.9493 [0.0429] & 0.9481 [0.0389] & 0.9516 [0.0654] \\
10000 & CM & 0.9478 [0.0499] & 0.9504 [0.0514] & 0.9487 [0.0423] & 0.9474 [0.0385] & 0.9514 [0.0645] \\
10000 & MC & 0.9504 [0.0960] & 0.9519 [0.1144] & 0.9503 [0.0860] & 0.9491 [0.0726] & 0.9531 [0.1285] \\
\bottomrule

\end{tabular}
\end{table}

\begin{table}[!p]
\centering
\caption{Empirical coverage and mean width of nominal 95\% pointwise
confidence intervals for the derivative, averaged over the 61 grid
coordinates.  Each entry gives coverage, followed by mean width in brackets,
over \(5{,}000\) replications.}
\label{tab:sim-pointwise-grid-summary-derivative}
\scriptsize
\setlength{\tabcolsep}{2.8pt}
\begin{tabular}{@{}rrccccc@{}}
\toprule
\(n\) & Scenario & PDR--PI1 & PDR--PI2 & PDR--CV1 & PDR--CV2 & ProxLQ--US\\
\midrule
% Sources: frozen PDR results in configuration 4e7d28a1 and ProxLQ-US results in configuration 0bebf0fb.
% Schema: n & scenario & PDR--PI1 coverage [mean CI width] & PDR--PI2 coverage [mean CI width] & PDR--CV1 coverage [mean CI width] & PDR--CV2 coverage [mean CI width] & ProxLQ--US coverage [mean CI width] \\
% Scope: arithmetic means over the fixed 61-point grid; 5000 replications; scenarios CC, CM, and MC.
% Rounding: coverage and mean CI width are rounded to four decimal places.
2500 & CC & 0.9539 [0.4139] & 0.9511 [0.4509] & 0.9466 [2.2205] & 0.9480 [1.3785] & 0.9409 [6.1643] \\
2500 & CM & 0.9538 [0.4040] & 0.9510 [0.4390] & 0.9464 [2.1660] & 0.9480 [1.3481] & 0.9408 [6.0094] \\
2500 & MC & 0.9518 [1.0168] & 0.9492 [1.1309] & 0.9436 [5.0154] & 0.9441 [3.2301] & 0.9413 [14.3950] \\
5000 & CC & 0.9521 [0.3288] & 0.9496 [0.3668] & 0.9472 [1.7993] & 0.9480 [1.1255] & 0.9447 [5.7297] \\
5000 & CM & 0.9524 [0.3215] & 0.9495 [0.3577] & 0.9476 [1.7574] & 0.9480 [1.1014] & 0.9445 [5.6292] \\
5000 & MC & 0.9520 [0.8050] & 0.9499 [0.9223] & 0.9450 [3.9840] & 0.9454 [2.6097] & 0.9471 [13.3406] \\
7500 & CC & 0.9527 [0.2870] & 0.9505 [0.3255] & 0.9486 [1.6129] & 0.9489 [1.0054] & 0.9455 [5.4067] \\
7500 & CM & 0.9526 [0.2808] & 0.9505 [0.3176] & 0.9484 [1.5822] & 0.9489 [0.9854] & 0.9460 [5.2851] \\
7500 & MC & 0.9511 [0.7045] & 0.9487 [0.8195] & 0.9447 [3.6157] & 0.9448 [2.3175] & 0.9468 [12.7297] \\
10000 & CC & 0.9522 [0.2612] & 0.9497 [0.2992] & 0.9494 [1.4649] & 0.9486 [0.9195] & 0.9470 [5.2231] \\
10000 & CM & 0.9519 [0.2555] & 0.9496 [0.2920] & 0.9487 [1.4332] & 0.9491 [0.8945] & 0.9470 [5.0843] \\
10000 & MC & 0.9527 [0.6371] & 0.9503 [0.7536] & 0.9460 [3.2176] & 0.9459 [2.0965] & 0.9484 [12.0594] \\

\end{tabular}
\end{table}

% \begin{figure}[!p]
% \centering
% \includegraphics[width=\textwidth]{\simfigdir/mean_dose_response.pdf}
% \caption{Simulation mean dose-response estimates over the 61-point grid.
% Rows index nuisance-model scenario, columns index sample size, and the dashed
% curve is \(\theta_0\).}
% \label{fig:sim-mean-curves-dose}
% \end{figure}

% \begin{figure}[!p]
% \centering
% \includegraphics[width=\textwidth]{\simfigdir/mean_derivative.pdf}
% \caption{Simulation mean derivative estimates over the 61-point grid.  Rows
% index nuisance-model scenario, columns index sample size, and the dashed
% curve is \(\theta_0'\).}
% \label{fig:sim-mean-curves-derivative}
% \end{figure}

\begin{figure}[!p]
\centering
\includegraphics[width=\textwidth]{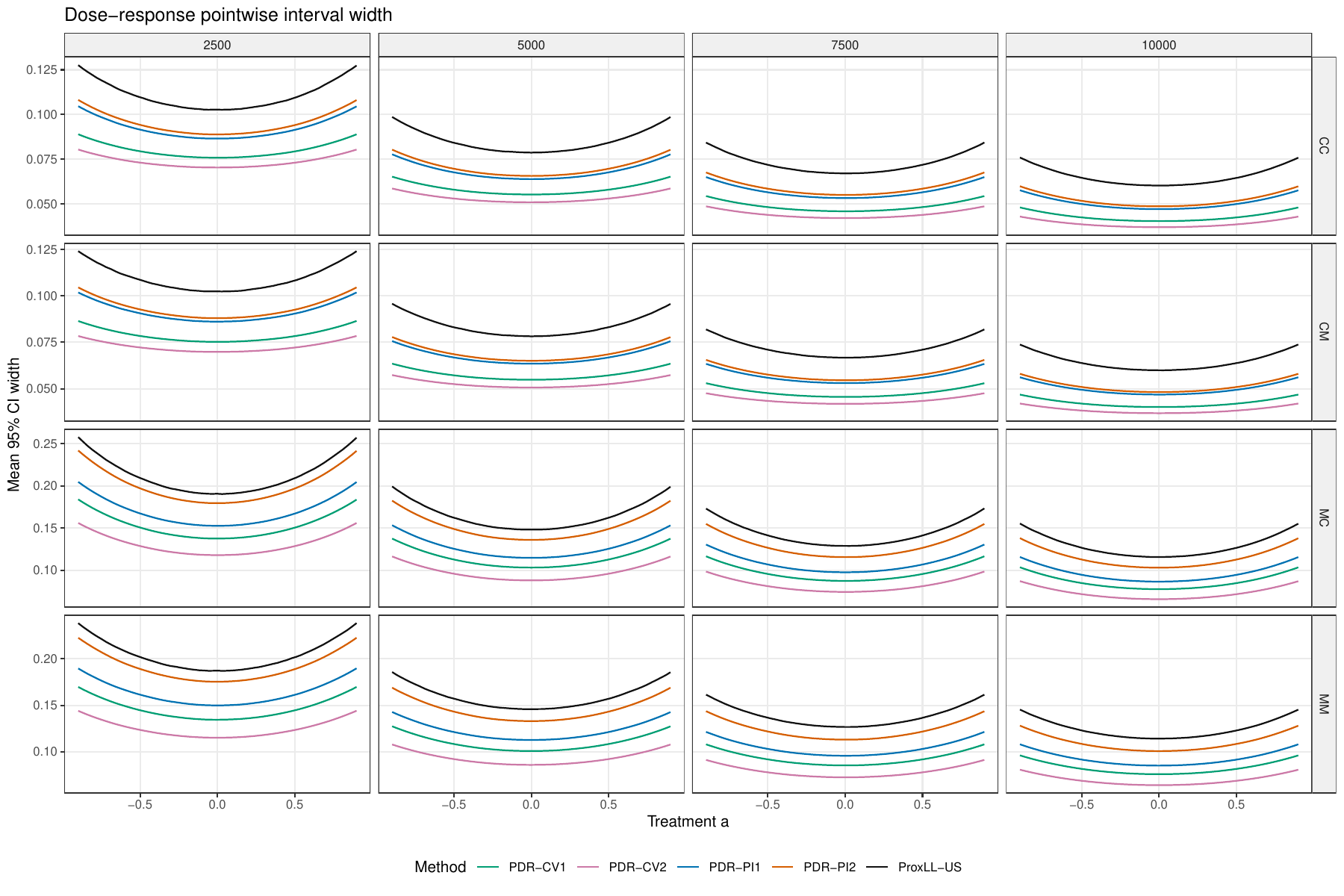}
\caption{Mean width of nominal 95\% pointwise confidence intervals for the
dose-response function over the 61-point grid.}
\label{fig:sim-pointwise-width-dose}
\end{figure}

\begin{figure}[!p]
\centering
\includegraphics[width=\textwidth]{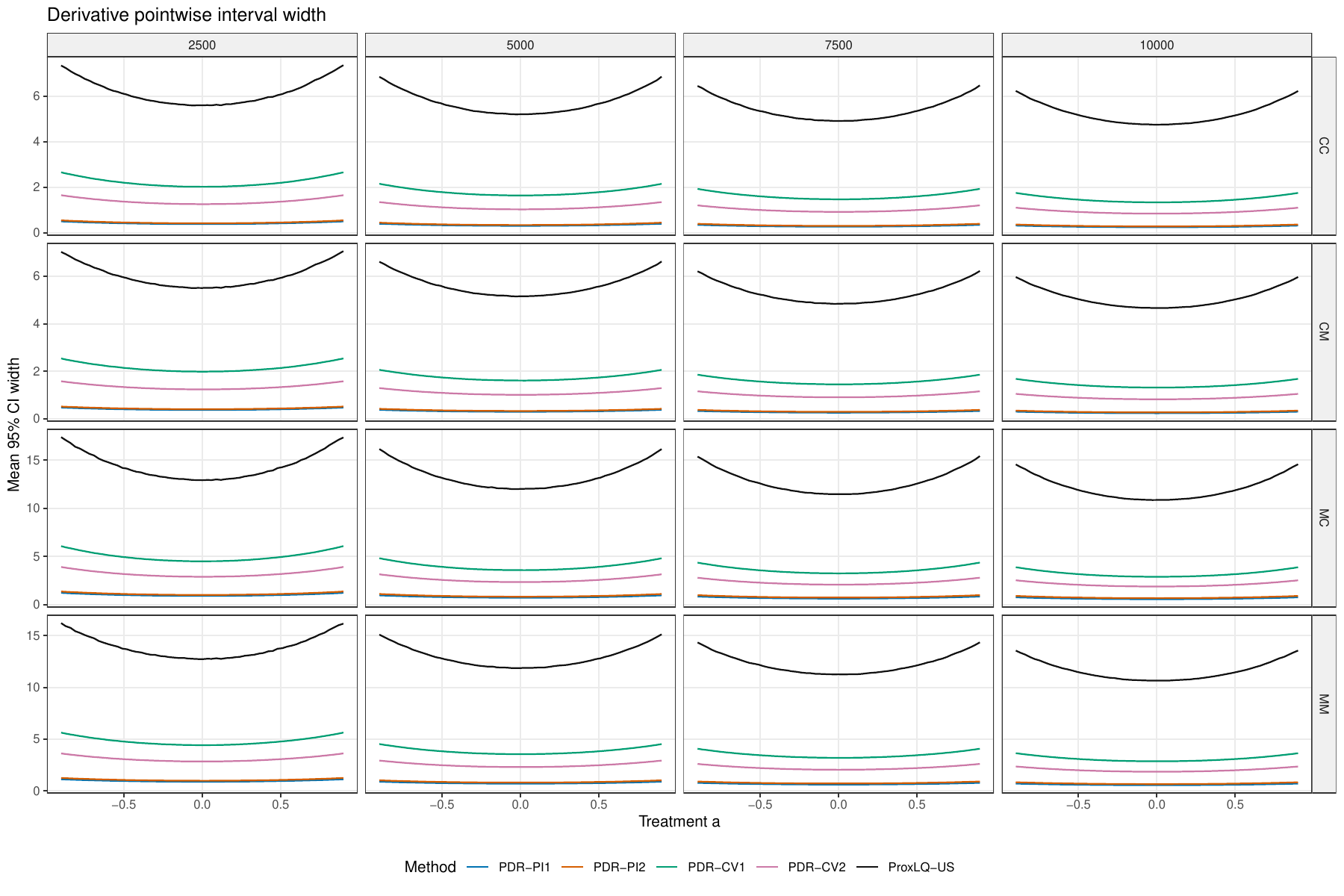}
\caption{Mean width of nominal 95\% pointwise confidence intervals for the
derivative over the 61-point grid.}
\label{fig:sim-pointwise-width-derivative}
\end{figure}

\begin{figure}[!p]
\centering
\includegraphics[width=\textwidth]{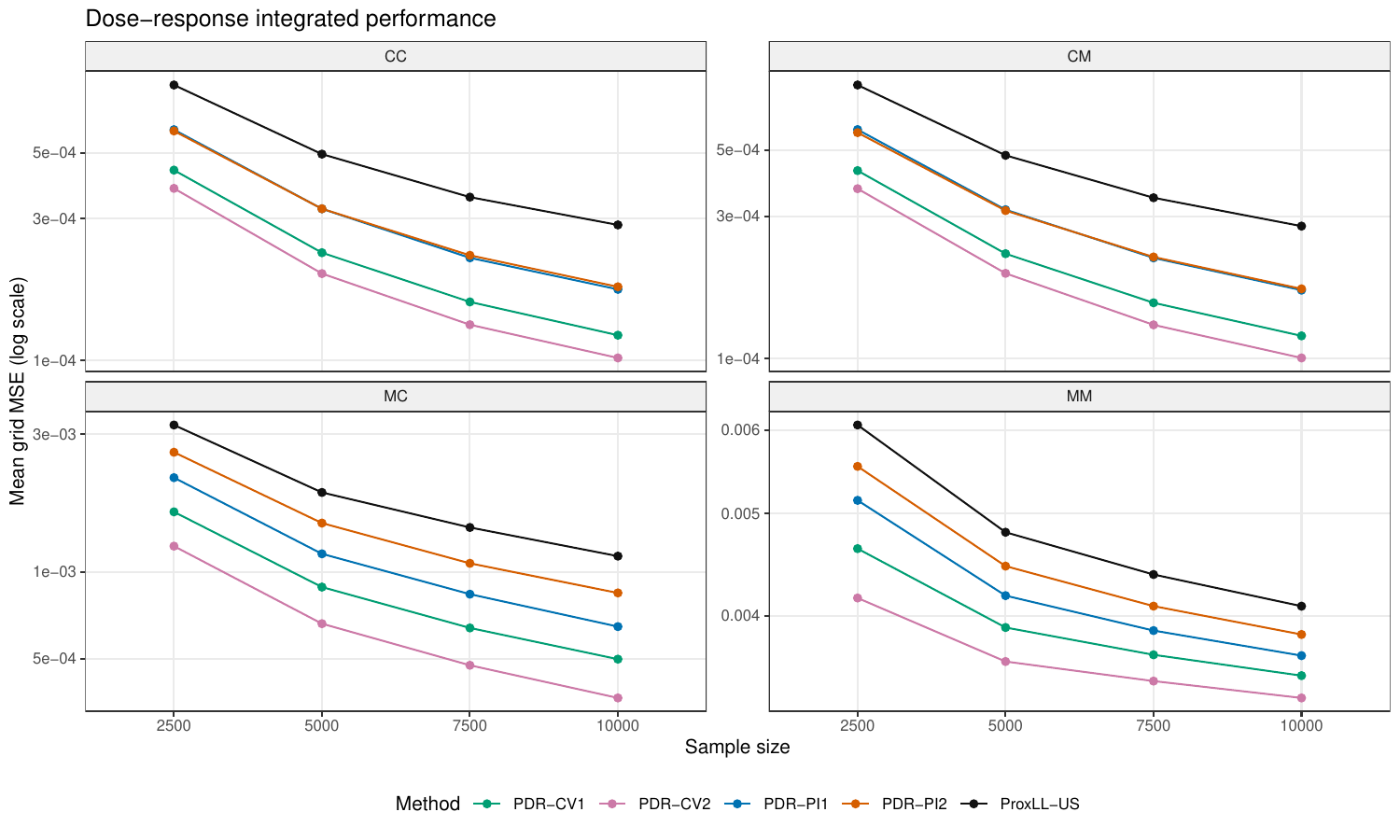}
\caption{Mean grid MSE for the dose-response function.  Each value is the
mean of coordinatewise MSE over the fixed 61-point grid.}
\label{fig:sim-integrated-error-dose}
\end{figure}

\begin{figure}[!p]
\centering
\includegraphics[width=\textwidth]{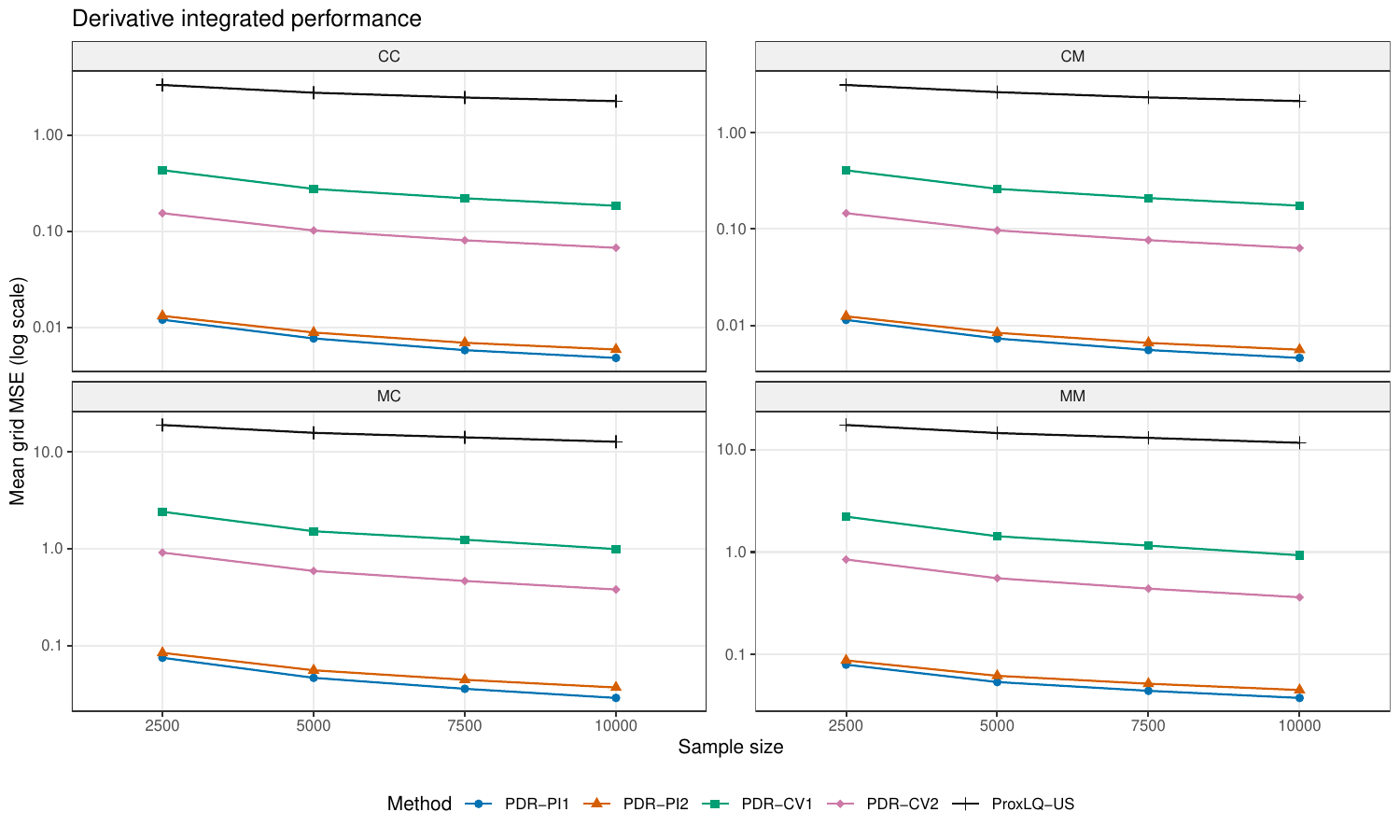}
\caption{Mean grid MSE for the derivative.  Each value is the
mean of coordinatewise MSE over the fixed 61-point grid.}
\label{fig:sim-integrated-error-derivative}
\end{figure}

% \begin{figure}[!p]
% \centering
% \includegraphics[width=\textwidth]{\simfigdir/pointwise_coverage_dose_response.pdf}
% \caption{Empirical coverage of nominal 95\% pointwise confidence intervals
% for the dose-response function in all four nuisance-model scenarios.  The
% dashed line is \(0.95\), and the gray region is \([.9438,.9560]\).}
% \label{fig:sim-pointwise-coverage-dose-all}
% \end{figure}

% \begin{figure}[!p]
% \centering
% \includegraphics[width=\textwidth]{\simfigdir/pointwise_coverage_derivative.pdf}
% \caption{Empirical coverage of nominal 95\% pointwise confidence intervals
% for the derivative in all four nuisance-model scenarios.  The dashed line is
% \(0.95\), and the gray region is \([.9438,.9560]\).}
% \label{fig:sim-pointwise-coverage-derivative-all}
% \end{figure}

\begin{landscape}
\fontsize{6}{7}\selectfont
\setlength{\tabcolsep}{1.0pt}
\renewcommand{\arraystretch}{1.00}
\begin{longtable}{@{}rrl*{9}{c}@{}}
\caption{Empirical coverage and mean width of nominal 95\% pointwise
confidence intervals for the dose-response function at nine prespecified
coordinates.  Each entry gives coverage, followed by mean width in brackets,
over \(5{,}000\) replications.}
\label{tab:sim-ci-nine-point}\\
\toprule
\(n\) & Scenario & Method
& \(-0.90\) & \(-0.69\) & \(-0.45\) & \(-0.21\) & \(0\)
& \(0.21\) & \(0.45\) & \(0.69\) & \(0.90\)\\
\midrule
\endfirsthead
\multicolumn{12}{c}{\tablename\ \thetable\ (continued)}\\
\toprule
\(n\) & Scenario & Method
& \(-0.90\) & \(-0.69\) & \(-0.45\) & \(-0.21\) & \(0\)
& \(0.21\) & \(0.45\) & \(0.69\) & \(0.90\)\\
\midrule
\endhead
\midrule
\multicolumn{12}{r}{Continued on next page}\\
\endfoot
\bottomrule
\endlastfoot
% Source: Derivative R study/joint_simulation/results/config_4e7d28a1/summary/ci_nine_point_performance.csv (configuration 4e7d28a1).
% Schema: n & scenario & method & nine pointwise coverages [mean CI widths], in the coordinate order given below \\
% Coordinates: -0.90, -0.69, -0.45, -0.21, 0.00, 0.21, 0.45, 0.69, 0.90.
% Scope: 5000 replications; scenarios CC, CM, and MC.
% Rounding: coverage and mean CI width are rounded to four decimal places.
% n = 2500
2500 & CC & PDR--PI1 & 0.9464 [0.1047] & 0.9476 [0.0965] & 0.9480 [0.0906] & 0.9472 [0.0874] & 0.9468 [0.0866] & 0.9488 [0.0874] & 0.9450 [0.0905] & 0.9424 [0.0965] & 0.9432 [0.1046] \\
2500 & CC & PDR--PI2 & 0.9496 [0.1082] & 0.9498 [0.0994] & 0.9494 [0.0932] & 0.9522 [0.0898] & 0.9492 [0.0889] & 0.9472 [0.0898] & 0.9522 [0.0931] & 0.9490 [0.0994] & 0.9474 [0.1081] \\
2500 & CC & PDR--CV1 & 0.9472 [0.0890] & 0.9470 [0.0830] & 0.9442 [0.0788] & 0.9416 [0.0765] & 0.9474 [0.0759] & 0.9458 [0.0765] & 0.9494 [0.0787] & 0.9446 [0.0830] & 0.9430 [0.0890] \\
2500 & CC & PDR--CV2 & 0.9464 [0.0805] & 0.9444 [0.0759] & 0.9426 [0.0726] & 0.9422 [0.0708] & 0.9398 [0.0704] & 0.9474 [0.0709] & 0.9436 [0.0726] & 0.9406 [0.0758] & 0.9380 [0.0804] \\
2500 & CC & ProxLL--US & 0.9472 [0.1276] & 0.9478 [0.1165] & 0.9486 [0.1083] & 0.9504 [0.1038] & 0.9508 [0.1028] & 0.9476 [0.1039] & 0.9452 [0.1080] & 0.9480 [0.1164] & 0.9502 [0.1274] \\
2500 & CM & PDR--PI1 & 0.9458 [0.1019] & 0.9474 [0.0949] & 0.9454 [0.0897] & 0.9484 [0.0870] & 0.9478 [0.0862] & 0.9452 [0.0869] & 0.9418 [0.0896] & 0.9420 [0.0948] & 0.9416 [0.1019] \\
2500 & CM & PDR--PI2 & 0.9488 [0.1046] & 0.9486 [0.0972] & 0.9482 [0.0918] & 0.9522 [0.0889] & 0.9500 [0.0880] & 0.9462 [0.0888] & 0.9498 [0.0917] & 0.9490 [0.0971] & 0.9480 [0.1046] \\
2500 & CM & PDR--CV1 & 0.9478 [0.0865] & 0.9478 [0.0815] & 0.9432 [0.0778] & 0.9418 [0.0758] & 0.9478 [0.0753] & 0.9482 [0.0758] & 0.9472 [0.0777] & 0.9436 [0.0814] & 0.9422 [0.0865] \\
2500 & CM & PDR--CV2 & 0.9474 [0.0784] & 0.9458 [0.0746] & 0.9416 [0.0718] & 0.9400 [0.0702] & 0.9402 [0.0698] & 0.9470 [0.0702] & 0.9446 [0.0717] & 0.9424 [0.0745] & 0.9372 [0.0784] \\
2500 & CM & ProxLL--US & 0.9480 [0.1241] & 0.9462 [0.1145] & 0.9454 [0.1074] & 0.9498 [0.1035] & 0.9522 [0.1024] & 0.9480 [0.1034] & 0.9464 [0.1071] & 0.9440 [0.1144] & 0.9486 [0.1241] \\
2500 & MC & PDR--PI1 & 0.9528 [0.2047] & 0.9456 [0.1814] & 0.9478 [0.1645] & 0.9418 [0.1552] & 0.9482 [0.1527] & 0.9466 [0.1551] & 0.9462 [0.1643] & 0.9484 [0.1814] & 0.9500 [0.2046] \\
2500 & MC & PDR--PI2 & 0.9498 [0.2419] & 0.9500 [0.2140] & 0.9510 [0.1937] & 0.9460 [0.1826] & 0.9492 [0.1796] & 0.9462 [0.1825] & 0.9524 [0.1935] & 0.9488 [0.2139] & 0.9430 [0.2415] \\
2500 & MC & PDR--CV1 & 0.9534 [0.1840] & 0.9496 [0.1632] & 0.9480 [0.1481] & 0.9418 [0.1398] & 0.9464 [0.1376] & 0.9456 [0.1397] & 0.9462 [0.1480] & 0.9470 [0.1632] & 0.9496 [0.1838] \\
2500 & MC & PDR--CV2 & 0.9528 [0.1560] & 0.9512 [0.1390] & 0.9474 [0.1266] & 0.9422 [0.1197] & 0.9446 [0.1179] & 0.9418 [0.1197] & 0.9438 [0.1265] & 0.9450 [0.1390] & 0.9496 [0.1560] \\
2500 & MC & ProxLL--US & 0.9510 [0.2578] & 0.9496 [0.2276] & 0.9560 [0.2058] & 0.9482 [0.1938] & 0.9544 [0.1907] & 0.9476 [0.1937] & 0.9520 [0.2054] & 0.9490 [0.2274] & 0.9498 [0.2572] \\
% n = 5000
5000 & CC & PDR--PI1 & 0.9444 [0.0778] & 0.9448 [0.0715] & 0.9448 [0.0670] & 0.9438 [0.0645] & 0.9446 [0.0639] & 0.9432 [0.0646] & 0.9470 [0.0669] & 0.9484 [0.0714] & 0.9468 [0.0777] \\
5000 & CC & PDR--PI2 & 0.9502 [0.0804] & 0.9498 [0.0737] & 0.9462 [0.0689] & 0.9456 [0.0664] & 0.9490 [0.0657] & 0.9472 [0.0664] & 0.9498 [0.0689] & 0.9514 [0.0737] & 0.9502 [0.0803] \\
5000 & CC & PDR--CV1 & 0.9472 [0.0653] & 0.9498 [0.0607] & 0.9462 [0.0575] & 0.9472 [0.0558] & 0.9472 [0.0553] & 0.9414 [0.0558] & 0.9434 [0.0575] & 0.9462 [0.0607] & 0.9498 [0.0653] \\
5000 & CC & PDR--CV2 & 0.9444 [0.0587] & 0.9478 [0.0552] & 0.9432 [0.0526] & 0.9476 [0.0513] & 0.9476 [0.0510] & 0.9432 [0.0513] & 0.9458 [0.0527] & 0.9472 [0.0552] & 0.9494 [0.0587] \\
5000 & CC & ProxLL--US & 0.9518 [0.0986] & 0.9492 [0.0896] & 0.9474 [0.0832] & 0.9460 [0.0796] & 0.9486 [0.0787] & 0.9488 [0.0797] & 0.9502 [0.0832] & 0.9488 [0.0896] & 0.9514 [0.0986] \\
5000 & CM & PDR--PI1 & 0.9472 [0.0757] & 0.9480 [0.0702] & 0.9422 [0.0662] & 0.9420 [0.0640] & 0.9464 [0.0635] & 0.9422 [0.0641] & 0.9428 [0.0662] & 0.9454 [0.0701] & 0.9434 [0.0756] \\
5000 & CM & PDR--PI2 & 0.9516 [0.0779] & 0.9494 [0.0721] & 0.9470 [0.0679] & 0.9448 [0.0656] & 0.9496 [0.0650] & 0.9462 [0.0657] & 0.9488 [0.0679] & 0.9508 [0.0721] & 0.9502 [0.0778] \\
5000 & CM & PDR--CV1 & 0.9492 [0.0635] & 0.9464 [0.0596] & 0.9454 [0.0567] & 0.9480 [0.0552] & 0.9466 [0.0548] & 0.9416 [0.0552] & 0.9452 [0.0567] & 0.9470 [0.0596] & 0.9498 [0.0635] \\
5000 & CM & PDR--CV2 & 0.9424 [0.0574] & 0.9448 [0.0543] & 0.9412 [0.0521] & 0.9478 [0.0510] & 0.9436 [0.0507] & 0.9412 [0.0510] & 0.9436 [0.0521] & 0.9478 [0.0543] & 0.9496 [0.0573] \\
5000 & CM & ProxLL--US & 0.9530 [0.0958] & 0.9518 [0.0880] & 0.9472 [0.0823] & 0.9474 [0.0791] & 0.9482 [0.0782] & 0.9472 [0.0791] & 0.9520 [0.0823] & 0.9466 [0.0879] & 0.9494 [0.0957] \\
5000 & MC & PDR--PI1 & 0.9504 [0.1534] & 0.9526 [0.1360] & 0.9438 [0.1234] & 0.9514 [0.1164] & 0.9450 [0.1147] & 0.9466 [0.1165] & 0.9434 [0.1233] & 0.9512 [0.1360] & 0.9570 [0.1532] \\
5000 & MC & PDR--PI2 & 0.9516 [0.1825] & 0.9484 [0.1616] & 0.9496 [0.1464] & 0.9528 [0.1381] & 0.9516 [0.1360] & 0.9456 [0.1382] & 0.9462 [0.1464] & 0.9570 [0.1617] & 0.9492 [0.1824] \\
5000 & MC & PDR--CV1 & 0.9528 [0.1375] & 0.9548 [0.1221] & 0.9460 [0.1108] & 0.9546 [0.1047] & 0.9474 [0.1032] & 0.9480 [0.1048] & 0.9460 [0.1108] & 0.9548 [0.1222] & 0.9582 [0.1374] \\
5000 & MC & PDR--CV2 & 0.9528 [0.1164] & 0.9556 [0.1036] & 0.9432 [0.0943] & 0.9494 [0.0893] & 0.9468 [0.0880] & 0.9468 [0.0893] & 0.9452 [0.0943] & 0.9536 [0.1037] & 0.9584 [0.1162] \\
5000 & MC & ProxLL--US & 0.9536 [0.1994] & 0.9540 [0.1761] & 0.9490 [0.1596] & 0.9508 [0.1503] & 0.9490 [0.1482] & 0.9504 [0.1505] & 0.9480 [0.1595] & 0.9602 [0.1762] & 0.9532 [0.1990] \\
% n = 7500
7500 & CC & PDR--PI1 & 0.9444 [0.0651] & 0.9476 [0.0597] & 0.9554 [0.0559] & 0.9508 [0.0539] & 0.9508 [0.0533] & 0.9524 [0.0539] & 0.9476 [0.0559] & 0.9466 [0.0598] & 0.9464 [0.0651] \\
7500 & CC & PDR--PI2 & 0.9476 [0.0676] & 0.9520 [0.0620] & 0.9554 [0.0579] & 0.9480 [0.0557] & 0.9490 [0.0551] & 0.9514 [0.0557] & 0.9486 [0.0579] & 0.9498 [0.0620] & 0.9464 [0.0676] \\
7500 & CC & PDR--CV1 & 0.9410 [0.0544] & 0.9480 [0.0506] & 0.9548 [0.0478] & 0.9524 [0.0464] & 0.9478 [0.0460] & 0.9512 [0.0463] & 0.9496 [0.0478] & 0.9502 [0.0506] & 0.9496 [0.0544] \\
7500 & CC & PDR--CV2 & 0.9414 [0.0487] & 0.9480 [0.0457] & 0.9502 [0.0436] & 0.9544 [0.0425] & 0.9460 [0.0422] & 0.9494 [0.0425] & 0.9508 [0.0436] & 0.9512 [0.0457] & 0.9464 [0.0487] \\
7500 & CC & ProxLL--US & 0.9472 [0.0843] & 0.9508 [0.0766] & 0.9524 [0.0709] & 0.9502 [0.0679] & 0.9568 [0.0671] & 0.9504 [0.0679] & 0.9504 [0.0710] & 0.9500 [0.0766] & 0.9530 [0.0844] \\
7500 & CM & PDR--PI1 & 0.9428 [0.0634] & 0.9504 [0.0587] & 0.9516 [0.0553] & 0.9478 [0.0535] & 0.9498 [0.0530] & 0.9512 [0.0535] & 0.9470 [0.0554] & 0.9454 [0.0587] & 0.9460 [0.0634] \\
7500 & CM & PDR--PI2 & 0.9458 [0.0655] & 0.9510 [0.0606] & 0.9550 [0.0570] & 0.9488 [0.0551] & 0.9498 [0.0546] & 0.9510 [0.0551] & 0.9484 [0.0570] & 0.9496 [0.0606] & 0.9472 [0.0655] \\
7500 & CM & PDR--CV1 & 0.9424 [0.0530] & 0.9480 [0.0496] & 0.9540 [0.0472] & 0.9494 [0.0459] & 0.9462 [0.0456] & 0.9518 [0.0459] & 0.9496 [0.0472] & 0.9500 [0.0496] & 0.9468 [0.0530] \\
7500 & CM & PDR--CV2 & 0.9394 [0.0476] & 0.9484 [0.0450] & 0.9484 [0.0431] & 0.9536 [0.0421] & 0.9438 [0.0419] & 0.9472 [0.0421] & 0.9500 [0.0431] & 0.9486 [0.0450] & 0.9426 [0.0476] \\
7500 & CM & ProxLL--US & 0.9422 [0.0820] & 0.9502 [0.0751] & 0.9532 [0.0701] & 0.9498 [0.0675] & 0.9538 [0.0667] & 0.9490 [0.0674] & 0.9516 [0.0702] & 0.9516 [0.0752] & 0.9548 [0.0820] \\
7500 & MC & PDR--PI1 & 0.9502 [0.1305] & 0.9490 [0.1157] & 0.9496 [0.1050] & 0.9446 [0.0992] & 0.9480 [0.0976] & 0.9488 [0.0992] & 0.9482 [0.1050] & 0.9500 [0.1158] & 0.9566 [0.1305] \\
7500 & MC & PDR--PI2 & 0.9492 [0.1549] & 0.9506 [0.1372] & 0.9484 [0.1244] & 0.9466 [0.1174] & 0.9464 [0.1156] & 0.9528 [0.1174] & 0.9532 [0.1244] & 0.9518 [0.1373] & 0.9556 [0.1549] \\
7500 & MC & PDR--CV1 & 0.9566 [0.1165] & 0.9510 [0.1035] & 0.9494 [0.0940] & 0.9472 [0.0888] & 0.9470 [0.0874] & 0.9464 [0.0888] & 0.9492 [0.0939] & 0.9458 [0.1035] & 0.9564 [0.1165] \\
7500 & MC & PDR--CV2 & 0.9568 [0.0984] & 0.9524 [0.0877] & 0.9456 [0.0798] & 0.9472 [0.0755] & 0.9434 [0.0743] & 0.9434 [0.0755] & 0.9500 [0.0798] & 0.9486 [0.0877] & 0.9568 [0.0985] \\
7500 & MC & ProxLL--US & 0.9546 [0.1731] & 0.9498 [0.1532] & 0.9526 [0.1388] & 0.9538 [0.1309] & 0.9470 [0.1289] & 0.9532 [0.1309] & 0.9574 [0.1387] & 0.9540 [0.1532] & 0.9564 [0.1733] \\
% n = 10000
10000 & CC & PDR--PI1 & 0.9524 [0.0577] & 0.9454 [0.0529] & 0.9432 [0.0494] & 0.9518 [0.0476] & 0.9446 [0.0471] & 0.9502 [0.0476] & 0.9460 [0.0494] & 0.9500 [0.0529] & 0.9458 [0.0577] \\
10000 & CC & PDR--PI2 & 0.9536 [0.0599] & 0.9486 [0.0548] & 0.9458 [0.0512] & 0.9552 [0.0492] & 0.9472 [0.0487] & 0.9570 [0.0492] & 0.9514 [0.0511] & 0.9526 [0.0548] & 0.9466 [0.0599] \\
10000 & CC & PDR--CV1 & 0.9504 [0.0481] & 0.9472 [0.0446] & 0.9502 [0.0421] & 0.9540 [0.0408] & 0.9482 [0.0405] & 0.9480 [0.0408] & 0.9472 [0.0421] & 0.9502 [0.0446] & 0.9502 [0.0481] \\
10000 & CC & PDR--CV2 & 0.9472 [0.0430] & 0.9442 [0.0403] & 0.9508 [0.0383] & 0.9532 [0.0373] & 0.9490 [0.0370] & 0.9442 [0.0373] & 0.9452 [0.0383] & 0.9454 [0.0403] & 0.9512 [0.0430] \\
10000 & CC & ProxLL--US & 0.9580 [0.0760] & 0.9536 [0.0689] & 0.9464 [0.0638] & 0.9486 [0.0610] & 0.9478 [0.0603] & 0.9536 [0.0610] & 0.9528 [0.0638] & 0.9518 [0.0689] & 0.9492 [0.0759] \\
10000 & CM & PDR--PI1 & 0.9526 [0.0562] & 0.9436 [0.0520] & 0.9438 [0.0489] & 0.9516 [0.0473] & 0.9444 [0.0469] & 0.9518 [0.0473] & 0.9466 [0.0489] & 0.9488 [0.0520] & 0.9448 [0.0562] \\
10000 & CM & PDR--PI2 & 0.9550 [0.0580] & 0.9466 [0.0536] & 0.9444 [0.0504] & 0.9540 [0.0487] & 0.9474 [0.0482] & 0.9558 [0.0487] & 0.9506 [0.0504] & 0.9526 [0.0536] & 0.9492 [0.0580] \\
10000 & CM & PDR--CV1 & 0.9504 [0.0468] & 0.9484 [0.0438] & 0.9504 [0.0416] & 0.9522 [0.0405] & 0.9484 [0.0401] & 0.9484 [0.0404] & 0.9464 [0.0416] & 0.9500 [0.0438] & 0.9486 [0.0468] \\
10000 & CM & PDR--CV2 & 0.9472 [0.0420] & 0.9446 [0.0397] & 0.9482 [0.0380] & 0.9514 [0.0370] & 0.9498 [0.0368] & 0.9440 [0.0370] & 0.9454 [0.0379] & 0.9464 [0.0397] & 0.9494 [0.0420] \\
10000 & CM & ProxLL--US & 0.9546 [0.0738] & 0.9532 [0.0676] & 0.9464 [0.0630] & 0.9488 [0.0606] & 0.9484 [0.0600] & 0.9530 [0.0606] & 0.9508 [0.0631] & 0.9542 [0.0676] & 0.9474 [0.0738] \\
10000 & MC & PDR--PI1 & 0.9522 [0.1157] & 0.9518 [0.1026] & 0.9492 [0.0930] & 0.9476 [0.0879] & 0.9502 [0.0865] & 0.9532 [0.0879] & 0.9522 [0.0930] & 0.9508 [0.1025] & 0.9580 [0.1156] \\
10000 & MC & PDR--PI2 & 0.9522 [0.1380] & 0.9516 [0.1223] & 0.9510 [0.1108] & 0.9454 [0.1047] & 0.9540 [0.1030] & 0.9512 [0.1046] & 0.9524 [0.1108] & 0.9478 [0.1222] & 0.9616 [0.1379] \\
10000 & MC & PDR--CV1 & 0.9542 [0.1035] & 0.9496 [0.0919] & 0.9518 [0.0834] & 0.9450 [0.0788] & 0.9528 [0.0776] & 0.9512 [0.0788] & 0.9460 [0.0833] & 0.9510 [0.0918] & 0.9584 [0.1034] \\
10000 & MC & PDR--CV2 & 0.9536 [0.0871] & 0.9526 [0.0775] & 0.9504 [0.0705] & 0.9460 [0.0667] & 0.9528 [0.0657] & 0.9472 [0.0667] & 0.9414 [0.0704] & 0.9488 [0.0774] & 0.9552 [0.0870] \\
10000 & MC & ProxLL--US & 0.9546 [0.1554] & 0.9514 [0.1375] & 0.9524 [0.1244] & 0.9512 [0.1175] & 0.9540 [0.1157] & 0.9538 [0.1175] & 0.9534 [0.1245] & 0.9542 [0.1374] & 0.9586 [0.1552] \\

\end{longtable}
\end{landscape}

\begin{landscape}
\fontsize{6}{7}\selectfont
\setlength{\tabcolsep}{1.0pt}
\renewcommand{\arraystretch}{1.00}
\begin{longtable}{@{}rrl*{9}{c}@{}}
\caption{Empirical coverage and mean width of nominal 95\% pointwise
confidence intervals for the derivative at nine prespecified coordinates.
Each entry gives coverage, followed by mean width in brackets, over
\(5{,}000\) replications.}
\label{tab:sim-ci-nine-point-derivative}\\
\toprule
\(n\) & Scenario & Method
& \(-0.90\) & \(-0.69\) & \(-0.45\) & \(-0.21\) & \(0\)
& \(0.21\) & \(0.45\) & \(0.69\) & \(0.90\)\\
\midrule
\endfirsthead
\multicolumn{12}{c}{\tablename\ \thetable\ (continued)}\\
\toprule
\(n\) & Scenario & Method
& \(-0.90\) & \(-0.69\) & \(-0.45\) & \(-0.21\) & \(0\)
& \(0.21\) & \(0.45\) & \(0.69\) & \(0.90\)\\
\midrule
\endhead
\midrule
\multicolumn{12}{r}{Continued on next page}\\
\endfoot
\bottomrule
\endlastfoot
% Sources: frozen PDR results in configuration 4e7d28a1 and ProxLQ-US results in configuration 0bebf0fb.
% Schema: n & scenario & method & nine pointwise coverages [mean CI widths], in the coordinate order given below \\
% Coordinates: -0.90, -0.69, -0.45, -0.21, 0.00, 0.21, 0.45, 0.69, 0.90.
% Scope: 5000 replications; scenarios CC, CM, and MC.
% Rounding: coverage and mean CI width are rounded to four decimal places.
% n = 2500
2500 & CC & PDR--PI1 & 0.9466 [0.4923] & 0.9558 [0.4408] & 0.9572 [0.4024] & 0.9556 [0.3815] & 0.9562 [0.3758] & 0.9534 [0.3813] & 0.9530 [0.4020] & 0.9566 [0.4400] & 0.9508 [0.4925] \\
2500 & CC & PDR--PI2 & 0.9444 [0.5365] & 0.9512 [0.4802] & 0.9550 [0.4383] & 0.9550 [0.4155] & 0.9532 [0.4092] & 0.9498 [0.4153] & 0.9496 [0.4377] & 0.9520 [0.4793] & 0.9502 [0.5366] \\
2500 & CC & PDR--CV1 & 0.9432 [2.6471] & 0.9478 [2.3625] & 0.9476 [2.1602] & 0.9488 [2.0446] & 0.9462 [2.0134] & 0.9444 [2.0417] & 0.9468 [2.1528] & 0.9486 [2.3664] & 0.9438 [2.6442] \\
2500 & CC & PDR--CV2 & 0.9452 [1.6433] & 0.9448 [1.4678] & 0.9512 [1.3398] & 0.9498 [1.2694] & 0.9472 [1.2499] & 0.9422 [1.2670] & 0.9448 [1.3373] & 0.9494 [1.4675] & 0.9526 [1.6418] \\
2500 & CC & ProxLQ--US & 0.9424 [7.3350] & 0.9354 [6.5645] & 0.9414 [5.9833] & 0.9444 [5.6724] & 0.9418 [5.5868] & 0.9426 [5.6714] & 0.9486 [5.9613] & 0.9396 [6.5848] & 0.9348 [7.3513] \\
2500 & CM & PDR--PI1 & 0.9472 [0.4715] & 0.9542 [0.4273] & 0.9556 [0.3940] & 0.9550 [0.3757] & 0.9560 [0.3707] & 0.9538 [0.3755] & 0.9530 [0.3937] & 0.9528 [0.4268] & 0.9520 [0.4718] \\
2500 & CM & PDR--PI2 & 0.9460 [0.5127] & 0.9490 [0.4645] & 0.9548 [0.4283] & 0.9536 [0.4083] & 0.9528 [0.4028] & 0.9502 [0.4081] & 0.9484 [0.4277] & 0.9496 [0.4638] & 0.9492 [0.5130] \\
2500 & CM & PDR--CV1 & 0.9386 [2.5374] & 0.9452 [2.2913] & 0.9478 [2.1142] & 0.9480 [2.0119] & 0.9454 [1.9849] & 0.9450 [2.0096] & 0.9416 [2.1072] & 0.9530 [2.2950] & 0.9438 [2.5360] \\
2500 & CM & PDR--CV2 & 0.9428 [1.5783] & 0.9468 [1.4265] & 0.9536 [1.3150] & 0.9478 [1.2531] & 0.9464 [1.2354] & 0.9430 [1.2511] & 0.9466 [1.3121] & 0.9518 [1.4264] & 0.9514 [1.5777] \\
2500 & CM & ProxLQ--US & 0.9404 [7.0197] & 0.9370 [6.3604] & 0.9396 [5.8544] & 0.9434 [5.5757] & 0.9430 [5.5033] & 0.9424 [5.5751] & 0.9464 [5.8330] & 0.9372 [6.3819] & 0.9372 [7.0515] \\
2500 & MC & PDR--PI1 & 0.9488 [1.2280] & 0.9456 [1.0888] & 0.9514 [0.9849] & 0.9580 [0.9296] & 0.9450 [0.9144] & 0.9564 [0.9294] & 0.9546 [0.9845] & 0.9534 [1.0875] & 0.9472 [1.2278] \\
2500 & MC & PDR--PI2 & 0.9460 [1.3668] & 0.9456 [1.2112] & 0.9516 [1.0951] & 0.9526 [1.0334] & 0.9436 [1.0165] & 0.9516 [1.0332] & 0.9540 [1.0948] & 0.9492 [1.2095] & 0.9452 [1.3667] \\
2500 & MC & PDR--CV1 & 0.9444 [6.0752] & 0.9432 [5.3629] & 0.9446 [4.8588] & 0.9482 [4.5843] & 0.9398 [4.5001] & 0.9502 [4.5722] & 0.9424 [4.8627] & 0.9400 [5.3691] & 0.9444 [6.0814] \\
2500 & MC & PDR--CV2 & 0.9456 [3.9139] & 0.9450 [3.4561] & 0.9422 [3.1266] & 0.9476 [2.9517] & 0.9434 [2.8987] & 0.9520 [2.9454] & 0.9418 [3.1298] & 0.9420 [3.4545] & 0.9478 [3.9176] \\
2500 & MC & ProxLQ--US & 0.9436 [17.3770] & 0.9396 [15.4756] & 0.9470 [13.9296] & 0.9416 [13.1445] & 0.9426 [12.9114] & 0.9394 [13.1148] & 0.9386 [13.9137] & 0.9356 [15.4305] & 0.9314 [17.3384] \\
% n = 5000
5000 & CC & PDR--PI1 & 0.9498 [0.3904] & 0.9558 [0.3498] & 0.9522 [0.3195] & 0.9490 [0.3032] & 0.9502 [0.2988] & 0.9540 [0.3031] & 0.9534 [0.3193] & 0.9532 [0.3498] & 0.9496 [0.3905] \\
5000 & CC & PDR--PI2 & 0.9492 [0.4356] & 0.9528 [0.3902] & 0.9492 [0.3563] & 0.9472 [0.3381] & 0.9464 [0.3333] & 0.9530 [0.3380] & 0.9502 [0.3562] & 0.9514 [0.3901] & 0.9494 [0.4357] \\
5000 & CC & PDR--CV1 & 0.9480 [2.1443] & 0.9434 [1.9148] & 0.9450 [1.7490] & 0.9470 [1.6576] & 0.9484 [1.6311] & 0.9510 [1.6565] & 0.9468 [1.7489] & 0.9456 [1.9133] & 0.9496 [2.1414] \\
5000 & CC & PDR--CV2 & 0.9462 [1.3400] & 0.9458 [1.1981] & 0.9484 [1.0935] & 0.9512 [1.0366] & 0.9470 [1.0208] & 0.9542 [1.0367] & 0.9456 [1.0936] & 0.9528 [1.1976] & 0.9492 [1.3391] \\
5000 & CC & ProxLQ--US & 0.9428 [6.8388] & 0.9436 [6.1036] & 0.9470 [5.5658] & 0.9490 [5.2672] & 0.9480 [5.1942] & 0.9388 [5.2793] & 0.9408 [5.5621] & 0.9430 [6.0897] & 0.9378 [6.8479] \\
5000 & CM & PDR--PI1 & 0.9472 [0.3756] & 0.9520 [0.3401] & 0.9514 [0.3134] & 0.9516 [0.2989] & 0.9494 [0.2951] & 0.9540 [0.2989] & 0.9550 [0.3132] & 0.9518 [0.3400] & 0.9498 [0.3754] \\
5000 & CM & PDR--PI2 & 0.9470 [0.4180] & 0.9518 [0.3784] & 0.9496 [0.3486] & 0.9484 [0.3325] & 0.9444 [0.3282] & 0.9518 [0.3324] & 0.9526 [0.3484] & 0.9486 [0.3782] & 0.9482 [0.4179] \\
5000 & CM & PDR--CV1 & 0.9468 [2.0580] & 0.9470 [1.8604] & 0.9452 [1.7133] & 0.9482 [1.6326] & 0.9518 [1.6090] & 0.9514 [1.6314] & 0.9468 [1.7130] & 0.9462 [1.8579] & 0.9518 [2.0580] \\
5000 & CM & PDR--CV2 & 0.9466 [1.2904] & 0.9458 [1.1660] & 0.9490 [1.0734] & 0.9478 [1.0227] & 0.9478 [1.0090] & 0.9534 [1.0231] & 0.9450 [1.0735] & 0.9496 [1.1652] & 0.9504 [1.2893] \\
5000 & CM & ProxLQ--US & 0.9416 [6.6042] & 0.9438 [5.9637] & 0.9460 [5.4796] & 0.9486 [5.2194] & 0.9504 [5.1538] & 0.9406 [5.2245] & 0.9434 [5.4816] & 0.9426 [5.9508] & 0.9362 [6.6179] \\
5000 & MC & PDR--PI1 & 0.9560 [0.9704] & 0.9516 [0.8607] & 0.9542 [0.7796] & 0.9524 [0.7367] & 0.9530 [0.7252] & 0.9496 [0.7370] & 0.9534 [0.7803] & 0.9520 [0.8607] & 0.9516 [0.9700] \\
5000 & MC & PDR--PI2 & 0.9558 [1.1123] & 0.9498 [0.9861] & 0.9534 [0.8931] & 0.9502 [0.8438] & 0.9544 [0.8308] & 0.9474 [0.8440] & 0.9516 [0.8937] & 0.9488 [0.9862] & 0.9498 [1.1118] \\
5000 & MC & PDR--CV1 & 0.9468 [4.8227] & 0.9424 [4.2621] & 0.9498 [3.8561] & 0.9468 [3.6390] & 0.9484 [3.5838] & 0.9440 [3.6415] & 0.9520 [3.8558] & 0.9496 [4.2670] & 0.9452 [4.8134] \\
5000 & MC & PDR--CV2 & 0.9488 [3.1520] & 0.9486 [2.7925] & 0.9518 [2.5264] & 0.9450 [2.3837] & 0.9458 [2.3480] & 0.9444 [2.3862] & 0.9472 [2.5280] & 0.9504 [2.7933] & 0.9422 [3.1515] \\
5000 & MC & ProxLQ--US & 0.9446 [16.1396] & 0.9478 [14.2455] & 0.9452 [12.9423] & 0.9468 [12.1718] & 0.9422 [12.0051] & 0.9546 [12.2191] & 0.9482 [12.9127] & 0.9454 [14.2937] & 0.9422 [16.1366] \\
% n = 7500
7500 & CC & PDR--PI1 & 0.9522 [0.3408] & 0.9512 [0.3052] & 0.9532 [0.2789] & 0.9506 [0.2647] & 0.9492 [0.2608] & 0.9478 [0.2647] & 0.9546 [0.2788] & 0.9572 [0.3053] & 0.9496 [0.3407] \\
7500 & CC & PDR--PI2 & 0.9508 [0.3867] & 0.9486 [0.3463] & 0.9536 [0.3163] & 0.9504 [0.3002] & 0.9462 [0.2957] & 0.9486 [0.3001] & 0.9510 [0.3162] & 0.9560 [0.3464] & 0.9506 [0.3865] \\
7500 & CC & PDR--CV1 & 0.9496 [1.9206] & 0.9490 [1.7162] & 0.9498 [1.5679] & 0.9472 [1.4852] & 0.9522 [1.4635] & 0.9428 [1.4851] & 0.9508 [1.5679] & 0.9520 [1.7183] & 0.9474 [1.9191] \\
7500 & CC & PDR--CV2 & 0.9532 [1.1972] & 0.9458 [1.0701] & 0.9470 [0.9773] & 0.9506 [0.9256] & 0.9506 [0.9126] & 0.9440 [0.9256] & 0.9504 [0.9775] & 0.9560 [1.0704] & 0.9510 [1.1967] \\
7500 & CC & ProxLQ--US & 0.9502 [6.4368] & 0.9466 [5.7466] & 0.9480 [5.2574] & 0.9412 [4.9778] & 0.9492 [4.9053] & 0.9498 [4.9779] & 0.9482 [5.2544] & 0.9508 [5.7573] & 0.9456 [6.4650] \\
7500 & CM & PDR--PI1 & 0.9542 [0.3281] & 0.9494 [0.2970] & 0.9524 [0.2738] & 0.9496 [0.2611] & 0.9498 [0.2576] & 0.9472 [0.2611] & 0.9538 [0.2736] & 0.9566 [0.2970] & 0.9526 [0.3280] \\
7500 & CM & PDR--PI2 & 0.9512 [0.3713] & 0.9484 [0.3360] & 0.9520 [0.3096] & 0.9498 [0.2953] & 0.9466 [0.2913] & 0.9468 [0.2952] & 0.9520 [0.3095] & 0.9536 [0.3360] & 0.9508 [0.3711] \\
7500 & CM & PDR--CV1 & 0.9502 [1.8543] & 0.9488 [1.6742] & 0.9472 [1.5429] & 0.9456 [1.4690] & 0.9520 [1.4492] & 0.9416 [1.4687] & 0.9518 [1.5427] & 0.9540 [1.6756] & 0.9506 [1.8530] \\
7500 & CM & PDR--CV2 & 0.9494 [1.1550] & 0.9444 [1.0433] & 0.9468 [0.9606] & 0.9526 [0.9147] & 0.9506 [0.9030] & 0.9408 [0.9145] & 0.9510 [0.9610] & 0.9576 [1.0431] & 0.9496 [1.1541] \\
7500 & CM & ProxLQ--US & 0.9492 [6.2026] & 0.9442 [5.5886] & 0.9480 [5.1537] & 0.9414 [4.9058] & 0.9496 [4.8394] & 0.9498 [4.9014] & 0.9476 [5.1490] & 0.9512 [5.5980] & 0.9498 [6.2187] \\
7500 & MC & PDR--PI1 & 0.9502 [0.8484] & 0.9538 [0.7531] & 0.9564 [0.6827] & 0.9484 [0.6450] & 0.9536 [0.6348] & 0.9470 [0.6451] & 0.9432 [0.6828] & 0.9572 [0.7532] & 0.9456 [0.8482] \\
7500 & MC & PDR--PI2 & 0.9492 [0.9875] & 0.9510 [0.8763] & 0.9510 [0.7941] & 0.9466 [0.7501] & 0.9488 [0.7382] & 0.9460 [0.7502] & 0.9420 [0.7942] & 0.9560 [0.8763] & 0.9464 [0.9872] \\
7500 & MC & PDR--CV1 & 0.9468 [4.3667] & 0.9404 [3.8687] & 0.9420 [3.5035] & 0.9426 [3.3069] & 0.9456 [3.2515] & 0.9388 [3.3056] & 0.9466 [3.5051] & 0.9432 [3.8703] & 0.9384 [4.3681] \\
7500 & MC & PDR--CV2 & 0.9452 [2.7973] & 0.9432 [2.4796] & 0.9458 [2.2453] & 0.9404 [2.1191] & 0.9440 [2.0854] & 0.9442 [2.1178] & 0.9454 [2.2464] & 0.9468 [2.4795] & 0.9422 [2.7984] \\
7500 & MC & ProxLQ--US & 0.9464 [15.3589] & 0.9426 [13.6175] & 0.9482 [12.3205] & 0.9506 [11.6193] & 0.9460 [11.4704] & 0.9480 [11.6431] & 0.9476 [12.3336] & 0.9402 [13.5963] & 0.9464 [15.4090] \\
% n = 10000
10000 & CC & PDR--PI1 & 0.9514 [0.3102] & 0.9532 [0.2778] & 0.9542 [0.2538] & 0.9578 [0.2409] & 0.9524 [0.2375] & 0.9506 [0.2408] & 0.9466 [0.2536] & 0.9444 [0.2776] & 0.9496 [0.3099] \\
10000 & CC & PDR--PI2 & 0.9488 [0.3555] & 0.9506 [0.3183] & 0.9514 [0.2907] & 0.9534 [0.2760] & 0.9492 [0.2720] & 0.9484 [0.2759] & 0.9464 [0.2906] & 0.9426 [0.3181] & 0.9470 [0.3553] \\
10000 & CC & PDR--CV1 & 0.9464 [1.7448] & 0.9502 [1.5600] & 0.9524 [1.4223] & 0.9508 [1.3497] & 0.9508 [1.3299] & 0.9512 [1.3491] & 0.9488 [1.4224] & 0.9464 [1.5583] & 0.9508 [1.7434] \\
10000 & CC & PDR--CV2 & 0.9462 [1.0953] & 0.9522 [0.9788] & 0.9506 [0.8933] & 0.9488 [0.8476] & 0.9530 [0.8351] & 0.9492 [0.8471] & 0.9440 [0.8922] & 0.9496 [0.9784] & 0.9472 [1.0942] \\
10000 & CC & ProxLQ--US & 0.9452 [6.2222] & 0.9446 [5.5666] & 0.9494 [5.0793] & 0.9504 [4.8132] & 0.9454 [4.7414] & 0.9470 [4.8015] & 0.9408 [5.0738] & 0.9460 [5.5648] & 0.9460 [6.2185] \\
10000 & CM & PDR--PI1 & 0.9506 [0.2988] & 0.9572 [0.2703] & 0.9572 [0.2490] & 0.9576 [0.2376] & 0.9530 [0.2345] & 0.9490 [0.2375] & 0.9472 [0.2489] & 0.9446 [0.2702] & 0.9472 [0.2987] \\
10000 & CM & PDR--PI2 & 0.9486 [0.3416] & 0.9528 [0.3089] & 0.9538 [0.2846] & 0.9534 [0.2715] & 0.9486 [0.2679] & 0.9476 [0.2714] & 0.9466 [0.2844] & 0.9410 [0.3089] & 0.9464 [0.3415] \\
10000 & CM & PDR--CV1 & 0.9462 [1.6796] & 0.9452 [1.5178] & 0.9516 [1.3960] & 0.9484 [1.3308] & 0.9484 [1.3132] & 0.9528 [1.3304] & 0.9476 [1.3960] & 0.9478 [1.5164] & 0.9488 [1.6799] \\
10000 & CM & PDR--CV2 & 0.9490 [1.0484] & 0.9494 [0.9467] & 0.9498 [0.8714] & 0.9492 [0.8310] & 0.9518 [0.8199] & 0.9498 [0.8306] & 0.9438 [0.8706] & 0.9452 [0.9467] & 0.9476 [1.0478] \\
10000 & CM & ProxLQ--US & 0.9488 [5.9615] & 0.9468 [5.3883] & 0.9492 [4.9567] & 0.9482 [4.7205] & 0.9450 [4.6570] & 0.9486 [4.7096] & 0.9394 [4.9545] & 0.9420 [5.3899] & 0.9466 [5.9647] \\
10000 & MC & PDR--PI1 & 0.9498 [0.7676] & 0.9570 [0.6810] & 0.9556 [0.6175] & 0.9524 [0.5836] & 0.9548 [0.5744] & 0.9526 [0.5835] & 0.9556 [0.6171] & 0.9518 [0.6808] & 0.9508 [0.7667] \\
10000 & MC & PDR--PI2 & 0.9468 [0.9087] & 0.9550 [0.8058] & 0.9514 [0.7303] & 0.9490 [0.6900] & 0.9552 [0.6792] & 0.9514 [0.6900] & 0.9508 [0.7298] & 0.9492 [0.8053] & 0.9492 [0.9075] \\
10000 & MC & PDR--CV1 & 0.9456 [3.8849] & 0.9440 [3.4416] & 0.9508 [3.1143] & 0.9464 [2.9447] & 0.9494 [2.8972] & 0.9488 [2.9436] & 0.9500 [3.1158] & 0.9458 [3.4389] & 0.9480 [3.8864] \\
10000 & MC & PDR--CV2 & 0.9450 [2.5316] & 0.9428 [2.2431] & 0.9492 [2.0304] & 0.9450 [1.9185] & 0.9456 [1.8871] & 0.9480 [1.9183] & 0.9476 [2.0294] & 0.9458 [2.2409] & 0.9468 [2.5295] \\
10000 & MC & ProxLQ--US & 0.9506 [14.5446] & 0.9484 [12.8822] & 0.9508 [11.6901] & 0.9506 [11.0143] & 0.9476 [10.8623] & 0.9476 [11.0185] & 0.9428 [11.6918] & 0.9510 [12.8995] & 0.9516 [14.5608] \\

\end{longtable}
\end{landscape}

\subsubsection{Simultaneous inference}

Tables~\ref{tab:sim-band-performance}
and~\ref{tab:sim-band-performance-derivative} report simultaneous coverage
and mean band width over the fixed grid \(\mathcal G\).

\begin{table}[!p]
\centering
\caption{Empirical simultaneous coverage and mean width of nominal 95\%
confidence bands for the dose-response function over the fixed 61-point
grid.  Each entry gives coverage, followed by mean width in brackets, over
\(5{,}000\) replications.}
\label{tab:sim-band-performance}
\small
\setlength{\tabcolsep}{3.4pt}
\begin{tabular}{@{}rrcccc@{}}
\toprule
\(n\) & Scenario & PDR--PI1 & PDR--PI2 & PDR--CV1 & PDR--CV2\\
\midrule
% Source: Derivative R study/joint_simulation/results/config_4e7d28a1/summary/band_performance.csv (configuration 4e7d28a1).
% Schema: n & scenario & PDR--PI1 coverage [mean band width] & PDR--PI2 coverage [mean band width] & PDR--CV1 coverage [mean band width] & PDR--CV2 coverage [mean band width] \\
% Scope: simultaneous coverage over the fixed 61-point grid; 5000 replications; scenarios CC, CM, and MC.
% Rounding: coverage and mean band width are rounded to four decimal places.
2500 & CC & 0.9372 [0.1392] & 0.9508 [0.1448] & 0.9440 [0.1138] & 0.9406 [0.0999] \\
2500 & CM & 0.9400 [0.1376] & 0.9498 [0.1421] & 0.9398 [0.1119] & 0.9372 [0.0983] \\
2500 & MC & 0.9440 [0.2564] & 0.9482 [0.3117] & 0.9408 [0.2257] & 0.9366 [0.1861] \\
5000 & CC & 0.9368 [0.1047] & 0.9480 [0.1089] & 0.9436 [0.0844] & 0.9410 [0.0737] \\
5000 & CM & 0.9372 [0.1032] & 0.9490 [0.1069] & 0.9432 [0.0830] & 0.9406 [0.0727] \\
5000 & MC & 0.9434 [0.1947] & 0.9486 [0.2386] & 0.9476 [0.1711] & 0.9470 [0.1405] \\
7500 & CC & 0.9440 [0.0881] & 0.9514 [0.0923] & 0.9476 [0.0709] & 0.9434 [0.0615] \\
7500 & CM & 0.9420 [0.0870] & 0.9518 [0.0906] & 0.9452 [0.0697] & 0.9428 [0.0606] \\
7500 & MC & 0.9472 [0.1670] & 0.9480 [0.2041] & 0.9470 [0.1461] & 0.9430 [0.1198] \\
10000 & CC & 0.9454 [0.0785] & 0.9508 [0.0821] & 0.9462 [0.0629] & 0.9430 [0.0546] \\
10000 & CM & 0.9438 [0.0775] & 0.9510 [0.0806] & 0.9454 [0.0620] & 0.9426 [0.0539] \\
10000 & MC & 0.9468 [0.1487] & 0.9460 [0.1827] & 0.9484 [0.1304] & 0.9440 [0.1063] \\
\bottomrule

\end{tabular}
\end{table}

\begin{table}[!p]
\centering
\caption{Empirical simultaneous coverage and mean width of nominal 95\%
confidence bands for the derivative over the fixed 61-point grid.  Each
entry gives coverage, followed by mean width in brackets, over \(5{,}000\)
replications.}
\label{tab:sim-band-performance-derivative}
\small
\setlength{\tabcolsep}{3.4pt}
\begin{tabular}{@{}rrcccc@{}}
\toprule
\(n\) & Scenario & PDR--PI1 & PDR--PI2 & PDR--CV1 & PDR--CV2\\
\midrule
% Source: Derivative R study/joint_simulation/results/config_4e7d28a1/summary/band_performance.csv (configuration 4e7d28a1).
% Schema: n & scenario & PDR--PI1 coverage [mean band width] & PDR--PI2 coverage [mean band width] & PDR--CV1 coverage [mean band width] & PDR--CV2 coverage [mean band width] \\
% Scope: simultaneous coverage over the fixed 61-point grid; 5000 replications; scenarios CC, CM, and MC.
% Rounding: coverage and mean band width are rounded to four decimal places.
2500 & CC & 0.9590 [0.6425] & 0.9514 [0.7035] & 0.9304 [3.7297] & 0.9376 [2.2851] \\
2500 & CM & 0.9560 [0.6270] & 0.9496 [0.6849] & 0.9280 [3.6382] & 0.9360 [2.2351] \\
2500 & MC & 0.9526 [1.5766] & 0.9474 [1.7648] & 0.9166 [8.4064] & 0.9280 [5.3445] \\
5000 & CC & 0.9496 [0.5143] & 0.9454 [0.5775] & 0.9356 [3.0293] & 0.9388 [1.8747] \\
5000 & CM & 0.9478 [0.5029] & 0.9446 [0.5631] & 0.9324 [2.9586] & 0.9360 [1.8345] \\
5000 & MC & 0.9484 [1.2574] & 0.9472 [1.4524] & 0.9324 [6.6937] & 0.9332 [4.3379] \\
7500 & CC & 0.9520 [0.4508] & 0.9494 [0.5151] & 0.9458 [2.7186] & 0.9466 [1.6788] \\
7500 & CM & 0.9500 [0.4412] & 0.9478 [0.5026] & 0.9482 [2.6671] & 0.9480 [1.6456] \\
7500 & MC & 0.9486 [1.1052] & 0.9430 [1.2970] & 0.9346 [6.0843] & 0.9358 [3.8623] \\
10000 & CC & 0.9532 [0.4115] & 0.9482 [0.4752] & 0.9468 [2.4702] & 0.9430 [1.5378] \\
10000 & CM & 0.9534 [0.4027] & 0.9502 [0.4638] & 0.9450 [2.4169] & 0.9452 [1.4956] \\
10000 & MC & 0.9518 [1.0021] & 0.9482 [1.1969] & 0.9372 [5.4155] & 0.9400 [3.4975] \\
\bottomrule

\end{tabular}
\end{table}

\begin{figure}[!p]
\centering
\includegraphics[width=\textwidth]{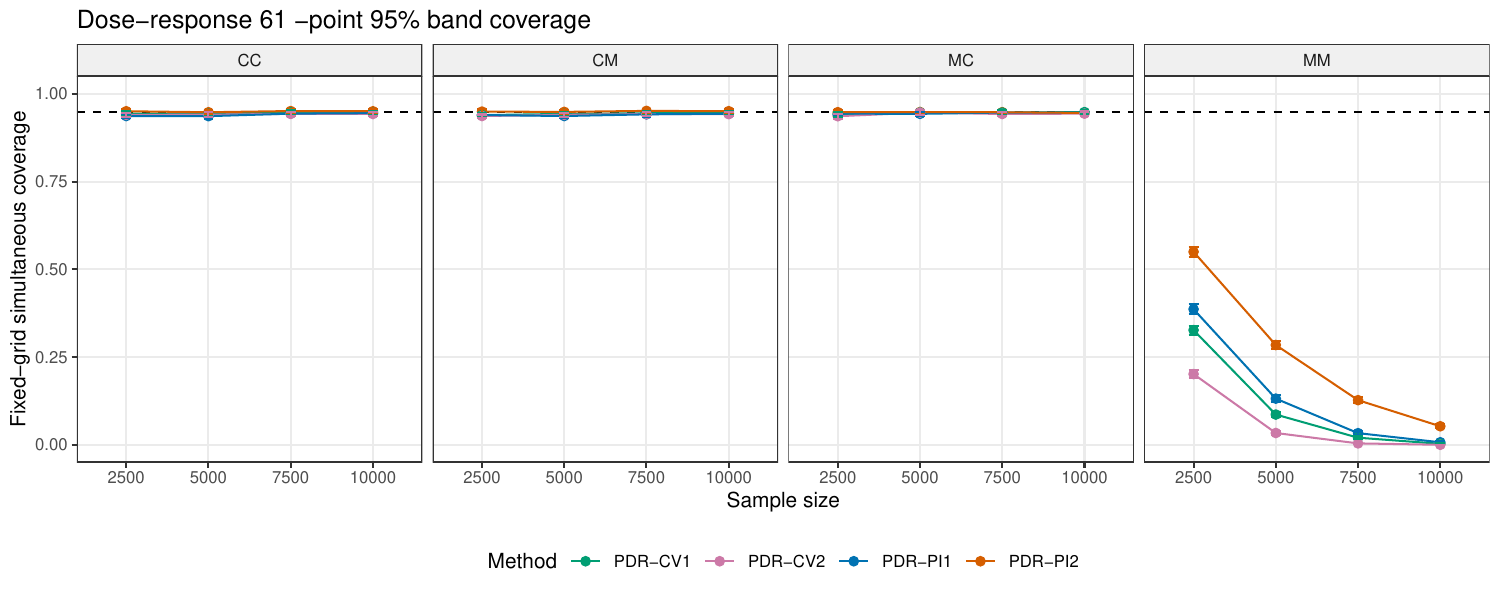}
\caption{Empirical coverage of nominal 95\% fixed-grid simultaneous
confidence bands for the dose-response function in all four nuisance-model
scenarios.  Error bars equal coverage plus or minus \(1.96\) Monte Carlo
standard errors.}
\label{fig:sim-band-coverage-dose-all}
\end{figure}

\begin{figure}[!p]
\centering
\includegraphics[width=\textwidth]{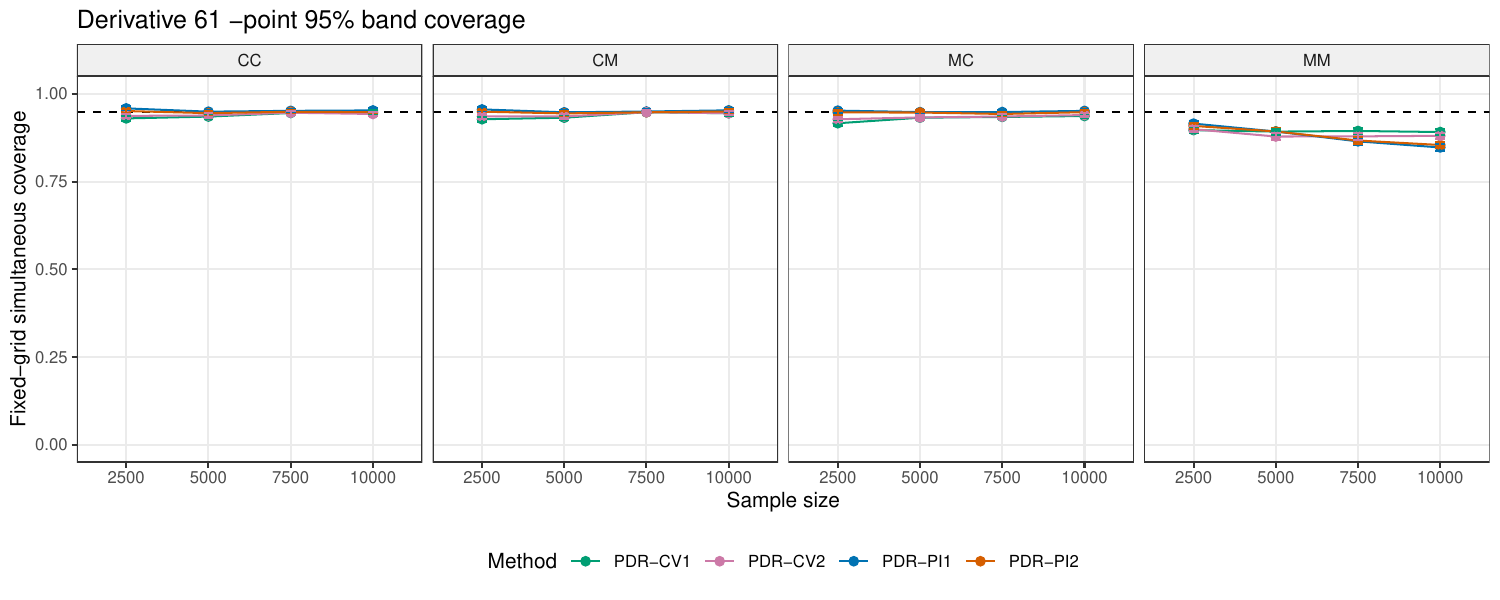}
\caption{Empirical coverage of nominal 95\% fixed-grid simultaneous
confidence bands for the derivative in all four nuisance-model scenarios.
Error bars equal coverage plus or minus \(1.96\) Monte Carlo standard
errors.}
\label{fig:sim-band-coverage-derivative-all}
\end{figure}

\begin{figure}[!p]
\centering
\includegraphics[width=\textwidth]{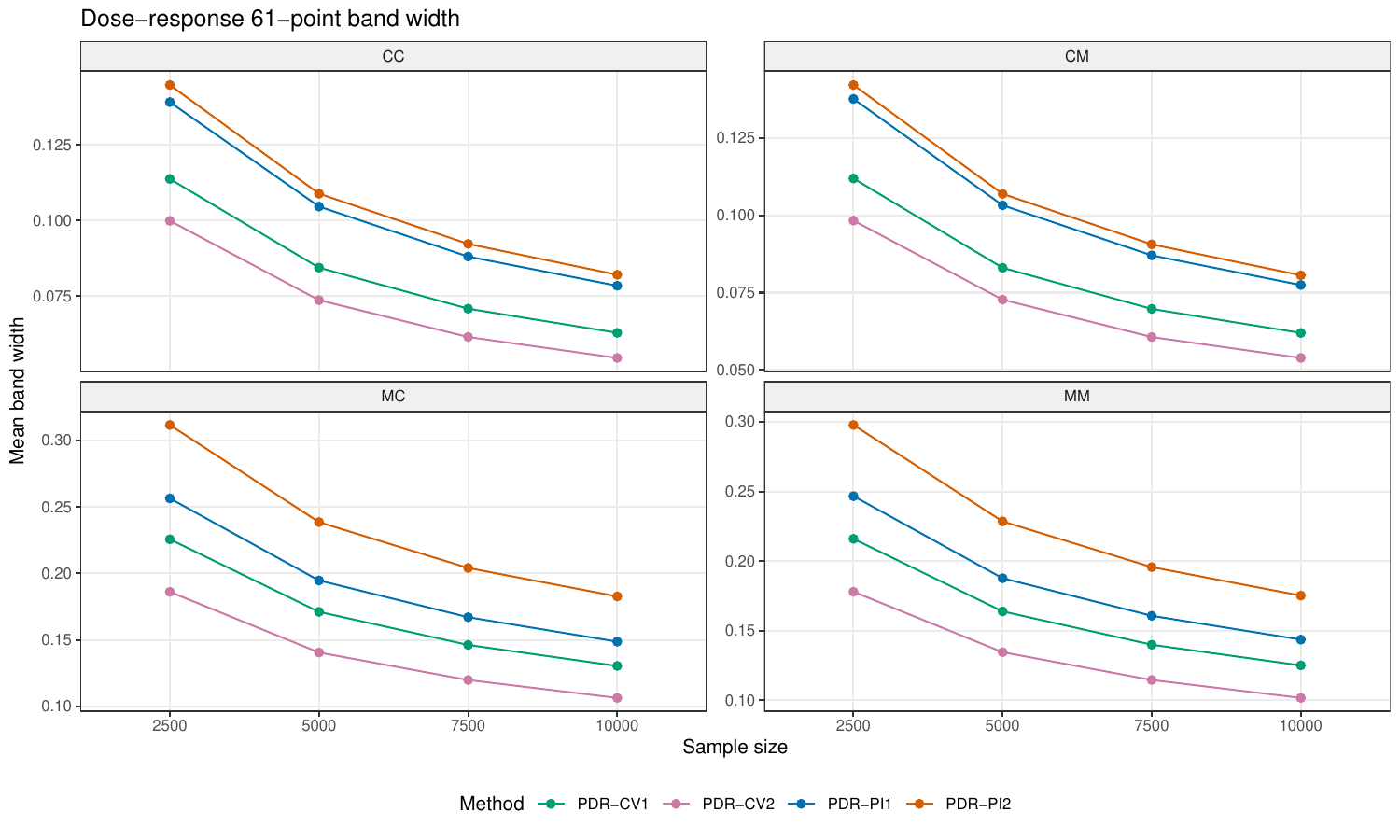}
\caption{Mean width of the dose-response simultaneous confidence bands over
the fixed 61-point grid.}
\label{fig:sim-band-width-dose}
\end{figure}

\begin{figure}[!p]
\centering
\includegraphics[width=\textwidth]{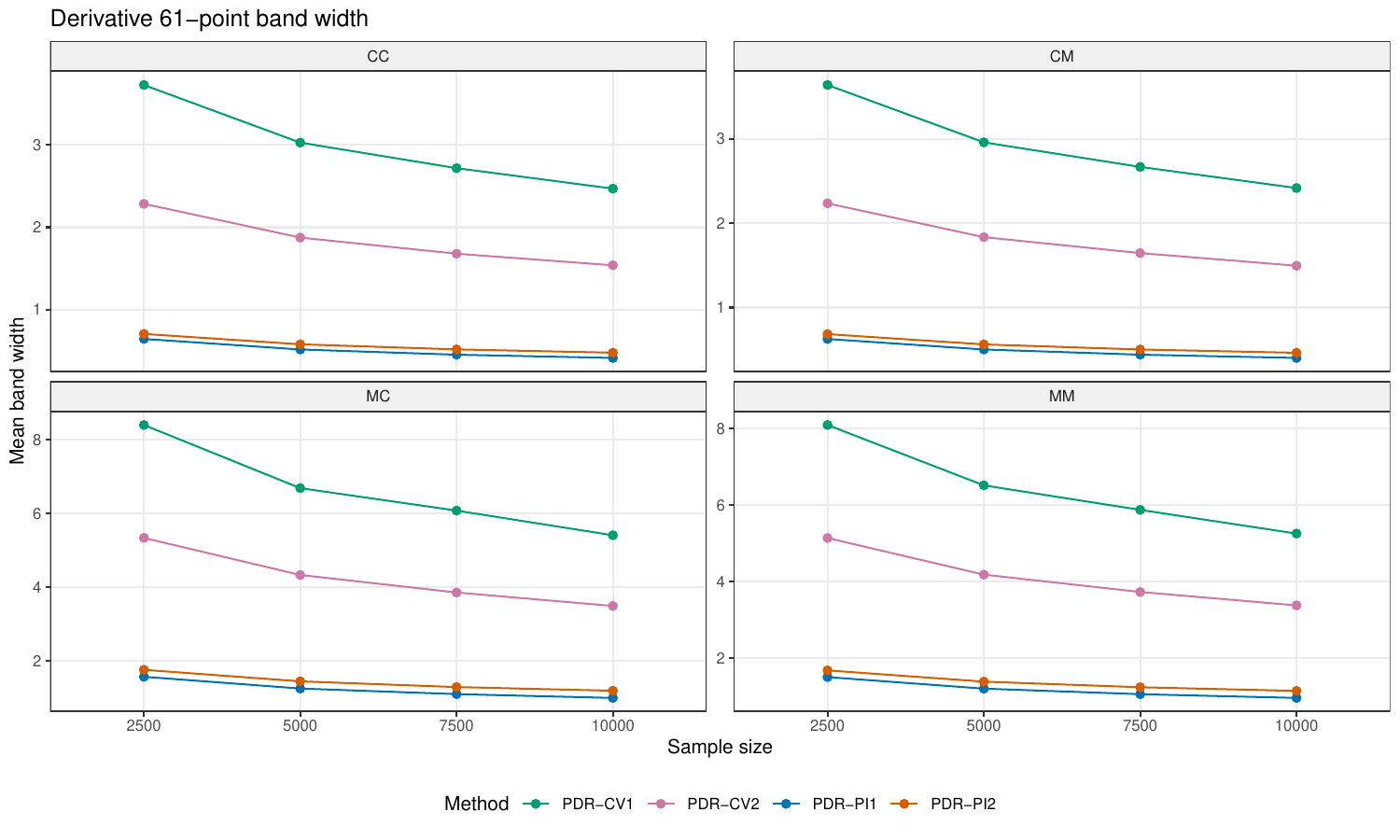}
\caption{Mean width of the derivative simultaneous confidence bands over the
fixed 61-point grid.}
\label{fig:sim-band-width-derivative}
\end{figure}

\subsubsection{Under complete misspecification}\label{sec: MM reports}

The scenario MM has neither bridge model correctly specified.
Table
\ref{tab:sim-mm-diagnostic} reports estimation and inference results at
\(n=10{,}000\). 
It shows appreciable mean
absolute biases of similar magnitude across methods within each estimand
at the largest sample size, but substantially different widths and coverage.
The dose-response band coverage falls sharply as $n$ increases.
The derivative band coverage declines less sharply; the higher pointwise
coverage of ProxLQ--US accompanies much wider intervals, while its bias
is similar to that of the other methods.
These results illustrate the sensitivity to joint bridge misspecification.

\begin{table}[!t]
\centering
\caption{Scenario-MM diagnostics at \(n=10{,}000\).  Pointwise quantities
are averages over the 61 coordinates.  Band coverage refers to the fixed
61-point grid.}
\label{tab:sim-mm-diagnostic}
\scriptsize
\setlength{\tabcolsep}{3.0pt}
\begin{tabular}{@{}llcccc@{}}
\toprule
Estimand & Method & Mean absolute bias & Coverage & Mean CI width & Band coverage\\
\midrule
% Sources: frozen results in configuration 4e7d28a1 and ProxLQ-US results in configuration 0bebf0fb.
% Schema: estimand & method & mean absolute pointwise bias & mean pointwise coverage & mean pointwise CI width & fixed-grid band coverage \\
% Scope: n=10000, scenario MM, 61-point grid, 5000 replications.
% The first three statistics are arithmetic means over the 61 coordinates; -- denotes a band that was not computed.
% Rounding: all reported statistics are rounded to four decimal places.
Dose response & PDR--PI1 & 0.0477 & 0.4846 & 0.0926 & 0.0074 \\
Dose response & PDR--PI2 & 0.0477 & 0.5813 & 0.1097 & 0.0530 \\
Dose response & PDR--CV1 & 0.0477 & 0.4315 & 0.0825 & 0.0030 \\
Dose response & PDR--CV2 & 0.0473 & 0.3607 & 0.0696 & 0.0002 \\
Dose response & ProxLL--US & 0.0474 & 0.6338 & 0.1243 & -- \\
Derivative & PDR--PI1 & 0.1048 & 0.8990 & 0.6118 & 0.8476 \\
Derivative & PDR--PI2 & 0.1048 & 0.9123 & 0.7217 & 0.8546 \\
Derivative & PDR--CV1 & 0.1049 & 0.9243 & 3.1218 & 0.8916 \\
Derivative & PDR--CV2 & 0.1049 & 0.9167 & 2.0254 & 0.8808 \\
Derivative & ProxLQ--US & 0.1037 & 0.9470 & 11.6024 & -- \\
\bottomrule

\end{tabular}
\end{table}

\subsubsection{Bandwidth-selector diagnostics}
\label{sec:simulation-bandwidth-diagnostics}

Figures~\ref{fig:sim-bandwidths-dose}
and~\ref{fig:sim-bandwidths-derivative} display the selected bandwidths.

\begin{figure}[!p]
\centering
\includegraphics[width=\textwidth]{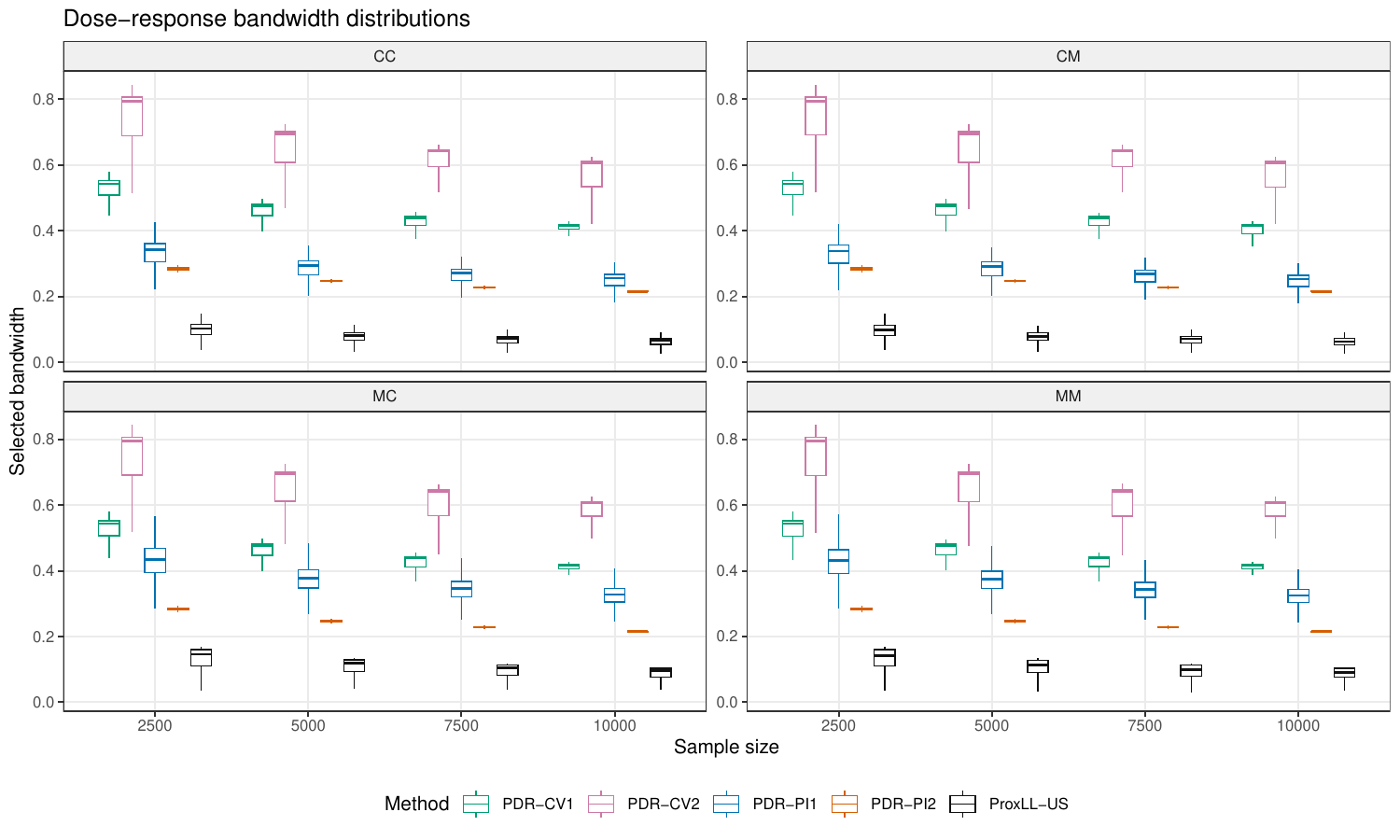}
\caption{Selected dose-response bandwidth distributions by sample size,
nuisance-model scenario, and estimator.}
\label{fig:sim-bandwidths-dose}
\end{figure}

\begin{figure}[!p]
\centering
\includegraphics[width=\textwidth]{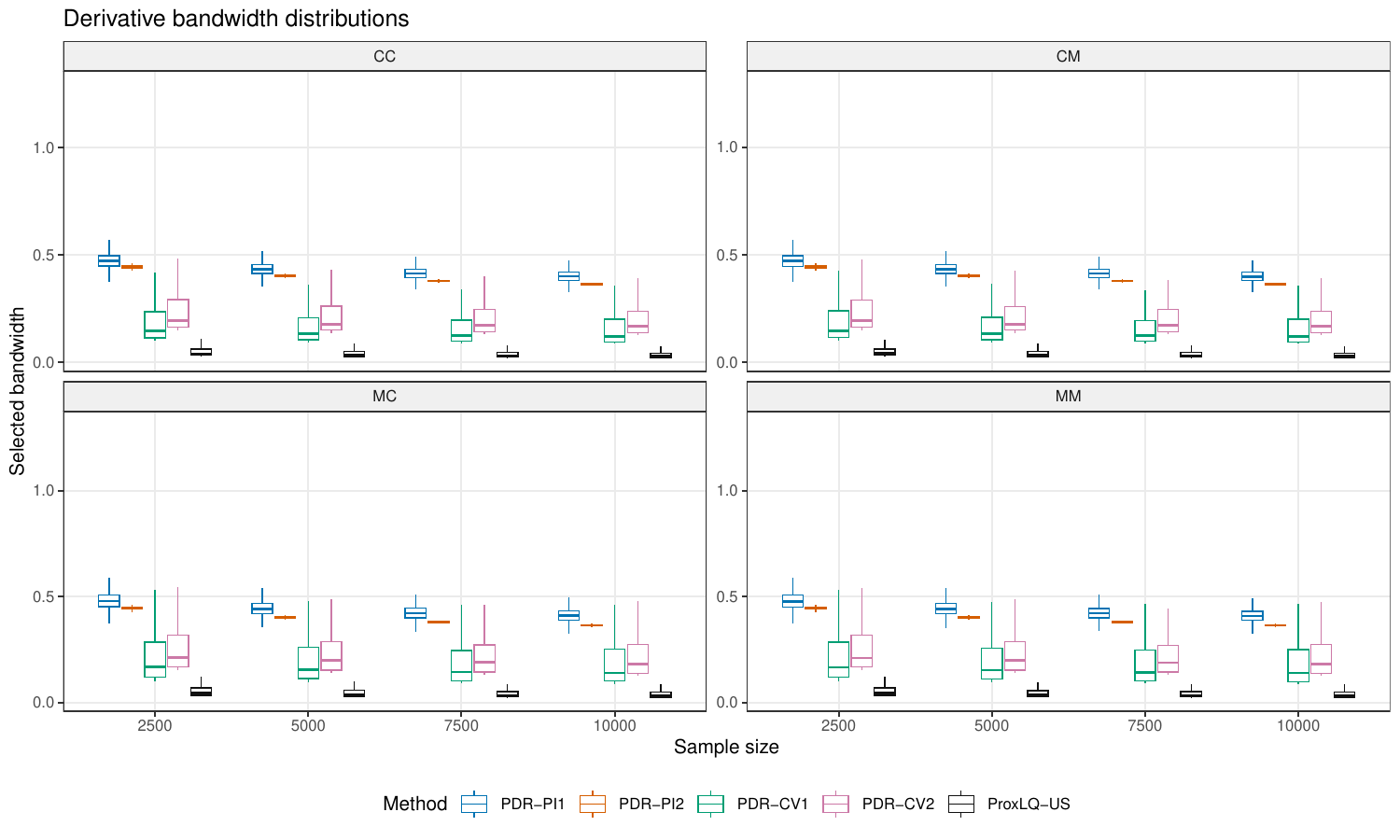}
\caption{Selected derivative bandwidth distributions by sample size,
nuisance-model scenario, and estimator.}
\label{fig:sim-bandwidths-derivative}
\end{figure}

\subsection{Finite-sample performance of the parametric-null tests}
\label{sec:simulation-parametric-null-main}

\subsubsection{Design}
The primitive variables and the first three equations in
\eqref{eqs:simulation design} are preserved.  
In particular,
\(A\sim N(0,V_A)\), where \(V_A=105/128\).  For
\(d\in\{0,1,2\}\), define
\[
\begin{split}
m_0(a)&:=2.5,\\
m_1(a)&:=2.5+0.4a,\\
m_2(a)&:=2.5+0.4a+0.25(a^2-V_A),
\end{split}
\qquad
v_d(a):=(1,a,\ldots,a^d)^\top.
\]
The same departure direction is used for all three null families.  Let
\[
\widetilde r(a):=a^3-3V_Aa,
\qquad
w_T(a):=\{1-(a/0.9)^2\}^2 1(|a|\leq0.9),
\]
and define
\[
\kappa^2:=\int_{-0.9}^{0.9}\widetilde r(a)^2w_T(a)\,da,
\qquad
r(a):=\frac{\widetilde r(a)}{\kappa}.
\]
The Gaussian moment identities give
\[
E\{v_d(A)\widetilde r(A)\}=0,
\qquad d=0,1,2,
\]
and the normalization yields
\(\int_{-0.9}^{0.9}r(a)^2w_T(a)\,da=1\).  The dose-response
function and outcome are
\begin{equation}
\theta_{d,\delta}(a):=m_d(a)+\delta r(a),
\qquad
Y:=\theta_{d,\delta}(A)+0.4\tanh(W)+\varepsilon_Y.
\label{eq:simulation-parametric-null-dgp}
\end{equation}
Thus the \(F_A\)-projection of \(\theta_{d,\delta}\) onto the
degree-\(d\) null space is \(m_d\), and the weighted integrated squared
distance from that projection is \(\delta^2\).

Relative to Section~\ref{sec:simulation}, only the outcome curve and the coefficient of \(\tanh(W)\) change.  For every \(d\) and \(\delta\), the
outcome-bridge GMM fit uses
\[
\{1,A,A^2-V_A,A^3-3V_AA,X_1,X_2,\tanh(W)\},
\]
and the misspecified version omits \(\tanh(W)\).  The five-fold
cross-fitting and treatment-bridge fitting algorithms are unchanged.
The outcome bridge is
\[
h_{0,d,\delta}(w,a,x)
=\theta_{d,\delta}(a)+0.4\tanh(w).
\]
Consequently,
\[
Y-h_{0,d,\delta}(W,A,X)=\varepsilon_Y,
\qquad
E\{h_{0,d,\delta}(W,a,X)\}=\theta_{d,\delta}(a),
\]
where the second equality follows from symmetry of \(W\).  The outcome-bridge
calculation in Section~\ref{sec:simulation-bridge-compatibility} follows
after this substitution.  The joint law determining \(q_0\) is unchanged,
so the treatment-bridge calculation is unchanged.
So we do not repeat the details here.

We use \(n\in\{2{,}500,5{,}000,7{,}500,10{,}000\}\), \(1{,}000\)
replications at each sample size, and
\[
\delta\in\{0,.004,.008,.012,.016,.024,.032,.040,.052,.064,.080\}.
\]
The test integral is evaluated on \([-0.9,0.9]\).  
The 61-point grid is still used for the PI1 pilot integral.
All smoothers use the Epanechnikov kernel.  The four PDR bandwidth rules are
PDR--PI1, PDR--PI2, PDR--CV1, and PDR--CV2, as in the dose-response
experiment, with \(b_n=h_n\).  ProxLL--raw uses the refined degree-one
CV1 bandwidth, and ProxLL--US again divides that bandwidth by \(\log_{10}(n)\).
The bandwidth selection uses the pooled cross-fitted pseudo-outcomes.
The selected bandwidth for each method is shared by all
folds and is held fixed in every bootstrap draw.

The CV calculations follow the same rule as in the previous simulations. 
Among the finite-risk candidates, the smallest CV risk is selected, with the smaller bandwidth used in a tie.

\subsubsection{Bootstrap procedure}
We use the proposed wild response bootstrap in
Section~\ref{sec:parametric-null-main}, motivated by
\citet{Doss2024nonparametric}.  For each evaluation fold \(I_k\),
\(\hat\xi_n^{(-k)}\) is constructed using bridges fitted outside \(I_k\)
and the same-fold empirical marginal in
Section~\ref{sec:cross-fitted-estimator}.  For the degree-\(d\) family,
\(m(a;\beta)=v_d(a)^\top\beta\), the null fit uses all observations in
the evaluation fold:
\[
\hat\beta_{n,k}
=\left[P_{n,k}\{v_d(A)v_d(A)^\top\}\right]^{-1}
P_{n,k}\{v_d(A)\hat\xi_n^{(-k)}(O)\}.
\]
The smoothing moments and null fit are computed separately within each
fold. Recall that the PDR statistic is
\[
\begin{aligned}
\widehat T_{n,\mathcal M}^{DB}
=\sum_{k=1}^K\frac{n_k}{n}\,n_k\sqrt{h_n}\int_{-0.9}^{0.9}\bigl\{\hat\theta_{n,k}^{DB}(a)
-\hat m_{n,k,h_n,b_n}(a;\hat\beta_{n,k})\bigr\}^2w_T(a)\,da.
\end{aligned}
\]
Thus the fold statistics are averaged after squaring and integration.
For \(b_n=h_n\), the debiased local-linear estimator equals the
local-quadratic intercept on nonsingular local windows.  Its fold operator
reproduces constant, linear, and quadratic functions, so the smoothed null
equals \(m(a;\hat\beta_{n,k})\) at the test coordinates.

The ProxLL procedures use the corresponding local-linear fold smoother.
The canonical ProxLL procedure is \(m(a;\hat\beta_{n,k})\), whereas the
matched ProxLL is that fitted polynomial passed through the same
fold smoother.  
These competitors agree for the constant and linear
nulls.  
However, they do not agree for the quadratic nulls, since the local linear mechanism does not reproduce polynomials beyond the first order. 
Under the quadratic null, both comparators are reported for
ProxLL--raw and ProxLL--US.

Conditional on the data, we generate independent two-point multipliers with
\[
\omega_i=
\begin{cases}
(1-\sqrt5)/2,&\text{with probability }(\sqrt5+1)/(2\sqrt5),\\
(1+\sqrt5)/2,&\text{with probability }(\sqrt5-1)/(2\sqrt5).
\end{cases}
\]
Recall that they satisfy \(E(\omega_i)=0\) and
\(E(\omega_i^2)=E(\omega_i^3)=1\).  For \(i\in I_k\), the PDR bootstrap
response is
\[
\hat\xi_{n,i}^{*}
=m(A_i;\hat\beta_{n,k})+\omega_i
\{\hat\xi_n^{(-k)}(O_i)-m(A_i;\hat\beta_{n,k})\}.
\]
The same response construction is used for ProxLL.  The responses are thus
generated around the fitted null for every method.

In every bootstrap draw, the null coefficient is refit separately on each
fold and the bootstrap response is passed through the same selected fold
smoother.  The resulting PDR bootstrap statistic is
\[
\begin{aligned}
\widehat T_{n,\mathcal M}^{DB,*}
=\sum_{k=1}^K\frac{n_k}{n}\,n_k\sqrt{h_n}\int_{-0.9}^{0.9}\bigl\{\hat\theta_{n,k}^{DB,*}(a)
-\hat m_{n,k,h_n,b_n}(a;\hat\beta_{n,k}^{*})\bigr\}^2w_T(a)\,da.
\end{aligned}
\]
For ProxLL, the bootstrap estimator is the refitted polynomial in a canonical version and its smoothed version in a matched competitor.
In particular, the canonical quadratic comparison has the smoothing bias of the fitted null, since it is not replaced by the matched comparison.
The canonical and matched versions share the same response draws.

The folds, fitted bridges, selected bandwidths, and fold smoothing operators
are fixed across bootstrap draws.  The same multiplier draws are used for
all methods within each replication.  We use \(B=4{,}999\) draws and reject when the
corrected bootstrap \(p\)-value
\[
\frac{1+\sum_{\ell=1}^{B}
1\{\widehat T_{n,\mathcal M}^{DB,*(\ell)}
\geq\widehat T_{n,\mathcal M}^{DB}\}}{B+1}
\]
is no greater than \(\alpha_0=.05\), with the corresponding statistics
used for ProxLL.  At \(\delta=0\), rejection probability estimates size.
At positive \(\delta\), it estimates power under the stated fixed
departure.

\begin{remark}\label{rem:debiased-local-linear-variance}
Consider interior-point estimation and the leading stochastic terms of
the integrated-square tests, using the same pseudo-outcome,
kernel, fold weights, and curve bandwidth.  With the debiased procedure
in the numerator, the pointwise variance ratio is
$\Lambda_\tau^{(2)}(0)/\int K(u)^2\,du$, and the ratio for the normalized
global statistics is
$\Lambda_\tau^{(4)}(0)/\int\{\int K(v)K(v+r)\,dv\}^2\,dr$.
The latter compares the stochastic null limits when the local-linear
null-smoothing discrepancy is absent or negligible.  The common factors
involving $f_A$, $\sigma_0^2$, and the fold weights cancel.
Here $\tau$ is the limiting bandwidth ratio within PDR, $h_n/b_n\to\tau$.
For the Epanechnikov kernel and $h_n=b_n$, these variance ratios are
$25/12$ and approximately $2.38$, respectively.
Thus debiasing increases the leading
stochastic variance in this same-bandwidth comparison while removing the leading curve-smoothing bias. 
These constants are specific to the kernel and $\tau$ and do not give a universal ordering across those choices.

However, the comparison changes when the procedures use different bandwidths.
For the dose-response curve estimation, the variance ratio additionally contains
$h_{\mathrm{LL}}/h_{\mathrm{PDR}}$.  For derivative estimation, the
leading variance is of order $(nh_n^3)^{-1}$, so the corresponding canonical-to-debiased bandwidth ratio enters cubically.  
Thus, undersmoothing changes variance as well as bias.
The dose-response and derivative comparisons in
Sections~\ref{sec:simulation-dose-response} and
\ref{sec:simulation-derivative} are consistent with this joint bias--variance
trade-off, rather than the kernel constants alone.

For the global statistics, each procedure is normalized by its own
$n\sqrt{h_n}$.  Consequently, the leading null variance comparison has does not have additional factor $h_{\mathrm{LL}}/h_{\mathrm{PDR}}$; the bandwidth
dependence of the unnormalized integrated-square variance is canceled by
this normalization.  
The bandwidth still affects the centering and the alternative mean shift.  The variance constants therefore do not order
testing power.
\end{remark}

\subsubsection{Empirical rejection probabilities at the null}
\label{sec:simulation-parametric-null-size}

Figure~\ref{fig:sim-parametric-null-size} shows the null rejection probabilities for each setting.
Let \(\tilde X=\sum_{r=1}^{1000}1\{\text{replication }r\text{ rejects}\}\).
If the true size is exactly \(.05\), then
\[
\tilde X\sim\operatorname{Binomial}(1000,.05),
\qquad \tilde p=\tilde X/1000.
\]
Writing \(\tilde  q_\gamma\) for the \(\gamma\)-quantile of this distribution,
\[
\left[\tilde q_{.025}/1000,\tilde  q_{.975}/1000\right]=[.037,.064].
\]
The gray region in Figure~\ref{fig:sim-parametric-null-size} is this pointwise Monte Carlo sampling
range when the true rejection probability is exactly \(.05\).

The PDR and ProxLL procedures have broadly
near-nominal rejection probabilities in CC and CM under all three nulls.
PDR--PI1 and PDR--CV2 tend to reject more often under
the quadratic null in CC, while ProxLL--US is generally less liberal than
ProxLL--raw.  Between two ProxLLs, matching the quadratic null has marginal effect for the
undersmoothed bandwidth and does not uniformly improve size for the raw
bandwidth.
The scenario MC has larger variation across sample sizes.  Under the constant and linear
nulls, PDR--CV2 and ProxLL--raw tend to reject more often than ProxLL--US.
The corresponding values are comprehensively reported in
Tables~\ref{tab:sim-parametric-null-size}--\ref{tab:sim-parametric-null-size-proxll}.

\begin{landscape}
\begin{figure}[!p]
\centering
\includegraphics[width=.78\linewidth]{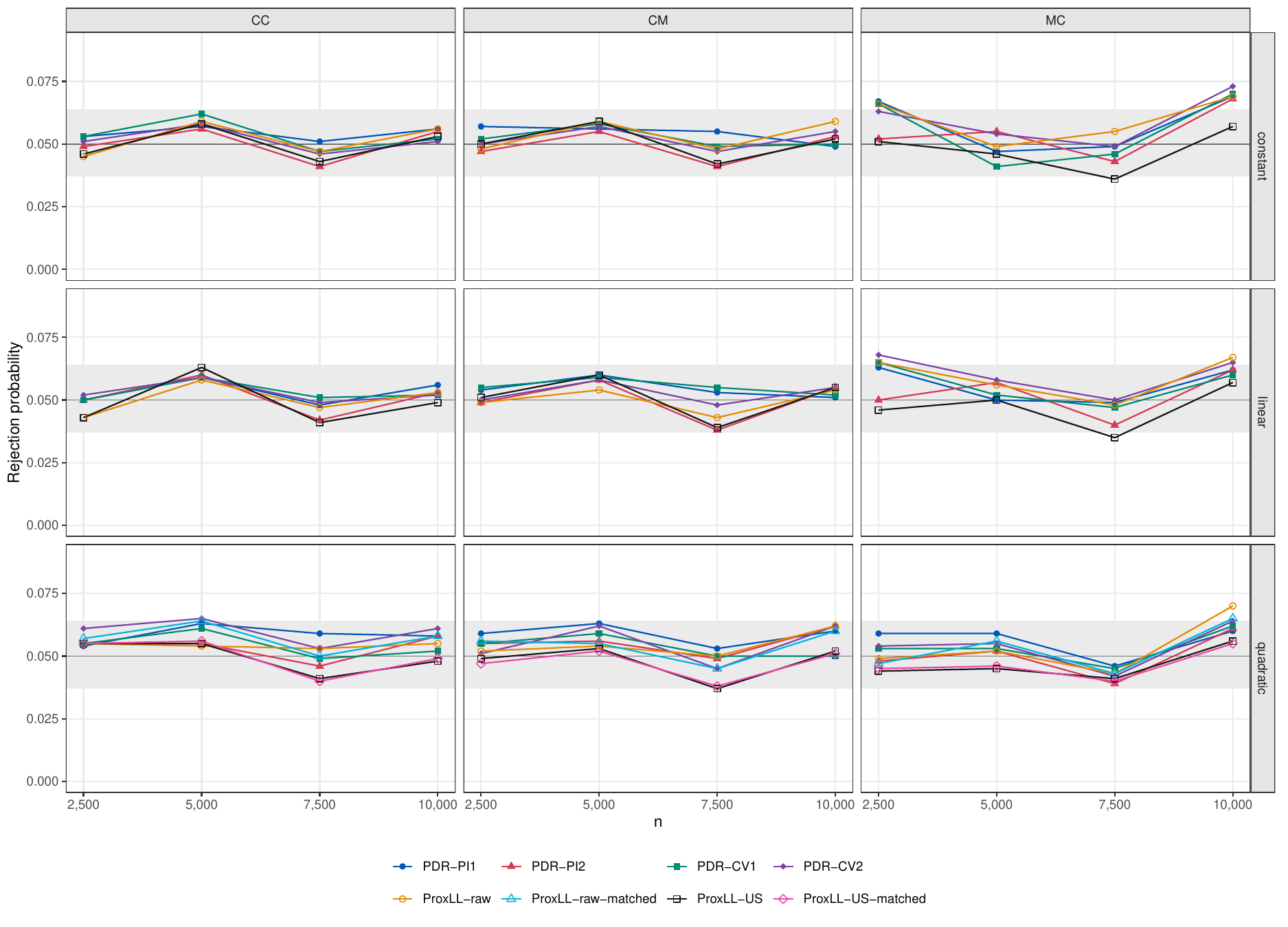}
\caption{Empirical rejection probabilities at the null.  Rows
index the null family and columns index the nuisance-model scenario.  The
gray region is the pointwise Monte Carlo range defined in
Section~\ref{sec:simulation-parametric-null-size}.  Matched ProxLL comparisons are reported only under the quadratic null.}
\label{fig:sim-parametric-null-size}
\end{figure}
\end{landscape}

\begin{landscape}
\footnotesize
\setlength{\tabcolsep}{5pt}
\renewcommand{\arraystretch}{.82}
\begin{longtable}{lllrrrr}
\caption{Empirical size of the PDR procedures. Each entry is based on $1{,}000$ replications.}
\label{tab:sim-parametric-null-size} \\
\toprule
Scenario & Null & $n$ & PI1 & PI2 & CV1 & CV2 \\
\midrule
\endfirsthead
\caption[]{Empirical size of the PDR procedures. (continued).} \\
\toprule
Scenario & Null & $n$ & PI1 & PI2 & CV1 & CV2 \\
\midrule
\endhead
\midrule \multicolumn{7}{r}{Continued on next page} \\
\endfoot
\bottomrule
\endlastfoot
CC & constant & $2{,}500$ & .053 & .049 & .053 & .051 \\
CC & constant & $5{,}000$ & .057 & .056 & .062 & .058 \\
CC & constant & $7{,}500$ & .051 & .041 & .047 & .046 \\
CC & constant & $10{,}000$ & .056 & .055 & .052 & .051 \\
CC & linear & $2{,}500$ & .050 & .050 & .050 & .052 \\
CC & linear & $5{,}000$ & .059 & .060 & .059 & .059 \\
CC & linear & $7{,}500$ & .048 & .042 & .051 & .049 \\
CC & linear & $10{,}000$ & .056 & .053 & .052 & .052 \\
CC & quadratic & $2{,}500$ & .054 & .055 & .055 & .061 \\
CC & quadratic & $5{,}000$ & .063 & .055 & .061 & .065 \\
CC & quadratic & $7{,}500$ & .059 & .046 & .049 & .053 \\
CC & quadratic & $10{,}000$ & .058 & .058 & .052 & .061 \\
CM & constant & $2{,}500$ & .057 & .047 & .052 & .050 \\
CM & constant & $5{,}000$ & .056 & .055 & .058 & .057 \\
CM & constant & $7{,}500$ & .055 & .041 & .049 & .047 \\
CM & constant & $10{,}000$ & .049 & .053 & .050 & .055 \\
CM & linear & $2{,}500$ & .054 & .049 & .055 & .050 \\
CM & linear & $5{,}000$ & .060 & .058 & .059 & .058 \\
CM & linear & $7{,}500$ & .053 & .038 & .055 & .048 \\
CM & linear & $10{,}000$ & .051 & .055 & .052 & .055 \\
CM & quadratic & $2{,}500$ & .059 & .055 & .055 & .051 \\
CM & quadratic & $5{,}000$ & .063 & .056 & .059 & .062 \\
CM & quadratic & $7{,}500$ & .053 & .049 & .050 & .045 \\
CM & quadratic & $10{,}000$ & .060 & .062 & .050 & .062 \\
MC & constant & $2{,}500$ & .067 & .052 & .066 & .063 \\
MC & constant & $5{,}000$ & .047 & .055 & .041 & .054 \\
MC & constant & $7{,}500$ & .049 & .043 & .046 & .049 \\
MC & constant & $10{,}000$ & .069 & .068 & .070 & .073 \\
MC & linear & $2{,}500$ & .063 & .050 & .065 & .068 \\
MC & linear & $5{,}000$ & .050 & .057 & .052 & .058 \\
MC & linear & $7{,}500$ & .049 & .040 & .047 & .050 \\
MC & linear & $10{,}000$ & .062 & .062 & .060 & .065 \\
MC & quadratic & $2{,}500$ & .059 & .048 & .053 & .054 \\
MC & quadratic & $5{,}000$ & .059 & .052 & .053 & .055 \\
MC & quadratic & $7{,}500$ & .046 & .039 & .045 & .042 \\
MC & quadratic & $10{,}000$ & .060 & .061 & .062 & .064 \\
\end{longtable}
\end{landscape}

\begin{landscape}
\footnotesize
\setlength{\tabcolsep}{5pt}
\renewcommand{\arraystretch}{.82}
\begin{longtable}{lllrrrr}
\caption{Empirical size of the ProxLL procedures. Each entry is based on $1{,}000$ replications. Matched comparisons are reported only for the quadratic null.}
\label{tab:sim-parametric-null-size-proxll} \\
\toprule
Scenario & Null & $n$ & Raw & Raw--matched & US & US--matched \\
\midrule
\endfirsthead
\caption[]{Empirical size of the ProxLL procedures. (continued).} \\
\toprule
Scenario & Null & $n$ & Raw & Raw--matched & US & US--matched \\
\midrule
\endhead
\midrule \multicolumn{7}{r}{Continued on next page} \\
\endfoot
\bottomrule
\endlastfoot
CC & constant & $2{,}500$ & .045 & -- & .046 & -- \\
CC & constant & $5{,}000$ & .059 & -- & .058 & -- \\
CC & constant & $7{,}500$ & .047 & -- & .043 & -- \\
CC & constant & $10{,}000$ & .056 & -- & .053 & -- \\
CC & linear & $2{,}500$ & .043 & -- & .043 & -- \\
CC & linear & $5{,}000$ & .058 & -- & .063 & -- \\
CC & linear & $7{,}500$ & .047 & -- & .041 & -- \\
CC & linear & $10{,}000$ & .053 & -- & .049 & -- \\
CC & quadratic & $2{,}500$ & .055 & .057 & .055 & .055 \\
CC & quadratic & $5{,}000$ & .054 & .064 & .055 & .056 \\
CC & quadratic & $7{,}500$ & .053 & .050 & .041 & .040 \\
CC & quadratic & $10{,}000$ & .055 & .058 & .048 & .049 \\
CM & constant & $2{,}500$ & .048 & -- & .050 & -- \\
CM & constant & $5{,}000$ & .059 & -- & .059 & -- \\
CM & constant & $7{,}500$ & .048 & -- & .042 & -- \\
CM & constant & $10{,}000$ & .059 & -- & .052 & -- \\
CM & linear & $2{,}500$ & .049 & -- & .051 & -- \\
CM & linear & $5{,}000$ & .054 & -- & .060 & -- \\
CM & linear & $7{,}500$ & .043 & -- & .039 & -- \\
CM & linear & $10{,}000$ & .054 & -- & .055 & -- \\
CM & quadratic & $2{,}500$ & .052 & .056 & .049 & .047 \\
CM & quadratic & $5{,}000$ & .054 & .055 & .053 & .052 \\
CM & quadratic & $7{,}500$ & .050 & .045 & .037 & .038 \\
CM & quadratic & $10{,}000$ & .062 & .060 & .052 & .051 \\
MC & constant & $2{,}500$ & .066 & -- & .051 & -- \\
MC & constant & $5{,}000$ & .049 & -- & .046 & -- \\
MC & constant & $7{,}500$ & .055 & -- & .036 & -- \\
MC & constant & $10{,}000$ & .069 & -- & .057 & -- \\
MC & linear & $2{,}500$ & .065 & -- & .046 & -- \\
MC & linear & $5{,}000$ & .056 & -- & .050 & -- \\
MC & linear & $7{,}500$ & .048 & -- & .035 & -- \\
MC & linear & $10{,}000$ & .067 & -- & .057 & -- \\
MC & quadratic & $2{,}500$ & .049 & .047 & .044 & .045 \\
MC & quadratic & $5{,}000$ & .052 & .056 & .045 & .046 \\
MC & quadratic & $7{,}500$ & .043 & .043 & .041 & .040 \\
MC & quadratic & $10{,}000$ & .070 & .065 & .056 & .055 \\
\end{longtable}
\end{landscape}

\subsubsection{Fixed-departure rejection probabilities}

Figures~\ref{fig:sim-parametric-null-power-cc}--\ref{fig:sim-parametric-null-power-mc}
display empirical rejection probabilities across departure values,
null families, scenarios, and sample sizes.
Tables~\ref{tab:power-cc-constant}--\ref{tab:power-mc-quadratic}
report the corresponding values for departures up to $\delta=.040$.
Throughout, power generally increases with sample size and departure, and the
procedures approach unit power at the largest departures.
The powers are similar in CC and CM and generally
lower in MC.  
Across the three scenarios and null families, PDR--CV2
generally has the highest PDR power, followed by PDR--CV1 and PDR--PI1,
with PDR--PI2 lowest.  The difference between the CV procedures and the
plug-in procedures is more pronounced under the quadratic null.
ProxLL--raw is generally more powerful than ProxLL--US under the constant and linear
nulls.  PDR--PI2 and ProxLL--US have similar power in those two cases,
whereas PDR--PI2 is generally more powerful under the quadratic null.

\begin{landscape}
\begin{figure}[!p]
\centering
\includegraphics[width=.82\linewidth]{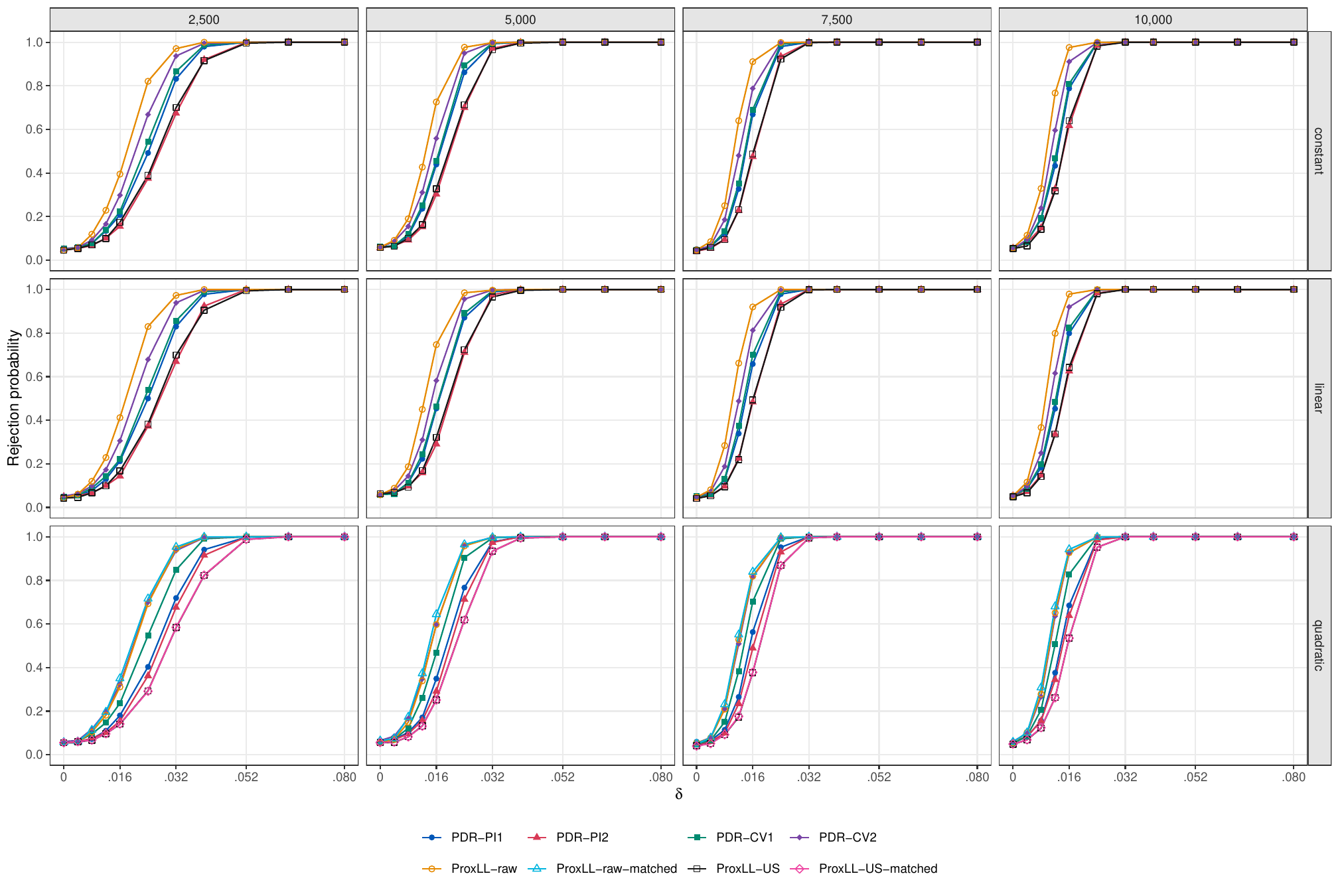}
\caption{Empirical rejection probability in CC.
Rows index the null family and columns index sample
size.}
\label{fig:sim-parametric-null-power-cc}
\end{figure}
\end{landscape}

\begin{landscape}
\begin{figure}[!p]
\centering
\includegraphics[width=.82\linewidth]{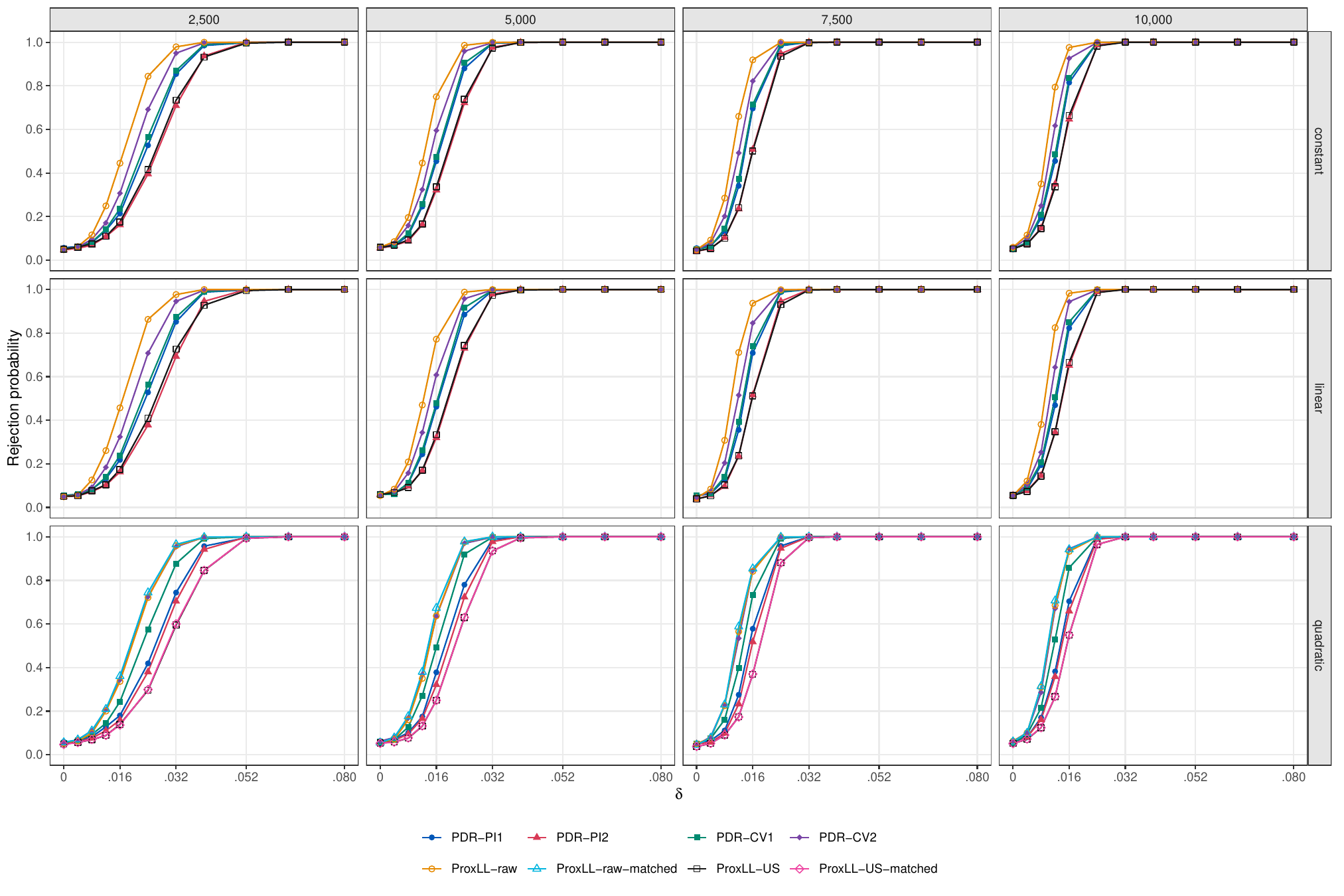}
\caption{Empirical rejection probability in CM.
Rows index the null family and columns index sample
size.}
\label{fig:sim-parametric-null-power-cm}
\end{figure}
\end{landscape}

\begin{landscape}
\begin{figure}[!p]
\centering
\includegraphics[width=.82\linewidth]{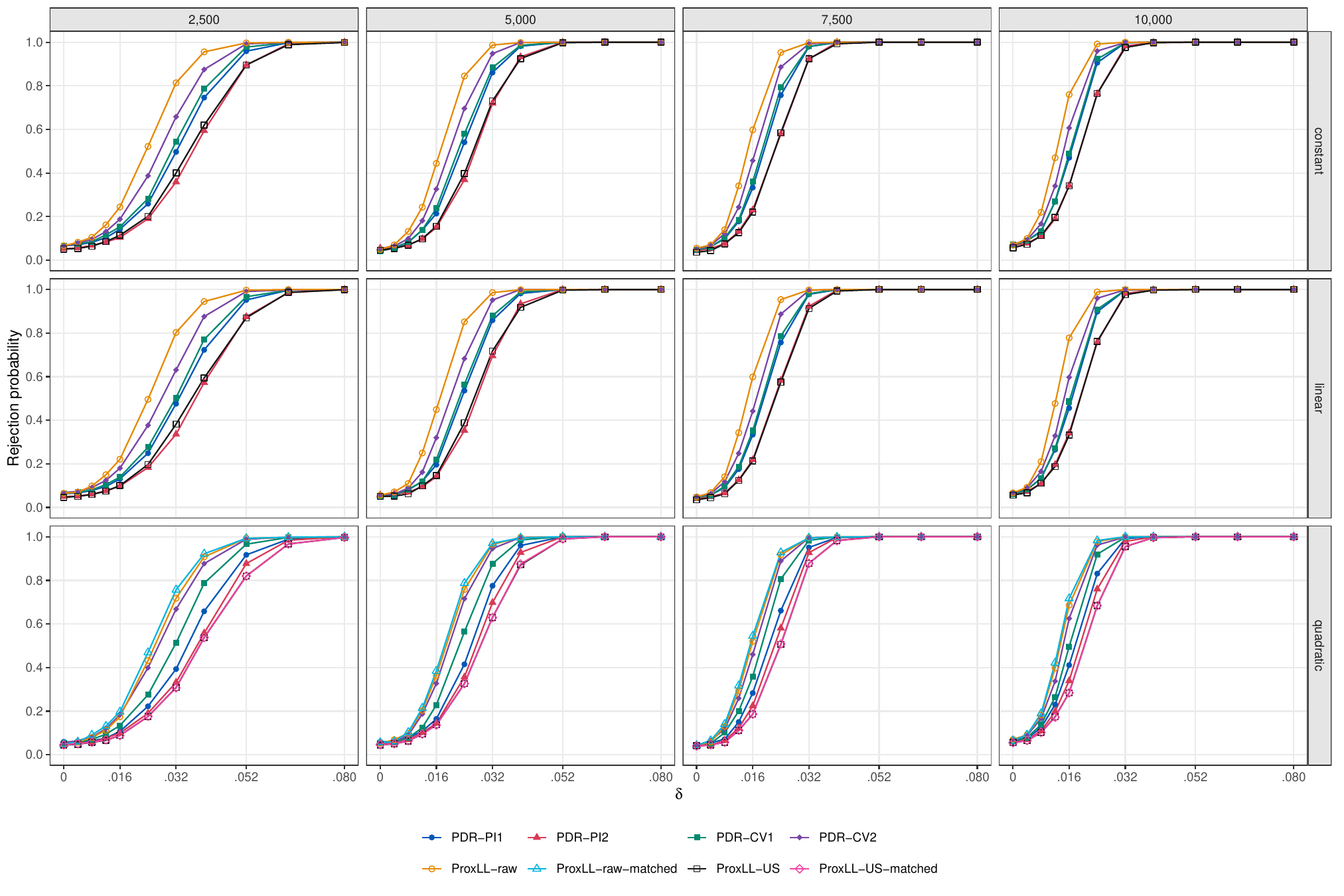}
\caption{Empirical rejection probability in MC.
Rows index the null family and columns index sample
size.}
\label{fig:sim-parametric-null-power-mc}
\end{figure}
\end{landscape}

%\input{\simtabledir/parametric-null-power.tex}

% CC: constant
\begin{table}[p]
\centering
\begingroup
\footnotesize
\setlength{\tabcolsep}{4pt}
\renewcommand{\arraystretch}{1.05}
\caption{Empirical rejection probabilities for the constant-null test in CC.
Each entry uses $1{,}000$ replications.}
\label{tab:power-cc-constant}
\begin{tabular}{l*{8}{r}}
\toprule
Method & $0$ & $.004$ & $.008$ & $.012$ & $.016$ & $.024$ & $.032$ & $.040$ \\
\midrule
\multicolumn{9}{l}{$n=2{,}500$}\\
PDR--PI1    & .053 & .057 & .077 & .134 & .207 & .492 & .832 & .980 \\
PDR--PI2    & .049 & .054 & .067 & .101 & .155 & .375 & .674 & .920 \\
PDR--CV1    & .053 & .058 & .078 & .139 & .223 & .543 & .866 & .987 \\
PDR--CV2    & .051 & .058 & .092 & .166 & .298 & .668 & .937 & .995 \\
ProxLL--raw & .045 & .057 & .119 & .229 & .395 & .821 & .971 & 1.000 \\
ProxLL--US  & .046 & .054 & .070 & .099 & .174 & .389 & .700 & .916 \\
\midrule
\multicolumn{9}{l}{$n=5{,}000$}\\
PDR--PI1    & .057 & .067 & .111 & .236 & .439 & .862 & .993 & 1.000 \\
PDR--PI2    & .056 & .067 & .092 & .153 & .303 & .700 & .972 & .999 \\
PDR--CV1    & .062 & .066 & .119 & .252 & .455 & .894 & .996 & 1.000 \\
PDR--CV2    & .058 & .086 & .155 & .311 & .559 & .950 & .999 & 1.000 \\
ProxLL--raw & .059 & .091 & .190 & .427 & .726 & .977 & .998 & 1.000 \\
ProxLL--US  & .058 & .064 & .100 & .163 & .328 & .713 & .966 & .996 \\
\midrule
\multicolumn{9}{l}{$n=7{,}500$}\\
PDR--PI1    & .051 & .063 & .120 & .327 & .669 & .980 & 1.000 & 1.000 \\
PDR--PI2    & .041 & .055 & .094 & .228 & .475 & .935 & 1.000 & 1.000 \\
PDR--CV1    & .047 & .063 & .132 & .353 & .691 & .991 & 1.000 & 1.000 \\
PDR--CV2    & .046 & .067 & .186 & .480 & .788 & .998 & 1.000 & 1.000 \\
ProxLL--raw & .047 & .085 & .250 & .640 & .911 & .999 & 1.000 & 1.000 \\
ProxLL--US  & .043 & .058 & .095 & .231 & .488 & .923 & .998 & 1.000 \\
\midrule
\multicolumn{9}{l}{$n=10{,}000$}\\
PDR--PI1    & .056 & .085 & .188 & .433 & .788 & .998 & 1.000 & 1.000 \\
PDR--PI2    & .055 & .080 & .148 & .327 & .617 & .986 & 1.000 & 1.000 \\
PDR--CV1    & .052 & .085 & .194 & .467 & .810 & .999 & 1.000 & 1.000 \\
PDR--CV2    & .051 & .095 & .239 & .596 & .911 & .999 & 1.000 & 1.000 \\
ProxLL--raw & .056 & .114 & .329 & .767 & .976 & 1.000 & 1.000 & 1.000 \\
ProxLL--US  & .053 & .065 & .141 & .318 & .639 & .983 & 1.000 & 1.000 \\
\bottomrule
\end{tabular}
\endgroup
\end{table}

% CC: linear
\begin{table}[p]
\centering
\begingroup
\footnotesize
\setlength{\tabcolsep}{4pt}
\renewcommand{\arraystretch}{1.05}
\caption{Empirical rejection probabilities for the linear-null test in CC.
Each entry uses $1{,}000$ replications.}
\label{tab:power-cc-linear}
\begin{tabular}{l*{8}{r}}
\toprule
Method & $0$ & $.004$ & $.008$ & $.012$ & $.016$ & $.024$ & $.032$ & $.040$ \\
\midrule
\multicolumn{9}{l}{$n=2{,}500$}\\
PDR--PI1    & .050 & .055 & .077 & .129 & .212 & .500 & .830 & .978 \\
PDR--PI2    & .050 & .055 & .066 & .100 & .144 & .373 & .669 & .924 \\
PDR--CV1    & .050 & .055 & .089 & .143 & .221 & .539 & .857 & .990 \\
PDR--CV2    & .052 & .063 & .098 & .173 & .306 & .679 & .939 & .996 \\
ProxLL--raw & .043 & .061 & .120 & .229 & .412 & .830 & .973 & 1.000 \\
ProxLL--US  & .043 & .046 & .067 & .099 & .167 & .383 & .699 & .905 \\
\midrule
\multicolumn{9}{l}{$n=5{,}000$}\\
PDR--PI1    & .059 & .065 & .113 & .223 & .454 & .871 & .989 & 1.000 \\
PDR--PI2    & .060 & .071 & .098 & .159 & .291 & .711 & .977 & .999 \\
PDR--CV1    & .059 & .062 & .114 & .245 & .464 & .891 & .994 & 1.000 \\
PDR--CV2    & .059 & .078 & .144 & .310 & .582 & .956 & .998 & 1.000 \\
ProxLL--raw & .058 & .088 & .187 & .450 & .747 & .985 & .997 & 1.000 \\
ProxLL--US  & .063 & .068 & .093 & .169 & .321 & .724 & .966 & .996 \\
\midrule
\multicolumn{9}{l}{$n=7{,}500$}\\
PDR--PI1    & .048 & .062 & .126 & .339 & .658 & .979 & 1.000 & 1.000 \\
PDR--PI2    & .042 & .055 & .097 & .225 & .483 & .933 & 1.000 & 1.000 \\
PDR--CV1    & .051 & .060 & .131 & .376 & .701 & .990 & 1.000 & 1.000 \\
PDR--CV2    & .049 & .073 & .188 & .488 & .813 & .997 & 1.000 & 1.000 \\
ProxLL--raw & .047 & .081 & .284 & .662 & .920 & 1.000 & 1.000 & 1.000 \\
ProxLL--US  & .041 & .054 & .095 & .220 & .493 & .918 & .999 & 1.000 \\
\midrule
\multicolumn{9}{l}{$n=10{,}000$}\\
PDR--PI1    & .056 & .086 & .183 & .453 & .799 & .996 & 1.000 & 1.000 \\
PDR--PI2    & .053 & .072 & .148 & .336 & .626 & .984 & 1.000 & 1.000 \\
PDR--CV1    & .052 & .087 & .199 & .484 & .824 & .998 & 1.000 & 1.000 \\
PDR--CV2    & .052 & .097 & .250 & .616 & .920 & .999 & 1.000 & 1.000 \\
ProxLL--raw & .053 & .116 & .367 & .799 & .979 & 1.000 & 1.000 & 1.000 \\
ProxLL--US  & .049 & .067 & .142 & .337 & .643 & .981 & 1.000 & 1.000 \\
\bottomrule
\end{tabular}
\endgroup
\end{table}

% CC: quadratic
\begin{table}[p]
\centering
\begingroup
\footnotesize
\setlength{\tabcolsep}{4pt}
\renewcommand{\arraystretch}{1.05}
\caption{Empirical rejection probabilities for the quadratic-null test in CC.
Each entry uses $1{,}000$ replications.}
\label{tab:power-cc-quadratic}
\begin{tabular}{l*{8}{r}}
\toprule
Method & $0$ & $.004$ & $.008$ & $.012$ & $.016$ & $.024$ & $.032$ & $.040$ \\
\midrule
\multicolumn{9}{l}{$n=2{,}500$}\\
PDR--PI1             & .054 & .055 & .072 & .110 & .180 & .403 & .719 & .941 \\
PDR--PI2             & .055 & .058 & .076 & .101 & .154 & .362 & .676 & .915 \\
PDR--CV1             & .055 & .061 & .095 & .147 & .236 & .547 & .848 & .991 \\
PDR--CV2             & .061 & .065 & .117 & .199 & .323 & .700 & .937 & .996 \\
ProxLL--raw          & .055 & .059 & .104 & .181 & .311 & .693 & .944 & .998 \\
ProxLL--US           & .055 & .060 & .066 & .096 & .140 & .293 & .583 & .823 \\
ProxLL--raw--matched & .057 & .060 & .113 & .196 & .350 & .716 & .953 & .999 \\
ProxLL--US--matched  & .055 & .060 & .067 & .096 & .141 & .291 & .584 & .824 \\
\midrule
\multicolumn{9}{l}{$n=5{,}000$}\\
PDR--PI1             & .063 & .071 & .102 & .171 & .349 & .767 & .977 & 1.000 \\
PDR--PI2             & .055 & .068 & .096 & .157 & .291 & .713 & .974 & 1.000 \\
PDR--CV1             & .061 & .071 & .122 & .261 & .468 & .904 & .995 & 1.000 \\
PDR--CV2             & .065 & .085 & .166 & .348 & .597 & .962 & .998 & 1.000 \\
ProxLL--raw          & .054 & .070 & .149 & .339 & .597 & .958 & .997 & 1.000 \\
ProxLL--US           & .055 & .057 & .082 & .132 & .252 & .619 & .934 & .994 \\
ProxLL--raw--matched & .064 & .074 & .173 & .373 & .644 & .965 & .997 & 1.000 \\
ProxLL--US--matched  & .056 & .058 & .082 & .131 & .251 & .618 & .934 & .994 \\
\midrule
\multicolumn{9}{l}{$n=7{,}500$}\\
PDR--PI1             & .059 & .065 & .114 & .265 & .563 & .952 & 1.000 & 1.000 \\
PDR--PI2             & .046 & .060 & .099 & .234 & .489 & .930 & 1.000 & 1.000 \\
PDR--CV1             & .049 & .066 & .150 & .384 & .703 & .991 & 1.000 & 1.000 \\
PDR--CV2             & .053 & .079 & .210 & .509 & .815 & .996 & 1.000 & 1.000 \\
ProxLL--raw          & .053 & .077 & .209 & .528 & .819 & .996 & 1.000 & 1.000 \\
ProxLL--US           & .041 & .052 & .092 & .173 & .377 & .868 & .995 & 1.000 \\
ProxLL--raw--matched & .050 & .077 & .231 & .551 & .839 & .997 & 1.000 & 1.000 \\
ProxLL--US--matched  & .040 & .051 & .092 & .176 & .378 & .868 & .995 & 1.000 \\
\midrule
\multicolumn{9}{l}{$n=10{,}000$}\\
PDR--PI1             & .058 & .083 & .156 & .376 & .685 & .989 & 1.000 & 1.000 \\
PDR--PI2             & .058 & .083 & .155 & .344 & .638 & .985 & 1.000 & 1.000 \\
PDR--CV1             & .052 & .085 & .206 & .506 & .828 & .999 & 1.000 & 1.000 \\
PDR--CV2             & .061 & .098 & .266 & .636 & .926 & .999 & 1.000 & 1.000 \\
ProxLL--raw          & .055 & .102 & .278 & .651 & .927 & .999 & 1.000 & 1.000 \\
ProxLL--US           & .048 & .068 & .122 & .262 & .535 & .952 & 1.000 & 1.000 \\
ProxLL--raw--matched & .058 & .101 & .308 & .680 & .942 & .999 & 1.000 & 1.000 \\
ProxLL--US--matched  & .049 & .068 & .124 & .261 & .535 & .952 & 1.000 & 1.000 \\
\bottomrule
\end{tabular}
\endgroup
\end{table}

% CM: constant
\begin{table}[p]
\centering
\begingroup
\footnotesize
\setlength{\tabcolsep}{4pt}
\renewcommand{\arraystretch}{1.05}
\caption{Empirical rejection probabilities for the constant-null test in CM.
Each entry uses $1{,}000$ replications.}
\label{tab:power-cm-constant}
\begin{tabular}{l*{8}{r}}
\toprule
Method & $0$ & $.004$ & $.008$ & $.012$ & $.016$ & $.024$ & $.032$ & $.040$ \\
\midrule
\multicolumn{9}{l}{$n=2{,}500$}\\
PDR--PI1    & .057 & .063 & .080 & .137 & .214 & .527 & .854 & .985 \\
PDR--PI2    & .047 & .054 & .069 & .108 & .162 & .396 & .709 & .937 \\
PDR--CV1    & .052 & .060 & .083 & .142 & .236 & .566 & .870 & .990 \\
PDR--CV2    & .050 & .059 & .096 & .171 & .307 & .692 & .950 & .997 \\
ProxLL--raw & .048 & .062 & .116 & .249 & .445 & .844 & .979 & 1.000 \\
ProxLL--US  & .050 & .060 & .075 & .109 & .175 & .416 & .734 & .932 \\
\midrule
\multicolumn{9}{l}{$n=5{,}000$}\\
PDR--PI1    & .056 & .066 & .114 & .246 & .454 & .880 & .995 & 1.000 \\
PDR--PI2    & .055 & .067 & .088 & .162 & .322 & .722 & .977 & .999 \\
PDR--CV1    & .058 & .069 & .123 & .256 & .474 & .905 & .996 & 1.000 \\
PDR--CV2    & .057 & .077 & .159 & .324 & .595 & .959 & .998 & 1.000 \\
ProxLL--raw & .059 & .085 & .196 & .446 & .750 & .986 & 1.000 & 1.000 \\
ProxLL--US  & .059 & .068 & .093 & .166 & .336 & .739 & .973 & .999 \\
\midrule
\multicolumn{9}{l}{$n=7{,}500$}\\
PDR--PI1    & .055 & .066 & .131 & .341 & .696 & .984 & 1.000 & 1.000 \\
PDR--PI2    & .041 & .053 & .104 & .235 & .508 & .948 & 1.000 & 1.000 \\
PDR--CV1    & .049 & .063 & .144 & .372 & .714 & .991 & 1.000 & 1.000 \\
PDR--CV2    & .047 & .079 & .202 & .492 & .822 & .999 & 1.000 & 1.000 \\
ProxLL--raw & .048 & .092 & .285 & .660 & .919 & 1.000 & 1.000 & 1.000 \\
ProxLL--US  & .042 & .054 & .101 & .241 & .501 & .935 & .998 & 1.000 \\
\midrule
\multicolumn{9}{l}{$n=10{,}000$}\\
PDR--PI1    & .049 & .083 & .193 & .455 & .815 & .999 & 1.000 & 1.000 \\
PDR--PI2    & .053 & .075 & .147 & .349 & .647 & .987 & 1.000 & 1.000 \\
PDR--CV1    & .050 & .082 & .208 & .486 & .835 & .999 & 1.000 & 1.000 \\
PDR--CV2    & .055 & .100 & .250 & .617 & .926 & .999 & 1.000 & 1.000 \\
ProxLL--raw & .059 & .115 & .350 & .794 & .976 & 1.000 & 1.000 & 1.000 \\
ProxLL--US  & .052 & .074 & .144 & .336 & .664 & .983 & 1.000 & 1.000 \\
\bottomrule
\end{tabular}
\endgroup
\end{table}

% CM: linear
\begin{table}[p]
\centering
\begingroup
\footnotesize
\setlength{\tabcolsep}{4pt}
\renewcommand{\arraystretch}{1.05}
\caption{Empirical rejection probabilities for the linear-null test in CM.
Each entry uses $1{,}000$ replications.}
\label{tab:power-cm-linear}
\begin{tabular}{l*{8}{r}}
\toprule
Method & $0$ & $.004$ & $.008$ & $.012$ & $.016$ & $.024$ & $.032$ & $.040$ \\
\midrule
\multicolumn{9}{l}{$n=2{,}500$}\\
PDR--PI1    & .054 & .056 & .083 & .132 & .218 & .528 & .852 & .988 \\
PDR--PI2    & .049 & .053 & .072 & .102 & .163 & .378 & .693 & .946 \\
PDR--CV1    & .055 & .061 & .082 & .139 & .239 & .564 & .873 & .991 \\
PDR--CV2    & .050 & .059 & .092 & .184 & .324 & .708 & .946 & .997 \\
ProxLL--raw & .049 & .056 & .126 & .261 & .457 & .863 & .977 & 1.000 \\
ProxLL--US  & .051 & .053 & .076 & .104 & .174 & .409 & .727 & .928 \\
\midrule
\multicolumn{9}{l}{$n=5{,}000$}\\
PDR--PI1    & .060 & .064 & .111 & .244 & .462 & .885 & .995 & 1.000 \\
PDR--PI2    & .058 & .069 & .097 & .167 & .320 & .730 & .979 & 1.000 \\
PDR--CV1    & .059 & .061 & .113 & .262 & .478 & .917 & .997 & 1.000 \\
PDR--CV2    & .058 & .073 & .158 & .344 & .608 & .958 & .998 & 1.000 \\
ProxLL--raw & .054 & .084 & .209 & .470 & .772 & .988 & 1.000 & 1.000 \\
ProxLL--US  & .060 & .068 & .091 & .171 & .333 & .743 & .974 & .998 \\
\midrule
\multicolumn{9}{l}{$n=7{,}500$}\\
PDR--PI1    & .053 & .063 & .128 & .356 & .709 & .988 & 1.000 & 1.000 \\
PDR--PI2    & .038 & .056 & .096 & .236 & .516 & .946 & 1.000 & 1.000 \\
PDR--CV1    & .055 & .065 & .140 & .393 & .740 & .992 & 1.000 & 1.000 \\
PDR--CV2    & .048 & .073 & .204 & .515 & .846 & .997 & 1.000 & 1.000 \\
ProxLL--raw & .043 & .084 & .308 & .711 & .937 & .999 & 1.000 & 1.000 \\
ProxLL--US  & .039 & .053 & .104 & .237 & .512 & .930 & .998 & 1.000 \\
\midrule
\multicolumn{9}{l}{$n=10{,}000$}\\
PDR--PI1    & .051 & .084 & .194 & .469 & .823 & .998 & 1.000 & 1.000 \\
PDR--PI2    & .055 & .077 & .146 & .345 & .652 & .989 & 1.000 & 1.000 \\
PDR--CV1    & .052 & .092 & .208 & .506 & .851 & .999 & 1.000 & 1.000 \\
PDR--CV2    & .055 & .104 & .253 & .643 & .944 & .999 & 1.000 & 1.000 \\
ProxLL--raw & .054 & .121 & .381 & .825 & .983 & 1.000 & 1.000 & 1.000 \\
ProxLL--US  & .055 & .073 & .144 & .347 & .666 & .986 & 1.000 & 1.000 \\
\bottomrule
\end{tabular}
\endgroup
\end{table}

% CM: quadratic
\begin{table}[p]
\centering
\begingroup
\footnotesize
\setlength{\tabcolsep}{4pt}
\renewcommand{\arraystretch}{1.05}
\caption{Empirical rejection probabilities for the quadratic-null test in CM.
Each entry uses $1{,}000$ replications.}
\label{tab:power-cm-quadratic}
\begin{tabular}{l*{8}{r}}
\toprule
Method & $0$ & $.004$ & $.008$ & $.012$ & $.016$ & $.024$ & $.032$ & $.040$ \\
\midrule
\multicolumn{9}{l}{$n=2{,}500$}\\
PDR--PI1             & .059 & .065 & .082 & .128 & .180 & .419 & .744 & .957 \\
PDR--PI2             & .055 & .059 & .076 & .110 & .159 & .380 & .704 & .942 \\
PDR--CV1             & .055 & .065 & .093 & .144 & .243 & .575 & .877 & .992 \\
PDR--CV2             & .051 & .069 & .110 & .207 & .344 & .725 & .954 & .998 \\
ProxLL--raw          & .052 & .064 & .099 & .200 & .336 & .721 & .959 & .997 \\
ProxLL--US           & .049 & .056 & .069 & .088 & .139 & .296 & .595 & .845 \\
ProxLL--raw--matched & .056 & .069 & .110 & .209 & .360 & .744 & .965 & .998 \\
ProxLL--US--matched  & .047 & .056 & .070 & .089 & .136 & .299 & .600 & .846 \\
\midrule
\multicolumn{9}{l}{$n=5{,}000$}\\
PDR--PI1             & .063 & .069 & .101 & .175 & .378 & .780 & .987 & 1.000 \\
PDR--PI2             & .056 & .065 & .095 & .166 & .321 & .723 & .978 & 1.000 \\
PDR--CV1             & .059 & .068 & .127 & .271 & .493 & .920 & .998 & 1.000 \\
PDR--CV2             & .062 & .079 & .169 & .368 & .634 & .970 & .998 & 1.000 \\
ProxLL--raw          & .054 & .071 & .160 & .350 & .637 & .976 & 1.000 & 1.000 \\
ProxLL--US           & .053 & .058 & .076 & .131 & .250 & .630 & .935 & .995 \\
ProxLL--raw--matched & .055 & .078 & .176 & .380 & .673 & .978 & 1.000 & 1.000 \\
ProxLL--US--matched  & .052 & .058 & .077 & .132 & .249 & .630 & .935 & .995 \\
\midrule
\multicolumn{9}{l}{$n=7{,}500$}\\
PDR--PI1             & .053 & .065 & .111 & .275 & .578 & .958 & 1.000 & 1.000 \\
PDR--PI2             & .049 & .057 & .098 & .232 & .518 & .947 & 1.000 & 1.000 \\
PDR--CV1             & .050 & .074 & .161 & .399 & .734 & .993 & 1.000 & 1.000 \\
PDR--CV2             & .045 & .082 & .228 & .534 & .845 & .997 & 1.000 & 1.000 \\
ProxLL--raw          & .050 & .076 & .225 & .565 & .841 & .999 & 1.000 & 1.000 \\
ProxLL--US           & .037 & .052 & .090 & .174 & .369 & .880 & .996 & 1.000 \\
ProxLL--raw--matched & .045 & .077 & .230 & .588 & .854 & .999 & 1.000 & 1.000 \\
ProxLL--US--matched  & .038 & .052 & .090 & .173 & .370 & .880 & .996 & 1.000 \\
\midrule
\multicolumn{9}{l}{$n=10{,}000$}\\
PDR--PI1             & .060 & .080 & .170 & .382 & .704 & .993 & 1.000 & 1.000 \\
PDR--PI2             & .062 & .079 & .161 & .358 & .659 & .989 & 1.000 & 1.000 \\
PDR--CV1             & .050 & .093 & .215 & .528 & .858 & .999 & 1.000 & 1.000 \\
PDR--CV2             & .062 & .101 & .286 & .671 & .945 & .999 & 1.000 & 1.000 \\
ProxLL--raw          & .062 & .100 & .310 & .685 & .933 & 1.000 & 1.000 & 1.000 \\
ProxLL--US           & .052 & .070 & .124 & .266 & .549 & .965 & 1.000 & 1.000 \\
ProxLL--raw--matched & .060 & .099 & .314 & .707 & .942 & 1.000 & 1.000 & 1.000 \\
ProxLL--US--matched  & .051 & .072 & .124 & .267 & .549 & .965 & 1.000 & 1.000 \\
\bottomrule
\end{tabular}
\endgroup
\end{table}

% MC: constant
\begin{table}[p]
\centering
\begingroup
\footnotesize
\setlength{\tabcolsep}{4pt}
\renewcommand{\arraystretch}{1.05}
\caption{Empirical rejection probabilities for the constant-null test in MC.
Each entry uses $1{,}000$ replications.}
\label{tab:power-mc-constant}
\begin{tabular}{l*{8}{r}}
\toprule
Method & $0$ & $.004$ & $.008$ & $.012$ & $.016$ & $.024$ & $.032$ & $.040$ \\
\midrule
\multicolumn{9}{l}{$n=2{,}500$}\\
PDR--PI1    & .067 & .070 & .082 & .103 & .142 & .259 & .497 & .746 \\
PDR--PI2    & .052 & .056 & .068 & .083 & .104 & .190 & .359 & .594 \\
PDR--CV1    & .066 & .075 & .085 & .116 & .153 & .283 & .544 & .786 \\
PDR--CV2    & .063 & .079 & .095 & .131 & .188 & .387 & .658 & .875 \\
ProxLL--raw & .066 & .082 & .105 & .162 & .244 & .522 & .814 & .956 \\
ProxLL--US  & .051 & .054 & .064 & .085 & .114 & .201 & .400 & .620 \\
\midrule
\multicolumn{9}{l}{$n=5{,}000$}\\
PDR--PI1    & .047 & .058 & .082 & .138 & .214 & .541 & .861 & .982 \\
PDR--PI2    & .055 & .059 & .069 & .098 & .150 & .369 & .720 & .933 \\
PDR--CV1    & .041 & .056 & .083 & .139 & .238 & .580 & .884 & .986 \\
PDR--CV2    & .054 & .064 & .099 & .181 & .326 & .696 & .948 & .999 \\
ProxLL--raw & .049 & .070 & .132 & .243 & .444 & .845 & .987 & .999 \\
ProxLL--US  & .046 & .054 & .069 & .098 & .156 & .397 & .731 & .925 \\
\midrule
\multicolumn{9}{l}{$n=7{,}500$}\\
PDR--PI1    & .049 & .062 & .098 & .178 & .333 & .757 & .980 & .999 \\
PDR--PI2    & .043 & .053 & .076 & .131 & .230 & .579 & .923 & 1.000 \\
PDR--CV1    & .046 & .060 & .103 & .185 & .361 & .792 & .981 & 1.000 \\
PDR--CV2    & .049 & .068 & .121 & .243 & .457 & .886 & .993 & 1.000 \\
ProxLL--raw & .055 & .070 & .140 & .341 & .598 & .953 & .998 & 1.000 \\
ProxLL--US  & .036 & .044 & .075 & .126 & .221 & .585 & .924 & .993 \\
\midrule
\multicolumn{9}{l}{$n=10{,}000$}\\
PDR--PI1    & .069 & .090 & .134 & .268 & .470 & .906 & .998 & 1.000 \\
PDR--PI2    & .068 & .078 & .109 & .190 & .342 & .764 & .983 & .999 \\
PDR--CV1    & .070 & .089 & .131 & .271 & .489 & .923 & .998 & 1.000 \\
PDR--CV2    & .073 & .093 & .167 & .341 & .607 & .960 & .999 & 1.000 \\
ProxLL--raw & .069 & .099 & .220 & .470 & .760 & .992 & 1.000 & 1.000 \\
ProxLL--US  & .057 & .073 & .112 & .196 & .341 & .764 & .976 & .998 \\
\bottomrule
\end{tabular}
\endgroup
\end{table}

% MC: linear
\begin{table}[p]
\centering
\begingroup
\footnotesize
\setlength{\tabcolsep}{4pt}
\renewcommand{\arraystretch}{1.05}
\caption{Empirical rejection probabilities for the linear-null test in MC.
Each entry uses $1{,}000$ replications.}
\label{tab:power-mc-linear}
\begin{tabular}{l*{8}{r}}
\toprule
Method & $0$ & $.004$ & $.008$ & $.012$ & $.016$ & $.024$ & $.032$ & $.040$ \\
\midrule
\multicolumn{9}{l}{$n=2{,}500$}\\
PDR--PI1    & .063 & .066 & .078 & .096 & .131 & .249 & .476 & .723 \\
PDR--PI2    & .050 & .052 & .061 & .076 & .098 & .183 & .336 & .573 \\
PDR--CV1    & .065 & .069 & .081 & .105 & .139 & .276 & .504 & .771 \\
PDR--CV2    & .068 & .074 & .091 & .124 & .180 & .377 & .631 & .876 \\
ProxLL--raw & .065 & .070 & .098 & .150 & .221 & .496 & .803 & .945 \\
ProxLL--US  & .046 & .051 & .059 & .075 & .100 & .197 & .382 & .594 \\
\midrule
\multicolumn{9}{l}{$n=5{,}000$}\\
PDR--PI1    & .050 & .055 & .080 & .120 & .196 & .536 & .859 & .982 \\
PDR--PI2    & .057 & .061 & .070 & .095 & .141 & .353 & .696 & .934 \\
PDR--CV1    & .052 & .056 & .083 & .119 & .219 & .565 & .880 & .989 \\
PDR--CV2    & .058 & .061 & .089 & .162 & .320 & .683 & .952 & 1.000 \\
ProxLL--raw & .056 & .071 & .110 & .250 & .449 & .852 & .986 & .999 \\
ProxLL--US  & .050 & .052 & .062 & .099 & .146 & .388 & .718 & .918 \\
\midrule
\multicolumn{9}{l}{$n=7{,}500$}\\
PDR--PI1    & .049 & .057 & .091 & .176 & .334 & .757 & .978 & .999 \\
PDR--PI2    & .040 & .049 & .066 & .126 & .218 & .583 & .923 & .997 \\
PDR--CV1    & .047 & .056 & .097 & .187 & .352 & .787 & .981 & 1.000 \\
PDR--CV2    & .050 & .063 & .118 & .248 & .442 & .887 & .996 & 1.000 \\
ProxLL--raw & .048 & .068 & .141 & .343 & .599 & .954 & .998 & 1.000 \\
ProxLL--US  & .035 & .045 & .062 & .124 & .213 & .575 & .913 & .993 \\
\midrule
\multicolumn{9}{l}{$n=10{,}000$}\\
PDR--PI1    & .062 & .082 & .133 & .265 & .456 & .898 & .997 & 1.000 \\
PDR--PI2    & .062 & .074 & .108 & .196 & .341 & .758 & .982 & .999 \\
PDR--CV1    & .060 & .081 & .134 & .271 & .486 & .909 & .998 & 1.000 \\
PDR--CV2    & .065 & .090 & .165 & .329 & .597 & .960 & .999 & 1.000 \\
ProxLL--raw & .067 & .092 & .210 & .477 & .778 & .989 & 1.000 & 1.000 \\
ProxLL--US  & .057 & .067 & .112 & .188 & .332 & .762 & .976 & .997 \\
\bottomrule
\end{tabular}
\endgroup
\end{table}

% MC: quadratic
\begin{table}[p]
\centering
\begingroup
\footnotesize
\setlength{\tabcolsep}{4pt}
\renewcommand{\arraystretch}{1.05}
\caption{Empirical rejection probabilities for the quadratic-null test in MC.
Each entry uses $1{,}000$ replications.}
\label{tab:power-mc-quadratic}
\begin{tabular}{l*{8}{r}}
\toprule
Method & $0$ & $.004$ & $.008$ & $.012$ & $.016$ & $.024$ & $.032$ & $.040$ \\
\midrule
\multicolumn{9}{l}{$n=2{,}500$}\\
PDR--PI1             & .059 & .058 & .064 & .076 & .107 & .222 & .393 & .658 \\
PDR--PI2             & .048 & .054 & .057 & .076 & .099 & .190 & .332 & .558 \\
PDR--CV1             & .053 & .058 & .070 & .094 & .134 & .275 & .513 & .787 \\
PDR--CV2             & .054 & .065 & .081 & .119 & .183 & .399 & .668 & .877 \\
ProxLL--raw          & .049 & .057 & .079 & .112 & .174 & .432 & .718 & .908 \\
ProxLL--US           & .044 & .048 & .057 & .066 & .088 & .175 & .308 & .537 \\
ProxLL--raw--matched & .047 & .059 & .091 & .131 & .197 & .469 & .757 & .922 \\
ProxLL--US--matched  & .045 & .050 & .057 & .068 & .088 & .176 & .307 & .538 \\
\midrule
\multicolumn{9}{l}{$n=5{,}000$}\\
PDR--PI1             & .059 & .061 & .078 & .108 & .164 & .415 & .775 & .960 \\
PDR--PI2             & .052 & .053 & .069 & .101 & .144 & .356 & .698 & .928 \\
PDR--CV1             & .053 & .060 & .077 & .124 & .228 & .564 & .877 & .986 \\
PDR--CV2             & .055 & .061 & .089 & .187 & .327 & .716 & .946 & .997 \\
ProxLL--raw          & .052 & .069 & .098 & .203 & .362 & .758 & .965 & .997 \\
ProxLL--US           & .045 & .050 & .063 & .095 & .137 & .326 & .630 & .871 \\
ProxLL--raw--matched & .056 & .060 & .103 & .214 & .384 & .787 & .970 & .995 \\
ProxLL--US--matched  & .046 & .050 & .064 & .095 & .136 & .326 & .629 & .875 \\
\midrule
\multicolumn{9}{l}{$n=7{,}500$}\\
PDR--PI1             & .046 & .054 & .071 & .150 & .283 & .661 & .951 & .997 \\
PDR--PI2             & .039 & .044 & .065 & .122 & .224 & .581 & .927 & .997 \\
PDR--CV1             & .045 & .051 & .101 & .200 & .359 & .806 & .984 & .999 \\
PDR--CV2             & .042 & .065 & .118 & .259 & .460 & .889 & .996 & .999 \\
ProxLL--raw          & .043 & .058 & .130 & .293 & .520 & .917 & .995 & .999 \\
ProxLL--US           & .041 & .045 & .057 & .111 & .186 & .506 & .878 & .983 \\
ProxLL--raw--matched & .043 & .065 & .140 & .317 & .545 & .928 & .995 & .999 \\
ProxLL--US--matched  & .040 & .045 & .057 & .112 & .186 & .508 & .878 & .982 \\
\midrule
\multicolumn{9}{l}{$n=10{,}000$}\\
PDR--PI1             & .060 & .074 & .124 & .230 & .411 & .831 & .992 & 1.000 \\
PDR--PI2             & .061 & .069 & .109 & .198 & .339 & .760 & .979 & .999 \\
PDR--CV1             & .062 & .078 & .138 & .265 & .496 & .920 & .998 & 1.000 \\
PDR--CV2             & .064 & .088 & .166 & .337 & .625 & .962 & .999 & 1.000 \\
ProxLL--raw          & .070 & .091 & .182 & .398 & .686 & .978 & 1.000 & 1.000 \\
ProxLL--US           & .056 & .064 & .102 & .174 & .284 & .684 & .955 & .996 \\
ProxLL--raw--matched & .065 & .091 & .189 & .422 & .717 & .983 & 1.000 & 1.000 \\
ProxLL--US--matched  & .055 & .065 & .102 & .172 & .284 & .684 & .956 & .996 \\
\bottomrule
\end{tabular}
\endgroup
\end{table}

\subsection{Bridge compatibility of the continuous-treatment simulation design}
\label{sec:simulation-bridge-compatibility}

We verify that \eqref{eq:simulation-h0} and \eqref{eq:simulation-q0}
satisfy \eqref{eq:Bh} and \eqref{eq:Bq}, respectively.  Since
\[
Y-h_0(W,A,X)=\varepsilon_Y
\]
and \(\varepsilon_Y\) is independent of \((Z,A,X,W)\),
\[
E(Y\mid Z,A,X)
=
E\{h_0(W,A,X)\mid Z,A,X\}.
\]
Thus \(h_0\) satisfies \eqref{eq:Bh}.  Also, $W\sim N(0,161/128)$.
The symmetry of \(W\) and the oddness of the hyperbolic tangent give
$E\{\tanh(W)\}=0$,
which yields \eqref{eq:simulation-theta0}.

For the treatment bridge,  we get
\[
A\mid U=u,X=x
\sim
N\!\left\{0.25(x_1+x_2)+0.15u,\,0.79\right\},
\qquad
A\sim N(0,V_A).
\]
It follows that
\[
\frac{f_A(a)}{f(a\mid u,x)}
=
\left(\frac{0.79}{V_A}\right)^{1/2}
\exp\!\left\{
-\frac{a^2}{2V_A}
+\frac{\{d(a,x)-0.15u\}^2}{2(0.79)}
\right\}.
\]
Conditional on \(U=u,A=a,X=x\), we have
\[
r(Z,a,x)
=
u-\frac{d(a,x)}{0.15}+\varepsilon_Z.
\]
For \(R\sim N(\mu,\tau^2)\) and \(v>\tau^2\), one has
\[
E\!\left\{\exp\left(\frac{R^2}{2v}\right)\right\}
=
\left(\frac{v}{v-\tau^2}\right)^{1/2}
\exp\!\left\{\frac{\mu^2}{2(v-\tau^2)}\right\}.
\]
Since
\[
V_R-0.5^2=\frac{0.79}{0.15^2},
\qquad
Q\left(\frac{V_R}{V_R-0.5^2}\right)^{1/2}
=
\left(\frac{0.79}{V_A}\right)^{1/2},
\]
we obtain
\[
\begin{aligned}
&E\{q_0(Z,a,X)\mid U=u,A=a,X=x\}\\
&\quad=
\left(\frac{0.79}{V_A}\right)^{1/2}
\exp\!\left\{
-\frac{a^2}{2V_A}
+\frac{\{d(a,x)-0.15u\}^2}{2(0.79)}
\right\}\\
&\quad=
\frac{f_A(a)}{f(a\mid u,x)}.
\end{aligned}
\]

Finally, \(W\perp\!\!\!\perp(Z,A)\mid(U,X)\) and Bayes' rule yields
\[
\begin{aligned}
E\{q_0(Z,a,X)\mid W=w,A=a,X=x\}
&=
\int
\frac{f_A(a)}{f(a\mid u,x)}
\,dF(u\mid w,a,x)\\
&=
\int
\frac{f_A(a)}{f(a\mid u,x)}
\frac{f(a\mid u,x)}{f(a\mid w,x)}
\,dF(u\mid w,x)\\
&=
\frac{f_A(a)}{f(a\mid w,x)}.
\end{aligned}
\]
Thus \(q_0\) satisfies \eqref{eq:Bq}.

  \clearpage
\fi
\bibliography{allref}

@article{Hall1984central,
  author  = {Hall, Peter},
  title   = {Central limit theorem for integrated square error of multivariate nonparametric density estimators},
  journal = {Journal of Multivariate Analysis},
  year    = {1984},
  volume  = {14},
  number  = {1},
  pages   = {1--16},
  doi     = {10.1016/0047-259X(84)90044-7}
}

@article{Chernozhukov2014Gaussian,
author = {Victor Chernozhukov and Denis Chetverikov and Kengo Kato},
title = {{Gaussian approximation of suprema of empirical processes}},
volume = {42},
journal = {The Annals of Statistics},
number = {4},
publisher = {Institute of Mathematical Statistics},
pages = {1564 -- 1597},
year = {2014}
}

@article{ai2022estimation,
  title={Estimation and inference for the counterfactual distribution and quantile functions in continuous treatment models},
  author={Ai, Chunrong and Linton, Oliver and Zhang, Zheng},
  journal={Journal of Econometrics},
  volume={228},
  number={1},
  pages={39--61},
  year={2022},
  publisher={Elsevier}
}

@article{hardle1993comparing,
  title={Comparing nonparametric versus parametric regression fits},
  author={H{\"a}rdle, Wolfgang and Mammen, Enno},
  journal={The Annals of Statistics},
  pages={1926--1947},
  year={1993},
  publisher={JSTOR}
}

@article{alcala1999goodness,
  title={Goodness-of-fit test for linear models based on local polynomials},
  author={Alcal{\'a}, JT and Crist{\'o}bal, JA and Gonz{\'a}lez-Manteiga, W},
  journal={Statistics \& Probability Letters},
  volume={42},
  number={1},
  pages={39--46},
  year={1999},
  publisher={Elsevier}
}

@article{huang2023nonparametric,
  title={Nonparametric estimation of the continuous treatment effect with measurement error},
  author={Huang, Wei and Zhang, Zheng},
  journal={Journal of the Royal Statistical Society Series B: Statistical Methodology},
  volume={85},
  number={2},
  pages={474--496},
  year={2023},
  publisher={Oxford University Press US}
}

@article{calonico2018effect,
  title={On the effect of bias estimation on coverage accuracy in nonparametric inference},
  author={Calonico, Sebastian and Cattaneo, Matias D and Farrell, Max H},
  journal={Journal of the American Statistical Association},
  volume={113},
  number={522},
  pages={767--779},
  year={2018},
  publisher={Taylor \& Francis}
}

@article{cui2024semiparametric,
  title={Semiparametric proximal causal inference},
  author={Cui, Yifan and Pu, Hongming and Shi, Xu and Miao, Wang and Tchetgen Tchetgen, Eric},
  journal={Journal of the American Statistical Association},
  volume={119},
  number={546},
  pages={1348--1359},
  year={2024},
  publisher={Taylor \& Francis}
}

@article{Westling2020causal,
    author = {Westling, Ted and Gilbert, Peter and Carone, Marco},
    title = {Causal Isotonic Regression},
    journal = {Journal of the Royal Statistical Society Series B: Statistical Methodology},
    volume = {82},
    number = {3},
    pages = {719-747},
    year = {2020},
    month = {05}
}

@article{miao2018identifying,
  title        = {Identifying Causal Effects with Proxy Variables of an Unmeasured Confounder},
  author       = {Miao, Wang and Geng, Zhi and Tchetgen Tchetgen, Eric J.},
  journal      = {Biometrika},
  volume       = {105},
  number       = {4},
  pages        = {987--993},
  year         = {2018}
}

@incollection{carrasco2007linear,
  title        = {Linear Inverse Problems in Structural Econometrics Estimation Based on Spectral Decomposition and Regularization},
  author       = {Carrasco, Marine and Florens, Jean-Pierre and Renault, Eric},
  booktitle    = {Handbook of Econometrics},
  editor       = {Heckman, James J. and Leamer, Edward E.},
  volume       = {6B},
  pages        = {5633--5751},
  publisher    = {Elsevier},
  address      = {Amsterdam},
  year         = {2007},
  doi          = {10.1016/S1573-4412(07)06077-1}
}

@book{kress2014linear,
  title        = {Linear Integral Equations},
  author       = {Kress, Rainer},
  series       = {Applied Mathematical Sciences},
  volume       = {82},
  edition      = {3},
  publisher    = {Springer},
  address      = {New York},
  year         = {2014},
  isbn         = {978-1-4614-9592-5}
}

@article{Takatsu2025debiased,
  author  = {Takatsu, Kenta and Westling, Ted},
  title   = {Debiased inference for a covariate-adjusted regression function},
  journal = {Journal of the Royal Statistical Society Series B: Statistical Methodology},
  volume  = {87},
  number  = {1},
  pages   = {33--55},
  year    = {2025}
}

@article{Doss2024nonparametric,
author = {Charles R. Doss and Guangwei Weng and Lan Wang and Ira Moscovice and Tongtan Chantarat},
title = {{A nonparametric doubly robust test for a continuous treatment effect}},
volume = {52},
journal = {The Annals of Statistics},
number = {4},
publisher = {Institute of Mathematical Statistics},
pages = {1592 -- 1615},
year = {2024}
}

@article{belloni2018uniformly,
  author  = {Belloni, Alexandre and Chernozhukov, Victor and Chetverikov, Denis and Wei, Ying},
  title   = {Uniformly Valid Post-Regularization Confidence Regions for Many Functional Parameters in Z-Estimation Framework},
  journal = {The Annals of Statistics},
  year    = {2018},
  volume  = {46},
  number  = {6B},
  pages   = {3643--3675}
}

@article{kennedy2023towards,
  title={Towards optimal doubly robust estimation of heterogeneous causal effects},
  author={Kennedy, Edward H},
  journal={Electronic Journal of Statistics},
  volume={17},
  number={2},
  pages={3008--3049},
  year={2023},
  publisher={The Institute of Mathematical Statistics and the Bernoulli Society}
}

@article{doss2026doubly,
  title={Doubly robust pointwise confidence intervals for a monotonic continuous treatment effect curve},
  author={Doss, Charles R},
  journal={Journal of the American Statistical Association},
  pages={1--12},
  year={2026},
  publisher={Taylor \& Francis}
}

@book{wasserman2006all,
  title={All of nonparametric statistics},
  author={Wasserman, Larry},
  year={2006},
  publisher={Springer}
}

@article{kennedy2017non,
  title={Non-parametric methods for doubly robust estimation of continuous treatment effects},
  author={Kennedy, Edward H and Ma, Zongming and McHugh, Matthew D and Small, Dylan S},
  journal={Journal of the Royal Statistical Society Series B: Statistical Methodology},
  volume={79},
  number={4},
  pages={1229--1245},
  year={2017},
  publisher={Oxford University Press}
}

@article{Arcones1993limit,
  author    = {Arcones, Miguel A. and Gin{\'e}, Evarist},
  title     = {Limit Theorems for {$U$}-Processes},
  journal   = {The Annals of Probability},
  volume    = {21},
  number    = {3},
  pages     = {1494--1542},
  year      = {1993},
  month     = {jul},
  publisher = {Institute of Mathematical Statistics},
  url       = {https://www.jstor.org/stable/2244585}
}

@book{VanderVaartWellner2023weak,
  author    = {{van der Vaart}, Aad W. and Wellner, Jon A.},
  title     = {Weak Convergence and Empirical Processes: With Applications to Statistics},
  edition   = {2},
  series    = {Springer Series in Statistics},
  publisher = {Springer Nature Switzerland},
  address   = {Cham},
  year      = {2023},
  isbn      = {978-3-031-29038-1}
}

@article{donohue2001impact,
  author  = {Donohue, John J. and Levitt, Steven D.},
  title   = {The Impact of Legalized Abortion on Crime},
  journal = {The Quarterly Journal of Economics},
  volume  = {116},
  number  = {2},
  pages   = {379--420},
  year    = {2001},
  doi     = {10.1162/00335530151144050}
}

@inproceedings{mastouri2021proximal,
  title={Proximal Causal Learning with Kernels: Two-Stage Estimation and Moment Restriction},
  author={Mastouri, Afsaneh and Zhu, Yuchen and Gultchin, Limor and Korba, Anna and Silva, Ricardo and Kusner, Matt J. and Gretton, Arthur and Muandet, Krikamol},
  booktitle={Proceedings of the 38th International Conference on Machine Learning},
  pages={7512--7523},
  year={2021},
  volume={139},
  series={Proceedings of Machine Learning Research},
  publisher={PMLR}
}

@article{woody2020estimating,
  title={Estimating Heterogeneous Effects of Continuous Exposures Using Bayesian Tree Ensembles: Revisiting the Impact of Abortion Rates on Crime},
  author={Woody, Spencer and Carvalho, Carlos M. and Hahn, P. Richard and Murray, Jared S.},
  journal={arXiv preprint arXiv:2007.09845},
  year={2020}
}

@inproceedings{wu2024doubly,
  title={Doubly Robust Proximal Causal Learning for Continuous Treatments},
  author={Wu, Yong and Fu, Yanwei and Wang, Shouyan and Sun, Xinwei},
  booktitle={The Twelfth International Conference on Learning Representations},
  year={2024},
  url={https://openreview.net/forum?id=TjGJFkU3xL}
}

@incollection{hirano2004propensity,
  author    = {Hirano, Keisuke and Imbens, Guido W.},
  title     = {The Propensity Score with Continuous Treatments},
  booktitle = {Applied Bayesian Modeling and Causal Inference from Incomplete-Data Perspectives},
  editor    = {Gelman, Andrew and Meng, Xiao-Li},
  publisher = {Wiley},
  pages     = {73--84},
  year      = {2004}
}

@article{ImaiVanDyk2004,
  title   = {Causal Inference with General Treatment Regimes: Generalizing the Propensity Score},
  author  = {Imai, Kosuke and Van Dyk, David A.},
  journal = {Journal of the American Statistical Association},
  volume  = {99},
  number  = {467},
  pages   = {854--866},
  year    = {2004}
}

@article{GalvaoWang2015,
  title   = {Uniformly Semiparametric Efficient Estimation of Treatment Effects with a Continuous Treatment},
  author  = {Galvao, Antonio F. and Wang, Liang},
  journal = {Journal of the American Statistical Association},
  volume  = {110},
  number  = {512},
  pages   = {1528--1542},
  year    = {2015}
}

@article{ColangeloLee2026,
  title   = {Double Debiased Machine Learning Nonparametric Inference with Continuous Treatments},
  author  = {Colangelo, Kyle and Lee, Ying-Ying},
  journal = {Journal of Business \& Economic Statistics},
  volume  = {44},
  number  = {1},
  pages   = {67--79},
  year    = {2026}
}

@book{pfanzagl2012contributions,
  title={Contributions to a general asymptotic statistical theory},
  author={Pfanzagl, Johann},
  year={2012},
  publisher={Springer Science \& Business Media}
}

@article{bickel1982adaptive,
  title={On adaptive estimation},
  author={Bickel, Peter J},
  journal={The Annals of Statistics},
  pages={647--671},
  year={1982},
  publisher={JSTOR}
}

@article{Chernozhukov2018,
  title   = {Double/Debiased Machine Learning for Treatment and Structural Parameters},
  author  = {Chernozhukov, Victor and Chetverikov, Denis and Demirer, Mert and Duflo, Esther and Hansen, Christian and Newey, Whitney and Robins, James},
  journal = {The Econometrics Journal},
  volume  = {21},
  number  = {1},
  pages   = {C1--C68},
  year    = {2018}
}

@incollection{vanderlaan2011crossvalidated,
  author    = {van der Laan, Mark J. and Rose, Sherri and Zheng, Wenjing},
  title     = {Cross-Validated Targeted Minimum-Loss-Based Estimation},
  booktitle = {Targeted Learning: Causal Inference for Observational and Experimental Data},
  editor    = {van der Laan, Mark J. and Rose, Sherri},
  publisher = {Springer},
  year      = {2011},
  pages     = {459--474}
}

@article{tchetgen2024introduction,
  title   = {An Introduction to Proximal Causal Inference},
  author  = {Tchetgen Tchetgen, Eric J. and Ying, Andrew and Cui, Yifan and Shi, Xu and Miao, Wang},
  journal = {Statistical Science},
  volume  = {39},
  number  = {3},
  pages   = {375--390},
  year    = {2024}
}

@article{Ying2023,
  title   = {Proximal Causal Inference for Complex Longitudinal Studies},
  author  = {Ying, Andrew and Miao, Wang and Shi, Xu and Tchetgen Tchetgen, Eric J.},
  journal = {Journal of the Royal Statistical Society Series B: Statistical Methodology},
  volume  = {85},
  number  = {3},
  pages   = {684--704},
  year    = {2023}
}

@article{Ying2024,
  title   = {Proximal Survival Analysis to Handle Dependent Right Censoring},
  author  = {Ying, Andrew},
  journal = {Journal of the Royal Statistical Society Series B: Statistical Methodology},
  volume  = {86},
  number  = {5},
  pages   = {1414--1434},
  year    = {2024}
}

@article{Dukes2023,
  title   = {Proximal Mediation Analysis},
  author  = {Dukes, Oliver and Shpitser, Ilya and Tchetgen Tchetgen, Eric J.},
  journal = {Biometrika},
  volume  = {110},
  number  = {4},
  pages   = {973--987},
  year    = {2023}
}

@article{shi2026theory,
  title   = {Theory for Identification and Inference with Synthetic Controls: A Proximal Causal Inference Framework},
  author  = {Shi, Xu and Li, Kendrick Qijun and Yu, Myeonghun and Miao, Wang and Kuchibhotla, Arun Kumar and Hu, Mengtong and Tchetgen Tchetgen, Eric},
  journal = {Journal of the American Statistical Association},
  pages   = {1--23},
  year    = {2026}
}

@article{qi2024proximal,
  title   = {Proximal Learning for Individualized Treatment Regimes under Unmeasured Confounding},
  author  = {Qi, Zhengling and Miao, Rui and Zhang, Xiaoke},
  journal = {Journal of the American Statistical Association},
  volume  = {119},
  number  = {546},
  pages   = {915--928},
  year    = {2024}
}

@inproceedings{Shen2023,
  title     = {Optimal Treatment Regimes for Proximal Causal Learning},
  author    = {Shen, Tao and Cui, Yifan},
  booktitle = {Advances in Neural Information Processing Systems},
  volume    = {36},
  pages     = {47735--47748},
  year      = {2023}
}

@inproceedings{shi2022minimax,
  title     = {A Minimax Learning Approach to Off-Policy Evaluation in Confounded Partially Observable Markov Decision Processes},
  author    = {Shi, Chengchun and Uehara, Masatoshi and Huang, Jiawei and Jiang, Nan},
  booktitle = {International Conference on Machine Learning},
  pages     = {20057--20094},
  year      = {2022},
  organization = {PMLR}
}

@inproceedings{sverdrup2023proximal,
  title     = {Proximal Causal Learning of Conditional Average Treatment Effects},
  author    = {Sverdrup, Erik and Cui, Yifan},
  booktitle = {International Conference on Machine Learning},
  pages     = {33285--33298},
  year      = {2023},
  organization = {PMLR}
}

@article{Zhang2026,
  title   = {On Identification of Optimal Dynamic Treatment Regimes with Proxies of Hidden Confounders},
  author  = {Zhang, Jeffrey and Tchetgen Tchetgen, Eric},
  journal = {Observational Studies},
  volume  = {12},
  number  = {1},
  pages   = {1--15},
  year    = {2026}
}

@article{wang2026blessing,
  title   = {Blessing from Human-{AI} Interaction: Super Policy Learning in Confounded Environments},
  author  = {Wang, Jiayi and Shi, Chengchun and Qi, Zhengling},
  journal = {Journal of the American Statistical Association},
  pages   = {1--14},
  year    = {2026}
}

@article{bai2026proximal,
  title={Proximal Path-Specific Inference},
  author={Bai, Yang and Wu, Sihan and Sun, Baoluo and Cui, Yifan},
  journal={arXiv preprint arXiv:2605.09462},
  year={2026}
}

@article{ghassami2025causal,
  title={Causal inference with hidden mediators},
  author={Ghassami, Amiremad and Yang, Alan and Shpitser, Ilya and Tchetgen Tchetgen, Eric},
  journal={Biometrika},
  volume={112},
  number={1},
  pages={asae037},
  year={2025},
  publisher={Oxford University Press}
}

@article{wu2026proximal,
  title={Proximal Mediation Analysis with Hidden Recanting Witnesses},
  author={Wu, Sihan and Bai, Yang and Cui, Yifan},
  journal={arXiv preprint arXiv:2606.17600},
  year={2026}
}

@article{gao2025multiple,
  title={On Multiple Robustness of Proximal Dynamic Treatment Regimes},
  author={Gao, Yuanshan and Bai, Yang and Cui, Yifan},
  journal={arXiv preprint arXiv:2510.20451},
  year={2025}
}

@article{yang2023forster,
  title={Forster-warmuth counterfactual regression: A unified learning approach},
  author={Yang, Yachong and Kuchibhotla, Arun Kumar and Tchetgen Tchetgen, Eric},
  journal={arXiv preprint arXiv:2307.16798},
  year={2023}
}

@article{zeng2025nonparametric,
  title={Nonparametric estimation of local treatment effects with continuous instruments},
  author={Zeng, Zhenghao and Levis, Alexander W and Lee, JungHo and Kennedy, Edward H and Keele, Luke},
  journal={arXiv preprint arXiv:2504.03063},
  year={2025}
}

@article{zhang2025doubly,
  title={Doubly robust inference on causal derivative effects for continuous treatments},
  author={Zhang, Yikun and Chen, Yen-Chi},
  journal={arXiv preprint arXiv:2501.06969},
  year={2025}
}

@article{park2026proximal,
  title={Proximal causal inference for conditional separable effects},
  author={Park, Chan and Stensrud, Mats J and Tchetgen Tchetgen, Eric J},
  journal={Journal of the Royal Statistical Society Series B: Statistical Methodology},
  pages={qkag051},
  year={2026},
  publisher={Oxford University Press UK}
}

@article{yu2025fortified,
  title={Fortified proximal causal inference with many invalid proxies},
  author={Yu, Myeonghun and Shi, Xu and Tchetgen, Eric J Tchetgen},
  journal={arXiv preprint arXiv:2506.13152},
  year={2025}
}
\end{document}